\documentclass[letter,times,10pt,usenames,dvipsnames]{IEEEtran}

\usepackage{multicol}
\usepackage{amsmath}
\usepackage{amssymb}
\usepackage{balance}

\usepackage{pgf}
\usepackage{tikz}
\usetikzlibrary{arrows,automata}
\usetikzlibrary{positioning}
\usetikzlibrary{shapes.geometric}
\usetikzlibrary{plotmarks}
\usetikzlibrary{patterns}
\usetikzlibrary{calc}
\usetikzlibrary{fit}
\usetikzlibrary{shapes.misc}
\usetikzlibrary{backgrounds}
\usetikzlibrary{decorations.markings}
\usetikzlibrary{fit}

\usetikzlibrary{pgfplots.groupplots}

\usepackage{amsthm}

\makeatletter
\newtheoremstyle{sig}
  {}
  {}
  {\itshape}
  {}
  {\scshape}
  {.}
  {.5em}
  {#1\@ifnotempty{#2}{ #2}\thmnote{\quad(#3)}}%
\makeatother

\theoremstyle{sig}

\usepackage{mathpartir}
\usepackage{thmtools}
\usepackage{thm-restate}

\usepackage{paralist}
\usepackage{listings}
\usepackage{color}
\usepackage{graphicx}
\usepackage{float}
\usepackage{textcomp}
\usepackage{xspace}
\usepackage{pifont}
\usepackage{multirow}
\usepackage{algorithm}
\usepackage[noend]{algpseudocode}
\usepackage{pgfplots,booktabs}

\usepackage{enumitem}
\usetikzlibrary{snakes}
\tikzstyle{printersafe}=[snake=snake,segment amplitude=0 pt]
\usepackage{dirtree}
\usepackage{listings}
\usepackage{syntax}
\usepackage{bm}
\usepackage{bigfoot}
\usepackage{footnote}
\usepackage{stmaryrd}
\usepackage{mathtools}
\usepackage{diagbox}
\usepackage{tabularx}
\usepackage{xifthen}

\usepackage{xcolor}

\usepackage{mathptmx} 
\usepackage{fancyhdr}
\usepackage[normalem]{ulem}
\usepackage[hyphens]{url}
\usepackage[sort,nocompress]{cite}
\usepackage[keeplastbox]{flushend}

\usepackage{enumitem}
\setlistdepth{15} 
\newlist{longitemize}{itemize}{15}
\setlist[longitemize,1]{label=\textbullet}
\setlist[longitemize,2]{label=\textbullet}
\setlist[longitemize,3]{label=\textbullet}
\setlist[longitemize,4]{label=\textbullet}
\setlist[longitemize,5]{label=\textbullet}
\setlist[longitemize,6]{label=\textbullet}
\setlist[longitemize,7]{label=\textbullet}
\setlist[longitemize,8]{label=\textbullet}
\setlist[longitemize,9]{label=\textbullet}
\setlist[longitemize,10]{label=\textbullet}
\setlist[longitemize,11]{label=\textbullet}
\setlist[longitemize,12]{label=\textbullet}
\setlist[longitemize,13]{label=\textbullet}
\setlist[longitemize,14]{label=\textbullet}
\setlist[longitemize,15]{label=\textbullet}

\usepackage{rotating}
\usepackage{multirow}

\usepackage{subcaption}

\definecolor{ao(english)}{rgb}{0.0, 0.5, 0.0}
\definecolor{royalblue(web)}{rgb}{0.25, 0.41, 0.88}

\newcommand{\comment}[2][1=]
{}

\newcommand{\setCommentColor}[1]{%
	\ifthenelse{\equal{#1}{bk}}%
		{\colorlet{colorVar}{red!50}}%
		{\ifthenelse{\equal{#1}{pv}}%
			{\colorlet{colorVar}{blue}}%
			{\ifthenelse{\equal{#1}{mg}}%
				{\colorlet{colorVar}{ao(english)}}%
			{\ifthenelse{\equal{#1}{jr}}%
				{\colorlet{colorVar}{magenta}}%
				{}%
			}%
		}%
	}%
}

\newcommand{\commentAuthor}[1]{%
	\ifthenelse{\equal{#1}{bk}}%
		{Boris:\ }%
		{\ifthenelse{\equal{#1}{pv}}%
			{Pepe:\ }%
			{\ifthenelse{\equal{#1}{mg}}%
				{Marco:\ }%
			{\ifthenelse{\equal{#1}{jr}}%
				{Jan:\ }%
				{}%
			}%
		}%
	}%
}

\definecolor{codegray}{gray}{0.95}

\theoremstyle{definition}
\newtheorem{definition}{Definition}
\newtheorem{example}{Example}

\theoremstyle{theorem}

\newtheorem{lemma}{Lemma}

\newcommand{\tup}[1]{\langle #1 \rangle}

\newcommand{\pc}{\mathit{pc}}
\newcommand{\buf}{\mathit{buf}}

\newcommand{\exprEval}[2]{(\!\!| #1 |\!\!)(#2)}

\newcommand{\lowequiv}{\simeq_{\pi}}
\newcommand{\policy}{\pi}%
\newcommand{\high}{H}
\newcommand{\low}{L}

\newcommand{\srclang}{$\mu$\textsc{Asm}}

\newcommand{\elt}[2]{#1|_{#2}}
\newcommand{\emptysequence}{\varepsilon}
\newcommand{\concat}{\cdot}
\newcommand{\decrement}[1]{\mathit{decr}(#1)}

\renewcommand{\paragraph}[1]{\smallskip\noindent\textbf{#1:\ }}

\newcommand{\highlightBox}[1]{\colorbox{blue!10}{$#1$}}

\newcommand{\Var}{\mathit{Regs}}
\newcommand{\Val}{\mathit{Vals}}
\newcommand{\Nat}{\mathbb{N}}

\newcommand{\Conf}{\mathit{State}}

\newcommand{\Prg}{\mathit{p}}

\newcommand{\lbl}{\ell}
\newcommand{\kywd}[1]{\mathbf{#1}}
\newcommand{\skipKywd}{\kywd{skip}}
\newcommand{\storeKywd}{\kywd{store}}
\newcommand{\loadKywd}{\kywd{load}}
\newcommand{\jmpKywd}{\kywd{jmp}}
\newcommand{\jzKywd}{\kywd{beqz}}
\newcommand{\pcKywd}{\kywd{pc}}
\newcommand{\barrierKywd}{\kywd{spbarr}}
\newcommand{\timerKywd}{\kywd{probe}}
\renewcommand{\pc}{\pcKywd}
\definecolor{Blue3}{HTML}{0000CD}
\newcommand{\obsKywd}[1]{\textcolor{Blue3}{\mathtt{#1}}}
\newcommand{\loadObsKywd}{\obsKywd{load}}
\newcommand{\storeObsKywd}{\obsKywd{store}}
\newcommand{\pcObsKywd}{\obsKywd{pc}}
\newcommand{\loadObs}[1]{\loadObsKywd\ #1}
\newcommand{\storeObs}[1]{\storeObsKywd\ #1}
\newcommand{\pcObs}[1]{\pcObsKywd\ #1}
\newcommand{\startObsKywd}[1]{\obsKywd{start}}
\newcommand{\commitObsKywd}[1]{\obsKywd{commit}} 
\newcommand{\rollbackObsKywd}[1]{\obsKywd{rollback}}

\newcommand{\commitObs}[1]{\commitObsKywd{}} %
\newcommand{\rollbackObs}[1]{\rollbackObsKywd{}} %
\newcommand{\unaryOp}[1]{\ominus #1}
\newcommand{\binaryOp}[2]{#1 \otimes #2}

\newcommand{\pseq}[2]{#1 ; #2}
\newcommand{\pskip}{\skipKywd{}}
\newcommand{\passign}[2]{#1 \leftarrow #2}
\newcommand{\pmarkedassign}[2]{#1 \leftarrow^* #2}
\newcommand{\pload}[2]{\loadKywd\ #1, #2}
\newcommand{\pstore}[2]{\storeKywd\ #1, #2}
\newcommand{\pcondassign}[3]{#1 \xleftarrow{#3?} #2}
\newcommand{\pjmp}[1]{\jmpKywd\ #1}
\newcommand{\pjz}[2]{\jzKywd\ #1, #2}
\newcommand{\pbarrier}{\barrierKywd}
\newcommand{\ptimer}[1]{#1 \leftarrow \timerKywd}
\newcommand{\select}[2]{#1(#2)}

\newcommand{\ite}[3]{\mathbf{ite}(#1,#2,#3)}
\newcommand{\tagged}[2]{#1@#2}
\newcommand{\notags}{\varepsilon}

\newcommand{\conf}{\mathit{\sigma}}
\newcommand{\muconf}{\mu}

\newcommand{\Directives}{\mathit{Dir}}
\newcommand{\fetch}[1]{
\ifthenelse{\equal{#1}{}}{\kywd{fetch}}{\kywd{fetch}}%
}
\newcommand{\execute}[1]{\kywd{execute}\ #1}
\newcommand{\retire}{\kywd{retire}}

\newcommand{\apply}[2]{\mathit{apl}(#1,#2)}

\newcommand{\mask}[1]{\mathit{mask}(#1)}
\newcommand{\drop}[1]{\mathit{drop}(#1)}

\newcommand{\labelNda}[3]{\mathit{label}_{\mathit{nda}}(#1, #2, #3)}

\newcommand{\mi}[1]{\ensuremath{\mathit{#1}}}
\newcommand{\mr}[1]{\ensuremath{\mathrm{#1}}}

\newcommand{\mf}[1]{\ensuremath{\mathbf{#1}}}

\newcommand{\ms}[1]{\ensuremath{\mathsf{#1}}}

\newcommand{\neutcol}[0]{black}
\newcommand{\stlccol}[0]{RoyalBlue}
\newcommand{\ulccol}[0]{RedOrange}

\newcommand{\idecol}[0]{Emerald}

\newcommand{\col}[2]{\ensuremath{{\color{#1}{#2}}}}

\newcommand{\archStyle}[1]{\ms{\col{\stlccol}{#1}}}
\newcommand{\muarchStyle}[1]{{\mf{\col{\ulccol }{#1}}}}
\newcommand{\interfStyle}[1]{{\mr{\col{\idecol }{#1}}}}
\newcommand{\neut}[1]{{\mi{\col{\neutcol }{#1}}}}

\newcommand{\archStep}[2]{\archStyle{\xrightarrow[\neut{#2}]{\neut{#1}}}}
\newcommand{\archStepCompact}{\archStyle{\rightarrow}}

\newcommand{\muarchSem}[1]{\muarchStyle{\{\!\!|} #1\muarchStyle{ |\!\!\} } }
\newcommand{\muarchStep}[2]{\muarchStyle{\xRightarrow{\neut{#1}}}}
\newcommand{\muarchStepCompact}{\muarchStyle{\Rightarrow}}

\newcommand{\interfSem}[1]{\interfStyle{\llbracket} #1\interfStyle{\rrbracket} }
\newcommand{\interfStep}[2]{\interfStyle{\xrightharpoonup{\neut{#1}}}}
\newcommand{\interfStepCompact}{\interfStyle{\rightharpoonup}}

\newcommand{\CtSeqInterf}[1]{\interfStyle{\llbracket} #1\interfStyle{\rrbracket}_{\interfStyle{ct}}^{\interfStyle{seq}}}
\newcommand{\CtSeqInterfStep}[2]{ { \interfStep{#1}{#2}_{\interfStyle{ct}}^{\interfStyle{seq}} } } 
\newcommand{\ArchSeqInterf}[1]{\interfStyle{\llbracket} #1\interfStyle{\rrbracket}_{\interfStyle{arch}}^{\interfStyle{seq}}}
\newcommand{\ArchSeqInterfStep}[2]{ { \interfStep{#1}{#2}_{\interfStyle{arch}}^{\interfStyle{seq}} } }
\newcommand{\CtPcSpecInterf}[1]{\interfStyle{\llbracket} #1\interfStyle{\rrbracket}_{\interfStyle{ct\text{-}pc}}^{\interfStyle{seq\text{-}spec}}}
\newcommand{\PcSpecInterf}[1]{\interfStyle{\llbracket} #1\interfStyle{\rrbracket}_{\interfStyle{pc}}^{\interfStyle{spec}}}
\newcommand{\CtPcSpecInterfStep}[2]{ { \interfStep{#1}{#2}^{\interfStyle{spec, ct-pc}} } }
\newcommand{\CtSpecInterf}[1]{\interfStyle{\llbracket} #1\interfStyle{\rrbracket}_{\interfStyle{ct}}^{\interfStyle{spec}}}
\newcommand{\CtSpecInterfStep}[2]{ { \interfStep{#1}{#2}_{\interfStyle{ct}}^{\interfStyle{spec}} } }
\newcommand{\ArchSpecInterf}[1]{\interfStyle{\llbracket} #1\interfStyle{\rrbracket}_{\interfStyle{arch}}^{\interfStyle{spec}}}
\newcommand{\ArchSpecInterfStep}[2]{ { \interfStep{#1}{#2}^{\interfStyle{spec, arch}} } }

\newcommand{\InftyInterf}[1]{\interfStyle{\llbracket} #1\interfStyle{\rrbracket}_{\interfStyle{\bot}}}

\newcommand{\hsni}[2]{#2 \vdash #1}
\newcommand{\hsniViolation}[2]{#2 \not\vdash #1}
\newcommand{\ct}[3]{#1 \vdash \mathit{NI}(#2,#3)}

\newcommand{\ProcConfs}{\neut{HwStates}}

\newcommand{\CtxMuarchStep}[3]{ {\muarchStepCompact}_{\muarchStyle{#3}} }
\newcommand{\CtxMuarchSem}[2]{ {\muarchSem{#1}}_{\muarchStyle{#2}} }

\newcommand{\SeqProcMuarchStep}[2]{ \CtxMuarchStep{#1}{#2}{seq} }  
\newcommand{\SeqProcMuarchSem}[1]{ \CtxMuarchSem{#1}{seq} }

\newcommand{\LoadDelayMuarchStep}[2]{ \CtxMuarchStep{#1}{#2}{loadDelay} }  
\newcommand{\LoadDelayMuarchSem}[1]{ \CtxMuarchSem{#1}{loadDelay} }

\newcommand{\TtMuarchStep}[2]{ \CtxMuarchStep{#1}{#2}{tt} }  
\newcommand{\TtMuarchSem}[1]{ \CtxMuarchSem{#1}{tt} }

\newcommand{\update}[2]{{#1} \uplus {#2}}

\newcommand{\adversary}{\mathcal{A}}

\newcommand{\CacheStates}{\mathit{CacheStates}}
\newcommand{\CacheState}{\mathit{cs}}
\newcommand{\CacheAccess}{\mathit{access}}
\newcommand{\CacheUpdate}{\mathit{update}}
\newcommand{\CacheHit}{\mathtt{Hit}}
\newcommand{\CacheMiss}{\mathtt{Miss}}

\newcommand{\BpStates}{\mathit{BpStates}}
\newcommand{\BpState}{\mathit{bp}}
\newcommand{\BpUpdate}{\mathit{update}}
\newcommand{\BpPredict}{\mathit{predict}}

\newcommand{\ReorderBuffers}{\mathit{Bufs}}
\newcommand{\BufProject}[1]{ {#1}\!\!\downarrow }
\newcommand{\resolved}{\mathtt{R}}
\newcommand{\unresolved}{\mathtt{UR}}

\newcommand{\SchedStates}{\mathit{ScStates}}
\newcommand{\SchedState}{\mathit{sc}}
\newcommand{\SchedUpdate}{\mathit{update}}
\newcommand{\SchedNext}{\mathit{next}}

\newcommand{\wMuarch}{\muarchStyle{w}}
\newcommand{\wInterf}{\interfStyle{w}}

\newcommand{\unlabelNda}[2]{\mathit{unlbl}(#1,#2)}
\renewcommand{\labelNda}[3]{\mathit{lbl}(#1,#2,#3)}
\newcommand{\labelled}[2]{\langle #1\rangle_{#2}}

\newcommand{\transmitGadget}[1]{\mathit{transmit}(#1)}
\newcommand{\labels}[1]{\mathit{labels}(#1)}
\newcommand{\deriveLabels}[2]{\mathit{derive}(#1,#2)}

\newcommand{\toExecute}{\mathit{exec}}

\definecolor{Blue3}{HTML}{0000CD}
\definecolor{Green4}{HTML}{008B00}
\definecolor{Red3}{HTML}{CD0000}
\definecolor{orange}{rgb}{0.8, 0.47, 0.196}
\lstdefinestyle{Cstyle}
{
	frame = tb,
  belowskip=.4\baselineskip,
  aboveskip=.4\baselineskip,
  	showstringspaces = false,
  	breaklines = true,
  	breakatwhitespace = true,
  	tabsize = 3,
  	numbers = left,
    stepnumber = 1,
    numberstyle = \tiny\color{gray},
    language = {[ANSI]C},
    alsoletter={.\$},
    basicstyle={\ttfamily\color{black}},
    keywordstyle={\ttfamily\color{Blue3}},
    keywordstyle=[2]{\ttfamily\color{Green4}},
    keywordstyle=[3]{\ttfamily\color{orange}},
    keywordstyle=[4]{\ttfamily\color{violet}},
    otherkeywords = {skip,not},
    morekeywords = [2]{A,B},
    morekeywords = [3]{},
	morekeywords = [4]{y,x,z,w, size,size_A,k,temp},
	morecomment=[l][\small\itshape\color{purple!40!black}]{//},
	sensitive=true,
}

\usepackage{letltxmacro}
\newcommand*{\SavedLstInline}{}
\LetLtxMacro\SavedLstInline\lstinline
\DeclareRobustCommand*{\lstinline}{%
  \ifmmode
    \let\SavedBGroup\bgroup
    \def\bgroup{%
      \let\bgroup\SavedBGroup
      \hbox\bgroup
    }%
  \fi
  \SavedLstInline
}
\newcommand{\inlineCcode}[1]{\lstinline[style=Cstyle]|#1|}

\newcommand{\archstate}{architectural state}
\newcommand{\uarchstate}{microarchitectural state}

\newcommand{\Sat}{Y}
\newcommand{\NSat}{N}

\newcommand{\DeepProject}[1]{{#1}\!\!\Downarrow}

\newcommand{\wellformed}[1]{\mathit{wf}(#1)}

\newcommand{\nrPc}[1]{\mathit{\#Instr}(#1)}

\newcommand{\prefixes}[1]{\mathit{prefixes}(#1)}

\newcommand{\headConf}[1]{\mathit{headConf}(#1)}
\newcommand{\headWindow}[1]{\mathit{headWndw}(#1)}
\newcommand{\maxOf}[1]{\mathit{max}(#1)}
\newcommand{\minOf}[1]{\mathit{min}(#1)}

\newcommand{\nrMispred}[2]{\#\mathit{mispr}(#1, #2)}
\newcommand{\nextPairs}[1]{\mathit{next}(#1)}
\newcommand{\map}[3]{\mathit{corr}_{#1,#2}(#3)}
\newcommand{\hrun}{\muarchStyle{hr}}
\newcommand{\crun}{\interfStyle{cr}}
\newcommand{\hrunp}{\muarchStyle{hr'}}
\newcommand{\crunp}{\interfStyle{cr'}}

\newcommand{\Mispred}[1]{\mathit{mispred}(#1)}
\newcommand{\bufEquiv}[1]{\equiv_{#1}}

\newcommand{\techReportAppendix}[1]{%
	\ifthenelse{\equal{\shortVersion}{true}}%
		{\cite{technicalReport}}%
        {Appendix~\ref{#1}}%
}%

\newcommand{\techReportAppendices}[2]{%
	\ifthenelse{\equal{\shortVersion}{true}}%
		{\cite{technicalReport}}%
		{Appendices~\ref{#1}--\ref{#2}}%
}%

\newcommand{\onlyTechReport}[1]{%
	\ifthenelse{\equal{\shortVersion}{true}}%
		{}%
		{#1}%
}%

\newcommand{\onlyShortVersion}[1]{%
	\ifthenelse{\equal{\shortVersion}{true}}%
		{#1}%
		{}%
}%

\newcommand{\shortVersion}{true} 

\def\tablename{Table}
\title{Hardware-Software Contracts for\\Secure Speculation}

\author{\IEEEauthorblockN{Marco Guarnieri\IEEEauthorrefmark{1},
Boris K\"{o}pf\IEEEauthorrefmark{2}, 
Jan Reineke\IEEEauthorrefmark{3}, and
Pepe Vila\IEEEauthorrefmark{1}}\\
\IEEEauthorblockA{\IEEEauthorrefmark{1}\textit{IMDEA Software Institute}\qquad
\IEEEauthorrefmark{2}\textit{Microsoft Research}\qquad
\IEEEauthorrefmark{3}\textit{Saarland University}}}

\begin{document}
\maketitle

\begin{abstract}
Since the discovery of Spectre, a large number of hardware mechanisms for secure
speculation has been proposed. Intuitively, more defensive mechanisms are less
efficient but can securely execute a larger class of programs, while more
permissive mechanisms may offer more performance but require more defensive
programming. Unfortunately, there are no hardware-software contracts that would
turn this intuition into a basis for principled co-design.

In this paper, we put forward a framework for specifying such contracts, and we
demonstrate its expressiveness and flexibility. 

On the hardware side, we use the
framework to provide the first formalization and comparison of the security
guarantees provided by a representative class of mechanisms for secure
speculation.

On the software side, we use the framework to characterize program properties
that guarantee secure co-design in two scenarios  traditionally investigated in
isolation: (1) ensuring that a benign program does not leak information while
computing on confidential data, and (2) ensuring that a potentially malicious
program cannot read outside of its designated sandbox. 
Finally, we show how the properties corresponding to both scenarios can be
checked based on existing tools for software verification, and we use them to
validate our findings on executable code.
\end{abstract}

\vspace{-2mm}
\section{Introduction}\label{sec:introduction}

Speculative execution avoids expensive pipeline stalls by predicting the outcome of branching (and other) decisions, and by continuing the execution based on these predictions.
When a prediction turns out to be incorrect, the processor rolls back the effects of speculatively executed instructions on the architectural state consisting of registers, flags, and main memory.\looseness=-1

However, the microarchitectural state, which includes the content of various caches and buffers, is not (or only partially) rolled back.
This side effect can leak information about the speculatively accessed data and thus violate confidentiality, see Figure~\ref{figure:v1-vanilla}.
Spectre attacks~\cite{Kocher2018spectre, spectreSoK}  demonstrate that this vulnerability affects all modern general-purpose processors and poses a serious threat for platforms with multiple tenants.

A multitude of hardware mechanisms for secure speculation have been proposed.  %
They are based on a number of basic ideas, such as delaying load operations until they cannot be squashed~\cite{specshadow2019}, delaying operations that depend on speculatively loaded data~\cite{nda2019weisse,STT2019}, limiting the effect of speculatively executed instructions~\cite{invisispec2018,safespec2019,DAWG2018,ainsworth2020muontrap}, or rolling back the microarchitectural state when a misprediction is detected~\cite{saileshwar2019cleanupspec}.

Intuitively, more defensive mechanisms are less efficient but can securely execute a larger class of programs, while more permissive mechanisms offer more performance but require more defensive programming.
We refer to this intuition as (*).\looseness=-1

For example, consider the  variant of Spectre v1 shown in Figure~\ref{figure:v1-variant}, where  array  \inlineCcode{A} is accessed before the bounds check.
\begin{figure}	
	\begin{subfigure}[t]{0.238\textwidth}
		\small
	\begin{lstlisting}[basicstyle=\small,style=Cstyle]
if (y < size_A) 
 x = A[y];
 temp &= B[x * 64];
	\end{lstlisting}
	\caption{Program $P_1$}\label{figure:v1-vanilla}
	\end{subfigure}
	\begin{subfigure}[t]{0.238\textwidth}
	\small
	\begin{lstlisting}[basicstyle=\small,style=Cstyle]
x = A[y];
if (y < size_A)
 temp &= B[x * 64];
	\end{lstlisting}
	\caption{Program $P_2$}\label{figure:v1-variant}
	\end{subfigure}
		\caption{
		Program $P_1$ is the vanilla Spectre v1 example, where \inlineCcode{A[y]} can be speculatively read and leaked into the data cache via an access to array \inlineCcode{B}, for \inlineCcode{y >= size_A}.	
		Program $P_2$, is a variant where \inlineCcode{A[y]} is accessed non-speculatively before the bounds check but the leak occurs during speculative execution.
		}\label{figure:v1-ctrl-flow}
		\end{figure}
Mechanisms  delaying loads until they cannot be squashed~\cite{specshadow2019} prevent speculatively leaking \inlineCcode{A[y]}, for $\inlineCcode{y}\ge \inlineCcode{size_A}$. In contrast, more permissive mechanisms that delay only loads depending on speculatively accessed data~\cite{nda2019weisse,STT2019} do {\em not} prevent the leak, because~$\inlineCcode{A[y]}$ is accessed non-speculatively.

While the performance characteristics of secure speculation mechanisms  are well-studied,
  there has been little work on
\begin{inparaenum}[(1)]
	\item \label{it:hw}characterizing the security guarantees they provide, and in particular on
	\item \label{it:sw}investigating how these guarantees can be effectively leveraged by software to achieve global security guarantees.\footnote{A notable exception to~\eqref{it:hw} is STT~\cite{STT2019}, which is backed by a security property that guarantees the confidentiality of speculatively loaded data. However, this property alone does not provide an actionable basis for \eqref{it:sw}, as preventing leakage of non-speculatively accessed data (as in Figure~\ref{figure:v1-variant}) is declared out of scope~\cite[Section 4]{STT2019}.}
\end{inparaenum}
That is, we lack hardware-software contracts that support principled co-design for secure speculation, and that would formalize the intuition (*) described above.

\paragraph{Contracts}
In this paper, we put forward a framework for specifying such contracts, %
 based on three basic building blocks: an ISA language, a model of the microarchitecture, and an adversary model specifying which microarchitectural components (such as caches or branch predictor state) are observable via side-channels.

Contracts specify which program executions a side-channel adversary can distinguish.
A contract in our framework is defined in terms of {\em executions} and {\em observations} made on these executions, and it is formalized in terms of a labelled ISA semantics.  A CPU satisfies a contract if, whenever two program executions agree on all observations, they are guaranteed to be indistinguishable by the adversary at the microarchitectural level. The contract semantics can mandate exploration of mispredicted paths, effectively requiring agreement on observations corresponding to transient instructions.

Secrets at the program level must not affect contract observations, because then they can become visible to the adversary. Hence, contracts exposing more observations correspond to hardware with weaker security guarantees, whereas contracts exposing fewer observations correspond to hardware with stronger guarantees. The extreme case is a contract with no observations, which is satisfied by an ideal side-channel resilient platform that can securely execute every program.

\paragraph{Software Side}
Our framework provides a basis for deriving requirements that {\em software} needs to satisfy to run securely on a specific platform.
For deriving such requirements, we consider two scenarios  typically considered in the literature:
\begin{asparaitem}
	\item In the first scenario, called ``constant-time programming'', the goal is to ensure that a benign program, such as a cryptographic algorithm, does not leak information while computing on confidential data.
	\item In the second scenario, which we call ``sandboxing'',  the goal is  to  restrict the memory region that a potentially malicious program, such as a Web application, can read from.\footnote{In the terminology of~\cite{spectreSoK}, sandboxing aims to block disclosure gadgets.}
\end{asparaitem}

For each scenario, we identify program-level properties that guarantee security on hardware that satisfies a given contract.
We stress that secure speculation  approaches usually {\em either} consider constant-time programming~\cite{spectector2020,SyslevelNI,constanttime2019,balliu2019inspectre} {\em or} sandboxing~\cite{loadhardening2018,microsoft-mitigation}. In contrast, our framework supports {\em both} goals through  program-level properties. %

We provide tool support for automatically checking if programs are secure in both scenarios.
For this, we extend a static analysis tool for detecting speculative leaks~\cite{spectector2020} to cater for  different contracts, and we use it to validate all examples used in the paper on x86 executable code.

\paragraph{Hardware Side}
We use our framework to define contracts for a comprehensive set of recent hardware mechanisms for secure speculation: disabling speculation, delaying speculative load operations~\cite{specshadow2019}, and speculative taint tracking~\cite{nda2019weisse,STT2019}.

To this end, we formalize each mechanism in the context of a variant of the simple speculative out-of-order processor from~\cite{constanttime2019} and we prove that it satisfies  specific contracts against an adversary that observes caches, predictors, and (part of) the reorder buffer during  execution.
We show that the contracts we define form a lattice, and we use this to give, for the first time, a rigorous comparison of the security guarantees offered by different secure speculation mechanisms.

Our analysis highlights that the studied mechanisms~\cite{specshadow2019,nda2019weisse,STT2019} prevent leaks of speculatively accessed data, and confirms the results of~\cite{STT2019}. For software, this means that ``sandboxing'' is supported out-of-the-box, in the sense that programs only need to place appropriate bounds checks, but no speculation barriers.\looseness=-1

Our analysis also shows that the mechanisms offer no support for ``constant-time programming''. This means
that programs that are constant-time in the traditional sense~\cite{AlmeidaBBDE16} still require additional checks~\cite{spectector2020,constanttime2019} or insertion of speculation barriers~\cite{blade}, even if hardware mechanisms for secure speculation are deployed.

\paragraph{Summary of contributions}
We propose a novel framework for expressing security contracts between hardware and software.
Our framework is expressive enough to (1) characterize the security guarantees provided by recent proposals for secure speculation, and (2) provide program-level properties formalizing how to leverage these hardware guarantees to achieve global, end-to-end security for different scenarios.
From a theoretical perspective, we provide the first characterization of security for a comprehensive class of hardware mechanisms for secure speculation. %
From a practical perspective, we show how to automate checks for programs to run securely on top of these mechanisms.

\onlyShortVersion{
\paragraph{Bonus material} 
A technical report containing a full formalization and proofs of all technical results is available at~\cite{technicalReport}.
}

\section{ISA language, semantics, and adversaries}\label{sec:model}

We introduce the  foundations for specifying hardware-software contracts: an ISA language (\S\ref{sec:contracts:isa}),  its architectural semantics  (\S\ref{sec:contracts:semantics}), a general notion of hardware semantics (\S\ref{sec:model:hw-semantics}), and an adversary  model  capturing  which  aspects  of  the microarchitecture are observable via side channels (\S\ref{sec:model:adversary}).\looseness=-1

\subsection{ISA language}\label{sec:contracts:isa}

For modeling the ISA we rely on \srclang{}, a simple assembly language from~\cite{spectector2020} with the following syntax:\\[-2mm] %

	{\centering
	\begin{tabular}{llcl}
	\multicolumn{4}{l}{\bf Basic Types} \\
	\textit{(Registers)} 	&  $x$		& $\in$ & $\Var$ \\
	\textit{(Values)} 		&  $n, \lbl$ 		& $\in$ & $\Val = \Nat \cup \{\bot\}$  \\
	\multicolumn{4}{l}{\bf Syntax} \\
	\textit{(Expressions)} 	&  $e$		& $:=$ & $n \mid x \mid \unaryOp{e} \mid \binaryOp{e_1}{e_2} \mid \ite{e_1}{e_2}{e_3}$ \\
	\textit{(Instructions)} 	&  $i$ 		& $:=$ & $\pskip{} \mid \passign{x}{e} \mid \pload{x}{e} \mid \pstore{x}{e}$ \\
	& & &  $\mid \pjmp{e} \mid  \pjz{x}{\lbl} \mid \pbarrier{}$ \\ %
	\textit{(Programs)}		&  $p$		& $:=$ & $i \mid \pseq{p_1}{p_2} $%
	\end{tabular}
	}

\begin{asparaitem}
\item \srclang{} expressions are built from a set of register identifiers $\Var$, which contains a designated element $\pc$ representing the program counter, and a set of values~$\Val$, which consists of the natural numbers and $\bot$.
\item \srclang{} instructions include assignments, load and store instructions,  indirect jumps, branching instructions, and a speculation barrier $\pbarrier{}$.
\item \srclang{} programs are sequences of instructions. %
\end{asparaitem}

\subsection{Architectural semantics $\archStepCompact$}
\label{sec:contracts:semantics}

The architectural semantics models the execution of \srclang{} programs at the architectural level. 
It is defined in terms of \textit{architectural states} $\conf=\tup{m, a}$ consisting of a \textit{memory} $m$ and a {\em register assignment}  $a$.
Memories $m$ map memory addresses, represented by natural numbers,  to values in $\Val$.
Register assignments $a$ map register identifiers to values in $\Val$.
We signal program termination by assigning the special value $\bot$ to the program counter $\pc$. %

The architectural semantics is a deterministic binary relation $\sigma \archStepCompact \sigma'$, which we formalize in \techReportAppendix{appendix:arch-sem},  mapping an \archstate{} $\sigma$ to its successor $\sigma'$.  
A {\em run} is a finite sequence of states $\sigma_0, \ldots, \sigma_n$ with $\conf_0 \archStepCompact \dots \archStepCompact \conf_n$ such that $\conf_0$ is initial  (that is, all registers including $\pc$ have value $0$) and $\conf_n$ is final (that is, $\conf_n(\pc) = \bot$).

\subsection{Hardware semantics $\muarchStepCompact$}\label{sec:model:hw-semantics}
A hardware semantics models the execution of \srclang{} programs at the microarchitectural level. 
Here we describe a general notion of hardware semantics with the key aspects necessary for explaining hardware-software contracts; we provided multiple, concrete hardware semantics modeling different processors and countermeasures in \S\ref{sec:hw-side}--\ref{sec:countermeasures}.

Hardware semantics are defined in terms of {\em hardware states} $\tup{\sigma,\mu}$ consisting of an \archstate{} $\sigma$ (as before) and a {\em microarchitectural state} $\mu$, which models the state of components like predictors, caches, and  reorder buffer. %

A hardware semantics is a deterministic relation $\muarchStepCompact$ mapping hardware states $\tup{\conf,\mu}$  to their successors $\tup{\conf',\muconf'}$. 
A {\em hardware run} is a sequence $\tup{\conf_0,\muconf_0} \muarchStepCompact \dots \muarchStepCompact \tup{\conf_n,\muconf_n}$ such that $\tup{\conf_0,\muconf_0}$ is initial  and $\tup{\conf_n, \muconf_n}$ is final. For this, we assume that there is a fixed, initial microarchitectural state $\mu_0$, where, for instance, the reorder buffer is empty and all caches have been invalidated.

\subsection{Adversary model}\label{sec:model:adversary}
We consider adversaries that can observe parts of the \uarchstate{} during execution.
We model hardware observations as projections to parts of the \uarchstate{}.
For instance, a cache adversary can be modeled as a function $\adversary$ projecting $\mu$ to its cache component.
In the paper, we consider an adversary $\adversary$ that has access to the state of caches, predictors, and (part of) the reorder buffer; we formalize $\adversary$ in 
Section~\ref{sec:hw-side:adversaries}.

Given a program $p$,   $\muarchSem{p}(\sigma)$ denotes the trace $\adversary(\muconf_0) \concat \ldots \concat \adversary(\muconf_n)$ of hardware observations produced in the run $\tup{\conf,\muconf_0} \muarchStepCompact \dots \muarchStepCompact \tup{\conf_n,\muconf_n}$. %
We refer to $\muarchSem{p}$ as the {\em hardware trace semantics} (hardware semantics for short) of program $p$.

\section{Hardware-software contracts}\label{sec:contracts:contracts}

The purpose of a contract is to split the responsibilities for preventing side channels between software and hardware.

We first formalize the general notion of contracts and we specify  when a hardware platform satisfies a contract.
Then we present several fundamental contracts for secure speculation.

\subsection{Formalizing contracts}\label{sec:contracts:contracts:formalizing}
A {\em contract} is a labeled, deterministic semantics $\interfStepCompact$ for the~ISA.
Given a program~$p$ and an initial \archstate{}~$\sigma_0$, the labels on the transitions of the corresponding run $\sigma_0 \overset{l_1}{\interfStepCompact} \sigma_1 \overset{l_2}{\interfStepCompact} \dots \overset{l_n}{\interfStepCompact} \sigma_n$ define the {\em trace} $\interfSem{p}(\sigma_0) = l_1 l_2 \dots l_n$.
In this paper, we only consider terminating programs. 
We leave the extension to non-terminating programs as future work.

The traces of a contract $\interfSem{p}$ capture which \archstate{}s are guaranteed to be indistinguishable to an attacker on a hardware satisfying the contract, which is formalized below.

\begin{definition}[$\hsni{\interfSem{\cdot}}{\muarchSem{\cdot}}$]\label{def:hni}
A hardware semantics $\muarchSem{\cdot}$ {\em satisfies a contract  $\interfSem{\cdot}$} if, for all programs~$\Prg$ and all initial \archstate{}s~$\sigma,\sigma'$, if $\interfSem{\Prg}(\sigma) = \interfSem{\Prg}(\sigma')$, then $\muarchSem{\Prg}(\sigma)= \muarchSem{\Prg}(\sigma')$.
\end{definition}

Different contracts correspond to different divisions of security obligations between software and hardware:
secrets at the program level must not affect contract observations, because then they can become visible to the adversary. Hence, contracts exposing more observations correspond to hardware with weaker security guarantees, whereas contracts exposing fewer observations correspond to hardware with stronger security guarantees.
A degenerate case is a contract with no observations, which is satisfied by an ideal side-channel resilient platform that securely executes every program.

\subsection{Contracts for secure speculation}\label{sec:contracts:contracts:contracts}

We now define four fundamental contracts that characterize the security guarantees offered by mechanisms for secure speculation. 
We derive our contracts, which we fully formalize in \techReportAppendix{appendix:contracts},  as the combination of two kinds of building blocks.

\subsubsection{Building blocks for contracts}

The first building block are {\em observer modes}, which govern what information a contract exposes. We  define them via labels on the contract semantics.
\begin{asparaitem}
\item The {\em constant-time} observer mode ($\interfStyle{ct}$ for short) is commonly used when reasoning about side channels in cryptographic algorithms. It uses labels $\pcObs{\ell}$, $\loadObs{n}$, and $\storeObs{n}$ to expose the value $\ell$ of the program counter and the addresses $n$ of load and store operations. The $\interfStyle{ct}$ observer mode can be augmented with support for variable-latency instructions by additionally exposing the operands of those instructions as observations, or refined to capture adversaries that can infer addresses of memory accesses only up to the granularity of cache banks, lines, or pages~\cite{dkpldi17}. We forgo both extensions for simplicity.
\item The {\em architectural} observer mode ($\interfStyle{arch}$ for short) additionally exposes the {\em value}~$v$ that is loaded from memory location~$n$ via the label $\loadObs{n=v}$ upon each load instruction. 
As registers are set to zero in the initial \archstate{},   $\interfStyle{arch}$ traces effectively determine the values of all registers during execution.

\end{asparaitem}

The second building block are {\em execution modes} that characterize which paths need to be explored to collect observations. For processors with speculative execution, depending on the presence and effectiveness of hardware-level countermeasures, it is necessary to go beyond paths covered by the architectural semantics.
\begin{asparaitem}
\item In the {\em sequential} execution mode ($\interfStyle{seq}$ for short), programs are executed sequentially and in-order following the architectural semantics.
\item In the {\em always-mispredict} execution mode ($\interfStyle{spec}$ for short), programs are executed sequentially, but incorrect branches are also executed for a bounded number of steps before backtracking.
This execution mode is based on~\cite{spectector2020} and can be used to explore the effects of speculatively executed instructions at the ISA level. 
\end{asparaitem}

\subsubsection{Contract $\CtSeqInterf{\cdot}$}\label{contract:ct-seq}
This contract exposes the program counter and the locations of memory accesses on sequential, non-speculative paths; see Figure~\ref{figure:interface:seq-ct}.
$\CtSeqInterf{\cdot}$ is a fundamental baseline that is often implicitly assumed in practice, and that has also been formalized in~\cite{BartheBCL19,AlmeidaBBDE16}.

In Section~\ref{sec:countermeasures:sequential} we show that $\CtSeqInterf{\cdot}$ is satisfied by a simple in-order processor without speculation. However, modern out-of-order processors do {\em not} satisfy $\CtSeqInterf{\cdot}$, as shown below.

\begin{example}\label{ex:out-of-order}
Consider the vanilla Spectre v1 snippet from Figure~\ref{figure:v1-vanilla}, compiled to \srclang{}:
{%
\begin{lstlisting}[mathescape=true,style=Cstyle]
x $\leftarrow$ y < size_A
$\mathbf{beqz}$ x, $\bot$ //checking y < size_A
$\mathbf{load}$ z,A + y //accessing A[y]
z $\leftarrow$ z*64
$\mathbf{load}$ w, B+z //accessing B[A[y]*64]
\end{lstlisting}
}
Consider \archstate{}s $\sigma$ and $\sigma'$ that agree on the observations on trace $\pcObs{3}\concat\loadObs{(\inlineCcode{A+y})}\concat \loadObs{(\inlineCcode{B+z})}$ (and hence on the content of array \inlineCcode{A} within bounds), but for which $\sigma(\inlineCcode{A + y})  = 0$ and $\sigma'(\inlineCcode{A + y}) = 1$ for some \inlineCcode{y>size_A}. On processors with speculation, an adversary with cache access can distinguish $\sigma$ and $\sigma'$, as shown by Spectre attacks~\cite{Kocher2018spectre}.
\end{example}

Perhaps surprisingly, processors deploying recent proposals for secure speculation still violate $\CtSeqInterf{\cdot}$, see \S\ref{sec:countermeasures}.

\begin{figure}
\begin{mathpar}
	\inferrule[Load]
	{
	\select{p}{a(\pc)} = \pload{x}{e} \\
	\tup{m, a} \archStep{}{} \tup{m',a'}
	}
	{
	\tup{m, a} \CtSeqInterfStep{\loadObs{\exprEval{e}{a}}}{} \tup{m', a'}
	}

	\inferrule[Store]
	{
	\select{p}{a(\pc)} = \pstore{x}{e} \\
	\tup{m, a} \archStep{}{} \tup{m',a'}
	}
	{
	\tup{m, a} \CtSeqInterfStep{\storeObs{ \exprEval{e}{a} }}{} \tup{ m', a'}
	}

	\inferrule[Beqz-Sat]
	{
	\select{p}{a(\pc)} = \pjz{x}{\lbl} \\
	\tup{m, a} \archStep{}{} \tup{m',a'}
	}
	{
	\tup{m, a} \CtSeqInterfStep{\pcObs{a'(\pc)}}{} \tup{ m', a'}
	}
\end{mathpar}
\caption{ $\CtSeqInterf{\cdot}$ contract for a program $p$ - selected rules (here $\exprEval{e}{a}$ is the result of expression $e$ given assignment $a$). The contract is
obtained by augmenting the architectural semantics with observations $\loadObs{n}$, $\storeObs{n}$, and $\pcObs{\ell}$ exposing the addresses of loads, stores, and the program counter, respectively.
}\label{figure:interface:seq-ct}
\end{figure}

\subsubsection{Contract $\CtSpecInterf{\cdot}$}\label{contract:ct-spec}
	This contract additionally exposes the program counter and the locations of all memory accesses on speculatively executed paths. It is based on the speculative semantics from~\cite{spectector2020} and formalized in Figure~\ref{figure:interface:spec-ct}.

	In Section~\ref{sec:countermeasures}, we show that speculative out-of-order processors (with and without mechanisms for secure speculation) satisfy $\CtSpecInterf{\cdot}$.\looseness=-1 
	
	Consider again Example~\ref{ex:out-of-order}: by exposing observations on mispredicted paths, $\CtSpecInterf{\cdot}$ makes the states $\sigma, \sigma'$ distinguishable at the contract level, effectively delegating the responsibility of ensuring that $\inlineCcode{A + y}$ does not carry secret information for \inlineCcode{y >= size_A} to software.

\begin{figure*}[h]
	\begin{mathpar}
		\inferrule[Step]
		{
		p(\sigma(\pc))\neq \pjz{x}{\lbl}\\
		 \conf \CtSeqInterfStep{\tau}{} \conf'\\
		}
		{
			\tup{\sigma, \omega +1} \cdot s \CtSpecInterfStep{\tau}{} \tup{\sigma', \omega}\cdot s
		}

		\inferrule[Rollback]
		{
		s=\tup{\sigma',\omega'}\cdot s'\\
		}
		{
			\tup{\sigma,0 }\cdot s \CtSpecInterfStep{\pcObs{\sigma'(\pc)}}{} s
		}

		\inferrule[Barrier]
		{
		p(\sigma(\pc))= \pbarrier \\
		 \conf \CtSeqInterfStep{\tau}{} \conf'\\
		}
		{
			\tup{\sigma, \omega +1} \cdot s \CtSpecInterfStep{\tau}{} \tup{\sigma', 0}\cdot s
		}
		
		\inferrule[Branch]
		{
		p(\sigma(\pc))=\pjz{x}{\lbl}\quad
		\ell_{\mathit{correct}} =
		{
			\begin{cases}
				\lbl & \text{if}\ \sigma(x) = 0\\
				\sigma(\pc) + 1 & \text{otherwise}
			\end{cases}
		}\quad
		\ell_{\mathit{mispred}} \in \{\lbl, \sigma(\pc) + 1\} \setminus \ell_{\mathit{correct}}\quad
		\omega_\mathit{mispred}={	\begin{cases}
				\wInterf & \text{if } \omega=\infty\\
				\omega & \text{otherwise}
			\end{cases}
		}
		}
		{
			\tup{\sigma,\omega+1 }\cdot s \CtSpecInterfStep{\pcObs{\lbl_\mathit{mispred}}}{} \tup{\sigma[\pc \mapsto \lbl_\mathit{mispred}],\omega_\mathit{mispred}} \cdot\tup{\sigma[\pc \mapsto \lbl_\mathit{correct}],\omega} \cdot s
		}
	\end{mathpar}
	\caption{
	Definition of $\CtSpecInterf{\cdot}$ contract.
	Configurations are stacks of $\tup{\sigma, \omega}$, where $\omega \in \Nat\cup \{\infty\}$ is the speculative window denoting how many instructions are left to be executed. (initial \archstate{}s $\sigma$ are treated as $\tup{\sigma,\infty}$). %
	At each computation step, the $\omega$ at the top of the stack is reduced by $1$ (rules \textsc{Step} and \textsc{Branch}).
	When executing a branch instruction (rule \textsc{Branch}), the state $\tup{\sigma[\pc \mapsto \lbl_\mathit{mispred}],\omega_\mathit{mispred}}$ is pushed on top of the stack, thereby allowing the exploration of the mispredicted branch for $\omega_\mathit{mispred}$ steps.
	The correct branch $\tup{\sigma[\pc \mapsto \lbl_\mathit{correct}],\omega}$ is also recorded on the stack; allowing to later roll back speculatively executed statements.
	When the $\omega$ at the top of the stack reaches $0$, we pop it (i.e., we backtrack and discard the changes) and we continue the computation (rule \textsc{Rollback}).
	Speculation barriers trigger a roll back by setting $\omega$ to $0$ (rule \textsc{Barrier}).
	}\label{figure:interface:spec-ct}
	\end{figure*}

\subsubsection{Contract $\ArchSeqInterf{\cdot}$}\label{contract:arch-seq}
This contract exposes the program counter, the location of all loads and stores, and the values of all data loaded from memory on standard, i.e., non-speculative, program paths.
The contract is obtained by modifying the \textsc{Load} %
rule from Figure~\ref{figure:interface:seq-ct} as follows:
\begin{mathpar}
	\inferrule[Load]
	{
	\select{p}{a(\pc)} = \pload{x}{e} \\
	\tup{m, a} \archStep{}{} \tup{m',a'}
	}
	{
	\tup{m, a} \ArchSeqInterfStep{\loadObs{\exprEval{e}{a}= m(\exprEval{e}{a})}}{} \tup{m', a'}
	}

\end{mathpar}

As we assume that register values are zeroed in the initial state, the $\ArchSeqInterf{\cdot}$ trace effectively exposes the contents of registers during execution. 
While this does not seem to guarantee any kind of security,  $\ArchSeqInterf{\cdot}$ {\em does} guarantee the confidentiality of data that is {\em only transiently} loaded, thus effectively preventing speculative disclosure gadgets. In that sense, the contract $\ArchSeqInterf{\cdot}$ is a simple and clean formulation of the idea behind {\em transient noninterference}~\cite{STT2019}, making it comparable to the guarantees offered by other contracts, and providing an actionable interface to software.

\subsubsection{Special contracts}
We informally present a number of contracts that illustrate our framework's expressiveness:
\begin{asparaitem}
	\item $\interfSem{\cdot}_{\interfStyle{\top}}$ is the contract that does not expose any observations and corresponds to a hypothetical side-channel resilient processor that can securely execute every program.
	\item $\CtPcSpecInterf{\cdot}$ exposes program counter and addresses of loads during sequential execution, and only the program counter during speculative execution. That is, it may intuitively be understood as  $\CtSeqInterf{\cdot} + \PcSpecInterf{\cdot}$.
		This contract corresponds, for instance, to processors vulnerable to speculative port-contention attacks like~\cite{smotherspectre2019}.
	\item $\ArchSpecInterf{\cdot}$ exposes the values of data loaded from memory also during speculatively executed instructions. It corresponds to a processor that does not offer any confidentiality guarantees for accessed data.
	\item $\InftyInterf{\cdot}$ exposes all \archstate{} and corresponds to a hypothetical processor that provides no confidentiality guarantees whatsoever.
\end{asparaitem}

\subsection{A lattice of contracts}\label{sec:contracts:contracts:lattice}

Contracts can be compared in terms of the security guarantees that they offer to software. Intuitively, a contract is stronger than another, if it guarantees to leak less information to a microarchitectural adversary.
For instance,  $\interfSem{\cdot}_{\interfStyle{\top}}$, which exposes no observations, is stronger than  $\InftyInterf{\cdot}$, which exposes the entire \archstate{}, written $\interfSem{\cdot}_{\interfStyle{\top}} \sqsupseteq \InftyInterf{\cdot}$.

\begin{definition}[$\interfSem{\cdot}_{\interfStyle{1}} \sqsupseteq \interfSem{\cdot}_{\interfStyle{2}}$]\label{def:contract-strenght}
A contract $\interfSem{\cdot}_{\interfStyle{1}}$ is {\em stronger} than a contract $\interfSem{\cdot}_{\interfStyle{2}}$ if $\interfSem{p}_{\interfStyle{2}}(\sigma) = \interfSem{p}_{\interfStyle{2}}(\sigma')  \Rightarrow  \interfSem{p}_{\interfStyle{1}}(\sigma) = \interfSem{p}_{\interfStyle{1}}(\sigma') $
for all programs $p$ and all initial \archstate{}s $\sigma,\sigma'$.
\end{definition}
Equivalently, $\interfSem{\cdot}_{\interfStyle{1}} \sqsupseteq \interfSem{\cdot}_{\interfStyle{2}}$ holds whenever two \archstate{}s that can be distinguished by $\interfSem{\cdot}_{\interfStyle{1}}$'s traces can also be distinguished by $\interfSem{\cdot}_{\interfStyle{2}}$'s traces.

Note that if $\interfSem{\cdot}_{\interfStyle{1}}$ exposes only a subset of the labels of $\interfSem{\cdot}_{\interfStyle{2}}$, then $\interfSem{\cdot}_{\interfStyle{1}}$ is stronger than $\interfSem{\cdot}_{\interfStyle{2}}$ according to Definition~\ref{def:contract-strenght}. For example, the instructions explored by $\interfStyle{spec}$ are also explored by
$\interfStyle{seq}$, and the observations of $\interfStyle{ct}$ are contained in the observations of $\interfStyle{arch}$. This enables us to arrange all contracts defined in~\S\ref{sec:contracts:contracts:contracts} in the lattice~\cite{lattice}
shown in Figure~\ref{figure:lattice-contracts}.

Finally, as expected, a hardware platform that satisfies a contract $\interfSem{\cdot}_{\interfStyle{1}}$ also satisfies all weaker contracts $\interfSem{\cdot}_{\interfStyle{2}}$.

\begin{restatable}{prop}{contractHniOrdering}
    \label{proposition:contract-hni-ordering}
	If $\hsni{\interfSem{\cdot}_{\interfStyle{1}}}{\CtxMuarchSem{\cdot}{}}$ and $ \interfSem{\cdot}_{\interfStyle{1}} \sqsupseteq \interfSem{\cdot}_{\interfStyle{2}}$, then $\hsni{\interfSem{\cdot}_{\interfStyle{2}}}{\CtxMuarchSem{\cdot}{}}$.
\end{restatable}

This implies that processors with stronger contracts $\interfSem{\cdot}_{\interfStyle{1}}$ are backward-compatible in the sense that they can securely execute any side-channel resilient legacy code that was already secure under weaker contracts $\interfSem{\cdot}_{\interfStyle{2}}$.

Full proofs of Proposition~\ref{proposition:contract-hni-ordering} and of the results in Figure~\ref{figure:lattice-contracts} are available in \techReportAppendix{appendix:contracts}.

\newcommand{\ydistance}{0.88}
\begin{figure}
    \centering

    \begin{tikzpicture}[->,>=stealth',shorten >=1pt,auto, semithick]

	\node[fill=none,draw=none, shape = rectangle, rounded corners, inner sep=0pt, outer sep=0pt, minimum height = 16pt, minimum width = 35pt]  (top) at (+2,0) {$\interfSem{\cdot}_{\interfStyle{\top}}$};

    \node[fill=none,draw=none, shape = rectangle, rounded corners, inner sep=0pt, outer sep=0pt, minimum height = 16pt, minimum width = 55pt]  (seqCt) at (0,0) {$\CtSeqInterf{\cdot}$};

    \node[fill=none,draw=none, shape = rectangle, rounded corners, inner sep=0pt, outer sep=0pt, minimum height = 16pt, minimum width = 55pt]  (specArch) at ($(seqCt) + (-3,-2.5*\ydistance)$) {$\ArchSpecInterf{\cdot}$};

    \node[fill=none,draw=none, shape = rectangle, rounded corners, inner sep=0pt, outer sep=0pt, minimum height = 16pt, minimum width = 55pt]  (seqArch) at ($(seqCt) + (-3,-0.5*\ydistance)$) {$\ArchSeqInterf{\cdot}$};

    \node[fill=none,draw=none, shape = rectangle, rounded corners, inner sep=0pt, outer sep=0pt, minimum height = 16pt, minimum width = 55pt]  (specPcCt) at ($(seqCt) + (0,-\ydistance)$) {$\CtPcSpecInterf{\cdot}$};

    \node[fill=none,draw=none, shape = rectangle, rounded corners, inner sep=0pt, outer sep=0pt, minimum height = 16pt, minimum width = 55pt]  (specCt) at ($(specPcCt) + (0,-\ydistance)$) {$\CtSpecInterf{\cdot}$};

	\node[fill=none,draw=none, shape = rectangle, rounded corners, inner sep=0pt, outer sep=0pt, minimum height = 16pt, minimum width = 35pt]  (bot) at ($(specArch) + (-2,0)$) {$\InftyInterf{\cdot}$};
	
	\path (seqCt) edge[] node[left] {} (top);
    \path (specPcCt) edge[] node[left] {} (seqCt);
    \path (seqArch) edge[] node[left] {} (seqCt);
    \path (specArch) edge[] node[left] {} (seqArch);
    \path (specCt) edge[] node[left] {} (specPcCt);
    \path (specArch) edge[] node[left] {} (specCt);
    \path (bot) edge[] node[left] {} (specArch);

    \end{tikzpicture}
	\caption{Lattice of contracts. An edge from $\interfSem{\cdot}_{\interfStyle{2}}$ to $\interfSem{\cdot}_{\interfStyle{1}}$ means that $\interfSem{\cdot}_{\interfStyle{1}} \sqsupseteq \interfSem{\cdot}_{\interfStyle{2}}$, that is,  $\interfSem{\cdot}_{\interfStyle{1}}$ is stronger than $\interfSem{\cdot}_{\interfStyle{2}}$. The top element $\interfSem{\cdot}_{\interfStyle{\top}}$ of the lattice exposes no observations, while its bottom element $\InftyInterf{\cdot}$ exposes the entire architectural state.}
    \label{figure:lattice-contracts}
    \end{figure}

 \section{Programming against contracts}\label{sec:sw-side}

Contracts are the basis for secure programming.
Here, we consider two scenarios that are both instances of secure programming:
In the first, which we call ``constant-time programming'', the goal is to ensure that a benign program does not leak confidential data to an adversary while computing on this data.
In the second, which we call ``sandboxing'', the goal is to prevent a potentially malicious program from accessing confidential data.
Proofs of this section's results are  in~\techReportAppendix{appendix:contracts}.

\subsection{Secure programming}

We begin by framing secure programming as an information-flow property.
To distinguish confidential data from public data, we rely on a policy $\policy\colon \Val \rightarrow \{\low, \high\}$ that labels memory locations as high ($\high$) or low ($\low$), encoding whether locations store confidential data or not.
Two architectural states~$\sigma, \sigma'$ are \emph{low-equivalent}, written $\conf \lowequiv \conf'$, iff the values of all low memory locations are the same.

\begin{definition}[$\ct{p}{\policy}{\interfSem{\cdot}}$]\label{def:nochannel}
Program $\Prg$ is {\em non-interferent} w.r.t. contract $\interfSem{\cdot}$ and policy $\policy$ %
if for all initial architectural states~$\conf, \conf'$:
$\conf \lowequiv \conf' \Rightarrow \interfSem{p}(\conf) = \interfSem{p}(\conf')$.
\end{definition}

That is, a program is non-interferent w.r.t. a contract and a policy,
if low-equivalent \archstate{}s are indistinguishable under the contract, i.e., no information about high memory locations leaks into the contract's traces.

Similarly to Definition~\ref{def:nochannel}, one can define a notion of non-interference w.r.t. a hardware semantics $\muarchSem{\cdot}$, written $\ct{p}{\policy}{\muarchSem{\cdot}}$,  where information about high memory locations cannot flow into hardware observations.

The following proposition, capturing leakage at the hardware level, follows by composition of Definitions~\ref{def:hni} and~\ref{def:nochannel}:

\begin{restatable}{prop}{secProgEndToEnd}
    \label{thm:sec-prog:end-to-end}
	If $\ct{\Prg}{\policy}{\interfSem{\cdot}}$ and $\hsni{\interfSem{\cdot}}{ \muarchSem{\cdot}}$, then $\ct{\Prg}{\policy}{\muarchSem{\cdot}}$.
\end{restatable}

\subsection{Sandboxing}

The goal of sandboxing is to enable the safe execution of untrusted, potentially malicious code. This is achieved by ensuring that the untrusted code is confined to a set of tightly controlled resources.
Here we focus on one important aspect: preventing code from reading outside of its own subset of the address space. 
To achieve this,  just-in-time compilers enforce access-control policies by inserting checks to ensure that all memory accesses happen within the sandbox's bounds.

We describe sandboxes using policies $\policy$, where memory outside of the sandbox is declared high. 
To account for programs that may escape the sandbox by exploiting speculation across access-control checks, we make the following distinction: 
\begin{asparaitem}
	\item Traditional sandboxing approaches~\cite{nacl2010,wasm2017}
	   check/enforce {\em vanilla sandboxing}:
	A program~$\Prg$ is {\em vanilla-sandboxed} w.r.t.~$\policy$ if $\Prg$ never  accesses high memory locations when executing under the architectural semantics $\archStepCompact$.
	In our framework, being vanilla-sandboxed is equivalent to  $\ct{p}{\policy}{\ArchSeqInterf{\cdot}}$, i.e., being non-interferent w.r.t. $\ArchSeqInterf{\cdot}$.
	This  follows from $\ArchSeqInterf{\cdot}$ exposing the value of accessed high memory locations. %
	\item To faithfully reason about sandboxing on out-of-order and speculative processors, one needs to go beyond vanilla sandboxing and make sure that the program does not leak any information that is outside of its sandbox through a covert channel.
	We say that a program is {\em generally-sandboxed} w.r.t. contract $\interfSem{\cdot}$, if it is vanilla-sandboxed and in addition non-interferent w.r.t $\interfSem{\cdot}$, i.e., $\ct{p}{\policy}{\interfSem{\cdot}}$. 
	General sandboxing together with Proposition~\ref{thm:sec-prog:end-to-end} guarantees that no data outside of the sandbox affects what a microarchitectural adversary (including the sandboxed program $\Prg$ itself, via probing) %
	can observe on any platform satisfying $\interfSem{\cdot}$. %
\end{asparaitem}

Definition~\ref{def:wsni} enables to bridge the gap between {vanilla sandboxing}
and {general sandboxing} for a given program. %
\begin{definition}\label{def:wsni}
	Program~$\Prg$ satisfies {\em weak speculative non-interference} (wSNI) with respect to $\interfSem{\cdot}$
	if for all initial architectural states~$\conf, \conf'$:
		$\ArchSeqInterf{\Prg}(\sigma) = \ArchSeqInterf{\Prg}(\sigma') \Rightarrow \interfSem{\Prg}(\sigma) = \interfSem{\Prg}(\sigma')$.
\end{definition}

Weak speculative non-interference %
 is a variant of {\em speculative non-interference}, the security property checked by Spectector~\cite{spectector2020}.
Proposition~\ref{prop:sandbox} shows how wSNI bridges the gap between vanilla and general sandboxing.

\begin{restatable}{prop}{propSandbox}
    \label{prop:sandbox}
	If program~$\Prg$ is vanilla-sandboxed w.r.t. $\policy$ and wSNI w.r.t. $\interfSem{\cdot}$, then $\Prg$ is generally-sandboxed w.r.t. $\policy$ and $\interfSem{\cdot}$.
\end{restatable}

Hence, to check  whether a program~$\Prg$ is generally-sandboxed w.r.t. $\interfSem{\cdot}$ and $\policy$ one can: (1) check/enforce that $\Prg$ is vanilla-sandboxed w.r.t. $\policy$, and (2) verify whether~$\Prg$ is wSNI.\looseness=-1

\subsection{Constant-time programming}

Constant-time programming is a coding discipline for the implementation of code like cryptographic algorithms that needs to compute over secret data without leaks. %
Code without (1) secret-dependent control flow, (2) secret-dependent memory accesses, and (3) secret-dependent inputs to variable-latency instructions is traditionally understood as ``constant time''.
As discussed before this corresponds to $\CtSeqInterf{\cdot}$, which exposes control flow and memory accesses.%

Again, considering only $\CtSeqInterf{\cdot}$ is insufficient to reason about constant-time on modern processors.
For this, we make the following distinction:
\begin{asparaitem}
	\item Existing constant-time approaches (type systems~\cite{Rodrigues16}, static analyses~\cite{MolnarPSW05,AlmeidaBBDE16}, and techniques for secure compilation~\cite{Cauligi19fact,Barthe19constanttimecompiler}) check/enforce \emph{vanilla-constant-time}.
	In our framework, a program $\Prg$ is \emph{vanilla-constant-time} w.r.t. $\policy$ if $\ct{p}{\policy}{\CtSeqInterf{\cdot}}$, i.e., $\Prg$ non-interferent w.r.t. $\CtSeqInterf{\cdot}$.
	
	\item More generally, a program $\Prg$ is \emph{generally-constant-time} w.r.t. contract $\interfSem{\cdot}$ iff $\ct{p}{\policy}{\interfSem{\cdot}}$, i.e., constant-time coincides with non-interference w.r.t. a contract.
\end{asparaitem}

One possibility for checking general-constant-time is devising dedicated tools~\cite{constanttime2019}.
Alternatively, one can reuse vanilla-constant-time tools~\cite{MolnarPSW05,AlmeidaBBDE16} and then bridge the gap between vanilla and general-constant-time.
To bridge this gap, one can rely on the following generalization of \emph{speculative non-interference} from~\cite{spectector2020}:
\begin{definition}[Speculative non-interference~\cite{spectector2020}]
Program $\Prg$ is {\em speculatively non-interferent} (SNI) w.r.t. policy $\policy$ and contract $\interfSem{\cdot}$ if for all initial architectural states~$\conf, \conf'$:
\vspace{-1mm}
 \begin{equation*}
	\conf \lowequiv \conf' \wedge \CtSeqInterf{\Prg}(\sigma) = \CtSeqInterf{\Prg}(\sigma') \Rightarrow \interfSem{\Prg}(\sigma) = \interfSem{\Prg}(\sigma').
 \end{equation*}
\end{definition}

Proposition~\ref{prop:ct} shows how SNI bridges the gap between vanilla and general constant-time.

\begin{restatable}{prop}{propCt}
    \label{prop:ct}
	If program~$\Prg$ is vanilla-constant-time w.r.t. $\policy$ and SNI w.r.t.  $\policy$ and $\interfSem{\cdot}$, then $\Prg$ is generally-constant-time w.r.t. $\policy$ and $\interfSem{\cdot}$.
\end{restatable}

Thus to check whether a program~$\Prg$ is generally-constant-time w.r.t. $\interfSem{\cdot}$ and $\policy$ one can (1)  check vanilla-constant-time, and (2) verify whether~$\Prg$ is SNI w.r.t. $\interfSem{\cdot}$ and $\policy$.

Observe, however, that not all contracts are useful for general-constant-time.
Remarkably, the $\ArchSeqInterf{\cdot}$ contract, which naturally corresponds to the guarantees provided by state-of-the-art hardware-level countermeasures like STT~\cite{STT2019} and NDA~\cite{nda2019weisse} is inherently inadequate for constant-time programming:
A program that is non-interferent w.r.t. $\ArchSeqInterf{\cdot}$ may not access any secret data.
However, accessing and computing on secret data is the whole point of constant-time programming.

\vspace{-0.5mm}

\subsection{Experiments}

In this section, we illustrate how our framework can be used to support secure programming, for both the sandboxing and constant-time scenarios, w.r.t. the contracts from \S\ref{sec:contracts:contracts}.

\begin{figure}	
\begin{subfigure}[t]{0.238\textwidth}
\begin{lstlisting}[basicstyle=\small,style=Cstyle]
if (y < size_A)
  x = A[y];
  if (x)
    temp &= B[0];
\end{lstlisting}
\caption{Program $P_1'$}\label{figure:v1-ctrl-flow:p1}
\end{subfigure}
\begin{subfigure}[t]{0.238\textwidth}
\begin{lstlisting}[basicstyle=\small,style=Cstyle]
x = A[y];
if (y < size_A)
  if (x)
    temp &= B[0];
\end{lstlisting}
\caption{Program $P_2'$}\label{figure:v1-ctrl-flow:p2}
\end{subfigure}
	\caption{Variants of Spectre v1 that leak information through the control-flow statement in line 3.}\label{figure:v1-ctrl-flow}
	\end{figure}

\paragraph{Tooling}
To automate our analysis we adapted Spectector~\cite{spectector2020},  which can already check SNI for the $\CtSpecInterf{\cdot}$ contract, to support checking SNI and wSNI w.r.t. all the contracts from \S\ref{sec:contracts:contracts}, i.e., $\ArchSeqInterf{\cdot}, \CtSeqInterf{\cdot}, \CtSpecInterf{\cdot}, \CtPcSpecInterf{\cdot}$.

Propositions~\ref{prop:sandbox}--\ref{prop:ct} present a clear path to check (general) sandboxing/constant-time: (1) use existing tools to verify vanilla sandboxing/constant-time, and (2) verify wSNI/SNI using Spectector.

\paragraph{Experimental setup}
We analyze 4 different programs:
\begin{asparaitem}
	\item $P_1$ and $P_2$ are the Spectre~v1 snippet from Figure~\ref{figure:v1-vanilla} and its variant from Figure~\ref{figure:v1-variant}, respectively. 
	\item $P_1'$ and $P_2'$ are modifications of  $P_1$ and $P_2$ that leak information through control-flow statements. %
	The programs are shown in Figure~\ref{figure:v1-ctrl-flow}.
\end{asparaitem}

We compile each program with Clang at \texttt{-O2} optimization level.
We also compile each program with a countermeasure that automatically injects \texttt{lfence} speculation barriers after each branch instruction.\footnote{The countermeasure is enabled with the \texttt{-x86}\texttt{-speculative}\texttt{-load}\texttt{-hardening} \texttt{-x86}\-\texttt{-slh}\-\texttt{-lfence}  flags.}
We denote by $P^f$ the program $P$ with \texttt{lfence}s.

As a result, we have eight small x86 programs %
that we analyze with the help of our enhanced version of Spectector. %

\paragraph{Sandboxing}
We analyze programs $P_1, P_1', P_1^f, P_1'^f$ w.r.t. the policy $\policy$ that declares the contents of \inlineCcode{A[i]} as {\em low} for all~\inlineCcode{i} that are within the array bounds, and as {\em high} otherwise.

Our goal is to determine whether these programs satisfy the general-sandboxing property w.r.t. the contracts in \S\ref{sec:contracts:contracts}.
We remark that all variants of $P_1$ are vanilla-sandboxed w.r.t. $\policy$: they never access out-of-bound locations under the architectural semantics~$\archStepCompact$ thanks to the bounds check. %

Table~\ref{tab:results:sb} summarizes our findings, which we discuss below:
\begin{asparaitem}
	\item For $\interfSem{\cdot}\in\{\ArchSeqInterf{\cdot},\CtSeqInterf{\cdot}\}$, the fact that $\ArchSeqInterf{\cdot}$ and $\CtSeqInterf{\cdot}$ are stronger than $\ArchSeqInterf{\cdot}$ (see \S\ref{sec:contracts:contracts:lattice}) directly implies wSNI w.r.t. these contracts for 
\emph{any} program (denoted by ``\Sat, $\sqsupseteq$'' in the table).
	Therefore, programs $P_1, P_1', P_1^f, P_1'^f$ all satisfy general-sandboxing (see Proposition~\ref{prop:sandbox}) without further analysis.

	\item For $\interfSem{\cdot}\in\{\CtSpecInterf{\cdot},\CtPcSpecInterf{\cdot}\}$, we check whether wSNI holds using Spectector. Table entries ``\Sat{}, wSNI'' denote a successful check, which implies (via Proposition~\ref{prop:sandbox}) that the program is generally-sandboxed w.r.t. $\interfSem{\cdot}$.
	In several cases, denoted by ``\NSat{}'', the wSNI check fails.
	While this is not generally the case, the counterexamples to wSNI show that the respective programs are indeed not sandboxed w.r.t. $\interfSem{\cdot}$.
	\begin{asparaitem}
	\item Program $P_1$ fails the wSNI check w.r.t. $\CtSpecInterf{\cdot}$, due to the speculative secret-dependent load (line 3 in Figure~\ref{figure:v1-variant}), but it satisfies wSNI w.r.t. the stronger contract $\CtPcSpecInterf{\cdot}$  that ensures confidentiality of secret-dependent speculative loads. %
	\item In contrast, program $P_1'$ violates wSNI due to the speculative branch on line 3 in Figure~\ref{figure:v1-ctrl-flow} w.r.t.   $\CtSpecInterf{\cdot}$ and $\CtPcSpecInterf{\cdot}$.
	\item	Finally,  programs $P_1^f$ and $P_1'^f$, where \texttt{lfence}s are inserted after the branch,  satisfy wSNI w.r.t.   $\CtSpecInterf{\cdot}$ and $\CtPcSpecInterf{\cdot}$.
	\end{asparaitem}
	\end{asparaitem}

	\begin{table}
		\caption{Sandboxing analysis w.r.t. different contracts.%
	}\label{tab:results:sb}
		\centering
		\begin{tabular}{lllll}
		\toprule
		  & $\CtSeqInterf{\cdot}$ & $\ArchSeqInterf{\cdot}$  & $\CtSpecInterf{\cdot}$ & $\CtPcSpecInterf{\cdot}$   \\ \midrule
		$P_1$ &  \Sat, $\sqsupseteq$ & \Sat, $\sqsupseteq$   & \NSat       &  \Sat, wSNI           \\
		$P_1^f$ &   \Sat, $\sqsupseteq$ & \Sat, $\sqsupseteq$    & \Sat, wSNI      &   \Sat, wSNI    \\
		$P_1'$ & \Sat, $\sqsupseteq$ & \Sat, $\sqsupseteq$ &  \NSat &  \NSat \\
		$P_1'^f$ &  \Sat, $\sqsupseteq$ & \Sat, $\sqsupseteq$ & \Sat, wSNI & \Sat, wSNI \\
		\bottomrule
		\end{tabular}
	\end{table}

\paragraph{Constant-time}
We analyze programs $P_2, P_2', P_2^f, P_2'^f$ w.r.t. the same policy $\policy$ as before. %

This time, our goal is to determine whether these programs are constant-time w.r.t. the contracts in \S\ref{sec:contracts:contracts}.
We remark that $P_2, P_2', P_2^f, P_2'^f$ are vanilla-constant-time w.r.t. $\policy$, while none of these programs is vanilla-sandboxed w.r.t. $\policy$.

Table~\ref{tab:results:ct} summarizes our findings, which we discuss below:

\begin{asparaitem}
	\item For $\CtSeqInterf{\cdot}$, all programs are constant-time w.r.t. $\CtSeqInterf{\cdot}$ as they are vanilla-constant-time (denoted by ``$\Sat, \sqsupseteq$'' in the table).
	
	\item For $\ArchSeqInterf{\cdot}$, constant-time is violated for all programs, with and without \texttt{lfence}, due to the non-speculative load of a secret into the architectural state. %

	\item For $\interfSem{\cdot}\in\{ \CtSpecInterf{\cdot}, \CtPcSpecInterf{\cdot}\}$, %
	Table entries ``\Sat{}, SNI'' denote a successful check using Spectector, which implies (via Proposition~\ref{prop:ct}) that the program is constant-time w.r.t. $\interfSem{\cdot}$.
	Again, while this is not true in general, the counterexamples to SNI for these particular programs turn out to be proofs that the programs are not constant-time w.r.t. $\interfSem{\cdot}$.

	Program $P_2$ violates SNI w.r.t. $\CtSpecInterf{\cdot}$ but satisfies it under the stronger contract $ \CtPcSpecInterf{\cdot}$ that does not expose the address of the speculative load (line 3 in Figure~\ref{figure:v1-variant}).
	In contrast, $P_2'$ violates SNI against both contracts.
	Finally, the programs with fences ($P_2^f$ and $P_2'^f$) satisfy  SNI w.r.t. $\CtSpecInterf{\cdot}$ and $ \CtPcSpecInterf{\cdot}$.

	\end{asparaitem}

\begin{table}
\caption{
	Constant-time analysis results w.r.t. diff. contracts. %
	}\label{tab:results:ct}
	\centering
	\begin{tabular}{lllll}
		\toprule
	   & $\CtSeqInterf{\cdot}$ & $\ArchSeqInterf{\cdot}$ & $\CtSpecInterf{\cdot}$ & $\CtPcSpecInterf{\cdot}$  \\    \midrule
	$P_2$ &   \Sat, $\sqsupseteq$    &  \NSat      &  \NSat      &  \Sat, SNI          \\
	$P_2^f$ &   \Sat, $\sqsupseteq$   &   \NSat &  \Sat, SNI     &  \Sat, SNI          \\
	$P_2'$ & \Sat, $\sqsupseteq$ & \NSat  & \NSat & \NSat \\
	$P_2'^f$ & \Sat, $\sqsupseteq$ & \NSat  &  \Sat, SNI & \Sat, SNI  \\
		\bottomrule
	\end{tabular}
\end{table}

\section{Modeling microarchitecture and adversaries}\label{sec:hw-side}

This section presents a hardware semantics for \srclang{} programs.
The semantics is based on the semantics from~\cite{constanttime2019,blade} and it models the execution of \srclang{} programs by a simple out-of-order processor with a unified cache for data and instructions and a branch predictor for speculative execution over branch instructions.
The purpose of this semantics is to allow us to model and reason about hardware-level Spectre countermeasures; see \S\ref{sec:countermeasures}.
To this end, it strives to achieve the following design goals: (1) To faithfully capture the key features of speculative and out-of-order execution, while (2) keeping it simple, and (3) supporting large classes of microarchitectural features like caches  and branch predictors.
The latter aspect allows us to focus on hardware-level countermeasures in the context of  arbitrary caching algorithms and branch-prediction strategies.

We start by formalizing hardware configurations (Section~\ref{sec:hw-side:hw-configurations}) that extend \archstate{}s with the state of the microarchitectural components, i.e., cache, reorder buffer, and branch predictor. 
Next, we formalize the semantics of the pipeline steps (Section~\ref{sec:hw-side:hw-semantics}).
This semantics describes how instructions are fetched, executed, and retired under our semantics as well as how hardware configurations are updated during the execution.
We conclude by formalizing the adversary that we consider in our security analysis (Section~\ref{sec:hw-side:adversaries}). %

\subsection{Hardware configurations}\label{sec:hw-side:hw-configurations}
Each \emph{hardware configuration} $\tup{\sigma,\mu}$ consists of its \archstate{}~$\sigma$, recording the memory and register assignments, and of its \uarchstate{} $\mu$, which we formalize next.

The  \uarchstate{} consists of a reorder buffer, which stores the state of in-flight instructions, a cache, a branch predictor, and a scheduler, which orchestrates the pipeline during the computation.
Note that, in our model, cache states track which memory blocks are stored in the cache (i.e., they store metadata) but they do \emph{not} store the data itself.
While we fix the behavior of the reorder buffer in \S\ref{sec:hw-side:hw-configurations:buf}, our semantics is parametric in the models of caches, branch predictors, and the pipeline scheduler; see \S\ref{sec:hw-side:hw-configurations:parameters}. 
Theorem statements in \S\ref{sec:countermeasures} (except where explicitly stated) hold for \emph{all} possible choices of cache, predictor, and scheduler in our model.\looseness=-1

\subsubsection{Reorder buffers}\label{sec:hw-side:hw-configurations:buf}
Reorder buffers store the state of in-flight, i.e., not yet retired, instructions.
Initially instructions are \emph{unresolved}, e.g., a load $\pload{x}{y+z}$ that has not yet been performed or an assignment $\passign{z}{2 + k}$ whose right-hand side has not yet been evaluated.
Executing an unresolved instruction can transform it into a \emph{resolved} instruction, where all expressions are replaced with their values.
Additionally, to model speculative control flow, reorder buffer entries may be \emph{tagged} with the address of a branch instruction~$\ell$.
We write $\tagged{\passign{\pc}{v}}{\ell}$, whenever the assignment of $v$ to the $\pc$ is the result of a call to the branch predictor when fetching the branch at address $\ell$ (only assignments to the program counter register $\pc$ are tagged since branch prediction is the sole source of speculation in our semantics).
Instructions are \emph{untagged}, written $\tagged{i}{\epsilon}$, if they are not the result of a prediction.

We model reorder buffers as sequences of \textit{commands} of length at most $\wMuarch$ denoting the buffer's maximal length: %
\begin{center}
\begin{tabular}{llcl}
\textit{(Tags)}				& $T$ & $:=$ & $\notags \mid \lbl$\\
\textit{(Commands)}		&  $c$		& $:=$ & $\tagged{i}{T}$ \\
\textit{(Reorder buffers)}		&  $\buf$		& $:=$ & $\varepsilon \mid c \concat \buf$
\end{tabular}
\end{center}

A reorder buffer captures the state of execution of in-flight instructions.
Consider the buffer  $\buf := \tagged{\passign{k}{25}}{\notags} \concat \tagged{\pload{x}{y + z}}{\notags} \concat \tagged{\passign{z}{2 + k}}{\notags}$.
It records that there are three in-flight instructions: one of them ($\tagged{\passign{k}{25}}{\notags}$) has been resolved and is ready to be retired, while the remaining two are still unresolved.
Executing the third command  would result in the new buffer $\buf' := \tagged{\passign{k}{25}}{\notags} \concat \tagged{\pload{x}{y + z}}{\notags} \concat \tagged{\passign{z}{27}}{\notags}$.

Given a buffer $\buf$,  its data-independent projection  $\BufProject{\buf}$ is obtained by replacing all resolved (respectively unresolved) expressions in instructions with $\resolved$ (respectively $\unresolved$).
For instance, the data-independent projection of the buffer $\buf$ from above is $\tagged{\passign{k}{\resolved}}{\notags} \concat \tagged{\pload{x}{\unresolved}}{\notags} \concat \tagged{\passign{z}{\unresolved}}{\notags}$.

\subsubsection{Caches, Branch predictors, and Schedulers}\label{sec:hw-side:hw-configurations:parameters}
Rather than providing a fixed model for caches, branch predictors and schedulers, our semantics is parametric in such components.
To this end, we only fix the interface to these components, which is given in Table~\ref{table:muarch:components}, constraining how the semantics may interact with these components.
Each of these components is defined by a set of states, an initial state, and uninterpreted functions modeling their relevant behavior:
\begin{asparaitem}
\item Caches are equipped with a function $\CacheAccess(\CacheState, \ell) \in  \{\CacheHit,\CacheMiss\}$ that captures whether accessing memory address~$\ell$ in cache state~$\CacheState$ results in a cache hit ($\CacheHit$) or miss ($\CacheMiss$), and a function $\CacheUpdate(\CacheState, \ell) = \CacheState'$ that updates the state of the cache based on the access to address~$\ell$.
We stress that cache states $\CacheState$ track only the memory addresses of the blocks in the cache, \emph{not} the blocks themselves.
\item Branch predictors are equipped with a function $\BpUpdate(\BpState,\ell,b)$ that updates the state $\BpState$ of the branch predictor by recording that the branch at program counter $\ell$ has been resolved to value $b$, and $\BpPredict(\BpState,\ell)$ that, given a predictor state $\BpState$, predicts the outcome of the branch at address $\ell$. 
\item Schedulers determine which pipeline stages to activate next. Following~\cite{constanttime2019,blade}, we model this choice using three types of directives: 
\begin{inparaenum}[(a)]
\item $\fetch{b}$ is used to fetch and decode the next instruction pointed by the program counter register~$\pc$,
\item $\execute{i}$ is used to execute the $i$-th command in the reorder buffer $\buf$, and
\item $\retire$ is used to retire (i.e., apply the changes to the memory and register file) the first command in the buffer.
\end{inparaenum} 
Schedulers are equipped with an $\SchedNext( \SchedState)$ function that produces the next directive given the scheduler's state $\SchedState$, and an $\SchedUpdate(\SchedState,\buf)$ function that updates the scheduler's state based on the state of the reorder buffer.

\end{asparaitem}

\subsubsection{Microarchitectural states}\label{sec:hw-side:hw-configurations:configurations}
A \emph{\uarchstate} $\mu$ is a 4-tuple $\tup{\buf,  \CacheState, \BpState, \SchedState}$ where $\buf$ is a reorder buffer, $\CacheState$ is the state of the unified cache (for data and instructions), $\BpState$ is the branch predictor state, and $\SchedState$ is the scheduler state.

A \uarchstate{} $\mu$ is \emph{initial} if $\buf = \emptysequence$ and the microarchitectural components are in their initial states.
Similarly, $\mu$ is \emph{final} if $\buf = \emptysequence$. %
Hence, a hardware configuration $\tup{\sigma, \mu}$ is initial (respectively final) if $
\sigma$ and $\mu$ are so.

For simplicity, we write $\tup{m,a, \buf,  \CacheState, \BpState, \SchedState }$ to represent the hardware configuration  $\tup{ \tup{m,a}, \tup{\buf, \CacheState, \BpState, \SchedState }}$.

\begin{table*}
    \caption{Signatures of the microarchitectural components}\label{table:muarch:components}
    \centering \small
    \begin{tabular}{l l l l l}
    	\toprule
        \textbf{Component} & \textbf{States} & \textbf{Initial state}\hspace{-1.2mm} & \textbf{Functions} \\ \midrule
       {\emph{Cache}} & {$\CacheStates$}\hspace{-2mm} & {$\CacheState_{0}$} & $\CacheAccess: \CacheStates \times \Val \to \{\CacheHit,\CacheMiss\}$\hspace{-2mm} & $\CacheUpdate: \CacheStates  \times \Val \to \CacheStates$ \\
        
        {\emph{Branch predictor}} & {$\BpStates$} & {$\BpState_0$} &  $\BpPredict: \BpStates \times \Val \to \Val$ & $\BpUpdate: \BpStates \times \Val \times \Val \to \BpStates$\\
        
        {\emph{Pipeline scheduler}}\hspace{-1.2mm} & {$\SchedStates$} &  {$\SchedState_0$} & $\SchedNext: \SchedStates \to \Directives$ & $\SchedUpdate: \SchedStates \times \ReorderBuffers \to \SchedStates$ \\
    \bottomrule
    \end{tabular}    
\end{table*}

\subsection{Hardware semantics}\label{sec:hw-side:hw-semantics}

We formalize the hardware semantics of a \srclang{} program~$p$  using a binary relation $\muarchStepCompact \subseteq \ProcConfs \times \ProcConfs$ that maps hardware states to their successors:
\begin{mathpar}
\inferrule[Step]
{
	\tup{m,a, \buf,  \CacheState, \BpState } \muarchStep{d}{} \tup{m',a', \buf', \CacheState', \BpState'  }	\\
	d = \SchedNext(\SchedState)\\
	\SchedState' = \SchedUpdate(\SchedState,\BufProject{\buf'})
}
{
	\tup{m,a, \buf,  \CacheState, \BpState, \SchedState } \muarchStepCompact \tup{m',a', \buf',  \CacheState', \BpState', \SchedState' }
}
\end{mathpar}
The rule captures one execution step at the microarchitectural level.
The scheduler is queried to determine the directive $d = \SchedNext(\SchedState)$ indicating which pipeline step to execute.
Next, the \uarchstate{} is updated by performing one step of the auxiliary relation $\tup{m,a, \buf,  \CacheState, \BpState } \muarchStep{d}{} \tup{m',a', \buf', \CacheState', \BpState' }$, which depends on the directive $d$ and is formalized below.
Finally, the scheduler state is updated based on the data-independent projection of the reorder buffer, i.e., $\SchedState' = \SchedUpdate(\SchedState,\BufProject{\buf'})$.
This formalizes the crucial assumption that the scheduler's decisions may depend upon the dependencies between the instructions in the reorder buffer, but not on the values computed thus far.

For each directive, i.e., $\fetch{b}, \execute{i}$, and $\retire$, we sketch below the rules that govern the definition of the auxiliary relations $\muarchStep{\fetch{b}}{}$, $\muarchStep{\execute{i}}{}$, and $\muarchStep{\retire}{}$.
We provide a full formalization of the rules in~\techReportAppendix{appendix:hardware-semantics}.

\subsubsection{Fetch}
Instructions are fetched in-order. 
Here we present selected rules modeling instruction fetch:
    \begin{mathpar}
    \inferrule[Fetch-Branch-Hit]
    {
    a' = \apply{\buf}{a} \\
    |\buf| < \wMuarch \\
    a'(\pc) \neq \bot \\
    p(a'(\pc)) = \pjz{x}{\lbl} \\
    \lbl' = \BpPredict(\BpState, a'(\pc))\\
    \CacheAccess(\CacheState, a'(\pc)) = \CacheHit\\
    \CacheUpdate(\CacheState, a'(\pc)) = \CacheState{}'
    }
    {
        \tup{m,a,\buf,  \CacheState, \BpState} \muarchStep{\fetch{b}}{} \tup{m,a,\buf \concat \tagged{\passign{\pc}{\lbl'}}{a'(\pc)},  \CacheState',\BpState}	
    }
    \end{mathpar}
    \begin{mathpar}
    \inferrule[Fetch-Miss]
    {
    |\buf| < \wMuarch \\
    a' = \apply{\buf}{a} \\
    a'(\pc) \neq \bot \\
    \CacheAccess(\CacheState,  a'(\pc)) = \CacheMiss\\
   \CacheUpdate(\CacheState,  a'(\pc)) =  \CacheState{}'
    }
    {
        \tup{m,a,\buf, \CacheState, \BpState} \muarchStep{\fetch{b}}{}  \tup{m,a,\buf, \CacheState', \BpState}
    }
    \end{mathpar}
In these rules, and in those  described later, $\apply{\buf}{a}$ denotes the assignment $a'$ obtained by updating $a$ with the changes performed by the commands in $\buf$.
Concretely, $\apply{\buf}{a}$ iteratively applies the pending changes for all commands in $\buf$ as follows: 
\begin{inparaenum}[(a)]
\item Assignments $\tagged{\passign{x}{e}}{T}$ set the value of $a'(x)$ to $e$ if the assignment is resolved (i.e., $e \in \Val$) and to $\bot$ otherwise (denoting unresolved values).
\item Load operations  $\tagged{\pload{x}{e}}{T}$ set the value of $a'(x)$ to $\bot$ (since the load operation has not been performed yet).
\item Whenever $\buf$ contains a speculation barrier $\tagged{\pbarrier}{T}$,  $\apply{\buf}{a} = \lambda x \in \Var.\ \bot$. 
\item Other instructions are ignored.
\end{inparaenum}

The rule \textsc{Fetch-Branch-Hit} models the fetch of a branch instruction $\pjz{x}{\lbl}$.
Whenever the reorder buffer $\buf$ is not full ($|\buf| < \wMuarch$), $\pc$ is defined ($a'(\pc) \neq \bot$), and the instruction is in the cache ($\CacheAccess(\CacheState, a'(\pc)) = \CacheHit$), the branch predictor is queried to obtain the next program counter $\lbl' = \BpPredict(\BpState, a'(\pc))$.
Next, the cache and the reorder buffer states are updated.
The latter is updated by appending the command $\tagged{\passign{\pc}{\lbl'}}{a'(\pc)}$, which records the change to the program counter as well as the label of the branch instruction whose target was predicted.
The semantics also contains rules for fetching jumps $\pjmp{e}$, which append the command $\tagged{\passign{\pc}{e}}{\notags}$ to the buffer, and other instructions $i$, which append the commands  $\tagged{i}{\notags} \concat \tagged{\passign{\pc}{a'(\pc)+1}}{\notags}$ to the buffer.

The rule \textsc{Fetch-Miss} models a cache miss when loading the next instruction.
In this case, the cache is updated while the reorder buffer is not modified.
A subsequent $\fetch{}$ triggered by the scheduler would result in a cache hit and a corresponding change to the reorder buffer.

\subsubsection{Execute}
Commands in-flight are executed out-of-order, where the $\execute{i}$ directive triggers the execution of the $i$-th command in the buffer.
Selected rules are given in Figure~\ref{figure:uops:execute}.
 
\begin{figure*}[h]
    \begin{mathpar}
    \inferrule[Execute-Load-Hit]
    {
        |\buf| = i-1 \\
        a' = \apply{\buf}{a} \\
        \pbarrier \not\in \buf \\
        \pstore{x'}{e'} \not\in \buf \\
        x \neq \pc\\
        \exprEval{e}{a'} \neq \bot\\
        \CacheAccess(\CacheState, \exprEval{e}{a'}) = \CacheHit\\
         \CacheUpdate(\CacheState, \exprEval{e}{a'}) = \CacheState'
    }
    {
        \tup{m,a,\buf \concat  \tagged{\pload{x}{e}}{T} \concat \buf',  \CacheState, \BpState} \muarchStep{\execute{i}}{} \tup{m,a,\buf \concat \tagged{\passign{x}{ m(\exprEval{e}{a'})  }}{T} \concat \buf',  \CacheState', \BpState}	
    }
    
    \inferrule[Execute-Branch-Rollback]
    {
        |\buf| = i-1 \\
        a' = \apply{\buf}{a} \\
        \pbarrier \not\in \buf \\
        \ell_0 \neq \notags \\
        p(\lbl_0) = \pjz{x}{\lbl''} \\
        (a'(x) = 0 \wedge \lbl \neq \lbl'') \vee (a'(x) \in \Val \setminus \{0,\bot\} \wedge \lbl \neq \ell_0+1)\\ 
        \lbl' \in \{ \lbl'',\lbl_0 +1 \} \setminus \{ \ell \} \\
        \BpState' = \BpUpdate(\BpState, \ell_0, \ell')  \\
    }
    {
        \tup{m,a, \buf \concat \tagged{\passign{\pc}{\lbl}}{\lbl_0} \concat \buf', \CacheState, \BpState} \muarchStep{\execute{i}}{} \tup{m,a,\buf \concat \tagged{\passign{\pc}{ \lbl'  }}{\varepsilon} , \CacheState, \BpState'}	
    }
    \end{mathpar}
    \caption{Selected rules for $\execute{i}$}	\label{figure:uops:execute}
    \end{figure*}

The rule \textsc{Execute-Load-Hit} models the successful execution of a load ($\tagged{\pload{x}{e}}{T}$) that results in a cache hit.
In the rule, $\exprEval{e}{a'}$ denotes the result of evaluating  $e$ in the context of the assignment $a'$ obtained by applying to $a$ all earlier in-flight commands in  $\buf$.
Whenever the address is resolved, i.e., $\exprEval{e}{a'} \neq \bot$, and accessing the address results in a cache hit ($\CacheAccess(\CacheState, \exprEval{e}{a'}) = \CacheHit$), the reorder buffer is updated by replacing $\tagged{\pload{x}{e}}{T}$ with $\tagged{\passign{x}{ m(\exprEval{e}{a'})  }}{T}$, thereby recording that the load operation has been executed and that the value of $x$ is now $m(\exprEval{e}{a'})$.
The cache state is also updated to account for the memory access to $\exprEval{e}{a'}$.

In contrast, the \textsc{Execute-Branch-Rollback} rule models the resolution of a mis-speculated branch instruction %
that results in rolling back the speculatively executed instructions by dropping their entries from the reorder buffer.
Whenever the predicted value $\lbl$ disagrees with the outcome $\ell'$ of the instruction $\pjz{x}{\ell''}$ at address $\ell_0$, the buffer is updated by (1) recording the new value of $\pc$ (by replacing $\tagged{\passign{\pc}{\lbl}}{\lbl_0}$ with $\tagged{\passign{\pc}{\ell'}}{\notags}$), and (2) squashing all later buffer entries (by discarding the buffer suffix $\buf'$).
Moreover, the branch predictor's state is updated by recording that the branch at address $\ell_0$ has been resolved to $\ell'$.

\subsubsection{Retire}
Instructions are retired in-order.
This is done by retiring only commands $\tagged{i}{T}$ at the head of the reorder buffer where the instruction~$i$ has been resolved and the tag $T$ is $\epsilon$ indicating that there are no  unresolved predictions.
Selected rules for the $\retire$ directive are given below:
    \begin{mathpar}
    \inferrule[Retire-Assignment]
    {
        \buf = \tagged{\passign{x}{v}}{\notags} \concat \buf'\\
        v \in \Val
    }
    {
        \tup{m,a,\buf,\CacheState,\BpState} \muarchStep{\retire}{} \tup{m, a[x \mapsto v], \buf', \CacheState,\BpState}	
    }
    \end{mathpar}
    \begin{mathpar}
    \inferrule[Retire-Store]
    {
        \buf = \tagged{\pstore{v}{n}}{\notags} \concat \buf'\\
    v,n \in \Val\\
   \CacheUpdate( \CacheState,n) = \CacheState' 
    }
    {
        \tup{m,a,\buf,\CacheState,\BpState} \muarchStep{\retire}{  }
         \tup{m[n \mapsto v], a, \buf' ,  \CacheState', \BpState}	
    }
    \end{mathpar}

The rule \textsc{Retire-Assignment} models the retirement of a command $\tagged{\passign{x}{v}}{\notags}$, where the assignment $a$ is permanently updated by recording that $x$'s value is now $v$.
In contrast, \textsc{Retire-Store} models the retirement of store commands $\tagged{\pstore{v}{n}}{\notags}$.
In this case, the memory $m$ is permanently updated by writing the value $v$ to address $n$ and the cache state is updated. %
Finally, we have rules \textsc{Retire-Skip} and \textsc{Retire-Barrier} modeling the retirement of $\pskip$ and $\pbarrier$ instructions, which are removed from the  reorder buffer without modifying the \archstate{}.

\subsection{Formalizing the adversary model}\label{sec:hw-side:adversaries}

We conclude by formalizing the adversary model that we use in the security analysis in Section~\ref{sec:countermeasures}.

In our analysis, we consider an adversary $\adversary$ that can observe almost the entire microarchitectural state.
Specifically, it can observe (1) the data-independent projection of the reorder buffer (i.e., which instructions are in-flight, but not to what values they are resolved), (2) the state of cache (which stores only the addresses of the blocks in the cache, not the blocks themselves), branch predictor, and scheduler.
We formalize this as $\adversary(\tup{m,a,\buf, \CacheState,\BpState, \SchedState}) = \tup{\BufProject{\buf}, \CacheState,\BpState, \SchedState}$.

\section{Mechanisms for secure speculation}\label{sec:countermeasures}

In this section, we show how several recent proposals for hardware-level secure speculation can be cast within our framework and we study their security.

We analyze three  countermeasures: (1) disabling speculation ($\muarchStyle{seq}$ in \S\ref{sec:countermeasures:sequential}), (2) delaying \emph{all} speculative loads ($\muarchStyle{loadDelay}$ in \S\ref{sec:countermeasures:load-delay}), and (3) employing hardware-level taint tracking and selectively delaying tainted instructions ($\muarchStyle{tt}$ in \S\ref{sec:countermeasures:taint-tracking}).
For each countermeasure $\muarchStyle{ctx}$, we formalize its semantics using a  relation $\CtxMuarchStep{}{}{ctx}$ obtained by modifying the hardware semantics from \S\ref{sec:hw-side} (which induces the corresponding trace semantics $\CtxMuarchSem{\cdot}{ctx}$ in the usual way).
Additionally, we characterize their security guarantees by showing which of the contracts from \S\ref{sec:contracts:contracts} they satisfy;
see Figure~\ref{figure:guarantees} for a summary of the results.
An overview of the proofs is available in Appendix~\ref{appendix:proof:overview}, whereas detailed proofs of all  results are given in \techReportAppendices{appendix:proofs:general}{appendix:proofs:taint-tracking}.

Unless otherwise specified, all theorems hold for any instantiation of  cache, branch predictor, and scheduler.

Before analyzing the countermeasures, we observe that \emph{all} possible instances of the hardware semantics satisfy the  $\CtSpecInterf{\cdot}$ contract, as stated in Theorem~\ref{theorem:hni:all}.

\begin{restatable}{thm}{allCtSpec}
    \label{theorem:hni:all}
    $\hsni{\CtSpecInterf{\cdot}}{\muarchSem{\cdot}}$.
\end{restatable}

From this,
it immediately follows that \emph{all} countermeasures presented below satisfy the $\CtSpecInterf{\cdot}$ contract as well.

\subsection{$\muarchStyle{seq}$: Disabling speculation}\label{sec:countermeasures:sequential}
A first, drastic countermeasure against speculative execution attacks is disabling speculative and out-of-order execution.
To model this, we instantiate the hardware semantics by providing a sequential scheduler that produces directives in a $\fetch{} - \execute{1} - \retire{}$ order.
The sequential scheduler, formalized in \techReportAppendix{appendix:seq-scheduler}, works as follows: 
\begin{asparaitem}
\item Whenever the reorder buffer is empty, the scheduler selects the $\fetch{}$ directive that adds entries to the buffer.
\item If the first entry in the buffer is not resolved, the scheduler selects the $\execute{1}$ directive. Thus, the instruction is executed and, potentially, resolved.
\item If the first entry in the buffer is resolved, the scheduler selects the $\retire$ directive. Therefore, the instruction is retired and its changes are written into the architectural state.
\end{asparaitem}
That is, the sequential scheduler ensures that instructions are executed in an in-order, non-speculative fashion.

As expected, instantiating the hardware semantics with the sequential scheduler (denoted with $\muarchStyle{seq}$) results in strong security guarantees.
As stated in Theorem~\ref{theorem:hni:sequential}, $\muarchStyle{seq}$ implements the $\CtSeqInterf{\cdot}$ interface that exposes only the program counter and the location of memory accesses under sequential execution. %

\begin{restatable}{thm}{sequentialGuarantees}
    \label{theorem:hni:sequential}
    $\hsni{\CtSeqInterf{\cdot}}{\SeqProcMuarchSem{\cdot}}$.
\end{restatable}

\subsection{$\muarchStyle{loadDelay}$: Delaying all speculative loads}\label{sec:countermeasures:load-delay}

Sakalis et al.~\cite{specshadow2019} propose a family of countermeasures that delay  memory loads to avoid leakage.
In the following, we analyze the \textit{eager delay of (speculative) loads} countermeasure.
This countermeasure consists in delaying loads until all sources of mis-speculation have been resolved.
We remark that the hardware semantics of Section~\ref{sec:hw-side} supports speculation only over branch instructions.
Therefore, we model the $\muarchStyle{loadDelay}$ countermeasure by preventing loads whenever there are preceding, unresolved branch instructions in the reorder buffer. 
Using the terminology of~\cite{specshadow2019}, loads are delayed as long as they are under a so-called \emph{control-shadow}. %

We formalize the $\muarchStyle{loadDelay}$ countermeasure by modifying the \textsc{Step} rule of the hardware semantics as follows (changes are highlighted in blue):
\begin{mathpar}
    \inferrule[Step-Others]
    {
        \tup{m,a,\buf, \CacheState,\BpState} \muarchStep{d}{} \tup{m',a',\buf', \CacheState',\BpState'}	\\
        d = \SchedNext(\SchedState)\\
        \SchedState' = \SchedUpdate(\SchedState,\BufProject{\buf'})\\\\
        \highlightBox{ d \in \{\fetch{}, \retire\} \vee (d = \execute{i} \wedge \elt{\buf}{i} \neq \pload{x}{e}) }
    }
    {
        \tup{m,a,\buf, \CacheState,\BpState, \SchedState} \LoadDelayMuarchStep{}{} \tup{m',a',\buf', \CacheState',\BpState', \SchedState'}
    }
\end{mathpar}
\begin{mathpar}
    \inferrule[Step-Eager-Delay]
    {
        \tup{m,a,\buf, \CacheState,\BpState} \muarchStep{d}{} \tup{m',a',\buf', \CacheState',\BpState'}	\\
        d = \SchedNext(\SchedState)\\
        \SchedState' = \SchedUpdate(\SchedState,\BufProject{\buf'})\\
        \highlightBox{d = \execute{i}}\\
        \highlightBox{\elt{\buf}{i} = \pload{x}{e}}\\
        \highlightBox{\forall \tagged{\passign{\pc}{\lbl}}{T} \in \buf[0..i-1].\  T = \notags}
    }
    {
        \tup{m,a,\buf, \CacheState,\BpState, \SchedState } \LoadDelayMuarchStep{}{} \tup{m',a',\buf', \CacheState',\BpState', \SchedState'}
    }
\end{mathpar}
Fetching, retiring, and executing all instructions that are not loads work as before (see \textsc{Step-Others} rule).
However, load instructions are executed only if all prior branch instructions are resolved (see \textsc{Step-Naive-Delay} rule).
This is captured by requiring that all branch instructions in the buffer prefix have tag $\notags$, i.e., $ \forall \tagged{\passign{\pc}{\lbl}}{T} \in \buf[0..i-1].\ T = \notags$.

Thus, loads are delayed until they are guaranteed to be executed, while other instructions may be freely executed speculatively and out-of-order.
Hence, no data memory accesses are performed on mis-speculated paths.
However, maybe surprisingly, parts of the architectural state can still be leaked on mis-speculated paths as nested conditional branches may modify the instruction cache and the branch predictor state.

As a consequence, $\muarchStyle{loadDelay}$ violates the $\CtSeqInterf{\cdot}$~contract capturing the standard constant-time requirements. %

\begin{example}\label{example:naive-delay}
This program illustrates that $\hsniViolation{\CtSeqInterf{\cdot}}{\LoadDelayMuarchSem{\cdot}}$: %
{%
\begin{lstlisting}[style=Cstyle]
x = A[10]
y = not (A[20] | 1)
if (y) //branch always unsatisfied 
  if (x) //only executed speculatively
    skip
\end{lstlisting}
}
Consider two configurations $\sigma$ and $\sigma'$ such that $\sigma(\inlineCcode{A+10})  = 0$ and $\sigma'(\inlineCcode{A+10}) = 1$. 
Then, $\CtSeqInterf{p}(\sigma) = \CtSeqInterf{p}(\sigma') = \loadObs{\inlineCcode{A+10}} \concat \loadObs{\inlineCcode{A+20}} \concat \pcObs{\bot}$.
However, the hardware can leak information through, e.g., the instruction cache if the branch at line 3 is speculatively taken. %
Then, the result of branch at line $4$,
 which determines whether or not $\pskip{}$ at line~$5$ is fetched, leaks whether $\inlineCcode{A[10]}$ (stored in $\inlineCcode{x}$) is $0$ or not, thereby distinguishing $\sigma$ and $\sigma'$.
\end{example}

To capture the guarantees offered by the eager-delay countermeasure, we can use the $\CtPcSpecInterf{\cdot}$ contract, which may intuitively be understood as $\CtSeqInterf{\cdot} + \PcSpecInterf{\cdot}$, i.e., control-flow and memory accesses are leaked under sequential execution, and in addition, the program counter is leaked during speculative execution.
This new contract is satisfied by the countermeasure,
 leading to  Theorem~\ref{theorem:hni:load-delayone}.

 \begin{restatable}{thm}{loadDelayOne}
    \label{theorem:hni:load-delayone}
    $\hsni{\CtPcSpecInterf{\cdot}}{\LoadDelayMuarchSem{\cdot}}$.
\end{restatable}

As the control flow during speculative execution may only depend upon data previously loaded non-speculatively, the security of the countermeasure can also be captured by $\ArchSeqInterf{\cdot}$. %

\begin{restatable}{thm}{loadDelayTwo}
    \label{theorem:hni:load-delaytwo}
    $\hsni{\ArchSeqInterf{\cdot}}{\LoadDelayMuarchSem{\cdot}}$.
\end{restatable}

\subsection{$\muarchStyle{tt}$: Taint tracking of speculative values}\label{sec:countermeasures:taint-tracking}

Recent work~\cite{nda2019weisse,STT2019} propose to track transient computations and to selectively delay instructions involving tainted information.
While these proposals slightly differ in how instructions are labelled and on the effects of different labels, they share the same building blocks and provide similar guarantees.

For this reason, we start by presenting an overview of the Speculative Taint Tracking (STT)~\cite{STT2019} and  Non-speculative Data Access (NDA)~\cite{nda2019weisse} countermeasures.
Next, we introduce a general extension to the hardware semantics from Section~\ref{sec:hw-side} for supporting taint tracking schemes.
We continue by formalizing a countermeasure inspired by STT and we discuss its security guarantees, and  we conclude by discussing  NDA.

\subsubsection{Overview}\label{sec:countermeasures:taint-tracking:overview}
STT~\cite{STT2019} and NDA~\cite{nda2019weisse} are two recent taint tracking proposals for secure speculation.
These countermeasures extend a processor with hardware-level taint tracking to track whether data has been retrieved by a speculatively executed instruction.
The taint tracking mechanism propagates taint through the computation and whenever operations are no longer transient, the taint is removed.
Finally, both NDA and STT selectively delay tainted operations to avoid leaks.

The main difference between the two approaches is that while STT delays the \emph{execution} of tainted {transmit instructions} (that is, instructions like loads that might leak information), NDA adopts a more conservative approach that delays the \emph{propagation} of data from tainted instructions.

\subsubsection{Supporting taint tracking}\label{sec:countermeasures:taint-tracking:general}

To support taint tracking, we label entries in the reorder buffer with sets of reorder buffer's indexes.
A labeled command is of the form $\labelled{\tagged{i}{T}}{L}$ where $\tagged{i}{T}$ is a reorder buffer entry and $L \subset \Nat$ is a label, i.e., the set of indexes of the entries $\tagged{i}{T}$ depends on.

Existing proposals differ in (1) how labels are assigned and propagated, and (2) how labels affect the processor's execution.
To accommodate different variants for (1) and (2), we formalize these aspects using two functions:
\begin{asparaitem}
\item The \emph{labeling function} $\labelNda{\buf_{ul}}{\buf}{d}$ computes the new labels associated with the (unlabeled) buffer $\buf_{ul}$ given the old labeled buffer $\buf$ and the directive $d$ determining the activated pipeline step.
This function models how the tracking works, i.e., how labels are assigned to new instructions and how they are propagated.

\item The \emph{unlabeling function} $\unlabelNda{\buf}{d}$ produces an unlabeled buffer $\buf_{ul}$ starting from a labeled buffer $\buf$ and a directive $d$.
This function models how labels affect the processor's semantics in terms of changes to the reorder buffer (and these changes might depend on the executed pipeline step modeled by $d$).
\end{asparaitem}
We describe later how these functions can be instantiated to model STT and NDA.

We formalize the  $\muarchStyle{tt}$ countermeasure by modifying the \textsc{Step} rule as follows (changes are highlighted in blue):
\begin{mathpar}
\inferrule[Step]
{
	d = \SchedNext(\SchedState)\\
	\highlightBox{\buf_{ul} = \unlabelNda{\buf}{d}}\\
	\tup{m,a, \highlightBox{\buf_{ul}}, \CacheState,\BpState} \muarchStep{d}{} \tup{m',a',\highlightBox{\buf_{ul}'}, \CacheState',\BpState'}	\\
    	\highlightBox{\buf' = \labelNda{\buf_{ul}'}{\buf}{d}}\\
    \SchedState' = \SchedUpdate(\SchedState,{\BufProject{{\buf'}}})
}
{
	\tup{m,a,\buf, \CacheState,\BpState, \SchedState} \TtMuarchStep{}{} \tup{m',a',\buf', \CacheState',\BpState', \SchedState'}
}
\end{mathpar}
The rule differs from the standard \textsc{Step} rule in three ways:
\begin{asparaitem}
\item Entries in the reorder buffer are labelled.
\item Before activating a step in the pipeline, i.e., before applying one step of $\muarchStep{d}{}$, we use the unlabeling function to derive an unlabeled buffer $\buf_{ul} = \unlabelNda{\buf}{d}$ representing how labels affect the reorder buffer entries.
\item The buffer produced by the application of $\muarchStep{d}{}$ is labeled by invoking the labeling function $\buf' = \labelNda{\buf_{ul}'}{\buf}{d}$.
Therefore, the labels in $\buf'$ are updated to track the information flows through the computation.
\end{asparaitem}

\subsubsection{Speculative taint tracking}\label{sec:countermeasures:taint-tracking:stt}

Here we present how to model a countermeasure inspired by STT~\cite{STT2019}. %
As mentioned above, STT tracks whether data depends on speculatively accessed data and delays the execution of transient transmit instructions.
These features are reflected in our model:
\begin{asparaitem}
\item In \srclang{}, there are three kinds of \emph{transmit instructions}: loads $\pload{x}{e}$, stores $\pstore{x}{e}$, and assignments to the program counter $\passign{\pc}{e}$.
We write $\transmitGadget{\tagged{i}{T}}$ whenever the instruction $i$ is a transmit instruction. 

\item The labeling function, formalized in \techReportAppendix{appendix:stt-labeling}, specifies how newly fetched instructions are labeled as well as how labels are updated during computation, and it  works as follows:\looseness=-1
\begin{asparaitem}
\item Newly fetched $\pload{x}{e}$ instructions are labelled with the indexes of the unresolved branch instruction in the buffer.
	That is, non-transient loads are labelled with $\emptyset$, whereas potentially transient loads have a non-empty label.
 	In contrast, newly fetched assignments $\passign{x}{e}$ are labelled with the union of the labels associated with the registers occurring in $e$.
 	That is, assignments that depend only on non-transient values are labelled with $\emptyset$, whereas those depending on potentially transient values have a non-empty label.
	All other newly fetched instructions are labelled with an empty label.
\item When we retire an instruction, all indexes in labels are decremented by $1$ and indexes reaching $0$ are removed.
This ensures that indexes are still consistent with the, now shorter, reorder buffer.
\item When we execute non-branch instructions,  labels are preserved.
\item When we execute and resolve a branch instruction (thereby eliminating one of the sources of speculation), we remove its index from later commands' labels.
	As a result, some later commands may now have an empty label, i.e., they are certainly non-transient.
\end{asparaitem}
Overall, the labeling function ensures that reorder buffer entries that may depend on transiently retrieved data are labelled with a non-empty label at every point of the computation.

\item To delay \emph{only} transmit instructions, the unlabeling function, defined in Figure~\ref{figure:stt:unlabeling}, replaces assignments $\passign{x}{e}$ whose label is non-empty with $\passign{x}{\bot}$ for $\fetch{}$ and $\execute{i}$ directives when the $i$-th entry in the buffer is a transmit instruction.
This ensures that transmit instructions are not executed whenever they depend on possibly transient data, which are now mapped to $\bot$.
In contrast, the unlabeling function simply strips the taint tracking labels for $\retire{}$ and $\execute{i}$ directives whenever the $i$-the entry is not a transmit instruction; thereby allowing the hardware to freely execute non-transmit instructions.

\begin{figure}
\begin{align*}
        \unlabelNda{\buf}{\fetch{}} &= \mask{\buf}\\
        \unlabelNda{\buf}{\retire{}} &= \drop{\buf}\\
        \unlabelNda{\buf}{\execute{i}} &=  
        {
            \begin{cases}
                \mask{\buf} & \text{if } \transmitGadget{\elt{\buf}{i}}\\
                \drop{\buf} & \text{otherwise}
            \end{cases}
        }
        \\
        \drop{\emptysequence} & := \emptysequence\\
        \drop{ \labelled{\tagged{i}{T}}{L} \concat \buf } & := \tagged{i}{T} \concat \drop{\buf}\\
        \mask{\emptysequence} &:= \emptysequence \\ 	
        \mask{ \labelled{\tagged{i}{T}}{L} \concat \buf } & := 
        {
            \begin{cases}
                \tagged{\passign{x}{\bot}}{T} \concat \mask{\buf} & \text{if } L = \emptyset  \wedge  \\
                & \ i = \passign{x}{e}  \\ 
                \tagged{i}{T} \concat \mask{\buf} & \text{otherwise}
            \end{cases}
        }
    \end{align*}
\caption{Unlabeling function $\unlabelNda{\buf}{d }$ for STT}\label{figure:stt:unlabeling}
\end{figure}
\end{asparaitem}

Concretely, $\muarchStyle{tt}$ delays all transmit instructions that depend on transiently retrieved data.
However, $\muarchStyle{tt}$ does not delay transient loads that depend on non-transient data, as acknowledged also in~\cite{STT2019}.
This means that parts of the architectural state can be leaked using speculatively executed instructions.
As Example~\ref{example:nda-stt} shows, $\muarchStyle{tt}$ violates the $\CtSeqInterf{\cdot}$ contract. %
\begin{example}\label{example:nda-stt}
    Consider the Spectre v1 variant from Figure~\ref{figure:v1-variant}, compiled to \srclang{}:
{%
\begin{lstlisting}[mathescape=true,style=Cstyle]
$\mathbf{load}$ z,A + y //accessing A[y]
x $\leftarrow$ y < size_A
$\mathbf{beqz}$ x, $\bot$ //checking y < size_A
z $\leftarrow$ z*64
$\mathbf{load}$ w, B+z //accessing B[A[y]*64]
\end{lstlisting}
}
    Consider two configurations $\sigma$ and $\sigma'$ that agree on the values of $\inlineCcode{A}$, $\inlineCcode{B}$, $\inlineCcode{y}$, and $\inlineCcode{size_A}$ and for which $\sigma(\inlineCcode{y}) > \sigma(\inlineCcode{size_A})$, i.e., the array~$\inlineCcode{A}$ is speculatively accessed out of bounds.
    Furthermore, assume that $\sigma(\inlineCcode{A + y})  = 0$ and $\sigma'(\inlineCcode{A + y}) = 1$.
    Then, $\CtSeqInterf{p}(\sigma) = \CtSeqInterf{p}(\sigma') = \loadObs{\inlineCcode{A+y}} \concat \pcObs{\bot}$.
    However, the hardware semantics can potentially leak information through the data cache if the hardware speculatively executes the load on line 5.
    Indeed, the load on line~1 is labeled with $\emptyset$ since it is \emph{not} transient.
    Hence, the load operation on line 5, which depends on the result of line $1$, is not delayed (even though transient operations relying on its result would be delayed since line $5$ is labeled with the index corresponding to the unresolved branch in line 3).
    By probing the state of the cache an attacker can determine whether $\inlineCcode{A[y]  = 0}$ or $\inlineCcode{A[y] = 1}$, thereby distinguishing $\sigma$ and $\sigma'$.\looseness=-1
\end{example}

One way to characterize the guarantees provided by the $\muarchStyle{tt}$ countermeasure is with the $\CtSpecInterf{\cdot}$ contract.
\begin{restatable}{thm}{sttOne}
    \label{theorem:hni:stt:one}
    $\hsni{\CtSpecInterf{\cdot}}{ \TtMuarchSem{\cdot} }$.
\end{restatable}

However, we remark that this contract is already satisfied by the baseline hardware defined in Section~\ref{sec:hw-side} without any countermeasures. 
A more meaningful characterization of $\muarchStyle{tt}$'s guarantees, stated in Theorem~\ref{theorem:hni:stt:two}, is via the $\ArchSeqInterf{\cdot}$ contract. %
Intuitively,  $\muarchStyle{tt}$ satisfies $\ArchSeqInterf{\cdot}$ as it prevents the execution of transmit instructions based on transiently retrieved data.

\begin{restatable}{thm}{sttTwo}
    \label{theorem:hni:stt:two}
    $\hsni{\ArchSeqInterf{\cdot}}{ \TtMuarchSem{\cdot} }$.
\end{restatable}

Theorem~\ref{theorem:hni:stt:two} confirms the results of~\cite{STT2019} and provides a clean characterization of the {\em transient noninterference}~\cite{STT2019} guarantees in terms of the $\ArchSeqInterf{\cdot}$ contract.

\subsubsection{Non-speculative data access}\label{sec:countermeasures:taint-tracking:nda}

Weisse et al.~\cite{nda2019weisse} propose NDA, a family of countermeasures for secure speculation that also relies on hardware taint tracking.
In a nutshell, NDA delays the propagation of speculatively executed instructions until the corresponding speculation sources have been resolved.
NDA comes with two different propagation strategies---\emph{strict} and \emph{permissive} propagation---that can be modeled  as follows:
\begin{asparaitem}
\item For both propagation strategies, the unlabeling function simply replaces all assignments $\passign{x}{e}$ whose label is non-empty with $\passign{x}{\bot}$, thereby preventing the propagation of potentially transient data.
This differs from STT where labels are sometimes removed to allow the propagation of potentially transient data, as long as the propagation does not lead to leaks. %
\item The labeling function differs from the one in $\muarchStyle{tt}$ in how newly fetched instructions are labeled.
For the  strict strategy, \textit{all} newly fetched transient instructions are labelled with the indexes of unresolved branch instructions.
In contrast, \textit{only} newly fetched transient $\loadKywd$s are labeled with the indexes of unresolved branch instructions under the permissive strategy.
\end{asparaitem}

Despite these differences, NDA provides similar guarantees to~$\muarchStyle{tt}$.
That is, it satisfies the $\CtSpecInterf{\cdot}$ and $\ArchSeqInterf{\cdot}$ contracts. %

\subsection{Summary}

Figure~\ref{figure:guarantees} summarizes the results of this section in the lattice structure established in~\S\ref{sec:contracts:contracts:lattice}. This yields the first rigorous comparison of the security guarantees of mechanisms for secure speculation, and it translates the results from~\S\ref{sec:sw-side} into a principled basis for programming them securely.

\renewcommand{\ydistance}{0.9}

\begin{figure}[t!]
    \centering

    \begin{tikzpicture}[->,>=stealth',shorten >=1pt,auto, semithick]

    \node[fill=none,draw=none, shape = rectangle, rounded corners, inner sep=0pt, outer sep=0pt, minimum height = 1pt, minimum width = 1pt]  (seqCt) at (0,0) {$\CtSeqInterf{\cdot}$};

    \node[fill=none,draw=none, shape = rectangle, rounded corners, inner sep=0pt, outer sep=0pt, minimum height = 1pt, minimum width = 1pt]  (specArch) at ($(seqCt) + (-3,-2.5*\ydistance)$) {$\ArchSpecInterf{\cdot}$};

    \node[fill=none,draw=none, shape = rectangle, rounded corners, inner sep=0pt, outer sep=0pt, minimum height = 16pt, minimum width = 1pt]  (seqArch) at ($(seqCt) + (-3,-0.5*\ydistance)$) {$\ArchSeqInterf{\cdot}$};

    \node[fill=none,draw=none, shape = rectangle, rounded corners, inner sep=0pt, outer sep=0pt, minimum height = 16pt, minimum width = 1pt]  (specPcCt) at ($(seqCt) + (0,-\ydistance)$) {$\CtPcSpecInterf{\cdot}$};

    \node[fill=none,draw=none, shape = rectangle, rounded corners, inner sep=0pt, outer sep=0pt, minimum height = 16pt, minimum width = 1pt]  (specCt) at ($(specPcCt) + (0,-\ydistance)$) {$\CtSpecInterf{\cdot}$};

    \path (specPcCt) edge[] node[left] {} (seqCt);
    \path (seqArch) edge[] node[left] {} (seqCt);
    \path (specArch) edge[] node[left] {} (seqArch);
    \path (specCt) edge[] node[left] {} (specPcCt);
    \path (specArch) edge[] node[left] {} (specCt);

            \begin{scope}[on background layer]
                \node[draw=gray!80,dashed, fill=none, fill opacity=0.5, rounded corners, rotate fit=-8
                , minimum height=0.1em, fit=(seqArch) (specPcCt)](Fit1) {};
                \node[draw=gray!80,dashed, fill=none, fill opacity=0.5, rounded corners, rotate fit=-24, minimum height=0.1em, fit=(seqArch) (specCt)](Fit2) {};
                \node[draw=gray!80,dashed, fill=none, fill opacity=0.5, rounded corners, rotate fit=0, minimum height=0.1em, fit=(seqCt)](Fit3) {};
                \node  (label1) at ($(Fit1.east) + (1,0)$) {$\LoadDelayMuarchSem{\cdot}$};
                \node  (label2) at ($(Fit2.east) + (0.5,0)$) {$\TtMuarchSem{\cdot}$};
                \node  (label3) at ($(Fit3.east) + (0.7,0)$) {$\SeqProcMuarchSem{\cdot}$};
            \end{scope}

    \end{tikzpicture}

    \caption{Security \hspace{-0.125mm}guarantees \hspace{-0.125mm}of \hspace{-0.125mm}secure-speculation \hspace{-0.125mm}mechanisms. %
    }\label{figure:guarantees}
\end{figure}
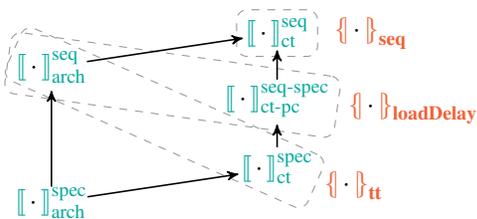

\section{Discussion}\label{sec:discussion}

\subsection{Scope of the model}

With our modeling of a generic microarchitecture and corresponding side-channel adversaries (\S\ref{sec:hw-side}), we aim to strike a balance between capturing the central aspects of attacks on speculative and out-of-order processors, while obtaining a general and tractable model.

As a consequence, we simplify many aspects of modern processors.
For instance, we model only a simple 3-stage pipeline, single threaded, %
and with conditional branch prediction as the only source of speculation. Likewise, we consider an adversary that can observe instructions in the reorder buffer and memory blocks in the cache, but not the data they carry.

This modelling is adequate for reasoning about protections against variants of Spectre v1. However, it does not encompass features such as store-to-load forwarding or prediction over memory aliasing, or adversaries that can observe leaks from internal processor buffers, such as those exploited in data-sampling attacks~\cite{ridl2019, zombieload2019}.

As a consequence, Theorems~\ref{theorem:hni:all}--\ref{theorem:hni:stt:two} need not extend to these scenarios.
However, our framework for expressing contracts is not limited to this simple model, as we discuss next. %

\subsection{Beyond Spectre v1}
We now discuss how to extend our framework to  other transient execution attacks.
For each attack, we discuss how to (1) extend our contracts, and (2) adjust our hardware semantics:\looseness=-1
\begin{asparaitem}
\item \textit{Spectre-BTB and Spectre-RSB}: These variants speculate respectively over indirect jumps and return instructions.
To support them, the \interfStyle{spec}-contracts can be extended to explore  all possible mispredicted paths for a bounded number of steps before rolling back (similarly to the \textsc{Branch} rule in Figure~\ref{figure:interface:spec-ct}).
Moreover, our hardware semantics $\muarchSem{\cdot}$ can also easily be extended to handle these new forms of speculation.
For instance, speculation over indirect jumps could be modeled similarly to the \textsc{Fetch-Branch-Hit} rule in \S\ref{sec:hw-side}.
\item \textit{Spectre-STL}: This variant speculates over memory aliasing over in-flight store and load operations.
Extending our contracts to handle this new kind of speculation requires to modify the $\interfStyle{spec}$-contracts to model the effects of store-to-load forwarding resulting from memory aliasing predictions.
This could be done similarly to Pitchfork~\cite{constanttime2019}.
That is, the $\interfStyle{spec}$ semantics could keep track of the  issued $\pstore{x}{e}$ instructions.
Then, whenever a $\pload{y}{e'}$ instruction is executed, one could explore multiple paths representing all possible aliasing predictions for a fixed number of steps and later roll back.
Finally, the  $\muarchSem{\cdot}$ semantics can be extended to support Spectre-STL  similarly to other semantics~\cite{spectre-heretostay,balliu2019inspectre,constanttime2019}.

\item \textit{Straight-line speculation}: Some CPUs can speculatively execute instructions that follow straight after unconditional jumps or function returns\cite{arm-straightline}, which may result in speculative leaks. A corresponding contract can be captured by an execution mode (similar to $\interfStyle{spec}$) that explores instructions following unconditional changes in control flow up to a bounded depth and exposes the corresponding observations, whereas our hardware semantics $\muarchSem{\cdot}$ can be extended to support this form of speculation similarly to \textsc{Fetch-Branch-Hit} rule in \S\ref{sec:hw-side}.

\item \textit{Meltdown and MDS}: In Spectre-type attacks, transient execution is caused by control- and data-flow mispredictions. In Meltdown-type~\cite{spectreSoK} attacks  (encompassing both Meltdown~\cite{meltdown2018} and data sampling attacks~\cite{ridl2019, zombieload2019}), transient execution is caused by instruction faults or microcode assists.
Contracts for processors that are vulnerable to this kind of attacks would need to expose large portions of the memory space, which makes secure programming challenging. %
\end{asparaitem}

\subsection{Uses of contracts}

The contracts we propose in this paper are designed to adequately capture the security guarantees offered by existing mechanisms for secure speculation, while exposing tractable verification conditions for software. We envision hardware vendors to produce such contracts for their CPUs, to enable users to reason about software security without exposing details of the microarchitecture, and to provide a baseline against which to validate the vendors' security claims.

Moreover, rather than trying to infer contracts for a microarchitecture that has not been designed with security in mind, our framework can serve as a basis for a clean-slate approach, where one starts from a desired security contract and aims to design microarchitectures that optimize performance within these constraints.

\section{Related work}\label{sec:related}

\paragraph{Speculative execution attacks}
These attacks exploit microarchitectural side-effects of speculatively executed instructions to leak information.
There exist many Spectre~\cite{Kocher2018spectre} variants  that differ in the exploited speculation sources~\cite{Maisuradze:2018:RSE:3243734.3243761, 220586, spectreV4}, the covert channels~\cite{trippel2018meltdownprime,schwarz2018netspectre, stecklina2018lazyfp} used, or the target platforms~\cite{chen2018sgxpectre}. 
We refer to~\cite{spectreSoK, xiong2020survey} for a survey.\looseness=-1%

\paragraph{Hardware-level countermeasures}
Here, we review proposals that  we have not formalized in \S\ref{sec:countermeasures}:
\begin{asparaitem}
    \item {``Redo''-based countermeasures}~\cite{invisispec2018,safespec2019,ainsworth2020muontrap}   execute speculative memory operations on shadow cache structures. 
    Once a memory operation becomes non-speculative, its effects are replicated on the standard cache hierarchy by re-executing the operation.
    While these countermeasures satisfy $\CtSpecInterf{\cdot}$, they likely violate $\ArchSeqInterf{\cdot}$ as they still modify other parts of the \uarchstate{} such as the reorder buffer. 
    This intuition is confirmed by speculative interference attacks~\cite{specinterference20}, which demonstrate how to convert seemingly transient changes to the reorder buffer into persistent cache-state changes.
    
   \item In contrast, {``Undo''-based countermeasures}~\cite{saileshwar2019cleanupspec}  mitigate Spectre attacks by rolling back the effects of speculatively executed instructions on the cache.
    Such countermeasures provide security against adversaries that observe the final cache state, but they likely do not provide guarantees against the trace-based attackers we consider in this paper.

    \item Delay-based mitigations selectively delay the execution of some instructions to prevent speculative leaks.
    In addition to the \muarchStyle{loadDelay} countermeasure studied in \S\ref{sec:countermeasures:load-delay},  Sakalis et al.~\cite{specshadow2019}  propose a more permissive scheme, similar to conditional speculation~\cite{conditionalspeculation},   where only loads resulting in cache misses are delayed. %
    These countermeasures, however, would violate the $\ArchSeqInterf{\cdot}$ and $\CtPcSpecInterf{\cdot}$ contracts because cache hits would still leak information.
    
    SpecShield~\cite{specshield} proposes two  countermeasures: one similar to eager-delay and the other similar to NDA's permissive strategy, with similar guarantees as those of $\muarchStyle{loadDelay}$ and $\muarchStyle{tt}$.

    Finally, some proposals, like~\cite{contextsensitivefencing,context}, improve efficiency by only delaying instructions that may leak program-level sensitive information.
    This is achieved by either considering all user-provided data as untrusted~\cite{contextsensitivefencing} or by allowing the specification of program-level policies~\cite{context}.

\end{asparaitem}

\paragraph{Formal microarchitectural models}
While several works~\cite{conf/popl19/armstrong,degenbaev2012x86,Goel2017} present formal architectural models for (parts of) the ARMv8-A, RISC-V,  MIPS, and x86 ISAs, only recently researchers started to focus on formal models of  microarchitectural aspects.
For instance, Coppelia~\cite{zhang2018coppelia} is a tool to automatically generate software exploits for hardware designs.\looseness=-1

The speculative semantics from~\cite{spectector2020} forms the basis for the $\CtSpecInterf{\cdot}$ contract that exposes the effects of speculatively executed instructions.
In contrast to~\cite{spectector2020}, other  semantics~\cite{spectre-heretostay,constanttime2019,blade,balliu2019inspectre}  more closely resemble the actual microarchitectural behavior of  out-of-order processors with multiple pipeline stages, rather than concisely capturing the resulting leakage.
Specifically, the hardware semantics $\CtxMuarchSem{\cdot}{}$ in \S\ref{sec:hw-side} extends~\cite{constanttime2019,blade}'s semantics by making explicit the dependencies with caches,  predictors, and pipeline scheduler. 
Finally, Disselkoen et al.~\cite{disselkoen2018code} presents a speculative semantics based on ideas from recent advances on relaxed memory models.

Fadiheh et al.~\cite{Fadiheh19upec} propose a SAT-based boundel model checking methodology to check whether a given register-transfer level (RTL) processor design exhibits covert channel vulnerabilities.
In the terminology of our work, they check whether the processor design satisfies the $\ArchSeqInterf{\cdot}$ contract.

\paragraph{HW-SW contracts for side channels}
Recently, researchers~\cite{Heiser18,Ge2018} have been calling for new hardware-software contracts that expose security-relevant microarchitectural details. 
We answer this call by providing contracts for secure speculation and by showing how they can be leveraged at the software level.

Recent work~\cite{YuHHF19,ZagieboyloSM19} presents extensions to the RISC-V ISA where data is %
 labeled, e.g., as \emph{Public} or \emph{Secret}; labels are tracked during the computation; and the microarchitecture ensures that secret data does not leak. 
This work is orthogonal to ours in that we characterize the security of different hardware-level countermeasures for a standard ISA. %

\section{Conclusions}\label{sec:conclusions}

Motivated by a lack of hardware-software contracts that support principled co-design for secure speculation, we presented a framework for specifying such contracts.

On the hardware side, we used our framework to provide the first uniform characterization of guarantees provided by a representative set of mechanisms for secure speculation. 

On the software side, we used our framework to characterize secure programming in two scenarios---``constant-time programming'' and ``sandboxing''---and we show how to automate checks for programs to run securely on top of these mechanisms.

\medskip
\subsubsection*{Acknowledgments}
Pepe Vila's work was done while at Microsoft Research.
We would like to thank David Chisnall, Muntaquim Chowdhury, Matthew Fernandez, C{\'e}dric Fournet, Carlos Rozas, and Gururaj Saileshwar for feedback and discussions.   
This work was supported by a grant from Intel Corporation,
Atracci\'on de Talento Investigador grant 2018-T2/TIC-11732A,
Juan de la Cierva-Formaci\'on grant FJC2018-036513-I,
Spanish project RTI2018-102043-B-I00 SCUM,
and Madrid regional project S2018/TCS-4339 BLOQUES.

\bibliographystyle{IEEEtran}
\balance
\bibliography{biblio}
\appendices
\section{Proof overview}\label{appendix:proof:overview}
Theorems~\ref{theorem:hni:all}--\ref{theorem:hni:stt:two} are statements about contract satisfaction of the form $\hsni{\interfSem{\cdot}}{ \muarchSem{\cdot} }$. For their proof, we need to show that, given an arbitrary program  $p$ and two arbitrary \archstate{}s $\sigma$, $\sigma'$ such that $\interfSem{p}(\sigma)=\interfSem{p}(\sigma')$, then $\muarchSem{p}(\sigma)=\muarchSem{p}(\sigma')$.
The proofs follow a common structure, which we outline in this section. For full details, see~\techReportAppendices{appendix:proofs:general}{appendix:proofs:taint-tracking}.

\paragraph{Notation}
In the following,  $\crun$, $\crunp$ denote the contract runs corresponding to the traces $\interfSem{p}(\sigma)$ and $\interfSem{p}(\sigma')$. Similarly, $\hrun$, $\hrunp$ denote the hardware runs corresponding to $\muarchSem{p}(\sigma)$ and $\muarchSem{p}(\sigma')$.
Moreover, $\hrun(i)$ and $\crun(i)$ respectively denote $\hrun$'s and $\crun$'s $i$-th state.
All proofs rely on four main components:
\begin{asparaitem}
\item an indistinguishability relation $\!\muarchStyle{\approx}\!$ between hardware states,
\item relations  $\equiv_j$ between hardware and contract states,
\item a correspondence function $\map{\crun}{\hrun}{\cdot}$ that maps hardware states in $\hrun$ to $\equiv_j$-related contract states in $\crun$, and
\item an indistinguishability lemma capturing under which conditions one step of the hardware semantics $\muarchStepCompact{}$ preserves state indistinguishability $\muarchStyle{\approx}$.
\end{asparaitem}
We next describe each of these components and outline how they are combined for proofs of $\hsni{\interfSem{\cdot}}{ \muarchSem{\cdot} }$.

\paragraph{State indistinguishability relation $\muarchStyle{\approx}$}
Two hardware states related by $\muarchStyle{\approx}$ must be indistinguishable by the microarchitectural adversary from \S\ref{sec:hw-side:adversaries}. In particular, this implies that they execute the same pipeline step.
Additionally, $\muarchStyle{\approx}$ captures proof invariants that are associated with specific hardware semantics.
For instance, for Theorem~\ref{theorem:hni:all},  $\muarchStyle{\approx}$ requires agreement on $\pc$-values, whereas for Theorem~\ref{theorem:hni:stt:two}, $\muarchStyle{\approx}$ requires agreement on register assignments and reorder buffers entries with empty labels.\looseness=-1

\paragraph{Hardware-contract relation $\equiv_j$}
This is a family of relations between contract states $\interfStyle{c}$ and hardware states $\muarchStyle{h}$, where, intuitively, $\interfStyle{c} \equiv_j \muarchStyle{h} $ holds if contract state $\interfStyle{c}$ is related to hardware state $\muarchStyle{h}$ when considering the prefix of $\muarchStyle{h}$'s reorder buffer of length $j$.

The specific definition of $\equiv_j$ depends on the considered contract:
For instance, in Theorems~\ref{theorem:hni:sequential}, \ref{theorem:hni:load-delaytwo}, and~\ref{theorem:hni:stt:two}, which consider contracts with $\interfStyle{seq}$ execution mode, $\interfStyle{c} \equiv_j \muarchStyle{h} $ iff $\interfStyle{c}$ is equivalent to the architectural state obtained by applying all  commands in $\muarchStyle{h}$'s reorder buffer up to the $j$-th instruction to the memory and registers in $\muarchStyle{h}$.
In contrast, for Theorems~\ref{theorem:hni:all},~\ref{theorem:hni:load-delayone}, and~\ref{theorem:hni:stt:one}, which consider contracts with $\interfStyle{spec}$ execution mode, $\interfStyle{c} \equiv_j \muarchStyle{h} $ additionally require that invariants related with speculative execution are satisfied.
For instance, $\interfStyle{c} \equiv_j \muarchStyle{h} $ ensures that all mispredicted branch instructions in $\muarchStyle{h}$'s buffer correspond to configurations whose speculative window is not $\interfStyle{\infty}$ (denoting non-speculative execution) in $\interfStyle{c}$.

\paragraph{Correspondence $\map{\crun}{\hrun}{\cdot}$}
The reorder buffer of a hardware state can contain data corresponding to multiple contract-level instructions. We use the function  $\map{\crun}{\hrun}{\cdot}$ to map {\em prefixes} of the reorder buffer of hardware states to unique contract states.
Specifically, $\map{\crun}{\hrun}{i}(j) = k$ indicates that the reorder buffer prefix of length $j$ of $\hrun(i)$ is mapped to $\crun(k)$.
In~\techReportAppendices{appendix:proofs:general}{appendix:proofs:taint-tracking}, we construct correspondences for the $\interfStyle{seq}$ and $\interfStyle{spec}$ execution modes such that: 
\begin{compactitem}
    \item if $\map{\crun}{\hrun}{i}(j) = k$, then $\crun(k) \equiv_{j} \hrun(i)$, and 
    \item $\map{\crun}{\hrun}{i} = \map{\crunp}{\hrunp}{i}$ for all $i$. 
\end{compactitem}

\paragraph{Indistinguishability lemma}
The indistinguishability lemma is the key intermediate result in our contract satisfaction proofs.
The lemma states that, given two reachable and indistinguishable hardware states $\muarchStyle{h}$, $\muarchStyle{h'}$, if for all prefixes $\buf$ of  $\muarchStyle{h}$'s reorder buffer, there are  contract states $\interfStyle{c}$ and $\interfStyle{c'}$ such that (a) $\interfStyle{c} \equiv_{|\buf|} \muarchStyle{h}$ and $\interfStyle{c'} \equiv_{|\buf|} \muarchStyle{h'}$, and (b) $\interfStyle{c}$ and $\interfStyle{c'}$ produce the same observation $\interfStyle{l}$ when executing one step of the contract semantics $\interfStepCompact$, then either $\muarchStyle{h}$ and $\muarchStyle{h'}$ are stuck or doing one step of the hardware semantics $\muarchStepCompact$ preserves the indistinguishability. See Figure~\ref{figure:indistinguishability-lemma} for a visualization. In~\techReportAppendices{appendix:proofs:general}{appendix:proofs:taint-tracking}, we prove indistinguishability lemmata for each combination of hardware semantics and contract. 

\paragraph{Contract satisfaction outline}
To prove $\hsni{\interfSem{\cdot}}{ \muarchSem{\cdot} }$, we need to show that $\muarchSem{p}(\sigma)=\muarchSem{p}(\sigma')$ if $\interfSem{p}(\sigma)=\interfSem{p}(\sigma')$.
To this end, we prove by induction that hardware states $\hrun(i)$ and $\hrunp(i)$ are indistinguishable for all $i$, which by definition of~$\muarchStyle{\approx}$ implies $\muarchSem{p}(\sigma)=\muarchSem{p}(\sigma')$.

(Induction basis): The initial hardware states $\hrun(0)$ and $\hrunp(0)$ are indistinguishable by definition as they agree on their microarchitectural components.

(Inductive step):
Assume that $\hrun(i) \muarchStyle{\approx} \hrunp(i)$.
Note that (1)~$\crun$, $\crunp$ agree on observations since $\interfSem{p}(\sigma)=\interfSem{p}(\sigma')$, %
and (2)~$\map{\crun}{\hrun}{i} = \map{\crunp}{\hrunp}{i}$ holds by construction.
Therefore, the correspondence mappings $\map{\crun}{\hrun}{i}$ and $\map{\crunp}{\hrunp}{i}$ provide contract states that satisfy conditions (a) and (b) of the indistinguishability lemma for the indistinguishable hardware states $\hrun(i)$ and $\hrunp(i)$, and we can conclude $\hrun(i+1) \muarchStyle{\approx} \hrunp(i+1)$.\looseness=-1

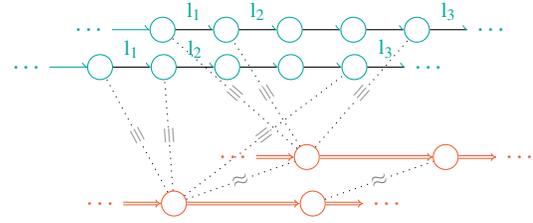
\begin{figure}
\centering
\begin{tikzpicture}

	{
		\node[draw=none,rounded corners](c1){\interfStyle{\ldots}};
		\node[draw=\idecol,circle,font=\footnotesize, right = .5 of c1](c2){};
		\node[draw=\idecol,circle,font=\footnotesize, right = .5 of c2](c3){};
		\node[draw=\idecol,circle,font=\footnotesize, right = .5 of c3](c4){};		
		\node[draw=\idecol,circle,font=\footnotesize, right = .5 of c4](c5){};
		\node[draw=\idecol,circle,font=\footnotesize, right = .5 of c5](c6){};
		\node[draw=none,rounded corners, right = .5 of c6](c9){\interfStyle{\ldots}};
		
 		\draw[draw=\idecol, font=\footnotesize, -left to] (c1) -- (c2);   
		\path[draw=\idecol, font=\footnotesize, -left to] (c2) edge [above] node[name=l1, inner xsep = 0]  {$\interfStyle{\footnotesize l_1}$} (c3);   
		\path[draw=\idecol, font=\footnotesize, -left to] (c3) edge [above] node[name=l2, inner xsep = 0]   {$\interfStyle{\footnotesize l_2}$} (c4);   
		\path[draw=\idecol, font=\footnotesize, -left to] (c4) edge [above] node[name=l3, inner xsep = 0]   {} (c5);   
		\path[draw=\idecol, font=\footnotesize, -left to] (c5) edge [above] node[name=l4, inner xsep = 0]   {} (c6); 
		\path[draw=\idecol, font=\footnotesize, -left to] (c6) edge [above] node[name=l5, inner xsep = 0]   {$\interfStyle{\footnotesize l_3}$} (c9);

		\node[draw=none,rounded corners, below = 1.5 of c2](h1){\muarchStyle{\ldots}};
		\node[draw=\ulccol,circle,font=\footnotesize, right = .5 of h1](h2){};
		\node[draw=\ulccol,circle,font=\footnotesize, right = 1.5 of h2](h3){};
		\node[draw=none,rounded corners, right = .5 of h3](h9){\muarchStyle{\ldots}};
		
		\draw[draw=\ulccol, double, -{implies}] (h1) -- (h2);   
		\draw[draw=\ulccol, double, -{implies}] (h2) -- (h3);
		\draw[draw=\ulccol, double, -{implies}] (h3) -- (h9);

		\node[draw=none,rounded corners, above right = 0.3 of c1](c1p){\interfStyle{\ldots}};
		\node[draw=\idecol,circle,font=\footnotesize, right = .5 of c1p](c2p){};
		\node[draw=\idecol,circle,font=\footnotesize, right = .5 of c2p](c3p){};
		\node[draw=\idecol,circle,font=\footnotesize, right = .5 of c3p](c4p){};		
		\node[draw=\idecol,circle,font=\footnotesize, right = .5 of c4p](c5p){};
		\node[draw=\idecol,circle,font=\footnotesize, right = .5 of c5p](c6p){};
		\node[draw=none,rounded corners, right = .5 of c6p](c9p){\interfStyle{\ldots}};
		
 		\draw[draw=\idecol, -left to] (c1p) -- (c2p);   
		\path[draw=\idecol, -left to] (c2p) edge [above, font=\footnotesize] node[name=l1p, inner xsep = 0]  {$\interfStyle{ l_1}$} (c3p);   
		\path[draw=\idecol, -left to] (c3p) edge [above, font=\footnotesize] node[name=l2p, inner xsep = 0]  {$\interfStyle{ l_2}$} (c4p);   
		\path[draw=\idecol, -left to] (c4p) edge [above, font=\footnotesize] node[name=l3p, inner xsep = 0] {} (c5p);   
		\path[draw=\idecol, -left to] (c5p) edge [above, font=\footnotesize] node[name=l4p, inner xsep = 0]  {} (c6p); 
		\path[draw=\idecol, -left to] (c6p) edge [above, font=\footnotesize] node[name=l5p, inner xsep = 0] {$\interfStyle{ l_3}$} (c9p);

		\node[draw=none,rounded corners, above right = 0.5 of h2](h1p){\muarchStyle{\ldots}};
		\node[draw=\ulccol,circle,font=\footnotesize, right = .5 of h1p](h2p){};
		\node[draw=\ulccol,circle,font=\footnotesize, right = 1.5 of h2p](h3p){};
		\node[draw=none,rounded corners, right = .5 of h3p](h9p){\muarchStyle{\ldots}};
		
		\draw[draw=\ulccol, double, -{implies}] (h1p) -- (h2p);   
		\draw[draw=\ulccol, double, -{implies}] (h2p) -- (h3p);   
		\draw[draw=\ulccol, double, -{implies}] (h3p) -- (h9p);

		\path[draw=gray, font=\footnotesize, dotted] (h2p) edge [sloped] node [pos=0.5, anchor=center, fill=white, inner sep = 0]  {$\textcolor{gray}{\equiv}$} (c2p);   
		\path[draw=gray, font=\footnotesize, dotted] (h2p) edge [ sloped] node [pos=0.5, anchor=center, fill=white, inner sep = 0]   {$\textcolor{gray}{\equiv}$} (c3p);   
		\path[draw=gray, font=\footnotesize, dotted] (h2p) edge [ sloped] node [pos=0.5, anchor=center, fill=white, inner sep = 0]   {$\textcolor{gray}{\equiv}$} (c6p);
		
		\path[draw=black, font=\footnotesize, dotted] (h2) edge [sloped] node [pos=0.5, anchor=center, fill=white, inner sep = 0]  {$\textcolor{gray}{\equiv}$} (c2);   
		\path[draw=black, font=\footnotesize, dotted, inner sep = 0] (h2) edge [ sloped] node [pos=0.5, anchor=center, fill=white]   {$\textcolor{gray}{\equiv}$} (c3);   
		\path[draw=black, font=\footnotesize, dotted, inner sep = 0] (h2) edge [ sloped] node [pos=0.5, anchor=center, fill=white]   {$\textcolor{gray}{\equiv}$} (c6);

		\path[draw=gray, font=\footnotesize, dotted] (h2) edge [sloped] node [pos=0.5, anchor=center, fill=white, inner sep = 0, name=indist1]  {$\textcolor{gray}{\approx}$} (h2p);
		
		\path[draw=gray, font=\footnotesize, dotted] (h3) edge [sloped] node [pos=0.5, anchor=center, fill=white, inner sep = 0, name=indist2]  {$\textcolor{gray}{\approx}$} (h3p);

	}
	
\end{tikzpicture}
\caption{Indistinguishability lemma}\label{figure:indistinguishability-lemma}
\end{figure}

\onlyTechReport{
\newpage
\onecolumn
\section{Architectural semantics}\label{appendix:arch-sem}

The architectural semantics for \srclang{} programs is presented in Figure~\ref{figure:arch-sem}.

\begin{figure*}[hbtp!]
    \begin{tabular}{l c l c l c l c}
        \multicolumn{8}{l}{\bf Expression evaluation}\\
            $\exprEval{n}{a} = n$ & & $\exprEval{x}{a} = a(x)$ &
             $\exprEval{\unaryOp{e}}{a} = \unaryOp{\exprEval{e}{a}}$ & & $\exprEval{\binaryOp{e_1}{e_2}}{a} = \binaryOp{\exprEval{e_1}{a}}{\exprEval{e_2}{a}}$\\\\
        \multicolumn{8}{l}{\bf Instruction evaluation}\\
        \multicolumn{8}{l}{
        {
        \begin{mathpar}
        \inferrule[Skip]
        {
        \select{p}{a(\pc)} = \pskip
        }
        {
        \tup{m, a} \archStepCompact \tup{m, a[\pc \mapsto a(\pc)+1]}
        }
        
        \inferrule[Barrier]
        {
            p(a(\pc)) = \pbarrier
        }
        {
            \tup{m,a} \archStepCompact \tup{m,a[\pc \mapsto a(\pc) + 1]}
        }
        
        \inferrule[Assign]
        {
        \select{p}{a(\pc)} = \passign{x}{e}\\
        x \neq \pc
        }
        {
        \tup{m, a} \archStepCompact \tup{m, a[\pc \mapsto a(\pc)+1,x \mapsto \exprEval{e}{a}]}
        }
        
        \inferrule[ConditionalUpdate-Sat]
        {
            p(a(\pc)) = \pcondassign{x}{e}{e'}\\
            \exprEval{e'}{a} = 0\\
            x \neq \pc
        }
        {
            \tup{m,a} \archStepCompact \tup{m,a[\pc \mapsto a(\pc) + 1, x \mapsto \exprEval{e}{a}]}
        }
        
        \inferrule[ConditionalUpdate-Unsat]
        {
            p(a(\pc)) = \pcondassign{x}{e}{e'}\\
            \exprEval{e'}{a} \neq 0\\
            x \neq \pc
        }
        {
            \tup{m,a} \archStepCompact \tup{m,a[\pc \mapsto a(\pc) + 1]}
        }
        
        \inferrule[Terminate]
        {
            \select{p}{a(\pc)} = \bot
        }
        {
        \tup{m, a} \archStepCompact \tup{m, a[\pc \mapsto \bot]}
        }
        
        \inferrule[Load]
        {
        \select{p}{a(\pc)} = \pload{x}{e} \\
        x \neq \pc \\
        n = \exprEval{e}{a}
        }
        {
        \tup{m, a} \archStepCompact \tup{m, a[\pc \mapsto a(\pc)+1, x \mapsto m(n)]}
        }
        
        \inferrule[Store]
        {
        \select{p}{a(\pc)} = \pstore{x}{e} \\
        n = \exprEval{e}{a}
        }
        {
        \tup{m, a} \archStepCompact \tup{ m[n \mapsto a(x)], a[\pc \mapsto a(\pc)+1]}
        }
        
        \inferrule[Beqz-Sat]
        {
        \select{p}{a(\pc)} = \pjz{x}{\lbl} \\
        a(x) = 0
        }
        {
        \tup{m, a} \archStepCompact \tup{ m, a[\pc \mapsto \lbl]}
        }
        
        \inferrule[Beqz-Unsat]
        {
        \select{p}{a(\pc)} = \pjz{x}{\lbl} \\
        a(x) \neq 0
        }
        {
        \tup{ m, a} \archStepCompact \tup{ m, a[\pc \mapsto a(\pc) +1]}
        }
        
        \inferrule[Jmp]
        {
        \select{p}{a(\pc)} = \pjmp{e} \\
        \lbl = \exprEval{e}{a}
        }
        {
        \tup{m, a} \archStepCompact \tup{ m, a[\pc \mapsto \lbl]}
        }
        \end{mathpar}
        }
        }
        \end{tabular}
\caption{Architectural semantics for a \srclang{} program $\Prg$}\label{figure:arch-sem}
\end{figure*}
\newpage
\section{Sequential scheduler}\label{appendix:seq-scheduler}

Here, we formalize the sequential scheduler from \S\ref{sec:countermeasures:sequential}.
The sequential scheduler $Seq$ is defined as the 4-tuple $\tup{\SchedStates, \SchedState_0, \SchedNext, \SchedUpdate}$ where the components are as follows:
    \begin{align*}
        &\qquad \SchedStates := \{ \BufProject{\buf} \mid \buf \in \ReorderBuffers \}\\
        &\qquad \SchedState_0 := \emptysequence \\
        &\qquad \SchedNext :  \SchedStates \to \Directives := \\
        &\qquad\quad  \SchedNext(\emptysequence) = \fetch{} \\
        &\qquad\quad  \SchedNext(c \concat \buf) = 
        {
            \begin{cases}
                \execute{1} & \text{if}\ \toExecute(c) \\
                \retire		& \text{otherwise}
            \end{cases}
        }\\ 
        &\qquad\quad  \toExecute(\tagged{\pskip{}}{T}) = \bot\\
        &\qquad\quad  \toExecute(\tagged{\pbarrier{}}{T}) = \bot\\
        &\qquad\quad   \toExecute(\tagged{\ptimer{x}}{T}) = \top\\
        &\qquad\quad  \toExecute(\tagged{\passign{x}{e}}{T}) = 
        {
            \begin{cases}
                \top &\text{if}\ e = \unresolved \vee T \neq \notags\\
                \bot & \text{otherwise}
            \end{cases}
        }\\
        &\qquad\quad  \toExecute(\tagged{\pload{x}{e}}{T}) = \top\\
        &\qquad\quad  \toExecute(\tagged{\pstore{x}{e}}{T}) = 
        {
            \begin{cases}
                \top &\text{if}\ e = \unresolved \vee x = \unresolved\\
                \bot & \text{otherwise}
            \end{cases}
        }\\
        &\qquad \SchedUpdate: \SchedStates \times \ReorderBuffers \to \SchedStates := \\
        & \qquad \quad \SchedUpdate(\SchedState, \buf) = \buf
        \end{align*}

\newpage
\section{Labeling function for $\muarchStyle{tt}$}\label{appendix:stt-labeling}

The labeling function $\labelNda{\buf'}{\buf}{d}$ for the STT countermeasure is as follows:
    \begin{align*}
            \labels{\buf} &= \deriveLabels{\buf}{\lambda x \in \Var.\ \emptyset}\\
            \labels{\buf}(x) &= \deriveLabels{\buf}{\lambda x \in \Var.\ \emptyset}(x)\\
            \labels{\buf}(e) &= \bigcup_{x \in \mathit{vars}(e)} \deriveLabels{\buf}{\lambda x \in \Var.\ \emptyset}(x)\\
            \\
            \deriveLabels{\emptysequence}{\Lambda} &= \Lambda \\ 
            \deriveLabels{ \labelled{\tagged{i}{T}}{L} \concat \buf }{\Lambda} &=  
            {
                \begin{cases}
                    \deriveLabels{  \buf }{\Lambda[x \mapsto L ]} & \text{if } i = \pload{x}{e} \vee i = \passign{x}{e}\\
                    \deriveLabels{  \buf }{\Lambda} & \text{otherwise}
                \end{cases}
            }\\
            \\
        \decrement{\emptysequence} &:= \emptysequence \\
        \decrement{\labelled{\tagged{i}{T}}{L} \concat \buf } &:=  \labelled{\tagged{i}{T}}{ \{ j - 1 \mid j \in L \wedge j-1 > 0\} } \concat \decrement{\buf} \\
        \\
        \mathit{strip}(\emptysequence, j) &:= \emptysequence\\
        \mathit{strip}(\labelled{\tagged{i}{T}}{L} \concat \buf , j) &:= \labelled{\tagged{i}{T}}{L \setminus \{ j \}} \concat  \mathit{strip}( \buf , j)\\
        \\
        \labelNda{\buf_{ul}}{\buf}{\fetch{}} &= \buf \\
                \text{where }& |\buf_{ul}| = |\buf| \\
        \labelNda{\buf_{ul} \concat \tagged{\passign{\pc}{e}}{\notags}}{\buf}{\fetch{}} &= \buf \concat \labelled{\tagged{\passign{\pc}{e}}{\notags}}{\emptyset}\\
        \text{where }& |\buf_{ul}| = |\buf| + 1 \\
        \labelNda{\buf_{ul} \concat \tagged{\passign{\pc}{\ell}}{\ell_0}}{\buf}{\fetch{}} &= \buf \concat \labelled{\tagged{\passign{\pc}{\ell}}{\ell_0}}{\emptyset}\\
        \text{where }& |\buf_{ul}| = |\buf| + 1 \\
        \labelNda{\buf_{ul} \concat \tagged{i}{T} \concat \tagged{\pmarkedassign{\pc}{\ell}}{\notags} }{\buf}{\fetch{}} &= \buf \concat \labelled{\tagged{i}{T}}{\emptyset} \concat \labelled{\tagged{\pmarkedassign{\pc}{\ell}}{\notags}}{\emptyset}\\
            \text{where }& 
            i \neq \pload{x}{e} \wedge i \neq \passign{x}{e} \wedge |\buf_{ul}| = |\buf|+2\\
        \labelNda{\buf_{ul} \concat \tagged{ \passign{x}{e} }{T} \concat \tagged{\pmarkedassign{\pc}{\ell}}{\notags} }{\buf}{\fetch{}} &= \buf \concat \labelled{\tagged{ \passign{x}{e} }{T}}{\labels{\buf}(e)} \concat \labelled{\tagged{\pmarkedassign{\pc}{\ell}}{\notags}}{\emptyset}\\
        \text{where }& |\buf_{ul}| = |\buf| + 2 \\
        \labelNda{\buf_{ul} \concat \tagged{\pload{x}{e}}{T} \concat \tagged{\passign{\pc}{\ell}}{\notags} }{\buf}{\fetch{}} &= \buf \concat \labelled{\tagged{\pload{x}{e}}{T}}{ \{ j \in \Nat \mid \exists y, \ell,\ell_0.\  (\elt{\buf}{j} = \tagged{\passign{y}{\ell}}{\ell_0} \wedge \ell_0 \neq \notags) \} } \concat \labelled{\tagged{\passign{\pc}{\ell}}{\notags}}{\emptyset}\\
        \text{where }& |\buf_{ul}| = |\buf| + 2 
            \\
            \\
        \labelNda{\buf_{ul}}{\buf}{\execute{i}} &= \buf[0..i-1] \concat \labelled{\elt{\buf_{ul}}{i}}{ L } \concat \buf[i+1 .. |\buf|] \\
        \text{where }& |\buf_{ul}| = |\buf| \wedge \elt{\buf}{i} = \labelled{c}{ L } \wedge 
        \forall \ell, \ell_0.\ 	(c = \tagged{\passign{\pc}{\ell}}{\ell_0} \rightarrow \ell_0 = \emptysequence)
        \\
        \labelNda{\buf_{ul}}{\buf}{\execute{i}} &= \buf[0..i-1] \concat \labelled{\elt{\buf_{ul}}{i}}{ \emptyset } \concat 
        \mathit{strip}(\buf[i+1 .. |\buf|], i) \\
        \text{where }& |\buf_{ul}| = |\buf| \wedge \elt{\buf}{i} = \labelled{\tagged{\passign{\pc}{\ell}}{\ell_0}}{ L } \wedge \ell_0 \neq  \notags \wedge \elt{\buf_{ul}}{i} = \tagged{\passign{\pc}{\ell'}}{\notags}
        \\
        \labelNda{\buf_{ul}}{\buf}{\execute{i}} &= \buf[0..i-1] \concat \labelled{\elt{\buf_{ul}}{i}}{ L } \concat \buf[i+1 .. |\buf|] \\
        \text{where }& |\buf_{ul}| = |\buf| \wedge \elt{\buf}{i} = \labelled{\tagged{\passign{\pc}{\ell}}{\ell_0}}{ L } \wedge \ell_0 \neq \notags  \wedge \elt{\buf_{ul}}{i} = \tagged{\passign{\pc}{\ell}}{\ell_0} 
        \\
        \labelNda{\buf_{ul}}{\buf}{\execute{i}} &= \buf[0..i-1] \concat \labelled{\elt{\buf_{ul}}{i}}{ \emptyset} \\
        \text{where }& |\buf_{ul}| = i \wedge \elt{\buf}{i} = \labelled{\tagged{\passign{\pc}{\ell}}{\ell_0}}{L} \wedge \ell_0 \neq \notags \wedge \elt{\buf_{ul}}{i} = \tagged{\passign{\pc}{\ell'}}{\notags}
        \\
        \\
        \labelNda{ \buf_{ul} }{\labelled{\tagged{i}{T}}{L} \concat \buf}{\retire{}} &= \decrement{\buf}
        \end{align*}

\newpage
\section{Hardware semantics}\label{appendix:hardware-semantics}

Here we present the full  rules for the hardware semantics.

Differently from the main paper, we slightly modify the hardware semantics to syntactically mark some of the assignments (specifically, those increasing the program counter $\pc$ by 1)  in the reorder buffer as $\tagged{\pmarkedassign{x}{e}}{T}$.
These assignments behave exactly as a normal assignment $\tagged{\passign{x}{e}}{T}$. 
We syntactically mark them to distinguish them from changes to the program counter resulting from branch and jump instructions and to to simplify proofs.

\subsection{Data-independent projection for reorder buffers}

Given a reorder buffer $\buf$, we denote by $\BufProject{\buf}$ its data-independent projection, which is defined as follows:
\begin{align*}
	\BufProject{\emptysequence}	&:= \emptysequence\\
	\BufProject{\tagged{\pskip{}}{T}}  &:=  \tagged{\pskip{}}{T}\\
	\BufProject{\tagged{\pbarrier{}}{T}}  &:=  \tagged{\pbarrier{}}{T}\\
	\BufProject{\tagged{\passign{x}{e}}{T}}  &:= 
	{
		\begin{cases}
		\tagged{\passign{x}{\resolved}}{T} & \text{if}\ e \in \Val\\
		\tagged{\passign{x}{\unresolved}}{T} & \text{otherwise}
		\end{cases}
    }\\
    \BufProject{\tagged{\pmarkedassign{x}{e}}{T}}  &:= 
	{
		\begin{cases}
		\tagged{\pmarkedassign{x}{\resolved}}{T} & \text{if}\ e \in \Val\\
		\tagged{\pmarkedassign{x}{\unresolved}}{T} & \text{otherwise}
		\end{cases}
	}\\
	\BufProject{\tagged{\pload{x}{e}}{T}}  &:= 
	{
		\begin{cases}
		\tagged{\pload{x}{\resolved}}{T} & \text{if}\ e \in \Val\\
		\tagged{\pload{x}{\unresolved}}{T} & \text{otherwise}
		\end{cases}
	}\\
	\BufProject{\tagged{\pstore{x}{e}}{T}}  &:= 
	{
		\begin{cases}
		\tagged{\pstore{\resolved}{\resolved}}{T} & \text{if}\ x,e \in \Val\\
		\tagged{\pstore{\unresolved}{\resolved}}{T} & \text{if}\ x \not\in \Val \wedge e \in \Val\\
		\tagged{\pstore{\resolved}{\unresolved}}{T} & \text{if}\ x \in \Val \wedge e \not\in \Val\\
		\tagged{\pstore{\unresolved}{\unresolved}}{T} & \text{otherwise}
		\end{cases}
	}\\
	\BufProject{c \concat \buf} &:= \BufProject{c} \concat \BufProject{\buf}
\end{align*}

\subsection{Apply function $\apply{\cdot}{\cdot}$}

The $\apply{\cdot}{\cdot}$ function takes as input a reorder buffer $\buf$ and an assignment $a$ and returns a new assignment $a'$ obtained by applying the changes performed by the in-flight commands in $\buf$ to $a$.
Formally, $\apply{\cdot}{\cdot}$ is defined as follows:
\begin{align*}
	\apply{\emptysequence}{a} & = a\\
	\apply{ \tagged{\passign{x}{e}}{T} \concat \buf }{a} & = 
	{
	 \begin{cases}
 		\apply{\buf}{a[x \mapsto e]} & \text{if}\ e \in \Val\\
 		\apply{\buf}{a[x \mapsto \bot]} & \text{if}\ e \not\in \Val\\
	 \end{cases}
	}
\\
\apply{ \tagged{\pmarkedassign{x}{e}}{T} \concat \buf }{a} & = 
{
 \begin{cases}
     \apply{\buf}{a[x \mapsto e]} & \text{if}\ e \in \Val\\
     \apply{\buf}{a[x \mapsto \bot]} & \text{if}\ e \not\in \Val\\
 \end{cases}
}
\\
	\apply{ \tagged{\pload{x}{e}}{T} \concat \buf }{a} & = \apply{\buf}{a[x \mapsto \bot]}\\
	\apply{ \tagged{\pstore{x}{e}}{T} \concat \buf }{a} & = \apply{\buf}{a}\\
	\apply{ \tagged{\pskip}{T} \concat \buf }{a} & = \apply{\buf}{a}\\
	\apply{ \tagged{\pbarrier}{T} \concat \buf }{a} & = \lambda x \in \Var.\ \bot
\end{align*}
That is, assignments $\tagged{\passign{x}{e}}{T}$ and $\tagged{\pmarkedassign{x}{e}}{T}$ update $a$ by either providing a new value for the variable $x$ (if the expression $e$ has already been resolved, i.e., $e \in \Val$) or by setting it to $\bot$ (if the expression $e$ has not yet been resolved).
Load operations $\tagged{\pload{x}{e}}{T}$, instead, always update $a$ by setting the value of $x$ to $\bot$ since the loaded value has not been retrieved yet.
Store and skip operations do not modify the assignments $a$ (because they do not alter the values of registers).
Finally, whenever a speculation barrier is in the buffer, the result of $\apply{\cdot}{\cdot}$ maps all registers to $\bot$.

\subsection{Expression evaluation $\exprEval{e}{a}$}

For the hardware semantics, expression evaluation $\exprEval{e}{a}$ works similarly to the architectural semantics except that now it handles the $\bot$ undefined value.
Whenever, a sub-expression $e'$ evaluates to $\bot$, then $\exprEval{e}{a}$  evaluates to $\bot$ as well.
Formally, the expression evaluation is defined as follows:
\begin{align*}
    \exprEval{n}{a} &= n\\
    \exprEval{\bot}{a} &= \bot\\
    \exprEval{x}{a} &= a(x) \\
    \exprEval{\unaryOp{e}}{a} &= 
    {
        \begin{cases}
            \unaryOp{\exprEval{e}{a}} & \text{if } \exprEval{e}{a} \neq \bot \\
            \bot & \text{otherwise}
        \end{cases}
    }\\
    \exprEval{\binaryOp{e_1}{e_2}}{a} &= 
    {
        \begin{cases}
            \binaryOp{\exprEval{e_1}{a}}{\exprEval{e_2}{a}} & \text{if } \exprEval{e_1}{a} \neq \bot \wedge  \exprEval{e_2}{a} \neq \bot \\
            \bot & \text{otherwise}
        \end{cases}
    }
\end{align*}

\subsection{Rules for \textbf{fetch} stage}

The rules formalizing the $\fetch{}$ stage are given below:

\begin{mathpar}
\inferrule[Fetch-Branch-Hit]
{
a' = \apply{\buf}{a} \\
|\buf| < \wMuarch \\
a'(\pc) \neq \bot \\
p(a'(\pc)) = \pjz{x}{\lbl} \\
\lbl' = \BpPredict(\BpState, a'(\pc))\\
\CacheAccess(\CacheState, a'(\pc)) = \CacheHit\\
\CacheUpdate(\CacheState, a'(\pc)) = \CacheState{}'
}
{
    \tup{m,a,\buf,  \CacheState, \BpState} \muarchStep{\fetch{b}}{} \tup{m,a,\buf \concat \tagged{\passign{\pc}{\lbl'}}{a'(\pc)},  \CacheState',\BpState}	
}

\inferrule[Fetch-Jump-Hit]
{
a' = \apply{\buf}{a} \\
|\buf| < \wMuarch \\
a'(\pc) \neq \bot \\
p(a'(\pc)) = \pjmp{e} \\
\CacheAccess(\CacheState, a'(\pc)) = \CacheHit\\
\CacheUpdate(\CacheState, a'(\pc)) = \CacheState{}'
}
{
	\tup{m,a,\buf,\CacheState, \BpState} \muarchStep{\fetch{b}}{} \tup{m,a,\buf \concat  \tagged{\passign{\pc}{e}}{\notags}, \CacheState', \BpState}	
}

\inferrule[Fetch-Others-Hit]
{
|\buf| < \wMuarch-1 \\
a' = \apply{\buf}{a} \\
a'(\pc) \neq \bot \\
p(a'(\pc)) \neq \pjz{x}{\lbl} \\
p(a'(\pc)) \neq \pjmp{e} \\
v = \exprEval{\pc + 1}{a'}\\
\CacheAccess(\CacheState, a'(\pc)) = \CacheHit\\
\CacheUpdate(\CacheState, a'(\pc)) = \CacheState{}'
}
{
	\tup{m,a,\buf,\CacheState, \BpState} \muarchStep{\fetch{b}}{} \tup{m,a,\buf \concat \tagged{p(a'(\pc))}{\notags} \concat \tagged{\pmarkedassign{\pc}{v}}{\notags}, \CacheState', \BpState}	
}

\inferrule[Fetch-Miss]
{
|\buf| < \wMuarch \\
a' = \apply{\buf}{a} \\
a'(\pc) \neq \bot \\
\CacheAccess(\CacheState,  a'(\pc)) = \CacheMiss\\
\CacheUpdate(\CacheState,  a'(\pc)) =  \CacheState{}'
}
{
    \tup{m,a,\buf, \CacheState, \BpState} \muarchStep{\fetch{b}}{}  \tup{m,a,\buf, \CacheState', \BpState}
}
\end{mathpar}

\subsection{Rules for \textbf{execute} stage}

The rules formalizing the $\retire$ state are given below:

\begin{mathpar}
    \inferrule[Execute-Load-Hit]
    {
        |\buf| = i-1 \\
        a' = \apply{\buf}{a} \\
        \tagged{\pbarrier}{T'} \not\in \buf \\
        \tagged{\pstore{x'}{e'}}{T''} \not\in \buf \\
        x \neq \pc\\
        \exprEval{e}{a'} \neq \bot\\
        \CacheAccess(\CacheState, \exprEval{e}{a'}) = \CacheHit\\
         \CacheUpdate(\CacheState, \exprEval{e}{a'}) = \CacheState'
    }
    {
        \tup{m,a,\buf \concat  \tagged{\pload{x}{e}}{T} \concat \buf',  \CacheState, \BpState} \muarchStep{\execute{i}}{} \tup{m,a,\buf \concat \tagged{\passign{x}{ m(\exprEval{e}{a'})  }}{T} \concat \buf',  \CacheState', \BpState}	
    }
\end{mathpar}
\begin{mathpar}
    \inferrule[Execute-Load-Miss]
    {
        |\buf| = i-1 \\
        a' = \apply{\buf}{a} \\
        \tagged{\pbarrier}{T'} \not\in \buf \\
        \tagged{\pstore{x'}{e'}}{T''} \not\in \buf \\
        x \neq \pc\\
        \exprEval{e}{a'} \neq \bot\\
        \CacheAccess(\CacheState, \exprEval{e}{a'}) = \CacheMiss\\
         \CacheUpdate(\CacheState, \exprEval{e}{a'}) = \CacheState'
    }
    {
        \tup{m,a,\buf \concat  \tagged{\pload{x}{e}}{T} \concat \buf',  \CacheState, \BpState} \muarchStep{\execute{i}}{} \tup{m,a,\buf \concat\tagged{\pload{x}{e}}{T} \concat \buf',  \CacheState', \BpState}	
    }
\end{mathpar}
\begin{mathpar}
    \inferrule[Execute-Branch-Commit]
    {
        |\buf| = i-1 \\
        a' = \apply{\buf}{a} \\
        \tagged{\pbarrier}{T'} \not\in \buf \\
        \ell_0 \neq \notags \\
        p(\lbl_0) = \pjz{x}{\lbl''} \\
        (a'(x) = 0 \wedge \lbl = \lbl'') \vee (a'(x) \in \Val \setminus \{0,\bot\} \wedge \lbl = \ell_0+1)\\ 
        \BpState' = \BpUpdate(\BpState, \ell_0, \ell)  
    }
    {
        \tup{m,a, \buf \concat \tagged{\passign{\pc}{\lbl}}{\lbl_0} \concat \buf', \CacheState, \BpState} \muarchStep{\execute{i}}{} \tup{m,a,\buf \concat \tagged{\passign{\pc}{ \lbl  }}{\varepsilon} \concat \buf' , \CacheState, \BpState'}	
    }
\end{mathpar}
\begin{mathpar}   
    \inferrule[Execute-Branch-Rollback]
    {
        |\buf| = i-1 \\
        a' = \apply{\buf}{a} \\ 
        \tagged{\pbarrier}{T'} \not\in \buf \\
        \ell_0 \neq \notags \\
        p(\lbl_0) = \pjz{x}{\lbl''} \\
        (a'(x) = 0 \wedge \lbl \neq \lbl'') \vee (a'(x) \in \Val \setminus \{0,\bot\} \wedge \lbl \neq \ell_0+1)\\ 
        \lbl' \in \{ \lbl'',\lbl_0 +1 \} \setminus \{ \ell \} \\
        \BpState' = \BpUpdate(\BpState, \ell_0, \ell')  \\
    }
    {
        \tup{m,a, \buf \concat \tagged{\passign{\pc}{\lbl}}{\lbl_0} \concat \buf', \CacheState, \BpState} \muarchStep{\execute{i}}{} \tup{m,a,\buf \concat \tagged{\passign{\pc}{ \lbl'  }}{\varepsilon} , \CacheState, \BpState'}	
    }
\end{mathpar}
\begin{mathpar}
    \inferrule[Execute-Assignment]
    {
        |\buf| = i-1 \\
        a' = \apply{\buf}{a} \\
        \tagged{\pbarrier}{T'} \not\in \buf  \\
        \exprEval{e}{a'} = v \\
        v \neq \bot
    }
    {
        \tup{m,a,\buf \concat \tagged{\passign{x}{e}}{\notags} \concat \buf', \CacheState, \BpState} \muarchStep{\execute{i}}{} \tup{m,a,\buf \concat \tagged{\passign{x}{v}}{\notags} \concat \buf', \CacheState, \BpState}
    }
\end{mathpar}
\begin{mathpar}
    \inferrule[Execute-Marked-Assignment]
    {
        |\buf| = i-1 \\
        a' = \apply{\buf}{a} \\
        \tagged{\pbarrier}{T'} \not\in \buf  \\
        \exprEval{e}{a'} = v \\
        v \neq \bot
    }
    {
        \tup{m,a,\buf \concat \tagged{\pmarkedassign{x}{e}}{\notags} \concat \buf', \CacheState, \BpState} \muarchStep{\execute{i}}{} \tup{m,a,\buf \concat \tagged{\pmarkedassign{x}{v}}{\notags} \concat \buf', \CacheState, \BpState}
    }
\end{mathpar}
\begin{mathpar}
    \inferrule[Execute-Store]
    {
        |\buf| = i-1 \\
        a' = \apply{\buf}{a} \\
        \tagged{\pbarrier}{T'} \not\in \buf  \\
        \exprEval{e}{a'} = n \\
        n \neq \bot\\
        a'(x) = v \\
        v \neq \bot 
    }
    {
        \tup{m,a,\buf \concat \tagged{\pstore{x}{e}}{T} \concat \buf', \CacheState, \BpState} \muarchStep{\execute{i}}{} \tup{m,a,\buf \concat \tagged{\pstore{v}{n}}{T} \concat \buf', \CacheState, \BpState}
    }
\end{mathpar}
\begin{mathpar}
    \inferrule[Execute-Skip]
    {
        |\buf| = i-1 
    }
    {
        \tup{m,a,\buf \concat \tagged{\pskip}{T} \concat \buf', \CacheState, \BpState} \muarchStep{\execute{i}}{} \tup{m,a,\buf \concat \tagged{\pskip}{T} \concat \buf', \CacheState, \BpState}	
    }
\end{mathpar}
\begin{mathpar}
    \inferrule[Execute-Barrier]
    {
        |\buf| = i-1 
    }
    {
        \tup{m,a,\buf \concat \tagged{\pbarrier}{T} \concat \buf', \CacheState, \BpState} \muarchStep{\execute{i}}{} \tup{m,a,\buf \concat \tagged{\pbarrier}{T} \concat \buf', \CacheState, \BpState}	
    }
    \end{mathpar}

\subsection{Rules for \textbf{retire} stage}

The rules formalizing the $\retire$ state are given below:
\begin{mathpar}
    \inferrule[Retire-Skip]
    {
        \buf = \tagged{\pskip{}}{\notags} \concat \buf'
    }
    {
        \tup{m,a,\buf,\CacheState,\BpState} \muarchStep{\retire}{} \tup{m, a, \buf', \CacheState,\BpState}	
    }

    \inferrule[Retire-Fence]
    {
        \buf = \tagged{\pbarrier{}}{\notags} \concat \buf'
    }
    {
        \tup{m,a,\buf,\CacheState,\BpState} \muarchStep{\retire}{} \tup{m, a, \buf', \CacheState,\BpState}	
    }

    \inferrule[Retire-Assignment]
    {
        \buf = \tagged{\passign{x}{v}}{\notags} \concat \buf'\\
        v \in \Val
    }
    {
        \tup{m,a,\buf,\CacheState,\BpState} \muarchStep{\retire}{} \tup{m, a[x \mapsto v], \buf', \CacheState,\BpState}	
    }

    \inferrule[Retire-Marked-Assignment]
    {
        \buf = \tagged{\pmarkedassign{x}{v}}{\notags} \concat \buf'\\
        v \in \Val
    }
    {
        \tup{m,a,\buf,\CacheState,\BpState} \muarchStep{\retire}{} \tup{m, a[x \mapsto v], \buf', \CacheState,\BpState}	
    }

    \inferrule[Retire-Store]
    {
        \buf = \tagged{\pstore{v}{n}}{\notags} \concat \buf'\\
    v,n \in \Val\\
   \CacheUpdate( \CacheState,n) = \CacheState' 
    }
    {
        \tup{m,a,\buf,\CacheState,\BpState} \muarchStep{\retire}{  }
         \tup{m[n \mapsto v], a, \buf' ,  \CacheState', \BpState}	
    }
\end{mathpar}
\newpage
\section{Well-formed conditions}\label{appendix:well-formedness-conditions}

In our results, we consider only \textit{well-formed} \srclang{} programs, that is, programs where (a) all assignments $\passign{x}{e}$ and loads $\pload{x}{e}$ statements are such that $x \neq \pc$, and (b) whenever $p(\ell) = \pjz{x}{\ell'}$ then $\ell' \neq \ell +1$.

Additionally, we consider only \textit{well-formed} branch predictors that for all $\BpState$ and $\lbl$, if $p(\lbl) = \pjz{x}{\ell'}$, then $\BpPredict(\BpState, \lbl) \in \{\ell', \lbl+1\}$, i.e., well-formed predictors predict as next instruction only one of the two branch outcomes.

\newpage
\section{Contracts  for secure speculation}\label{appendix:contracts}

In this section, we present the full formalization of our contracts for secure speculation.

\subsection{Contract $\CtSeqInterf{\cdot}$}
The contract is induced by the following inference rules defining the $\CtSeqInterfStep{}{}$ relation:
\begin{mathpar}
	\inferrule[Load]
	{
	\select{p}{a(\pc)} = \pload{x}{e} \\
	\tup{m, a} \archStep{}{} \tup{m',a'}
	}
	{
	\tup{m, a} \CtSeqInterfStep{\loadObs{\exprEval{e}{a}}}{} \tup{m', a'}
	}

	\inferrule[Store]
	{
	\select{p}{a(\pc)} = \pstore{x}{e} \\
	n = \exprEval{e}{a}\\
	\tup{m, a} \archStep{}{} \tup{m',a'}
	}
	{
	\tup{m, a} \CtSeqInterfStep{\storeObs{n}}{} \tup{ m', a'}
	}

	\inferrule[Beqz]
	{
	\select{p}{a(\pc)} = \pjz{x}{\lbl} \\
	\tup{m, a} \archStep{}{} \tup{m',a'}
	}
	{
	\tup{m, a} \CtSeqInterfStep{\pcObs{a'(\pc)}}{} \tup{ m', a'}
	}

	\inferrule[Jmp]
	{
	\select{p}{a(\pc)} = \pjmp{e} \\
	\tup{m, a} \archStep{}{} \tup{m',a'}
	}
	{
	\tup{m, a} \CtSeqInterfStep{\pcObs{a'(\pc)}}{} \tup{ m', a'}
	}

	\inferrule[Others]
	{
	\select{p}{a(\pc)} \neq \pjmp{e} \\
	\select{p}{a(\pc)} \neq \pjz{x}{\lbl} \\
	\select{p}{a(\pc)} \neq \pstore{x}{e} \\
	\select{p}{a(\pc)} \neq \pload{x}{e} \\
	\tup{m, a} \archStep{}{} \tup{m',a'}
	}
	{
	\tup{m, a} \CtSeqInterfStep{}{} \tup{ m', a'}
	}
\end{mathpar}

Then,  $\CtSeqInterf{p}(\sigma)$, where $\sigma$ is an initial \archstate{}, is the trace of contract observations $\tau$ such that $\sigma \CtSeqInterfStep{\tau}{}^* \sigma'$ and $\sigma'$ is a final \archstate{}.

\subsection{Contract $\ArchSeqInterf{\cdot}$}
The contract $\ArchSeqInterfStep{}{}$ is obtained by modifying the  \textsc{Load} rule from  $\CtSeqInterfStep{}{}$ as follows (the remaining rules are identical):
\begin{mathpar}
	\inferrule[Load]
	{
	\select{p}{a(\pc)} = \pload{x}{e} \\
	\tup{m, a} \archStep{}{} \tup{m',a'}
	}
	{
	\tup{m, a} \ArchSeqInterfStep{\loadObs{\exprEval{e}{a}= m(\exprEval{e}{a})}}{} \tup{m', a'}
	}
\end{mathpar}

Again,  $\ArchSeqInterf{p}(\sigma)$, where $\sigma$ is an initial \archstate{}, is the trace of contract observations $\tau$ such that $\sigma \ArchSeqInterfStep{\tau}{}^* \sigma'$ and $\sigma'$ is a final \archstate{}.

\subsection{Contract $\CtSpecInterf{\cdot}$}
The contract is induced by the relation $\CtSpecInterfStep{}{}$ defined by  the following inference rules:
\begin{mathpar}
	\inferrule[Step]
	{
	p(\sigma(\pc))\neq \pjz{x}{\lbl}\\
	 \conf \CtSeqInterfStep{\tau}{} \conf'\\
	}
	{
		\tup{\sigma, \omega +1} \cdot s \CtSpecInterfStep{\tau}{} \tup{\sigma', \omega}\cdot s
	}

	\inferrule[Rollback]
	{
	s=\tup{\sigma',\omega'}\cdot s'\\
	}
	{
		\tup{\sigma,0 }\cdot s \CtSpecInterfStep{\pcObs{\sigma'(\pc)}}{} s
	}

	\inferrule[Barrier]
	{
	p(\sigma(\pc))= \pbarrier \\
	 \conf \CtSeqInterfStep{\tau}{} \conf'\\
	}
	{
		\tup{\sigma, \omega +1} \cdot s \CtSpecInterfStep{\tau}{} \tup{\sigma', 0}\cdot s
	}
	
	\inferrule[Branch]
	{
	p(\sigma(\pc))=\pjz{x}{\lbl}\quad
	\ell_{\mathit{correct}} =
	{
		\begin{cases}
			\lbl & \text{if}\ \sigma(x) = 0\\
			\sigma(\pc) + 1 & \text{otherwise}
		\end{cases}
	}\quad
	\ell_{\mathit{mispred}} \in \{\lbl, \sigma(\pc) + 1\} \setminus \ell_{\mathit{correct}}\quad
	\omega_\mathit{mispred}={	\begin{cases}
			\wInterf & \text{if } \omega=\infty\\
			\omega & \text{otherwise}
		\end{cases}
	}
	}
	{
		\tup{\sigma,\omega+1 }\cdot s \CtSpecInterfStep{\pcObs{\lbl_\mathit{mispred}}}{} \tup{\sigma[\pc \mapsto \lbl_\mathit{mispred}],\omega_\mathit{mispred}} \cdot\tup{\sigma[\pc \mapsto \lbl_\mathit{correct}],\omega} \cdot s
	}
\end{mathpar}

Configurations are stacks of $\tup{\sigma, \omega}$, where $\omega \in \Nat\cup \{\infty\}$ is the speculative window denoting how many instructions are left to be executed. (initial \archstate{}s $\sigma$ are treated as $\tup{\sigma,\infty}$). %
	At each computation step, the $\omega$ at the top of the stack is reduced by $1$ (rules \textsc{Step} and \textsc{Branch}).
	For the sake of notation's simplicity, we assume the $\infty + 1 = \infty$.
	Therefore, in rules we write $\omega+1$ to denote either $\infty$ or any non-zero $\omega \in \Nat$.
	When executing a branch instruction (rule \textsc{Branch}), the state $\tup{\sigma[\pc \mapsto \lbl_\mathit{mispred}],\omega_\mathit{mispred}}$ is pushed on top of the stack, thereby allowing the exploration of the mispredicted branch for $\omega_\mathit{mispred}$ steps.
	The correct branch $\tup{\sigma[\pc \mapsto \lbl_\mathit{correct}],\omega}$ is also recorded on the stack; allowing to later roll back speculatively executed statements.
	When the $\omega$ at the top of the stack reaches $0$, we pop it (i.e., we backtrack and discard the changes) and we continue the computation.
	Speculation barriers trigger a roll back by setting $\omega$ to $0$ (rule \textsc{Barrier}).

Then,  $\CtSpecInterf{p}(\sigma)$, where $\sigma$ is an initial \archstate{}, is the trace of contract observations $\tau$ such that $\tup{\sigma, \infty} \CtSpecInterfStep{\tau}{}^* \tup{\sigma', \infty}$ and $\sigma'$ is a final \archstate{}.

In the following, given a contract configuration $s \in (\Conf\times \Nat \cup \{\infty\})^*$, we write $s(x)$ for some $x \in \Var$ to denote $\sigma(x)$, where $s = \tup{\sigma, \omega} \concat s'$.

\subsection{Contract $\CtPcSpecInterf{\cdot}$}
The contract $\CtPcSpecInterfStep{}{}$ is obtained from the contract $\CtSpecInterfStep{}{}$ by modifying the rule \textsc{Step} as follows (changes are highlighted in blue):
\begin{mathpar}
	\inferrule[Step]
	{
	p(\sigma(\pc))\neq \pjz{x}{\lbl}\\
	\conf \CtSeqInterfStep{\tau}{} \conf'\\
	\highlightBox{\tau' =
	{
		\begin{cases}
		\tau & \text{if}\ \omega = \infty \vee (\tau	 \neq \loadObs{n} \wedge \tau \neq \storeObs{n}) \\
		\emptysequence & \text{otherwise}
		\end{cases}
	}
	}
	}
	{
		\tup{\sigma, \omega+1} \cdot s \CtSpecInterfStep{ \highlightBox{ \tau'}}{} \tup{\sigma', \omega}\cdot s
	}
	\end{mathpar}
Then,  $\CtPcSpecInterf{p}(\sigma)$, where $\sigma$ is an initial \archstate{}, is the trace of contract observations $\tau$ such that $\tup{\sigma, \infty} \CtPcSpecInterfStep{\tau}{}^* \tup{\sigma', \infty}$ and $\sigma'$ is a final \archstate{}.

\subsection{Contract $\ArchSpecInterf{\cdot}$}
The contract $\ArchSpecInterfStep{}{}$ is obtained from the contract $\CtSpecInterfStep{}{}$ by modifying the rule \textsc{Step} as follows (changes are highlighted in blue):
\begin{mathpar}
	\inferrule[Step]
	{
	p(\sigma(\pc))\neq \pjz{x}{\lbl}\\
    \highlightBox{\conf \ArchSeqInterfStep{\tau}{} \conf'}
	}
	{
		\tup{\sigma, \omega+1} \cdot s \CtSpecInterfStep{\tau}{} \tup{\sigma', \omega}\cdot s
	}
\end{mathpar}
Then,  $\ArchSpecInterf{p}(\sigma)$, where $\sigma$ is an initial \archstate{}, is the trace of contract observations $\tau$ such that $\tup{\sigma, \infty} \ArchSpecInterfStep{\tau}{}^* \tup{\sigma', \infty}$ and $\sigma'$ is a final \archstate{}.

\subsection{A lattice of contracts -- proofs}\label{appendix:contracts:lattice}

We start by proving Proposition~\ref{proposition:contract-hni-ordering}.

\contractHniOrdering*

\begin{proof}
Consider a hardware semantics 	$\CtxMuarchSem{\cdot}{}$ and two contracts $\interfSem{\cdot}_{\interfStyle{1}}$, $\interfSem{\cdot}_{\interfStyle{2}}$ such that $\hsni{\interfSem{\cdot}_{\interfStyle{1}}}{\CtxMuarchSem{\cdot}{}}$ and $\interfSem{\cdot}_{\interfStyle{1}} \sqsupseteq \interfSem{\cdot}_{\interfStyle{2}}$ hold.
Let $p$ be a program and $\sigma, \sigma'$ be two arbitrary initial architectural states such that $\interfSem{p}_{\interfStyle{2}}(\sigma) = \interfSem{p}_{\interfStyle{2}}(\sigma')$.
Therefore, we get $\interfSem{p}_{\interfStyle{1}}(\sigma) = \interfSem{p}_{\interfStyle{1}}(\sigma')$ from $\interfSem{\cdot}_{\interfStyle{1}} \sqsupseteq \interfSem{\cdot}_{\interfStyle{2}}$.
Finally, $\CtxMuarchSem{p}{}(\sigma) = \CtxMuarchSem{p}{}(\sigma')$ follows from $\hsni{\interfSem{\cdot}_{\interfStyle{1}}}{\CtxMuarchSem{\cdot}{}}$.
Since $p$, $\sigma, \sigma'$ are arbitrary such that $\interfSem{p}_{\interfStyle{2}}(\sigma) = \interfSem{p}_{\interfStyle{2}}(\sigma')$, we have that $\interfSem{p}_{\interfStyle{2}}(\sigma) = \interfSem{p}_{\interfStyle{2}}(\sigma') \Rightarrow\CtxMuarchSem{p}{}(\sigma) = \CtxMuarchSem{p}{}(\sigma')$, i.e., $\hsni{\interfSem{\cdot}_{\interfStyle{2}}}{\CtxMuarchSem{\cdot}{}}$ holds.
\end{proof}

We now prove the results leading to the lattice in Figure~\ref{figure:lattice-contracts}.

\begin{restatable}{prop}{latticeOrd1}
$\CtSeqInterf{\cdot} \sqsupseteq \ArchSeqInterf{\cdot}$ and 
$\CtSeqInterf{\cdot} \sqsupseteq \CtPcSpecInterf{\cdot}$. 
\end{restatable}

\begin{proof}
\begin{description}
	\item[$\CtSeqInterf{\cdot} \sqsupseteq \ArchSeqInterf{\cdot}$:]
		Consider two \archstate{}s $\sigma,\sigma'$ and a program $p$ such that $\ArchSeqInterf{p}(\sigma) = \ArchSeqInterf{p}(\sigma')$.
		Then, $\CtSeqInterf{p}(\sigma) = \CtSeqInterf{p}(\sigma')$ immediately follows from (1) $\CtSeqInterfStep{}{}$ and $\ArchSeqInterfStep{}{}$ explore the same ``execution'', (2) all non-load observations are produced, identically, under the two contracts, and (3) load observations in $\ArchSeqInterfStep{}{}$ fully determine the corresponding load observation in $\CtSeqInterfStep{}{}$.
		
	\item[$\CtSeqInterf{\cdot} \sqsupseteq \CtPcSpecInterf{\cdot}$:]
		Consider two \archstate{}s $\sigma,\sigma'$ and a program $p$ such that $\CtPcSpecInterf{p}(\sigma) = \CtPcSpecInterf{p}(\sigma')$.
		Then, $\CtSeqInterf{p}(\sigma) = \CtSeqInterf{p}(\sigma')$ immediately follows from (1) $\CtPcSpecInterf{\cdot}{}$ explores also the speculatively executed instructions ignored by  $\CtSeqInterf{\cdot}{}$, (2) the observations produced non-speculatively (i.e., whenever there is only one state on the stack) in $\CtSeqInterf{\cdot}{}$ corresponds to the trace for $\CtSeqInterf{\cdot}{}$. 
\end{description}
\end{proof}

\begin{restatable}{prop}{latticeOrd2}
	$\ArchSeqInterf{\cdot} \sqsupseteq \ArchSpecInterf{\cdot}$. 
\end{restatable}

\begin{proof}
	Consider two \archstate{}s $\sigma,\sigma'$ and a program $p$ such that $\ArchSpecInterf{p}(\sigma) = \ArchSpecInterf{p}(\sigma')$.
	Then, $\ArchSeqInterf{p}(\sigma) = \ArchSeqInterf{p}(\sigma')$ immediately follows from (1) $\ArchSpecInterf{\cdot}{}$ explores also the speculatively executed instructions ignored by  $\ArchSeqInterf{\cdot}{}$, (2) the observations produced non-speculatively (i.e., whenever there is only one state on the stack) in $\ArchSpecInterf{\cdot}{}$  corresponds to the trace for  $\ArchSeqInterf{\cdot}{}$. 
\end{proof}

\begin{restatable}{prop}{latticeOrd3}
	$\CtSpecInterf{\cdot} \sqsupseteq \ArchSpecInterf{\cdot}$. 
\end{restatable}

\begin{proof}
	Consider two \archstate{}s $\sigma,\sigma'$ and a program $p$ such that $\ArchSpecInterf{p}(\sigma) = \ArchSpecInterf{p}(\sigma')$.
	Then, $\CtSpecInterf{p}(\sigma) = \CtSpecInterf{p}(\sigma')$ immediately follows from (1) $\CtSpecInterf{\cdot}{}$ and $\ArchSpecInterf{\cdot}{}$ explore the same ``execution'', (2) all non-load observations are produced, identically, under the two contracts, and (3) load observations in $\ArchSpecInterf{\cdot}{}$ fully determine the corresponding load observation in $\CtSpecInterf{\cdot}{}$.
\end{proof}

\begin{restatable}{prop}{latticeOrd4}
	$\CtPcSpecInterf{\cdot} \sqsupseteq \CtSpecInterf{\cdot}$. 
\end{restatable}

\begin{proof}
	Consider two \archstate{}s $\sigma,\sigma'$ and a program $p$ such that $\CtSpecInterf{p}(\sigma) = \CtSpecInterf{p}(\sigma')$.
	Then, $\CtSpecInterf{p}(\sigma) = \CtSpecInterf{p}(\sigma')$ immediately follows from (1) $\CtSpecInterf{\cdot}{}$ and $\CtPcSpecInterf{\cdot}{}$ explore the same ``execution'', and (2) traces in $\CtPcSpecInterf{\cdot}{}$ can be derived from those in $\CtSpecInterf{\cdot}{}$ by dropping load and store observations when mispredicting.
\end{proof}

\begin{restatable}{prop}{latticeOrd5}
	$\interfSem{\cdot}_{\interfStyle{\top}} \sqsupseteq \CtSeqInterf{\cdot}$. 
\end{restatable}

\begin{proof}
	Consider two \archstate{}s $\sigma,\sigma'$ and a program $p$ such that $\CtSeqInterf{p}(\sigma) = \CtSeqInterf{p}(\sigma')$.
	Then, $\interfSem{p}_{\interfStyle{\top}}(\sigma) = \interfSem{p}_{\interfStyle{\top}}(\sigma')$ immediately follows from $\interfSem{p}_{\interfStyle{\top}}$ not producing any observations.
\end{proof}

\begin{restatable}{prop}{latticeOrd6}
	$\ArchSeqInterf{\cdot} \sqsupseteq \InftyInterf{\cdot}$. 
\end{restatable}

\begin{proof}
	Consider two \archstate{}s $\sigma,\sigma'$ and a program $p$ such that $\InftyInterf{p}(\sigma) = \InftyInterf{p}(\sigma')$.
	Since $\InftyInterf{\cdot}$ discloses the entire \archstate{}, from $\InftyInterf{p}(\sigma) = \InftyInterf{p}(\sigma')$ we have that $\sigma = \sigma'$.
	Then, $\ArchSeqInterf{p}(\sigma) = \ArchSeqInterf{p}(\sigma')$ immediately follows from $\sigma = \sigma'$ and the determinacy of $\ArchSeqInterf{\cdot}$.
\end{proof}

\subsection{Programming against contracts -- proofs}\label{appendix:contracts:programming}

We start by proving Proposition~\ref{thm:sec-prog:end-to-end}, which connects non-interference at contract and  hardware levels.

\secProgEndToEnd*

\begin{proof}
Let $\Prg$ be a program, $\policy$ be a policy, $\interfSem{\cdot}$ be a contract, and $\muarchSem{\cdot}$ be a hardware semantics.
Assume that $\ct{\Prg}{\policy}{\interfSem{\cdot}}$ and $\hsni{\interfSem{\cdot}}{ \muarchSem{\cdot}}$ hold.
Let $\sigma, \sigma'$ be two initial architectural states such that $\sigma \lowequiv \sigma'$.
From $\ct{\Prg}{\policy}{\interfSem{\cdot}}$, we get $\interfSem{p}(\sigma) = \interfSem{p}(\sigma')$.
From $\hsni{\interfSem{\cdot}}{ \muarchSem{\cdot}}$, we get $\muarchSem{p}(\sigma) = \muarchSem{p}(\sigma')$.
Since $\sigma, \sigma'$ are arbitrary states such that $\sigma \lowequiv \sigma'$, we have that $\sigma \lowequiv \sigma' \Rightarrow \muarchSem{p}(\sigma) = \muarchSem{p}(\sigma')$, that is, $\ct{\Prg}{\policy}{\muarchSem{\cdot}}$.
\end{proof}

Next, we prove Proposition~\ref{prop:sandbox}, which shows how wSNI bridges the gap between vanilla and general sandboxing.

\propSandbox*

\begin{proof}
Let $\Prg$ be a program, $\policy$ be a policy, and	$\interfSem{\cdot}$ be a contract.
Assume that $\Prg$ is vanilla-sandboxed w.r.t. $\policy$ (i.e., $\ct{p}{\policy}{\ArchSeqInterf{\cdot}}$) and wSNI w.r.t. $\interfSem{\cdot}$.
Let $\sigma, \sigma'$ be two arbitrary initial architectural states such that $\sigma \lowequiv \sigma'$.
From $\ct{p}{\policy}{\ArchSeqInterf{\cdot}}$, we get that $\ArchSeqInterf{p}(\sigma) = \ArchSeqInterf{p'}(\sigma')$.
From $\Prg$ is wSNI w.r.t. $\interfSem{\cdot}$, we get that $\interfSem{\Prg}(\sigma) = \interfSem{\Prg}(\sigma')$.
Since $\sigma, \sigma'$ are two arbitrary states such that $\sigma \lowequiv \sigma'$, we have $\sigma \lowequiv \sigma' \Rightarrow \interfSem{\Prg}(\sigma) = \interfSem{\Prg}(\sigma')$, that is, $\ct{p}{\policy}{\interfSem{\cdot}}$.
Since $\ct{p}{\policy}{\ArchSeqInterf{\cdot}}$ and $\ct{p}{\policy}{\interfSem{\cdot}}$ hold, $\Prg$ is generally-sandboxed w.r.t. $\policy$ and $\interfSem{\cdot}$. 
\end{proof}

We conclude by proving Proposition~\ref{prop:ct}, which shows how SNI bridges the gap between vanilla and general constant-time.

\propCt*

\begin{proof}
Let $\Prg$ be a program, $\policy$ be a policy, and	$\interfSem{\cdot}$ be a contract.
Assume that $\Prg$ is vanilla-constant-time w.r.t. $\policy$ (i.e., $\ct{p}{\policy}{\CtSeqInterf{\cdot}}$) and SNI w.r.t.  $\policy$ and $\interfSem{\cdot}$.
Let $\sigma, \sigma'$ be two arbitrary initial architectural states such that $\sigma \lowequiv \sigma'$.
From $\ct{p}{\policy}{\CtSeqInterf{\cdot}}$, we get $\CtSeqInterf{p}(\sigma) = \CtSeqInterf{p'}(\sigma')$.
From $\Prg$ is  SNI w.r.t.  $\policy$ and $\interfSem{\cdot}$, we then get $\interfSem{\Prg}(\sigma) = \interfSem{\Prg}(\sigma')$.
Since $\sigma, \sigma'$ are two arbitrary states such that $\sigma \lowequiv \sigma'$, we have $\sigma \lowequiv \sigma' \Rightarrow \interfSem{\Prg}(\sigma) = \interfSem{\Prg}(\sigma')$, that is, $\ct{p}{\policy}{\interfSem{\cdot}}$ (i.e., $\Prg$ is generally-constant-time w.r.t. $\policy$ and $\interfSem{\cdot}$).
\end{proof}

\newpage
\section{General processor -- Proof of Theorem~\ref{theorem:hni:all}}\label{appendix:proofs:general}

In this section, we prove that all instances of our hardware semantics satisfy the $\CtSpecInterf{\cdot}$ contract.

In the following we assume given an arbitrary cache $C$, branch predictor $Bp$, and scheduler $S$.
Hence, our proof holds for arbitrary caches $C$, branch predictors $Bp$, and schedulers $S$.

The semantics $\muarchSem{p}$ for a program $p$ is defined as follows:
$\muarchSem{p}(\tup{m,a})$ is $\tup{\emptysequence, \CacheState_0,\BpState_0, \SchedState_0}$ $ \cdot \tup{\BufProject{\buf_1}, \CacheState_1,\BpState_1, \SchedState_1}$ $  \cdot \tup{\BufProject{\buf_2}, \CacheState_2,\BpState_2, \SchedState_2}$ $  \cdot \ldots \cdot $ $ \tup{\BufProject{\buf_n}, \CacheState_n,\BpState_n, \SchedState_n}$ where $	\tup{m,a,\emptysequence, \CacheState_0,\BpState_0, \SchedState_0 } \muarchStep{}{} \tup{m_1,a_1,\buf_1, \CacheState_1,\BpState_1, \SchedState_1}  \muarchStep{}{} \tup{m_2,a_2,\buf_2, \CacheState_2,\BpState_2, \SchedState_2} \muarchStep{}{}$ $\ldots$ $\muarchStep{}{}  \tup{m_n,a_n,\buf_n, \CacheState_n,\BpState_n, \SchedState_n}$ is the complete hardware run obtained starting from $\tup{m,a}$ and terminating in $\tup{m_n,a_n,\buf_n, \CacheState_n,\BpState_n, \SchedState_n}$, which is a final hardware state. Otherwise, $\muarchSem{p}(\tup{m,a})$ is undefined.

\textbf{ASSUMPTION: } In the following, we assume that $\wInterf > \wMuarch + 1$.

\allCtSpec*

\begin{proof}
	Let $p$ be an arbitrary well-formed program.
	Moreover, let $\sigma = \tup{m,a},\sigma' = \tup{m',a'}$ be two arbitrary initial configurations.
	There are two cases:
	\begin{compactitem}
	\item[$\CtSpecInterf{\Prg}(\sigma) \neq \CtSpecInterf{\Prg}(\sigma')$:] Then, 	$\CtSpecInterf{\Prg}(\sigma) = \CtSpecInterf{\Prg}(\sigma') \Rightarrow \muarchSem{\Prg}(\sigma) = \muarchSem{\Prg}(\sigma')$ trivially holds.
	\item[$\CtSpecInterf{\Prg}(\sigma) = \CtSpecInterf{\Prg}(\sigma')$:]
		By unrolling the notion of $\CtSpecInterf{\Prg}(\sigma)$ (together with all changes to the program counter $\pc$ being visible on traces), we obtain that there are runs $\crun:= s$ $\CtSpecInterf{o_1}{}$ $s_1$ $\CtSpecInterf{o_2}{}$ $\ldots$  $\CtSpecInterf{o_{n-1}}{}$  $s_n$ and $\crunp:= s'\CtSpecInterf{o_1'}{} s_1' \CtSpecInterf{o_2'}{} \ldots  \CtSpecInterf{o_{n-1}'}{}  s_n'$ such that $o_i = o_i'$ for all $0 < i < n$.
		Since $\wInterf > \wMuarch + 1$, by applying Lemma~\ref{lemma:vanilla:main-lemma}, we immediately get that $\muarchSem{\Prg}(\sigma) = \muarchSem{\Prg}(\sigma')$ (because either both run terminate producing indistinguishable sequences of processor configurations or they both get stuck).
		Therefore, $\CtSpecInterf{\Prg}(\sigma) = \CtSpecInterf{\Prg}(\sigma') \Rightarrow \muarchSem{\Prg}(\sigma) = \muarchSem{\Prg}(\sigma')$ holds.
	\end{compactitem}
	Hence, $\CtSpecInterf{\Prg}(\sigma) = \CtSpecInterf{\Prg}(\sigma') \Rightarrow \muarchSem{\Prg}(\sigma) = \muarchSem{\Prg}(\sigma')$ holds for all programs $p$ and initial configurations $\sigma,\sigma'$.
	Therefore, $\hsni{\CtSpecInterf{\cdot}}{\muarchSem{\cdot}}$ holds.	
\end{proof}

\subsection{Preliminary definitions}

\begin{definition}[Deep-update]
    Let $p$ be a program,  $\tup{m,a}$ be an \archstate{}, and $\buf$ be a buffer.
    The \emph{deep-update of $\tup{m,a}$ given $\buf$} is defined as follows:
        \begin{align*}
        \update{\tup{m,a}}{\emptysequence} &:= \tup{m,a} \\
        \update{\tup{m,a}}{ \tagged{\passign{x}{e}}{T}}  &:= 
             \tup{m, a[x \mapsto \exprEval{e}{a}]} 
        \\
        \update{\tup{m,a}}{ \tagged{\pmarkedassign{x}{e}}{T}}  &:= \tup{m, a[x \mapsto \exprEval{e}{a}]}\\
        \update{\tup{m,a}}{ \tagged{\pload{x}{e}}{T}} &:= \tup{m,a[x \mapsto m(\exprEval{e}{a})] } \\
        \update{\tup{m,a}}{ \tagged{\pstore{x}{e}}{T}} &:= \tup{m[\exprEval{e}{a} \mapsto a(x)],a}\\
        \update{\tup{m,a}}{\tagged{\pskip{}}{T}} &:= \tup{m,a}\\
        \update{\tup{m,a}}{\tagged{\pbarrier{}}{T}} &:= \tup{m,a}\\
        \update{\tup{m,a}}{(\tagged{i}{T} \concat \buf)} &:= 				\update{ (\update{\tup{m,a}}{\tagged{i}{T}}) }{ \buf }
        \end{align*}
\end{definition}

\begin{definition}[Well-formed buffers]
    A reorder buffer $\buf$ is \emph{well-formed for  a well-formed program $p$ and an \archstate{}  $\tup{m,a}$}, written $\wellformed{\buf,\tup{m,a}}$, if the following conditions hold:
    \begin{align*}
    	\wellformed{\emptysequence,\tup{m,a}} & \\
        \wellformed{ \tagged{\passign{\pc}{e}}{T} \concat \buf, \tup{m,a} } & \text{ if } \wellformed{\buf, \update{\tup{m,a}}{\tagged{\passign{\pc}{e}}{T}}} \wedge p(a(\pc)) \sim_{\tup{m,a}} \tagged{\passign{\pc}{e}}{\notags} \\
        \wellformed{\tagged{\passign{\pc}{\ell}}{\ell_0} \concat \buf, \tup{m,a}} & \text{ if } \wellformed{\buf, \update{\tup{m,a}}{\tagged{\passign{\pc}{\ell}}{\ell_0}}} \wedge  \ell_0 \in \Val \wedge p(\ell_0) = \pjz{x}{\ell'} \wedge \\& \quad \ell \in \{\ell', \ell_0+1\} \wedge p(a(\pc)) \sim_{\tup{m,a}} \tagged{\passign{\pc}{\ell}}{\ell_0} \wedge \ell_0 = a(\pc)\\
    	\wellformed{ \tagged{\pmarkedassign{\pc}{\ell}}{\notags} \concat \buf , \tup{m,a} } & \text{ if } \wellformed{\buf,\update{\tagged{\pmarkedassign{\pc}{\ell}}{\notags}}{\tup{m,a}}} \wedge \ell \in \Val \\
        \wellformed{ \tagged{i}{\notags} \concat \tagged{\pmarkedassign{\pc}{\ell}}{\notags}  \concat \buf , \tup{m,a} } & \text{ if }  
        \wellformed{\buf, \update{\tup{m,a}}{(\tagged{i}{\notags} \concat \tagged{\pmarkedassign{\pc}{\ell}}{\notags} )}} \wedge 
        \ell \in \Val \wedge (\forall e.\ i \neq \passign{\pc}{e}) \wedge  \\
        & \quad (\forall x,e.\ i \neq \pmarkedassign{x}{e}) \wedge (\forall x,e.\ i \neq \pload{\pc}{e}) \wedge p(a(\pc)) \sim_{\tup{m,a}} \tagged{i}{\notags} \wedge \\ & \quad \ell = a(\pc)+1
    \end{align*}
    where the instruction-compatibility relation $\sim_{\tup{m,a}}$ is defined as follows:
    \begin{align*}
        \pskip &\sim_{\tup{m,a}} \tagged{\pskip}{\notags} \\ \allowdisplaybreaks
        \pbarrier &\sim_{\tup{m,a}} \tagged{\pbarrier}{\notags} \\ \allowdisplaybreaks
        \passign{x}{e} &\sim_{\tup{m,a}} \tagged{\passign{x}{e}}{\notags} \\ \allowdisplaybreaks
        \passign{x}{e} &\sim_{\tup{m,a}} \tagged{\passign{x}{v}}{\notags} \text{ if } v = \exprEval{e}{a} \\ \allowdisplaybreaks
        \pload{x}{e} &\sim_{\tup{m,a}} \tagged{\pload{x}{e}}{\notags} \\ \allowdisplaybreaks
        \pload{x}{e} &\sim_{\tup{m,a}} \tagged{\passign{x}{v}}{\notags} \text{ if } v = m(\exprEval{e}{a})\\ \allowdisplaybreaks
        \pstore{x}{e} &\sim_{\tup{m,a}} \tagged{\pstore{x}{e}}{\notags} \\ \allowdisplaybreaks
        \pstore{x}{e} &\sim_{\tup{m,a}} \tagged{\pstore{v}{n}}{\notags} \text{ if } v = a(x) \wedge n = \exprEval{e}{a} \\ \allowdisplaybreaks
        \pjmp{e} &\sim_{\tup{m,a}} \tagged{\passign{\pc}{e}}{\notags} \\ \allowdisplaybreaks
        \pjmp{e} &\sim_{\tup{m,a}} \tagged{\passign{\pc}{v}}{\notags} \text{ if } v = \exprEval{e}{a} \\ \allowdisplaybreaks
        \pjz{x}{\ell} &\sim_{\tup{m,a}} \tagged{\passign{\pc}{\ell}}{a(\pc)} \\ \allowdisplaybreaks
        \pjz{x}{\ell} &\sim_{\tup{m,a}} \tagged{\passign{\pc}{a(\pc)+1}}{a(\pc)} \\ \allowdisplaybreaks
        \pjz{x}{\ell} &\sim_{\tup{m,a}} \tagged{\passign{\pc}{a(\pc)+1}}{\notags} \wedge a(x) \neq 0\\ \allowdisplaybreaks
        \pjz{x}{\ell} &\sim_{\tup{m,a}} \tagged{\passign{\pc}{\ell}}{\notags} \wedge a(x) = 0
    \end{align*}
    \end{definition}

    Lemma~\ref{lemma:vanilla:buffers-well-formedness} states that all reorder buffers occurring in hardware runs are well-formed.

\begin{lemma}[Reorder buffers are well-formed]\label{lemma:vanilla:buffers-well-formedness}
	Let $p$ be a well-formed program, $\sigma_0 = \tup{m,a}$ be an initial \archstate, $\CacheState_0$ be the initial cache state, $\BpState_0$ be the initial branch predictor state, and $\SchedState_0$ be the initial scheduler state.
	For all hardware runs  $\hrun := C_0 \muarchStep{}{} C_1 \muarchStep{}{} C_2 \muarchStep{}{} \ldots \muarchStep{}{} C_k$ and all $0 \leq i \leq k$, then $\wellformed{\buf_i,\tup{m_i,a_i}}$, where $C_0 = \tup{m,a,\emptysequence, \CacheState_0, \BpState_0, \SchedState_0}$ and $C_i = \tup{m_i, a_i,\buf_i, \CacheState_i, \BpState_i, \SchedState_i}$.
\end{lemma}

\begin{proof}
 The lemma follows by (1) induction on $i$, and (2) inspection of the rules defining $\muarchStep{}{}$.
\end{proof}

    \begin{definition}[Prefixes of buffers]\label{def:vanilla:ct-spec:prefixes}
        The prefixes of a well-formed buffer $\buf$ are defined as follows:
        \begin{align*}
            \prefixes{\emptysequence} &= \{ \emptysequence \}\\
            \prefixes{  \tagged{\passign{\pc}{e}}{\notags} \concat \buf } &= 
            \{ \emptysequence \} \cup \{  \tagged{\passign{\pc}{e}}{T} \concat \buf' \mid \buf' \in \prefixes{\buf } \} \\
            \prefixes{  \tagged{\passign{\pc}{\ell}}{\ell_0} \concat \buf } &= 
            \{ \emptysequence \} \cup \{  \tagged{\passign{\pc}{\ell}}{\ell_0} \concat \buf' \mid \buf' \in \prefixes{\buf }  \} \\
            \prefixes{  \tagged{i}{\notags} \concat \tagged{\pmarkedassign{\pc}{\ell}}{\notags} \concat \buf } &= \{ \emptysequence \} \cup \{   \tagged{i}{\notags} \concat \tagged{\pmarkedassign{\pc}{\ell}}{\notags} \concat \buf' \mid \buf' \in \prefixes{\buf } \} \\
            \prefixes{   \tagged{\pmarkedassign{\pc}{\ell}}{\notags} \concat \buf } &= \{  \tagged{\pmarkedassign{\pc}{\ell}}{\notags} \} \cup \{ \tagged{\pmarkedassign{\pc}{\ell}}{\notags} \concat \buf' \mid \buf' \in \prefixes{\buf } \} 
        \end{align*}
    \end{definition}

\begin{definition}[Equivalence between contract states and hardware configurations]
    We say that a hardware configuration $C = \tup{m,a,\CacheState, \BpState, \SchedState}$ and a contract state $s =  \tup{\tup{m',a'},\omega} \concat s'$ are \emph{$i$-equivalent}, for an integer $0 \leq i \leq |\buf|$, written $C \bufEquiv{i} \sigma$, iff $\update{\tup{m,a}}{\buf[0..i]} = \tup{m',a'}$ and $\nrMispred{\tup{m,a}}{\buf[0..i]} = 0 \leftrightarrow \omega = \infty$.
\end{definition}

\subsection{Mapping lemma}

\begin{definition}[$\crun-\hrun$ mapping]\label{def:vanilla:mapping}
	Let $p$ be a well-formed program, $\sigma_0 = \tup{m,a}$ be an initial \archstate, $\CacheState_0$ be the initial cache state, $\BpState_0$ be the initial branch predictor state, and $\SchedState_0$ be the initial scheduler state.
	Furthermore, let:
	\begin{compactitem}
		\item $\crun := s_0 \CtSpecInterfStep{o_1}{} s_1 \CtSpecInterfStep{o_2}{} s_2 \CtSpecInterfStep{o_3}{} \ldots \CtSpecInterfStep{o_{n-1}}{} s_n$ be the longest $\interfStyle{spec-ct}$ contract run obtained starting from $s_0 := \tup{\sigma_0, \infty}$.
		\item $\hrun := C_0 \muarchStep{}{} C_1 \muarchStep{}{} C_2 \muarchStep{}{} \ldots \muarchStep{}{} C_k$ be the longest  hardware run obtained starting from $C_0 = \tup{m,a,\emptysequence, \CacheState_0, \BpState_0, \SchedState_0}$.
		\item $\hrun(i)$ be the $i$-th hardware configuration in $\hrun$.
		\item $\crun(i)$ be the $i$-th contract configuration in $\crun$ (note that $\crun(i) = s_n$ for all $i>n$).
	\end{compactitem}
	The \emph{$\crun-\hrun$ mapping}, which maps hardware configurations in $\hrun$ to contract configurations in $\crun$, is defined as follows:
	\begin{align*}
		\map{\crun}{\hrun}{0} &:= \{0 \mapsto 0\} \\ 
		\map{\crun}{\hrun}{i} &:= {
			\begin{cases}
			\map{\crun}{\hrun}{i-1} & \text{if } \SchedNext(\hrun(i-1)) = \fetch{} \wedge ln(\hrun(i-1)) = ln(\hrun(i))\\
			fetch_{\crun,\hrun}(i) & \text{if } \SchedNext(\hrun(i-1)) = \fetch{} \wedge ln(\hrun(i-1)) < ln(\hrun(i)) \\
            \map{\crun}{\hrun}{i-1} & \text{if } \SchedNext(\hrun(i-1)) = \execute{j} \wedge  \neg isRb_{\hrun}(i,j) \\ 
            rollback_{\crun,\hrun}(i,j) & \text{if } \SchedNext(\hrun(i-1)) = \execute{j} \wedge   isRb_{\hrun}(i,j)\\ 
			shift(\map{\crun}{\hrun}{i -1}) & \text{if } \SchedNext(\hrun(i-1)) = \retire{} \\
			\end{cases}
		}\\ 
		fetch_{\crun,\hrun}(i) &=
				\map{\crun}{\hrun}{i-1}[ln(\hrun(i-1))+2 \mapsto \map{\crun}{\hrun}{i-1}(ln(\hrun(i-1)))+1]\\
				& \qquad \text{if }
							p(\mathit{lstPc}(\hrun(i -1))) \neq \pjz{x}{\lbl} \wedge 
							p(\mathit{lstPc}(\hrun(i -1))) \neq \pjmp{e} 
							\\
		fetch_{\crun,\hrun}(i) &=
				\map{\crun}{\hrun}{i-1}[ln(\hrun({i-1}))+1 \mapsto \map{\crun}{\hrun}{i-1}(ln(\hrun(i-1)))+1]\\
				& \qquad \text{if }
							p(\mathit{lstPc}(\hrun(i -1))) = \pjmp{e}
                            \\
        fetch_{\crun,\hrun}(i) &=
        \map{\crun}{\hrun}{i-1}[ln(\hrun({i-1}))+1 \mapsto \map{\crun}{\hrun}{i-1}(ln(\hrun(i-1)))+1]\\
        & \qquad \text{if }
                    p(\mathit{lstPc}(\hrun(i -1))) = \pjz{x}{\ell} \wedge \BpPredict(\hrun(i-1)) \neq \ell \wedge \crun(\map{\crun}{\hrun}{i-1}(ln(\hrun({i-1}))))(x) = 0
                    \\
        fetch_{\crun,\hrun}(i) &=
        \map{\crun}{\hrun}{i-1}[ln(\hrun({i-1}))+1 \mapsto \map{\crun}{\hrun}{i-1}(ln(\hrun(i-1)))+1]\\
        & \qquad \text{if }
                    p(\mathit{lstPc}(\hrun(i -1))) = \pjz{x}{\ell} \wedge \BpPredict(\hrun(i-1)) = \ell \wedge \crun(\map{\crun}{\hrun}{i-1}(ln(\hrun({i-1}))))(x) \neq 0
                    \\
        fetch_{\crun,\hrun}(i) &=
        \map{\crun}{\hrun}{i-1}[ln(\hrun({i-1}))+1 \mapsto rb_{\crun,\hrun}(\map{\crun}{\hrun}{i-1}(ln(\hrun({i-1}))) ]\\
        & \qquad \text{if }
                    p(\mathit{lstPc}(\hrun(i -1))) = \pjz{x}{\ell} \wedge \BpPredict(\hrun(i-1)) \neq \ell \wedge \crun(\map{\crun}{\hrun}{i-1}(ln(\hrun({i-1}))))(x) \neq 0
                    \\
        fetch_{\crun,\hrun}(i) &=
        \map{\crun}{\hrun}{i-1}[ln(\hrun({i-1}))+1 \mapsto rb_{\crun}(\map{\crun}{\hrun}{i-1}(ln(\hrun({i-1}))) ]\\
        & \qquad \text{if }
                    p(\mathit{lstPc}(\hrun(i -1))) = \pjz{x}{\ell} \wedge \BpPredict(\hrun(i-1)) = \ell \wedge \crun(\map{\crun}{\hrun}{i-1}(ln(\hrun({i-1}))))(x) = 0\\
        rollback_{\crun,\hrun}(i,j) &= clip( \map{\crun}{\hrun}{i-1} , j  )[j \mapsto rb_{\crun}( \map{\crun}{\hrun}{i-1}(j-1) ) ]
        \end{align*}
        where:
        \begin{align*}
        rb_{\crun}(i) &:= min(\{ i' \in \Nat \mid i' > i \wedge |\crun(i)| = |\crun(i')| \})\\
		ln(\tup{m,a,\buf,\CacheState,\BpState,\SchedState}) &= |\buf|\\
        \SchedNext(\tup{m,a,\buf,\CacheState,\BpState,\SchedState}) &= \SchedNext(\SchedState)\\
        \BpPredict(\tup{m,a,\buf,\CacheState,\BpState,\SchedState}) &= \BpPredict(\BpState)\\
        \elt{\tup{m,a,\buf,\CacheState,\BpState,\SchedState}}{j} &= \elt{\buf}{j}\\
        (\tup{\sigma, \omega} \concat s)(x) &= \sigma(x) \\
		shift(map) &= \lambda i \in \Nat.\ map(i +1 )\\
        \mathit{lstPc}(\tup{m,a,\buf,\CacheState,\BpState,\SchedState}) &= (\update{\tup{m,a}}{\buf})(\pc)\\
        clip(map, n) &= \lambda i \in \Nat.\
            {
            \begin{cases}
                map(n) & \text{if }\ i < n \\
                \bot & \text{otherwise}
            \end{cases}
            }\\
        isRb_{\hrun}(i,j) &= {
            \begin{cases}
            \top & \elt{\hrun(i)}{j} = \tagged{\passign{\pc}{\ell'}}{\notags} \wedge \elt{\hrun(i-1)}{j} = \tagged{\passign{\pc}{\ell}}{\ell_0} \wedge \ell_0 \neq \notags \wedge \ell \neq \ell'\\
            \bot & \text{otherwise}
            \end{cases}
        }
	\end{align*}
    \end{definition}

Additionally, we define the following auxiliary functions and sets:

\newcommand{\nrInstr}[1]{\nrPc{#1}}

\begin{definition}
Let $p$ be a well-formed program. 
The function $\nrMispred{\sigma}{\buf}$, where $\sigma$ is an \archstate{} and $\buf$ is a buffer, is defined as follows:
\begin{align*}
    \nrMispred{\sigma}{\emptysequence} &:= 0 \\
    \nrMispred{\sigma}{\tagged{i}{T} \concat \buf } &= {
        \begin{cases}
            \nrMispred{\update{\sigma}{\tagged{i}{T}}}{ \buf } + 1 & \text{if } i = \passign{\pc}{\ell} \wedge T \neq \notags \wedge p(T) = \pjz{x}{\ell'} \wedge \sigma(x) = 0 \wedge \ell \neq \ell' \\
            \nrMispred{\update{\sigma}{\tagged{i}{T}}}{ \buf } + 1 & \text{if } i = \passign{\pc}{\ell} \wedge T \neq \notags \wedge p(T) = \pjz{x}{\ell'} \wedge \sigma(x)\in \Val \setminus\{ 0\} \wedge \ell \neq T +1 \\
            \nrMispred{\update{\sigma}{\tagged{i}{T}}}{ \buf } & \text{otherwise} 
        \end{cases}
    }
\end{align*}
The function $\nrInstr{\buf}$, where $\buf$ is a buffer, is defined as follows:
\begin{align*}
    \nrInstr{\emptysequence} &:= 0 \\
    \nrInstr{\tagged{i}{T} \concat \buf} &= {
    \begin{cases}
        \nrPc{\buf} +1 & \text{if}\ i = \passign{\pc}{e} \\
        \nrPc{\buf} +1 & \text{if}\ i = \pmarkedassign{\pc}{e} \\
        \nrPc{\buf} & \text{otherwise}
    \end{cases}
    } 
\end{align*}
The functions $\headConf{s}$ and $\headWindow{s}$ predicates, where $s$ is a contract configuration, are defined as follows:
\begin{align*}
    \headConf{\tup{\sigma, \omega} \concat s} &:= \sigma \\
    \headWindow{\tup{\sigma, \omega} \concat s} &:= \omega
\end{align*}
We lift both predicates to sets of configurations $SS$ as standard, e.g., $\headWindow{SS} = \{ \headWindow{ss} \mid  ss\in SS \}$.\\
The functions $\min( map, \crun )$ and $\max( map, \crun )$, where $map$ is a (possibly partial) function $\Nat \to \Nat$ and $\crun$ is a contract run, are as follows:
\begin{align*}
    \headWindow{map, \crun} &:= \{ \headWindow{\crun(map(n))} \mid n \in \Nat \wedge map(n) \neq \bot \}\\
    \minOf{ map, \crun }&:= \min( \headWindow{map, \crun} ) \\ 
    \maxOf{ map, \crun } &:= \max( \headWindow{map, \crun} )  
\end{align*}
Finally, the set $\nextPairs{map}$, where $map$ is a possibly partial function $\Nat \to \Nat$, is defined as follows:
\begin{align*}
    \nextPairs{map} &:= \{ \tup{j,j'} \in \Nat^2 \mid map(j) \neq \bot \wedge map(j') \neq \bot \wedge j < j' \wedge \neg \exists j''.\ ( map(j'') \neq \bot \wedge j < j'' < j' )\}
\end{align*}
\end{definition}

Lemma~\ref{lemma:vanilla:contract-rollback} states the relation between rolled back speculative transactions and number of stack elements in contract's configurations.

\begin{lemma}\label{lemma:vanilla:contract-rollback}
    Let $p$ be a well-formed program and $\mathit{s}_1 = \tup{\sigma_1, \omega_1} \concat \mathit{s}$.
    If $\mathit{s}_1 \CtSpecInterfStep{o_1}{} \mathit{s}_2 \CtSpecInterfStep{o_2}{} \ldots \mathit{s}_{n-1} \CtSpecInterfStep{o_{n-1}}{} \mathit{s}_n$,
    $|\mathit{s}_n| = |\mathit{s}_1|$,
    and for all $1 < j < n$, $|\mathit{s}_j| > |\mathit{s}_1|$,
    then
    $\mathit{s}_2 =  \tup{\sigma_1[\pc \mapsto \ell'],\omega_2} \concat \tup{\sigma_1[\pc \mapsto \ell],\omega_1-1} \concat \mathit{s}$ and $\mathit{s}_n = \tup{\sigma_1[\pc \mapsto \ell],\omega_1-1} \concat \mathit{s}$,
    where $p(\sigma_1(\pc)) = \pjz{x}{\ell_0}$ and $\ell = \sigma(\pc) + 1$ if $\sigma_1(x) \neq 0$ and $\ell = \ell_0$ if $\sigma_1(x) = 0$, $\ell' \in \{\ell_0,\sigma_1(\pc)+1\} \setminus \{\ell\}$, and $\omega_2 = \wInterf$ if $\omega_1 = \infty$ and $\omega_2 = \omega_1 -1 $ otherwise.
\end{lemma}
    
\begin{proof}
    The lemma follows from (1) the rule \textsc{Branch} in $\CtSpecInterf{\cdot}$ is the only one increasing the number of configurations in the stack, (2) the rule \textsc{Branch} in $\CtSpecInterf{\cdot}$ ensures that $\mathit{s}_2 := \tup{\sigma', \omega'} \concat \tup{\sigma_1[\pc \mapsto \ell],\omega_1-1} \concat \mathit{s}$, and (3) all rules modify only the topmost element in the stack.
    A more detailed proof of this result (albeit using a slightly different notation and terminology) can be found in \cite{spectector2020}'s technical report.
\end{proof}

We are now ready to prove Lemma~\ref{lemma:vanilla:mapping-is-correct}, the  lemma showing the correctness of the $\crun-\hrun$ mapping.

\begin{lemma}[Correctness of $\crun-\hrun$ mapping]\label{lemma:vanilla:mapping-is-correct}
    Let $p$ be a well-formed program, $\sigma_0 = \tup{m,a}$ be an initial \archstate, $\CacheState_0$ be the initial cache state, $\BpState_0$ be the initial branch predictor state, and $\SchedState_0$ be the initial scheduler state.
    Furthermore, let:
    \begin{compactitem}
        \item $\crun := s_0 \CtSpecInterfStep{o_1}{} s_1 \CtSpecInterfStep{o_2}{} s_2 \CtSpecInterfStep{o_3}{} \ldots \CtSpecInterfStep{o_{n-1}}{} s_n$ be the longest $\interfStyle{spec}$ contract run obtained starting from $s_0 = \tup{\sigma_0,\infty}$.
        \item $\hrun := C_0 \muarchStep{}{}  C_1 \muarchStep{}{} C_2 \muarchStep{}{} \ldots \muarchStep{}{} C_k$ be the longest  hardware run obtained starting from $C_0 = \tup{m,a,\emptysequence, \CacheState_0, \BpState_0, \SchedState_0}$.
        \item $\hrun(i)$ be the $i$-th hardware configuration in $\hrun$.
        \item $\crun(i)$ be the $i$-th contract configuration in $\crun$ (note that $\crun(i) = s_n$ for all $i>n$).
        \item $ \map{\crun}{\hrun}{\cdot}$ be the mapping from Definition~\ref{def:vanilla:mapping}.
    \end{compactitem}
    If $\wInterf > \wMuarch + 1$, the following conditions hold:
    \begin{compactenum}[(1)]
    \item $C_0$ is an initial hardware configuration.
    \item $C_k$ is a final hardware configuration or there is no $C_{k'}$ such that $C_k \muarchStep{}{} C_{k'}$.
    \item for all $0 \leq i \leq k$, given $C_i = \tup{m_i,a_i, \buf_i, \CacheState_i, \BpState_i, \SchedState_i}$ the following conditions hold:
        \begin{compactenum}[(a)]
            \item for all $\buf \in \prefixes{\buf_i}$,  (i) $\update{\tup{m_i,a_i}}{\buf} = \headConf{\crun( \map{\crun}{\hrun}{i}(|\buf|) )}$, (ii) $\nrMispred{\tup{m_i,a_i}}{\buf}+1 \geq |\crun( \map{\crun}{\hrun}{i}(|\buf|) )|$, (iii) $\nrMispred{\tup{m_i,a_i}}{\buf} = 0 \leftrightarrow \headWindow{\crun( \map{\crun}{\hrun}{i}(|\buf|) )} = \infty$, and (iv) $\headWindow{\crun( \map{\crun}{\hrun}{i}(|\buf|) )} > 0$.
            \item for all $\tup{j,j'} \in \nextPairs{\map{\crun}{\hrun}{i}}$,
            $\headWindow{ \crun( \map{\crun}{\hrun}{i}(j) ) } = \headWindow{ \crun( \map{\crun}{\hrun}{i}(j') )  } = \infty$ or
            $\headWindow{ \crun( \map{\crun}{\hrun}{i}(j) ) } = \infty$ and $\headWindow{ \crun( \map{\crun}{\hrun}{i}(j') ) } = \wInterf$ or
            $\headWindow{ \crun( \map{\crun}{\hrun}{i}(j) ) } \neq \infty \wedge \headWindow{ \crun( \map{\crun}{\hrun}{i}(j') ) } = \headWindow{ \crun( \map{\crun}{\hrun}{i}(j) ) } - 1$.
        \end{compactenum}
    \end{compactenum}
\end{lemma}

\begin{proof}
Let $p$ be a well-formed program, $\sigma_0 = \tup{m,a}$ be an initial \archstate, $\CacheState_0$ be the initial cache state, $\BpState_0$ be the initial branch predictor state, and $\SchedState_0$ be the initial scheduler state.
Furthermore, let:
\begin{compactitem}
    \item $\crun := s_0 \CtSpecInterfStep{o_1}{} s_1 \CtSpecInterfStep{o_2}{} s_2 \CtSpecInterfStep{o_3}{} \ldots \CtSpecInterfStep{o_{n-1}}{} s_n$ be the longest $\interfStyle{spec}$ contract run obtained starting from $s_0 = \tup{\sigma_0,\infty}$.
    \item $\hrun := C_0 \muarchStep{}{}  C_1 \muarchStep{}{} C_2 \muarchStep{}{} \ldots \muarchStep{}{} C_k$ be the longest  hardware run obtained starting from $C_0 = \tup{m,a,\emptysequence, \CacheState_0, \BpState_0, \SchedState_0}$.
    \item $\hrun(i)$ be the $i$-th hardware configuration in $\hrun$.
    \item $\crun(i)$ be the $i$-th contract configuration in $\crun$ (note that $\crun(i) = s_n$ for all $i > n$).
    \item $ \map{\crun}{\hrun}{\cdot}$ be the mapping from Definition~\ref{def:vanilla:mapping}.
\end{compactitem}
Finally,  we assume that $\wInterf > \wMuarch + 1$ (if this is not the case, then the lemma trivially holds).

Observe that $C_0$ is an initial hardware configuration by construction. 
Observe also that $C_k$ is either a final hardware configuration (for which the semantics cannot further proceed) or the computation is stuck (since $\hrun$ is the longest run by construction).
Therefore, (1) and (2) hold.

We now prove, by induction on $i$, that (3) holds.
\begin{description}
    \item[Base case:]
    For the base case, $i = 0$.
    Observe that $C_0$ is $\tup{m,a,\emptysequence, \CacheState_0, \BpState_0, \SchedState_0 }$  and $s_0 = \tup{\tup{m,a},\infty}$.
    Moreover, $\map{\crun}{\hrun}{0} = \{ 0 \mapsto 0\}$ and  $\prefixes{\emptysequence} = \{\emptysequence\}$.
    We now show that all conditions hold:
    \begin{description}
        \item[(a):] 
        For (i), we need to show that $\update{\tup{m,a}}{\emptysequence} = \headConf{ \crun(\map{\crun}{\hrun}{0}(|\emptysequence|))}$.
        This directly follows from $\update{\tup{m,a}}{\emptysequence} = \tup{m,a}$, $\map{\crun}{\hrun}{0}(|\emptysequence|) = 0$, $\crun(0) =  s_0$, and $\headConf{s_0} = \sigma$.

        For (ii), we  need to show $\nrMispred{\tup{m_0,a_0}}{\emptysequence}+1 \geq | \crun(\map{\crun}{\hrun}{0}(|\emptysequence|)) |$.
        This follows from $\nrMispred{\tup{m_0,a_0}}{\emptysequence} = 0$, $\map{\crun}{\hrun}{0}(|\emptysequence|) = 0$, $\crun(0) =  s_0$, and $|s_0| = 1$.

        For (iii), we need to show that $\nrMispred{\tup{m_0,a_0}}{\emptysequence} = 0 \leftrightarrow \crun(\map{\crun}{\hrun}{0}(|\emptysequence|)) = \infty$.
        This follows from $\nrMispred{\tup{m_0,a_0}}{\emptysequence} = 0$, $\map{\crun}{\hrun}{0}(|\emptysequence|) = 0$, $\crun(0) =  s_0$, and $s_0 = \tup{\tup{m,a},\infty}$.
        
        For (iv), we need to show that  $\headWindow{\crun(\map{\crun}{\hrun}{0}(|\emptysequence|)} > 0$.
        This follows from $\map{\crun}{\hrun}{0}(|\emptysequence|) = 0$, $\crun(0) =  s_0$, and $s_0 = \tup{\tup{m,a},\infty}$.

        \item[(b):]
        This point trivially follows from $\nextPairs{\map{\crun}{\hrun}{0}} = \emptyset$. 
    \end{description}
    This concludes the proof of the base case.

    \item[Induction step:] 
	For the induction step, we assume that the claim holds for all $i' < i$ and we show that it holds for $i$ as well.
	In the following, let $C_i = \tup{m_i,a_i, \buf_i, \CacheState_i, \BpState_i, \SchedState_i}$ and we refer to the induction hypothesis as H.
	Similarly, we write H.3.a to denote, for instance, that fact 3.a holds for the induction hypothesis.
	
	We proceed by case distinction on the directive $\SchedNext(C_{i-1})$ used to derive $C_i$.
    There are three cases:
    \begin{description}
        \item[$\SchedNext(C_{i-1}) = \fetch{}$:]
        We now proceed by case distinction on the applied rule:
		\begin{description}
			\item[Rule \textsc{Fetch-Branch-Hit}:]
            Then, $\buf_i = \tagged{\passign{\pc}{\lbl'}}{\apply{a_{i-1}}{\buf_{i-1}} (\pc)} \concat \buf_{i-1}$, $a_{i} = a_{i-1}$, $m_{i} = m_{i-1}$, $p( \apply{a_{i-1}}{\buf_{i-1}} (\pc)) = \pjz{x}{\lbl}$, and $\lbl' = \BpPredict(\hrun(i-1))$.
            There are four cases:
            \begin{description}
                \item[$\lbl' = \lbl \wedge \update{\tup{m_{i-1},a_{i-1}}}{\buf_{i-1}}(x) = 0 $:] 
                From $\lbl' = \lbl \wedge \update{\tup{m_{i-1},a_{i-1}}}{\buf_{i-1}}(x) = 0 $ and (H.a), we immediately get that $\BpPredict(\hrun(i-1)) = \lbl \wedge \crun(\map{\crun}{\hrun}{i-1}(ln(\hrun(i-1))))(x) = 0$.
                Therefore, $\map{\crun}{\hrun}{i} = \map{\crun}{\hrun}{i-1}[ln(\hrun({i-1}))+1 \mapsto rb_{\crun}(\map{\crun}{\hrun}{i-1}(ln(\hrun({i-1}))) ]$.

                We now show that all conditions hold:
                \begin{description}
                    \item[(a):]
                    Let $\buf$ be an arbitrary prefix in  $\prefixes{\buf_i}$.
                    There are two cases:
                    \begin{description}
                        \item[$\buf \in \prefixes{\buf_{i-1}}$:]
                        Then, $|\buf| \geq |\buf_i|$.
                        We now show that (i)--(iv) hold:
                        \begin{description}
                            \item[(i):]
                            We need to show $\update{\tup{m_i,a_i}}{\buf} =  \headConf{\crun(\map{\crun}{\hrun}{i}(|\buf|))}$.
                            From (H.3.a.i) and $\buf \in \prefixes{\buf_{i-1}}$, we have $\update{\tup{m_{i-1},a_{i-1}}}{\buf} =  \headConf{\crun(\map{\crun}{\hrun}{i-1}(|\buf|))}$.
                            Finally, from this,  $m_{i} = m_{a-1}$, $a_{i} = a_{i-1}$, and $\map{\crun}{\hrun}{i}(\buf) = \map{\crun}{\hrun}{i-1}(\buf)$, we get $\update{\tup{m_i,a_i}}{\buf} =  \headConf{\crun(\map{\crun}{\hrun}{i}(|\buf|))}$.
        
                            \item[(ii):]
                            We need to show $\nrMispred{\tup{m_{i},a_{i}}}{\buf} + 1 \geq |\crun(\map{\crun}{\hrun}{i}(|\buf|))|$.
                            From (H.3.a.ii) and $\buf \in \prefixes{\buf_{i-1}}$,  we have $\nrMispred{\tup{m_{i-1},a_{i-1}}}{\buf} + 1 \geq |\crun(\map{\crun}{\hrun}{i-1}(|\buf|))|$.
                            Finally, from this,  $m_{i} = m_{a-1}$, $a_{i} = a_{i-1}$, and $\map{\crun}{\hrun}{i}(\buf) = \map{\crun}{\hrun}{i-1}(\buf)$, we get $\nrMispred{\tup{m_{i},a_{i}}}{\buf} + 1 \geq |\crun(\map{\crun}{\hrun}{i}(|\buf|))|$.
        
                            \item[(iii):]
                            We need to show $\nrMispred{\tup{m_{i},a_{i}}}{\buf} = 0 \leftrightarrow \headWindow{\crun(\map{\crun}{\hrun}{i}(|\buf|))} = \infty$.
                            From (H.3.a.iii) and $\buf \in \prefixes{\buf_{i-1}}$,  we have $\nrMispred{\tup{m_{i-1},a_{i-1}}}{\buf} = 0 \leftrightarrow \headWindow{\crun(\map{\crun}{\hrun}{i-1}(|\buf|))} = \infty$.
                            Finally, from this,  $m_{i} = m_{a-1}$, $a_{i} = a_{i-1}$, and $\map{\crun}{\hrun}{i}(\buf) = \map{\crun}{\hrun}{i-1}(\buf)$, we get $\nrMispred{\tup{m_{i},a_{i}}}{\buf} = 0 \leftrightarrow \headWindow{\crun(\map{\crun}{\hrun}{i}(|\buf|))} = \infty$.
                            
							\item[(iv):]
	                        We need to show $\headWindow{\crun(\map{\crun}{\hrun}{i}(|\buf|))} > 0$.
	                        From (H.3.a.iv) and $\buf \in \prefixes{\buf_{i-1}}$, we get $\headWindow{\crun(\map{\crun}{\hrun}{i-1}(|\buf|))} > 0$.
	                        From this and $\map{\crun}{\hrun}{i}(\buf) = \map{\crun}{\hrun}{i-1}(\buf)$, we get $\headWindow{\crun(\map{\crun}{\hrun}{i}(|\buf|))} > 0$.
        
                        \end{description} 

                        \item[$\buf \not\in \prefixes{\buf_{i-1}}$:] 
                        Then, $\buf = \buf_i =  \tagged{\passign{\pc}{\lbl'}}{\apply{a_{i-1}}{\buf_{i-1}} (\pc)} \concat \buf_{i-1}$, $|\buf| = |\buf_{i-1}| +1 = ln(\hrun(i-1))+1$, and $\map{\crun}{\hrun}{i}(|\buf|) =  rb_{\crun}(\map{\crun}{\hrun}{i-1}(ln(\hrun({i-1})))$.
                        From (H.3.a.i) and $\buf_{i-1} \in \prefixes{\buf_{i-1}}$, we have $\update{\tup{m_{i-1},a_{i-1}}}{\buf_{i-1}} =  \headConf{\crun(\map{\crun}{\hrun}{i-1}(|\buf_{i-1}|))}$.
                        Observe that:
                        \begin{description}
                            \item[(1):] 
                            From  $\update{\tup{m_{i-1},a_{i-1}}}{\buf_{i-1}} =  \headConf{\crun(\map{\crun}{\hrun}{i-1}(|\buf_{i-1}|))}$, we have, in particular, that $\update{\tup{m_{i-1},a_{i-1}}}{\buf_{i-1}}(\pc) = \headConf{\crun(\map{\crun}{\hrun}{i-1}(|\buf_{i-1}|))}(\pc)$.
                            \item[(2):] 
                            From (H.a.iii) applied to $\buf_{i-1}$ and the well-formedness of buffers, we get that $\maxOf{\map{\crun}{\hrun}{i-1}, \crun} = \infty$ (because the shortest prefix is always associated with window $\infty$).
                            From (H.b) and $\maxOf{\map{\crun}{\hrun}{i-1}, \crun} = \infty$, we have that $\minOf{\map{\crun}{\hrun}{i-1}, \crun} \geq \wInterf - |\buf_{i-1}|$ (because one can decrement $\wInterf$ at most once per prefix and there are  at most $|\buf_{i-1}|$ prefixes whose window is not $\infty$).
                            From this, $|\buf| \leq \wMuarch$, and $\wInterf > \wMuarch + 1$, we get $\minOf{\map{\crun}{\hrun}{i-1}, \crun} \geq \wInterf - \wMuarch > 1$.
                        \end{description}

                        We now show that (i)--(iv) hold:
                        \begin{description}
                            \item[(i):] 
                            We need to show $\update{\tup{m_i,a_i}}{\buf} =  \headConf{\crun(\map{\crun}{\hrun}{i}(|\buf|))}$.
                            From (H.3.a.i) and $\buf_{i-1} \in \prefixes{\buf_{i-1}}$, we have $\update{\tup{m_{i-1},a_{i-1}}}{\buf_{i-1}} =  \headConf{\crun(\map{\crun}{\hrun}{i-1}(|\buf_{i-1}|))}$.
                            From (2) and $p( \apply{a_{i-1}}{\buf_{i-1}} (\pc)) = \pjz{x}{\lbl}$, we get that $\crun(\map{\crun}{\hrun}{i-1}(|\buf_{i-1}|) + 1)$ is obtained by executing the \textsc{Branch} rule of $\CtSpecInterfStep{}{}$ starting from $\crun(\map{\crun}{\hrun}{i-1}(|\buf_{i-1}|))$.
                            From this and Lemma~\ref{lemma:vanilla:contract-rollback}, we have that $\headConf{rb_{\crun}(\map{\crun}{\hrun}{i-1}(ln(\hrun({i-1}))) } =  \headConf{\crun(\map{\crun}{\hrun}{i-1}(|\buf_{i-1}|) )}[\pc \mapsto \ell ]$  because $\headConf{\crun(\map{\crun}{\hrun}{i-1}(ln(\hrun(i-1))))}(x) = 0$.
                            From $\update{\tup{m_{i-1},a_{i-1}}}{\buf_{i-1}} =  \headConf{\crun(\map{\crun}{\hrun}{i-1}(|\buf_{i-1}|))}$, we therefore get $\headConf{rb_{\crun}(\map{\crun}{\hrun}{i-1}(ln(\hrun({i-1}))) } = \update{\tup{m_{i-1},a_{i-1}}}{\buf_{i-1}}[\pc \mapsto \ell ]$.
                            By leveraging $\update{\cdot}{\cdot}$'s definition and $\ell = \ell'$, we get $\headConf{rb_{\crun}(\map{\crun}{\hrun}{i-1}(ln(\hrun({i-1}))) } = \update{\tup{m_{i-1},a_{i-1}}}{(\buf_{i-1} \concat \tagged{\passign{\pc}{\ell'}}{\apply{a_{i-1}}{\buf_{i-1}} (\pc)}) } $.
                            From $\buf_i =  \tagged{\passign{\pc}{\lbl'}}{\apply{a_{i-1}}{\buf_{i-1}} (\pc)} \concat \buf_{i-1}$, $a_{i} = a_{i-1}$, and $m_i = m_{i-1}$, we get $\headConf{rb_{\crun}(\map{\crun}{\hrun}{i-1}(ln(\hrun({i-1}))) } = \update{\tup{m_{i},a_{i}}}{\buf_{i} } $.
                            Finally, from $\map{\crun}{\hrun}{i} = \map{\crun}{\hrun}{i-1}[\mathit{ln}(\hrun(i-1) ) +1 \mapsto rb_{\crun}(\map{\crun}{\hrun}{i-1}(ln(\hrun({i-1}))) ]$,  $|\buf_i|  = \mathit{ln}(\hrun(i-1))+1$, $\buf = \buf_i$, we get $ \update{\tup{m_{i},a_{i}}}{\buf } = \headConf{\crun(\map{\crun}{\hrun}{i}(|\buf|))}$.

                            \item[(ii):]
                            We need to show $\nrMispred{\tup{m_{i},a_{i}}}{\buf} + 1 \geq |\crun(\map{\crun}{\hrun}{i}(|\buf|))|$.
                            From (H.3.a.ii) and $\buf_{i-1} \in \prefixes{\buf_{i-1}}$,  we have $\nrMispred{\tup{m_{i-1},a_{i-1}}}{\buf_{i-1}} + 1 \geq |\crun(\map{\crun}{\hrun}{i-1}(|\buf_{i-1}|))|$.
                            From $\buf = \tagged{\passign{\pc}{\lbl'}}{\apply{a_{i-1}}{\buf_{i-1}} (\pc)} \concat \buf_{i-1}$ and $\headConf{\crun(\map{\crun}{\hrun}{i-1}(ln(\hrun(i-1))))}(x) = 0$, we get that $\nrMispred{\tup{m_{i-1},a_{i-1}}}{\buf} = \nrMispred{\tup{m_{i-1},a_{i-1}}}{\buf_{i-1}}$.
                            From (2) and $p( \apply{a_{i-1}}{\buf_{i-1}} (\pc)) = \pjz{x}{\lbl}$, we get that $\crun(\map{\crun}{\hrun}{i-1}(|\buf_{i-1}|) + 1)$ is obtained by executing the \textsc{Branch} rule of $\CtSpecInterfStep{}{}$ starting from $\crun(\map{\crun}{\hrun}{i-1}(|\buf_{i-1}|))$.
                            From this and Lemma~\ref{lemma:vanilla:contract-rollback}, $ |rb_{\crun}(\map{\crun}{\hrun}{i-1}(ln(\hrun({i-1})))| = |\crun(\map{\crun}{\hrun}{i-1}(|\buf_{i-1}|))|$.
                            Therefore,  we get $\nrMispred{\tup{m_{i-1},a_{i-1}}}{\buf} + 1 \geq |rb_{\crun}(\map{\crun}{\hrun}{i-1}(ln(\hrun({i-1})))| $.
                            Finally, from  $a_{i} = a_{i-1}$,  $m_i = m_{i-1}$, $\map{\crun}{\hrun}{i} = \map{\crun}{\hrun}{i-1}[\mathit{ln}(\hrun(i-1) ) +1 \mapsto rb_{\crun}(\map{\crun}{\hrun}{i-1}(ln(\hrun({i-1}))) ]$, and $|\buf|  = \mathit{ln}(\hrun(i-1))+1$, we get $\nrMispred{\tup{m_{i},a_{i}}}{\buf} + 1 \geq |\crun(\map{\crun}{\hrun}{i}(|\buf|))|$.

                            \item[(iii):]
                            We need to show $\nrMispred{\tup{m_{i},a_{i}}}{\buf} = 0 \leftrightarrow \headWindow{\crun(\map{\crun}{\hrun}{i}(|\buf|))} = \infty$.
                            From (H.3.a.iii) and $\buf_{i-1} \in \prefixes{\buf_{i-1}}$,  we have $\nrMispred{\tup{m_{i-1},a_{i-1}}}{\buf_{i-1}} = 0 \leftrightarrow \headWindow{\crun(\map{\crun}{\hrun}{i-1}(|\buf_{i-1}|))} = \infty$.
                            From (2) and $p( \apply{a_{i-1}}{\buf_{i-1}} (\pc)) = \pjz{x}{\lbl}$, we get that $\crun(\map{\crun}{\hrun}{i-1}(|\buf_{i-1}|) + 1)$ is obtained by executing the \textsc{Branch} rule of $\CtSpecInterfStep{}{}$ starting from $\crun(\map{\crun}{\hrun}{i-1}(|\buf_{i-1}|))$.
                            From this and Lemma~\ref{lemma:vanilla:contract-rollback}, we get $\headWindow{\crun(rb_{\crun}(\map{\crun}{\hrun}{i-1}(ln(\hrun({i-1}))))} = \headWindow{\crun(\map{\crun}{\hrun}{i-1}(|\buf_{i-1}|))} - 1$.
                            From $\buf = \tagged{\passign{\pc}{\lbl'}}{\apply{a_{i-1}}{\buf_{i-1}} (\pc)} \concat \buf_{i-1}$ and $\headConf{\crun(\map{\crun}{\hrun}{i-1}(ln(\hrun(i-1))))}(x) = 0$, we get that $\nrMispred{\tup{m_{i-1},a_{i-1}}}{\buf} = \nrMispred{\tup{m_{i-1},a_{i-1}}}{\buf_{i-1}}$.
                            From $a_i =  a_{i-1}$, $m_i = m_{i-1}$, $\map{\crun}{\hrun}{i} = \map{\crun}{\hrun}{i-1}[\mathit{ln}(\hrun(i-1) ) +1 \mapsto \crun(rb_{\crun}(\map{\crun}{\hrun}{i-1}(ln(\hrun({i-1}))))]$, and $|\buf|  = \mathit{ln}(\hrun(i-1))+1$, we get that $\nrMispred{\tup{m_{i},a_{i}}}{\buf} = \nrMispred{\tup{m_{i-1},a_{i-1}}}{\buf_{i-1}}$ and                $\headWindow{\crun(\map{\crun}{\hrun}{i}(|\buf|))} = \headWindow{\crun(\map{\crun}{\hrun}{i-1}(|\buf_{i-1}|))}$.
                            Therefore, $\nrMispred{\tup{m_{i},a_{i}}}{\buf} = 0 \leftrightarrow \headWindow{\crun(\map{\crun}{\hrun}{i}(|\buf|))} = \infty$ immediately follows from $\nrMispred{\tup{m_{i-1},a_{i-1}}}{\buf_{i-1}} = 0 \leftrightarrow \headWindow{\crun(\map{\crun}{\hrun}{i-1}(|\buf_{i-1}|))} = \infty$.
                            
                            \item[(iv):]
                            We need to show $\headWindow{\crun(\map{\crun}{\hrun}{i}(|\buf|))} > 0$.
                           	From Lemma~\ref{lemma:vanilla:contract-rollback}, we get that  either $\headWindow{ \crun( \map{\crun}{\hrun}{i}(|\buf|) ) } = \infty$ or  $\headWindow{ \crun( \map{\crun}{\hrun}{i}(|\buf|) ) } = \headWindow{ \crun( \map{\crun}{\hrun}{i}(|\buf_{i-1}|) ) } - 1$. 
                           	In the first case, then $\headWindow{\crun(\map{\crun}{\hrun}{i}(|\buf|))} > 0$ holds.
                           	In the second case, from (2), $\map{\crun}{\hrun}{i}(|\buf_{i-1}|) =  \map{\crun}{\hrun}{i-1}( |\buf_{i-1}| )) $, and $\minOf{\map{\crun}{\hrun}{i-1}, \crun} > 1$, we get  $\headWindow{\crun(\map{\crun}{\hrun}{i}(|\buf|))} > 0$.
                        \end{description}
                    \end{description}

                    \item[(b):]
                    From (H.b), we have that  $\tup{j,j'} \in \nextPairs{\map{\crun}{\hrun}{i-1}}$,
                    $\headWindow{ \crun( \map{\crun}{\hrun}{i-1}(j) ) } = \headWindow{ \crun( \map{\crun}{\hrun}{i-1}(j') )  } = \infty$ or
                    $\headWindow{ \crun( \map{\crun}{\hrun}{i-1}(j) ) } = \infty$ and $\headWindow{ \crun( \map{\crun}{\hrun}{i-1}(j') ) } = \wInterf$ or
                    $\headWindow{ \crun( \map{\crun}{\hrun}{i-1}(j) ) } \neq \infty \wedge \headWindow{ \crun( \map{\crun}{\hrun}{i-1}(j') ) } = \headWindow{ \crun( \map{\crun}{\hrun}{i-1}(j) ) } - 1$.
                    From $\map{\crun}{\hrun}{i} = \map{\crun}{\hrun}{i-1}[ln(\hrun({i-1}))+1 \mapsto rb_{\crun}(\map{\crun}{\hrun}{i-1}(ln(\hrun({i-1}))) ]$, the only interesting pair is $\tup{j,j'}  = \tup{|\buf_{i-1}|, |\buf_{i-1}|+1}$ since for all others pairs the relation follows from (H.b).
                    For this pair, we have $\map{\crun}{\hrun}{i}(|\buf_{i-1}|) =  \map{\crun}{\hrun}{i-1}(\mathit{ln}(\hrun(i-1) )) $ and $ \map{\crun}{\hrun}{i}(|\buf_{i-1}|+1)   = rb_{\crun}(\map{\crun}{\hrun}{i-1}(ln(\hrun({i-1})))$.
                    From Lemma~\ref{lemma:vanilla:contract-rollback}, we get that  either $\headWindow{ \crun( \map{\crun}{\hrun}{i}(j) ) } = \headWindow{ \crun( \map{\crun}{\hrun}{i}(j') )  } = \infty$ or  $\headWindow{ \crun( \map{\crun}{\hrun}{i}(j) ) } \neq \infty \wedge \headWindow{ \crun( \map{\crun}{\hrun}{i}(j') ) } = \headWindow{ \crun( \map{\crun}{\hrun}{i}(j) ) } - 1$.
                    Hence, (b) holds.
                \end{description}                
                
                \item[$\lbl' = \lbl \wedge \update{\tup{m_{i-1},a_{i-1}}}{\buf_{i-1}}(x) \neq 0 $:] 
                From $\lbl' = \lbl \wedge \update{\tup{m_{i-1},a_{i-1}}}{\buf_{i-1}}(x) \neq 0 $ and (H.a), we immediately get that $\BpPredict(\hrun(i-1)) = \lbl \wedge \crun(\map{\crun}{\hrun}{i-1}(ln(\hrun(i-1))))(x) \neq 0$.
                Therefore, $\map{\crun}{\hrun}{i} = \map{\crun}{\hrun}{i-1}[ln(\hrun({i-1}))+1 \mapsto \map{\crun}{\hrun}{i-1}(ln(\hrun(i-1)))+1]$.

                We now show that all conditions hold:
                \begin{description}
                    \item[(a):]
                    Let $\buf$ be an arbitrary prefix in  $\prefixes{\buf_i}$.
                    There are two cases:
                    \begin{description}
                        \item[$\buf \in \prefixes{\buf_{i-1}}$:]
                        Then, $|\buf| \geq |\buf_i|$.
                        We now show that (i)--(iv) hold:
                        \begin{description}
                            \item[(i):]
                            We need to show $\update{\tup{m_i,a_i}}{\buf} =  \headConf{\crun(\map{\crun}{\hrun}{i}(|\buf|))}$.
                            From (H.3.a.i) and $\buf \in \prefixes{\buf_{i-1}}$, we have $\update{\tup{m_{i-1},a_{i-1}}}{\buf} =  \headConf{\crun(\map{\crun}{\hrun}{i-1}(|\buf|))}$.
                            Finally, from this,  $m_{i} = m_{a-1}$, $a_{i} = a_{i-1}$, and $\map{\crun}{\hrun}{i}(\buf) = \map{\crun}{\hrun}{i-1}(\buf)$, we get $\update{\tup{m_i,a_i}}{\buf} =  \headConf{\crun(\map{\crun}{\hrun}{i}(|\buf|))}$.
        
                            \item[(ii):]
                            We need to show $\nrMispred{\tup{m_{i},a_{i}}}{\buf} + 1 \geq |\crun(\map{\crun}{\hrun}{i}(|\buf|))|$.
                            From (H.3.a.ii) and $\buf \in \prefixes{\buf_{i-1}}$,  we have $\nrMispred{\tup{m_{i-1},a_{i-1}}}{\buf} + 1 \geq |\crun(\map{\crun}{\hrun}{i-1}(|\buf|))|$.
                            Finally, from this,  $m_{i} = m_{a-1}$, $a_{i} = a_{i-1}$, and $\map{\crun}{\hrun}{i}(\buf) = \map{\crun}{\hrun}{i-1}(\buf)$, we get $\nrMispred{\tup{m_{i},a_{i}}}{\buf} + 1 \geq |\crun(\map{\crun}{\hrun}{i}(|\buf|))|$.
        
                            \item[(iii):]
                            We need to show $\nrMispred{\tup{m_{i},a_{i}}}{\buf} = 0 \leftrightarrow \headWindow{\crun(\map{\crun}{\hrun}{i}(|\buf|))} = \infty$.
                            From (H.3.a.iii) and $\buf \in \prefixes{\buf_{i-1}}$,  we have $\nrMispred{\tup{m_{i-1},a_{i-1}}}{\buf} = 0 \leftrightarrow \headWindow{\crun(\map{\crun}{\hrun}{i-1}(|\buf|))} = \infty$.
                            Finally, from this,  $m_{i} = m_{a-1}$, $a_{i} = a_{i-1}$, and $\map{\crun}{\hrun}{i}(\buf) = \map{\crun}{\hrun}{i-1}(\buf)$, we get $\nrMispred{\tup{m_{i},a_{i}}}{\buf} = 0 \leftrightarrow \headWindow{\crun(\map{\crun}{\hrun}{i}(|\buf|))} = \infty$.
                            
							\item[(iv):]
	                        We need to show $\headWindow{\crun(\map{\crun}{\hrun}{i}(|\buf|))} > 0$.
	                        From (H.3.a.iv) and $\buf \in \prefixes{\buf_{i-1}}$, we get $\headWindow{\crun(\map{\crun}{\hrun}{i-1}(|\buf|))} > 0$.
	                        From this and $\map{\crun}{\hrun}{i}(\buf) = \map{\crun}{\hrun}{i-1}(\buf)$, we get $\headWindow{\crun(\map{\crun}{\hrun}{i}(|\buf|))} > 0$.
        
                        \end{description} 

                        \item[$\buf \not\in \prefixes{\buf_{i-1}}$:]
                        Then, $\buf = \buf_i =  \tagged{\passign{\pc}{\lbl'}}{\apply{a_{i-1}}{\buf_{i-1}} (\pc)} \concat \buf_{i-1}$, $|\buf| = |\buf_{i-1}| +1 = ln(\hrun(i-1))+1$, and $\map{\crun}{\hrun}{i}(|\buf|) = \map{\crun}{\hrun}{i-1}(ln(\hrun(i-1)))+1$.
                        From (H.3.a.i) and $\buf_{i-1} \in \prefixes{\buf_{i-1}}$, we have $\update{\tup{m_{i-1},a_{i-1}}}{\buf_{i-1}} =  \headConf{\crun(\map{\crun}{\hrun}{i-1}(|\buf_{i-1}|))}$.
                        Observe that:
                        \begin{description}
                            \item[(1):] 
                            From  $\update{\tup{m_{i-1},a_{i-1}}}{\buf_{i-1}} =  \headConf{\crun(\map{\crun}{\hrun}{i-1}(|\buf_{i-1}|))}$, we have, in particular, that $\update{\tup{m_{i-1},a_{i-1}}}{\buf_{i-1}}(\pc) = \headConf{\crun(\map{\crun}{\hrun}{i-1}(|\buf_{i-1}|))}(\pc)$.
                            \item[(2):] 
                            From (H.a.iii) applied to $\buf_{i-1}$ and the well-formedness of buffers, we get that $\maxOf{\map{\crun}{\hrun}{i-1}, \crun} = \infty$ (because the shortest prefix is always associated with window $\infty$).
                            From (H.b) and $\maxOf{\map{\crun}{\hrun}{i-1}, \crun} = \infty$, we have that $\minOf{\map{\crun}{\hrun}{i-1}, \crun} \geq \wInterf - |\buf_{i-1}|$ (because one can decrement $\wInterf$ at most once per prefix and there are  at most $|\buf_{i-1}|$ prefixes).
                            From this, $|\buf| \leq \wMuarch$, and $\wInterf > \wMuarch + 1$, we get $\minOf{\map{\crun}{\hrun}{i-1}, \crun} \geq \wInterf - \wMuarch > 1$.
                        \end{description}

                        We now show that (i)--(iv) hold:
                        \begin{description}
                            \item[(i):]
                            We need to show $\update{\tup{m_i,a_i}}{\buf} =  \headConf{\crun(\map{\crun}{\hrun}{i}(|\buf|))}$.
                            From (H.3.a.i) and $\buf_{i-1} \in \prefixes{\buf_{i-1}}$, we have $\update{\tup{m_{i-1},a_{i-1}}}{\buf_{i-1}} =  \headConf{\crun(\map{\crun}{\hrun}{i-1}(|\buf_{i-1}|))}$.
                            From (2) and $p( \apply{a_{i-1}}{\buf_{i-1}} (\pc)) = \pjz{x}{\lbl}$, we get that $\crun(\map{\crun}{\hrun}{i-1}(|\buf_{i-1}|) + 1)$ is obtained by executing the \textsc{Branch} rule of $\CtSpecInterfStep{}{}$ starting from $\crun(\map{\crun}{\hrun}{i-1}(|\buf_{i-1}|))$.
                            Therefore,  $\headConf{\crun(\map{\crun}{\hrun}{i-1}(|\buf_{i-1}|) + 1)} = \headConf{\crun(\map{\crun}{\hrun}{i-1}(|\buf_{i-1}|) )}[\pc \mapsto \ell ]$ because $\headConf{\crun(\map{\crun}{\hrun}{i-1}(ln(\hrun(i-1))))}(x) \neq 0$.
                            From $\update{\tup{m_{i-1},a_{i-1}}}{\buf_{i-1}} =  \headConf{\crun(\map{\crun}{\hrun}{i-1}(|\buf_{i-1}|))}$, we therefore get $\headConf{\crun(\map{\crun}{\hrun}{i-1}(|\buf_{i-1}|) + 1)} = \update{\tup{m_{i-1},a_{i-1}}}{\buf_{i-1}}[\pc \mapsto \ell ]$.
                            By leveraging $\update{\cdot}{\cdot}$'s definition and $\ell = \ell'$, we get $\headConf{\crun(\map{\crun}{\hrun}{i-1}(|\buf_{i-1}|) + 1)} = \update{\tup{m_{i-1},a_{i-1}}}{(\buf_{i-1} \concat \tagged{\passign{\pc}{\ell'}}{\apply{a_{i-1}}{\buf_{i-1}} (\pc)}) } $.
                            From $\buf_i =  \tagged{\passign{\pc}{\lbl'}}{\apply{a_{i-1}}{\buf_{i-1}} (\pc)} \concat \buf_{i-1}$, $a_{i} = a_{i-1}$, and $m_i = m_{i-1}$, we get $\headConf{\crun(\map{\crun}{\hrun}{i-1}(|\buf_{i-1}|) + 1)} = \update{\tup{m_{i},a_{i}}}{\buf_{i} } $.
                            Finally, from $\map{\crun}{\hrun}{i} = \map{\crun}{\hrun}{i-1}[\mathit{ln}(\hrun(i-1) ) +1 \mapsto \map{\crun}{\hrun}{i-1}(\mathit{ln}(\hrun(i-1) )) + 1]$,  $|\buf_i|  = \mathit{ln}(\hrun(i-1))+1$, $\buf = \buf_i$, we get $ \update{\tup{m_{i},a_{i}}}{\buf } = \headConf{\crun(\map{\crun}{\hrun}{i}(|\buf|))}$.
                            
                            \item[(ii):]
                            We need to show $\nrMispred{\tup{m_{i},a_{i}}}{\buf} + 1 \geq |\crun(\map{\crun}{\hrun}{i}(|\buf|))|$.
                            From (H.3.a.ii) and $\buf_{i-1} \in \prefixes{\buf_{i-1}}$,  we have $\nrMispred{\tup{m_{i-1},a_{i-1}}}{\buf_{i-1}} + 1 \geq |\crun(\map{\crun}{\hrun}{i-1}(|\buf_{i-1}|))|$.
                            From $\buf = \tagged{\passign{\pc}{\lbl'}}{\apply{a_{i-1}}{\buf_{i-1}} (\pc)} \concat \buf_{i-1}$ and $\headConf{\crun(\map{\crun}{\hrun}{i-1}(ln(\hrun(i-1))))}(x) \neq 0$, we get that $\nrMispred{\tup{m_{i-1},a_{i-1}}}{\buf} = \nrMispred{\tup{m_{i-1},a_{i-1}}}{\buf_{i-1}}+1$.
                            From (2) and $p( \apply{a_{i-1}}{\buf_{i-1}} (\pc)) = \pjz{x}{\lbl}$, we get that $\crun(\map{\crun}{\hrun}{i-1}(|\buf_{i-1}|) + 1)$ is obtained by executing the \textsc{Branch} rule of $\CtSpecInterfStep{}{}$ starting from $\crun(\map{\crun}{\hrun}{i-1}(|\buf_{i-1}|))$.
                            Therefore, $ |\crun(\map{\crun}{\hrun}{i-1}(|\buf_{i-1}|)+1)| = |\crun(\map{\crun}{\hrun}{i-1}(|\buf_{i-1}|))|  + 1$.
                            By adding $1$ to both sides in $\nrMispred{\tup{m_{i-1},a_{i-1}}}{\buf_{i-1}} + 1 \geq |\crun(\map{\crun}{\hrun}{i-1}(|\buf_{i-1}|))|$, we get $\nrMispred{\tup{m_{i-1},a_{i-1}}}{\buf_{i-1}} +  1 +1 \geq |\crun(\map{\crun}{\hrun}{i-1}(|\buf_{i-1}|))|+1$.
                            Therefore,  we get $\nrMispred{\tup{m_{i-1},a_{i-1}}}{\buf} + 1 \geq |\crun(\map{\crun}{\hrun}{i-1}(|\buf_{i-1}|)+1)|$ (from $\nrMispred{\tup{m_{i-1},a_{i-1}}}{\buf} = \nrMispred{\tup{m_{i-1},a_{i-1}}}{\buf_{i-1}}+1$ and $ |\crun(\map{\crun}{\hrun}{i-1}(|\buf_{i-1}|)+1)| = |\crun(\map{\crun}{\hrun}{i-1}(|\buf_{i-1}|))|  + 1$).
                            Finally, from  $a_{i} = a_{i-1}$,  $m_i = m_{i-1}$, $\map{\crun}{\hrun}{i} = \map{\crun}{\hrun}{i-1}[\mathit{ln}(\hrun(i-1) ) +1 \mapsto \map{\crun}{\hrun}{i-1}(\mathit{ln}(\hrun(i-1) )) + 1]$, and $|\buf|  = \mathit{ln}(\hrun(i-1))+1$, we get $\nrMispred{\tup{m_{i},a_{i}}}{\buf} + 1 \geq |\crun(\map{\crun}{\hrun}{i}(|\buf|))|$.

                            \item[(iii):] 
                            We need to show $\nrMispred{\tup{m_{i},a_{i}}}{\buf} = 0 \leftrightarrow \headWindow{\crun(\map{\crun}{\hrun}{i}(|\buf|))} = \infty$.
                            From (H.3.a.iii) and $\buf_{i-1} \in \prefixes{\buf_{i-1}}$,  we have $\nrMispred{\tup{m_{i-1},a_{i-1}}}{\buf_{i-1}} = 0 \leftrightarrow \headWindow{\crun(\map{\crun}{\hrun}{i-1}(|\buf_{i-1}|))} = \infty$.
                            From (2) and $p( \apply{a_{i-1}}{\buf_{i-1}} (\pc)) = \pjz{x}{\lbl}$, we get that $\crun(\map{\crun}{\hrun}{i-1}(|\buf_{i-1}|) + 1)$ is obtained by executing the \textsc{Branch} rule of $\CtSpecInterfStep{}{}$ starting from $\crun(\map{\crun}{\hrun}{i-1}(|\buf_{i-1}|))$.
                            Therefore, $\headWindow{\crun(\map{\crun}{\hrun}{i-1}(|\buf_{i-1}|)+1)} \neq \infty$ (because it is either $\wInterf$ or $\headWindow{\crun(\map{\crun}{\hrun}{i-1}(|\buf_{i-1}|))}-1$).
                            From $\buf = \tagged{\passign{\pc}{\lbl'}}{\apply{a_{i-1}}{\buf_{i-1}} (\pc)} \concat \buf_{i-1}$ and $\headConf{\crun(\map{\crun}{\hrun}{i-1}(ln(\hrun(i-1))))}(x) \neq 0$, we get that $\nrMispred{\tup{m_{i-1},a_{i-1}}}{\buf} = \nrMispred{\tup{m_{i-1},a_{i-1}}}{\buf_{i-1}}+1$.
                            Therefore, $\nrMispred{\tup{m_{i-1},a_{i-1}}}{\buf} \neq 0$.
                            Thus, from this, $\headWindow{\crun(\map{\crun}{\hrun}{i-1}(|\buf_{i-1}|)+1)} \neq \infty$,  $m_i = m_{i-1}$, $\map{\crun}{\hrun}{i} = \map{\crun}{\hrun}{i-1}[\mathit{ln}(\hrun(i-1) ) +1 \mapsto \map{\crun}{\hrun}{i-1}(\mathit{ln}(\hrun(i-1) )) + 1]$, and $|\buf|  = \mathit{ln}(\hrun(i-1))+1$, we get $\nrMispred{\tup{m_{i},a_{i}}}{\buf} = 0 \leftrightarrow \headWindow{\crun(\map{\crun}{\hrun}{i}(|\buf|))} = \infty$.
                            
                            \item[(iv):]
                            We need to show $\headWindow{\crun(\map{\crun}{\hrun}{i}(|\buf|))} > 0$.
                            As we have shown before, $\crun(\map{\crun}{\hrun}{i-1}(\mathit{ln}(\hrun(i-1) )) + 1)$ is obtained from $\crun(\map{\crun}{\hrun}{i-1}(\mathit{ln}(\hrun(i-1) )))$ by applying the \textsc{Branch} rule of $\CtSpecInterfStep{}{}$.
                            Then, we get that either $\headWindow{ \crun( \map{\crun}{\hrun}{i}(|\buf|) )  } = \wInterf$ or  $\headWindow{ \crun( \map{\crun}{\hrun}{i}(|\buf|) ) } = \headWindow{ \crun( \map{\crun}{\hrun}{i}(|\buf_{i-1}|) ) } - 1$. 
                           	In the first case, then $\headWindow{\crun(\map{\crun}{\hrun}{i}(|\buf|))} > 0$ follows from $\wInterf > \wMuarch + 1$ and $\wMuarch \geq 0$.
                           	In the second case, from (2), $\map{\crun}{\hrun}{i}(|\buf_{i-1}|) =  \map{\crun}{\hrun}{i-1}( |\buf_{i-1}| )) $, and $\minOf{\map{\crun}{\hrun}{i-1}, \crun} > 1$, we get  $\headWindow{\crun(\map{\crun}{\hrun}{i}(|\buf|))} > 0$.
                        \end{description}
                    \end{description}

                    \item[(b):] 
                    From (H.b), we have that  $\tup{j,j'} \in \nextPairs{\map{\crun}{\hrun}{i-1}}$,
                    $\headWindow{ \crun( \map{\crun}{\hrun}{i-1}(j) ) } = \headWindow{ \crun( \map{\crun}{\hrun}{i-1}(j') )  } = \infty$ or
                    $\headWindow{ \crun( \map{\crun}{\hrun}{i-1}(j) ) } = \infty$ and $\headWindow{ \crun( \map{\crun}{\hrun}{i-1}(j') ) } = \wInterf$ or
                    $\headWindow{ \crun( \map{\crun}{\hrun}{i-1}(j) ) } \neq \infty \wedge \headWindow{ \crun( \map{\crun}{\hrun}{i-1}(j') ) } = \headWindow{ \crun( \map{\crun}{\hrun}{i-1}(j) ) } - 1$.
                    From $\map{\crun}{\hrun}{i} = \map{\crun}{\hrun}{i-1}[\mathit{ln}(\hrun(i-1) ) +1 \mapsto \map{\crun}{\hrun}{i-1}(\mathit{ln}(\hrun(i-1) )) + 1]$, the only interesting pair is $\tup{j,j'}  = \tup{|\buf_{i-1}|, |\buf_{i-1}|+1}$ since for all others pairs the relation follows from (H.b).
                    For this pair, we have $\map{\crun}{\hrun}{i}(|\buf_{i-1}|) =  \map{\crun}{\hrun}{i-1}(\mathit{ln}(\hrun(i-1) )) $ and $ \map{\crun}{\hrun}{i}(|\buf_{i-1}|+1)   = \map{\crun}{\hrun}{i-1}(\mathit{ln}(\hrun(i-1) )) + 1$.
                    As we have shown before, $\crun(\map{\crun}{\hrun}{i-1}(\mathit{ln}(\hrun(i-1) )) + 1)$ is obtained from $\crun(\map{\crun}{\hrun}{i-1}(\mathit{ln}(\hrun(i-1) )))$ by applying the \textsc{Branch} rule of $\CtSpecInterfStep{}{}$.
                    Therefore, either $\headWindow{ \crun( \map{\crun}{\hrun}{i}(j) ) } = \infty$ and $\headWindow{ \crun( \map{\crun}{\hrun}{i}(j') )  } = \wInterf$ or  $\headWindow{ \crun( \map{\crun}{\hrun}{i}(j) ) } \neq \infty \wedge \headWindow{ \crun( \map{\crun}{\hrun}{i}(j') ) } = \headWindow{ \crun( \map{\crun}{\hrun}{i}(j) ) } - 1$.
                    Hence, (b) holds.
                \end{description}

                \item[$\lbl' = \apply{a_{i-1}}{\buf_{i-1}} (\pc) + 1 \wedge \update{\tup{m_{i-1},a_{i-1}}}{\buf_{i-1}}(x) = 0$:]
                The proof of this case is similar to that of the  $\lbl' = \lbl \wedge \update{\tup{m_{i-1},a_{i-1}}}{\buf_{i-1}}(x) \neq 0 $ case.

                \item[$\lbl' = \apply{a_{i-1}}{\buf_{i-1}} (\pc) + 1 \wedge \update{\tup{m_{i-1},a_{i-1}}s}{\buf_{i-1}}(x) \neq 0 $:]  
                The proof of this case is similar to that of the  $\lbl' = \lbl \wedge \update{\tup{m_{i-1},a_{i-1}}}{\buf_{i-1}}(x) = 0 $ case. 
            \end{description}

			\item[Rule \textsc{Fetch-Jump-Hit}:]
            Then, $\buf_i = \tagged{\passign{\pc}{e}}{\notags} \concat \buf_{i-1}$, $a_{i} = a_{i-1}$, $m_{i} = m_{i-1}$, $p(\apply{a_{i-1}}{\buf_{i-1}}(\pc) ) = \pjmp{e}$, and $\map{\crun}{\hrun}{i} = \map{\crun}{\hrun}{i-1}[\mathit{ln}(\hrun(i-1) ) +1 \mapsto \map{\crun}{\hrun}{i-1}(\mathit{ln}(\hrun(i-1) )) + 1]$.

            We now show that all conditions hold:
            \begin{description}
                \item[(a):]
                Let $\buf$ be an arbitrary prefix in $\prefixes{\buf_{i}}$.
                There are two cases:
                \begin{description}
                    \item[$\buf \in \prefixes{\buf_{i-1}}$:]
                    Then, $|\buf| \geq |\buf_i|$.
                    We now show that (i)--(iv) hold:
                    \begin{description}
                        \item[(i):]
                        We need to show $\update{\tup{m_i,a_i}}{\buf} =  \headConf{\crun(\map{\crun}{\hrun}{i}(|\buf|))}$.
                        From (H.3.a.i) and $\buf \in \prefixes{\buf_{i-1}}$, we have $\update{\tup{m_{i-1},a_{i-1}}}{\buf} =  \headConf{\crun(\map{\crun}{\hrun}{i-1}(|\buf|))}$.
                        Finally, from this,  $m_{i} = m_{a-1}$, $a_{i} = a_{i-1}$, and $\map{\crun}{\hrun}{i}(\buf) = \map{\crun}{\hrun}{i-1}(\buf)$, we get $\update{\tup{m_i,a_i}}{\buf} =  \headConf{\crun(\map{\crun}{\hrun}{i}(|\buf|))}$.
    
                        \item[(ii):]
                        We need to show $\nrMispred{\tup{m_{i},a_{i}}}{\buf} + 1 \geq |\crun(\map{\crun}{\hrun}{i}(|\buf|))|$.
                        From (H.3.a.ii) and $\buf \in \prefixes{\buf_{i-1}}$,  we have $\nrMispred{\tup{m_{i-1},a_{i-1}}}{\buf} + 1 \geq |\crun(\map{\crun}{\hrun}{i-1}(|\buf|))|$.
                        Finally, from this,  $m_{i} = m_{a-1}$, $a_{i} = a_{i-1}$, and $\map{\crun}{\hrun}{i}(\buf) = \map{\crun}{\hrun}{i-1}(\buf)$, we get $\nrMispred{\tup{m_{i},a_{i}}}{\buf} + 1 \geq |\crun(\map{\crun}{\hrun}{i}(|\buf|))|$.
    
                        \item[(iii):]
                        We need to show $\nrMispred{\tup{m_{i},a_{i}}}{\buf} = 0 \leftrightarrow \headWindow{\crun(\map{\crun}{\hrun}{i}(|\buf|))} = \infty$.
                        From (H.3.a.iii) and $\buf \in \prefixes{\buf_{i-1}}$,  we have $\nrMispred{\tup{m_{i-1},a_{i-1}}}{\buf} = 0 \leftrightarrow \headWindow{\crun(\map{\crun}{\hrun}{i-1}(|\buf|))} = \infty$.
                        Finally, from this,  $m_{i} = m_{a-1}$, $a_{i} = a_{i-1}$, and $\map{\crun}{\hrun}{i}(\buf) = \map{\crun}{\hrun}{i-1}(\buf)$, we get $\nrMispred{\tup{m_{i},a_{i}}}{\buf} = 0 \leftrightarrow \headWindow{\crun(\map{\crun}{\hrun}{i}(|\buf|))} = \infty$.

    					\item[(iv):]
                        We need to show $\headWindow{\crun(\map{\crun}{\hrun}{i}(|\buf|))} > 0$.
                        From (H.3.a.iv) and $\buf \in \prefixes{\buf_{i-1}}$, we get $\headWindow{\crun(\map{\crun}{\hrun}{i-1}(|\buf|))} > 0$.
                        From this and $\map{\crun}{\hrun}{i}(\buf) = \map{\crun}{\hrun}{i-1}(\buf)$, we get $\headWindow{\crun(\map{\crun}{\hrun}{i}(|\buf|))} > 0$.
                        
                    \end{description}

                    \item[$\buf \not\in \prefixes{\buf_{i-1}}$:]    
                    From $\buf_i = \tagged{\passign{\pc}{e}}{\notags} \concat \buf_{i-1}$, we get that $\buf = \buf_i$ and $|\buf| = |buf_{i-1}| + 1 = \mathit{ln}(\hrun(i-1))+1$.
                    From (H.3.a.i) and $\buf_{i-1} \in \prefixes{\buf_{i-1}}$, we have $\update{\tup{m_{i-1},a_{i-1}}}{\buf_{i-1}} =  \headConf{\crun(\map{\crun}{\hrun}{i-1}(|\buf_{i-1}|))}$.
                    Observe that:
                    \begin{description}
                        \item[(1):] 
                        From  $\update{\tup{m_{i-1},a_{i-1}}}{\buf_{i-1}} =  \headConf{\crun(\map{\crun}{\hrun}{i-1}(|\buf_{i-1}|))}$, we have, in particular, that $\update{\tup{m_{i-1},a_{i-1}}}{\buf_{i-1}}(\pc) = \headConf{\crun(\map{\crun}{\hrun}{i-1}(|\buf_{i-1}|))}(\pc)$.
                        \item[(2):] 
                        From (H.a.iii) applied to $\buf_{i-1}$ and the well-formedness of buffers, we get that $\maxOf{\map{\crun}{\hrun}{i-1}, \crun} = \infty$ (because the shortest prefix is always associated with window $\infty$).
                        From (H.b) and $\maxOf{\map{\crun}{\hrun}{i-1}, \crun} = \infty$, we have that $\minOf{\map{\crun}{\hrun}{i-1}, \crun} \geq \wInterf - |\buf_{i-1}|$ (because one can decrement $\wInterf$ at most once per prefix and there are  at most $|\buf_{i-1}|$ prefixes).
                        From this, $|\buf| \leq \wMuarch$, and $\wInterf > \wMuarch + 1$, we get $\minOf{\map{\crun}{\hrun}{i-1}, \crun} \geq \wInterf - \wMuarch > 0$.
                    \end{description}

                    We now show that (i)--(iv) hold:
                    \begin{description}
                        \item[(i):]
                        We need to show $\update{\tup{m_i,a_i}}{\buf} =  \headConf{\crun(\map{\crun}{\hrun}{i}(|\buf|))}$.
                        From (H.3.a.i) and $\buf_{i-1} \in \prefixes{\buf_{i-1}}$, we have $\update{\tup{m_{i-1},a_{i-1}}}{\buf_{i-1}} =  \headConf{\crun(\map{\crun}{\hrun}{i-1}(|\buf_{i-1}|))}$.
                        From (2), we get that $\crun(\map{\crun}{\hrun}{i-1}(|\buf_{i-1}|) + 1)$ is obtained by executing the \textsc{Step} rule of $\CtSpecInterfStep{}{}$ starting from $\crun(\map{\crun}{\hrun}{i-1}(|\buf_{i-1}|))$.
                        From (1) and $p(\apply{a_{i-1}}{\buf_{i-1}}(\pc) ) = \pjmp{e}$, we have that $\headConf{\crun(\map{\crun}{\hrun}{i-1}(|\buf_{i-1}|) + 1)} = \headConf{\crun(\map{\crun}{\hrun}{i-1}(|\buf_{i-1}|) )}[\pc \mapsto \exprEval{e}{\headConf{\crun(\map{\crun}{\hrun}{i-1}(|\buf_{i-1}|) )}}]$.
                        From $\update{\tup{m_{i-1},a_{i-1}}}{\buf_{i-1}} =  \headConf{\crun(\map{\crun}{\hrun}{i-1}(|\buf_{i-1}|))}$, we therefore get $\headConf{\crun(\map{\crun}{\hrun}{i-1}(|\buf_{i-1}|) + 1)} = \update{\tup{m_{i-1},a_{i-1}}}{\buf_{i-1}}[\pc \mapsto \exprEval{e}{ \update{\tup{m_{i-1},a_{i-1}}}{\buf_{i-1}} }]$.
                        By leveraging $\update{\cdot}{\cdot}$'s definition, we get $\headConf{\crun(\map{\crun}{\hrun}{i-1}(|\buf_{i-1}|) + 1)} = \update{\tup{m_{i-1},a_{i-1}}}{(\buf_{i-1} \concat \tagged{\passign{\pc}{e}}{\notags}) } $.
                        From $\buf_i = \tagged{\passign{\pc}{e}}{\notags} \concat \buf_{i-1}$, $a_{i} = a_{i-1}$, and $m_i = m_{i-1}$, we get $\headConf{\crun(\map{\crun}{\hrun}{i-1}(|\buf_{i-1}|) + 1)} = \update{\tup{m_{i},a_{i}}}{\buf_{i} } $.
                        Finally, from $\map{\crun}{\hrun}{i} = \map{\crun}{\hrun}{i-1}[\mathit{ln}(\hrun(i-1) ) +1 \mapsto \map{\crun}{\hrun}{i-1}(\mathit{ln}(\hrun(i-1) )) + 1]$,  $|\buf_i|  = \mathit{ln}(\hrun(i-1))+1$, $\buf = \buf_i$, we get $ \update{\tup{m_{i},a_{i}}}{\buf } = \headConf{\crun(\map{\crun}{\hrun}{i}(|\buf|))}$.

                        \item[(ii):]
                        We need to show $\nrMispred{\tup{m_{i},a_{i}}}{\buf} + 1 \geq |\crun(\map{\crun}{\hrun}{i}(|\buf|))|$.
                        From (H.3.a.ii) and $\buf_{i-1} \in \prefixes{\buf_{i-1}}$,  we have $\nrMispred{\tup{m_{i-1},a_{i-1}}}{\buf_{i-1}} + 1 \geq |\crun(\map{\crun}{\hrun}{i-1}(|\buf_{i-1}|))|$.
                        From $\buf =  \tagged{\passign{\pc}{e}}{\notags} \concat \buf_{i-1}$, we get that $\nrMispred{\tup{m_{i-1},a_{i-1}}}{\buf_{i-1}} = \nrMispred{\tup{m_{i-1},a_{i-1}}}{\buf}$.
                        From (2), we get that $\crun(\map{\crun}{\hrun}{i-1}(|\buf_{i-1}|) + 1)$ is obtained by executing the \textsc{Step} rule of $\CtSpecInterfStep{}{}$ starting from $\crun(\map{\crun}{\hrun}{i-1}(|\buf_{i-1}|))$.
                        Therefore, $|\crun(\map{\crun}{\hrun}{i-1}(|\buf_{i-1}|))| = |\crun(\map{\crun}{\hrun}{i-1}(|\buf_{i-1}|)+1)|$.
                        Thus, we get $\nrMispred{\tup{m_{i-1},a_{i-1}}}{\buf} + 1 \geq |\crun(\map{\crun}{\hrun}{i-1}(|\buf_{i-1}|)+1)|$.
                        Finally, from  $a_{i} = a_{i-1}$,  $m_i = m_{i-1}$, $\map{\crun}{\hrun}{i} = \map{\crun}{\hrun}{i-1}[\mathit{ln}(\hrun(i-1) ) +1 \mapsto \map{\crun}{\hrun}{i-1}(\mathit{ln}(\hrun(i-1) )) + 1]$, and $|\buf|  = \mathit{ln}(\hrun(i-1))+1$, we get $\nrMispred{\tup{m_{i},a_{i}}}{\buf} + 1 \geq |\crun(\map{\crun}{\hrun}{i}(|\buf|))|$.

                        \item[(iii):]
                        We need to show $\nrMispred{\tup{m_{i},a_{i}}}{\buf} = 0 \leftrightarrow \headWindow{\crun(\map{\crun}{\hrun}{i}(|\buf|))} = \infty$.
                        From (H.3.a.iii) and $\buf_{i-1} \in \prefixes{\buf_{i-1}}$,  we have $\nrMispred{\tup{m_{i-1},a_{i-1}}}{\buf_{i-1}} = 0 \leftrightarrow \headWindow{\crun(\map{\crun}{\hrun}{i-1}(|\buf_{i-1}|))} = \infty$.
                        From (2), we get that $\crun(\map{\crun}{\hrun}{i-1}(|\buf_{i-1}|) + 1)$ is obtained by executing the \textsc{Step} rule of $\CtSpecInterfStep{}{}$ starting from $\crun(\map{\crun}{\hrun}{i-1}(|\buf_{i-1}|))$.
                        Therefore, $\headWindow{\crun(\map{\crun}{\hrun}{i-1}(|\buf_{i-1}|)+1)} = \headWindow{\crun(\map{\crun}{\hrun}{i-1}(|\buf_{i-1}|))} - 1$.
                        Moreover, we also have that $\nrMispred{\tup{m_{i-1},a_{i-1}}}{\buf_{i-1}} = \nrMispred{\tup{m_{i-1},a_{i-1}}}{\buf}$ because $\buf =  \tagged{\passign{\pc}{e}}{\notags} \concat \buf_{i-1}$.
                        There are  two cases:
                        \begin{description}
                            \item[$\nrMispred{\tup{m_{i-1},a_{i-1}}}{\buf_{i-1}} = 0$:]
                            Then $\headWindow{\crun(\map{\crun}{\hrun}{i-1}(|\buf_{i-1}|))} = \infty$ from (H.3.a.iii).
                            Therefore, we have that $\headWindow{\crun(\map{\crun}{\hrun}{i-1}(|\buf_{i-1}|)+1)} = \infty$ as well.
                            From this, $\nrMispred{\tup{m_{i-1},a_{i-1}}}{\buf} = 0$, $a_{i} = a_{i-1}$,  $m_i = m_{i-1}$, $\map{\crun}{\hrun}{i} = \map{\crun}{\hrun}{i-1}[\mathit{ln}(\hrun(i-1) ) +1 \mapsto \map{\crun}{\hrun}{i-1}(\mathit{ln}(\hrun(i-1) )) + 1]$, and $|\buf|  = \mathit{ln}(\hrun(i-1))+1$, we get $\nrMispred{\tup{m_{i},a_{i}}}{\buf} = 0 \leftrightarrow \headWindow{\crun(\map{\crun}{\hrun}{i}(|\buf|))} = \infty$.

                            \item[$\nrMispred{\tup{m_{i-1},a_{i-1}}}{\buf_{i-1}} \neq 0$:]
                            Then $\headWindow{\crun(\map{\crun}{\hrun}{i-1}(|\buf_{i-1}|))} \neq \infty$ from (H.3.a.iii).
                            Therefore, we have that $\headWindow{\crun(\map{\crun}{\hrun}{i-1}(|\buf_{i-1}|)+1)} \neq \infty$ as well.
                            From this, $\nrMispred{\tup{m_{i-1},a_{i-1}}}{\buf} \neq 0$, $a_{i} = a_{i-1}$,  $m_i = m_{i-1}$, $\map{\crun}{\hrun}{i} = \map{\crun}{\hrun}{i-1}[\mathit{ln}(\hrun(i-1) ) +1 \mapsto \map{\crun}{\hrun}{i-1}(\mathit{ln}(\hrun(i-1) )) + 1]$, and $|\buf|  = \mathit{ln}(\hrun(i-1))+1$, we get $\nrMispred{\tup{m_{i},a_{i}}}{\buf} = 0 \leftrightarrow \headWindow{\crun(\map{\crun}{\hrun}{i}(|\buf|))} = \infty$.
                        \end{description}
                        
                        \item[(iv):]
                        We need to show $\headWindow{\crun(\map{\crun}{\hrun}{i}(|\buf|))} > 0$.
                        As we have shown before, $\crun(\map{\crun}{\hrun}{i-1}(\mathit{ln}(\hrun(i-1) )) + 1)$ is obtained from $\crun(\map{\crun}{\hrun}{i-1}(\mathit{ln}(\hrun(i-1) )))$ by applying the \textsc{Step} rule of $\CtSpecInterfStep{}{}$.
                        Then, we get that either $\headWindow{ \crun( \map{\crun}{\hrun}{i}(|\buf|) )  } = \infty$ or  $\headWindow{ \crun( \map{\crun}{\hrun}{i}(|\buf|) ) } = \headWindow{ \crun( \map{\crun}{\hrun}{i}(|\buf_{i-1}|) ) } - 1$. 
                       	In the first case, then $\headWindow{\crun(\map{\crun}{\hrun}{i}(|\buf|))} > 0$ holds.
                       	In the second case, from $\map{\crun}{\hrun}{i}(|\buf_{i-1}|) =  \map{\crun}{\hrun}{i-1}( |\buf_{i-1}| )) $ and $\minOf{\map{\crun}{\hrun}{i-1}, \crun} > 1$, which follows from (2), we get  $\headWindow{\crun(\map{\crun}{\hrun}{i}(|\buf|))} > 0$.
                    \end{description}                  
                \end{description}

                \item[(b):] 
                From (H.b), we have that  $\tup{j,j'} \in \nextPairs{\map{\crun}{\hrun}{i-1}}$,
                $\headWindow{ \crun( \map{\crun}{\hrun}{i-1}(j) ) } = \headWindow{ \crun( \map{\crun}{\hrun}{i-1}(j') )  } = \infty$ or
                $\headWindow{ \crun( \map{\crun}{\hrun}{i-1}(j) ) } = \infty$ and $\headWindow{ \crun( \map{\crun}{\hrun}{i-1}(j') ) } = \wInterf$ or
                $\headWindow{ \crun( \map{\crun}{\hrun}{i-1}(j) ) } \neq \infty \wedge \headWindow{ \crun( \map{\crun}{\hrun}{i-1}(j') ) } = \headWindow{ \crun( \map{\crun}{\hrun}{i-1}(j) ) } - 1$.
                From $\map{\crun}{\hrun}{i} = \map{\crun}{\hrun}{i-1}[\mathit{ln}(\hrun(i-1) ) +1 \mapsto \map{\crun}{\hrun}{i-1}(\mathit{ln}(\hrun(i-1) )) + 1]$, the only interesting pair is $\tup{j,j'}  = \tup{|\buf_{i-1}|, |\buf_{i-1}|+1}$ since for all others pairs the relation follows from (H.b).
                For this pair, we have $\map{\crun}{\hrun}{i}(|\buf_{i-1}|) =  \map{\crun}{\hrun}{i-1}(\mathit{ln}(\hrun(i-1) )) $ and $ \map{\crun}{\hrun}{i}(|\buf_{i-1}|+1)   = \map{\crun}{\hrun}{i-1}(\mathit{ln}(\hrun(i-1) )) + 1$.
                As we have shown before, $\crun(\map{\crun}{\hrun}{i-1}(\mathit{ln}(\hrun(i-1) )) + 1)$ is obtained from $\crun(\map{\crun}{\hrun}{i-1}(\mathit{ln}(\hrun(i-1) )))$ by applying the \textsc{Step} rule of $\CtSpecInterfStep{}{}$.
                Therefore, either $\headWindow{ \crun( \map{\crun}{\hrun}{i}(j) ) } = \headWindow{ \crun( \map{\crun}{\hrun}{i}(j') )  } = \infty$ or  $\headWindow{ \crun( \map{\crun}{\hrun}{i}(j) ) } \neq \infty \wedge \headWindow{ \crun( \map{\crun}{\hrun}{i}(j') ) } = \headWindow{ \crun( \map{\crun}{\hrun}{i}(j) ) } - 1$.
                Hence, (b) holds.
            \end{description}

			\item[Rule \textsc{Fetch-Others-Hit}:]
            Then, $\buf_i = \tagged{p(  \apply{a_{i-1}}{\buf_{i-1}}(\pc) )}{\notags} \concat \tagged{\pmarkedassign{\pc}{  \apply{a_{i-1}}{\buf_{i-1}}(\pc) +1}}{\notags}\concat \buf_{i-1}$, $a_{i} = a_{i-1}$, $m_{i} = m_{i-1}$, $p( \apply{a_{i-1}}{\buf_{i-1}}  (\pc)) \neq \pjz{x}{\lbl}$, $p( \apply{a_{i-1}}{\buf_{i-1}}  (\pc)) \neq \pjmp{e}$, and $\map{\crun}{\hrun}{i} = \map{\crun}{\hrun}{i-1}[\mathit{ln}(\hrun(i-1) ) +2 \mapsto \map{\crun}{\hrun}{i-1}(\mathit{ln}(\hrun(i-1) )) + 1]$.

            We now show that all conditions hold:
            \begin{description}
                \item[(a):]
                Let $\buf$ be an arbitrary prefix in $\prefixes{\buf_{i}}$.
                There are two cases:
                \begin{description}
                    \item[$\buf \in \prefixes{\buf_{i-1}}$:]
                    The proof of this case is identical to the corresponding case for the  \textsc{Fetch-Jump-Hit} rule.

                    \item[$\buf \not\in \prefixes{\buf_{i-1}}$:]    
                    From $\buf_i = \tagged{p(  \apply{a_{i-1}}{\buf_{i-1}}(\pc))}{\notags} \concat \tagged{\pmarkedassign{\pc}{ \apply{a_{i-1}}{\buf_{i-1}}(\pc)+1}}{\notags}\concat \buf_{i-1}$, we get that $\buf = \buf_i$ and $|\buf| = |buf_{i-1}| + 2 = \mathit{ln}(\hrun(i-1))+2$.
                    From (H.3.a.i) and $\buf_{i-1} \in \prefixes{\buf_{i-1}}$, we have $\update{\tup{m_{i-1},a_{i-1}}}{\buf_{i-1}} =  \headConf{\crun(\map{\crun}{\hrun}{i-1}(|\buf_{i-1}|))}$.
                    Observe that:
                    \begin{description}
                        \item[(1):] 
                        From  $\update{\tup{m_{i-1},a_{i-1}}}{\buf_{i-1}} =  \headConf{\crun(\map{\crun}{\hrun}{i-1}(|\buf_{i-1}|))}$, we have, in particular, that $\update{\tup{m_{i-1},a_{i-1}}}{\buf_{i-1}}(\pc) = \headConf{\crun(\map{\crun}{\hrun}{i-1}(|\buf_{i-1}|))}(\pc)$.

                        \item[(2):] 
                        From (H.a.iii) applied to $\buf_{i-1}$ and the well-formedness of buffers, we get that $\maxOf{\map{\crun}{\hrun}{i-1}, \crun} = \infty$ (because the shortest prefix is always associated with window $\infty$).
                        From (H.b) and $\maxOf{\map{\crun}{\hrun}{i-1}, \crun} = \infty$, we have that $\minOf{\map{\crun}{\hrun}{i-1}, \crun} \geq \wInterf - |\buf_{i-1}|$ (because one can decrement $\wInterf$ at most once per prefix and there are  at most $|\buf_{i-1}|$ prefixes).
                        From this, $|\buf| \leq \wMuarch$, and $\wInterf > \wMuarch + 1$, we get $\minOf{\map{\crun}{\hrun}{i-1}, \crun} \geq \wInterf - \wMuarch > 1$.
                    \end{description}

                    We now show that (i)--(iv) hold:
                    \begin{description}
                        \item[(i):]
                        We need to show $\update{\tup{m_i,a_i}}{\buf} =  \headConf{\crun(\map{\crun}{\hrun}{i}(|\buf|))}$.
                        From (H.3.a.i) and $\buf_{i-1} \in \prefixes{\buf_{i-1}}$, we have $\update{\tup{m_{i-1},a_{i-1}}}{\buf_{i-1}} =  \headConf{\crun(\map{\crun}{\hrun}{i-1}(|\buf_{i-1}|))}$.
                        From (2), we get that $\crun(\map{\crun}{\hrun}{i-1}(|\buf_{i-1}|) + 1)$ is obtained by executing the \textsc{Step} rule of $\CtSpecInterfStep{}{}$ starting from $\crun(\map{\crun}{\hrun}{i-1}(|\buf_{i-1}|))$.
                        For simplicity, we assume that $p(\apply{a_{i-1}}{\buf_{i-1}}(\pc)) = \pload{x}{e}$ (the proof for the other cases is similar).
                        From (1) and $p(\apply{a_{i-1}}{\buf_{i-1}}(\pc)) = \pload{x}{e}$, we have that $\headConf{\crun(\map{\crun}{\hrun}{i-1}(|\buf_{i-1}|) + 1)} = \tup{m,a[x \mapsto m(\exprEval{e}{a}), \pc \mapsto a(\pc) +1] }$ where $\tup{m,a} = \headConf{\crun(\map{\crun}{\hrun}{i-1}(|\buf_{i-1}|) )}$.
                        From  $\update{\tup{m_{i-1},a_{i-1}}}{\buf_{i-1}} =  \headConf{\crun(\map{\crun}{\hrun}{i-1}(|\buf_{i-1}|))}$, we therefore get that $\headConf{\crun(\map{\crun}{\hrun}{i-1}(|\buf_{i-1}|) + 1)} = \update{\tup{m_{i-1},a_{i-1}}}{\buf_{i-1}}[x \mapsto m(\exprEval{e}{a}), \pc \mapsto a(\pc) +1]$ where $\tup{m,a} = \headConf{\crun(\map{\crun}{\hrun}{i-1}(|\buf_{i-1}|) )}$.
                        By leveraging $\update{\cdot}{\cdot}$'s definition and $p(\apply{a_{i-1}}{\buf_{i-1}}(\pc)) = \pload{x}{e}$, we get $\headConf{\crun(\map{\crun}{\hrun}{i-1}(|\buf_{i-1}|) + 1)} = \update{\tup{m_{i-1},a_{i-1}}}{(\tagged{p(  \apply{a_{i-1}}{\buf_{i-1}}(\pc))}{\notags} \concat \tagged{\pmarkedassign{\pc}{ \apply{a_{i-1}}{\buf_{i-1}}(\pc)+1}}{\notags}\concat \buf_{i-1}) } $.
                        From $\buf_i = \tagged{p(  \apply{a_{i-1}}{\buf_{i-1}}(\pc))}{\notags} \concat \tagged{\pmarkedassign{\pc}{ \apply{a_{i-1}}{\buf_{i-1}}(\pc)+1}}{\notags}\concat \buf_{i-1}$, $a_{i} = a_{i-1}$, and $m_i = m_{i-1}$, we get $\headConf{\crun(\map{\crun}{\hrun}{i-1}(|\buf_{i-1}|) + 1)} = \update{\tup{m_{i},a_{i}}}{\buf_{i} } $.
                        Finally, from $\map{\crun}{\hrun}{i} = \map{\crun}{\hrun}{i-1}[\mathit{ln}(\hrun(i-1) ) +2 \mapsto \map{\crun}{\hrun}{i-1}(\mathit{ln}(\hrun(i-1) )) + 1]$,  $|\buf_i|  = \mathit{ln}(\hrun(i-1))+2$, $\buf = \buf_i$, we get $ \update{\tup{m_{i},a_{i}}}{\buf } = \headConf{\crun(\map{\crun}{\hrun}{i}(|\buf|))}$.

                        \item[(ii):]
                        We need to show $\nrMispred{\tup{m_{i},a_{i}}}{\buf} + 1 \geq |\crun(\map{\crun}{\hrun}{i}(|\buf|))|$.
                        From (H.3.a.ii) and $\buf_{i-1} \in \prefixes{\buf_{i-1}}$,  we have $\nrMispred{\tup{m_{i-1},a_{i-1}}}{\buf_{i-1}} + 1 \geq |\crun(\map{\crun}{\hrun}{i-1}(|\buf_{i-1}|))|$.
                        From $\buf =  \tagged{\passign{\pc}{e}}{\notags} \concat \buf_{i-1}$, we get that $\nrMispred{\tup{m_{i-1},a_{i-1}}}{\buf_{i-1}} = \nrMispred{\tup{m_{i-1},a_{i-1}}}{\buf}$.
                        From (2), we get that $\crun(\map{\crun}{\hrun}{i-1}(|\buf_{i-1}|) + 1)$ is obtained by executing the \textsc{Step} rule of $\CtSpecInterfStep{}{}$ starting from $\crun(\map{\crun}{\hrun}{i-1}(|\buf_{i-1}|))$.
                        Therefore, $|\crun(\map{\crun}{\hrun}{i-1}(|\buf_{i-1}|))| = |\crun(\map{\crun}{\hrun}{i-1}(|\buf_{i-1}|)+1)|$.
                        Thus, we get $\nrMispred{\tup{m_{i-1},a_{i-1}}}{\buf} + 1 \geq |\crun(\map{\crun}{\hrun}{i-1}(|\buf_{i-1}|)+1)|$.
                        Finally, from  $a_{i} = a_{i-1}$,  $m_i = m_{i-1}$, $\map{\crun}{\hrun}{i} = \map{\crun}{\hrun}{i-1}[\mathit{ln}(\hrun(i-1) ) +2 \mapsto \map{\crun}{\hrun}{i-1}(\mathit{ln}(\hrun(i-1) )) + 1]$, and $|\buf|  = \mathit{ln}(\hrun(i-1))+2$, we get $\nrMispred{\tup{m_{i},a_{i}}}{\buf} + 1 \geq |\crun(\map{\crun}{\hrun}{i}(|\buf|))|$.

                        \item[(iii):]
                        We need to show $\nrMispred{\tup{m_{i},a_{i}}}{\buf} = 0 \leftrightarrow \headWindow{\crun(\map{\crun}{\hrun}{i}(|\buf|))} = \infty$.
                        From (H.3.a.iii) and $\buf_{i-1} \in \prefixes{\buf_{i-1}}$,  we have $\nrMispred{\tup{m_{i-1},a_{i-1}}}{\buf_{i-1}} = 0 \leftrightarrow \headWindow{\crun(\map{\crun}{\hrun}{i-1}(|\buf_{i-1}|))} = \infty$.
                        From (2), we get that $\crun(\map{\crun}{\hrun}{i-1}(|\buf_{i-1}|) + 1)$ is obtained by executing the \textsc{Step} rule of $\CtSpecInterfStep{}{}$ starting from $\crun(\map{\crun}{\hrun}{i-1}(|\buf_{i-1}|))$.
                        Therefore, $\headWindow{\crun(\map{\crun}{\hrun}{i-1}(|\buf_{i-1}|)+1)} = \headWindow{\crun(\map{\crun}{\hrun}{i-1}(|\buf_{i-1}|))} - 1$.
                        Moreover, we also have that $\nrMispred{\tup{m_{i-1},a_{i-1}}}{\buf_{i-1}} = \nrMispred{\tup{m_{i-1},a_{i-1}}}{\buf}$ because $\buf =  \tagged{\passign{\pc}{e}}{\notags} \concat \buf_{i-1}$.
                        There are  two cases:
                        \begin{description}
                            \item[$\nrMispred{\tup{m_{i-1},a_{i-1}}}{\buf_{i-1}} = 0$:]
                            Then $\headWindow{\crun(\map{\crun}{\hrun}{i-1}(|\buf_{i-1}|))} = \infty$ from (H.3.a.iii).
                            Therefore, we have that $\headWindow{\crun(\map{\crun}{\hrun}{i-1}(|\buf_{i-1}|)+1)} = \infty$ as well.
                            From this, $\nrMispred{\tup{m_{i-1},a_{i-1}}}{\buf} = 0$, $a_{i} = a_{i-1}$,  $m_i = m_{i-1}$, $\map{\crun}{\hrun}{i} = \map{\crun}{\hrun}{i-1}[\mathit{ln}(\hrun(i-1) ) +12\mapsto \map{\crun}{\hrun}{i-1}(\mathit{ln}(\hrun(i-1) )) + 1]$, and $|\buf|  = \mathit{ln}(\hrun(i-1))+2$, we get $\nrMispred{\tup{m_{i},a_{i}}}{\buf} = 0 \leftrightarrow \headWindow{\crun(\map{\crun}{\hrun}{i}(|\buf|))} = \infty$.

                            \item[$\nrMispred{\tup{m_{i-1},a_{i-1}}}{\buf_{i-1}} \neq 0$:]
                            Then $\headWindow{\crun(\map{\crun}{\hrun}{i-1}(|\buf_{i-1}|))} \neq \infty$ from (H.3.a.iii).
                            Therefore, we have that $\headWindow{\crun(\map{\crun}{\hrun}{i-1}(|\buf_{i-1}|)+1)} \neq \infty$ as well.
                            From this, $\nrMispred{\tup{m_{i-1},a_{i-1}}}{\buf} \neq 0$, $a_{i} = a_{i-1}$,  $m_i = m_{i-1}$, $\map{\crun}{\hrun}{i} = \map{\crun}{\hrun}{i-1}[\mathit{ln}(\hrun(i-1) ) +2 \mapsto \map{\crun}{\hrun}{i-1}(\mathit{ln}(\hrun(i-1) )) + 1]$, and $|\buf|  = \mathit{ln}(\hrun(i-1))+2$, we get $\nrMispred{\tup{m_{i},a_{i}}}{\buf} = 0 \leftrightarrow \headWindow{\crun(\map{\crun}{\hrun}{i}(|\buf|))} = \infty$.
                        \end{description}
                        
                        \item[(iv):]
                        We need to show $\headWindow{\crun(\map{\crun}{\hrun}{i}(|\buf|))} > 0$.
                        As we have shown before, $\crun(\map{\crun}{\hrun}{i-1}(\mathit{ln}(\hrun(i-1) )) + 1)$ is obtained from $\crun(\map{\crun}{\hrun}{i-1}(\mathit{ln}(\hrun(i-1) )))$ by applying the \textsc{Step} rule of $\CtSpecInterfStep{}{}$.
                        Then, we get that either $\headWindow{ \crun( \map{\crun}{\hrun}{i}(|\buf|) )  } = \infty$ or  $\headWindow{ \crun( \map{\crun}{\hrun}{i}(|\buf|) ) } = \headWindow{ \crun( \map{\crun}{\hrun}{i}(|\buf_{i-1}|) ) } - 1$. 
                       	In the first case, then $\headWindow{\crun(\map{\crun}{\hrun}{i}(|\buf|))} > 0$ holds.
                       	In the second case, from $\map{\crun}{\hrun}{i}(|\buf_{i-1}|) =  \map{\crun}{\hrun}{i-1}( |\buf_{i-1}| )) $ and $\minOf{\map{\crun}{\hrun}{i-1}, \crun} > 1$, which follows from (2), we get  $\headWindow{\crun(\map{\crun}{\hrun}{i}(|\buf|))} > 0$.
                    \end{description}                  
                \end{description}

                \item[(b):] 
                From (H.b), we have that  $\tup{j,j'} \in \nextPairs{\map{\crun}{\hrun}{i-1}}$,
                $\headWindow{ \crun( \map{\crun}{\hrun}{i-1}(j) ) } = \headWindow{ \crun( \map{\crun}{\hrun}{i-1}(j') )  } = \infty$ or
                $\headWindow{ \crun( \map{\crun}{\hrun}{i-1}(j) ) } = \infty$ and $\headWindow{ \crun( \map{\crun}{\hrun}{i-1}(j') ) } = \wInterf$ or
                $\headWindow{ \crun( \map{\crun}{\hrun}{i-1}(j) ) } \neq \infty \wedge \headWindow{ \crun( \map{\crun}{\hrun}{i-1}(j') ) } = \headWindow{ \crun( \map{\crun}{\hrun}{i-1}(j) ) } - 1$.
                From $\map{\crun}{\hrun}{i} = \map{\crun}{\hrun}{i-1}[\mathit{ln}(\hrun(i-1) ) +2 \mapsto \map{\crun}{\hrun}{i-1}(\mathit{ln}(\hrun(i-1) )) + 1]$, the only interesting pair is $\tup{j,j'}  = \tup{|\buf_{i-1}|, |\buf_{i-1}|+2}$ since for all others pairs the relation follows from (H.b).
                For this pair, we have $\map{\crun}{\hrun}{i}(|\buf_{i-1}|) =  \map{\crun}{\hrun}{i-1}(\mathit{ln}(\hrun(i-1) )) $ and $ \map{\crun}{\hrun}{i}(|\buf_{i-1}|+2)   = \map{\crun}{\hrun}{i-1}(\mathit{ln}(\hrun(i-1) )) + 1$.
                As we have shown before, $\crun(\map{\crun}{\hrun}{i-1}(\mathit{ln}(\hrun(i-1) )) + 1)$ is obtained from $\crun(\map{\crun}{\hrun}{i-1}(\mathit{ln}(\hrun(i-1) )))$ by applying the \textsc{Step} rule of $\CtSpecInterfStep{}{}$.
                Therefore, either $\headWindow{ \crun( \map{\crun}{\hrun}{i}(j) ) } = \headWindow{ \crun( \map{\crun}{\hrun}{i}(j') )  } = \infty$ or  $\headWindow{ \crun( \map{\crun}{\hrun}{i}(j) ) } \neq \infty \wedge \headWindow{ \crun( \map{\crun}{\hrun}{i}(j') ) } = \headWindow{ \crun( \map{\crun}{\hrun}{i}(j) ) } - 1$.
                Hence, (b) holds.
            \end{description}

			\item[Rule \textsc{Fetch-Miss}:]
            Then, $\buf_{i} = \buf_{i-1}$, $m_{i} = m_{a-1}$, and $a_{i} = a_{i-1}$.
            Moreover, $\map{\crun}{\hrun}{i} = \map{\crun}{\hrun}{i-1}$.
            We now show that all conditions hold:
            \begin{description}
                \item[(a):]
                Let $\buf$ be an arbitrary prefix in $\prefixes{\buf_{i}}$.
                We now show that (i)--(iv) hold:
                \begin{description}
                    \item[(i):]
                    We need to show $\update{\tup{m_i,a_i}}{\buf} =  \headConf{\crun(\map{\crun}{\hrun}{i}(|\buf|))}$.
                    From (H.3.a.i), $\buf \in \prefixes{\buf_{i}}$, and $\buf_i = \buf_{i-1}$, we have $\update{\tup{m_{i-1},a_{i-1}}}{\buf} =  \headConf{\crun(\map{\crun}{\hrun}{i-1}(|\buf|))}$.
                    Finally, from this,  $m_{i} = m_{a-1}$, $a_{i} = a_{i-1}$, and $\map{\crun}{\hrun}{i} = \map{\crun}{\hrun}{i-1}$, we get $\update{\tup{m_i,a_i}}{\buf} =  \headConf{\crun(\map{\crun}{\hrun}{i}(|\buf|))}$.

                    \item[(ii):]
                    We need to show $\nrMispred{\tup{m_{i},a_{i}}}{\buf} + 1 \geq |\crun(\map{\crun}{\hrun}{i}(|\buf|))|$.
                    From (H.3.a.ii), $\buf \in \prefixes{\buf_{i}}$, and $\buf_i = \buf_{i-1}$, we have $\nrMispred{\tup{m_{i-1},a_{i-1}}}{\buf} + 1 \geq |\crun(\map{\crun}{\hrun}{i-1}(|\buf|))|$.
                    Finally, from this,  $m_{i} = m_{a-1}$, $a_{i} = a_{i-1}$, and $\map{\crun}{\hrun}{i} = \map{\crun}{\hrun}{i-1}$, we get $\nrMispred{\tup{m_{i},a_{i}}}{\buf} + 1 \geq |\crun(\map{\crun}{\hrun}{i}(|\buf|))|$.

                    \item[(iii):]
                    We need to show $\nrMispred{\tup{m_{i},a_{i}}}{\buf} = 0 \leftrightarrow \headWindow{\crun(\map{\crun}{\hrun}{i}(|\buf|))} = \infty$.
                    From (H.3.a.iii), $\buf \in \prefixes{\buf_{i}}$, and $\buf_i = \buf_{i-1}$, we have $\nrMispred{\tup{m_{i-1},a_{i-1}}}{\buf} = 0 \leftrightarrow \headWindow{\crun(\map{\crun}{\hrun}{i-1}(|\buf|))} = \infty$.
                    Finally, from this,  $m_{i} = m_{a-1}$, $a_{i} = a_{i-1}$, and $\map{\crun}{\hrun}{i} = \map{\crun}{\hrun}{i-1}$, we get $\nrMispred{\tup{m_{i},a_{i}}}{\buf} = 0 \leftrightarrow \headWindow{\crun(\map{\crun}{\hrun}{i}(|\buf|))} = \infty$.

					\item[(iv):]
                    We need to show $\headWindow{\crun(\map{\crun}{\hrun}{i}(|\buf|))} > 0$.
                    From (H.3.a.iv), $\buf \in \prefixes{\buf_{i}}$, and $\buf_i = \buf_{i-1}$, we get $\headWindow{\crun(\map{\crun}{\hrun}{i-1}(|\buf|))} > 0$.
                    From this and $\map{\crun}{\hrun}{i} = \map{\crun}{\hrun}{i-1}$, we get $\headWindow{\crun(\map{\crun}{\hrun}{i}(|\buf|))} > 0$.
                \end{description}

                \item[(b):]
                (b) immediately follows from (H.3.b) and $\map{\crun}{\hrun}{i} = \map{\crun}{\hrun}{i-1}$.
            \end{description}
            
		\end{description}
		This completes the proof for the $\fetch{}$ case.

        \item[$\SchedNext(C_{i-1}) = \execute{j}$:] 
        We now proceed by case distinction on the applied $\execute{}$ rule:
        \begin{description}
            \item[Rule \textsc{Execute-Load-Hit}:]
            Therefore, we have that $\pbarrier \not\in \buf_{i-1}[0..j-1]$, $\pstore{x'}{e'} \not\in \buf_{i-1}[0..j-1]$, $m_i = m_{i-1}$, $a_i = a_{i-1}$, $\buf_{i} = \buf_{i-1}[0..j-1] \concat  \tagged{\passign{x}{m_{i-1}(\exprEval{e}{ \apply{\buf_{i-1}[0..j-1]}{a_{i-1} } }) }}{T} \concat \buf_{i-1}[j+1..|\buf_{i-1}|]$, and $\elt{\buf_{i-1}}{j} = \tagged{\pload{x}{e}}{T}$.
            Moreover, we also have that $\neg isRb_{\hrun}(i,j)$ holds and, therefore, $\map{\crun}{\hrun}{i} = \map{\crun}{\hrun}{i-1}$.

            We now show that all  conditions hold:
            \begin{description}
                \item[(a):]
                Let $\buf$ be an arbitrary prefix in $\prefixes{\buf_{i}}$.
                There are two cases:
                \begin{description}
                    \item[$\buf \in \prefixes{\buf_{i-1}}$:]
                    Then, (i)--(iv) follow from $\buf \in \prefixes{\buf_{i-1}}$, $m_i = m_{i-1}$, $a_i = a_{i-1}$, $\map{\crun}{\hrun}{i} = \map{\crun}{\hrun}{i-1}$, and (H.3.a.i)--(H.3.a.iv).

                    \item[$\buf \not\in \prefixes{\buf_{i-1}}$:] 
                    Then, $\buf = \buf_{i-1}[0..j-1] \concat \tagged{\passign{x}{m_{i-1}(\exprEval{e}{ \apply{\buf_{i-1}[0..j-1]}{a_{i-1} } }) }}{T} \concat \buf_{i-1}[j+1..k]$ for some $k \in \Nat$.
                    Observe also that $\buf_{i-1}[0..k] \in \prefixes{\buf_{i-1}}$.

                    We now show that (i)--(iv) hold:
                    \begin{description}
                        \item[(i):]
                        We need to show $\update{\tup{m_i,a_i}}{\buf} =  \headConf{\crun(\map{\crun}{\hrun}{i}(|\buf|))}$.
                        From (H.3.a.i) and $\buf_{i-1}[0..k] \in \prefixes{\buf_{i-1}}$, we get that $\update{\tup{m_{i-1},a_{i-1}}}{\buf_{i-1}[0..k]} =  \headConf{\crun(\map{\crun}{\hrun}{i-1}(|\buf_{i-1}[0..k]|))}$.
                        From $\pstore{x'}{e'} \not\in \buf_{i-1}[0..j-1]$ and $\buf = \buf_{i-1}[0..j-1] \concat \tagged{\passign{x}{m_{i-1}(\exprEval{e}{ \apply{\buf_{i-1}[0..j-1]}{a_{i-1} } }) }}{T} \concat \buf_{i-1}[j+1..k]$, we immediately get that $\update{\tup{m_{i-1},a_{i-1}}}{\buf_{i-1}[0..k]} = \update{\tup{m_{i-1},a_{i-1}}}{\buf}$.
                        Therefore, we have $\update{\tup{m_{i-1},a_{i-1}}}{\buf} =  \headConf{\crun(\map{\crun}{\hrun}{i-1}(|\buf_{i-1}[0..k]|))}$.
                        Finally, from this, $m_i = m_{i-1}$, $a_i = a_{i-1}$,  $|\buf| = |\buf_{i-1}[0..k]|$, and $\map{\crun}{\hrun}{i} = \map{\crun}{\hrun}{i-1}$, we have $\update{\tup{m_{i},a_{i}}}{\buf} =  \headConf{\crun(\map{\crun}{\hrun}{i-1}(|\buf|))}$.

                        \item[(ii):]
                        We need to show $\nrMispred{\tup{m_{i},a_{i}}}{\buf} + 1 \geq |\crun(\map{\crun}{\hrun}{i}(|\buf|))|$.
                        From (H.3.a.ii) and $\buf_{i-1}[0..k] \in \prefixes{\buf_{i-1}}$, we have $\nrMispred{\tup{m_{i-1},a_{i-1}}}{\buf_{i-1}[0..k]} + 1 \geq |\crun(\map{\crun}{\hrun}{i-1}(|\buf_{i-1}[0..k]|))|$.
                        From  this,  $m_{i} = m_{a-1}$, $a_{i} = a_{i-1}$, and $\map{\crun}{\hrun}{i} = \map{\crun}{\hrun}{i-1}$, we get $\nrMispred{\tup{m_{i},a_{i}}}{\buf_{i-1}[0..k]} + 1 \geq |\crun(\map{\crun}{\hrun}{i}(|\buf_{i-1}[0..k]|))|$.
                        Finally, from $\nrMispred{\tup{m_{i},a_{i}}}{\buf_{i-1}[0..k]} = \nrMispred{\tup{m_{i},a_{i}}}{\buf}$ and $|\buf_{i-1}[0..k]| = |\buf|$, we get $\nrMispred{\tup{m_{i},a_{i}}}{\buf} + 1 \geq |\crun(\map{\crun}{\hrun}{i}(|\buf|))|$.

                        \item[(iii):]
                        We need to show $\nrMispred{\tup{m_{i},a_{i}}}{\buf} = 0 \leftrightarrow \headWindow{\crun(\map{\crun}{\hrun}{i}(|\buf|))} = \infty$.
                        From (H.3.a.iii) and $\buf_{i-1}[0..k]  \in \prefixes{\buf_{i-1}}$, we have $\nrMispred{\tup{m_{i-1},a_{i-1}}}{\buf_{i-1}[0..k] } = 0 \leftrightarrow \headWindow{\crun(\map{\crun}{\hrun}{i-1}(|\buf_{i-1}[0..k] |))} = \infty$.
                        Finally, from this,  $m_{i} = m_{a-1}$, $a_{i} = a_{i-1}$,  $\map{\crun}{\hrun}{i} = \map{\crun}{\hrun}{i-1}$, $\nrMispred{\tup{m_{i},a_{i}}}{\buf_{i-1}[0..k]} = \nrMispred{\tup{m_{i},a_{i}}}{\buf}$, and $|\buf_{i-1}[0..k]| = |\buf|$, we get $\nrMispred{\tup{m_{i},a_{i}}}{\buf} = 0 \leftrightarrow \headWindow{\crun(\map{\crun}{\hrun}{i}(|\buf|))} = \infty$.
                        
                        \item[(iv):]
	                    We need to show $\headWindow{\crun(\map{\crun}{\hrun}{i}(|\buf|))} > 0$.
	                    From (H.3.a.iv) and $\buf_{i-1}[0..k]  \in \prefixes{\buf_{i-1}}$, we have $\headWindow{\crun(\map{\crun}{\hrun}{i-1}(|\buf|))} > 0$.
	                    From this and $\map{\crun}{\hrun}{i} = \map{\crun}{\hrun}{i-1}$, we get $\headWindow{\crun(\map{\crun}{\hrun}{i}(|\buf|))} > 0$.

                    \end{description}
                \end{description}
           
                \item[(b):]
                (b) immediately follows from (H.3.b) and  $\map{\crun}{\hrun}{i} = \map{\crun}{\hrun}{i-1}$. 
            \end{description}            

            \item[Rule \textsc{Execute-Load-Miss}:]
            Therefore, we have that $m_i = m_{i-1}$, $a_i = a_{i-1}$, $\buf_{i} = \buf_{i-1}$.
            Moreover, we also have that $\neg isRb_{\hrun}(i,j)$ holds and, therefore, $\map{\crun}{\hrun}{i} = \map{\crun}{\hrun}{i-1}$.
            Then, the proofs of (a) and (b) are similar to that of the rule \textsc{Fetch-Miss} in the $\SchedNext(C_{i-1}) = \fetch{}$ case.

            \item[Rule \textsc{Execute-Branch-Commit}:]
            Therefore, we have that $\pbarrier \not\in \buf_{i-1}[0..j-1]$, $m_i = m_{i-1}$, $a_i = a_{i-1}$, $\buf_{i} = \buf_{i-1}[0..j-1] \concat  \tagged{\passign{\pc}{ \lbl  }}{\notags}   \concat \buf_{i-1}[j+1..|\buf_{i-1}|]$, $\elt{\buf_{i-1}}{j} = \tagged{\passign{\pc}{\lbl}}{\lbl_0} $, $p(\lbl_0) = \pjz{x}{\lbl''}$, and $( \apply{\buf_{i-1}[0..j-1]}{a_{i-1}}(x) = 0 \wedge \lbl = \lbl'') \vee (\apply{\buf_{i-1}[0..j-1]}{a_{i-1}}(x) \in \Val \setminus \{0,\bot\} \wedge \lbl = \ell_0+1)$.
            Moreover, we also have that $\neg isRb_{\hrun}(i,j)$ holds (because $( \apply{\buf_{i-1}[0..j-1]}{a_{i-1}}(x) = 0 \wedge \lbl = \lbl'') \vee (\apply{\buf_{i-1}[0..j-1]}{a_{i-1}}(x) \in \Val \setminus \{0,\bot\} \wedge \lbl = \ell_0+1)$ holds) and, therefore, $\map{\crun}{\hrun}{i} = \map{\crun}{\hrun}{i-1}$.

            We now show that all  conditions hold:
            \begin{description}
                \item[(a):]
                Let $\buf$ be an arbitrary prefix in $\prefixes{\buf_{i}}$.
                There are two cases:
                \begin{description}
                    \item[$\buf \in \prefixes{\buf_{i-1}}$:]
                    Then, (i)--(iv) follow from $\buf \in \prefixes{\buf_{i-1}}$, $m_i = m_{i-1}$, $a_i = a_{i-1}$, $\map{\crun}{\hrun}{i} = \map{\crun}{\hrun}{i-1}$, and (H.3.a.i)--(H.3.a.iv).

                    \item[$\buf \not\in \prefixes{\buf_{i-1}}$:] 
                    Then, $\buf = \buf_{i-1}[0..j-1] \concat \tagged{\passign{\pc}{ \lbl  }}{\notags} \concat \buf_{i-1}[j+1..k]$ for some $k \in \Nat$.
                    Observe also that $\buf_{i-1}[0..k] \in \prefixes{\buf_{i-1}}$.

                    We now show that (i)--(iv) hold:
                    \begin{description}
                        \item[(i):]
                        We need to show $\update{\tup{m_i,a_i}}{\buf} =  \headConf{\crun(\map{\crun}{\hrun}{i}(|\buf|))}$.
                        From (H.3.a.i) and $\buf_{i-1}[0..k] \in \prefixes{\buf_{i-1}}$, we get that $\update{\tup{m_{i-1},a_{i-1}}}{\buf_{i-1}[0..k]} =  \headConf{\crun(\map{\crun}{\hrun}{i-1}(|\buf_{i-1}[0..k]|))}$.
                        From $\elt{\buf_{i-1}}{j} = \tagged{\passign{\pc}{\lbl}}{\lbl_0} $ and $\buf = \buf_{i-1}[0..j-1] \concat \tagged{\passign{\pc}{ \lbl  }}{\notags} \concat \buf_{i-1}[j+1..k]$, we immediately get that $\update{\tup{m_{i-1},a_{i-1}}}{\buf_{i-1}[0..k]} = \update{\tup{m_{i-1},a_{i-1}}}{\buf}$.
                        Therefore, we have $\update{\tup{m_{i-1},a_{i-1}}}{\buf} =  \headConf{\crun(\map{\crun}{\hrun}{i-1}(|\buf_{i-1}[0..k]|))}$.
                        Finally, from this, $m_i = m_{i-1}$, $a_i = a_{i-1}$,  $|\buf| = |\buf_{i-1}[0..k]|$, and $\map{\crun}{\hrun}{i} = \map{\crun}{\hrun}{i-1}$, we have $\update{\tup{m_{i},a_{i}}}{\buf} =  \headConf{\crun(\map{\crun}{\hrun}{i-1}(|\buf|))}$.

                        \item[(ii):]
                        We need to show $\nrMispred{\tup{m_{i},a_{i}}}{\buf} + 1 \geq |\crun(\map{\crun}{\hrun}{i}(|\buf|))|$.
                        From (H.3.a.ii) and $\buf_{i-1}[0..k] \in \prefixes{\buf_{i-1}}$, we have $\nrMispred{\tup{m_{i-1},a_{i-1}}}{\buf_{i-1}[0..k]} + 1 \geq |\crun(\map{\crun}{\hrun}{i-1}(|\buf_{i-1}[0..k]|))|$.
                        From  this,  $m_{i} = m_{a-1}$, $a_{i} = a_{i-1}$, and $\map{\crun}{\hrun}{i} = \map{\crun}{\hrun}{i-1}$, we get $\nrMispred{\tup{m_{i},a_{i}}}{\buf_{i-1}[0..k]} + 1 \geq |\crun(\map{\crun}{\hrun}{i}(|\buf_{i-1}[0..k]|))|$.
                        Finally, from $\nrMispred{\tup{m_{i},a_{i}}}{\buf_{i-1}[0..k]} = \nrMispred{\tup{m_{i},a_{i}}}{\buf}$ (which follows from $( \apply{\buf_{i-1}[0..j-1]}{a_{i-1}}(x) = 0 \wedge \lbl = \lbl'') \vee (\apply{\buf_{i-1}[0..j-1]}{a_{i-1}}(x) \in \Val \setminus \{0,\bot\} \wedge \lbl = \ell_0+1)$) and $|\buf_{i-1}[0..k]| = |\buf|$, we get $\nrMispred{\tup{m_{i},a_{i}}}{\buf} + 1 \geq |\crun(\map{\crun}{\hrun}{i}(|\buf|))|$.

                        \item[(iii):]
                        We need to show $\nrMispred{\tup{m_{i},a_{i}}}{\buf} = 0 \leftrightarrow \headWindow{\crun(\map{\crun}{\hrun}{i}(|\buf|))} = \infty$.
                        From (H.3.a.iii) and $\buf_{i-1}[0..k]  \in \prefixes{\buf_{i-1}}$, we have $\nrMispred{\tup{m_{i-1},a_{i-1}}}{\buf_{i-1}[0..k] } = 0 \leftrightarrow \headWindow{\crun(\map{\crun}{\hrun}{i-1}(|\buf_{i-1}[0..k] |))} = \infty$.
                        Finally, from this,  $m_{i} = m_{a-1}$, $a_{i} = a_{i-1}$,  $\map{\crun}{\hrun}{i} = \map{\crun}{\hrun}{i-1}$, $\nrMispred{\tup{m_{i},a_{i}}}{\buf_{i-1}[0..k]} = \nrMispred{\tup{m_{i},a_{i}}}{\buf}$, and $|\buf_{i-1}[0..k]| = |\buf|$, we get $\nrMispred{\tup{m_{i},a_{i}}}{\buf} = 0 \leftrightarrow \headWindow{\crun(\map{\crun}{\hrun}{i}(|\buf|))} = \infty$.
                        
                        \item[(iv):]
	                   	We need to show $\headWindow{\crun(\map{\crun}{\hrun}{i}(|\buf|))} > 0$.
	                    From (H.3.a.iv) and $\buf_{i-1}[0..k]  \in \prefixes{\buf_{i-1}}$, we have $\headWindow{\crun(\map{\crun}{\hrun}{i-1}(|\buf|))} > 0$.
	                    From this and $\map{\crun}{\hrun}{i} = \map{\crun}{\hrun}{i-1}$, we get $\headWindow{\crun(\map{\crun}{\hrun}{i}(|\buf|))} > 0$.
                    \end{description}
                \end{description}
           
                \item[(b):]
                (b) immediately follows from (H.3.b) and  $\map{\crun}{\hrun}{i} = \map{\crun}{\hrun}{i-1}$. 
            \end{description}

            \item[Rule \textsc{Execute-Branch-Rollback}:]
            Therefore, we have that $\pbarrier \not\in \buf_{i-1}[0..j-1]$, $m_i = m_{i-1}$, $a_i = a_{i-1}$, $\buf_{i} = \buf_{i-1}[0..j-1] \concat  \tagged{\passign{\pc}{ \lbl  }}{\notags}$, $\elt{\buf_{i-1}}{j} = \tagged{\passign{\pc}{\lbl}}{\lbl_0} $, $p(\lbl_0) = \pjz{x}{\lbl''}$, and $(a'(x) = 0 \wedge \lbl \neq \lbl'') \vee (a'(x) \in \Val \setminus \{0,\bot\} \wedge \lbl \neq \ell_0+1)$.
            Moreover, we also have that $ isRb_{\hrun}(i,j)$ holds (because $( \apply{\buf_{i-1}[0..j-1]}{a_{i-1}}(x) = 0 \wedge \lbl \neq \lbl'') \vee (\apply{\buf_{i-1}[0..j-1]}{a_{i-1}}(x) \in \Val \setminus \{0,\bot\} \wedge \lbl \neq \ell_0+1)$ holds) and, therefore, $\map{\crun}{\hrun}{i} =  clip( \map{\crun}{\hrun}{i-1} , j  )[j \mapsto rb_{\crun}( \map{\crun}{\hrun}{i-1}(j-1) ) ]$.

            We now show that all conditions hold:
            \begin{description}
                \item[(a):]
                Let $\buf$ be an arbitrary prefix in $\prefixes{\buf_{i}}$.
                There are two cases:
                \begin{description}
                    \item[$\buf \in \prefixes{\buf_{i-1}}$:]
                    Then, (i)--(iv) follow from $\buf \in \prefixes{\buf_{i-1}}$, $m_i = m_{i-1}$, $a_i = a_{i-1}$, $\map{\crun}{\hrun}{i} = \map{\crun}{\hrun}{i-1}$, and (H.3.a.i)--(H.3.a.iv).

                    \item[$\buf \not\in \prefixes{\buf_{i-1}}$:] 
                    Then, $\buf = \buf_i = \buf_{i-1}[0..j-1] \concat \tagged{\passign{\pc}{ \lbl  }}{\notags}$.
                    Observe that both $\buf_{i-1}[0..j-1]$ and $\buf_{i-1}[0..j-1]\concat \tagged{\passign{\pc}{\lbl}}{\lbl_0} $ are in $\prefixes{\buf_{i-1}}$.

                    We now show that (i)--(iv) hold:
                    \begin{description}
                        \item[(i):]
                        We need to show $\update{\tup{m_i,a_i}}{\buf} =  \headConf{\crun(\map{\crun}{\hrun}{i}(|\buf|))}$.
                        From (H.3.a.i) and $\buf_{i-1}[0..j-1] \in \prefixes{\buf_{i-1}}$, we get that $\update{\tup{m_{i-1},a_{i-1}}}{\buf_{i-1}[0..j-1]} =  \headConf{\crun(\map{\crun}{\hrun}{i-1}(|\buf_{i-1}[0..j-1]|))}$.
                        From applying Lemma~\ref{lemma:vanilla:contract-rollback} to $\headConf{\crun(\map{\crun}{\hrun}{i-1}(|\buf_{i-1}[0..j-1]|))}$ and $(a'(x) = 0 \wedge \lbl \neq \lbl'') \vee (a'(x) \in \Val \setminus \{0,\bot\} \wedge \lbl \neq \ell_0+1)$, we get that $\headConf{\crun(rb_{\crun}( \map{\crun}{\hrun}{i-1}(j-1) ))} = \headConf{\crun(\map{\crun}{\hrun}{i-1}(|\buf_{i-1}[0..j-1]|))}[\pc \mapsto \lbl]$.
                        From this and $\update{\tup{m_{i-1},a_{i-1}}}{\buf_{i-1}[0..j-1]} =  \headConf{\crun(\map{\crun}{\hrun}{i-1}(|\buf_{i-1}[0..j-1]|))}$, we get $\headConf{\crun(rb_{\crun}( \map{\crun}{\hrun}{i-1}(j-1) ))} = \update{\tup{m_{i-1},a_{i-1}}}{\buf_{i-1}[0..j-1]}[\pc \mapsto \lbl]$.
                        From this and $\update{\cdot}{\cdot}$'s definition, we get $\headConf{\crun(rb_{\crun}( \map{\crun}{\hrun}{i-1}(j-1) ))} = \update{\tup{m_{i-1},a_{i-1}}}{\buf_{i-1}[0..j-1] \concat \tagged{\passign{\pc}{\lbl}}{\notags}}$.
                        From this, $m_i = m_{i-1}$, $a_i = a_{i-1}$, and $\buf = \buf_{i-1}[0..j-1] \concat \tagged{\passign{\pc}{\lbl}}{\notags}$, we get $\update{\tup{m_{i},a_{i}}}{\buf} = \headConf{\crun(rb_{\crun}( \map{\crun}{\hrun}{i-1}(j-1) ))}$.
                        Finally, from $|\buf| = j$ and $\map{\crun}{\hrun}{i} = clip( \map{\crun}{\hrun}{i-1} , j  )[j \mapsto rb_{\crun}( \map{\crun}{\hrun}{i-1}(j-1) ) $, we get $\update{\tup{m_{i},a_{i}}}{\buf} = \headConf{\crun(\map{\crun}{\hrun}{i}(|\buf|) )}$.
                        
                        \item[(ii):]
                        We need to show $\nrMispred{\tup{m_{i},a_{i}}}{\buf} + 1 \geq |\crun(\map{\crun}{\hrun}{i}(|\buf|))|$.
                        From (H.3.a.ii) and $\buf_{i-1}[0..j-1] \in \prefixes{\buf_{i-1}}$, we have $\nrMispred{\tup{m_{i-1},a_{i-1}}}{\buf_{i-1}[0..j-1]} + 1 \geq |\crun(\map{\crun}{\hrun}{i-1}(|\buf_{i-1}[0..j-1]|))|$.
                        From applying Lemma~\ref{lemma:vanilla:contract-rollback} to $\crun(\map{\crun}{\hrun}{i-1}(|\buf_{i-1}[0..j-1]|))$, we get that $|\crun(\map{\crun}{\hrun}{i-1}(|\buf_{i-1}[0..j-1]|))| = |\crun(rb_{\crun}( \map{\crun}{\hrun}{i-1}(j-1) ))|$.
                        Moreover, from $\buf = \buf_{i-1}[0..j-1] \concat \tagged{\passign{\pc}{ \lbl  }}{\notags}$, we get $\nrMispred{\tup{m_{i-1},a_{i-1}}}{\buf} = \nrMispred{\tup{m_{i-1},a_{i-1}}}{\buf_{i-1}[0..j-1]}$.
                        Therefore, we get $\nrMispred{\tup{m_{i-1},a_{i-1}}}{\buf} + 1 \geq |\crun(rb_{\crun}( \map{\crun}{\hrun}{i-1}(j-1) ))|$.
                        Finally, from  this,  $m_{i} = m_{a-1}$, $a_{i} = a_{i-1}$, $|\buf| = j$, and $\map{\crun}{\hrun}{i} = clip( \map{\crun}{\hrun}{i-1} , j  )[j \mapsto rb_{\crun}( \map{\crun}{\hrun}{i-1}(j-1) ) $,  we get $\nrMispred{\tup{m_{i},a_{i}}}{\buf} + 1 \geq |\crun(\map{\crun}{\hrun}{i}(|\buf|) ))|$.
              
                        \item[(iii):]
                        We need to show $\nrMispred{\tup{m_{i},a_{i}}}{\buf} = 0 \leftrightarrow \headWindow{\crun(\map{\crun}{\hrun}{i}(|\buf|))} = \infty$.
                        From (H.3.a.iii) and $\buf_{i-1}[0..j-1]  \in \prefixes{\buf_{i-1}}$, we have $\nrMispred{\tup{m_{i-1},a_{i-1}}}{\buf_{i-1}[0..j-1] } = 0 \leftrightarrow \headWindow{\crun(\map{\crun}{\hrun}{i-1}(|\buf_{i-1}[0..j-1] |))} = \infty$.
                        From $\buf = \buf_{i-1}[0..j-1] \concat \tagged{\passign{\pc}{ \lbl  }}{\notags}$, we get $\nrMispred{\tup{m_{i-1},a_{i-1}}}{\buf} = \nrMispred{\tup{m_{i-1},a_{i-1}}}{\buf_{i-1}[0..j-1]}$.
                        From applying Lemma~\ref{lemma:vanilla:contract-rollback} to $\crun(\map{\crun}{\hrun}{i-1}(|\buf_{i-1}[0..j-1]|))$, we get that $\headWindow{rb_{\crun}( \map{\crun}{\hrun}{i-1}(j-1) )} = \headWindow{\crun(\map{\crun}{\hrun}{i-1}(|\buf_{i-1}[0..j-1] |))} -1 $.
                        There are two cases:
                        \begin{description}
                            \item[$\nrMispred{\tup{m_{i-1},a_{i-1}}}{\buf_{i-1}[0..j-1] } = 0$:]
                            Then, $\headWindow{\crun(\map{\crun}{\hrun}{i}(|\buf|))} = \infty$ and, therefore, $\headWindow{rb_{\crun}( \map{\crun}{\hrun}{i-1}(j-1) )} = \infty$.
                            Moreover, $\nrMispred{\tup{m_{i-1},a_{i-1}}}{\buf} = 0$ as well.
                            From this, $m_{i} = m_{a-1}$, $a_{i} = a_{i-1}$, $|\buf| = j$, and $\map{\crun}{\hrun}{i} = clip( \map{\crun}{\hrun}{i-1} , j  )[j \mapsto rb_{\crun}( \map{\crun}{\hrun}{i-1}(j-1) ) $, we get $\nrMispred{\tup{m_{i},a_{i}}}{\buf} = 0 \leftrightarrow \headWindow{\crun(\map{\crun}{\hrun}{i}(|\buf|))} = \infty$.
                            
                            \item[$\nrMispred{\tup{m_{i-1},a_{i-1}}}{\buf_{i-1}[0..j-1] } \neq 0$:]
                            Then, $\headWindow{\crun(\map{\crun}{\hrun}{i}(|\buf|))} \neq \infty$ and, therefore, $\headWindow{rb_{\crun}( \map{\crun}{\hrun}{i-1}(j-1) )} \neq \infty$.
                            Moreover, $\nrMispred{\tup{m_{i-1},a_{i-1}}}{\buf} \neq 0$ as well.
                            From this, $m_{i} = m_{a-1}$, $a_{i} = a_{i-1}$, $|\buf| = j$, and $\map{\crun}{\hrun}{i} = clip( \map{\crun}{\hrun}{i-1} , j  )[j \mapsto rb_{\crun}( \map{\crun}{\hrun}{i-1}(j-1) ) $, we get $\nrMispred{\tup{m_{i},a_{i}}}{\buf} = 0 \leftrightarrow \headWindow{\crun(\map{\crun}{\hrun}{i}(|\buf|))} = \infty$.
                        \end{description}
                        
                       \item[(iv):]
	                   We need to show $\headWindow{\crun(\map{\crun}{\hrun}{i}(|\buf|))} > 0$.
	                   We know that $\map{\crun}{\hrun}{i}(|\buf|) = rb_{\crun}( \map{\crun}{\hrun}{i-1}(j-1) )$.
		               From Lemma~\ref{lemma:vanilla:contract-rollback}, we have that  either $\headWindow{ \crun( \map{\crun}{\hrun}{i}(|\buf|) )  } = \infty$ or  $headWindow{ \crun( \map{\crun}{\hrun}{i}(|\buf|) ) } = \map{\crun}{\hrun}{i-1}(j)$.
		               In the first case, then $\headWindow{\crun(\map{\crun}{\hrun}{i}(|\buf|))} > 0$ holds.
		               In the second case, $\headWindow{\crun(\map{\crun}{\hrun}{i}(|\buf|))} > 0$ follows from $headWindow{ \crun( \map{\crun}{\hrun}{i}(|\buf|) ) } = \map{\crun}{\hrun}{i-1}(j)$, $\buf_{i-1}[0..j]\in \prefixes{\buf_{i-1}}$, $|\buf| = \buf_{i-1}[0..j] = j$, and (H.3.a.iv).

                    \end{description}
                \end{description}
                \item[(b):]
                From (H.b), we have that  $\tup{j,j'} \in \nextPairs{\map{\crun}{\hrun}{i-1}}$,
                $\headWindow{ \crun( \map{\crun}{\hrun}{i-1}(j) ) } = \headWindow{ \crun( \map{\crun}{\hrun}{i-1}(j') )  } = \infty$ or
                $\headWindow{ \crun( \map{\crun}{\hrun}{i-1}(j) ) } = \infty$ and $\headWindow{ \crun( \map{\crun}{\hrun}{i-1}(j') ) } = \wInterf$ or
                $\headWindow{ \crun( \map{\crun}{\hrun}{i-1}(j) ) } \neq \infty \wedge \headWindow{ \crun( \map{\crun}{\hrun}{i-1}(j') ) } = \headWindow{ \crun( \map{\crun}{\hrun}{i-1}(j) ) } - 1$.
                From $\map{\crun}{\hrun}{i} = clip( \map{\crun}{\hrun}{i-1} , j  )[j \mapsto rb_{\crun}( \map{\crun}{\hrun}{i-1}(j-1) ) $, the only interesting pair is $\tup{j,j'}  = \tup{j-1, j}$ since for all others pairs the relation follows from (H.b).
                For this pair, we have $\map{\crun}{\hrun}{i}(j-1) =  \map{\crun}{\hrun}{i-1}(j-1) $ and $\map{\crun}{\hrun}{i}(j) = rb_{\crun}( \map{\crun}{\hrun}{i-1}(j-1) )$.
                From Lemma~\ref{lemma:vanilla:contract-rollback}, we have that  either $\headWindow{ \crun( \map{\crun}{\hrun}{i}(j) ) } = \headWindow{ \crun( \map{\crun}{\hrun}{i}(j') )  } = \infty$ or  $\headWindow{ \crun( \map{\crun}{\hrun}{i}(j) ) } \neq \infty \wedge \headWindow{ \crun( \map{\crun}{\hrun}{i}(j') ) } = \headWindow{ \crun( \map{\crun}{\hrun}{i}(j) ) } - 1$.
                Hence, (b) holds. 

            \end{description}

            \item[Rule \textsc{Execute-Assignment}:]
            The proof of this case is similar to that of the \textsc{Execute-Load-Hit} rule.

            \item[Rule \textsc{Execute-Marked-Assignment}:]
            The proof of this case is similar to that of the \textsc{Execute-Load-Hit} rule.

            \item[Rule \textsc{Execute-Store}:]
            The proof of this case is similar to that of the \textsc{Execute-Load-Hit} rule.

            \item[Rule \textsc{Execute-Skip}:]
            The proof of this case is similar to that of the \textsc{Execute-Load-Miss} rule.

            \item[Rule \textsc{Execute-Barrier}:]
            The proof of this case is similar to that of the \textsc{Execute-Load-Miss} rule.

            \end{description}
            This completes the proof for the $\execute{j}$ case.
    
        \item[$\SchedNext(C_{i-1}) = \retire{}$:] 
		Then, $\map{\crun}{\hrun}{i} = \mathit{shift}(\map{\crun}{\hrun}{i-1})$.
		We now proceed by case distinction on the applied rule:
		\begin{description}
			\item[Rule \textsc{Retire-Skip}:]
			Then, $	\buf_{i-1} = \tagged{\pskip{}}{\notags} \concat \buf_i$, $m_{i} = m_{i-1}$, and $a_{i} = a_{i-1}$.

            We now show that all conditions hold:
            \begin{description}
                \item[(a):] 
                From the well-formedness of $\buf_{i-1}$ (see Lemma~\ref{lemma:vanilla:buffers-well-formedness}) and $	\buf_{i-1} = \tagged{\pskip{}}{\notags} \concat \buf_i$, we have that $\buf_{i-1} = \tagged{\pskip{}}{\notags} \concat \tagged{\pmarkedassign{\pc}{\ell}}{\notags} \concat \buf_i'$ and $\buf_i = \tagged{\pmarkedassign{\pc}{\ell}}{\notags} \concat \buf_i'$.
                Therefore, we have:
                \begin{align*}
                    \prefixes{\buf_{i-1}} & = \{ \emptysequence \} \cup \{  
                        \tagged{\pskip{}}{\notags} \concat \tagged{\pmarkedassign{\pc}{\ell}}{\notags} \concat \buf' \mid \buf' \in \prefixes{\buf_i'} \} \\
                    \prefixes{\buf_{i}} & = \{ \tagged{\pmarkedassign{\pc}{\ell}}{\notags} \} \cup \{ \tagged{\pmarkedassign{\pc}{\ell}}{\notags} \concat \buf' \mid \buf' \in \prefixes{\buf_i'} \} 
                \end{align*}
                Let $\buf$ be an arbitrary prefix in $\prefixes{\buf_{i}}$.
                From the definitions of $\prefixes{\buf_{i-1}}$ and $\prefixes{\buf_{i}}$ (and from the well-formedness of the buffers), we know that $\tagged{\pskip{}}{\notags} \concat \buf$ is a prefix in  $\prefixes{\buf_{i-1}}$.

                We need to show that (i)--(iv) hold:
                \begin{description}
                    \item[(i):]                 
                    From the induction hypothesis (H.3.a.i) and $\tagged{\pskip{}}{\notags} \concat \buf \in \prefixes{\buf_{i-1}}$,  we have $\update{\tup{m_{i-1},a_{i-1}}}{ (\tagged{\pskip{}}{\notags} \concat \buf) } = \headConf{\crun( \map{\crun}{\hrun}{i-1}(| \tagged{\pskip{}}{\notags} \concat \buf |) )}$.
                    From $m_{i} = m_{i-1}$ and $a_{i} = a_{i-1}$, we get $\update{\tup{m_{i},a_{i}}}{ (\tagged{\pskip{}}{\notags} \concat \buf) } = \headConf{\crun( \map{\crun}{\hrun}{i-1}(| \tagged{\pskip{}}{\notags} \concat \buf |) )}$.
                    By applying $\update{}{}$ to $\tagged{\pskip{}}{\notags}$, we get $\update{\tup{m_{i},a_{i}}}{ \buf } = \headConf{\crun( \map{\crun}{\hrun}{i-1}(| \tagged{\pskip{}}{\notags} \concat \buf |) )}$.
                    From $| \tagged{\pskip{}}{\notags} \concat \buf | = |\buf|+1$, we get $\update{\tup{m_{i},a_{i}}}{ \buf } = \headConf{\crun( \map{\crun}{\hrun}{i-1}(|\buf |+1) )}$.
                    From this and $\mathit{shift}$'s definition, we get $\update{\tup{m_{i},a_{i}}}{ \buf } = \headConf{\crun( \mathit{shift}(\map{\crun}{\hrun}{i-1}(|\buf |)) )}$.
                    From this and $\map{\crun}{\hrun}{i} = \mathit{shift}(\map{\crun}{\hrun}{i-1})$, we get  $\update{\tup{m_{i},a_{i}}}{ \buf } = \headConf{\crun( \map{\crun}{\hrun}{i}(|\buf |) )}$.
                    Since $\buf$ is an arbitrary prefix in $\prefixes{\buf_{i}}$, (i) holds. 

                    \item[(ii):] 
                    From (H.3.a.ii) and $\tagged{\pskip{}}{\notags} \concat \buf \in \prefixes{\buf_{i-1}}$, we have  $\nrMispred{\tup{m_{i-1},a_{i-1}}}{\tagged{\pskip}{\notags} \concat \buf}+1 \geq |\crun(\map{\crun}{\hrun}{i-1}(|\tagged{\pskip}{\notags} \concat \buf|) )| $.
                    From $m_{i} = m_{i-1}$ and $a_{i} = a_{i-1}$, we get $\nrMispred{\tup{m_{i},a_{i}}}{\tagged{\pskip}{\notags} \concat \buf}+1 \geq |\crun(\map{\crun}{\hrun}{i-1}(|\tagged{\pskip}{\notags} \concat \buf|) )| $.
                    Moreover, from $\nrMispred{\tup{m_{i},a_{i}}}{\tagged{\pskip}{\notags} \concat \buf}+1 = \nrMispred{\tup{m_{i},a_{i}}}{\buf}+1$, we get $\nrMispred{\tup{m_{i},a_{i}}}{\buf}+1 \geq |\crun(\map{\crun}{\hrun}{i-1}(|\tagged{\pskip}{\notags} \concat \buf|) )| $.
                    Finally, from $| \tagged{\pskip{}}{\notags} \concat \buf | = |\buf|+1$ and $\mathit{shift}$'s definition, we get  $\nrMispred{\tup{m_{i},a_{i}}}{\buf}+1 \geq |\crun( \mathit{shift}(\map{\crun}{\hrun}{i-1}(| \buf|)) )| $.
                    From this and $\map{\crun}{\hrun}{i} = \mathit{shift}(\map{\crun}{\hrun}{i-1})$, we get $\nrMispred{\tup{m_{i},a_{i}}}{\buf}+1 \geq |\crun( \map{\crun}{\hrun}{i}(| \buf|)))| $.
                    Since $\buf$ is an arbitrary prefix in $\prefixes{\buf_{i}}$, (ii) holds. 

                    \item[(iii):]
                    From (H.3.a.iii) and $\tagged{\pskip{}}{\notags} \concat \buf \in \prefixes{\buf_{i-1}}$, we have $\nrMispred{\tup{m_{i-1},a_{i-1}}}{ \tagged{\pskip{}}{\notags} \concat \buf } = 0 \leftrightarrow \headWindow{\crun( \map{\crun}{\hrun}{i-1}(| \tagged{\pskip{}}{\notags} \concat \buf |) )} = \infty$.
                    From $m_{i} = m_{i-1}$ and $a_{i} = a_{i-1}$, we get $\nrMispred{\tup{m_{i},a_{i}}}{ \tagged{\pskip{}}{\notags} \concat \buf } = 0 \leftrightarrow \headWindow{\crun( \map{\crun}{\hrun}{i-1}(| \tagged{\pskip{}}{\notags} \concat \buf |) )} = \infty$.
                    From $\nrMispred{\tup{m_{i},a_{i}}}{ \tagged{\pskip{}}{\notags} \concat \buf } = \nrMispred{\tup{m_{i},a_{i}}}{  \buf }$, we further get $\nrMispred{\tup{m_{i},a_{i}}}{  \buf } = 0 \leftrightarrow \headWindow{\crun( \map{\crun}{\hrun}{i-1}(| \tagged{\pskip{}}{\notags} \concat \buf |) )} = \infty$.
                    Finally, from $| \tagged{\pskip{}}{\notags} \concat \buf | = |\buf|+1$, $\mathit{shift}$'s definition, and $\map{\crun}{\hrun}{i} = \mathit{shift}(\map{\crun}{\hrun}{i-1})$, we get $\nrMispred{\tup{m_{i},a_{i}}}{  \buf } = 0 \leftrightarrow \headWindow{\crun( \map{\crun}{\hrun}{i}(|  \buf |) )} = \infty$.
                    Since $\buf$ is an arbitrary prefix in $\prefixes{\buf_{i}}$, (iii) holds.

				   \item[(iv):]
                   We need to show $\headWindow{\crun(\map{\crun}{\hrun}{i}(|\buf|))} > 0$.
                   From (H.3.a.iv) and $\tagged{\pskip{}}{\notags} \concat \buf \in \prefixes{\buf_{i-1}}$, we have $\headWindow{\crun(\map{\crun}{\hrun}{i-1}(|\tagged{\pskip{}}{\notags} \concat \buf|))} > 0$.
                   From this and $|\tagged{\pskip{}}{\notags} \concat \buf| = |\buf| +1$, we have $\headWindow{\crun(\map{\crun}{\hrun}{i-1}( |\buf|+1))} > 0$.
                   From $\mathit{shift}$'s definition, we get $\headWindow{\crun( \mathit{shift} (\map{\crun}{\hrun}{i-1}( |\buf|) ))} > 0$.
                   Finally, from this and $\map{\crun}{\hrun}{i} = \mathit{shift}(\map{\crun}{\hrun}{i-1})$, we get $\headWindow{\crun(\map{\crun}{\hrun}{i}(|\buf|))} > 0$.

                \end{description}
                This concludes the proof for (a).

                \item[(b):]
                For (b), we need to show  that for all $\tup{j,j'} \in \nextPairs{\map{\crun}{\hrun}{i}}$,
				$\headWindow{ \crun(\map{\crun}{\hrun}{i}(j)) } = \headWindow{ \crun(\map{\crun}{\hrun}{i}(j') } = \infty$ or
				$\headWindow{ \crun(\map{\crun}{\hrun}{i}(j) } = \infty$ and $\headWindow{ \crun(\map{\crun}{\hrun}{i}(j') } = \wInterf$ or
				$\headWindow{ \crun(\map{\crun}{\hrun}{i}(j) } \neq \infty \wedge \headWindow{ \crun(\map{\crun}{\hrun}{i}(j') } = \headWindow{ \crun(\map{\crun}{\hrun}{i}(j) } - 1$.
                This immediately follows from $\map{\crun}{\hrun}{i} = \mathit{shift}(\map{\crun}{\hrun}{i-1})$ and (H.b).

            \end{description}

			\item[Rule \textsc{Retire-Fence}:]
			The proof of this case is identical to that of the rule \textsc{Retire-Skip}.

			\item[Rule \textsc{Retire-Assignment}:]
			Then, $	\buf_{i-1} = \tagged{\passign{x}{v}}{\notags} \concat \buf_i$, $m_{i} = m_{i-1}$, $a_{i} = a_{i-1}[x \mapsto v]$, and $v \in \Val$.
            
            There are two cases:
            \begin{description}
                \item[$x \neq \pc$:]
                We now show that all conditions hold:
                \begin{description}
                    \item[(a):]
                    From the well-formedness of $\buf_{i-1}$ (see Lemma~\ref{lemma:vanilla:buffers-well-formedness}) and $	\buf_{i-1} = \tagged{\passign{x}{v}}{\notags} \concat \buf_i$, we have that $\buf_{i-1} =  \tagged{\passign{x}{v}}{\notags} \concat \tagged{\pmarkedassign{\pc}{\ell}}{\notags} \concat \buf_i'$ and $\buf_i = \tagged{\pmarkedassign{\pc}{\ell}}{\notags} \concat \buf_i'$.
                    Therefore, we have:
                    \begin{align*}
                        \prefixes{\buf_{i-1}} & = \{ \emptysequence \} \cup \{  
                            \tagged{\passign{x}{v}}{\notags} \concat \tagged{\pmarkedassign{\pc}{\ell}}{\notags} \concat \buf' \mid \buf' \in \prefixes{\buf_i'} \} \\
                        \prefixes{\buf_{i}} & = \{ \tagged{\pmarkedassign{\pc}{\ell}}{\notags} \} \cup \{ \tagged{\pmarkedassign{\pc}{\ell}}{\notags} \concat \buf' \mid \buf' \in \prefixes{\buf_i'} \} 
                    \end{align*}
                    Let $\buf$ be an arbitrary prefix in $\prefixes{\buf_{i}}$.
                    From the definitions of $\prefixes{\buf_{i-1}}$ and $\prefixes{\buf_{i}}$ (and from the well-formedness of the buffers), we know that $\tagged{\passign{x}{v}}{\notags} \concat \buf$ is a prefix in  $\prefixes{\buf_{i-1}}$.

                    We now show that (i)--(iv) hold:
                    \begin{description}
                        \item[(i):]
                        From (H.3.a.i) and $\tagged{\passign{x}{v}}{\notags} \concat \buf\in \prefixes{\buf_{i-1}}$,  we have $\update{\tup{m_{i-1},a_{i-1}}}{ (\tagged{\passign{x}{v}}{\notags} \concat \buf) } = \crun( \map{\crun}{\hrun}{i-1}(| \tagged{\passign{x}{v}}{\notags} \concat \buf |) )$.
                        By applying $\update{}{}$ to $\tagged{\passign{x}{v}}{\notags}$, we get $\update{\tup{m_{i-1},a_{i-1}[x \mapsto v]}}{  \buf } = \crun( \map{\crun}{\hrun}{i-1}(| \tagged{\passign{x}{v}}{\notags} \concat \buf |) )$.
                        From $m_{i} = m_{i-1}$ and $a_{i} = a_{i-1}[x \mapsto v]$, we get $\update{\tup{m_{i},a_{i}}}{  \buf } = \crun( \map{\crun}{\hrun}{i-1}(| \tagged{\passign{x}{v}}{\notags} \concat \buf |) )$.
                        From $| \tagged{\passign{x}{v}}{\notags} \concat \buf | = |\buf|+1$, we get $\update{\tup{m_{i},a_{i}}}{ \buf } = \crun( \map{\crun}{\hrun}{i-1}(|\buf |+1) )$.
                        From this and $\mathit{shift}$'s definition, we get $\update{\tup{m_{i},a_{i}}}{ \buf } = \crun( \mathit{shift}(\map{\crun}{\hrun}{i-1}(|\buf |)) )$.
                        From this and $\map{\crun}{\hrun}{i} = \mathit{shift}(\map{\crun}{\hrun}{i-1})$, we get  $\update{\tup{m_{i},a_{i}}}{ \buf } = \crun( \map{\crun}{\hrun}{i}(|\buf |) )$.
                        Since $\buf$ is an arbitrary prefix in $\prefixes{\buf_{i}}$, (i) holds.
             
                        \item[(ii):] 
                        From (H.3.a.ii) and $\tagged{\passign{x}{v}}{\notags} \concat \buf \in \prefixes{\buf_{i-1}}$, we have  $\nrMispred{\tup{m_{i-1},a_{i-1}}}{\tagged{\passign{x}{v}}{\notags} \concat \buf}+1 \geq |\crun(\map{\crun}{\hrun}{i-1}(|\tagged{\passign{x}{v}}{\notags} \concat \buf|) )| $.
                        Moreover, from $m_{i} = m_{i-1}$ and $\nrMispred{\tup{m_{i-1},a_{i-1}}}{\tagged{\passign{x}{v}}{\notags} \concat \buf}+1 = \nrMispred{\tup{m_{i},a_{i}[x \mapsto v]}}{\buf}+1$, we get $\nrMispred{\tup{m_{i},a_{i}}}{\buf}+1 \geq |\crun(\map{\crun}{\hrun}{i-1}(|\tagged{\passign{x}{v}}{\notags} \concat \buf|) )| $.
                        Finally, from $| \tagged{\passign{x}{v}}{\notags} \concat \buf | = |\buf|+1$ and $\mathit{shift}$'s definition, we get  $\nrMispred{\tup{m_{i},a_{i}}}{\buf}+1 \geq |\crun( \mathit{shift}(\map{\crun}{\hrun}{i-1}(| \buf|)) )| $.
                        From this and $\map{\crun}{\hrun}{i} = \mathit{shift}(\map{\crun}{\hrun}{i-1})$, we get $\nrMispred{\tup{m_{i},a_{i}}}{\buf}+1 \geq |\crun( \map{\crun}{\hrun}{i}(| \buf|)))| $.
                        Since $\buf$ is an arbitrary prefix in $\prefixes{\buf_{i}}$, (ii) holds.  

                        \item[(iii):]
                        From (H.3.a.iii) and $\tagged{\passign{x}{v}}{\notags} \concat \buf \in \prefixes{\buf_{i-1}}$, we have $\nrMispred{\tup{m_{i-1},a_{i-1}}}{ \tagged{\passign{x}{v}}{\notags} \concat \buf } = 0 \leftrightarrow \headWindow{\crun( \map{\crun}{\hrun}{i-1}(| \tagged{\pskip{}}{\notags} \concat \buf |) )} = \infty$.
                        From $m_i = m_{i-1}$, $a_i = a_{i-1}[x \mapsto v]$, and $\nrMispred{\tup{m_{i-1},a_{i-1}}}{ \tagged{\passign{x}{v}}{\notags} \concat \buf } = \nrMispred{\tup{m_{i-1},a_{i-1}[x\mapsto v]}}{  \buf }$, we get $\nrMispred{\tup{m_{i},a_{i}}}{  \buf } = 0 \leftrightarrow \headWindow{\crun( \map{\crun}{\hrun}{i-1}(| \tagged{\passign{x}{v}}{\notags} \concat \buf |) )} = \infty$.
                        Finally, from $| \tagged{\passign{x}{v}}{\notags} \concat \buf | = |\buf|+1$, $\mathit{shift}$'s definition, and $\map{\crun}{\hrun}{i} = \mathit{shift}(\map{\crun}{\hrun}{i-1})$, we get $\nrMispred{\tup{m_{i},a_{i}}}{  \buf } = 0 \leftrightarrow \headWindow{\crun( \map{\crun}{\hrun}{i}(|  \buf |) )} = \infty$.
                        Since $\buf$ is an arbitrary prefix in $\prefixes{\buf_{i}}$, (iii) holds. 
                        
    				   \item[(iv):]
	                   We need to show $\headWindow{\crun(\map{\crun}{\hrun}{i}(|\buf|))} > 0$.
	                   From (H.3.a.iv) and $\tagged{\passign{x}{v}}{\notags} \concat \buf \in \prefixes{\buf_{i-1}}$, we have $\headWindow{\crun(\map{\crun}{\hrun}{i-1}(|\tagged{\passign{x}{v}}{\notags} \concat \buf|))} > 0$.
	                   From this and $|\tagged{\passign{x}{v}}{\notags} \concat \buf| = |\buf| +1$, we have $\headWindow{\crun(\map{\crun}{\hrun}{i-1}( |\buf|+1))} > 0$.
	                   From $\mathit{shift}$'s definition, we get $\headWindow{\crun( \mathit{shift} (\map{\crun}{\hrun}{i-1}( |\buf|) ))} > 0$.
	                   Finally, from this and $\map{\crun}{\hrun}{i} = \mathit{shift}(\map{\crun}{\hrun}{i-1})$, we get $\headWindow{\crun(\map{\crun}{\hrun}{i}(|\buf|))} > 0$.
                    \end{description}
                    
                    \item[(b):] 
                    The proof of (b) is similar to that for the \textsc{Retire-Skip} rule - case (b).

                \end{description}

                \item[$x = \pc$:] 
                We now show that all conditions hold:
                \begin{description}
                    \item[(a):]
                    From the well-formedness of $\buf_{i-1}$ (see Lemma~\ref{lemma:vanilla:buffers-well-formedness}) and $	\buf_{i-1} = \tagged{\passign{\pc}{v}}{\notags} \concat \buf_i$, we have that $\buf_{i-1} =  \tagged{\passign{\pc}{v}}{\notags} \concat  \buf_i$.
                    Therefore, we have:
                    \begin{align*}
                        \prefixes{\buf_{i-1}} & = \{ \emptysequence \} \cup \{  
                            \tagged{\passign{\pc}{v}}{\notags} \concat \buf' \mid \buf' \in \prefixes{\buf_i} \} 
                    \end{align*}
                    Let $\buf$ be an arbitrary prefix in $\prefixes{\buf_{i}}$.
                    From the definitions of $\prefixes{\buf_{i-1}}$ and $\prefixes{\buf_{i}}$ (and from the well-formedness of the buffers), we know that $\tagged{\passign{\pc}{v}}{\notags} \concat \buf$ is a prefix in  $\prefixes{\buf_{i-1}}$.

                    We now show that (i)-(iii) hold:
                    \begin{description}
                        \item[(i):]
                        From (H.3.a.i) and $\tagged{\passign{\pc}{v}}{\notags} \concat \buf\in \prefixes{\buf_{i-1}}$,  we have $\update{\tup{m_{i-1},a_{i-1}}}{ (\tagged{\passign{\pc}{v}}{\notags} \concat \buf) } = \crun( \map{\crun}{\hrun}{i-1}(| \tagged{\passign{x}{v}}{\notags} \concat \buf |) )$.
                        By applying $\update{}{}$ to $\tagged{\passign{\pc}{v}}{\notags}$, we get $\update{\tup{m_{i-1},a_{i-1}[\pc \mapsto v]}}{  \buf } = \crun( \map{\crun}{\hrun}{i-1}(| \tagged{\passign{\pc}{v}}{\notags} \concat \buf |) )$.
                        From $m_{i} = m_{i-1}$ and $a_{i} = a_{i-1}[\pc \mapsto v]$, we get $\update{\tup{m_{i},a_{i}}}{  \buf } = \crun( \map{\crun}{\hrun}{i-1}(| \tagged{\passign{\pc}{v}}{\notags} \concat \buf |) )$.
                        From $| \tagged{\passign{\pc}{v}}{\notags} \concat \buf | = |\buf|+1$, we get $\update{\tup{m_{i},a_{i}}}{ \buf } = \crun( \map{\crun}{\hrun}{i-1}(|\buf |+1) )$.
                        From this and $\mathit{shift}$'s definition, we get $\update{\tup{m_{i},a_{i}}}{ \buf } = \crun( \mathit{shift}(\map{\crun}{\hrun}{i-1}(|\buf |)) )$.
                        From this and $\map{\crun}{\hrun}{i} = \mathit{shift}(\map{\crun}{\hrun}{i-1})$, we get  $\update{\tup{m_{i},a_{i}}}{ \buf } = \crun( \map{\crun}{\hrun}{i}(|\buf |) )$.
                        Since $\buf$ is an arbitrary prefix in $\prefixes{\buf_{i}}$, (i) holds.
             
                        \item[(ii):] 
                        From (H.3.a.ii) and $\tagged{\passign{\pc}{v}}{\notags} \concat \buf \in \prefixes{\buf_{i-1}}$, we have  $\nrMispred{\tup{m_{i-1},a_{i-1}}}{\tagged{\passign{\pc}{v}}{\notags} \concat \buf}+1 \geq |\crun(\map{\crun}{\hrun}{i-1}(|\tagged{\passign{\pc}{v}}{\notags} \concat \buf|) )| $.
                        Moreover, from $m_{i} = m_{i-1}$ and $\nrMispred{\tup{m_{i-1},a_{i-1}}}{\tagged{\passign{\pc}{v}}{\notags} \concat \buf}+1 = \nrMispred{\tup{m_{i},a_{i}[\pc \mapsto v]}}{\buf}+1$, we get $\nrMispred{\tup{m_{i},a_{i}}}{\buf}+1 \geq |\crun(\map{\crun}{\hrun}{i-1}(|\tagged{\passign{x}{v}}{\notags} \concat \buf|) )| $.
                        Finally, from $| \tagged{\passign{\pc}{v}}{\notags} \concat \buf | = |\buf|+1$ and $\mathit{shift}$'s definition, we get  $\nrMispred{\tup{m_{i},a_{i}}}{\buf}+1 \geq |\crun( \mathit{shift}(\map{\crun}{\hrun}{i-1}(| \buf|)) )| $.
                        From this and $\map{\crun}{\hrun}{i} = \mathit{shift}(\map{\crun}{\hrun}{i-1})$, we get $\nrMispred{\tup{m_{i},a_{i}}}{\buf}+1 \geq |\crun( \map{\crun}{\hrun}{i}(| \buf|)))| $.
                        Since $\buf$ is an arbitrary prefix in $\prefixes{\buf_{i}}$, (ii) holds.  

                        \item[(iii):]
                        From (H.3.a.iii) and $\tagged{\passign{\pc}{v}}{\notags} \concat \buf \in \prefixes{\buf_{i-1}}$, we have $\nrMispred{\tup{m_{i-1},a_{i-1}}}{ \tagged{\passign{\pc}{v}}{\notags} \concat \buf } = 0 \leftrightarrow \headWindow{\crun( \map{\crun}{\hrun}{i-1}(| \tagged{\pskip{}}{\notags} \concat \buf |) )} = \infty$.
                        From $m_i = m_{i-1}$, $a_i = a_{i-1}[\pc \mapsto v]$, and $\nrMispred{\tup{m_{i-1},a_{i-1}}}{ \tagged{\passign{\pc}{v}}{\notags} \concat \buf } = \nrMispred{\tup{m_{i-1},a_{i-1}[\pc\mapsto v]}}{  \buf }$ (because the retired assignment is untagged), we get $\nrMispred{\tup{m_{i},a_{i}}}{  \buf } = 0 \leftrightarrow \headWindow{\crun( \map{\crun}{\hrun}{i-1}(| \tagged{\passign{x}{v}}{\notags} \concat \buf |) )} = \infty$.
                        Finally, from $| \tagged{\passign{\pc}{v}}{\notags} \concat \buf | = |\buf|+1$, $\mathit{shift}$'s definition, and $\map{\crun}{\hrun}{i} = \mathit{shift}(\map{\crun}{\hrun}{i-1})$, we get $\nrMispred{\tup{m_{i},a_{i}}}{  \buf } = 0 \leftrightarrow \headWindow{\crun( \map{\crun}{\hrun}{i}(|  \buf |) )} = \infty$.
                        Since $\buf$ is an arbitrary prefix in $\prefixes{\buf_{i}}$, (iii) holds.  
                        
                       \item[(iv):]
	                   We need to show $\headWindow{\crun(\map{\crun}{\hrun}{i}(|\buf|))} > 0$.
	                   From (H.3.a.iv) and $\tagged{\passign{\pc}{v}}{\notags} \concat \buf \in \prefixes{\buf_{i-1}}$, we have $\headWindow{\crun(\map{\crun}{\hrun}{i-1}(|\tagged{\passign{x}{v}}{\notags} \concat \buf|))} > 0$.
	                   From this and $|\tagged{\passign{\pc}{v}}{\notags} \concat \buf| = |\buf| +1$, we have $\headWindow{\crun(\map{\crun}{\hrun}{i-1}( |\buf|+1))} > 0$.
	                   From $\mathit{shift}$'s definition, we get $\headWindow{\crun( \mathit{shift} (\map{\crun}{\hrun}{i-1}( |\buf|) ))} > 0$.
	                   Finally, from this and $\map{\crun}{\hrun}{i} = \mathit{shift}(\map{\crun}{\hrun}{i-1})$, we get $\headWindow{\crun(\map{\crun}{\hrun}{i}(|\buf|))} > 0$.
                    \end{description}
                    
                    \item[(b):] 
                    For (b), we need to show  that for all $\tup{j,j'} \in \nextPairs{\map{\crun}{\hrun}{i}}$,
                    $\headWindow{ \crun(\map{\crun}{\hrun}{i}(j)) } = \headWindow{ \crun(\map{\crun}{\hrun}{i}(j') } = \infty$ or
                    $\headWindow{ \crun(\map{\crun}{\hrun}{i}(j) } = \infty$ and $\headWindow{ \crun(\map{\crun}{\hrun}{i}(j') } = \wInterf$ or
                    $\headWindow{ \crun(\map{\crun}{\hrun}{i}(j) } \neq \infty \wedge \headWindow{ \crun(\map{\crun}{\hrun}{i}(j') } = \headWindow{ \crun(\map{\crun}{\hrun}{i}(j) } - 1$.
                    This immediately follows from $\map{\crun}{\hrun}{i} = \mathit{shift}(\map{\crun}{\hrun}{i-1})$ and (H.b). 
                \end{description} 

            \end{description}         
		
			\item[Rule \textsc{Retire-Marked-Assignment}:]
			The proof of this case is similar to that of the rule \textsc{Retire-Assignment} (when $x = \pc$).

			\item[Rule \textsc{Retire-Store}:]
			The proof of this case is similar to that of the rule \textsc{Retire-Assignment} (when $x\neq \pc$).
		\end{description}
		This completes the proof for the $\retire{}$ case.

    \end{description}
    This concludes the proof of the induction step.
    
\end{description}
This concludes the proof of our theorem.
\end{proof}

\subsection{Indistinguishability lemma}

\begin{definition}[Deep-indistinguishability of hardware configurations]\label{def:vanilla:deep-indistinguishability}
	We say that two hardware configurations $\tup{\sigma,\mu} = \tup{m,a,\buf, \CacheState,\BpState, \SchedState}$ and $\tup{\sigma',\mu'} = \tup{m',a',\buf', \CacheState',\BpState', \SchedState'}$ are \emph{deep-indistinguishable}, written $\tup{\sigma,\mu} \sim \tup{\sigma',\mu'}$, iff
	\begin{inparaenum}[(a)]
		\item $\apply{\buf}{a}(\pc) = \apply{\buf'}{a'}(\pc)$,
		\item $\DeepProject{\buf} = \DeepProject{\buf'}$,
		\item $\CacheState = \CacheState'$,
		\item $\BpState = \BpState'$, and
		\item $\SchedState = \SchedState'$,
		\end{inparaenum}
		where $\DeepProject{\buf}$ is inductively defined as follows:
		\begin{align*}
			\DeepProject{\emptysequence}	&:= \emptysequence\\
			\DeepProject{\tagged{\pskip{}}{T}}  &:=  \tagged{\pskip{}}{T}\\
			\DeepProject{\tagged{\pbarrier{}}{T}}  &:=  \tagged{\pbarrier{}}{T}\\
			\DeepProject{\tagged{\passign{x}{e}}{T}}  &:=
			{
				\begin{cases}
				\tagged{\passign{x}{\resolved}}{T} & \text{if}\ e \in \Val \wedge x \neq \pc\\
				\tagged{\passign{x}{e}}{T} & \text{if}\ e \not\in \Val \wedge x \neq \pc\\
				\tagged{\passign{x}{e}}{T} & \text{if}\ x = \pc
				\end{cases}
			}\\
			\DeepProject{\tagged{\pload{x}{e}}{T}}  &:=
			{
				\tagged{\pload{x}{e}}{T}
			}\\
			\DeepProject{\tagged{\pstore{x}{e}}{T}}  &:=
			{
				\begin{cases}
					\tagged{\pstore{x}{e}}{T} & \text{if}\ x \in \Var\\
					\tagged{\pstore{\resolved}{e}}{T} & \text{if}\ x \in \Val
				\end{cases}
			}\\
			\DeepProject{\tagged{i}{T} \concat \buf} &:= \DeepProject{\tagged{i}{T}} \concat \DeepProject{\buf}
		\end{align*}
\end{definition}

\begin{lemma}[Deep-buffer projection equality implies equality of data-independent projections]\label{lemma:vanilla:buffer-projections}
Let $\buf, \buf'$ be two reorder buffers.
If $\DeepProject{\buf} = \DeepProject{\buf'}$, then $\BufProject{\buf} = \BufProject{\buf'}$.
\end{lemma}

\begin{proof}
It follows from the fact that everything disclosed by $\DeepProject{\buf}$ is also disclosed by $\BufProject{\buf}$.
\end{proof}

\begin{lemma}[Observation equivalence preserves deep-indistinguishability]\label{lemma:vanilla:trace-equiv-implies-stepwise-indistinguishability}
    Let $p$ be a well-formed program and $C_0 = \tup{m_0,a_0,\buf_0,\CacheState_0, \BpState_0, \SchedState_0}$, $C_0' = \tup{m_0',a_0',\buf_0',\CacheState_0', \BpState_0', \SchedState_0'}$ be reachable hardware configurations.
    If
    \begin{inparaenum}[(a)]
        \item $C_0 \sim C_0'$, and
        \item for all $\buf \in \prefixes{\buf_0}$, $\buf' \in \prefixes{\buf_0'}$ such that $|\buf| = |\buf'|$, 
        there are $s_0, s_0', s_1, s_1', \tau, \tau'$ such that $s_0 \CtSpecInterfStep{\tau}{} s_1$, $s_0' \CtSpecInterfStep{\tau'}{} s_1'$, $\headWindow{s_0} > 0$, $\headWindow{s_0'} >0$,  $\tau = \tau'$, $C_0 \bufEquiv{|\buf|} \sigma_0$, and $C_0' \bufEquiv{|\buf'|} \sigma_0'$,
    \end{inparaenum}
    then either there are $C_1, C_1'$ such that $C_0 \muarchStep{}{} C_1$, $C_0' \muarchStep{}{} C_1'$, and $C_1 \sim C_1'$ or there is no $C_1$ such that $C_0 \muarchStep{}{} C_1$ and no $C_1'$ such that $C_0' \muarchStep{}{} C_1'$.
\end{lemma}

\begin{proof}
    Let $p$ be a well-formed program, $C_0 = \tup{m_0,a_0,\buf_0,\CacheState_0, \BpState_0, \SchedState_0}$, and $C_0' = \tup{m_0',a_0',\buf_0',\CacheState_0', \BpState_0', \SchedState_0'}$.
    Moreover, we assume that conditions (a) and (b) holds. 
    In the following, we denote by (c) the post-condition ``either there are $C_1, C_1'$ such that $C_0 \muarchStep{}{} C_1$, $C_0' \muarchStep{}{} C_1'$, and $C_1 \sim C_1'$ or there is no $C_1$ such that $C_0 \muarchStep{}{} C_1$ and no $C_1'$ such that $C_0' \muarchStep{}{} C_1'$.''
    
    From (a), it follows that $\SchedState_0 = \SchedState_0$.
    Therefore, the directive obtained from the scheduler is the same in both cases, i.e., $\SchedNext(\SchedState_0) = \SchedNext(\SchedState_0')$.
    We proceed by case distinction on the directive $d = \SchedNext(\SchedState_0)$:
    \begin{description}
    \item[$d = \fetch{}$:]
        Therefore, we can only apply one of the $\fetch{}$ rules depending on the current program counter.
        There are two cases: 
        \begin{description}
            \item[$\apply{\buf_0}{a_0}(\pc) \neq \bot \wedge |\buf_0| < \wMuarch$:]
            There are several cases:
            \begin{description}
                \item[$\CacheAccess(\CacheState_0,  \apply{\buf_0}{a_0}(\pc)) = \CacheHit \wedge p(\apply{\buf_0}{a_0}(\pc)) = \pjz{x}{\lbl}$:] 
                From (a), we get that $\CacheState_0' = \CacheState_0$ and $\apply{\buf_0'}{a_0'}(\pc) = \apply{\buf_0}{a_0}(\pc)$.
                Therefore, $\CacheAccess(\CacheState_0',  \apply{\buf_0'}{a_0'}(\pc)) = \CacheHit$.
                Moreover, from (a) we also get that $\apply{\buf_0}{a_0}(\pc)=\apply{\buf_0'}{a_0'}(\pc)$ and, therefore, $p(\apply{\buf_0'}{a_0'}(\pc)) = \pjz{x}{\lbl}$ as well.
                Therefore, we can apply the \textsc{Fetch-Branch-Hit} and \textsc{Step} rules to $C_0$ and $C_0'$ as follows:
                \begin{align*}
                    \lbl_0 &:= \BpPredict(\BpState_0, \apply{\buf_0}{a_0}(\pc))\\
                    \lbl_0' &:= \BpPredict(\BpState_0', \apply{\buf_0'}{a_0'}(\pc))\\
                    \buf_1 &:= \buf_0  \concat  \tagged{\passign{\pc}{\lbl_0}}{\apply{\buf_0}{a_0}(\pc)}\\
                    \buf_1' &:= \buf_0'  \concat \tagged{\passign{\pc}{\lbl_0'}}{\apply{\buf_0'}{a_0'}(\pc)}\\
                    \tup{m_0,a_0,\buf_0,\CacheState_0,\BpState_0} &\muarchStep{\fetch{}}{} \tup{m_0, a_0, \buf_1, \CacheUpdate(\CacheState_0, \apply{\buf_0}{a_0}(\pc)),\BpState_0}\\
                    \tup{m_0,a_0,\buf_0,\CacheState_0,\BpState_0, \SchedState_0} &\muarchStep{}{} \tup{m_0, a_0, \buf_1, \CacheUpdate(\CacheState_0, \apply{\buf_0}{a_0}(\pc)),\BpState_0, \SchedUpdate(\SchedState_0, \BufProject{\buf_1})}\\
                    \tup{m_0',a_0',\buf_0',\CacheState_0',\BpState_0'} &\muarchStep{\fetch{}}{} \tup{m_0', a_0', \buf_1', \CacheUpdate(\CacheState_0',  \apply{\buf_0'}{a_0'}(\pc)),\BpState_0'}\\
                    \tup{m_0',a_0',\buf_0',\CacheState_0',\BpState_0', \SchedState_0'} &\muarchStep{}{} \tup{m_0', a_0', \buf_1', \CacheUpdate(\CacheState_0',  \apply{\buf_0'}{a_0'}(\pc)),\BpState_0', \SchedUpdate(\SchedState_0', \BufProject{\buf_1'})}
                \end{align*}
                We now show that $C_1 =  \tup{m_0, a_0, \buf_1, \CacheUpdate(\CacheState_0, \apply{\buf_0}{a_0}(\pc)),\BpState_0, \SchedUpdate(\SchedState_0, \BufProject{\buf_1})}$ and $C_1' = \tup{m_0', a_0', \buf_1', \CacheUpdate(\CacheState_0',  \apply{\buf_0'}{a_0'}(\pc)),\BpState_0', \SchedUpdate(\SchedState_0', \BufProject{\buf_1'})}$ are indistinguishable, i.e., $C_1 \sim C_1'$.
                For this, we need to show that:
                \begin{description}
                    \item[$\apply{\buf_1}{a_0}(\pc) = \apply{\buf_1'}{a_0'}(\pc)$:]
                    We know that $\apply{\buf_1}{a_0}(\pc) = \lbl_0$ and $\apply{\buf_1'}{a_0'}(\pc) = \lbl_0'$.
                    From $\lbl_0 = \BpPredict(\BpState_0, \apply{\buf_0}{a_0}(\pc))$, $\lbl_0' = \BpPredict(\BpState_0', \apply{\buf_0'}{a_0'}(\pc))$, and (a), we immediately get $\lbl_0 = \lbl_0'$.

                    \item[$\DeepProject{\buf_1} = \DeepProject{\buf_1'}$:] 
                    This follows from $\DeepProject{\buf_0} = \DeepProject{\buf_0'}$, which in turn follows from (a), $\buf_1 = \buf_0  \concat  \tagged{\passign{\pc}{\lbl_0}}{\apply{\buf_0}{a_0}(\pc)}$, $\buf_1' = \buf_0'  \concat  \tagged{\passign{\pc}{\lbl_0'}}{\apply{\buf_0'}{a_0'}(\pc)}$, and $\lbl_0 = \lbl_0'$.
    
                    \item[$\CacheUpdate(\CacheState_0,  \apply{\buf_0}{a_0}(\pc)) = \CacheUpdate(\CacheState_0',  \apply{\buf_0'}{a_0'}(\pc))$:]
                    This follows from $\CacheState_0 = \CacheState_0'$, which in turn follows from (a), and $\apply{\buf_0}{a_0}(\pc) = \apply{\buf_0'}{a_0'}(\pc)$.
    
                    \item[$\BpState_0 = \BpState_0':$] 
                    This follows from (a).
    
                    \item[$\SchedUpdate(\SchedState_0, \BufProject{\buf_1}) = \SchedUpdate(\SchedState_0', \BufProject{\buf_1'})$:]
                    From (a), we have   $\SchedState_0 = \SchedState_0'$.
                    From $\DeepProject{\buf_1} = \DeepProject{\buf_1'}$ and Lemma~\ref{lemma:vanilla:buffer-projections}, we have $\BufProject{\buf_1} = \BufProject{\buf_1'}$.
                    Therefore, $\SchedUpdate(\SchedState_0, \BufProject{\buf_1}) = \SchedUpdate(\SchedState_0', \BufProject{\buf_1'})$.
                \end{description}
                Therefore, $C_1 \sim C_1'$ and (c) holds.
    
                \item[$\CacheAccess(\CacheState_0,  \apply{\buf_0}{a_0}(\pc)) = \CacheHit \wedge p(\apply{\buf_0}{a_0}(\pc)) = \pjmp{e}$:] 
                From (a), we get that $\CacheState_0' = \CacheState_0$ and $\apply{\buf_0'}{a_0'}(\pc) = \apply{\buf_0}{a_0}(\pc)$.
                Therefore, $\CacheAccess(\CacheState_0',  \apply{\buf_0'}{a_0'}(\pc)) = \CacheHit$.
                Moreover, from (a) we also get that $\apply{\buf_0}{a_0}(\pc)=\apply{\buf_0'}{a_0'}(\pc)$ and, therefore, $p(\apply{\buf_0'}{a_0'}(\pc)) = \pjmp{e}$ as well.
                Therefore, we can apply the \textsc{Fetch-Jump-Hit} and \textsc{Step} rules to $C_0$ and $C_0'$ as follows:
                \begin{align*}
                    \buf_1 &:= \buf_0  \concat  \tagged{\passign{\pc}{e}}{\notags}\\
                    \buf_1' &:= \buf_0'  \concat  \tagged{\passign{\pc}{e}}{\notags}\\
                    \tup{m_0,a_0,\buf_0,\CacheState_0,\BpState_0} &\muarchStep{\fetch{}}{} \tup{m_0, a_0, \buf_1, \CacheUpdate(\CacheState_0, \apply{\buf_0}{a_0}(\pc)),\BpState_0}\\
                    \tup{m_0,a_0,\buf_0,\CacheState_0,\BpState_0, \SchedState_0} &\muarchStep{}{} \tup{m_0, a_0, \buf_1, \CacheUpdate(\CacheState_0, \apply{\buf_0}{a_0}(\pc)),\BpState_0, \SchedUpdate(\SchedState_0, \BufProject{\buf_1})}\\
                    \tup{m_0',a_0',\buf_0',\CacheState_0',\BpState_0'} &\muarchStep{\fetch{}}{} \tup{m_0', a_0', \buf_1', \CacheUpdate(\CacheState_0',  \apply{\buf_0'}{a_0'}(\pc)),\BpState_0'}\\
                    \tup{m_0',a_0',\buf_0',\CacheState_0',\BpState_0', \SchedState_0'} &\muarchStep{}{} \tup{m_0', a_0', \buf_1', \CacheUpdate(\CacheState_0',  \apply{\buf_0'}{a_0'}(\pc)),\BpState_0', \SchedUpdate(\SchedState_0', \BufProject{\buf_1'})}
                \end{align*}
                We now show that $C_1 =  \tup{m_0, a_0, \buf_1, \CacheUpdate(\CacheState_0, \apply{\buf_0}{a_0}(\pc)),\BpState_0, \SchedUpdate(\SchedState_0, \BufProject{\buf_1})}$ and $C_1' = \tup{m_0', a_0', \buf_1', \CacheUpdate(\CacheState_0',  \apply{\buf_0'}{a_0'}(\pc)),\BpState_0', \SchedUpdate(\SchedState_0', \BufProject{\buf_1'})}$ are indistinguishable, i.e., $C_1 \sim C_1'$.
                For this, we need to show that:
                \begin{description}
                    \item[$\apply{\buf_1}{a_0}(\pc) = \apply{\buf_1'}{a_0'}(\pc)$:]
                    There are two cases.
                    If $e \in \Val$, then $\apply{\buf_1}{a_0}(\pc) = \apply{\buf_1'}{a_0'}(\pc) = e$.
                    Otherwise, $\apply{\buf_1}{a_0}(\pc) = \apply{\buf_1'}{a_0'}(\pc) = \bot$.
    
                    \item[$\DeepProject{\buf_1} = \DeepProject{\buf_1'}$:] 
                    This follows from $\DeepProject{\buf_0} = \DeepProject{\buf_0'}$, which in turn follows from (a), $\buf_1 = \buf_0  \concat  \tagged{\passign{\pc}{e}}{\notags}$, and $\buf_1' = \buf_0'  \concat  \tagged{\passign{\pc}{e}}{\notags}$.
    
                    \item[$\CacheUpdate(\CacheState_0,  \apply{\buf_0}{a_0}(\pc)) = \CacheUpdate(\CacheState_0',  \apply{\buf_0'}{a_0'}(\pc))$:]
                    This follows from $\CacheState_0 = \CacheState_0'$, which in turn follows from (a), and $\apply{\buf_0}{a_0}(\pc) = \apply{\buf_0'}{a_0'}(\pc)$.
    
                    \item[$\BpState_0 = \BpState_0':$] 
                    This follows from (a).
    
                    \item[$\SchedUpdate(\SchedState_0, \BufProject{\buf_1}) = \SchedUpdate(\SchedState_0', \BufProject{\buf_1'})$:]
                    From (a), we have   $\SchedState_0 = \SchedState_0'$.
                    From $\DeepProject{\buf_1} = \DeepProject{\buf_1'}$ and Lemma~\ref{lemma:vanilla:buffer-projections}, we have $\BufProject{\buf_1} = \BufProject{\buf_1'}$.
                    Therefore, $\SchedUpdate(\SchedState_0, \BufProject{\buf_1}) = \SchedUpdate(\SchedState_0', \BufProject{\buf_1'})$.
                \end{description}
                Therefore, $C_1 \sim C_1'$ and (c) holds.
                
                \item[$\CacheAccess(\CacheState_0,  \apply{\buf_0}{a_0}(\pc)) = \CacheHit \wedge p(\apply{\buf_0}{a_0}(\pc)) \neq \pjz{x}{\lbl} \wedge p(\apply{\buf_0}{a_0}(\pc)) \neq \pjmp{e}$:]
                Observe that from (a) it follows that $|\buf_0| \geq \wMuarch-1$ iff $|\buf_0'| \geq \wMuarch-1$.
                Therefore, if $|\buf_0| \geq \wMuarch-1$, then (c) holds since both computations are stuck.
                In the following, we assume that $|\buf_0| < \wMuarch-1$ and $|\buf_0'| < \wMuarch-1$.
                
                From (a), we get that $\CacheState_0' = \CacheState_0$ and $\apply{\buf_0'}{a_0'}(\pc) = \apply{\buf_0}{a_0}(\pc)$.
                Therefore, $\CacheAccess(\CacheState_0',  \apply{\buf_0'}{a_0'}(\pc)) = \CacheHit$.
                Moreover, from (a) we also get that $\apply{\buf_0}{a_0}(\pc)=\apply{\buf_0'}{a_0'}(\pc)$ and, therefore, $p(\apply{\buf_0'}{a_0'}(\pc)) \neq \pjz{x}{\lbl} \wedge p(\apply{\buf_0'}{a_0'}(\pc)) \neq \pjmp{e}$ as well.
                Therefore, we can apply the \textsc{Fetch-Others-Hit} and \textsc{Step} rules to $C_0$ and $C_0'$ as follows:
                \begin{align*}
                    v &:= \apply{\buf_0}{a_0}(\pc) +1 \\
                    v' &:= \apply{\buf_0'}{a_0'}(\pc) +1 \\
                    \buf_1 &:= \buf_0  \concat \tagged{p(\apply{\buf_0}{a_0}(\pc))}{\notags} \concat \tagged{\pmarkedassign{\pc}{v}}{\notags}\\
                    \buf_1' &:= \buf_0'  \concat \tagged{p(\apply{\buf_0'}{a_0'}(\pc))}{\notags} \concat \tagged{\pmarkedassign{\pc}{v'}}{\notags}\\
                    \tup{m_0,a_0,\buf_0,\CacheState_0,\BpState_0} &\muarchStep{\fetch{}}{} \tup{m_0, a_0, \buf_1, \CacheUpdate(\CacheState_0, \apply{\buf_0}{a_0}(\pc)),\BpState_0}\\
                    \tup{m_0,a_0,\buf_0,\CacheState_0,\BpState_0, \SchedState_0} &\muarchStep{}{} \tup{m_0, a_0, \buf_1, \CacheUpdate(\CacheState_0, \apply{\buf_0}{a_0}(\pc)),\BpState_0, \SchedUpdate(\SchedState_0, \BufProject{\buf_1})}\\
                    \tup{m_0',a_0',\buf_0',\CacheState_0',\BpState_0'} &\muarchStep{\fetch{}}{} \tup{m_0', a_0', \buf_1', \CacheUpdate(\CacheState_0',  \apply{\buf_0'}{a_0'}(\pc)),\BpState_0'}\\
                    \tup{m_0',a_0',\buf_0',\CacheState_0',\BpState_0', \SchedState_0'} &\muarchStep{}{} \tup{m_0', a_0', \buf_1', \CacheUpdate(\CacheState_0',  \apply{\buf_0'}{a_0'}(\pc)),\BpState_0', \SchedUpdate(\SchedState_0', \BufProject{\buf_1'})}
                \end{align*}
                We now show that $C_1 =  \tup{m_0, a_0, \buf_1, \CacheUpdate(\CacheState_0, \apply{\buf_0}{a_0}(\pc)),\BpState_0, \SchedUpdate(\SchedState_0, \BufProject{\buf_1})}$ and $C_1' = \tup{m_0', a_0', \buf_1', \CacheUpdate(\CacheState_0',  \apply{\buf_0'}{a_0'}(\pc)),\BpState_0', \SchedUpdate(\SchedState_0', \BufProject{\buf_1'})}$ are indistinguishable, i.e., $C_1 \sim C_1'$.
                For this, we need to show that:
                \begin{description}
                    \item[$\apply{\buf_1}{a_0}(\pc) = \apply{\buf_1'}{a_0'}(\pc)$:]
                    From (a), we get that $\apply{\buf_0}{a_0}(\pc) = \apply{\buf_0'}{a_0'}(\pc)$.
                    From this, we have that $v = v'$.
                    Therefore, $\apply{\buf_1}{a_0}(\pc) = \apply{\buf_1'}{a_0'}(\pc) = v$.
    
                    \item[$\DeepProject{\buf_1} = \DeepProject{\buf_1'}$:] 
                    This follows from $\DeepProject{\buf_0} = \DeepProject{\buf_0'}$, which in turn follows from (a), $\buf_1 = \buf_0  \concat \tagged{p(\apply{\buf_0}{a_0}(\pc))}{\notags} \concat \tagged{\pmarkedassign{\pc}{v}}{\notags}$, $\buf_1' = \buf_0'  \concat \tagged{p(\apply{\buf_0'}{a_0'}(\pc))}{\notags} \concat \tagged{\pmarkedassign{\pc}{v'}}{\notags}$,  $v = v'$, and $\apply{\buf_0}{a_0}(\pc) = \apply{\buf_0'}{a_0'}(\pc)$.
    
                    \item[$\CacheUpdate(\CacheState_0,  \apply{\buf_0}{a_0}(\pc)) = \CacheUpdate(\CacheState_0',  \apply{\buf_0'}{a_0'}(\pc))$:]
                    This follows from $\CacheState_0 = \CacheState_0'$, which in turn follows from (a), and $\apply{\buf_0}{a_0}(\pc) = \apply{\buf_0'}{a_0'}(\pc)$.
    
                    \item[$\BpState_0 = \BpState_0':$] 
                    This follows from (a).
    
                    \item[$\SchedUpdate(\SchedState_0, \BufProject{\buf_1}) = \SchedUpdate(\SchedState_0', \BufProject{\buf_1'})$:]
                    From (a), we have   $\SchedState_0 = \SchedState_0'$.
                    From $\DeepProject{\buf_1} = \DeepProject{\buf_1'}$ and Lemma~\ref{lemma:vanilla:buffer-projections}, we have $\BufProject{\buf_1} = \BufProject{\buf_1'}$.
                    Therefore, $\SchedUpdate(\SchedState_0, \BufProject{\buf_1}) = \SchedUpdate(\SchedState_0', \BufProject{\buf_1'})$.
                \end{description}
                Therefore, $C_1 \sim C_1'$ and (c) holds.
    
                \item[$\CacheAccess(\CacheState_0,  \apply{\buf_0}{a_0}(\pc)) = \CacheMiss$:]
                From (a), we get that $\CacheState_0' = \CacheState_0$ and $\apply{\buf_0'}{a_0'}(\pc) = \apply{\buf_0}{a_0}(\pc)$.
                Therefore, $\CacheAccess(\CacheState_0',  \apply{\buf_0'}{a_0'}(\pc)) = \CacheMiss$.
                Therefore, we can apply the \textsc{Fetch-Miss} and \textsc{Step} rules to $C_0$ and $C_0'$ as follows:
                \begin{align*}
                    \tup{m_0,a_0,\buf_0,\CacheState_0,\BpState_0} &\muarchStep{\fetch{}}{} \tup{m_0, a_0, \buf_0, \CacheUpdate(\CacheState_0, \apply{\buf_0}{a_0}(\pc)),\BpState_0}\\
                    \tup{m_0,a_0,\buf_0,\CacheState_0,\BpState_0, \SchedState_0} &\muarchStep{}{} \tup{m_0, a_0, \buf_0, \CacheUpdate(\CacheState_0,  \apply{\buf_0}{a_0}(\pc)),\BpState_0, \SchedUpdate(\SchedState_0, \BufProject{\buf_0})}\\
                    \tup{m_0',a_0',\buf_0',\CacheState_0',\BpState_0'} &\muarchStep{\fetch{}}{} \tup{m_0', a_0', \buf_0', \CacheUpdate(\CacheState_0',  \apply{\buf_0'}{a_0'}(\pc)),\BpState_0'}\\
                    \tup{m_0',a_0',\buf_0',\CacheState_0',\BpState_0', \SchedState_0'} &\muarchStep{}{} \tup{m_0', a_0', \buf_0', \CacheUpdate(\CacheState_0',  \apply{\buf_0'}{a_0'}(\pc)),\BpState_0', \SchedUpdate(\SchedState_0', \BufProject{\buf_0'})}
                \end{align*}
                We now show that $C_1 = \tup{m_0, a_0, \buf_0, \CacheUpdate(\CacheState_0,  \apply{\buf_0}{a_0}(\pc)),\BpState_0, \SchedUpdate(\SchedState_0, \BufProject{\buf_0})}$ and $C_1' = \tup{m_0', a_0', \buf_0', \CacheUpdate(\CacheState_0',  \apply{\buf_0'}{a_0'}(\pc)),\BpState_0', \SchedUpdate(\SchedState_0', \BufProject{\buf_0'})}$ are indistinguishable, i.e., $C_1 \sim C_1'$.
                For this, we need to show that:
                \begin{description}
                    \item[$\apply{\buf_0}{a_0}(\pc) = \apply{\buf_0'}{a_0'}(\pc)$:]
                    This follows from (a).
    
                    \item[$\DeepProject{\buf_0} = \DeepProject{\buf_0'}$:] 
                    This follows from (a).
    
                    \item[$\CacheUpdate(\CacheState_0,  \apply{\buf_0}{a_0}(\pc)) = \CacheUpdate(\CacheState_0',  \apply{\buf_0'}{a_0'}(\pc))$:]
                    This follows from $\CacheState_0 = \CacheState_0'$, which in turn follows from (a), and $\apply{\buf_0}{a_0}(\pc) = \apply{\buf_0'}{a_0'}(\pc)$.
    
                    \item[$\BpState_0 = \BpState_0':$] 
                    This follows from (a).
    
                    \item[$\SchedUpdate(\SchedState_0, \BufProject{\buf_0}) = \SchedUpdate(\SchedState_0', \BufProject{\buf_0'})$:]
                    From (a), we have   $\SchedState_0 = \SchedState_0'$.
                    From $\DeepProject{\buf_1} = \DeepProject{\buf_1'}$ and Lemma~\ref{lemma:vanilla:buffer-projections}, we have $\BufProject{\buf_1} = \BufProject{\buf_1'}$.
                    Therefore, $\SchedUpdate(\SchedState_0, \BufProject{\buf_1}) = \SchedUpdate(\SchedState_0', \BufProject{\buf_1'})$.
                \end{description}
                Therefore, $C_1 \sim C_1'$ and (c) holds.
            \end{description}
            
            \item[$\apply{\buf_0}{a_0}(\pc) = \bot \vee |\buf_0| \geq \wMuarch$:]
            Then, from (a), we immediately get that $\apply{\buf_0'}{a_0'}(\pc) = \bot \vee |\buf_0'| \geq \wMuarch$ holds as well.
            Therefore, both computations are stuck and (c) holds.
        \end{description}
        Therefore, (c) holds for all the cases.

    \item[$d = \execute{i}$:]
    Therefore, we can only apply one of the $\execute{}$ rules.
    There are two cases:
    \begin{description}
        \item[$i \leq |\buf_0| \wedge \pbarrier \not\in {\buf_0[0..i-1]}$:]
        There are several cases depending on the $i$-th command in the reorder buffer:
        \begin{description}
            \item[$\elt{\buf_0}{i} = \tagged{\pload{x}{e}}{T}$:]
            From (a), we also have that $\elt{\buf_0'}{i} =  \tagged{\pload{x}{e}}{T}$, $i \leq |\buf_0'|$, and $\pbarrier \not\in \buf_0'[0..i-1]$. 
            There are two cases:
            \begin{description}
                \item[$\pstore{x'}{e'} \not\in \buf_0{[0..i-1]}$:]
                We now show that $\exprEval{e}{\apply{\buf_0[0..i-1]}{a_0}} = \exprEval{e}{\apply{\buf_0'[0..i-1]}{a_0'}}$.
                Since $C_0, C_0'$ are reachable configurations, the buffers $\buf_0, \buf_0'$ are well-formed (see Lemma~\ref{lemma:vanilla:buffers-well-formedness}), and therefore  $\buf_0[0..i-1] \in \prefixes{\buf_0}$ and  $\buf_0'[0..i-1] \in \prefixes{\buf_0'}$.
                From (b), therefore, there are configurations $s_0, s_0', s_1, s_1'$ such that $C_0 \bufEquiv{|\buf_0[0..i-1]|} s_0$, $C_0' \bufEquiv{|\buf_0[0..i-1]|} s_0'$, $s_0 \CtSpecInterfStep{\tau}{} s_1$,  $s_0' \CtSpecInterfStep{\tau'}{} s_1'$, $\headWindow{s_0} > 0$, $\headWindow{s_0'} >0$, and $\tau = \tau'$. 
                From (a), $C_0 \bufEquiv{|\buf_0[0..i-1]|} s_0$, $C_0' \bufEquiv{|\buf_0[0..i-1]|} s_0'$, and the well-formedness of the buffers, we know that $p(\headConf{s_0}(\pc)) = p(\headConf{s_0'}(\pc)) = \pload{x}{e}$.
                From $\CtSpecInterf{\cdot}$, $\headWindow{s_0} > 0$, $\headWindow{s_0'} >0$, and $p(\headConf{s_0}(\pc)) = p(\headConf{s_0'}(\pc)) = \pload{x}{e}$, we have that $\tau = \loadObs{ \exprEval{e}{\headConf{s_0} } }$  and $\tau' = \loadObs{ \exprEval{e}{ \headConf{s_0'} } }$ (because $s_0 \CtSpecInterfStep{\tau}{} s_1$ and $s_0' \CtSpecInterfStep{\tau'}{} s_1'$ have been obtained by applying the \textsc{Step} rule of $\CtSpecInterfStep{}{}$ and the \textsc{Load} rule of $\CtSeqInterfStep{}{})$.
                From $\tau=\tau'$, we get that $\exprEval{e}{\headConf{s_0}} = \exprEval{e}{\headConf{s_0'}}$.
                From this, $C_0 \bufEquiv{|\buf_0[0..i-1]|} s_0$, and $C_0' \bufEquiv{|\buf_0'[0..i-1]|} s_0'$, we finally get $\exprEval{e}{\apply{\buf_0{[0..i-1]}}{a_0}} = \exprEval{e}{\apply{\buf_0'{[0..i-1]}}{a_0'}(x)}$.
    
                Let $n = \exprEval{e}{\apply{\buf_0{[0..i-1]}}{a_0}} = \exprEval{e}{\apply{\buf_0'{[0..i-1]}}{a_0'}(x)}$.
                There are two cases:
                \begin{description}
                    \item[$\CacheAccess(\CacheState_0, \exprEval{e}{ \apply{\buf_0{[0..i-1]}}{a_0} }) = \CacheHit$:]
                    From (a), we have that $\CacheState_0 = \CacheState_0'$.
                    Moreover, we have already shown that $\exprEval{e}{ \apply{\buf_0{[0..i-1]}}{a_0} } = \exprEval{e}{ \apply{\buf_0'{[0..i-1]}}{a_0'} }$.
                    Therefore, we can apply the \textsc{Execute-Load-Hit} and \textsc{Step} rules to $C_0$ and $C_0'$ as follows:
                    \begin{align*}
                        \buf_0 &:= \buf_0[0..i-1] \concat \tagged{\pload{x}{e}}{T} \concat \buf_0[i+1 .. |\buf_0|]\\
                        \buf_1 &:= \buf_0[0..i-1] \concat \tagged{\passign{x}{m_0(n)}}{T} \concat \buf_0[i+1 .. |\buf_0|]\\
                        \tup{m_0,a_0,\buf_0, \CacheState_0, \BpState_0} &\muarchStep{\execute{i}}{} \tup{m_0,a_0,\buf_1, \CacheUpdate(\CacheState_0,n), \BpState_0}\\
                        \tup{m_0,a_0,\buf_0, \CacheState_0,  \BpState_0 , \SchedState_0} &\muarchStep{}{} \tup{m_0,a_0,\buf_1, \CacheUpdate(\CacheState_0,n), \BpState_0,\SchedUpdate(\SchedState_0, \BufProject{\buf_0})}\\
                        \buf_0' &:= \buf_0'[0..i-1] \concat \tagged{\pload{x}{e}}{T} \concat \buf_0'[i+1 .. |\buf_0|]\\
                        \buf_1' &:= \buf_0'[0..i-1] \concat \tagged{\passign{x}{m_0'(n)}}{T} \concat \buf_0'[i+1 .. |\buf_0|]\\
                        \tup{m_0',a_0',\buf_0', \CacheState_0', \BpState_0'} &\muarchStep{\execute{i}}{} \tup{m_0',a_0',\buf_1', \CacheUpdate(\CacheState_0',n), \BpState_0'}\\
                        \tup{m_0',a_0',\buf_0', \CacheState_0', \BpState_0', \SchedState_0'} &\muarchStep{}{} \tup{m_0',a_0',\buf_1', \CacheUpdate(\CacheState_0',n), \BpState_0',\SchedUpdate(\SchedState_0', \BufProject{\buf_0'})}
                    \end{align*}
                    We now show that $C_1 = \tup{m_0,a_0,\buf_1, \CacheUpdate(\CacheState_0,n), \BpState_0,\SchedUpdate(\SchedState_0, \BufProject{\buf_0})}$ and $C_1' = \tup{m_0',a_0',\buf_1', \CacheUpdate(\CacheState_0',n), \BpState_0',\SchedUpdate(\SchedState_0', \BufProject{\buf_0'})}$ are indistinguishable, i.e., i.e., $C_1 \sim C_1'$.
        
                    For this, we need to show that:
                        \begin{description}
                            \item[$\apply{\buf_1}{a_0}(\pc) = \apply{\buf_1'}{a_0'}(\pc)$:]
                            This immediately follows from (a) and the fact that \textbf{load}s do not modify $\pc$.
                
                            \item[$\DeepProject{\buf_1} = \DeepProject{\buf_1'}$:] 
                            This immediately follows from (a) and $x \neq \pc$ (from the well-formedness of the buffers).

                            \item[$\CacheUpdate(\CacheState_0,n) = \CacheUpdate(\CacheState_0',n)$:]
                            This follows from $\CacheState_0 = \CacheState_0'$, which follows from (a).
                
                            \item[$\BpState_0= \BpState_0':$] 
                            This follows from  (a).
                
                            \item[$\SchedUpdate(\SchedState_0, \BufProject{\buf_1}) = \SchedUpdate(\SchedState_0', \BufProject{\buf_1'})$:]
                            From (a), we have   $\SchedState_0 = \SchedState_0'$.
                            From $\DeepProject{\buf_1} = \DeepProject{\buf_1'}$ and Lemma~\ref{lemma:vanilla:buffer-projections}, we have $\BufProject{\buf_1} = \BufProject{\buf_1'}$.
                            Therefore, $\SchedUpdate(\SchedState_0, \BufProject{\buf_1}) = \SchedUpdate(\SchedState_0', \BufProject{\buf_1'})$.
                        \end{description}
                        Therefore, $C_1 \sim C_1'$ and (c) holds.
    
                    \item[$\CacheAccess(\CacheState_0, \exprEval{e}{\apply{\buf_0{[0..i-1]}}{a_0}}) = \CacheMiss$:]
                    The proof of this case is similar to the one for the $\CacheHit$ case (except that we apply the \textsc{Execute-Load-Miss} rule). 
                \end{description}
    
                \item[$\pstore{x'}{e'} \in \buf_0{[0..i-1]}$:]
                From (a), we also have that $\pstore{x'}{e'} \in \buf_0'[0..i-1]$.
                Therefore, both computations are stuck and (c) holds.
    
            \end{description}
    
            \item[$\elt{\buf_0}{i} =  \tagged{\passign{\pc}{\lbl}}{\lbl_0} \wedge \ell_0 \neq \emptysequence$:]
            From (a), we also have that $\elt{\buf_0'}{i} =  \tagged{\passign{\pc}{\lbl}}{\lbl_0} \wedge \ell_0 \neq \emptysequence$, $i \leq |\buf_0'|$, and $\pbarrier \not\in \buf_0'[0..i-1]$. 
            Observe that $p(\lbl_0) = \pjz{x}{\lbl''}$.
    
            We now show that $\apply{\buf_0[0..i-1]}{a_0}(x) = \apply{\buf_0'[0..i-1]}{a_0'}(x)$.
            Since $C_0, C_0'$ are reachable configurations, the buffers $\buf_0, \buf_0'$ are well-formed (see Lemma~\ref{lemma:vanilla:buffers-well-formedness}), and therefore  $\buf_0[0..i-1] \in \prefixes{\buf_0}$ and  $\buf_0'[0..i-1] \in \prefixes{\buf_0'}$.
            From (b), therefore, there are configurations $s_0, s_0', s_1, s_1'$ such that $C_0 \bufEquiv{|\buf_0[0..i-1]|} s_0$, $C_0' \bufEquiv{|\buf_0[0..i-1]|} s_0'$, $s_0 \CtSpecInterfStep{\tau}{} s_1$,  $s_0' \CtSpecInterfStep{\tau'}{} s_1'$, $\headWindow{s_0} > 0$, $\headWindow{s_0'} >0$, and $\tau = \tau'$. 
            From (a), $C_0 \bufEquiv{|\buf_0[0..i-1]|} s_0$, $C_0' \bufEquiv{|\buf_0[0..i-1]|} s_0'$, and the well-formedness of the buffers, we know that $p(\headConf{s_0}(\pc)) = p(\headConf{s_0'}(\pc)) = \pjz{x}{\lbl''}$.
            From $\CtSpecInterf{\cdot}$, $\headWindow{s_0} > 0$, $\headWindow{s_0'} >0$, and  $p(\headConf{s_0}(\pc)) = p(\headConf{s_0'}(\pc)) = \pjz{x}{\lbl''}$, we have that $(\tau = \pcObs{ \lbl'' } \leftrightarrow \headConf{s_0}(x) = 0) \wedge (\tau = \pcObs{ \sigma_0(\pc)+1 } \leftrightarrow \headConf{s_0}(x) \neq 0)$  and $(\tau' = \pcObs{ \lbl'' } \leftrightarrow \headConf{s_0'}(x) = 0) \wedge (\tau' = \pcObs{ \sigma_0(\pc)+1 } \leftrightarrow \headConf{s_0'}(x) \neq 0)$ (because $s_0 \CtSpecInterfStep{\tau}{} s_1$ and  $s_0' \CtSpecInterfStep{\tau'}{} s_1'$ have been obtained by applying the \textsc{Branch} rule of $\CtSpecInterfStep{}{}$).
            From $\tau=\tau'$, we get that $\headConf{s_0}(x) = \headConf{s_0'}(x)$.
            From this, $C_0 \bufEquiv{|\buf_0[0..i-1]|} s_0$, and $C_0' \bufEquiv{|\buf_0'[0..i-1]|} s_0'$, we finally get $\apply{\buf_0{[0..i-1]}}{a_0}(x) = \apply{\buf_0'{[0..i-1]}}{a_0'}(x)$.
    
            Given that $\apply{\buf_0[0..i-1]}{a_0}(x) = \apply{\buf_0'[0..i-1]}{a_0'}(x)$, there are two cases:
            \begin{description}
                \item[$( \apply{\buf_0[0..i-1]}{a_0}(x) = 0 \wedge \lbl = \lbl'') \vee (\apply{\buf_0[0..i-1]}{a_0}(x) \in \Val \setminus \{0,\bot\} \wedge \lbl = \ell_0+1)$:]
                From $\apply{\buf_0{[0..i-1]}}{a_0}(x) = \apply{\buf_0'{[0..i-1]}}{a_0'}(x)$ and (a), we also get $( \apply{\buf_0'[0..i-1]}{a_0'}(x) = 0 \wedge \lbl = \lbl'') \vee (\apply{\buf_0'[0..i-1]}{a_0'}(x) \in \Val \setminus \{0,\bot\} \wedge \lbl = \ell_0+1)$.
                Therefore,  we can apply the \textsc{Execute-Branch-Commit} and \textsc{Step} rules to $C_0$ and $C_0'$ as follows:
                \begin{align*}
                    \buf_0 &:= \buf_0[0..i-1] \concat \tagged{\passign{\pc}{\ell}}{\ell_0} \concat \buf_0[i+1 .. |\buf_0|]\\
                    \buf_1 &:= \buf_0[0..i-1] \concat \tagged{\passign{\pc}{\ell}}{\notags} \concat \buf_0[i+1 .. |\buf_0|]\\
                    \tup{m_0,a_0,\buf_0, \CacheState_0, \BpState_0} &\muarchStep{\execute{i}}{} \tup{m_0,a_0,\buf_1, \CacheState_0, \BpState_0}\\
                    \tup{m_0,a_0,\buf_0, \CacheState_0,  \BpUpdate(\BpState_0, \ell_0, \ell) , \SchedState_0} &\muarchStep{}{} \tup{m_0,a_0,\buf_1, \CacheState_0, \BpUpdate(\BpState_0, \ell_0, \ell),\SchedUpdate(\SchedState_0, \BufProject{\buf_0})}\\
                    \buf_0' &:= \buf_0'[0..i-1] \concat \tagged{\passign{\pc}{\ell}}{\ell_0} \concat \buf_0'[i+1 .. |\buf_0|]\\
                    \buf_1' &:= \buf_0'[0..i-1] \concat \tagged{\passign{\pc}{\ell}}{\notags} \concat \buf_0'[i+1 .. |\buf_0|]\\
                    \tup{m_0',a_0',\buf_0', \CacheState_0', \BpUpdate(\BpState_0', \ell_0, \ell)} &\muarchStep{\execute{i}}{} \tup{m_0',a_0',\buf_1', \CacheState_0', \BpState_0'}\\
                    \tup{m_0',a_0',\buf_0', \CacheState_0', \BpState_0', \SchedState_0'} &\muarchStep{}{} \tup{m_0',a_0',\buf_1', \CacheState_0', \BpUpdate(\BpState_0', \ell_0, \ell),\SchedUpdate(\SchedState_0', \BufProject{\buf_0'})}
                \end{align*}
                We now show that $C_1 = \tup{m_0,a_0,\buf_1, \CacheState_0, \BpUpdate(\BpState_0, \ell_0, \ell),\SchedUpdate(\SchedState_0, \BufProject{\buf_0})}$ and $C_1' = \tup{m_0',a_0',\buf_1', \CacheState_0', \BpUpdate(\BpState_0', \ell_0, \ell),\SchedUpdate(\SchedState_0', \BufProject{\buf_0'})}$ are indistinguishable, i.e., i.e., $C_1 \sim C_1'$.
    
                For this, we need to show that:
                    \begin{description}
                        \item[$\apply{\buf_1}{a_0}(\pc) = \apply{\buf_1'}{a_0'}(\pc)$:]
                        This immediately follows from (a) and the fact that we set $\pc$ to $\ell$ in both computations.
            
                        \item[$\DeepProject{\buf_1} = \DeepProject{\buf_1'}$:] 
                        This immediately follows from (a).

                        \item[$\CacheState_0 = \CacheState_0'$:]
                        This follows from (a).
            
                        \item[$\BpUpdate(\BpState_0, \ell_0, \ell) = \BpUpdate(\BpState_0', \ell_0, \ell):$] 
                        This follows from $\BpState_0 = \BpState_0'$, which follows from (a).
            
                        \item[$\SchedUpdate(\SchedState_0, \BufProject{\buf_1}) = \SchedUpdate(\SchedState_0', \BufProject{\buf_1'})$:]
                        From (a), we have   $\SchedState_0 = \SchedState_0'$.
                        From $\DeepProject{\buf_1} = \DeepProject{\buf_1'}$ and Lemma~\ref{lemma:vanilla:buffer-projections}, we have $\BufProject{\buf_1} = \BufProject{\buf_1'}$.
                        Therefore, $\SchedUpdate(\SchedState_0, \BufProject{\buf_1}) = \SchedUpdate(\SchedState_0', \BufProject{\buf_1'})$.
                    \end{description}
                    Therefore, $C_1 \sim C_1'$ and (c) holds.
    
                \item[$( \apply{\buf_0[0..i-1]}{a_0}(x) = 0 \wedge \lbl \neq \lbl'') \vee (\apply{\buf_0[0..i-1]}{a_0}(x) \in \Val \setminus \{0,\bot\} \wedge \lbl \neq \ell_0+1)$:]
                The proof of this case is similar to the one of the $( \apply{\buf_0[0..i-1]}{a_0}(x) = 0 \wedge \lbl = \lbl'') \vee (\apply{\buf_0[0..i-1]}{a_0}(x) \in \Val \setminus \{0,\bot\} \wedge \lbl = \ell_0+1)$ (except that we apply the \textsc{Execute-Branch-Rollback} rule).
            \end{description}
    
            \item[$\elt{\buf_0}{i} = \tagged{\passign{x}{e}}{\notags}$:]
            From (a), we also have that $\elt{\buf_0'}{i} =  \tagged{\passign{x}{e'}}{\notags}$, $i \leq |\buf_0'|$, and $\pbarrier \not\in \buf_0'[0..i-1]$.
            There are two cases:
            \begin{description}
                \item[$\exprEval{e}{\apply{\buf_0[0..i-1]}{a_0}} \neq \bot$:]
                From (a), we also have that $\exprEval{e'}{\apply{\buf_0'[0..i-1]}{a_0'}} \neq \bot$ (indeed, if $e \in \Val$, then $e'$ has to be in $\Val$ as well, and if $e \not\in \Val$, then $e = e'$ and $e$'s dependencies must be resolved in both $\buf_0[0..i-1]$ and $\buf_0'[0..i-1]$ from (a)).
                Therefore,  we can apply the \textsc{Execute-Assignment} and \textsc{Step} rules to $C_0$ and $C_0'$ as follows:
                \begin{align*}
                    v &:= \exprEval{e}{\apply{\buf_0{[0..i-1]}}{a_0}}\\
                    \buf_0 &:= \buf_0[0..i-1] \concat \tagged{\passign{x}{e}}{T} \concat \buf_0[i+1 .. |\buf_0|]\\
                    \buf_1 &:= \buf_0[0..i-1] \concat \tagged{\passign{x}{v}}{T} \concat \buf_0[i+1 .. |\buf_0|]\\
                    \tup{m_0,a_0,\buf_0, \CacheState_0, \BpState_0} &\muarchStep{\execute{i}}{} \tup{m_0,a_0,\buf_1, \CacheState_0, \BpState_0}\\
                    \tup{m_0,a_0,\buf_0, \CacheState_0, \BpState_0, \SchedState_0} &\muarchStep{}{} \tup{m_0,a_0,\buf_1, \CacheState_0, \BpState_0,\SchedUpdate(\SchedState_0, \BufProject{\buf_0})}\\
                    v' &:= \exprEval{e'}{\apply{\buf_0'{[0..i-1]}}{a_0'}}\\
                    \buf_0' &:= \buf_0'[0..i-1] \concat \tagged{\passign{x}{e'}}{T} \concat \buf_0'[i+1 .. |\buf_0|]\\
                    \buf_1' &:= \buf_0'[0..i-1] \concat \tagged{\passign{x}{v'}}{T} \concat \buf_0'[i+1 .. |\buf_0|]\\
                    \tup{m_0',a_0',\buf_0', \CacheState_0', \BpState_0'} &\muarchStep{\execute{i}}{} \tup{m_0',a_0',\buf_1', \CacheState_0', \BpState_0'}\\
                    \tup{m_0',a_0',\buf_0', \CacheState_0', \BpState_0', \SchedState_0'} &\muarchStep{}{} \tup{m_0',a_0',\buf_1', \CacheState_0', \BpState_0',\SchedUpdate(\SchedState_0', \BufProject{\buf_0'})}
                \end{align*}
                We now show that $C_1 = \tup{m_0,a_0,\buf_1, \CacheState_0, \BpState_0,\SchedUpdate(\SchedState_0, \BufProject{\buf_0})}$ and $C_1' = \tup{m_0',a_0',\buf_1', \CacheState_0', \BpState_0',\SchedUpdate(\SchedState_0', \BufProject{\buf_0'})}$ are indistinguishable, i.e., i.e., $C_1 \sim C_1'$.
    
                There are two cases:
                \begin{description}
                    \item[$x \neq \pc$:]
                    For $C_1 \sim C_1'$, we need to show that:
                    \begin{description}
                        \item[$\apply{\buf_1}{a_0}(\pc) = \apply{\buf_1'}{a_0'}(\pc)$:]
                        This immediately follows from (a) and $x \neq \pc$.
            
                        \item[$\DeepProject{\buf_1} = \DeepProject{\buf_1'}$:] 
                        This immediately follows from (a) and $x \neq \pc$.

                        \item[$\CacheState_0 = \CacheState_0'$:]
                        This follows from (a).
            
                        \item[$\BpState_0 = \BpState_0':$] 
                        This follows from (a).
            
                        \item[$\SchedUpdate(\SchedState_0, \BufProject{\buf_1}) = \SchedUpdate(\SchedState_0', \BufProject{\buf_1'})$:]
                        From (a), we have   $\SchedState_0 = \SchedState_0'$.
                        From $\DeepProject{\buf_1} = \DeepProject{\buf_1'}$ and Lemma~\ref{lemma:vanilla:buffer-projections}, we have $\BufProject{\buf_1} = \BufProject{\buf_1'}$.
                        Therefore, $\SchedUpdate(\SchedState_0, \BufProject{\buf_1}) = \SchedUpdate(\SchedState_0', \BufProject{\buf_1'})$.
                    \end{description}
                    Therefore, $C_1 \sim C_1'$ and (c) holds.
    
                    \item[$x = \pc$:]
                    For $C_1 \sim C_1'$, we need to show that:
                    \begin{description}
                        \item[$\apply{\buf_1}{a_0}(\pc) = \apply{\buf_1'}{a_0'}(\pc)$:]
                        For this, we need to show that $v = v'$ (in case there are no later changes to the program counter).
                        Since $C_0, C_0'$ are reachable configurations, the buffers $\buf_0, \buf_0'$ are well-formed (see Lemma~\ref{lemma:vanilla:buffers-well-formedness}), and therefore  $\buf_0[0..i-1] \in \prefixes{\buf_0}$ and  $\buf_0'[0..i-1] \in \prefixes{\buf_0'}$.
                        From (b), therefore, there are configurations $s_0, s_0', s_1, s_1'$ such that $C_0 \bufEquiv{|\buf_0[0..i-1]|} s_0$, $C_0' \bufEquiv{|\buf_0[0..i-1]|} s_0'$, $s_0 \CtSpecInterfStep{\tau}{} s_1$,  $s_0' \CtSpecInterfStep{\tau'}{} s_1'$, $\headWindow{s_0} > 0$, $\headWindow{s_0'} >0$, and $\tau = \tau'$. 
                        There are two cases:
                        \begin{description}
                            \item[$e \in \Val$:] 
                            Then,  $e = e'$ follows from (a) and, therefore, we immediately have $v = v'$.
    
                            \item[$e \not\in \Val$:] 
                            Then, from (a) we have $e = e'$.
                            From (a), $C_0 \bufEquiv{|\buf_0[0..i-1]|} s_0$, $C_0' \bufEquiv{|\buf_0[0..i-1]|} s_0'$, and the well-formedness of the buffers, we know that $p(\headConf{s_0}(\pc)) = p(\headConf{s_0'}(\pc)) = \pjmp{e}$.
                            From $\CtSpecInterf{\cdot}$, $\headWindow{s_0} > 0$, and $\headWindow{s_0'} >0$, we have that $\tau = \pcObs{ \exprEval{e}{\headConf{s_0}} }$ and $\tau' = \pcObs{ \exprEval{e}{\headConf{s_0'}}}$ (because $s_0 \CtSpecInterfStep{\tau}{} s_1$ and $s_0' \CtSpecInterfStep{\tau'}{} s_1'$ have been obtained by applying the \textsc{Step} rule of $\CtSpecInterfStep{}{}$ and the \textsc{Jump} rule of $\CtSeqInterfStep{}{}$).
                            From $C_0 \bufEquiv{|\buf_0[0..i-1]|} s_0$ and $\tau = \pcObs{ \exprEval{e}{ \headConf{s_0} } }$, we have that $\tau = \pcObs{ \exprEval{e}{\apply{\buf_0{[0..i-1]}}{a_0}}}$.
                            Similarly, from $C_0' \bufEquiv{|\buf_0'[0..i-1]|} s_0'$ and $\tau' = \pcObs {\exprEval{e}{\headConf{s_0'}}}$, we have that $\tau' = \pcObs{\exprEval{e}{\apply{\buf_0'{[0..i-1]}}{a_0'}}}$.
                            Finally, from $\tau=\tau'$, we get $\exprEval{e}{\apply{\buf_0{[0..i-1]}}{a_0}} = \exprEval{e}{\apply{\buf_0'{[0..i-1]}}{a_0'}}$ and, therefore, $v = v'$.
                        \end{description}
                            
                        \item[$\DeepProject{\buf_1} = \DeepProject{\buf_1'}$:] 
                        This immediately follows from (a) and $v = v'$ (shown above).
            
                        \item[$\CacheState_0 = \CacheState_0'$:]
                        This follows from (a).
            
                        \item[$\BpState_0 = \BpState_0':$] 
                        This follows from (a).
            
                        \item[$\SchedUpdate(\SchedState_0, \BufProject{\buf_1}) = \SchedUpdate(\SchedState_0', \BufProject{\buf_1'})$:]
                        From (a), we have   $\SchedState_0 = \SchedState_0'$.
                        From $\DeepProject{\buf_1} = \DeepProject{\buf_1'}$ and Lemma~\ref{lemma:seq-processor:buffer-projections}, we have $\BufProject{\buf_1} = \BufProject{\buf_1'}$.
                        Therefore, $\SchedUpdate(\SchedState_0, \BufProject{\buf_1}) = \SchedUpdate(\SchedState_0', \BufProject{\buf_1'})$.
                    \end{description}
                    Therefore, $C_1 \sim C_1'$ and (c) holds.
                \end{description} 
    
                \item[$\exprEval{e}{\apply{\buf_0[0..i-1]}{a_0}} = \bot$:] 
                From this, it follows that $e \not\in \Val$.
                Therefore, from (a), we have that $e = e'$.
                
                Observe that $\exprEval{e}{\apply{\buf[0..i-1]}{a_0}} = \bot$ implies that one of the dependencies of $e$ is unresolved in $\buf_0[0..i-1]$.
                From this and (a), it follows that one of the dependencies of $e$ is unresolved in $\buf_0'[0..i-1]$.
                Therefore, $\exprEval{e}{\apply{\buf_0'{[0..i-1]}}{a_0'}} = \bot $ holds as well.
                Hence, both configurations are stuck and (c) holds.
    
            \end{description}

            \item[$\elt{\buf_0}{i} =   \tagged{\pmarkedassign{x}{e}}{\notags}$:]
            The proof of this case is similar to that of $\elt{\buf_0}{i} = \tagged{\passign{x}{e}}{\notags}$ (when $x = \pc$).
    
            \item[$\elt{\buf_0}{i} =  \tagged{\pstore{x}{e}}{T}$:]
            From (a), we also have that $\elt{\buf_0'}{i} =  \tagged{\pstore{x}{e}}{T}$, $i \leq |\buf_0'|$, and $\pbarrier \not\in \buf_0'[0..i-1]$.
            There are two cases:
            \begin{description}
                \item[$\exprEval{e}{\apply{\buf_0{[0..i-1]}}{a_0}} \neq \bot \wedge \apply{\buf_0{[0..i-1]}}{a_0}(x) \neq \bot$:] 
                From (a), we have that  $\exprEval{e}{\apply{\buf_0'{[0..i-1]}}{a_0'}} \neq \bot \wedge \apply{\buf_0'{[0..i-1]}}{a_0'}(x) \neq \bot$ holds as well.
                Therefore,  we can apply the \textsc{Execute-Store} and \textsc{Step} rules to $C_0$ and $C_0'$ as follows:
                \begin{align*}
                    v &:= \apply{\buf_0{[0..i-1]}}{a_0}(x) \\
                    n &:= \exprEval{e}{\apply{\buf_0{[0..i-1]}}{a_0}}\\
                    \buf_0 &:= \buf_0[0..i-1] \concat \tagged{\pstore{x}{e}}{T} \concat \buf_0[i+1 .. |\buf_0|]\\
                    \buf_1 &:= \buf_0[0..i-1] \concat \tagged{\pstore{v}{n}}{T} \concat \buf_0[i+1 .. |\buf_0|]\\
                    \tup{m_0,a_0,\buf_0, \CacheState_0, \BpState_0} &\muarchStep{\execute{i}}{} \tup{m_0,a_0,\buf_1, \CacheState_0, \BpState_0}\\
                    \tup{m_0,a_0,\buf_0, \CacheState_0, \BpState_0, \SchedState_0} &\muarchStep{}{} \tup{m_0,a_0,\buf_1, \CacheState_0, \BpState_0,\SchedUpdate(\SchedState_0, \BufProject{\buf_0})}\\
                    v' &:= \apply{\buf_0'{[0..i-1]}}{a_0'}(x) \\
                    n' &:= \exprEval{e}{\apply{\buf_0'{[0..i-1]}}{a_0'}}\\
                    \buf_0' &:= \buf_0'[0..i-1] \concat \tagged{\pstore{x}{e}}{T} \concat \buf_0'[i+1 .. |\buf_0|]\\
                    \buf_1' &:= \buf_0'[0..i-1] \concat \tagged{\pstore{v'}{n'}}{T} \concat \buf_0'[i+1 .. |\buf_0|]\\
                    \tup{m_0',a_0',\buf_0', \CacheState_0', \BpState_0'} &\muarchStep{\execute{i}}{} \tup{m_0',a_0',\buf_1', \CacheState_0', \BpState_0'}\\
                    \tup{m_0',a_0',\buf_0', \CacheState_0', \BpState_0', \SchedState_0'} &\muarchStep{}{} \tup{m_0',a_0',\buf_1', \CacheState_0', \BpState_0',\SchedUpdate(\SchedState_0', \BufProject{\buf_0'})}
                \end{align*}
                We now show that $C_1 = \tup{m_0,a_0,\buf_1, \CacheState_0, \BpState_0,\SchedUpdate(\SchedState_0, \BufProject{\buf_0})}$ and $C_1' = \tup{m_0',a_0',\buf_1', \CacheState_0', \BpState_0',\SchedUpdate(\SchedState_0', \BufProject{\buf_0'})}$ are indistinguishable, i.e., i.e., $C_1 \sim C_1'$.
                For this, we need to show that:
                \begin{description}
                    \item[$\apply{\buf_1}{a_0}(\pc) = \apply{\buf_1'}{a_0'}(\pc)$:]
                    This immediately follows from (a) and the fact that \textbf{store}s do not alter the value of $\pc$.
        
                    \item[$\DeepProject{\buf_1} = \DeepProject{\buf_1'}$:] 
                    For this, we need to show $n = n'$:
                    \begin{description}
                        \item[$n = n'$:]
                        Since $C_0, C_0'$ are reachable configurations, the buffers $\buf_0, \buf_0'$ are well-formed (see Lemma~\ref{lemma:vanilla:buffers-well-formedness}), and therefore  $\buf_0[0..i-1] \in \prefixes{\buf_0}$ and  $\buf_0'[0..i-1] \in \prefixes{\buf_0'}$.
                        From (b), therefore, there are configurations $s_0, s_0', s_1, s_1'$ such that $C_0 \bufEquiv{|\buf_0[0..i-1]|} s_0$, $C_0' \bufEquiv{|\buf_0[0..i-1]|} s_0'$, $s_0 \CtSpecInterfStep{\tau}{} s_1$,  $s_0' \CtSpecInterfStep{\tau'}{} s_1'$, $\headWindow{s_0} > 0$, $\headWindow{s_0'} >0$, and $\tau = \tau'$.  
                        From (a), $C_0 \bufEquiv{|\buf_0[0..i-1]|} s_0$, $C_0' \bufEquiv{|\buf_0[0..i-1]|} s_0'$, and the well-formedness of the buffers, we know that $p(\headConf{s_0}(\pc)) = p(\headConf{s_0'}(\pc)) = \pstore{x}{e}$.
                        From $\CtSpecInterf{\cdot}$, $\headWindow{s_0} > 0$, $\headWindow{s_0'} >0$, we have that $\tau = \storeObs{ \exprEval{e}{\sigma_0} }$ and $\tau' = \storeObs{ \exprEval{e}{\sigma_0'}}$ (because $s_0 \CtSpecInterfStep{\tau}{} s_1$ and  $s_0' \CtSpecInterfStep{\tau'}{} s_1'$ have been obtained by applying the \textsc{Step} rule of $\CtSpecInterfStep{}{}$ and the \textsc{Store} rule of $\CtSeqInterfStep{}{}$).
                        From $C_0 \bufEquiv{|\buf_0[0..i-1]|} s_0$ and $\tau = \storeObs{ \exprEval{e}{\headConf{s_0}} }$, we have that $\tau = \storeObs{ \exprEval{e}{\apply{\buf_0{[0..i-1]}}{a_0}}}$.
                        Similarly, from $C_0' \bufEquiv{|\buf_0'[0..i-1]|} \sigma_0'$ and $\tau' = \storeObs {\exprEval{e}{\headConf{s_0'}}}$, we have that $\tau' = \storeObs{\exprEval{e}{\apply{\buf_0'{[0..i-1]}}{a_0'}}}$.
                        Finally, from $\tau=\tau'$, we get $\exprEval{e}{\apply{\buf_0{[0..i-1]}}{a_0}} = \exprEval{e}{\apply{\buf_0'{[0..i-1]}}{a_0'}}$ and, therefore, $n = n'$.
    
                    \end{description}
                    From (a), $n = n'$, $\buf_0 = \buf_0[0..i-1] \concat \tagged{\pstore{x}{e}}{T} \concat \buf_0[i+1 .. |\buf_0|]$, $				\buf_1 = \buf_0[0..i-1] \concat \tagged{\pstore{v}{n}}{T} \concat \buf_0[i+1 .. |\buf_0|]$, $\buf_0' = \buf_0'[0..i-1] \concat \tagged{\pstore{x}{e}}{T} \concat \buf_0'[i+1 .. |\buf_0|]$, and $				\buf_1' = \buf_0'[0..i-1] \concat \tagged{\pstore{v'}{n'}}{T} \concat \buf_0'[i+1 .. |\buf_0|]$, we get $\DeepProject{\buf_1} = \DeepProject{\buf_1'}$.

                    \item[$\CacheState_0 = \CacheState_0'$:]
                    This follows from (a).
        
                    \item[$\BpState_0 = \BpState_0':$] 
                    This follows from (a).
        
                    \item[$\SchedUpdate(\SchedState_0, \BufProject{\buf_1}) = \SchedUpdate(\SchedState_0', \BufProject{\buf_1'})$:]
                    From (a), we have   $\SchedState_0 = \SchedState_0'$.
                    From $\DeepProject{\buf_1} = \DeepProject{\buf_1'}$ and Lemma~\ref{lemma:vanilla:buffer-projections}, we have $\BufProject{\buf_1} = \BufProject{\buf_1'}$.
                    Therefore, $\SchedUpdate(\SchedState_0, \BufProject{\buf_1}) = \SchedUpdate(\SchedState_0', \BufProject{\buf_1'})$.
                \end{description}
                Therefore, $C_1 \sim C_1'$ and (c) holds.
    
                \item[$\exprEval{e}{\apply{\buf_0{[0..i-1]}}{a_0}} = \bot \vee \apply{\buf_0{[0..i-1]}}{a_0}(x) = \bot$:] 
                Then, one of the dependencies of $e$ or $x$ are unresolved in $\buf_0[0..i-1]$.
                From this and (a), it follows that one of the dependencies of $e$ or $x$ are unresolved in $\buf_0'[0..i-1]$.
                Therefore, $\exprEval{e}{\apply{\buf_0'{[0..i-1]}}{a_0'}} = \bot \vee \apply{\buf_0'{[0..i-1]}}{a_0'}(x) = \bot$ holds as well.
                Hence, both configurations are stuck and (c) holds.
            \end{description}
    
            \item[$\elt{\buf_0}{i} =  \tagged{\pskip{}}{\notags}$:]
            The proof of this case is similar to that of $\elt{\buf_0}{i} =  \tagged{\passign{x}{e}}{\notags}$.
            \item[$\elt{\buf_0}{i} =  \tagged{\pbarrier}{T}$:]
            The proof of this case is similar to that of $\elt{\buf_0}{i} =  \tagged{\passign{x}{e}}{\notags}$.
        \end{description}
        Therefore, (c) holds in all cases.

        \item[$i > |\buf_0| \vee \pbarrier \in \buf_0{[0..i-1]}$:]
        From (a), it immediately follows that $i > |\buf_0'| \vee \pbarrier \in \buf_0'[0..i-1]$.
        Therefore, both configurations are stuck and (c) holds.
    \end{description}
    Therefore, (c) holds in all cases.

    \item[$d = \retire{}$:]
        Therefore, we can only apply one of the $\retire{}$ rules depending on the head of the reorder buffer in $\buf_0$.
        There are five cases:
        \begin{description}
            \item[$\buf_0 = \tagged{\pskip}{\notags} \concat \buf_1 $:] 
            From (a), we get that $\DeepProject{\buf_0} = \DeepProject{\buf_0'}$.
            Therefore, we have that $\buf_0' = \tagged{\pskip}{\notags} \concat \buf_1' $ and $\DeepProject{\buf_1} = \DeepProject{\buf_1'}$.
            Therefore, we can apply the \textsc{Retire-Skip} and \textsc{Step} rules to $C_0$ and $C_0'$ as follows:
            \begin{align*}
                \tup{m_0,a_0,\tagged{\pskip}{\notags} \concat \buf_1,\CacheState_0,\BpState_0} &\muarchStep{\retire}{} \tup{m_0, a_0, \buf_1, \CacheState_0,\BpState_0}\\
                \tup{m_0,a_0,\tagged{\pskip}{\notags} \concat \buf_1,\CacheState_0,\BpState_0, \SchedState_0} &\muarchStep{}{} \tup{m_0, a_0, \buf_1, \CacheState_0,\BpState_0, \SchedUpdate(\SchedState_0, \BufProject{\buf_1})}\\
                \tup{m_0',a_0',\tagged{\pskip}{\notags} \concat \buf_1',\CacheState_0',\BpState_0'} &\muarchStep{\retire}{} \tup{m_0', a_0', \buf_1', \CacheState_0',\BpState_0'}\\
                \tup{m_0',a_0',\tagged{\pskip}{\notags} \concat \buf_1',\CacheState_0',\BpState_0', \SchedState_0'} &\muarchStep{}{} \tup{m_0', a_0', \buf_1', \CacheState_0',\BpState_0', \SchedUpdate(\SchedState_0', \BufProject{\buf_1'})}
            \end{align*}
            We now show that $C_1 = \tup{m_0, a_0, \buf_1, \CacheState_0,\BpState_0, \SchedUpdate(\SchedState_0, \BufProject{\buf_1})}$ and $C_1' = \tup{m_0', a_0', \buf_1', \CacheState_0',\BpState_0', \SchedUpdate(\SchedState_0', \BufProject{\buf_1'})}$ are indistinguishable, i.e., $C_1 \sim C_1'$.
            For this, we need to show that:
            \begin{description}
                \item[$\apply{\buf_1}{a_0}(\pc) = \apply{\buf_1'}{a_0'}(\pc)$:]
                From (a), we have $\apply{\tagged{\pskip}{\notags} \concat \buf_1}{a_0}(\pc) = \apply{\tagged{\pskip}{\notags} \concat \buf_1'}{a_0'}(\pc)$.
                From this, we immediately get that $\apply{ \buf_1}{a_0}(\pc) = \apply{ \buf_1'}{a_0'}(\pc)$.
    
                \item[$\DeepProject{\buf_1} = \DeepProject{\buf_1'}$:] 
                This follows from $\buf_0 = \tagged{\pskip}{\notags} \concat \buf_1 $, $\buf_0' = \tagged{\pskip}{\notags} \concat \buf_1' $, and (a).
    
                \item[$\CacheState_0 = \CacheState_0'$:]
                This follows from (a).
    
                \item[$\BpState_0 = \BpState_0':$] 
                This follows from (a).
    
                \item[$\SchedUpdate(\SchedState_0, \BufProject{\buf_1}) = \SchedUpdate(\SchedState_0', \BufProject{\buf_1'})$:]
                From (a), we have   $\SchedState_0 = \SchedState_0'$.
                From $\DeepProject{\buf_1} = \DeepProject{\buf_1'}$ and Lemma~\ref{lemma:vanilla:buffer-projections}, we have $\BufProject{\buf_1} = \BufProject{\buf_1'}$.
                Therefore, $\SchedUpdate(\SchedState_0, \BufProject{\buf_1}) = \SchedUpdate(\SchedState_0', \BufProject{\buf_1'})$.
            \end{description}
            Therefore, $C_1 \sim C_1'$ and (c) holds.

            \item[$\buf_0 = \tagged{\pbarrier}{\notags} \concat \buf_1 $:]
            The proof of this case is similar to that of the case $\buf_0 = \tagged{\pskip}{\notags} \concat \buf_1 $.
    
            \item[$\buf_0 = \tagged{\passign{x}{v}}{\notags} \concat \buf_1 $:] 
            From (a), we get that $\DeepProject{\buf_0} = \DeepProject{\buf_0'}$.
            Therefore, we have that $\buf_0' = \tagged{\passign{x}{v'}}{\notags} \concat \buf_1' $ and $\DeepProject{\buf_1} = \DeepProject{\buf_1'}$.
            From $\DeepProject{\buf_0} = \DeepProject{\buf_0'}$, we also get that $v \in \Val$ iff $v' \in \Val$.
            Observe also that if $v \not\in \Val$ then both computations are stuck and (c) holds (since there is no $C_1$ such that $C_0 \SeqProcMuarchStep{}{} C_1$ and no $C_1'$ such that $C_0' \SeqProcMuarchStep{}{} C_1'$).
            In the following, therefore, we assume that $v,v' \in \Val$.
            Therefore, we can apply the \textsc{Retire-Assignment} and \textsc{Step} rules to $C_0$ and $C_0'$ as follows:
            \begin{align*}
                \tup{m_0,a_0,\tagged{\passign{x}{v}}{\notags} \concat \buf_1,\CacheState_0,\BpState_0} &\muarchStep{\retire}{} \tup{m_0, a_0[x\mapsto v], \buf_1, \CacheState_0,\BpState_0}\\
                \tup{m_0,a_0,\tagged{\passign{x}{v}}{\notags} \concat \buf_1,\CacheState_0,\BpState_0, \SchedState_0} &\muarchStep{}{} \tup{m_0, a_0[x \mapsto v], \buf_1, \CacheState_0,\BpState_0, \SchedUpdate(\SchedState_0, \BufProject{\buf_1})}\\
                \tup{m_0',a_0',\tagged{\passign{x}{v'}}{\notags} \concat \buf_1',\CacheState_0',\BpState_0'} &\muarchStep{\retire}{} \tup{m_0', a_0'[x \mapsto v'], \buf_1', \CacheState_0',\BpState_0'}\\
                \tup{m_0',a_0',\tagged{\passign{x}{v'}}{\notags} \concat \buf_1',\CacheState_0',\BpState_0', \SchedState_0'} &\muarchStep{}{} \tup{m_0', a_0'[x \mapsto v'], \buf_1', \CacheState_0',\BpState_0', \SchedUpdate(\SchedState_0', \BufProject{\buf_1'})}
            \end{align*}
            We now show that $C_1 = \tup{m_0, a_0[x \mapsto v], \buf_1, \CacheState_0,\BpState_0, \SchedUpdate(\SchedState_0, \BufProject{\buf_1})}$ and $C_1' = \tup{m_0', a_0'[x \mapsto v'], \buf_1', \CacheState_0',\BpState_0', \SchedUpdate(\SchedState_0', \BufProject{\buf_1'})}$ are indistinguishable, i.e., $C_1 \sim C_1'$.
            For this, we need to show that:
            \begin{description}
                \item[$\apply{\buf_1}{a_0[x \mapsto v]}(\pc) = \apply{\buf_1'}{a_0'[x\mapsto v']}(\pc)$:]
                From (a), we have $\apply{\tagged{\passign{x}{v}}{\notags} \concat \buf_1}{a_0}(\pc) = \apply{\tagged{\passign{x}{v'}}{\notags} \concat \buf_1'}{a_0'}(\pc)$.
                From this, we immediately get that $\apply{ \buf_1}{a_0[x \mapsto v]}(\pc) = \apply{ \buf_1'}{a_0'[x \mapsto v']}(\pc)$.
    
                \item[$\DeepProject{\buf_1} = \DeepProject{\buf_1'}$:] 
                This follows from $\buf_0 = \tagged{\passign{x}{v}}{\notags} \concat \buf_1 $, $\buf_0' = \tagged{\passign{x}{v'}}{\notags} \concat \buf_1' $, and (a).
    
                \item[$\CacheState_0 = \CacheState_0'$:]
                This follows from (a).
    
                \item[$\BpState_0 = \BpState_0':$] 
                This follows from (a).
    
                \item[$\SchedUpdate(\SchedState_0, \BufProject{\buf_1}) = \SchedUpdate(\SchedState_0', \BufProject{\buf_1'})$:]
                From (a), we have   $\SchedState_0 = \SchedState_0'$.
                From $\DeepProject{\buf_1} = \DeepProject{\buf_1'}$ and Lemma~\ref{lemma:vanilla:buffer-projections}, we have $\BufProject{\buf_1} = \BufProject{\buf_1'}$.
                Therefore, $\SchedUpdate(\SchedState_0, \BufProject{\buf_1}) = \SchedUpdate(\SchedState_0', \BufProject{\buf_1'})$.
            \end{description}
            Therefore, $C_1 \sim C_1'$ and (c) holds.
    
            \item[$\buf_0 = \tagged{\pmarkedassign{x}{v}}{\notags} \concat \buf_1 $:]
            The proof fo this case is similar to that of the case $\buf_0 = \tagged{\passign{x}{v}}{\notags} \concat \buf_1 $.
    
            \item[$\buf_0 = \tagged{\pstore{v}{n}}{\notags} \concat \buf_1 $:] 
            From (a), we get that $\DeepProject{\buf_0} = \DeepProject{\buf_0'}$.
            Therefore, we have that $\buf_0' = \tagged{\pstore{v'}{n}}{\notags} \concat \buf_1' $ and $\DeepProject{\buf_1} = \DeepProject{\buf_1'}$.
            Observe that (1) $v \in \Val \leftrightarrow v' \in \Val$ from (a), and (2) if $v \not\in \Val$ or $n \not\in \Val$, then both computations are stuck and (c) holds (since there is no $C_1$ such that $C_0 \muarchStep{}{} C_1$ and no $C_1'$ such that $C_0' \muarchStep{}{} C_1'$).
            In the following, therefore, we assume that $v,v',n \in \Val$.
            Therefore, we can apply the \textsc{Retire-Store} and \textsc{Step} rules to $C_0$ and $C_0'$ as follows:
            \begin{align*}
                \tup{m_0,a_0,\tagged{\pstore{v}{n}}{\notags} \concat \buf_1,\CacheState_0,\BpState_0} &\muarchStep{\retire}{} \tup{m_0[n\mapsto v], a_0, \buf_1, \CacheUpdate(\CacheState_0,n),\BpState_0}\\
                \tup{m_0,a_0,\tagged{\pstore{v}{n}}{\notags} \concat \buf_1,\CacheState_0,\BpState_0, \SchedState_0} &\muarchStep{}{} \tup{m_0[n\mapsto v], a_0, \buf_1, \CacheUpdate(\CacheState_0,n),\BpState_0, \SchedUpdate(\SchedState_0, \BufProject{\buf_1})}\\
                \tup{m_0',a_0',\tagged{\pstore{v'}{n}}{\notags} \concat \buf_1',\CacheState_0',\BpState_0'} &\muarchStep{\retire}{} \tup{m_0'[n \mapsto v'], a_0', \buf_1', \CacheUpdate(\CacheState_0',n),\BpState_0'}\\
                \tup{m_0',a_0',\tagged{\pstore{v'}{n}}{\notags} \concat \buf_1',\CacheState_0',\BpState_0', \SchedState_0'} &\muarchStep{}{} \tup{m_0'[n\mapsto v'], a_0', \buf_1', \CacheUpdate(\CacheState_0',n),\BpState_0', \SchedUpdate(\SchedState_0', \BufProject{\buf_1'})}
            \end{align*}
            We now show that $C_1 = \tup{m_0[n \mapsto v], a_0, \buf_1, \CacheUpdate(\CacheState_0,n),\BpState_0, \SchedUpdate(\SchedState_0, \BufProject{\buf_1})}$ and $C_1' =  \tup{m_0'[n\mapsto v'], a_0', \buf_1', \CacheUpdate(\CacheState_0',n),\BpState_0', \SchedUpdate(\SchedState_0', \BufProject{\buf_1'})}$ are indistinguishable, i.e., $C_1 \sim C_1'$.
            For this, we need to show that:
            \begin{description}
                \item[$\apply{\buf_1}{a_0}(\pc) = \apply{\buf_1'}{a_0'}(\pc)$:]
                From (a), we have $\apply{\tagged{\pstore{n}{v}}{\notags} \concat \buf_1}{a_0}(\pc) = \apply{\tagged{\pstore{n}{v'}}{\notags} \concat \buf_1'}{a_0'}(\pc)$.
                From this, we immediately get that $\apply{ \buf_1}{a_0}(\pc) = \apply{ \buf_1'}{a_0'}(\pc)$.
    
                \item[$\DeepProject{\buf_1} = \DeepProject{\buf_1'}$:] 
                This follows from $\buf_0 = \tagged{\pstore{v}{n}}{\notags} \concat \buf_1 $, $\buf_0' = \tagged{\pstore{v'}{n}}{\notags} \concat \buf_1' $, and (a).
    
                \item[$\CacheUpdate(\CacheState_0,n) = \CacheUpdate(\CacheState_0',n)$:]
                This follows from $\CacheState_0 = \CacheState_0'$, which, in turn, follows from (a).
    
                \item[$\BpState_0 = \BpState_0':$] 
                This follows from (a).
    
                \item[$\SchedUpdate(\SchedState_0, \BufProject{\buf_1}) = \SchedUpdate(\SchedState_0', \BufProject{\buf_1'})$:]
                From (a), we have   $\SchedState_0 = \SchedState_0'$.
                From $\DeepProject{\buf_1} = \DeepProject{\buf_1'}$ and Lemma~\ref{lemma:seq-processor:buffer-projections}, we have $\BufProject{\buf_1} = \BufProject{\buf_1'}$.
                Therefore, $\SchedUpdate(\SchedState_0, \BufProject{\buf_1}) = \SchedUpdate(\SchedState_0', \BufProject{\buf_1'})$.
            \end{description}
            Therefore, $C_1 \sim C_1'$ and (c) holds.
        \end{description}
        Therefore, (c) holds for all the cases.

    \end{description}
    Since (c) holds for all cases, this completes the proof of our lemma.
    \end{proof}

\subsection{Main lemma}

\begin{definition}[Indistinguishability of hardware configurations]\label{def:vanilla:indistinguishability}
	We say that two hardware configurations $\tup{\sigma,\mu} = \tup{m,a,\buf, \CacheState,\BpState, \SchedState}$ and $\tup{\sigma',\mu'} = \tup{m',a',\buf', \CacheState',\BpState', \SchedState'}$ are \emph{indistinguishable}, written $\tup{\sigma,\mu} \approx \tup{\sigma',\mu'}$, iff $\BufProject{\buf } = \BufProject{\buf'}$, $\CacheState = \CacheState'$, $\BpState = \BpState'$, and $\SchedState = \SchedState'$.
\end{definition}

\begin{lemma}[Deep-indistinguishability implies indistinguishability]\label{lemma:vanilla:deep-indistiguishability-implies-indistinguishability}
Let $C, C'$ be hardware configurations.
If $C \sim C'$, then $C \approx C'$.
\end{lemma}

\begin{proof}
It  follows from Definitions~\ref{def:vanilla:deep-indistinguishability} and~\ref{def:vanilla:indistinguishability}.
\end{proof}

\begin{lemma}\label{lemma:vanilla:main-lemma}
Let $p$ be a well-formed program, $\CacheState_0$ be the initial cache state, $\BpState_0$ be the initial branch predictor state, and $\SchedState_0$ be the initial scheduler state, $\sigma_0 = \tup{m_0,a_0}, \sigma_0' = \tup{m_0',a_0'}$ be initial \archstate{}s, and $C_0 = \tup{m_0,a_0, \emptysequence, \CacheState_0, \BpState_0, \SchedState_0}$ and $C_0' = \tup{m_0',a_0', \emptysequence, \CacheState_0, \BpState_0, \SchedState_0}$ be hardware configurations.
	Furthermore, let $\crun := s_0$ $\CtSpecInterfStep{o_1}{}$ $s_1$ $\CtSpecInterfStep{o_2}{}$ $\ldots$  $\CtSpecInterfStep{o_{n-1}}{}$  $s_n$ and $\crunp:=s_0'$ $\CtSpecInterfStep{o_1'}{}$  $s_1' $ $\CtSpecInterfStep{o_2'}{}$ $\ldots$  $\CtSpecInterfStep{o_{n-1}'}{} $ $s_n$ be two runs for the $\CtSpecInterf{\cdot}$ contract where  $s_n, s_{n}'$ are final contract states.
	If $\wInterf > \wMuarch + 1$ and $o_i = o_i'$ for all $0 < i < n$, then there is a $k \in \Nat$ and $C_0, \ldots, C_k, C_0', \ldots, C_k'$ such that $C_0 \muarchStep{}{} C_1 \muarchStep{}{} \ldots \muarchStep{}{} C_k$, $C_0' \muarchStep{}{} C_1' \muarchStep{}{} \ldots \muarchStep{}{} C_k'$, and one of the following conditions hold:
	\begin{compactenum}
		\item $C_0, C_0'$ are initial states, $\forall 0 \leq i \leq k.\ C_{i} \approx C_{i}'$, and $C_, C_k'$ are final states, or
		\item $C_0, C_0'$ are initial states, $\forall 0 \leq i \leq k.\ C_{i} \approx C_{i}'$, and there are no $C_{k+1}$ such that $C_{k+1} \neq C_k \wedge C_k \muarchStep{}{} C_{k+1}$ and no $C_{k+1}'$ such that $C_{k+1}' \neq C_k' \wedge C_k' \muarchStep{}{} C_{k+1}'$.
	\end{compactenum}
\end{lemma}

\begin{proof}
Let $p$ be a well-formed program, $\CacheState_0$ be the initial cache state, $\BpState_0$ be the initial branch predictor state, and $\SchedState_0$ be the initial scheduler state, $\sigma_0 = \tup{m_0,a_0}, \sigma_0' = \tup{m_0',a_0'}$ be initial \archstate{}s, and $C_0 = \tup{m,a, \emptysequence, \CacheState_0, \BpState_0, \SchedState_0}$ and $C_0' = \tup{m',a', \emptysequence, \CacheState_0, \BpState_0, \SchedState_0}$ be hardware configurations.
	Furthermore, let $\crun := s_0$ $\CtSpecInterfStep{o_1}{}$ $s_1$ $\CtSpecInterfStep{o_2}{}$ $\ldots$  $\CtSpecInterfStep{o_{n-1}}{}$  $s_n$ and $\crunp:=s_0'$ $\CtSpecInterfStep{o_1'}{}$  $s_1' $ $\CtSpecInterfStep{o_2'}{}$ $\ldots$  $\CtSpecInterfStep{o_{n-1}'}{} $ $s_n$ be two runs for the $\CtSpecInterf{\cdot}$ contract where $s_n, s_{n}'$ are final contract states.
	Finally, we assume that $\wInterf > \wMuarch + 1$ and  $o_i = o_i'$ for all $0 < i < n$.

	Consider the two hardware runs $\hrun, \hrunp$ obtained as follows: we start from $C_0= \tup{m_0,a_0, \emptysequence, \CacheState_0, \BpState_0, \SchedState_0}$ and $C_0'= \tup{m_0',a_0', \emptysequence, \CacheState_0, \BpState_0, \SchedState_0}$, and we apply one step of $\muarchStep{}{}$ until either $\hrun$ or $\hrunp$ reaches a final state or gets stuck.
	That is, for some $k \in \Nat$, we obtain the following runs:
	\begin{align*}
		C_0 &:=  \tup{m_0,a_0, \emptysequence, \CacheState_0, \BpState_0, \SchedState_0}\\
		C_0' &:= \tup{m_0',a_0', \emptysequence, \CacheState_0, \BpState_0, \SchedState_0}\\
		\hrun&:=C_0 \muarchStep{}{} C_1 \muarchStep{}{} \ldots  \muarchStep{}{} C_k \\
		\hrunp&:=C_0' \muarchStep{}{} C_1' \muarchStep{}{} \ldots  \muarchStep{}{} C_k' 
	\end{align*}
	Let $\map{\crun}{\hrun}{\cdot}$ and $\map{\crunp}{\hrunp}{\cdot}$ be the maps constructed according to Definition~\ref{def:vanilla:mapping}.
	We claim that that for all $0 \leq i \leq k$, $\map{\crun}{\hrun}{i} = \map{\crunp}{\hrunp}{i}$ and $ C_i \sim C_{i}'$ which implies $C_i \approx C_i'$ (see Lemma~\ref{lemma:vanilla:deep-indistiguishability-implies-indistinguishability}).
	Moreover, if $C_k$ is a stuck configuration (i.e., there is no $C'$ such that $C_k \muarchStep{}{} C'$) then so is $C_k'$ from $C_{k} \sim C_{k}'$.
	This concludes the proof of our lemma.
	
	We now prove our claim.
	That is, we show, by induction on $i$, that for all $0 \leq i \leq k$, $\map{\crun}{\hrun}{i} = \map{\crun'}{\hrun'}{i}$ and $ C_i \sim C_{i}'$.
	\begin{description}
		\item[Base case:]
		Then, $i = 0$.
		Therefore, $\map{\crun}{\hrun}{0} = \map{\crun'}{\hrun'}{0} = \{ 0 \mapsto 0\}$ by construction.
		Moreover, $C_{0} \sim C_{0}'$ immediately follows from $C_0 =  \tup{m_0,a_0, \emptysequence, \CacheState_0, \BpState_0, \SchedState_0}$ and $
		C_0'= \tup{m_0',a_0', \emptysequence, \CacheState_0, \BpState_0, \SchedState_0}$.
	
		\item[Induction step:]
		For the induction step, we assume that our claim holds for all $i' < i$, and we show that it holds for $i$ as well.
		From the induction hypothesis, we get that $\map{\crun}{\hrun}{i-1} = \map{\crunp}{\hrunp}{i-1}$ and $C_{i-1} \sim C_{i-1}'$.
		In the following, we denote $\map{\crun}{\hrun}{i-1} = \map{\crunp}{\hrunp}{i-1}$ as (IH.1) and $C_{i-1} \sim C_{i-1}'$ as (IH.2). 
		Moreover, from Lemma~\ref{lemma:vanilla:mapping-is-correct} applied to $\map{\crun}{\hrun}{i-1} = \map{\crunp}{\hrunp}{i-1}$, we have that for all $\buf \in \prefixes{\buf_{i-1}}$ and all $\buf' \in \prefixes{\buf_{i-1}'}$, $C_{i-1} \bufEquiv{i} \crun( \map{\crun}{\hrun}{i-1}(|\buf|) ) $, $C_{i-1}' \bufEquiv{i} \crunp( \map{\crunp}{\hrunp}{i-1}(|\buf'|) ) $, $\headWindow{ \map{\crun}{\hrun}{i-1}(|\buf|)} >0$, and $\headWindow{ \map{\crunp}{\hrunp}{i-1}(|\buf'|)} >0$.
		Let $\buf \in \prefixes{\buf_{i-1}}$ and all $\buf' \in \prefixes{\buf_{i-1}'}$ be arbitrary buffers such that $|\buf| = |\buf'|$.
		From (IH.1) and $|\buf| = |\buf'|$, we get that $\map{\crun}{\hrun}{i-1}(|\buf|)= \map{\crunp}{\hrunp}{i-1}(|\buf'|)$.
		Therefore,  $\crun( \map{\crun}{\hrun}{i-1}(|\buf|) )$ and $\crunp(\map{\crunp}{\hrunp}{i-1}(|\buf'|))$ are two configurations $s_j, s_j'$ for some $j = \map{\crun}{\hrun}{i-1}(|\buf|)$.
		From $\crun, \crunp$ having pairwise the same observations, we have $s_{j} \CtSpecInterfStep{o_{j+1}}{} s_{j+1}$, $s_{j}' \CtSpecInterfStep{o_{j+1}'}{} s_{j+1}'$, and $o_{j+1} = o_{j+1}'$ (note that we can always make a step in $s_{j},s_{j}'$ since if $j = k$ we can make silent steps thanks to the \textsc{Step} and \textsc{Terminate} rules).
		Since $\buf,\buf'$ have been selected arbitrarily, we know that for all $\buf \in \prefixes{\buf_{i-1}}$ and all $\buf' \in \prefixes{\buf_{i-1}'}$ such that $|\buf| = |\buf'|$, we have $s_{j} \CtSpecInterfStep{o_{j+1}}{} s_{j+1}$, $s_{j}' \CtSpecInterfStep{o_{j+1}'}{} s_{j+1}'$, $\headWindow{s_j} > 0$, $\headWindow{s_j'}>0$, and $o_{j+1} = o_{j+1}'$ where $j = \map{\crun}{\hrun}{i-1}(|\buf|)$.
		From this,  $\wInterf > \wMuarch + 1$,  and (IH.2), we can apply Lemma~\ref{lemma:seq-processor:trace-equiv-implies-stepwise-indistinguishability} to $C_{i-1}$ and $C_{i-1}'$.
		As a result, we obtain that either $C_{i} \sim C_{i}$ or that both computations are stuck.
		Moreover, $\map{\crun}{\hrun}{i} = \map{\crunp}{\hrunp}{i}$ follows from the fact that the maps at step $i$ are derived from the maps at step $i-1$, which are equivalent thanks to (IH.1), depending on the content of the projection of the buffer (which is the same thanks to $C_{i} \sim C_{i}$).
		This completes the proof of the induction step.
	\end{description}
\end{proof}

\newpage
\section{Sequential processor $\muarchStyle{seq}$ -- Proof of Theorem~\ref{theorem:hni:sequential}}\label{appendix:proofs:seq-processor}

In this section, we prove the security guarantees of the sequential processor $\muarchStyle{seq}$ given in \S\ref{sec:countermeasures:sequential}.

We remark that, for the sequential processor $\muarchStyle{seq}$, we fix the scheduler to the sequential scheduler from Appendix~\ref{appendix:seq-scheduler}.
Moreover, in the following we assume given an arbitrary cache $C$ and branch predictor $Bp$.
Hence, our proof holds for arbitrary caches $C$ and branch predictors $Bp$.

The semantics $\SeqProcMuarchSem{p}$ for a program $p$ is defined as follows:
$\SeqProcMuarchSem{p}(\tup{m,a})$ is $\tup{\emptysequence, \CacheState_0,\BpState_0, \SchedState_0}$ $ \cdot \tup{\BufProject{\buf_1}, \CacheState_1,\BpState_1, \SchedState_1}$ $  \cdot \tup{\BufProject{\buf_2}, \CacheState_2,\BpState_2, \SchedState_2}$ $  \cdot \ldots \cdot $ $ \tup{\BufProject{\buf_n}, \CacheState_n,\BpState_n, \SchedState_n}$ where $	\tup{m,a,\emptysequence, \CacheState_0,\BpState_0, \SchedState_0 } \muarchStep{}{} \tup{m_1,a_1,\buf_1, \CacheState_1,\BpState_1, \SchedState_1}  \muarchStep{}{} \tup{m_2,a_2,\buf_2, \CacheState_2,\BpState_2, \SchedState_2} \muarchStep{}{}$ $\ldots$ $\muarchStep{}{}  \tup{m_n,a_n,\buf_n, \CacheState_n,\BpState_n, \SchedState_n}$ is the complete hardware run obtained starting from $\tup{m,a}$ and terminating in $\tup{m_n,a_n,\buf_n, \CacheState_n,\BpState_n, \SchedState_n}$, which is a final hardware state. Otherwise, $\SeqProcMuarchSem{p}(\tup{m,a})$ is undefined.

\sequentialGuarantees*

\begin{proof}
	Let $p$ be an arbitrary well-formed program.
	Moreover, let $\sigma = \tup{m,a},\sigma' = \tup{m',a'}$ be two arbitrary initial configurations.
	There are two cases:
	\begin{compactitem}
	\item[$\CtSeqInterf{\Prg}(\sigma) \neq \CtSeqInterf{\Prg}(\sigma')$:] Then, 	$\CtSeqInterf{\Prg}(\sigma) = \CtSeqInterf{\Prg}(\sigma') \Rightarrow \SeqProcMuarchSem{\Prg}(\sigma) = \SeqProcMuarchSem{\Prg}(\sigma')$ trivially holds.
	\item[$\CtSeqInterf{\Prg}(\sigma) = \CtSeqInterf{\Prg}(\sigma')$:]
		By unrolling the notion of $\CtSeqInterf{\Prg}(\sigma)$ (together with all changes to the program counter $\pc$ being visible on traces), we obtain that there are runs $\crun:= \sigma$ $\CtSeqInterfStep{o_1}{}$ $\sigma_1$ $\CtSeqInterfStep{o_2}{}$ $\ldots$  $\CtSeqInterfStep{o_{n-1}}{}$  $\sigma_n$ and $\crunp:= \sigma'\CtSeqInterfStep{o_1'}{} \sigma_1' \CtSeqInterfStep{o_2'}{} \ldots  \CtSeqInterfStep{o_{n-1}'}{}  \sigma_n'$ such that $o_i = o_i'$ for all $0 < i < n$.
		By applying Lemma~\ref{lemma:seq-processor:main-lemma}, we immediately get that $\SeqProcMuarchSem{\Prg}(\sigma) = \SeqProcMuarchSem{\Prg}(\sigma')$ (because either both run terminate producing indistinguishable sequences of processor configurations or they both get stuck).
		Therefore, $\CtSeqInterf{\Prg}(\sigma) = \CtSeqInterf{\Prg}(\sigma') \Rightarrow \SeqProcMuarchSem{\Prg}(\sigma) = \SeqProcMuarchSem{\Prg}(\sigma')$ holds.
	\end{compactitem}
	Hence, $\CtSeqInterf{\Prg}(\sigma) = \CtSeqInterf{\Prg}(\sigma') \Rightarrow \SeqProcMuarchSem{\Prg}(\sigma) = \SeqProcMuarchSem{\Prg}(\sigma')$ holds for all programs $p$ and initial configurations $\sigma,\sigma'$.
	Therefore, $\hsni{\CtSeqInterf{\cdot}}{\SeqProcMuarchSem{\cdot}}$ holds.	
\end{proof}

\subsection{Mapping lemma}

	\begin{definition}[Well-formed buffers for $\muarchStyle{seq}$]
		A reorder buffer $\buf$ is \emph{well-formed for $\muarchStyle{seq}$ and an assignment  $a$}, written $\wellformed{\buf,a}$, if the following conditions hold:
		\begin{align*}
			\wellformed{\emptysequence,a} & \\
			\wellformed{ \tagged{\passign{\pc}{e}}{\notags} , a} & \text{ if } p(a(\pc)) \sim_a \tagged{\passign{\pc}{e}}{\notags} \\
			\wellformed{\tagged{\passign{\pc}{\ell}}{\ell_0} , a} & \text{ if } \ell_0 \in \Val \wedge p(\ell_0) = \pjz{x}{\ell'} \wedge \ell \in \{\ell', \ell_0+1\} \wedge p(a(\pc)) \sim_a \tagged{\passign{\pc}{\ell}}{\ell_0} \wedge \ell_0 = a(\pc)\\
			\wellformed{ \tagged{\pmarkedassign{\pc}{\ell}}{\notags} , a } & \text{ if } \ell \in \Val \\
			\wellformed{ \tagged{i}{\notags} \concat \tagged{\pmarkedassign{\pc}{\ell}}{\notags}  } & \text{ if }  \ell \in \Val \wedge (\forall e.\ i \neq \passign{\pc}{e}) \wedge (\forall x,e.\ i \neq \pmarkedassign{x}{e}) \wedge (\forall x,e.\ i \neq \pload{\pc}{e}) \wedge p(a(\pc)) \sim_{a} \tagged{i}{\notags} \wedge \ell = a(\pc)+1
		\end{align*}
		where the instruction-compatibility relation $\sim_{a}$ is defined as follows:
		\begin{align*}
			\pskip &\sim_{a} \tagged{\pskip}{\notags} \\ \allowdisplaybreaks
			\pbarrier &\sim_{a} \tagged{\pbarrier}{\notags} \\ \allowdisplaybreaks
			\passign{x}{e} &\sim_{a} \tagged{\passign{x}{e}}{\notags} \\ \allowdisplaybreaks
			\passign{x}{e} &\sim_{a} \tagged{\passign{x}{v}}{\notags} \text{ if } v\in \Val \\ \allowdisplaybreaks
			\pload{x}{e} &\sim_{a} \tagged{\pload{x}{e}}{\notags} \\ \allowdisplaybreaks
			\pload{x}{e} &\sim_{a} \tagged{\passign{x}{v}}{\notags} \text{ if } v\in \Val \\ \allowdisplaybreaks
			\pstore{x}{e} &\sim_{a} \tagged{\pstore{x}{e}}{\notags} \\ \allowdisplaybreaks
			\pstore{x}{e} &\sim_{a} \tagged{\pstore{v}{n}}{\notags} \text{ if } v,n\in \Val \\ \allowdisplaybreaks
			\pjmp{e} &\sim_{a} \tagged{\passign{\pc}{e}}{\notags} \\ \allowdisplaybreaks
			\pjmp{e} &\sim_{a} \tagged{\passign{\pc}{v}}{\notags} \text{ if } v \in \Val \\ \allowdisplaybreaks
			\pjz{x}{\ell} &\sim_{a} \tagged{\passign{\pc}{\ell}}{a(\pc)} \\ \allowdisplaybreaks
			\pjz{x}{\ell} &\sim_{a} \tagged{\passign{\pc}{a(\pc)+1}}{a(\pc)} \\ \allowdisplaybreaks
			\pjz{x}{\ell} &\sim_{a} \tagged{\passign{\pc}{a(\pc)+1}}{\notags} \\ \allowdisplaybreaks
			\pjz{x}{\ell} &\sim_{a} \tagged{\passign{\pc}{\ell}}{\notags} 
		\end{align*}
		\end{definition}

Lemma~\ref{lemma:seq-processor:buffers-well-formedness} states that all reorder buffers occurring in hardware runs are well-formed.

\begin{lemma}[Reorder buffers are well-formed]\label{lemma:seq-processor:buffers-well-formedness}
	Let $p$ be a well-formed program, $\sigma_0 = \tup{m,a}$ be an initial \archstate, $\CacheState_0$ be the initial cache state, $\BpState_0$ be the initial branch predictor state, and $\SchedState_0$ be the initial scheduler state.
	For all hardware runs  $\hrun := C_0 \SeqProcMuarchStep{}{} C_1 \SeqProcMuarchStep{}{} C_2 \SeqProcMuarchStep{}{} \ldots \SeqProcMuarchStep{}{} C_k$ and all $0 \leq i \leq k$, then $\wellformed{\buf_i,a_i}$, where $C_0 = \tup{m,a,\emptysequence, \CacheState_0, \BpState_0, \SchedState_0}$ and $C_i = \tup{m_i, a_i,\buf_i, \CacheState_i, \BpState_i, \SchedState_i}$.
\end{lemma}

\begin{proof}
 The lemma follows by (1) induction on $i$, and (2) inspection of the rules defining $\SeqProcMuarchStep{}{}$.
\end{proof}
	
	\begin{definition}[Prefixes of buffers]\label{def:seq-processor:prefixes}
		The prefixes of a well-formed buffer $\buf$ are defined as follows:
		\begin{align*}
			\prefixes{\emptysequence} &= \{ \emptysequence \}\\
			\prefixes{  \tagged{\passign{\pc}{e}}{T} \concat \buf } &= 
			\{ \emptysequence \} \cup \{  \tagged{\passign{\pc}{e}}{T} \concat \buf' \mid \buf' \in \prefixes{\buf} \} \\
			\prefixes{  \tagged{i}{\notags} \concat \tagged{\pmarkedassign{\pc}{\ell}}{\notags} \concat \buf } &= \{ \emptysequence \} \cup \{   \tagged{i}{\notags} \concat \tagged{\pmarkedassign{\pc}{\ell}}{\notags} \concat \buf' \mid \buf' \in \prefixes{\buf} \} \\
			\prefixes{   \tagged{\pmarkedassign{\pc}{\ell}}{\notags} \concat \buf } &= \{  \tagged{\pmarkedassign{\pc}{\ell}}{\notags} \} \cup \{ \tagged{\pmarkedassign{\pc}{\ell}}{\notags} \concat \buf' \mid \buf' \in \prefixes{\buf} \} \\
		\end{align*}
	\end{definition}
	
	\begin{definition}[Deep-update for \muarchStyle{seq}]
			Let $p$ be a program,  $\tup{m,a}$ be an \archstate{}, and $\buf$ be a buffer.
			The \emph{deep-update of $\tup{m,a}$ given $\buf$} is defined as follows:
				\begin{align*}
				\update{\tup{m,a}}{\emptysequence} &:= \tup{m,a} \\
				\update{\tup{m,a}}{ \tagged{\passign{x}{e}}{T}}  &:= 
					\begin{cases}
						\tup{m, a[x \mapsto \exprEval{e}{a}]} & \text{if } x \neq \pc \\
						\tup{m, a[x \mapsto \exprEval{e}{a}]} & \text{if } x = \pc  \wedge T = \notags \\
						\tup{m, a[x \mapsto \ell']} & \text{if } x = \pc  \wedge T = \ell \wedge p(\ell) = \pjz{y}{\ell'} \wedge a(y) = 0 \\ 
						\tup{m, a[x \mapsto \ell+1]} & \text{if } x = \pc  \wedge T = \ell \wedge p(\ell) = \pjz{y}{\ell'} \wedge a(y) \neq 0 
					\end{cases}
				\\
				\update{\tup{m,a}}{ \tagged{\pmarkedassign{x}{e}}{T}}  &:= \tup{m, a[x \mapsto \exprEval{e}{a}]}\\
				\update{\tup{m,a}}{ \tagged{\pload{x}{e}}{T}} &:= \tup{m,a[x \mapsto m(\exprEval{e}{a})] } \\
				\update{\tup{m,a}}{ \tagged{\pstore{x}{e}}{T}} &:= \tup{m[\exprEval{e}{a} \mapsto a(x)],a}\\
				\update{\tup{m,a}}{\tagged{\pskip{}}{T}} &:= \tup{m,a}\\
				\update{\tup{m,a}}{\tagged{\pbarrier{}}{T}} &:= \tup{m,a}\\
				\update{\tup{m,a}}{(\tagged{i}{T} \concat \buf)} &:= 				\update{ (\update{\tup{m,a}}{\tagged{i}{T}}) }{ \buf }
				\end{align*}
	\end{definition}
	
	\begin{definition}[$\crun-\hrun$ mapping]\label{def:seq:mapping}
	Let $p$ be a well-formed program, $\sigma_0 = \tup{m,a}$ be an initial \archstate, $\CacheState_0$ be the initial cache state, $\BpState_0$ be the initial branch predictor state, and $\SchedState_0$ be the initial scheduler state.
	Furthermore, let:
	\begin{compactitem}
		\item $\crun := \sigma_0 \CtSeqInterfStep{o_1}{} \sigma_1 \CtSeqInterfStep{o_2}{} \sigma_2 \CtSeqInterfStep{o_3}{} \ldots \CtSeqInterfStep{o_{n-1}}{} \sigma_n$ be the longest $\interfStyle{seq-ct}$ contract run obtained starting from $\sigma_0$.
		\item $\hrun := C_0 \SeqProcMuarchStep{}{} C_1 \SeqProcMuarchStep{}{} C_2 \SeqProcMuarchStep{}{} \ldots \SeqProcMuarchStep{}{} C_k$ be the longest $\muarchStyle{seq}$ hardware run obtained starting from $C_0 = \tup{m,a,\emptysequence, \CacheState_0, \BpState_0, \SchedState_0}$.
		\item $\hrun(i)$ be the $i$-th hardware configuration in $\hrun$.
		\item $\crun(i)$ be the $i$-th contract configuration in $\crun$  (note that $\crun(i) = \sigma_n$ for all $i > n$).
	\end{compactitem}
	The \emph{$\crun-\hrun$ mapping}, which maps hardware configurations in $\hrun$ to contract configurations in $\crun$, is defined as follows:
	\begin{align*}
		\map{\crun}{\hrun}{0} &:= \{0 \mapsto 0\} \\ 
		\map{\crun}{\hrun}{i} &:= {
			\begin{cases}
			\map{\crun}{\hrun}{i-1} & \text{if } \SchedNext(\hrun(i-1)) = \fetch{} \wedge ln(\hrun(i-1)) = ln(\hrun(i))\\
			fetch_{\crun,\hrun}(i) & \text{if } \SchedNext(\hrun(i-1)) = \fetch{} \wedge ln(\hrun(i-1)) < ln(\hrun(i))\\
			\map{\crun}{\hrun}{i-1} & \text{if } \SchedNext(\hrun(i-1)) = \execute{j} \\ 
			shift(\map{\crun}{\hrun}{i -1}) & \text{if } \SchedNext(\hrun(i-1)) = \retire{} \\
			\end{cases}
		}\\ 
		fetch_{\crun,\hrun}(i) &=
				\map{\crun}{\hrun}{i-1}[ln(\hrun(i-1))+2 \mapsto \map{\crun}{\hrun}{i-1}(ln(\hrun(i-1)))+1]\\
				& \qquad \text{if }
							p(\mathit{lstPc}(\hrun(i -1))) \neq \pjz{x}{\lbl} \wedge 
							p(\mathit{lstPc}(\hrun(i -1))) \neq \pjmp{e} 
							\\
		fetch_{\crun,\hrun}(i) &=
				\map{\crun}{\hrun}{i-1}[ln(\hrun({i-1}))+1 \mapsto \map{\crun}{\hrun}{i-1}(ln(\hrun(i-1)))+1]\\
				& \qquad \text{if }
							p(\mathit{lstPc}(\hrun(i -1))) = \pjmp{e} \vee p(\mathit{lstPc}(\hrun(i -1))) = \pjz{x}{\lbl} 
							\\
		ln(\tup{m,a,\buf,\CacheState,\BpState,\SchedState}) &= |\buf|\\
		\SchedNext(\tup{m,a,\buf,\CacheState,\BpState,\SchedState}) &= \SchedNext(\SchedState)\\
		shift(map) &= \lambda i \in \Nat.\ map(i +1 )\\
		\mathit{lstPc}(\tup{m,a,\buf,\CacheState,\BpState,\SchedState}) &= (\update{\tup{m,a}}{\buf})(\pc)\\
	\end{align*}
	\end{definition}

We are now ready to prove Lemma~\ref{lemma:seq-processor:mapping-is-correct}, the main lemma showing the correctness of the $\crun-\hrun$ mapping.

\begin{lemma}[Correctness of $\crun-\hrun$ mapping]\label{lemma:seq-processor:mapping-is-correct}
	Let $p$ be a well-formed program, $\sigma_0 = \tup{m,a}$ be an initial \archstate, $\CacheState_0$ be the initial cache state, $\BpState_0$ be the initial branch predictor state, and $\SchedState_0$ be the initial scheduler state.
	Furthermore, let:
	\begin{compactitem}
		\item $\crun := \sigma_0 \CtSeqInterfStep{o_1}{} \sigma_1 \CtSeqInterfStep{o_2}{} \sigma_2 \CtSeqInterfStep{o_3}{} \ldots \CtSeqInterfStep{o_{n-1}}{} \sigma_n$ be the longest $\interfStyle{seq-ct}$ contract run obtained starting from $\sigma_0$.
		\item $\hrun := C_0 \SeqProcMuarchStep{}{} C_1 \SeqProcMuarchStep{}{} C_2 \SeqProcMuarchStep{}{} \ldots \SeqProcMuarchStep{}{} C_k$ be the longest $\muarchStyle{seq}$ hardware run obtained starting from $C_0 = \tup{m,a,\emptysequence, \CacheState_0, \BpState_0, \SchedState_0}$.
		\item $\hrun(i)$ be the $i$-th hardware configuration in $\hrun$.
		\item $\crun(i)$ be the $i$-th contract configuration in $\crun$  (note that $\crun(i) = \sigma_n$ for all $i > n$).
		\item $ \map{\crun}{\hrun}{\cdot}$ be the mapping from Definition~\ref{def:seq:mapping}.
	\end{compactitem}
	The following conditions hold:
	\begin{compactenum}[(1)]
	\item $C_0$ is an initial hardware configuration.
	\item $C_k$ is a final hardware configuration or there is no $C_{k'}$ such that $C_k \SeqProcMuarchStep{}{} C_{k'}$.
	\item for all $0 \leq i \leq k$, given $C_i = \tup{m_i,a_i, \buf_i, \CacheState_i, \BpState_i, \SchedState_i}$ the following conditions hold:
		\begin{compactenum}[(a)]
			\item for all $\buf \in \prefixes{\buf_i}$,  $\update{\tup{m_i,a_i}}{\buf} = \crun( \map{\crun}{\hrun}{i}(|\buf|) )$.
		\end{compactenum}
	\end{compactenum}
\end{lemma}

\begin{proof}
Let $p$ be a well-formed program, $\sigma_0 = \tup{m,a}$ be an initial \archstate, $\CacheState_0$ be the initial cache state, $\BpState_0$ be the initial branch predictor state, and $\SchedState_0$ be the initial scheduler state.
Furthermore, let:
\begin{compactitem}
	\item $\crun := \sigma_0 \CtSeqInterfStep{o_1}{} \sigma_1 \CtSeqInterfStep{o_2}{} \sigma_2 \CtSeqInterfStep{o_3}{} \ldots \CtSeqInterfStep{o_{n-1}}{} \sigma_n$ be the longest $\interfStyle{seq-ct}$ contract run obtained starting from $\sigma_0$.
	\item $\hrun := C_0 \SeqProcMuarchStep{}{} C_1 \SeqProcMuarchStep{}{} C_2 \SeqProcMuarchStep{}{} \ldots \SeqProcMuarchStep{}{} C_k$ be the longest $\muarchStyle{seq}$ hardware run obtained starting from $C_0 = \tup{m,a,\emptysequence, \CacheState_0, \BpState_0, \SchedState_0}$.
	\item $\hrun(i)$ be the $i$-th hardware configuration in $\hrun$.
	\item $\crun(i)$ be the $i$-th contract configuration in $\crun$  (note that $\crun(i) = \sigma_n$ for all $i > n$).
	\item $\map{\crun}{\hrun}{\cdot}$ be the mapping from Definition~\ref{def:seq:mapping}.
\end{compactitem}
Observe that $C_0$ is an initial hardware configuration by construction.
Observe also that $C_k$ is either a final hardware configuration (for which the semantics cannot further proceed) or the computation is stuck (since $\hrun$ is the longest run by construction).
Therefore, (1) and (2) hold.

We now prove, by induction on $i$, that (3) holds, i.e., that for all $0 \leq i \leq k$ and for all $\buf \in \prefixes{\buf_i}$,  $\update{\tup{m_i,a_i}}{\buf} = \crun( \map{\crun}{\hrun}{i}(|\buf|) )$.
\begin{description}
	\item[Base case:]
	For the base case, we consider $i = 0$.
	Observe that $C_0 = \tup{m,a,\emptysequence, \CacheState_0, \BpState_0, \SchedState_0}$.
	Therefore, $\buf_0 = \emptysequence$ and $\prefixes{\buf_0} = \{ \emptysequence \}$ (see Definition~\ref{def:seq-processor:prefixes}).
	Then, (3.a) immediately follows from $\update{m,a}{\emptysequence} = \tup{m,a} = \sigma_0 = \crun(0)$.

	\item[Induction step:] 
	For the induction step, we assume that the claim holds for all $i' < i$ and we show that it holds for $i$ as well.
	In the following, let $C_i = \tup{m_i,a_i, \buf_i, \CacheState_i, \BpState_i, \SchedState_i}$ and we refer to the induction hypothesis as H.
	Similarly, we write H.3.a to denote that fact 3.a holds for the induction hypothesis.
	
	We proceed by case distinction on the directive $\SchedNext(C_{i-1})$ used to derive $C_i$.
	There are three cases:
	\begin{description}
		\item[$\SchedNext(C_{i-1}) = \fetch{}$:]
		We now proceed by case distinction on the applied rule:
		\begin{description}
			\item[Rule \textsc{Fetch-Branch-Hit}:]
			Then, $\buf_i = \tagged{\passign{\pc}{\lbl'}}{a_{i-1}(\pc)}$, $\buf_{i-1} = \emptysequence$, $a_{i} = a_{i-1}$, $m_{i} = m_{i-1}$, $p(a_{i-1}(\pc)) = \pjz{x}{\lbl}$, and $\map{\crun}{\hrun}{i} = \map{\crun}{\hrun}{i-1}[\mathit{ln}(|\buf_{i-1}|) +1 \mapsto \map{\crun}{\hrun}{i-1}(\mathit{ln}(|\buf_{i-1}|)) + 1]$ (observe that $\Mispred{\hrun(i-1)} = \top$ since $\buf_{i-1} = \emptysequence$).
			Therefore, we have:
			\begin{align*}
				\prefixes{\buf_i} &= \{ \emptysequence, \tagged{\passign{\pc}{\lbl'}}{a_{i-1}(\pc)} \} \\
				\prefixes{\buf_{i-1}} &= \{ \emptysequence\}
			\end{align*}
			Let $\buf$ be an arbitrary prefix in $\prefixes{\buf_i}$.
			There are two cases:
			\begin{description}
				\item[$\buf = \emptysequence$:]
				Then, $\buf \in \prefixes{\buf_{i-1}}$.
				From the induction hypothesis, we have that $\update{\tup{m_{i-1},a_{i-1}}}{\buf} = \crun( \map{\crun}{\hrun}{i-1}(|\buf|)  )$.
				From $a_{i} = a_{i-1}$ and $m_{i} = m_{i-1}$, we get that $\update{\tup{m_{i},a_{i}}}{\buf} = \crun( \map{\crun}{\hrun}{i-1}(|\buf|)  )$.
				From $|\buf| = 0$ and $\map{\crun}{\hrun}{i} = \map{\crun}{\hrun}{i-1}[\mathit{ln}(|\buf_{i-1}|) +1 \mapsto \map{\crun}{\hrun}{i-1}(\mathit{ln}(|\buf_{i-1}|)) + 1]$, we get that $ \map{\crun}{\hrun}{i}(|\buf|) = \map{\crun}{\hrun}{i-1}(|\buf|)$.
				Therefore, we have that $\update{\tup{m_{i},a_{i}}}{\buf} = \crun( \map{\crun}{\hrun}{i}(|\buf|)  )$.
				Then, (3.a) holds.

				\item[$\buf =  \tagged{\passign{\pc}{\lbl'}}{a_{i-1}(\pc)}$:]
				From the induction hypothesis, we have that $\update{\tup{m_{i-1},a_{i-1}}}{\buf_{i-1}} = \crun( \map{\crun}{\hrun}{i-1}(|\buf_{i-1}|)  )$, or, equivalently $\update{\tup{m_{i-1},a_{i-1}}}{\emptysequence} = \crun( \map{\crun}{\hrun}{i-1}(0)  )$.
				From this, we have that $a_{i-1}(\pc) = \crun( \map{\crun}{\hrun}{i-1}(0)  )(\pc)$.
				From this and $p(a_{i-1}(\pc)) = \pjz{x}{\lbl}$, we have that $\crun( \map{\crun}{\hrun}{i-1}(0) + 1  )$ has been obtained by applying the \textsc{Beqz} rule to $\crun( \map{\crun}{\hrun}{i-1}(0)  )$.
				There are two cases:
				\begin{description}
					\item[$a_{i-1}(x) = 0$:]
					Then,  $\crun( \map{\crun}{\hrun}{i-1}(0) + 1  ) = \crun( \map{\crun}{\hrun}{i-1}(0)  )[\pc \mapsto \ell]$.
					Moreover, $\update{\tup{m_{i-1},a_{i-1}}}{ \tagged{\passign{\pc}{\lbl'}}{a_{i-1}(\pc)} } = \tup{m_{i-1},a_{i-1}[\pc \mapsto \ell]}$.
					Therefore, we have $\update{\tup{m_{i-1},a_{i-1}}}{ \tagged{\passign{\pc}{\lbl'}}{a_{i-1}(\pc)} } = \crun( \map{\crun}{\hrun}{i-1}(0) + 1  )$.
					From $\buf =  \tagged{\passign{\pc}{\lbl'}}{a_{i-1}(\pc)}$, $|\buf| = 1$, and $\map{\crun}{\hrun}{i} = \map{\crun}{\hrun}{i-1}[\mathit{ln}(|\buf_{i-1}|) +1 \mapsto \map{\crun}{\hrun}{i-1}(\mathit{ln}(|\buf_{i-1}|)) + 1]$, we immediately get that $\update{\tup{m_{i},a_{i}}}{\buf} = \crun( \map{\crun}{\hrun}{i}(|\buf|)  )$.
					Then, (3.a) holds.

					\item[$a_{i-1}(x) \neq 0$:] 
					The proof of this case is similar to that of $a_{i-1}(x) = 0$.
				\end{description}
			\end{description}
			Since $\buf$ has been selected arbitrarily, (3.a) holds.
			
			\item[Rule \textsc{Fetch-Jump-Hit}:]
			The proof of this case is similar to that of the \textsc{Fetch-Branch-Hit} rule.

			\item[Rule \textsc{Fetch-Others-Hit}:]
			Then, $\buf_i = \tagged{p(a_{i-1}(\pc))}{\notags} \concat \tagged{\pmarkedassign{\pc}{a_{i-1}(\pc)+1}}{\notags}$, $\buf_{i-1} = \emptysequence$, $a_{i} = a_{i-1}$, $m_{i} = m_{i-1}$, $p(a_{i-1}(\pc)) \neq \pjz{x}{\lbl}$, $p(a_{i-1}(\pc)) \neq \pjmp{e}$, and $\map{\crun}{\hrun}{i} = \map{\crun}{\hrun}{i-1}[\mathit{ln}(|\buf_{i-1}|) +2 \mapsto \map{\crun}{\hrun}{i-1}(\mathit{ln}(|\buf_{i-1}|)) + 1]$ (observe that $\Mispred{\hrun(i-1)} = \top$ since $\buf_{i-1} = \emptysequence$).
			Therefore, we have:
			\begin{align*}
				\prefixes{\buf_i} &= \{ \emptysequence,  \tagged{p(a_{i-1}(\pc))}{\notags} \concat \tagged{\pmarkedassign{\pc}{a_{i-1}(\pc)+1}}{\notags} \} \\
				\prefixes{\buf_{i-1}} &= \{ \emptysequence\}
			\end{align*}
			Let $\buf$ be an arbitrary prefix in $\prefixes{\buf_i}$.
			There are two cases:
			\begin{description}
				\item[$\buf = \emptysequence$:]
				Then, $\buf \in \prefixes{\buf_{i-1}}$.
				From the induction hypothesis, we have that $\update{\tup{m_{i-1},a_{i-1}}}{\buf} = \crun( \map{\crun}{\hrun}{i-1}(|\buf|)  )$.
				From $a_{i} = a_{i-1}$ and $m_{i} = m_{i-1}$, we get that $\update{\tup{m_{i},a_{i}}}{\buf} = \crun( \map{\crun}{\hrun}{i-1}(|\buf|)  )$.
				From $|\buf| = 0$ and $\map{\crun}{\hrun}{i} = \map{\crun}{\hrun}{i-1}[\mathit{ln}(|\buf_{i-1}|) +1 \mapsto \map{\crun}{\hrun}{i-1}(\mathit{ln}(|\buf_{i-1}|)) + 1]$, we get that $ \map{\crun}{\hrun}{i}(|\buf|) = \map{\crun}{\hrun}{i-1}(|\buf|)$.
				Therefore, we have that $\update{\tup{m_{i},a_{i}}}{\buf} = \crun( \map{\crun}{\hrun}{i}(|\buf|)  )$.
				Then, (3.a) holds.

				\item[$\buf =  \tagged{p(a_{i-1}(\pc))}{\notags} \concat \tagged{\pmarkedassign{\pc}{a_{i-1}(\pc)+1}}{\notags}$:]
				From the induction hypothesis, we have that $\update{\tup{m_{i-1},a_{i-1}}}{\buf_{i-1}} = \crun( \map{\crun}{\hrun}{i-1}(|\buf_{i-1}|)  )$, or, equivalently $\update{\tup{m_{i-1},a_{i-1}}}{\emptysequence} = \crun( \map{\crun}{\hrun}{i-1}(0)  )$.
				From this, we have that $a_{i-1}(\pc) = \crun( \map{\crun}{\hrun}{i-1}(0)  )(\pc)$.
				From this, $p(a_{i-1}(\pc)) \neq \pjz{x}{\lbl}$, and $p(a_{i-1}(\pc)) \neq \pjmp{e}$, we have that $\crun( \map{\crun}{\hrun}{i-1}(0) + 1  )$ has been obtained by applying one of the rules \textsc{Load}, \textsc{Store}, or \textsc{Others}  to $\crun( \map{\crun}{\hrun}{i-1}(0)  )$.
				Here we consider the \textsc{Load} rule, i.e., $p(a_{i-1}(\pc)) = \pload{x}{e}$. 
				The proof for the \textsc{Store} and \textsc{Others} rules is similar and we omit it.
				Then, $\crun( \map{\crun}{\hrun}{i-1}(0) + 1  ) = \crun( \map{\crun}{\hrun}{i-1}(0)  )[x \mapsto \crun( \map{\crun}{\hrun}{i-1}(0)  )( \exprEval{e}{\crun( \map{\crun}{\hrun}{i-1}(0)  )}  ), \pc \mapsto \crun( \map{\crun}{\hrun}{i-1}(0)  )(\pc) +1]$.
				Given that $\crun( \map{\crun}{\hrun}{i-1}(0)  ) = \update{\tup{m_{i-1},a_{i-1}}}{\emptysequence} = \tup{m_{i-1},a_{i-1}}$, we have that $\crun( \map{\crun}{\hrun}{i-1}(0) + 1  ) = \tup{m_{i-1}, a_{i-1}[x \mapsto m_{i-1}( \exprEval{e}{a_{i-1}}  ), \pc \mapsto a_{i-1}(\pc) +1]}$.
				Observe also that $\update{\tup{m_{i-1},a_{i-1}}}{  \tagged{p(a_{i-1}(\pc))}{\notags} \concat \tagged{\pmarkedassign{\pc}{a_{i-1}(\pc)+1}}{\notags}} = \tup{m_{i-1}, a_{i-1}[x \mapsto m_{i-1}( \exprEval{e}{a_{i-1}}  ), \pc \mapsto a_{i-1}(\pc) +1]}$ since $p(a_{i-1}(\pc)) = \pload{x}{e}$. 
				Therefore, we have $\update{\tup{m_{i-1},a_{i-1}}}{\buf} = \crun( \map{\crun}{\hrun}{i-1}(0) + 1  )$.
				From $a_{i} = a_{i-1}$ and $m_{i} = m_{i-1}$, we get $\update{\tup{m_{i},a_{i}}}{\buf} = \crun( \map{\crun}{\hrun}{i-1}(0) + 1  )$.
				From $|\buf_i| = |\buf_{i-1}| + 2$, $\map{\crun}{\hrun}{i} = \map{\crun}{\hrun}{i-1}[\mathit{ln}(|\buf_{i-1}|) +2 \mapsto \map{\crun}{\hrun}{i-1}(\mathit{ln}(|\buf_{i-1}|)) + 1]$, we finally get $\update{\tup{m_{i},a_{i}}}{\buf} = \crun( \map{\crun}{\hrun}{i}(|\buf|)  )$.
				Therefore, (3.a) holds.

			\end{description}
			Since $\buf$ has been selected arbitrarily, (3.a) holds.

			\item[Rule \textsc{Fetch-Miss}:]
			Then, $\buf_{i} = \buf_{i-1}$, $m_{i} = m_{a-1}$, and $a_{i} = a_{i-1}$.
			Moreover, $\map{\crun}{\hrun}{i} = \map{\crun}{\hrun}{i-1}$.
			Let $\buf$ be an arbitrary prefix in $\prefixes{\buf_{i}}$.
			Then, (3.a) immediately follows from $\prefixes{\buf_i} = \prefixes{\buf_{i-1}}$, $\buf \in \prefixes{\buf_{i-1}}$, $m_{i-1} = m_i$, $a_{i-1} = a_i$, $\map{\crun}{\hrun}{i} = \map{\crun}{\hrun}{i-1}$, and (H.3.a).
			Since $\buf$ is an arbitrary prefix in $\prefixes{\buf_{i}}$, (3.a) holds.
		\end{description}
		This completes the proof for the $\fetch{}$ case.

		\item[$\SchedNext(C_{i-1}) = \execute{j}$:]
		Then, $\map{\crun}{\hrun}{i} = \map{\crun}{\hrun}{i-1}$.
		Observe also that $j = 1$ since $\muarchStyle{seq}$ uses the sequential scheduler of Appendix~\ref{appendix:seq-scheduler}.
		We now proceed by case distinction on the applied rule:
		\begin{description}
			\item[Rule \textsc{Execute-Load-Hit}:]
			Therefore, given the well-formedness of buffers (see Lemma~\ref{lemma:seq-processor:buffers-well-formedness}), we have that $\buf_{i-1} = \tagged{\pload{x}{e}}{\notags} \concat \tagged{\pmarkedassign{\pc}{\ell}}{\notags}$, $\buf_{i} =  \tagged{ \passign{x}{ m_{i-1}(\exprEval{e}{ \apply{\notags}{a_{i-1}} })}  }{\notags} \concat \tagged{\pmarkedassign{\pc}{\ell}}{\notags}$, $m_{i} = m_{i-1}$, and $a_{i} = a_{i-1}$.
			Therefore, we have that:
			\begin{align*}
				\prefixes{\buf_{i-1}} &= \{ \emptysequence, \tagged{\pload{x}{e}}{\notags} \concat \tagged{\pmarkedassign{\pc}{\ell}}{\notags}\}\\
				\prefixes{\buf_{i}} &= \{ \emptysequence, \tagged{ \passign{x}{ m_{i-1}(\exprEval{e}{ \apply{\notags}{a_{i-1}} })}  }{\notags} \concat \tagged{\pmarkedassign{\pc}{\ell}}{\notags} \}
			\end{align*}
			Let $\buf$ be an arbitrary prefix in $\prefixes{\buf_{i}}$.
			There are two cases:
			\begin{description}
				\item[$\buf = \emptysequence$:]
				Then (3.a) immediately follows from $\buf \in \prefixes{\buf_{i-1}}$, $m_{i-1} = m_i$, $a_{i-1} = a_i$, $\map{\crun}{\hrun}{i} = \map{\crun}{\hrun}{i-1}$, and (H.3.a).

				\item[$\buf = \tagged{ \passign{x}{ m_{i-1}(\exprEval{e}{ \apply{\notags}{a_{i-1}} })}  }{\notags} \concat \tagged{\pmarkedassign{\pc}{\ell}}{\notags}$:]
				From the induction hypothesis (H.3.a), we have $\update{\tup{m_{i-1}, a_{i-1}}}{  (\tagged{\pload{x}{e}}{\notags} \concat \tagged{\pmarkedassign{\pc}{\ell}}{\notags}) } = \crun( \map{\crun}{\hrun}{i-1}(|\tagged{\pload{x}{e}}{\notags} \concat \tagged{\pmarkedassign{\pc}{\ell}}{\notags}|) )$.
				By applying $\update{}{}$, we obtain that $
				\tup{m_{i-1}, a_{i-1}[x \mapsto m_{i-1}(\exprEval{e}{a_{i-1}}), \pc \mapsto \ell ] } = \crun( \map{\crun}{\hrun}{i-1}(|\tagged{\pload{x}{e}}{\notags} \concat \tagged{\pmarkedassign{\pc}{\ell}}{\notags}|) )$.
				Since $\buf = \tagged{ \passign{x}{ m_{i-1}(\exprEval{e}{ \apply{\notags}{a_{i-1}} })}  }{\notags} \concat \tagged{\pmarkedassign{\pc}{\ell}}{\notags}$ and $\update{\tup{m_{i-1}, a_{i-1}}}{ (\tagged{ \passign{x}{ m_{i-1}(\exprEval{e}{ \apply{\notags}{a_{i-1}} })}  }{\notags} \concat \tagged{\pmarkedassign{\pc}{\ell}}{\notags}) }  = \tup{m_{i-1}, a_{i-1}[x \mapsto m_{i-1}(\exprEval{e}{a_{i-1}}), \pc \mapsto \ell ] }$, we obtain that $\update{\tup{m_{i-1}, a_{i-1}}}{ \buf } = \crun( \map{\crun}{\hrun}{i-1}(|\tagged{\pload{x}{e}}{\notags} \concat \tagged{\pmarkedassign{\pc}{\ell}}{\notags}|) )$. 
				From this and $|\buf| = |\tagged{\pload{x}{e}}{\notags} \concat \tagged{\pmarkedassign{\pc}{\ell}}{\notags}|$, we have $\update{\tup{m_{i-1}, a_{i-1}}}{ \buf } = \crun( \map{\crun}{\hrun}{i-1}(|\buf|) )$.
				From $m_{i-1} = m_i$, $a_{i-1} = a_i$, and $\map{\crun}{\hrun}{i} = \map{\crun}{\hrun}{i-1}$, we finally get $\update{\tup{m_{i}, a_{i}}}{ \buf } = \crun( \map{\crun}{\hrun}{i}(|\buf|) )$.
				Therefore, (3.a) holds.
			\end{description}
			Since $\buf$ is an arbitrary prefix in $\prefixes{\buf_{i}}$, (3.a) holds.

			\item[Rule \textsc{Execute-Load-Miss}:]
			Therefore, given the well-formedness of buffers  (see Lemma~\ref{lemma:seq-processor:buffers-well-formedness}), $\buf_{i-1} = \tagged{\pload{x}{e}}{\notags} \concat \tagged{\pmarkedassign{\pc}{\ell}}{\notags}$, $\buf_{i} =  \tagged{\pload{x}{e}}{\notags} \concat \tagged{\pmarkedassign{\pc}{\ell}}{\notags}$, $m_{i} = m_{i-1}$, and $a_{i} = a_{i-1}$.
			Therefore, we have that:
			\begin{align*}
				\prefixes{\buf_{i-1}} &= \{ \emptysequence, \tagged{\pload{x}{e}}{\notags} \concat \tagged{\pmarkedassign{\pc}{\ell}}{\notags}\}\\
				\prefixes{\buf_{i}} &= \{ \emptysequence, \tagged{\pload{x}{e}}{\notags} \concat \tagged{\pmarkedassign{\pc}{\ell}}{\notags}\}
			\end{align*}
			Let $\buf$ be an arbitrary prefix in $\prefixes{\buf_{i}}$.
			Then, (3.a) immediately follows from $\prefixes{\buf_i} = \prefixes{\buf_{i-1}}$, $\buf \in \prefixes{\buf_{i-1}}$, $m_{i-1} = m_i$, $a_{i-1} = a_i$, $\map{\crun}{\hrun}{i} = \map{\crun}{\hrun}{i-1}$, and (H.3.a).
			Since $\buf$ is an arbitrary prefix in $\prefixes{\buf_{i}}$, (3.a) holds.

			\item[Rule \textsc{Execute-Branch-Commit}:]
			Therefore, given the well-formedness of buffers (see Lemma~\ref{lemma:seq-processor:buffers-well-formedness}), we have that  $\buf_{i-1} = \tagged{\passign{\pc}{\lbl}}{\lbl_0}$, $\buf_{i} =  \tagged{\passign{\pc}{ \lbl  }}{\varepsilon}$, $m_{i} = m_{i-1}$, $a_{i} = a_{i-1}$, $p(\lbl_0) = \pjz{x}{\lbl''}$, and $(a_{i-1}(x) = 0 \wedge \lbl = \lbl'') \vee (a_{i-1}(x) \in \Val \setminus \{0,\bot\} \wedge \lbl = \ell_0+1)$.
			Therefore, we have that:
			\begin{align*}
				\prefixes{\buf_{i-1}} &= \{ \tagged{\passign{\pc}{\lbl}}{\lbl_0} \}\\
				\prefixes{\buf_{i}} &= \{ \tagged{\passign{\pc}{ \lbl  }}{\varepsilon} \}
			\end{align*}
			Let $\buf$ be the only prefix in $\prefixes{\buf_i}$, i.e., $\buf = \tagged{\passign{\pc}{ \lbl  }}{\varepsilon} $.
			From the induction hypothesis (H.3.a), we have that $\update{\tup{m_{i-1}, a_{i-1}}}{  \tagged{\passign{\pc}{\lbl}}{\lbl_0} } = \crun( \map{\crun}{\hrun}{i-1}(|\tagged{\passign{\pc}{\lbl}}{\lbl_0}|) )$.
			From $|\tagged{\passign{\pc}{\lbl}}{\lbl_0}| = 1$, we have $\update{\tup{m_{i-1}, a_{i-1}}}{  \tagged{\passign{\pc}{\lbl}}{\lbl_0} } = \crun( \map{\crun}{\hrun}{i-1}(1) )$.
			By applying $\update{}{}$, there are two cases:
			\begin{description}
				\item[$a_{i-1}(x) = 0:$]
				Then,  $\update{\tup{m_{i-1}, a_{i-1}}}{  \tagged{\passign{\pc}{\lbl}}{\lbl_0} } = \tup{m_{i-1},a_{i-1}[\pc \mapsto \lbl'']}$. 
				From $a_{i-1}(x) = 0$ and $(a_{i-1}(x) = 0 \wedge \lbl = \lbl'') \vee (a_{i-1}(x) \in \Val \setminus \{0,\bot\} \wedge \lbl = \ell_0+1)$, we immediately get that $\lbl = \lbl''$ and, therefore, $\update{\tup{m_{i-1}, a_{i-1}}}{  \tagged{\passign{\pc}{\lbl}}{\lbl_0} } = \tup{m_{i-1}, a_{i-1}[\pc \mapsto \lbl]}$. 
				Therefore, we have $\tup{m_{i-1}, a_{i-1}[\pc \mapsto \lbl]} = \crun( \map{\crun}{\hrun}{i-1}(1) )$.
				From $\buf = \tagged{\passign{\pc}{ \lbl  }}{\varepsilon} $ and $\update{\tup{m_{i-1}, a_{i-1}}}{ \buf }  = \tup{m_{i-1}, a_{i-1}[ \pc \mapsto \ell ] }$, we have $\update{\tup{m_{i-1}, a_{i-1}}}{ \buf } = \crun( \map{\crun}{\hrun}{i-1}(1) )$.
				From this and $|\buf| = 1$, we get $\update{\tup{m_{i-1}, a_{i-1}}}{ \buf } = \crun( \map{\crun}{\hrun}{i-1}(|\buf|) )$.
				From $m_{i-1} = m_i$, $a_{i-1} = a_i$, and $\map{\crun}{\hrun}{i} = \map{\crun}{\hrun}{i-1}$, we finally get $\update{\tup{m_{i}, a_{i}}}{ \buf } = \crun( \map{\crun}{\hrun}{i}(|\buf|) )$.
				Therefore, (3.a) holds. 

				\item[$a_{i-1}(x) \neq 0:$]
				Then,  $\update{\tup{m_{i-1}, a_{i-1}}}{  \tagged{\passign{\pc}{\lbl}}{\lbl_0} } = \tup{m_{i-1},a_{i-1}[\pc \mapsto \lbl_0+1]}$. 
				From $a_{i-1}(x) \neq  0$, $a_{i-1}(x)$ being defined, and $(a_{i-1}(x) = 0 \wedge \lbl = \lbl'') \vee (a_{i-1}(x) \in \Val \setminus \{0,\bot\} \wedge \lbl = \ell_0+1)$, we immediately get that $\lbl = \ell_0 + 1$ and, therefore, $\update{\tup{m_{i-1}, a_{i-1}}}{  \tagged{\passign{\pc}{\lbl}}{\lbl_0} } = \tup{m_{i-1}, a_{i-1}[\pc \mapsto \lbl]}$. 
				Therefore, we have $\tup{m_{i-1}, a_{i-1}[\pc \mapsto \lbl]} = \crun( \map{\crun}{\hrun}{i-1}(1) )$.
				From $\buf = \tagged{\passign{\pc}{ \lbl  }}{\varepsilon} $ and $\update{\tup{m_{i-1}, a_{i-1}}}{ \buf }  = \tup{m_{i-1}, a_{i-1}[ \pc \mapsto \ell ] }$, we have $\update{\tup{m_{i-1}, a_{i-1}}}{ \buf } = \crun( \map{\crun}{\hrun}{i-1}(1) )$.
				From this and $|\buf| = 1$, we get $\update{\tup{m_{i-1}, a_{i-1}}}{ \buf } = \crun( \map{\crun}{\hrun}{i-1}(|\buf|) )$.
				From $m_{i-1} = m_i$, $a_{i-1} = a_i$, and $\map{\crun}{\hrun}{i} = \map{\crun}{\hrun}{i-1}$, we finally get $\update{\tup{m_{i}, a_{i}}}{ \buf } = \crun( \map{\crun}{\hrun}{i}(|\buf|) )$.
				Therefore, (3.a) holds.
			\end{description}

			\item[Rule \textsc{Execute-Branch-Rollback}:]
			Therefore, given the well-formedness of buffers (see Lemma~\ref{lemma:seq-processor:buffers-well-formedness}), we have that  $\buf_{i-1} = \tagged{\passign{\pc}{\lbl}}{\lbl_0}$, $\buf_{i} =  \tagged{\passign{\pc}{ \lbl'  }}{\varepsilon}$, $m_{i} = m_{i-1}$, $a_{i} = a_{i-1}$, $p(\lbl_0) = \pjz{x}{\lbl''}$, $(a_{i-1}(x) = 0 \wedge \lbl \neq \lbl'') \vee (a_{i-1}(x) \in \Val \setminus \{0,\bot\} \wedge \lbl \neq \ell_0+1)$, and $\lbl' \in \{ \lbl'',\lbl_0 +1 \} \setminus \{ \ell \}$.
			Therefore, we have that:
			\begin{align*}
				\prefixes{\buf_{i-1}} &= \{ \tagged{\passign{\pc}{\lbl}}{\lbl_0} \}\\
				\prefixes{\buf_{i}} &= \{ \tagged{\passign{\pc}{ \lbl'  }}{\varepsilon} \}
			\end{align*}
			Let $\buf$ be the only prefix in $\prefixes{\buf_i}$, i.e., $\buf = \tagged{\passign{\pc}{ \lbl'  }}{\varepsilon} $.
			From the induction hypothesis (H.3.a), we have that $\update{\tup{m_{i-1}, a_{i-1}}}{  \tagged{\passign{\pc}{\lbl}}{\lbl_0} } = \crun( \map{\crun}{\hrun}{i-1}(|\tagged{\passign{\pc}{\lbl}}{\lbl_0}|) )$.
			From $|\tagged{\passign{\pc}{\lbl}}{\lbl_0}| = 1$, we have $\update{\tup{m_{i-1}, a_{i-1}}}{  \tagged{\passign{\pc}{\lbl}}{\lbl_0} } = \crun( \map{\crun}{\hrun}{i-1}(1) )$.
			By applying $\update{}{}$, there are two cases:
			\begin{description}
				\item[$a_{i-1}(x) = 0:$]
				Then,  $\update{\tup{m_{i-1}, a_{i-1}}}{  \tagged{\passign{\pc}{\lbl}}{\lbl_0} } = \tup{m_{i-1},a_{i-1}[\pc \mapsto \lbl'']}$. 
				From $a_{i-1}(x) = 0$, $(a_{i-1}(x) = 0 \wedge \lbl \neq \lbl'') \vee (a_{i-1}(x) \in \Val \setminus \{0,\bot\} \wedge \lbl \neq \ell_0+1)$, $\lbl \in \{ \lbl'', \lbl_0 +1\}$ (from the well-formedness of buffers), $\lbl' \in \{ \lbl'',\lbl_0 +1 \} \setminus \{ \ell \}$, we get that $\lbl' = \lbl''$, and, therefore, $\update{\tup{m_{i-1}, a_{i-1}}}{  \tagged{\passign{\pc}{\lbl}}{\lbl_0} } = \tup{m_{i-1}, a_{i-1}[\pc \mapsto \lbl']}$. 
				Therefore, we have $\tup{m_{i-1}, a_{i-1}[\pc \mapsto \lbl']} = \crun( \map{\crun}{\hrun}{i-1}(1) )$.
				From $\buf = \tagged{\passign{\pc}{ \lbl'  }}{\varepsilon} $ and $\update{\tup{m_{i-1}, a_{i-1}}}{ \buf }  = \tup{m_{i-1}, a_{i-1}[ \pc \mapsto \ell' ] }$, we have $\update{\tup{m_{i-1}, a_{i-1}}}{ \buf } = \crun( \map{\crun}{\hrun}{i-1}(1) )$.
				From this and $|\buf| = 1$, we get $\update{\tup{m_{i-1}, a_{i-1}}}{ \buf } = \crun( \map{\crun}{\hrun}{i-1}(|\buf|) )$.
				From $m_{i-1} = m_i$, $a_{i-1} = a_i$, and $\map{\crun}{\hrun}{i} = \map{\crun}{\hrun}{i-1}$, we finally get $\update{\tup{m_{i}, a_{i}}}{ \buf } = \crun( \map{\crun}{\hrun}{i}(|\buf|) )$.
				Therefore, (3.a) holds. 

				\item[$a_{i-1}(x) \neq 0:$]
				Then,  $\update{\tup{m_{i-1}, a_{i-1}}}{  \tagged{\passign{\pc}{\lbl}}{\lbl_0} } = \tup{m_{i-1},a_{i-1}[\pc \mapsto \lbl_0+1]}$. 
				From $a_{i-1}(x) \neq 0$, $(a_{i-1}(x) = 0 \wedge \lbl \neq \lbl'') \vee (a_{i-1}(x) \in \Val \setminus \{0,\bot\} \wedge \lbl \neq \ell_0+1)$, $\lbl \in \{ \lbl'', \lbl_0 +1\}$ (from the well-formedness of buffers), $\lbl' \in \{ \lbl'',\lbl_0 +1 \} \setminus \{ \ell \}$, we get that $\lbl' = \ell_0+1$, and, therefore, $\update{\tup{m_{i-1}, a_{i-1}}}{  \tagged{\passign{\pc}{\lbl}}{\lbl_0} } = \tup{m_{i-1}, a_{i-1}[\pc \mapsto \lbl']}$. 
				Therefore, we have $\tup{m_{i-1}, a_{i-1}[\pc \mapsto \lbl']} = \crun( \map{\crun}{\hrun}{i-1}(1) )$.
				From $\buf = \tagged{\passign{\pc}{ \lbl'  }}{\varepsilon} $ and $\update{\tup{m_{i-1}, a_{i-1}}}{ \buf }  = \tup{m_{i-1}, a_{i-1}[ \pc \mapsto \ell' ] }$, we have $\update{\tup{m_{i-1}, a_{i-1}}}{ \buf } = \crun( \map{\crun}{\hrun}{i-1}(1) )$.
				From this and $|\buf| = 1$, we get $\update{\tup{m_{i-1}, a_{i-1}}}{ \buf } = \crun( \map{\crun}{\hrun}{i-1}(|\buf|) )$.
				From $m_{i-1} = m_i$, $a_{i-1} = a_i$, and $\map{\crun}{\hrun}{i} = \map{\crun}{\hrun}{i-1}$, we finally get $\update{\tup{m_{i}, a_{i}}}{ \buf } = \crun( \map{\crun}{\hrun}{i}(|\buf|) )$.
				Therefore, (3.a) holds.
			\end{description}

			\item[Rule \textsc{Execute-Assignment}:]
			Therefore, given the well-formedness of buffers, $\buf_{i-1} = \tagged{\passign{x}{e}}{\notags} \concat \tagged{\pmarkedassign{\pc}{\ell}}{\notags}$, $\buf_{i} =  \tagged{\passign{x}{\exprEval{e}{\apply{\emptysequence}{a_{i-1}}}}}{\notags} \concat \tagged{\pmarkedassign{\pc}{\ell}}{\notags}$, $m_{i} = m_{i-1}$, and $a_{i} = a_{i-1}$.
			From the well-formedness of $\buf_{i-1}$ and $\buf_i$ (see Lemma~\ref{lemma:seq-processor:buffers-well-formedness}), we have that:
			\begin{align*}
				\prefixes{\buf_{i-1}} &= \{ \emptysequence, \tagged{\passign{x}{e}}{\notags} \concat \tagged{\pmarkedassign{\pc}{\ell}}{\notags}\}\\
				\prefixes{\buf_{i}} &= \{ \emptysequence, \tagged{\passign{x}{\exprEval{e}{\apply{\emptysequence}{a_{i-1}}}}}{\notags} \concat \tagged{\pmarkedassign{\pc}{\ell}}{\notags} \}
			\end{align*}
			Let $\buf$ be an arbitrary prefix in $\prefixes{\buf_{i}}$.
			There are two cases:
			\begin{description}
				\item[$\buf = \emptysequence$:]
				Then (3.a) immediately follows from $\buf \in \prefixes{\buf_{i-1}}$, $m_{i-1} = m_i$, $a_{i-1} = a_i$, $\map{\crun}{\hrun}{i} = \map{\crun}{\hrun}{i-1}$, and (H.3.a).

				\item[$\buf = \tagged{\passign{x}{\exprEval{e}{\apply{\emptysequence}{a_{i-1}}}}}{\notags} \concat \tagged{\pmarkedassign{\pc}{\ell}}{\notags}$:]
				From the induction hypothesis (H.3.a), we have $\update{\tup{m_{i-1}, a_{i-1}}}{  ( \tagged{\passign{x}{e}}{\notags} \concat \tagged{\pmarkedassign{\pc}{\ell}}{\notags} ) } = \crun( \map{\crun}{\hrun}{i-1}(|\tagged{\passign{x}{e}}{\notags} \concat \tagged{\pmarkedassign{\pc}{\ell}}{\notags}|) )$.
				By applying $\update{}{}$, we obtain that $
				\tup{m_{i-1}, a_{i-1}[x \mapsto \exprEval{e}{a_{i-1}}, \pc \mapsto \ell ] } = \crun( \map{\crun}{\hrun}{i-1}(|\tagged{\passign{x}{e}}{\notags} \concat \tagged{\pmarkedassign{\pc}{\ell}}{\notags}|) )$.
				Since $\buf = \tagged{\passign{x}{\exprEval{e}{\apply{\emptysequence}{a_{i-1}}}}}{\notags} \concat \tagged{\pmarkedassign{\pc}{\ell}}{\notags}$ and $\update{\tup{m_{i-1}, a_{i-1}}}{ (\tagged{\passign{x}{\exprEval{e}{\apply{\emptysequence}{a_{i-1}}}}}{\notags} \concat \tagged{\pmarkedassign{\pc}{\ell}}{\notags}) }  = \tup{m_{i-1}, a_{i-1}[x \mapsto \exprEval{e}{a_{i-1}}, \pc \mapsto \ell ] }$, we obtain that $\update{\tup{m_{i-1}, a_{i-1}}}{ \buf } = \crun( \map{\crun}{\hrun}{i-1}(|\tagged{\passign{x}{e}}{\notags} \concat \tagged{\pmarkedassign{\pc}{\ell}}{\notags}|) )$. 
				From this and $|\buf| = |\tagged{\passign{x}{e}}{\notags} \concat \tagged{\pmarkedassign{\pc}{\ell}}{\notags}|$, we have $\update{\tup{m_{i-1}, a_{i-1}}}{ \buf } = \crun( \map{\crun}{\hrun}{i-1}(|\buf|) )$.
				From $m_{i-1} = m_i$, $a_{i-1} = a_i$, and $\map{\crun}{\hrun}{i} = \map{\crun}{\hrun}{i-1}$, we finally get $\update{\tup{m_{i}, a_{i}}}{ \buf } = \crun( \map{\crun}{\hrun}{i}(|\buf|) )$.
				Therefore, (3.a) holds.
			\end{description}
			Since $\buf$ is an arbitrary prefix in $\prefixes{\buf_{i}}$, (3.a) holds.

			\item[Rule \textsc{Execute-Marked-Assignment}:]
			Therefore, given the well-formedness of buffers, $\buf_{i-1} = \tagged{\pmarkedassign{\pc}{\ell}}{\notags}$, $\ell \in \Val$, $\buf_{i} =  \tagged{\pmarkedassign{\pc}{\ell}}{\notags}$, $m_{i} = m_{i-1}$, and $a_{i} = a_{i-1}$.
			From the well-formedness of $\buf_{i-1}$ and $\buf_i$ (see Lemma~\ref{lemma:seq-processor:buffers-well-formedness}), we have that:
			\begin{align*}
				\prefixes{\buf_{i-1}} &= \{ \tagged{\pmarkedassign{\pc}{\ell}}{\notags}\}\\
				\prefixes{\buf_{i}} &= \{ \tagged{\pmarkedassign{\pc}{\ell}}{\notags} \}
			\end{align*}
			Let $\buf$ be an arbitrary prefix in $\prefixes{\buf_{i}}$.
			Then, (3.a) immediately follows from $\prefixes{\buf_i} = \prefixes{\buf_{i-1}}$, $\buf \in \prefixes{\buf_{i-1}}$, $m_{i-1} = m_i$, $a_{i-1} = a_i$, $\map{\crun}{\hrun}{i} = \map{\crun}{\hrun}{i-1}$, and (H.3.a).
			Since $\buf$ is an arbitrary prefix in $\prefixes{\buf_{i}}$, (3.a) holds.

			\item[Rule \textsc{Execute-Store}:]
			The proof of this case is similar to that of the rule \textsc{Execute-Assignment}.
			
			\item[Rule \textsc{Execute-Skip}:]
			The proof of this case is similar to that of the rule \textsc{Execute-Load-Miss}.

			\item[Rule \textsc{Execute-Barrier}:] 
			The proof of this case is similar to that of the rule \textsc{Execute-Load-Miss}.
			
		\end{description}
		This completes the proof for the $\execute{j}$ case.

		\item[$\SchedNext(C_{i-1}) = \retire{}$:]
		Then, $\map{\crun}{\hrun}{i} = \mathit{shift}(\map{\crun}{\hrun}{i-1})$.
		We now proceed by case distinction on the applied rule:
		\begin{description}
			\item[Rule \textsc{Retire-Skip}:]
			Then, $	\buf_{i-1} = \tagged{\pskip{}}{\notags} \concat \buf_i$, $m_{i} = m_{i-1}$, and $a_{i} = a_{i-1}$.
			From the well-formedness of $\buf_{i-1}$ (see Lemma~\ref{lemma:seq-processor:buffers-well-formedness}) and $	\buf_{i-1} = \tagged{\pskip{}}{\notags} \concat \buf_i$, we have that $\buf_{i-1} = \tagged{\pskip{}}{\notags} \concat \tagged{\pmarkedassign{\pc}{\ell}}{\notags} \concat \buf_i'$ and $\buf_i = \tagged{\pmarkedassign{\pc}{\ell}}{\notags} \concat \buf_i'$.
			Therefore, we have:
			\begin{align*}
				\prefixes{\buf_{i-1}} & = \{ \emptysequence \} \cup \{  
					\tagged{\pskip{}}{\notags} \concat \tagged{\pmarkedassign{\pc}{\ell}}{\notags} \concat \buf' \mid \buf' \in \prefixes{\buf_i'} \} \\
				\prefixes{\buf_{i}} & = \{ \tagged{\pmarkedassign{\pc}{\ell}}{\notags} \} \cup \{ \tagged{\pmarkedassign{\pc}{\ell}}{\notags} \concat \buf' \mid \buf' \in \prefixes{\buf_i'} \} 
			\end{align*}
			Let $\buf$ be an arbitrary prefix in $\prefixes{\buf_{i}}$.
			From the definitions of $\prefixes{\buf_{i-1}}$ and $\prefixes{\buf_{i}}$ (and from the well-formedness of the buffers), we know that $\tagged{\pskip{}}{\notags} \concat \buf$ is a prefix in  $\prefixes{\buf_{i-1}}$.
			From the induction hypothesis (H.3.a) and $\tagged{\pskip{}}{\notags} \concat \buf \in \prefixes{\buf_{i-1}}$,  we have $\update{\tup{m_{i-1},a_{i-1}}}{ (\tagged{\pskip{}}{\notags} \concat \buf) } = \crun( \map{\crun}{\hrun}{i-1}(| \tagged{\pskip{}}{\notags} \concat \buf |) )$.
			From $m_{i} = m_{i-1}$ and $a_{i} = a_{i-1}$, we get $\update{\tup{m_{i},a_{i}}}{ (\tagged{\pskip{}}{\notags} \concat \buf) } = \crun( \map{\crun}{\hrun}{i-1}(| \tagged{\pskip{}}{\notags} \concat \buf |) )$.
			By applying $\update{}{}$ to $\tagged{\pskip{}}{\notags}$, we get $\update{\tup{m_{i},a_{i}}}{ \buf } = \crun( \map{\crun}{\hrun}{i-1}(| \tagged{\pskip{}}{\notags} \concat \buf |) )$.
			From $| \tagged{\pskip{}}{\notags} \concat \buf | = |\buf|+1$, we get $\update{\tup{m_{i},a_{i}}}{ \buf } = \crun( \map{\crun}{\hrun}{i-1}(|\buf |+1) )$.
			From this and $\mathit{shift}$'s definition, we get $\update{\tup{m_{i},a_{i}}}{ \buf } = \crun( \mathit{shift}(\map{\crun}{\hrun}{i-1}(|\buf |)) )$.
			From this and $\map{\crun}{\hrun}{i} = \mathit{shift}(\map{\crun}{\hrun}{i-1})$, we get  $\update{\tup{m_{i},a_{i}}}{ \buf } = \crun( \map{\crun}{\hrun}{i}(|\buf |) )$.
			Since $\buf$ is an arbitrary prefix in $\prefixes{\buf_{i}}$, (3.a) holds.

			\item[Rule \textsc{Retire-Fence}:]
			The proof of this case is identical to that of the rule \textsc{Retire-Skip}.

			\item[Rule \textsc{Retire-Assignment}:]
			Then, $	\buf_{i-1} = \tagged{\passign{x}{v}}{\notags} \concat \buf_i$, $m_{i} = m_{i-1}$, $a_{i} = a_{i-1}[x \mapsto v]$, and $v \in \Val$.
			From the well-formedness of $\buf_{i-1}$ (see Lemma~\ref{lemma:seq-processor:buffers-well-formedness}) and $	\buf_{i-1} = \tagged{\passign{x}{v}}{\notags} \concat \buf_i$, we have that $\buf_{i-1} =  \tagged{\passign{x}{v}}{\notags} \concat \tagged{\pmarkedassign{\pc}{\ell}}{\notags} \concat \buf_i'$ and $\buf_i = \tagged{\pmarkedassign{\pc}{\ell}}{\notags} \concat \buf_i'$.
			Therefore, we have:
			\begin{align*}
				\prefixes{\buf_{i-1}} & = \{ \emptysequence \} \cup \{  
					\tagged{\passign{x}{v}}{\notags} \concat \tagged{\pmarkedassign{\pc}{\ell}}{\notags} \concat \buf' \mid \buf' \in \prefixes{\buf_i'} \} \\
				\prefixes{\buf_{i}} & = \{ \tagged{\pmarkedassign{\pc}{\ell}}{\notags} \} \cup \{ \tagged{\pmarkedassign{\pc}{\ell}}{\notags} \concat \buf' \mid \buf' \in \prefixes{\buf_i'} \} 
			\end{align*}
			Let $\buf$ be an arbitrary prefix in $\prefixes{\buf_{i}}$.
			From the definitions of $\prefixes{\buf_{i-1}}$ and $\prefixes{\buf_{i}}$ (and from the well-formedness of the buffers), we know that $\tagged{\passign{x}{v}}{\notags} \concat \buf$ is a prefix in  $\prefixes{\buf_{i-1}}$.
			From the induction hypothesis (H.3.a) and $\tagged{\passign{x}{v}}{\notags} \concat \buf\in \prefixes{\buf_{i-1}}$,  we have $\update{\tup{m_{i-1},a_{i-1}}}{ (\tagged{\passign{x}{v}}{\notags} \concat \buf) } = \crun( \map{\crun}{\hrun}{i-1}(| \tagged{\passign{x}{v}}{\notags} \concat \buf |) )$.
			By applying $\update{}{}$ to $\tagged{\passign{x}{v}}{\notags}$, we get $\update{\tup{m_{i-1},a_{i-1}[x \mapsto v]}}{  \buf } = \crun( \map{\crun}{\hrun}{i-1}(| \tagged{\passign{x}{v}}{\notags} \concat \buf |) )$.
			From $m_{i} = m_{i-1}$ and $a_{i} = a_{i-1}[x \mapsto v]$, we get $\update{\tup{m_{i},a_{i}}}{  \buf } = \crun( \map{\crun}{\hrun}{i-1}(| \tagged{\passign{x}{v}}{\notags} \concat \buf |) )$.
			From $| \tagged{\passign{x}{v}}{\notags} \concat \buf | = |\buf|+1$, we get $\update{\tup{m_{i},a_{i}}}{ \buf } = \crun( \map{\crun}{\hrun}{i-1}(|\buf |+1) )$.
			From this and $\mathit{shift}$'s definition, we get $\update{\tup{m_{i},a_{i}}}{ \buf } = \crun( \mathit{shift}(\map{\crun}{\hrun}{i-1}(|\buf |)) )$.
			From this and $\map{\crun}{\hrun}{i} = \mathit{shift}(\map{\crun}{\hrun}{i-1})$, we get  $\update{\tup{m_{i},a_{i}}}{ \buf } = \crun( \map{\crun}{\hrun}{i}(|\buf |) )$.
			Since $\buf$ is an arbitrary prefix in $\prefixes{\buf_{i}}$, (3.a) holds.

			\item[Rule \textsc{Retire-Marked-Assignment}:]
			Then, from the well-formedness of $\buf_{i-1}$ (see Lemma~\ref{lemma:seq-processor:buffers-well-formedness}), we have $\buf_{i-1} = \tagged{\pmarkedassign{\pc}{\ell}}{\notags} \concat \buf_i$, $m_{i} = m_{i-1}$, $a_{i} = a_{i-1}[\pc \mapsto \ell]$, and $\ell \in \Val$.
			Therefore, we have that:
			\begin{align*}
				\prefixes{\buf_{i-1}} & = \{ \tagged{\pmarkedassign{\pc}{\ell}}{\notags} \} \cup \{ \tagged{\pmarkedassign{\pc}{\ell}}{\notags} \concat \buf' \mid \buf' \in \prefixes{\buf_i} \}
			\end{align*}
			Let $\buf$ be an arbitrary prefix in $\prefixes{\buf_{i}}$.
			From the definitions of $\prefixes{\buf_{i-1}}$, we know that $\tagged{\pmarkedassign{\pc}{\ell}}{\notags} \concat \buf$ is a prefix in  $\prefixes{\buf_{i-1}}$.
			From the induction hypothesis (H.3.a) and $\tagged{\pmarkedassign{\pc}{\ell}}{\notags} \concat \buf\in \prefixes{\buf_{i-1}}$,  we have $\update{\tup{m_{i-1},a_{i-1}}}{ (\tagged{\pmarkedassign{\pc}{\ell}}{\notags} \concat \buf) } = \crun( \map{\crun}{\hrun}{i-1}(| \pmarkedassign{\pc}{\ell} \concat \buf |) )$.
			By applying $\update{}{}$ to $\tagged{\pmarkedassign{\pc}{\ell}}{\notags}$, we get $\update{\tup{m_{i-1},a_{i-1}[\pc \mapsto \ell]}}{  \buf } = \crun( \map{\crun}{\hrun}{i-1}(| \tagged{\pmarkedassign{\pc}{\ell}}{\notags} \concat \buf |) )$.
			From $m_{i} = m_{i-1}$ and $a_{i} = a_{i-1}[\pc \mapsto \ell]$, we get $\update{\tup{m_{i},a_{i}}}{  \buf } = \crun( \map{\crun}{\hrun}{i-1}(| \tagged{\pmarkedassign{\pc}{\ell}}{\notags} \concat \buf |) )$.
			From $| \tagged{\pmarkedassign{\pc}{\ell}}{\notags} \concat \buf | = |\buf|+1$, we get $\update{\tup{m_{i},a_{i}}}{ \buf } = \crun( \map{\crun}{\hrun}{i-1}(|\buf |+1) )$.
			From this and $\mathit{shift}$'s definition, we get $\update{\tup{m_{i},a_{i}}}{ \buf } = \crun( \mathit{shift}(\map{\crun}{\hrun}{i-1}(|\buf |)) )$.
			From this and $\map{\crun}{\hrun}{i} = \mathit{shift}(\map{\crun}{\hrun}{i-1})$, we get  $\update{\tup{m_{i},a_{i}}}{ \buf } = \crun( \map{\crun}{\hrun}{i}(|\buf |) )$.
			Since $\buf$ is an arbitrary prefix in $\prefixes{\buf_{i}}$, (3.a) holds.

			\item[Rule \textsc{Retire-Store}:]
			The proof of this statement is similar to that of the rule \textsc{Retire-Assignment}.
		\end{description}
		This completes the proof for the $\retire{}$ case.

	\end{description}
	This completes the proof of the induction step.
\end{description}
This concludes the proof of our theorem.
\end{proof}

\subsection{Indistinguishability lemma}

\begin{definition}[Deep-indistinguishability of hardware configurations]\label{def:seq-processor:deep-indistinguishability}
	We say that two hardware configurations $\tup{\sigma,\mu} = \tup{m,a,\buf, \CacheState,\BpState, \SchedState}$ and $\tup{\sigma',\mu'} = \tup{m',a',\buf', \CacheState',\BpState', \SchedState'}$ are \emph{deep-indistinguishable}, written $\tup{\sigma,\mu} \sim \tup{\sigma',\mu'}$, iff
	\begin{inparaenum}[(a)]
		\item $\apply{\buf}{a}(\pc) = \apply{\buf'}{a'}(\pc)$,
		\item $\DeepProject{\buf} = \DeepProject{\buf'}$,
		\item $\CacheState = \CacheState'$,
		\item $\BpState = \BpState'$, and
		\item $\SchedState = \SchedState'$,
		\end{inparaenum}
		where $\DeepProject{\buf}$ is inductively defined as follows:
		\begin{align*}
			\DeepProject{\emptysequence}	&:= \emptysequence\\
			\DeepProject{\tagged{\pskip{}}{T}}  &:=  \tagged{\pskip{}}{T}\\
			\DeepProject{\tagged{\pbarrier{}}{T}}  &:=  \tagged{\pbarrier{}}{T}\\
			\DeepProject{\tagged{\passign{x}{e}}{T}}  &:=
			{
				\begin{cases}
				\tagged{\passign{x}{\resolved}}{T} & \text{if}\ e \in \Val \wedge x \neq \pc\\
				\tagged{\passign{x}{e}}{T} & \text{if}\ e \not\in \Val \wedge x \neq \pc\\
				\tagged{\passign{x}{e}}{T} & \text{if}\ x = \pc
				\end{cases}
			}\\
			\DeepProject{\tagged{\pload{x}{e}}{T}}  &:=
			{
				\tagged{\pload{x}{e}}{T}
			}\\
			\DeepProject{\tagged{\pstore{x}{e}}{T}}  &:=
			{
				\begin{cases}
					\tagged{\pstore{x}{e}}{T} & \text{if}\ x \in \Var\\
					\tagged{\pstore{\resolved}{e}}{T} & \text{if}\ x \in \Val
				\end{cases}
			}\\
			\DeepProject{\tagged{i}{T} \concat \buf} &:= \DeepProject{\tagged{i}{T}} \concat \DeepProject{\buf}
		\end{align*}
\end{definition}

\begin{definition}[Equivalence between \archstate{}s and hardware configurations]
We say that a hardware configuration $C = \tup{m,a,\CacheState, \BpState, \SchedState}$ and an \archstate{} $\sigma =  \tup{m',a'}$ are \emph{$i$-equivalent}, for an integer $0 \leq i \leq |\buf|$, written $C \bufEquiv{i} \sigma$, iff $\update{\tup{m,a}}{\buf[1..i]} = \tup{m',a'}$, where $\buf[1..0] = \emptysequence$.
\end{definition}

\begin{lemma}[Deep-buffer projection equality implies equality of data-independent projections]\label{lemma:seq-processor:buffer-projections}
Let $\buf, \buf'$ be two reorder buffers.
If $\DeepProject{\buf} = \DeepProject{\buf'}$, then $\BufProject{\buf} = \BufProject{\buf'}$.
\end{lemma}

\begin{proof}
It follows from the fact that everything disclosed by $\DeepProject{\buf}$ is also disclosed by $\BufProject{\buf}$.
\end{proof}

\begin{lemma}[Observation equivalence preserves deep-indistinguishability]\label{lemma:seq-processor:trace-equiv-implies-stepwise-indistinguishability}
Let $p$ be a well-formed program and $C_0 = \tup{m_0,a_0,\buf_0,\CacheState_0, \BpState_0, \SchedState_0}$, $C_0' = \tup{m_0',a_0',\buf_0',\CacheState_0', \BpState_0', \SchedState_0'}$ be reachable hardware configurations.
If
\begin{inparaenum}[(a)]
	\item $C_0 \sim C_0'$, and
	\item for all $\buf \in \prefixes{\buf_0}$, $\buf' \in \prefixes{\buf_0'}$ such that $|\buf| = |\buf'|$, 
	there are $\sigma_0, \sigma_0', \sigma_1, \sigma_1', \tau, \tau'$ such that $\sigma_0 \CtSeqInterfStep{\tau}{} \sigma_1$, $\sigma_0' \CtSeqInterfStep{\tau'}{} \sigma_1'$, $\tau = \tau'$, $C_0 \bufEquiv{|\buf|} \sigma_0$, and $C_0' \bufEquiv{|\buf'|} \sigma_0'$,
\end{inparaenum}
then either there are $C_1, C_1'$ such that $C_0 \SeqProcMuarchStep{}{} C_1$, $C_0' \SeqProcMuarchStep{}{} C_1'$, and $C_1 \sim C_1'$ or there is no $C_1$ such that $C_0 \SeqProcMuarchStep{}{} C_1$ and no $C_1'$ such that $C_0' \SeqProcMuarchStep{}{} C_1'$.
\end{lemma}

\begin{proof}
Let $p$ be a well-formed program, $C_0 = \tup{m_0,a_0,\buf_0,\CacheState_0, \BpState_0, \SchedState_0}$, and $C_0' = \tup{m_0',a_0',\buf_0',\CacheState_0', \BpState_0', \SchedState_0'}$.
Moreover, we assume that conditions (a) and (b) holds. 
In the following, we denote by (c) the post-condition ``either there are $C_1, C_1'$ such that $C_0 \SeqProcMuarchStep{}{} C_1$, $C_0' \SeqProcMuarchStep{}{} C_1'$, and $C_1 \sim C_1'$ or there is no $C_1$ such that $C_0 \SeqProcMuarchStep{}{} C_1$ and no $C_1'$ such that $C_0' \SeqProcMuarchStep{}{} C_1'$.''

From (a), it follows that $\SchedState_0 = \SchedState_0$.
Therefore, the directive obtained from the scheduler is the same in both cases, i.e., $\SchedNext(\SchedState_0) = \SchedNext(\SchedState_0')$.
We proceed by case distinction on the directive $d = \SchedNext(\SchedState_0)$:
\begin{description}
\item[$d = \fetch{}$:]
	Therefore, we can only apply one of the $\fetch{}$ rules depending on the current program counter.
	There are two cases: 
	\begin{description}
		\item[$\apply{\buf_0}{a_0}(\pc) \neq \bot \wedge |\buf_0| < \wMuarch$:]
		There are several cases:
		\begin{description}
			\item[$\CacheAccess(\CacheState_0,  \apply{\buf_0}{a_0}(\pc)) = \CacheHit \wedge p(\apply{\buf_0}{a_0}(\pc)) = \pjz{x}{\lbl}$:] 
			From (a), we get that $\CacheState_0' = \CacheState_0$ and $\apply{\buf_0'}{a_0'}(\pc) = \apply{\buf_0}{a_0}(\pc)$.
			Therefore, $\CacheAccess(\CacheState_0',  \apply{\buf_0'}{a_0'}(\pc)) = \CacheHit$.
			Moreover, from (a) we also get that $\apply{\buf_0}{a_0}(\pc)=\apply{\buf_0'}{a_0'}(\pc)$ and, therefore, $p(\apply{\buf_0'}{a_0'}(\pc)) = \pjz{x}{\lbl}$ as well.
			Therefore, we can apply the \textsc{Fetch-Branch-Hit} and \textsc{Step} rules to $C_0$ and $C_0'$ as follows:
			\begin{align*}
				\lbl_0 &:= \BpPredict(\BpState_0, \apply{\buf_0}{a_0}(\pc))\\
				\lbl_0' &:= \BpPredict(\BpState_0', \apply{\buf_0'}{a_0'}(\pc))\\
				\buf_1 &:= \buf_0  \concat  \tagged{\passign{\pc}{\lbl_0}}{\apply{\buf_0}{a_0}(\pc)}\\
				\buf_1' &:= \buf_0'  \concat \tagged{\passign{\pc}{\lbl_0'}}{\apply{\buf_0'}{a_0'}(\pc)}\\
				\tup{m_0,a_0,\buf_0,\CacheState_0,\BpState_0} &\muarchStep{\fetch{}}{} \tup{m_0, a_0, \buf_1, \CacheUpdate(\CacheState_0, \apply{\buf_0}{a_0}(\pc)),\BpState_0}\\
				\tup{m_0,a_0,\buf_0,\CacheState_0,\BpState_0, \SchedState_0} &\SeqProcMuarchStep{}{} \tup{m_0, a_0, \buf_1, \CacheUpdate(\CacheState_0, \apply{\buf_0}{a_0}(\pc)),\BpState_0, \SchedUpdate(\SchedState_0, \BufProject{\buf_1})}\\
				\tup{m_0',a_0',\buf_0',\CacheState_0',\BpState_0'} &\muarchStep{\fetch{}}{} \tup{m_0', a_0', \buf_1', \CacheUpdate(\CacheState_0',  \apply{\buf_0'}{a_0'}(\pc)),\BpState_0'}\\
				\tup{m_0',a_0',\buf_0',\CacheState_0',\BpState_0', \SchedState_0'} &\SeqProcMuarchStep{}{} \tup{m_0', a_0', \buf_1', \CacheUpdate(\CacheState_0',  \apply{\buf_0'}{a_0'}(\pc)),\BpState_0', \SchedUpdate(\SchedState_0', \BufProject{\buf_1'})}
			\end{align*}
			We now show that $C_1 =  \tup{m_0, a_0, \buf_1, \CacheUpdate(\CacheState_0, \apply{\buf_0}{a_0}(\pc)),\BpState_0, \SchedUpdate(\SchedState_0, \BufProject{\buf_1})}$ and $C_1' = \tup{m_0', a_0', \buf_1', \CacheUpdate(\CacheState_0',  \apply{\buf_0'}{a_0'}(\pc)),\BpState_0', \SchedUpdate(\SchedState_0', \BufProject{\buf_1'})}$ are indistinguishable, i.e., $C_1 \sim C_1'$.
			For this, we need to show that:
			\begin{description}
				\item[$\apply{\buf_1}{a_0}(\pc) = \apply{\buf_1'}{a_0'}(\pc)$:]
				We know that $\apply{\buf_1}{a_0}(\pc) = \lbl_0$ and $\apply{\buf_1'}{a_0'}(\pc) = \lbl_0'$.
				From $\lbl_0 = \BpPredict(\BpState_0, \apply{\buf_0}{a_0}(\pc))$, $\lbl_0' = \BpPredict(\BpState_0', \apply{\buf_0'}{a_0'}(\pc))$, and (a), we immediately get $\lbl_0 = \lbl_0'$.

				\item[$\DeepProject{\buf_1} = \DeepProject{\buf_1'}$:] 
				This follows from $\DeepProject{\buf_0} = \DeepProject{\buf_0'}$, which in turn follows from (a), $\buf_1 = \buf_0  \concat  \tagged{\passign{\pc}{\lbl_0}}{\apply{\buf_0}{a_0}(\pc)}$, $\buf_1' = \buf_0'  \concat  \tagged{\passign{\pc}{\lbl_0'}}{\apply{\buf_0'}{a_0'}(\pc)}$, and $\lbl_0 = \lbl_0'$.

				\item[$\CacheUpdate(\CacheState_0,  \apply{\buf_0}{a_0}(\pc)) = \CacheUpdate(\CacheState_0',  \apply{\buf_0'}{a_0'}(\pc))$:]
				This follows from $\CacheState_0 = \CacheState_0'$, which in turn follows from (a), and $\apply{\buf_0}{a_0}(\pc) = \apply{\buf_0'}{a_0'}(\pc)$.

				\item[$\BpState_0 = \BpState_0':$] 
				This follows from (a).

				\item[$\SchedUpdate(\SchedState_0, \BufProject{\buf_1}) = \SchedUpdate(\SchedState_0', \BufProject{\buf_1'})$:]
				From (a), we have   $\SchedState_0 = \SchedState_0'$.
				From $\DeepProject{\buf_1} = \DeepProject{\buf_1'}$ and Lemma~\ref{lemma:seq-processor:buffer-projections}, we have $\BufProject{\buf_1} = \BufProject{\buf_1'}$.
				Therefore, $\SchedUpdate(\SchedState_0, \BufProject{\buf_1}) = \SchedUpdate(\SchedState_0', \BufProject{\buf_1'})$.
			\end{description}
			Therefore, $C_1 \sim C_1'$ and (c) holds.

			\item[$\CacheAccess(\CacheState_0,  \apply{\buf_0}{a_0}(\pc)) = \CacheHit \wedge p(\apply{\buf_0}{a_0}(\pc)) = \pjmp{e}$:] 
			From (a), we get that $\CacheState_0' = \CacheState_0$ and $\apply{\buf_0'}{a_0'}(\pc) = \apply{\buf_0}{a_0}(\pc)$.
			Therefore, $\CacheAccess(\CacheState_0',  \apply{\buf_0'}{a_0'}(\pc)) = \CacheHit$.
			Moreover, from (a) we also get that $\apply{\buf_0}{a_0}(\pc)=\apply{\buf_0'}{a_0'}(\pc)$ and, therefore, $p(\apply{\buf_0'}{a_0'}(\pc)) = \pjmp{e}$ as well.
			Therefore, we can apply the \textsc{Fetch-Jump-Hit} and \textsc{Step} rules to $C_0$ and $C_0'$ as follows:
			\begin{align*}
				\buf_1 &:= \buf_0  \concat  \tagged{\passign{\pc}{e}}{\notags}\\
				\buf_1' &:= \buf_0'  \concat  \tagged{\passign{\pc}{e}}{\notags}\\
				\tup{m_0,a_0,\buf_0,\CacheState_0,\BpState_0} &\muarchStep{\fetch{}}{} \tup{m_0, a_0, \buf_1, \CacheUpdate(\CacheState_0, \apply{\buf_0}{a_0}(\pc)),\BpState_0}\\
				\tup{m_0,a_0,\buf_0,\CacheState_0,\BpState_0, \SchedState_0} &\SeqProcMuarchStep{}{} \tup{m_0, a_0, \buf_1, \CacheUpdate(\CacheState_0, \apply{\buf_0}{a_0}(\pc)),\BpState_0, \SchedUpdate(\SchedState_0, \BufProject{\buf_1})}\\
				\tup{m_0',a_0',\buf_0',\CacheState_0',\BpState_0'} &\muarchStep{\fetch{}}{} \tup{m_0', a_0', \buf_1', \CacheUpdate(\CacheState_0',  \apply{\buf_0'}{a_0'}(\pc)),\BpState_0'}\\
				\tup{m_0',a_0',\buf_0',\CacheState_0',\BpState_0', \SchedState_0'} &\SeqProcMuarchStep{}{} \tup{m_0', a_0', \buf_1', \CacheUpdate(\CacheState_0',  \apply{\buf_0'}{a_0'}(\pc)),\BpState_0', \SchedUpdate(\SchedState_0', \BufProject{\buf_1'})}
			\end{align*}
			We now show that $C_1 =  \tup{m_0, a_0, \buf_1, \CacheUpdate(\CacheState_0, \apply{\buf_0}{a_0}(\pc)),\BpState_0, \SchedUpdate(\SchedState_0, \BufProject{\buf_1})}$ and $C_1' = \tup{m_0', a_0', \buf_1', \CacheUpdate(\CacheState_0',  \apply{\buf_0'}{a_0'}(\pc)),\BpState_0', \SchedUpdate(\SchedState_0', \BufProject{\buf_1'})}$ are indistinguishable, i.e., $C_1 \sim C_1'$.
			For this, we need to show that:
			\begin{description}
				\item[$\apply{\buf_1}{a_0}(\pc) = \apply{\buf_1'}{a_0'}(\pc)$:]
				There are two cases.
				If $e \in \Val$, then $\apply{\buf_1}{a_0}(\pc) = \apply{\buf_1'}{a_0'}(\pc) = e$.
				Otherwise, $\apply{\buf_1}{a_0}(\pc) = \apply{\buf_1'}{a_0'}(\pc) = \bot$.

				\item[$\DeepProject{\buf_1} = \DeepProject{\buf_1'}$:] 
				This follows from $\DeepProject{\buf_0} = \DeepProject{\buf_0'}$, which in turn follows from (a), $\buf_1 = \buf_0  \concat  \tagged{\passign{\pc}{e}}{\notags}$, and $\buf_1' = \buf_0'  \concat  \tagged{\passign{\pc}{e}}{\notags}$.

				\item[$\CacheUpdate(\CacheState_0,  \apply{\buf_0}{a_0}(\pc)) = \CacheUpdate(\CacheState_0',  \apply{\buf_0'}{a_0'}(\pc))$:]
				This follows from $\CacheState_0 = \CacheState_0'$, which in turn follows from (a), and $\apply{\buf_0}{a_0}(\pc) = \apply{\buf_0'}{a_0'}(\pc)$.

				\item[$\BpState_0 = \BpState_0':$] 
				This follows from (a).

				\item[$\SchedUpdate(\SchedState_0, \BufProject{\buf_1}) = \SchedUpdate(\SchedState_0', \BufProject{\buf_1'})$:]
				From (a), we have   $\SchedState_0 = \SchedState_0'$.
				From $\DeepProject{\buf_1} = \DeepProject{\buf_1'}$ and Lemma~\ref{lemma:seq-processor:buffer-projections}, we have $\BufProject{\buf_1} = \BufProject{\buf_1'}$.
				Therefore, $\SchedUpdate(\SchedState_0, \BufProject{\buf_1}) = \SchedUpdate(\SchedState_0', \BufProject{\buf_1'})$.
			\end{description}
			Therefore, $C_1 \sim C_1'$ and (c) holds.
			
			\item[$\CacheAccess(\CacheState_0,  \apply{\buf_0}{a_0}(\pc)) = \CacheHit \wedge p(\apply{\buf_0}{a_0}(\pc)) \neq \pjz{x}{\lbl} \wedge p(\apply{\buf_0}{a_0}(\pc)) \neq \pjmp{e}$:]
			Observe that from (a) it follows that $|\buf_0| \geq \wMuarch-1$ iff $|\buf_0'| \geq \wMuarch-1$.
			Therefore, if $|\buf_0| \geq \wMuarch-1$, then (c) holds since both computations are stuck.
			In the following, we assume that $|\buf_0| < \wMuarch-1$ and $|\buf_0'| < \wMuarch-1$.
			
			From (a), we get that $\CacheState_0' = \CacheState_0$ and $\apply{\buf_0'}{a_0'}(\pc) = \apply{\buf_0}{a_0}(\pc)$.
			Therefore, $\CacheAccess(\CacheState_0',  \apply{\buf_0'}{a_0'}(\pc)) = \CacheHit$.
			Moreover, from (a) we also get that $\apply{\buf_0}{a_0}(\pc)=\apply{\buf_0'}{a_0'}(\pc)$ and, therefore, $p(\apply{\buf_0'}{a_0'}(\pc)) \neq \pjz{x}{\lbl} \wedge p(\apply{\buf_0'}{a_0'}(\pc)) \neq \pjmp{e}$ as well.
			Therefore, we can apply the \textsc{Fetch-Others-Hit} and \textsc{Step} rules to $C_0$ and $C_0'$ as follows:
			\begin{align*}
				v &:= \apply{\buf_0}{a_0}(\pc) +1 \\
				v' &:= \apply{\buf_0'}{a_0'}(\pc) +1 \\
				\buf_1 &:= \buf_0  \concat \tagged{p(\apply{\buf_0}{a_0}(\pc))}{\notags} \concat \tagged{\pmarkedassign{\pc}{v}}{\notags}\\
				\buf_1' &:= \buf_0'  \concat \tagged{p(\apply{\buf_0'}{a_0'}(\pc))}{\notags} \concat \tagged{\pmarkedassign{\pc}{v'}}{\notags}\\
				\tup{m_0,a_0,\buf_0,\CacheState_0,\BpState_0} &\muarchStep{\fetch{}}{} \tup{m_0, a_0, \buf_1, \CacheUpdate(\CacheState_0, \apply{\buf_0}{a_0}(\pc)),\BpState_0}\\
				\tup{m_0,a_0,\buf_0,\CacheState_0,\BpState_0, \SchedState_0} &\SeqProcMuarchStep{}{} \tup{m_0, a_0, \buf_1, \CacheUpdate(\CacheState_0, \apply{\buf_0}{a_0}(\pc)),\BpState_0, \SchedUpdate(\SchedState_0, \BufProject{\buf_1})}\\
				\tup{m_0',a_0',\buf_0',\CacheState_0',\BpState_0'} &\muarchStep{\fetch{}}{} \tup{m_0', a_0', \buf_1', \CacheUpdate(\CacheState_0',  \apply{\buf_0'}{a_0'}(\pc)),\BpState_0'}\\
				\tup{m_0',a_0',\buf_0',\CacheState_0',\BpState_0', \SchedState_0'} &\SeqProcMuarchStep{}{} \tup{m_0', a_0', \buf_1', \CacheUpdate(\CacheState_0',  \apply{\buf_0'}{a_0'}(\pc)),\BpState_0', \SchedUpdate(\SchedState_0', \BufProject{\buf_1'})}
			\end{align*}
			We now show that $C_1 =  \tup{m_0, a_0, \buf_1, \CacheUpdate(\CacheState_0, \apply{\buf_0}{a_0}(\pc)),\BpState_0, \SchedUpdate(\SchedState_0, \BufProject{\buf_1})}$ and $C_1' = \tup{m_0', a_0', \buf_1', \CacheUpdate(\CacheState_0',  \apply{\buf_0'}{a_0'}(\pc)),\BpState_0', \SchedUpdate(\SchedState_0', \BufProject{\buf_1'})}$ are indistinguishable, i.e., $C_1 \sim C_1'$.
			For this, we need to show that:
			\begin{description}
				\item[$\apply{\buf_1}{a_0}(\pc) = \apply{\buf_1'}{a_0'}(\pc)$:]
				From (a), we get that $\apply{\buf_0}{a_0}(\pc) = \apply{\buf_0'}{a_0'}(\pc)$.
				From this, we have that $v = v'$.
				Therefore, $\apply{\buf_1}{a_0}(\pc) = \apply{\buf_1'}{a_0'}(\pc) = v$.

				\item[$\DeepProject{\buf_1} = \DeepProject{\buf_1'}$:] 
				This follows from $\DeepProject{\buf_0} = \DeepProject{\buf_0'}$, which in turn follows from (a), $\buf_1 = \buf_0  \concat \tagged{p(\apply{\buf_0}{a_0}(\pc))}{\notags} \concat \tagged{\pmarkedassign{\pc}{v}}{\notags}$, $\buf_1' = \buf_0'  \concat \tagged{p(\apply{\buf_0'}{a_0'}(\pc))}{\notags} \concat \tagged{\pmarkedassign{\pc}{v'}}{\notags}$, and $v = v'$.

				\item[$\CacheUpdate(\CacheState_0,  \apply{\buf_0}{a_0}(\pc)) = \CacheUpdate(\CacheState_0',  \apply{\buf_0'}{a_0'}(\pc))$:]
				This follows from $\CacheState_0 = \CacheState_0'$, which in turn follows from (a), and $\apply{\buf_0}{a_0}(\pc) = \apply{\buf_0'}{a_0'}(\pc)$.

				\item[$\BpState_0 = \BpState_0':$] 
				This follows from (a).

				\item[$\SchedUpdate(\SchedState_0, \BufProject{\buf_1}) = \SchedUpdate(\SchedState_0', \BufProject{\buf_1'})$:]
				From (a), we have   $\SchedState_0 = \SchedState_0'$.
				From $\DeepProject{\buf_1} = \DeepProject{\buf_1'}$ and Lemma~\ref{lemma:seq-processor:buffer-projections}, we have $\BufProject{\buf_1} = \BufProject{\buf_1'}$.
				Therefore, $\SchedUpdate(\SchedState_0, \BufProject{\buf_1}) = \SchedUpdate(\SchedState_0', \BufProject{\buf_1'})$.
			\end{description}
			Therefore, $C_1 \sim C_1'$ and (c) holds.

			\item[$\CacheAccess(\CacheState_0,  \apply{\buf_0}{a_0}(\pc)) = \CacheMiss$:]
			From (a), we get that $\CacheState_0' = \CacheState_0$ and $\apply{\buf_0'}{a_0'}(\pc) = \apply{\buf_0}{a_0}(\pc)$.
			Therefore, $\CacheAccess(\CacheState_0',  \apply{\buf_0'}{a_0'}(\pc)) = \CacheMiss$.
			Therefore, we can apply the \textsc{Fetch-Miss} and \textsc{Step} rules to $C_0$ and $C_0'$ as follows:
			\begin{align*}
				\tup{m_0,a_0,\buf_0,\CacheState_0,\BpState_0} &\muarchStep{\fetch{}}{} \tup{m_0, a_0, \buf_0, \CacheUpdate(\CacheState_0, \apply{\buf_0}{a_0}(\pc)),\BpState_0}\\
				\tup{m_0,a_0,\buf_0,\CacheState_0,\BpState_0, \SchedState_0} &\SeqProcMuarchStep{}{} \tup{m_0, a_0, \buf_0, \CacheUpdate(\CacheState_0,  \apply{\buf_0}{a_0}(\pc)),\BpState_0, \SchedUpdate(\SchedState_0, \BufProject{\buf_0})}\\
				\tup{m_0',a_0',\buf_0',\CacheState_0',\BpState_0'} &\muarchStep{\fetch{}}{} \tup{m_0', a_0', \buf_0', \CacheUpdate(\CacheState_0',  \apply{\buf_0'}{a_0'}(\pc)),\BpState_0'}\\
				\tup{m_0',a_0',\buf_0',\CacheState_0',\BpState_0', \SchedState_0'} &\SeqProcMuarchStep{}{} \tup{m_0', a_0', \buf_0', \CacheUpdate(\CacheState_0',  \apply{\buf_0'}{a_0'}(\pc)),\BpState_0', \SchedUpdate(\SchedState_0', \BufProject{\buf_0'})}
			\end{align*}
			We now show that $C_1 = \tup{m_0, a_0, \buf_0, \CacheUpdate(\CacheState_0,  \apply{\buf_0}{a_0}(\pc)),\BpState_0, \SchedUpdate(\SchedState_0, \BufProject{\buf_0})}$ and $C_1' = \tup{m_0', a_0', \buf_0', \CacheUpdate(\CacheState_0',  \apply{\buf_0'}{a_0'}(\pc)),\BpState_0', \SchedUpdate(\SchedState_0', \BufProject{\buf_0'})}$ are indistinguishable, i.e., $C_1 \sim C_1'$.
			For this, we need to show that:
			\begin{description}
				\item[$\apply{\buf_0}{a_0}(\pc) = \apply{\buf_0'}{a_0'}(\pc)$:]
				This follows from (a).

				\item[$\DeepProject{\buf_0} = \DeepProject{\buf_0'}$:] 
				This follows from (a).

				\item[$\CacheUpdate(\CacheState_0,  \apply{\buf_0}{a_0}(\pc)) = \CacheUpdate(\CacheState_0',  \apply{\buf_0'}{a_0'}(\pc))$:]
				This follows from $\CacheState_0 = \CacheState_0'$, which in turn follows from (a), and $\apply{\buf_0}{a_0}(\pc) = \apply{\buf_0'}{a_0'}(\pc)$.

				\item[$\BpState_0 = \BpState_0':$] 
				This follows from (a).

				\item[$\SchedUpdate(\SchedState_0, \BufProject{\buf_0}) = \SchedUpdate(\SchedState_0', \BufProject{\buf_0'})$:]
				From (a), we have   $\SchedState_0 = \SchedState_0'$.
				From $\DeepProject{\buf_1} = \DeepProject{\buf_1'}$ and Lemma~\ref{lemma:seq-processor:buffer-projections}, we have $\BufProject{\buf_1} = \BufProject{\buf_1'}$.
				Therefore, $\SchedUpdate(\SchedState_0, \BufProject{\buf_1}) = \SchedUpdate(\SchedState_0', \BufProject{\buf_1'})$.
			\end{description}
			Therefore, $C_1 \sim C_1'$ and (c) holds.
		\end{description}
		
		\item[$\apply{\buf_0}{a_0}(\pc) = \bot \vee |\buf_0| \geq \wMuarch$:]
		Then, from (a), we immediately get that $\apply{\buf_0'}{a_0'}(\pc) = \bot \vee |\buf_0'| \geq \wMuarch$ holds as well.
		Therefore, both computations are stuck and (c) holds.
	\end{description}
	Therefore, (c) holds for all the cases.

\item[$d = \execute{i}$:]
Therefore, we can only apply one of the $\execute{}$ rules.
We remark that, since $\muarchStyle{seq}$ is equipped with the sequential scheduler from Appendix~\ref{appendix:seq-scheduler}, we have that $i = 1$ and $\buf_0[0..i-1] = \buf_0[0..i-1] = \emptysequence$.
There are two cases:
\begin{description}
	\item[$i \leq |\buf_0| \wedge \pbarrier \not\in {\buf_0[0..i-1]}$:]
	There are several cases depending on the $i$-th command in the reorder buffer:
	\begin{description}
		\item[$\elt{\buf_0}{i} = \tagged{\pload{x}{e}}{T}$:]
		From (a), we also have that $\elt{\buf_0'}{i} =  \tagged{\pload{x}{e}}{T}$, $i \leq |\buf_0'|$, and $\pbarrier \not\in \buf_0'[0..i-1]$. 
		There are two cases:
		\begin{description}
			\item[$\pstore{x'}{e'} \not\in \buf_0{[0..i-1]}$:]
			We now show that $\exprEval{e}{\apply{\buf_0[0..i-1]}{a_0}} = \exprEval{e}{\apply{\buf_0'[0..i-1]}{a_0'}}$.
			Since $C_0, C_0'$ are reachable configurations, the buffers $\buf_0, \buf_0'$ are well-formed (see Lemma~\ref{lemma:seq-processor:buffers-well-formedness}), and therefore  $\buf_0[0..i-1] \in \prefixes{\buf_0}$ and  $\buf_0'[0..i-1] \in \prefixes{\buf_0'}$.
			From (b), therefore, there are configurations $\sigma_0, \sigma_0', \sigma_1, \sigma_1'$ such that $C_0 \bufEquiv{|\buf_0[0..i-1]|} \sigma_0$, $C_0' \bufEquiv{|\buf_0[0..i-1]|} \sigma_0'$, $\sigma_0 \CtSeqInterfStep{\tau}{} \sigma_1$,  $\sigma_0' \CtSeqInterfStep{\tau'}{} \sigma_1'$, and $\tau = \tau'$. 
			From $C_0 \bufEquiv{|\buf_0[0..i-1]|} \sigma_0$, $C_0' \bufEquiv{|\buf_0[0..i-1]|} \sigma_0'$, and the well-formedness of the buffers, we know that $p(\sigma_0(\pc)) = p(\sigma_0'(\pc)) = \pload{x}{e}$.
			From $\CtSeqInterf{\cdot}$, we have that $\tau = \loadObs{ \exprEval{e}{\sigma_0} }$  and $\tau' = \loadObs{ \exprEval{e}{\sigma_0'} }$.
			From $\tau=\tau'$, we get that $\exprEval{e}{\sigma_0} = \exprEval{e}{\sigma_0'}$.
			From this, $C_0 \bufEquiv{|\buf_0[0..i-1]|} \sigma_0$, and $C_0' \bufEquiv{|\buf_0'[0..i-1]|} \sigma_0'$, we finally get $\exprEval{e}{\apply{\buf_0{[0..i-1]}}{a_0}} = \exprEval{e}{\apply{\buf_0'{[0..i-1]}}{a_0'}(x)}$.

			Let $n = \exprEval{e}{\apply{\buf_0{[0..i-1]}}{a_0}} = \exprEval{e}{\apply{\buf_0'{[0..i-1]}}{a_0'}(x)}$.
			There are two cases:
			\begin{description}
				\item[$\CacheAccess(\CacheState_0, \exprEval{e}{ \apply{\buf_0{[0..i-1]}}{a_0} }) = \CacheHit$:]
				From (a), we have that $\CacheState_0 = \CacheState_0'$.
				Moreover, we have already shown that $\exprEval{e}{ \apply{\buf_0{[0..i-1]}}{a_0} } = \exprEval{e}{ \apply{\buf_0'{[0..i-1]}}{a_0'} }$.
				Therefore, we can apply the \textsc{Execute-Load-Hit} and \textsc{Step} rules to $C_0$ and $C_0'$ as follows:
				\begin{align*}
					\buf_0 &:= \buf_0[0..i-1] \concat \tagged{\pload{x}{e}}{T} \concat \buf_0[i+1 .. |\buf_0|]\\
					\buf_1 &:= \buf_0[0..i-1] \concat \tagged{\passign{x}{m_0(n)}}{T} \concat \buf_0[i+1 .. |\buf_0|]\\
					\tup{m_0,a_0,\buf_0, \CacheState_0, \BpState_0} &\muarchStep{\execute{i}}{} \tup{m_0,a_0,\buf_1, \CacheUpdate(\CacheState_0,n), \BpState_0}\\
					\tup{m_0,a_0,\buf_0, \CacheState_0,  \BpState_0 , \SchedState_0} &\SeqProcMuarchStep{}{} \tup{m_0,a_0,\buf_1, \CacheUpdate(\CacheState_0,n), \BpState_0,\SchedUpdate(\SchedState_0, \BufProject{\buf_0})}\\
					\buf_0' &:= \buf_0'[0..i-1] \concat \tagged{\pload{x}{e}}{T} \concat \buf_0'[i+1 .. |\buf_0|]\\
					\buf_1' &:= \buf_0'[0..i-1] \concat \tagged{\passign{x}{m_0'(n)}}{T} \concat \buf_0'[i+1 .. |\buf_0|]\\
					\tup{m_0',a_0',\buf_0', \CacheState_0', \BpState_0'} &\muarchStep{\execute{i}}{} \tup{m_0',a_0',\buf_1', \CacheUpdate(\CacheState_0',n), \BpState_0'}\\
					\tup{m_0',a_0',\buf_0', \CacheState_0', \BpState_0', \SchedState_0'} &\SeqProcMuarchStep{}{} \tup{m_0',a_0',\buf_1', \CacheUpdate(\CacheState_0',n), \BpState_0',\SchedUpdate(\SchedState_0', \BufProject{\buf_0'})}
				\end{align*}
				We now show that $C_1 = \tup{m_0,a_0,\buf_1, \CacheUpdate(\CacheState_0,n), \BpState_0,\SchedUpdate(\SchedState_0, \BufProject{\buf_0})}$ and $C_1' = \tup{m_0',a_0',\buf_1', \CacheUpdate(\CacheState_0',n), \BpState_0',\SchedUpdate(\SchedState_0', \BufProject{\buf_0'})}$ are indistinguishable, i.e., i.e., $C_1 \sim C_1'$.
	
				For this, we need to show that:
					\begin{description}
						\item[$\apply{\buf_1}{a_0}(\pc) = \apply{\buf_1'}{a_0'}(\pc)$:]
						This immediately follows from (a) and the fact that \textbf{load}s do not modify $\pc$.
			
						\item[$\DeepProject{\buf_1} = \DeepProject{\buf_1'}$:] 
						This immediately follows from (a) and $x \neq \pc$ (from the well-formedness of the buffers).

						\item[$\CacheUpdate(\CacheState_0,n) = \CacheUpdate(\CacheState_0',n)$:]
						This follows from $\CacheState_0 = \CacheState_0'$, which follows from (a).
			
						\item[$\BpState_0= \BpState_0':$] 
						This follows from  (a).
			
						\item[$\SchedUpdate(\SchedState_0, \BufProject{\buf_1}) = \SchedUpdate(\SchedState_0', \BufProject{\buf_1'})$:]
						From (a), we have   $\SchedState_0 = \SchedState_0'$.
						From $\DeepProject{\buf_1} = \DeepProject{\buf_1'}$ and Lemma~\ref{lemma:seq-processor:buffer-projections}, we have $\BufProject{\buf_1} = \BufProject{\buf_1'}$.
						Therefore, $\SchedUpdate(\SchedState_0, \BufProject{\buf_1}) = \SchedUpdate(\SchedState_0', \BufProject{\buf_1'})$.
					\end{description}
					Therefore, $C_1 \sim C_1'$ and (c) holds.

				\item[$\CacheAccess(\CacheState_0, \exprEval{e}{\apply{\buf_0{[0..i-1]}}{a_0}}) = \CacheMiss$:]
				The proof of this case is similar to the one for the $\CacheHit$ case (except that we apply the \textsc{Execute-Load-Miss} rule). 
			\end{description}

			\item[$\pstore{x'}{e'} \in \buf_0{[0..i-1]}$:]
			From (a), we also have that $\pstore{x'}{e'} \in \buf_0'[0..i-1]$.
			Therefore, both computations are stuck and (c) holds.

		\end{description}

		\item[$\elt{\buf_0}{i} =  \tagged{\passign{\pc}{\lbl}}{\lbl_0} \wedge \ell_0 \neq \emptysequence$:]
		From (a), we also have that $\elt{\buf_0'}{i} =  \tagged{\passign{\pc}{\lbl}}{\lbl_0} \wedge \ell_0 \neq \emptysequence$, $i \leq |\buf_0'|$, and $\pbarrier \not\in \buf_0'[0..i-1]$. 
		Observe that $p(\lbl_0) = \pjz{x}{\lbl''}$.

		We now show that $\apply{\buf_0[0..i-1]}{a_0}(x) = \apply{\buf_0'[0..i-1]}{a_0'}(x)$.
		Since $C_0, C_0'$ are reachable configurations, the buffers $\buf_0, \buf_0'$ are well-formed (see Lemma~\ref{lemma:seq-processor:buffers-well-formedness}), and therefore  $\buf_0[0..i-1] \in \prefixes{\buf_0}$ and  $\buf_0'[0..i-1] \in \prefixes{\buf_0'}$.
		From (b), therefore, there are configurations $\sigma_0, \sigma_0', \sigma_1, \sigma_1'$ such that $C_0 \bufEquiv{|\buf_0[0..i-1]|} \sigma_0$, $C_0' \bufEquiv{|\buf_0[0..i-1]|} \sigma_0'$, $\sigma_0 \CtSeqInterfStep{\tau}{} \sigma_1$,  $\sigma_0' \CtSeqInterfStep{\tau'}{} \sigma_1'$, and $\tau = \tau'$. 
		From $C_0 \bufEquiv{|\buf_0[0..i-1]|} \sigma_0$, $C_0' \bufEquiv{|\buf_0[0..i-1]|} \sigma_0'$, and the well-formedness of the buffers, we know that $p(\sigma_0(\pc)) = p(\sigma_0'(\pc)) = \pjz{x}{\lbl''}$.
		From $\CtSeqInterf{\cdot}$, we have that $(\tau = \pcObs{ \lbl'' } \leftrightarrow \sigma_0(x) = 0) \wedge (\tau = \pcObs{ \sigma_0(\pc)+1 } \leftrightarrow \sigma_0(x) \neq 0)$  and $(\tau' = \pcObs{ \lbl'' } \leftrightarrow \sigma_0'(x) = 0) \wedge (\tau' = \pcObs{ \sigma_0(\pc)+1 } \leftrightarrow \sigma_0'(x) \neq 0)$.
		From $\tau=\tau'$, we get that $\sigma_0(x) = \sigma_0'(x)$.
		From this, $C_0 \bufEquiv{|\buf_0[0..i-1]|} \sigma_0$, and $C_0' \bufEquiv{|\buf_0'[0..i-1]|} \sigma_0'$, we finally get $\apply{\buf_0{[0..i-1]}}{a_0}(x) = \apply{\buf_0'{[0..i-1]}}{a_0'}(x)$.

		Given that $\apply{\buf_0[0..i-1]}{a_0}(x) = \apply{\buf_0'[0..i-1]}{a_0'}(x)$, there are two cases:
		\begin{description}
			\item[$( \apply{\buf_0[0..i-1]}{a_0}(x) = 0 \wedge \lbl = \lbl'') \vee (\apply{\buf_0[0..i-1]}{a_0}(x) \in \Val \setminus \{0,\bot\} \wedge \lbl = \ell_0+1)$:]
			From $\apply{\buf_0{[0..i-1]}}{a_0}(x) = \apply{\buf_0'{[0..i-1]}}{a_0'}(x)$ and (a), we also get $( \apply{\buf_0'[0..i-1]}{a_0'}(x) = 0 \wedge \lbl = \lbl'') \vee (\apply{\buf_0'[0..i-1]}{a_0'}(x) \in \Val \setminus \{0,\bot\} \wedge \lbl = \ell_0+1)$.
			Therefore,  we can apply the \textsc{Execute-Branch-Commit} and \textsc{Step} rules to $C_0$ and $C_0'$ as follows:
			\begin{align*}
				\buf_0 &:= \buf_0[0..i-1] \concat \tagged{\passign{\pc}{\ell}}{\ell_0} \concat \buf_0[i+1 .. |\buf_0|]\\
				\buf_1 &:= \buf_0[0..i-1] \concat \tagged{\passign{\pc}{\ell}}{\notags} \concat \buf_0[i+1 .. |\buf_0|]\\
				\tup{m_0,a_0,\buf_0, \CacheState_0, \BpState_0} &\muarchStep{\execute{i}}{} \tup{m_0,a_0,\buf_1, \CacheState_0, \BpState_0}\\
				\tup{m_0,a_0,\buf_0, \CacheState_0,  \BpUpdate(\BpState_0, \ell_0, \ell) , \SchedState_0} &\SeqProcMuarchStep{}{} \tup{m_0,a_0,\buf_1, \CacheState_0, \BpUpdate(\BpState_0, \ell_0, \ell),\SchedUpdate(\SchedState_0, \BufProject{\buf_0})}\\
				\buf_0' &:= \buf_0'[0..i-1] \concat \tagged{\passign{\pc}{\ell}}{\ell_0} \concat \buf_0'[i+1 .. |\buf_0|]\\
				\buf_1' &:= \buf_0'[0..i-1] \concat \tagged{\passign{\pc}{\ell}}{\notags} \concat \buf_0'[i+1 .. |\buf_0|]\\
				\tup{m_0',a_0',\buf_0', \CacheState_0', \BpUpdate(\BpState_0', \ell_0, \ell)} &\muarchStep{\execute{i}}{} \tup{m_0',a_0',\buf_1', \CacheState_0', \BpState_0'}\\
				\tup{m_0',a_0',\buf_0', \CacheState_0', \BpState_0', \SchedState_0'} &\SeqProcMuarchStep{}{} \tup{m_0',a_0',\buf_1', \CacheState_0', \BpUpdate(\BpState_0', \ell_0, \ell),\SchedUpdate(\SchedState_0', \BufProject{\buf_0'})}
			\end{align*}
			We now show that $C_1 = \tup{m_0,a_0,\buf_1, \CacheState_0, \BpUpdate(\BpState_0, \ell_0, \ell),\SchedUpdate(\SchedState_0, \BufProject{\buf_0})}$ and $C_1' = \tup{m_0',a_0',\buf_1', \CacheState_0', \BpUpdate(\BpState_0', \ell_0, \ell),\SchedUpdate(\SchedState_0', \BufProject{\buf_0'})}$ are indistinguishable, i.e., i.e., $C_1 \sim C_1'$.

			For this, we need to show that:
				\begin{description}
					\item[$\apply{\buf_1}{a_0}(\pc) = \apply{\buf_1'}{a_0'}(\pc)$:]
					This immediately follows from (a) and the fact that we set $\pc$ to $\ell$ in both computations.
		
					\item[$\DeepProject{\buf_1} = \DeepProject{\buf_1'}$:] 
					This immediately follows from (a).

					\item[$\CacheState_0 = \CacheState_0'$:]
					This follows from (a).
		
					\item[$\BpUpdate(\BpState_0, \ell_0, \ell) = \BpUpdate(\BpState_0', \ell_0, \ell):$] 
					This follows from $\BpState_0 = \BpState_0'$, which follows from (a).
		
					\item[$\SchedUpdate(\SchedState_0, \BufProject{\buf_1}) = \SchedUpdate(\SchedState_0', \BufProject{\buf_1'})$:]
					From (a), we have   $\SchedState_0 = \SchedState_0'$.
					From $\DeepProject{\buf_1} = \DeepProject{\buf_1'}$ and Lemma~\ref{lemma:seq-processor:buffer-projections}, we have $\BufProject{\buf_1} = \BufProject{\buf_1'}$.
					Therefore, $\SchedUpdate(\SchedState_0, \BufProject{\buf_1}) = \SchedUpdate(\SchedState_0', \BufProject{\buf_1'})$.
				\end{description}
				Therefore, $C_1 \sim C_1'$ and (c) holds.

			\item[$( \apply{\buf_0[0..i-1]}{a_0}(x) = 0 \wedge \lbl \neq \lbl'') \vee (\apply{\buf_0[0..i-1]}{a_0}(x) \in \Val \setminus \{0,\bot\} \wedge \lbl \neq \ell_0+1)$:]
			The proof of this case is similar to the one of the $( \apply{\buf_0[0..i-1]}{a_0}(x) = 0 \wedge \lbl = \lbl'') \vee (\apply{\buf_0[0..i-1]}{a_0}(x) \in \Val \setminus \{0,\bot\} \wedge \lbl = \ell_0+1)$ (except that we apply the \textsc{Execute-Branch-Rollback} rule).
		\end{description}

		\item[$\elt{\buf_0}{i} = \tagged{\passign{x}{e}}{\notags}$:]
		From (a), we also have that $\elt{\buf_0'}{i} =  \tagged{\passign{x}{e'}}{\notags}$, $i \leq |\buf_0'|$, and $\pbarrier \not\in \buf_0'[0..i-1]$.
		There are two cases:
		\begin{description}
			\item[$\exprEval{e}{\apply{\buf_0[0..i-1]}{a_0}} \neq \bot$:]
			From (a), we also have that $\exprEval{e'}{\apply{\buf_0'[0..i-1]}{a_0'}} \neq \bot$ (indeed, if $e \in \Val$, then $e'$ has to be in $\Val$ as well, and if $e \not\in \Val$, then $e = e'$ and $e$'s dependencies must be resolved in both $\buf_0[0..i-1]$ and $\buf_0'[0..i-1]$ from (a)).
			Therefore,  we can apply the \textsc{Execute-Assignment} and \textsc{Step} rules to $C_0$ and $C_0'$ as follows:
			\begin{align*}
				v &:= \exprEval{e}{\apply{\buf_0{[0..i-1]}}{a_0}}\\
				\buf_0 &:= \buf_0[0..i-1] \concat \tagged{\passign{x}{e}}{T} \concat \buf_0[i+1 .. |\buf_0|]\\
				\buf_1 &:= \buf_0[0..i-1] \concat \tagged{\passign{x}{v}}{T} \concat \buf_0[i+1 .. |\buf_0|]\\
				\tup{m_0,a_0,\buf_0, \CacheState_0, \BpState_0} &\muarchStep{\execute{i}}{} \tup{m_0,a_0,\buf_1, \CacheState_0, \BpState_0}\\
				\tup{m_0,a_0,\buf_0, \CacheState_0, \BpState_0, \SchedState_0} &\SeqProcMuarchStep{}{} \tup{m_0,a_0,\buf_1, \CacheState_0, \BpState_0,\SchedUpdate(\SchedState_0, \BufProject{\buf_0})}\\
				v' &:= \exprEval{e'}{\apply{\buf_0'{[0..i-1]}}{a_0'}}\\
				\buf_0' &:= \buf_0'[0..i-1] \concat \tagged{\passign{x}{e'}}{T} \concat \buf_0'[i+1 .. |\buf_0|]\\
				\buf_1' &:= \buf_0'[0..i-1] \concat \tagged{\passign{x}{v'}}{T} \concat \buf_0'[i+1 .. |\buf_0|]\\
				\tup{m_0',a_0',\buf_0', \CacheState_0', \BpState_0'} &\muarchStep{\execute{i}}{} \tup{m_0',a_0',\buf_1', \CacheState_0', \BpState_0'}\\
				\tup{m_0',a_0',\buf_0', \CacheState_0', \BpState_0', \SchedState_0'} &\SeqProcMuarchStep{}{} \tup{m_0',a_0',\buf_1', \CacheState_0', \BpState_0',\SchedUpdate(\SchedState_0', \BufProject{\buf_0'})}
			\end{align*}
			We now show that $C_1 = \tup{m_0,a_0,\buf_1, \CacheState_0, \BpState_0,\SchedUpdate(\SchedState_0, \BufProject{\buf_0})}$ and $C_1' = \tup{m_0',a_0',\buf_1', \CacheState_0', \BpState_0',\SchedUpdate(\SchedState_0', \BufProject{\buf_0'})}$ are indistinguishable, i.e., i.e., $C_1 \sim C_1'$.

			There are two cases:
			\begin{description}
				\item[$x \neq \pc$:]
				For $C_1 \sim C_1'$, we need to show that:
				\begin{description}
					\item[$\apply{\buf_1}{a_0}(\pc) = \apply{\buf_1'}{a_0'}(\pc)$:]
					This immediately follows from (a) and $x \neq \pc$.
		
					\item[$\DeepProject{\buf_1} = \DeepProject{\buf_1'}$:] 
					This immediately follows from (a) and $x \neq \pc$.

					\item[$\CacheState_0 = \CacheState_0'$:]
					This follows from (a).
		
					\item[$\BpState_0 = \BpState_0':$] 
					This follows from (a).
		
					\item[$\SchedUpdate(\SchedState_0, \BufProject{\buf_1}) = \SchedUpdate(\SchedState_0', \BufProject{\buf_1'})$:]
					From (a), we have   $\SchedState_0 = \SchedState_0'$.
					From $\DeepProject{\buf_1} = \DeepProject{\buf_1'}$ and Lemma~\ref{lemma:seq-processor:buffer-projections}, we have $\BufProject{\buf_1} = \BufProject{\buf_1'}$.
					Therefore, $\SchedUpdate(\SchedState_0, \BufProject{\buf_1}) = \SchedUpdate(\SchedState_0', \BufProject{\buf_1'})$.
				\end{description}
				Therefore, $C_1 \sim C_1'$ and (c) holds.

				\item[$x = \pc$:]
				For $C_1 \sim C_1'$, we need to show that:
				\begin{description}
					\item[$\apply{\buf_1}{a_0}(\pc) = \apply{\buf_1'}{a_0'}(\pc)$:]
					For this, we need to show that $v = v'$ (in case there are no later changes to the program counter).
					Since $C_0, C_0'$ are reachable configurations, the buffers $\buf_0, \buf_0'$ are well-formed (see Lemma~\ref{lemma:seq-processor:buffers-well-formedness}), and therefore  $\buf_0[0..i-1] \in \prefixes{\buf_0}$ and  $\buf_0'[0..i-1] \in \prefixes{\buf_0'}$.
					From (b), therefore, there are configurations $\sigma_0, \sigma_0', \sigma_1, \sigma_1'$ such that $C_0 \bufEquiv{|\buf_0[0..i-1]|} \sigma_0$, $C_0' \bufEquiv{|\buf_0[0..i-1]|} \sigma_0'$, $\sigma_0 \CtSeqInterfStep{\tau}{} \sigma_1$,  $\sigma_0' \CtSeqInterfStep{\tau'}{} \sigma_1'$, and $\tau = \tau'$. 
					There are two cases:
					\begin{description}
						\item[$e \in \Val$:] 
						Then,  $e = e'$ follows from (a) and, therefore, we immediately have $v = v'$.

						\item[$e \not\in \Val$:] 
						Then, from (a) we have $e = e'$.
						From $C_0 \bufEquiv{|\buf_0[0..i-1]|} \sigma_0$, $C_0' \bufEquiv{|\buf_0[0..i-1]|} \sigma_0'$, and the well-formedness of the buffers, we know that $p(\sigma_0(\pc)) = p(\sigma_0'(\pc)) = \pjmp{e}$.
						From $\CtSeqInterf{\cdot}$, we have that $\tau = \pcObs{ \exprEval{e}{\sigma_0} }$ and $\tau' = \pcObs{ \exprEval{e}{\sigma_0'}}$.
						From $C_0 \bufEquiv{|\buf_0[0..i-1]|} \sigma_0$ and $\tau = \pcObs{ \exprEval{e}{\sigma_0} }$, we have that $\tau = \pcObs{ \exprEval{e}{\apply{\buf_0{[0..i-1]}}{a_0}}}$.
						Similarly, from $C_0' \bufEquiv{|\buf_0'[0..i-1]|} \sigma_0'$ and $\tau' = \pcObs {\exprEval{e}{\sigma_0'}}$, we have that $\tau' = \pcObs{\exprEval{e}{\apply{\buf_0'{[0..i-1]}}{a_0'}}}$.
						Finally, from $\tau=\tau'$, we get $\exprEval{e}{\apply{\buf_0{[0..i-1]}}{a_0}} = \exprEval{e}{\apply{\buf_0'{[0..i-1]}}{a_0'}}$ and, therefore, $v = v'$.
					\end{description}
						
					\item[$\DeepProject{\buf_1} = \DeepProject{\buf_1'}$:] 
					This immediately follows from (a) and $v = v'$ (shown above).
		
					\item[$\CacheState_0 = \CacheState_0'$:]
					This follows from (a).
		
					\item[$\BpState_0 = \BpState_0':$] 
					This follows from (a).
		
					\item[$\SchedUpdate(\SchedState_0, \BufProject{\buf_1}) = \SchedUpdate(\SchedState_0', \BufProject{\buf_1'})$:]
					From (a), we have   $\SchedState_0 = \SchedState_0'$.
					From $\DeepProject{\buf_1} = \DeepProject{\buf_1'}$ and Lemma~\ref{lemma:seq-processor:buffer-projections}, we have $\BufProject{\buf_1} = \BufProject{\buf_1'}$.
					Therefore, $\SchedUpdate(\SchedState_0, \BufProject{\buf_1}) = \SchedUpdate(\SchedState_0', \BufProject{\buf_1'})$.
				\end{description}
				Therefore, $C_1 \sim C_1'$ and (c) holds.
			\end{description} 

			\item[$\exprEval{e}{\apply{\buf_0[0..i-1]}{a_0}} = \bot$:] 
			From this, it follows that $e \not\in \Val$.
			Therefore, from (a), we have that $e = e'$.
			
			Observe that $\exprEval{e}{\apply{\buf[0..i-1]}{a_0}} = \bot$ implies that one of the dependencies of $e$ is unresolved in $\buf_0[0..i-1]$.
			From this and (a), it follows that one of the dependencies of $e$ is unresolved in $\buf_0'[0..i-1]$.
			Therefore, $\exprEval{e}{\apply{\buf_0'{[0..i-1]}}{a_0'}} = \bot $ holds as well.
			Hence, both configurations are stuck and (c) holds.

		\end{description}

		\item[$\elt{\buf_0}{i} =   \tagged{\pmarkedassign{x}{e}}{\notags}$:]
		The proof of this case is similar to that of $\elt{\buf_0}{i} = \tagged{\passign{x}{e}}{\notags}$ (when $x = \pc$).

		\item[$\elt{\buf_0}{i} =  \tagged{\pstore{x}{e}}{T}$:]
		From (a), we also have that $\elt{\buf_0'}{i} =  \tagged{\pstore{x}{e}}{T}$, $i \leq |\buf_0'|$, and $\pbarrier \not\in \buf_0'[0..i-1]$.
		There are two cases:
		\begin{description}
			\item[$\exprEval{e}{\apply{\buf_0{[0..i-1]}}{a_0}} \neq \bot \wedge \apply{\buf_0{[0..i-1]}}{a_0}(x) \neq \bot$:] 
			From (a), we have that  $\exprEval{e}{\apply{\buf_0'{[0..i-1]}}{a_0'}} \neq \bot \wedge \apply{\buf_0'{[0..i-1]}}{a_0'}(x) \neq \bot$ holds as well.
			Therefore,  we can apply the \textsc{Execute-Store} and \textsc{Step} rules to $C_0$ and $C_0'$ as follows:
			\begin{align*}
				v &:= \apply{\buf_0{[0..i-1]}}{a_0}(x) \\
				n &:= \exprEval{e}{\apply{\buf_0{[0..i-1]}}{a_0}}\\
				\buf_0 &:= \buf_0[0..i-1] \concat \tagged{\pstore{x}{e}}{T} \concat \buf_0[i+1 .. |\buf_0|]\\
				\buf_1 &:= \buf_0[0..i-1] \concat \tagged{\pstore{v}{n}}{T} \concat \buf_0[i+1 .. |\buf_0|]\\
				\tup{m_0,a_0,\buf_0, \CacheState_0, \BpState_0} &\muarchStep{\execute{i}}{} \tup{m_0,a_0,\buf_1, \CacheState_0, \BpState_0}\\
				\tup{m_0,a_0,\buf_0, \CacheState_0, \BpState_0, \SchedState_0} &\SeqProcMuarchStep{}{} \tup{m_0,a_0,\buf_1, \CacheState_0, \BpState_0,\SchedUpdate(\SchedState_0, \BufProject{\buf_0})}\\
				v' &:= \apply{\buf_0'{[0..i-1]}}{a_0'}(x) \\
				n' &:= \exprEval{e}{\apply{\buf_0'{[0..i-1]}}{a_0'}}\\
				\buf_0' &:= \buf_0'[0..i-1] \concat \tagged{\pstore{x}{e}}{T} \concat \buf_0'[i+1 .. |\buf_0|]\\
				\buf_1' &:= \buf_0'[0..i-1] \concat \tagged{\pstore{v'}{n'}}{T} \concat \buf_0'[i+1 .. |\buf_0|]\\
				\tup{m_0',a_0',\buf_0', \CacheState_0', \BpState_0'} &\muarchStep{\execute{i}}{} \tup{m_0',a_0',\buf_1', \CacheState_0', \BpState_0'}\\
				\tup{m_0',a_0',\buf_0', \CacheState_0', \BpState_0', \SchedState_0'} &\SeqProcMuarchStep{}{} \tup{m_0',a_0',\buf_1', \CacheState_0', \BpState_0',\SchedUpdate(\SchedState_0', \BufProject{\buf_0'})}
			\end{align*}
			We now show that $C_1 = \tup{m_0,a_0,\buf_1, \CacheState_0, \BpState_0,\SchedUpdate(\SchedState_0, \BufProject{\buf_0})}$ and $C_1' = \tup{m_0',a_0',\buf_1', \CacheState_0', \BpState_0',\SchedUpdate(\SchedState_0', \BufProject{\buf_0'})}$ are indistinguishable, i.e., i.e., $C_1 \sim C_1'$.
			For this, we need to show that:
			\begin{description}
				\item[$\apply{\buf_1}{a_0}(\pc) = \apply{\buf_1'}{a_0'}(\pc)$:]
				This immediately follows from (a) and the fact that \textbf{store}s do not alter the value of $\pc$.
	
				\item[$\DeepProject{\buf_1} = \DeepProject{\buf_1'}$:] 
				For this, we need to show $n = n'$:
				\begin{description}
					\item[$n = n'$:]
					Since $C_0, C_0'$ are reachable configurations, the buffers $\buf_0, \buf_0'$ are well-formed (see Lemma~\ref{lemma:seq-processor:buffers-well-formedness}), and therefore  $\buf_0[0..i-1] \in \prefixes{\buf_0}$ and  $\buf_0'[0..i-1] \in \prefixes{\buf_0'}$.
					From (b), therefore, there are configurations $\sigma_0, \sigma_0', \sigma_1, \sigma_1'$ such that $C_0 \bufEquiv{|\buf_0[0..i-1]|} \sigma_0$, $C_0' \bufEquiv{|\buf_0[0..i-1]|} \sigma_0'$, $\sigma_0 \CtSeqInterfStep{\tau}{} \sigma_1$,  $\sigma_0' \CtSeqInterfStep{\tau'}{} \sigma_1'$, and $\tau = \tau'$. 
					From $C_0 \bufEquiv{|\buf_0[0..i-1]|} \sigma_0$, $C_0' \bufEquiv{|\buf_0[0..i-1]|} \sigma_0'$, and the well-formedness of the buffers, we know that $p(\sigma_0(\pc)) = p(\sigma_0'(\pc)) = \pstore{x}{e}$.
					From $\CtSeqInterf{\cdot}$, we have that $\tau = \storeObs{ \exprEval{e}{\sigma_0} }$ and $\tau' = \storeObs{ \exprEval{e}{\sigma_0'}}$.
					From $C_0 \bufEquiv{|\buf_0[0..i-1]|} \sigma_0$ and $\tau = \storeObs{ \exprEval{e}{\sigma_0} }$, we have that $\tau = \storeObs{ \exprEval{e}{\apply{\buf_0{[0..i-1]}}{a_0}}}$.
					Similarly, from $C_0' \bufEquiv{|\buf_0'[0..i-1]|} \sigma_0'$ and $\tau' = \storeObs {\exprEval{e}{\sigma_0'}}$, we have that $\tau' = \storeObs{\exprEval{e}{\apply{\buf_0'{[0..i-1]}}{a_0'}}}$.
					Finally, from $\tau=\tau'$, we get $\exprEval{e}{\apply{\buf_0{[0..i-1]}}{a_0}} = \exprEval{e}{\apply{\buf_0'{[0..i-1]}}{a_0'}}$ and, therefore, $n = n'$.

				\end{description}
				From (a), $n = n'$, $\buf_0 = \buf_0[0..i-1] \concat \tagged{\pstore{x}{e}}{T} \concat \buf_0[i+1 .. |\buf_0|]$, $				\buf_1 = \buf_0[0..i-1] \concat \tagged{\pstore{v}{n}}{T} \concat \buf_0[i+1 .. |\buf_0|]$, $\buf_0' = \buf_0'[0..i-1] \concat \tagged{\pstore{x}{e}}{T} \concat \buf_0'[i+1 .. |\buf_0|]$, and $				\buf_1' = \buf_0'[0..i-1] \concat \tagged{\pstore{v'}{n'}}{T} \concat \buf_0'[i+1 .. |\buf_0|]$, we get $\DeepProject{\buf_1} = \DeepProject{\buf_1'}$.

				\item[$\CacheState_0 = \CacheState_0'$:]
				This follows from (a).
	
				\item[$\BpState_0 = \BpState_0':$] 
				This follows from (a).
	
				\item[$\SchedUpdate(\SchedState_0, \BufProject{\buf_1}) = \SchedUpdate(\SchedState_0', \BufProject{\buf_1'})$:]
				From (a), we have   $\SchedState_0 = \SchedState_0'$.
				From $\DeepProject{\buf_1} = \DeepProject{\buf_1'}$ and Lemma~\ref{lemma:seq-processor:buffer-projections}, we have $\BufProject{\buf_1} = \BufProject{\buf_1'}$.
				Therefore, $\SchedUpdate(\SchedState_0, \BufProject{\buf_1}) = \SchedUpdate(\SchedState_0', \BufProject{\buf_1'})$.
			\end{description}
			Therefore, $C_1 \sim C_1'$ and (c) holds.

			\item[$\exprEval{e}{\apply{\buf_0{[0..i-1]}}{a_0}} = \bot \vee \apply{\buf_0{[0..i-1]}}{a_0}(x) = \bot$:] 
			Then, one of the dependencies of $e$ or $x$ are unresolved in $\buf_0[0..i-1]$.
			From this and (a), it follows that one of the dependencies of $e$ or $x$ are unresolved in $\buf_0'[0..i-1]$.
			Therefore, $\exprEval{e}{\apply{\buf_0'{[0..i-1]}}{a_0'}} = \bot \vee \apply{\buf_0'{[0..i-1]}}{a_0'}(x) = \bot$ holds as well.
			Hence, both configurations are stuck and (c) holds.
		\end{description}

		\item[$\elt{\buf_0}{i} =  \tagged{\pskip{}}{\notags}$:]
		The proof of this case is similar to that of $\elt{\buf_0}{i} =  \tagged{\passign{x}{e}}{\notags}$.
		\item[$\elt{\buf_0}{i} =  \tagged{\pbarrier}{T}$:]
		The proof of this case is similar to that of $\elt{\buf_0}{i} =  \tagged{\passign{x}{e}}{\notags}$.
	\end{description}
	Therefore, (c) holds in all cases.

	\item[$i > |\buf_0| \vee \pbarrier \in \buf_0{[0..i-1]}$:]
	From (a), it immediately follows that $i > |\buf_0'| \vee \pbarrier \in \buf_0'[0..i-1]$.
	Therefore, both configurations are stuck and (c) holds.
\end{description}
Therefore, (c) holds in all cases.

\item[$d = \retire{}$:]
	Therefore, we can only apply one of the $\retire{}$ rules depending on the head of the reorder buffer in $\buf_0$.
	There are five cases:
	\begin{description}
		\item[$\buf_0 = \tagged{\pskip}{\notags} \concat \buf_1 $:] 
		From (a), we get that $\DeepProject{\buf_0} = \DeepProject{\buf_0'}$.
		Therefore, we have that $\buf_0' = \tagged{\pskip}{\notags} \concat \buf_1' $ and $\DeepProject{\buf_1} = \DeepProject{\buf_1'}$.
		Therefore, we can apply the \textsc{Retire-Skip} and \textsc{Step} rules to $C_0$ and $C_0'$ as follows:
		\begin{align*}
			\tup{m_0,a_0,\tagged{\pskip}{\notags} \concat \buf_1,\CacheState_0,\BpState_0} &\muarchStep{\retire}{} \tup{m_0, a_0, \buf_1, \CacheState_0,\BpState_0}\\
			\tup{m_0,a_0,\tagged{\pskip}{\notags} \concat \buf_1,\CacheState_0,\BpState_0, \SchedState_0} &\SeqProcMuarchStep{}{} \tup{m_0, a_0, \buf_1, \CacheState_0,\BpState_0, \SchedUpdate(\SchedState_0, \BufProject{\buf_1})}\\
			\tup{m_0',a_0',\tagged{\pskip}{\notags} \concat \buf_1',\CacheState_0',\BpState_0'} &\muarchStep{\retire}{} \tup{m_0', a_0', \buf_1', \CacheState_0',\BpState_0'}\\
			\tup{m_0',a_0',\tagged{\pskip}{\notags} \concat \buf_1',\CacheState_0',\BpState_0', \SchedState_0'} &\SeqProcMuarchStep{}{} \tup{m_0', a_0', \buf_1', \CacheState_0',\BpState_0', \SchedUpdate(\SchedState_0', \BufProject{\buf_1'})}
		\end{align*}
		We now show that $C_1 = \tup{m_0, a_0, \buf_1, \CacheState_0,\BpState_0, \SchedUpdate(\SchedState_0, \BufProject{\buf_1})}$ and $C_1' = \tup{m_0', a_0', \buf_1', \CacheState_0',\BpState_0', \SchedUpdate(\SchedState_0', \BufProject{\buf_1'})}$ are indistinguishable, i.e., $C_1 \sim C_1'$.
		For this, we need to show that:
		\begin{description}
			\item[$\apply{\buf_1}{a_0}(\pc) = \apply{\buf_1'}{a_0'}(\pc)$:]
			From (a), we have $\apply{\tagged{\pskip}{\notags} \concat \buf_1}{a_0}(\pc) = \apply{\tagged{\pskip}{\notags} \concat \buf_1'}{a_0'}(\pc)$.
			From this, we immediately get that $\apply{ \buf_1}{a_0}(\pc) = \apply{ \buf_1'}{a_0'}(\pc)$.

			\item[$\DeepProject{\buf_1} = \DeepProject{\buf_1'}$:] 
			This follows from $\buf_0 = \tagged{\pskip}{\notags} \concat \buf_1 $, $\buf_0' = \tagged{\pskip}{\notags} \concat \buf_1' $, and (a).

			\item[$\CacheState_0 = \CacheState_0'$:]
			This follows from (a).

			\item[$\BpState_0 = \BpState_0':$] 
			This follows from (a).

			\item[$\SchedUpdate(\SchedState_0, \BufProject{\buf_1}) = \SchedUpdate(\SchedState_0', \BufProject{\buf_1'})$:]
			From (a), we have   $\SchedState_0 = \SchedState_0'$.
			From $\DeepProject{\buf_1} = \DeepProject{\buf_1'}$ and Lemma~\ref{lemma:seq-processor:buffer-projections}, we have $\BufProject{\buf_1} = \BufProject{\buf_1'}$.
			Therefore, $\SchedUpdate(\SchedState_0, \BufProject{\buf_1}) = \SchedUpdate(\SchedState_0', \BufProject{\buf_1'})$.
		\end{description}
		Therefore, $C_1 \sim C_1'$ and (c) holds.

		\item[$\buf_0 = \tagged{\pbarrier}{\notags} \concat \buf_1 $:]
		The proof of this case is similar to that of the case $\buf_0 = \tagged{\pskip}{\notags} \concat \buf_1 $.

		\item[$\buf_0 = \tagged{\passign{x}{v}}{\notags} \concat \buf_1 $:] 
		From (a), we get that $\DeepProject{\buf_0} = \DeepProject{\buf_0'}$.
		Therefore, we have that $\buf_0' = \tagged{\passign{x}{v'}}{\notags} \concat \buf_1' $ and $\DeepProject{\buf_1} = \DeepProject{\buf_1'}$.
		From $\DeepProject{\buf_0} = \DeepProject{\buf_0'}$, we also get that $v \in \Val$ iff $v' \in \Val$.
		Observe also that if $v \not\in \Val$ then both computations are stuck and (c) holds (since there is no $C_1$ such that $C_0 \SeqProcMuarchStep{}{} C_1$ and no $C_1'$ such that $C_0' \SeqProcMuarchStep{}{} C_1'$).
		In the following, therefore, we assume that $v,v' \in \Val$.
		Therefore, we can apply the \textsc{Retire-Assignment} and \textsc{Step} rules to $C_0$ and $C_0'$ as follows:
		\begin{align*}
			\tup{m_0,a_0,\tagged{\passign{x}{v}}{\notags} \concat \buf_1,\CacheState_0,\BpState_0} &\muarchStep{\retire}{} \tup{m_0, a_0[x\mapsto v], \buf_1, \CacheState_0,\BpState_0}\\
			\tup{m_0,a_0,\tagged{\passign{x}{v}}{\notags} \concat \buf_1,\CacheState_0,\BpState_0, \SchedState_0} &\SeqProcMuarchStep{}{} \tup{m_0, a_0[x \mapsto v], \buf_1, \CacheState_0,\BpState_0, \SchedUpdate(\SchedState_0, \BufProject{\buf_1})}\\
			\tup{m_0',a_0',\tagged{\passign{x}{v'}}{\notags} \concat \buf_1',\CacheState_0',\BpState_0'} &\muarchStep{\retire}{} \tup{m_0', a_0'[x \mapsto v'], \buf_1', \CacheState_0',\BpState_0'}\\
			\tup{m_0',a_0',\tagged{\passign{x}{v'}}{\notags} \concat \buf_1',\CacheState_0',\BpState_0', \SchedState_0'} &\SeqProcMuarchStep{}{} \tup{m_0', a_0'[x \mapsto v'], \buf_1', \CacheState_0',\BpState_0', \SchedUpdate(\SchedState_0', \BufProject{\buf_1'})}
		\end{align*}
		We now show that $C_1 = \tup{m_0, a_0[x \mapsto v], \buf_1, \CacheState_0,\BpState_0, \SchedUpdate(\SchedState_0, \BufProject{\buf_1})}$ and $C_1' = \tup{m_0', a_0'[x \mapsto v'], \buf_1', \CacheState_0',\BpState_0', \SchedUpdate(\SchedState_0', \BufProject{\buf_1'})}$ are indistinguishable, i.e., $C_1 \sim C_1'$.
		For this, we need to show that:
		\begin{description}
			\item[$\apply{\buf_1}{a_0[x \mapsto v]}(\pc) = \apply{\buf_1'}{a_0'[x\mapsto v']}(\pc)$:]
			From (a), we have $\apply{\tagged{\passign{x}{v}}{\notags} \concat \buf_1}{a_0}(\pc) = \apply{\tagged{\passign{x}{v'}}{\notags} \concat \buf_1'}{a_0'}(\pc)$.
			From this, we immediately get that $\apply{ \buf_1}{a_0[x \mapsto v]}(\pc) = \apply{ \buf_1'}{a_0'[x \mapsto v']}(\pc)$.

			\item[$\DeepProject{\buf_1} = \DeepProject{\buf_1'}$:] 
			This follows from $\buf_0 = \tagged{\passign{x}{v}}{\notags} \concat \buf_1 $, $\buf_0' = \tagged{\passign{x}{v'}}{\notags} \concat \buf_1' $, and (a).

			\item[$\CacheState_0 = \CacheState_0'$:]
			This follows from (a).

			\item[$\BpState_0 = \BpState_0':$] 
			This follows from (a).

			\item[$\SchedUpdate(\SchedState_0, \BufProject{\buf_1}) = \SchedUpdate(\SchedState_0', \BufProject{\buf_1'})$:]
			From (a), we have   $\SchedState_0 = \SchedState_0'$.
			From $\DeepProject{\buf_1} = \DeepProject{\buf_1'}$ and Lemma~\ref{lemma:seq-processor:buffer-projections}, we have $\BufProject{\buf_1} = \BufProject{\buf_1'}$.
			Therefore, $\SchedUpdate(\SchedState_0, \BufProject{\buf_1}) = \SchedUpdate(\SchedState_0', \BufProject{\buf_1'})$.
		\end{description}
		Therefore, $C_1 \sim C_1'$ and (c) holds.

		\item[$\buf_0 = \tagged{\pmarkedassign{x}{v}}{\notags} \concat \buf_1 $:]
		The proof fo this case is similar to that of the case $\buf_0 = \tagged{\passign{x}{v}}{\notags} \concat \buf_1 $.

		\item[$\buf_0 = \tagged{\pstore{v}{n}}{\notags} \concat \buf_1 $:] 
		From (a), we get that $\DeepProject{\buf_0} = \DeepProject{\buf_0'}$.
		Therefore, we have that $\buf_0' = \tagged{\pstore{v'}{n}}{\notags} \concat \buf_1' $ and $\DeepProject{\buf_1} = \DeepProject{\buf_1'}$.
		Observe that (1) $v \in \Val \leftrightarrow v' \in \Val$ from (a), and (2) if $v \not\in \Val$ or $n \not\in \Val$, then both computations are stuck and (c) holds (since there is no $C_1$ such that $C_0 \SeqProcMuarchStep{}{} C_1$ and no $C_1'$ such that $C_0' \SeqProcMuarchStep{}{} C_1'$).
		In the following, therefore, we assume that $v,v',n \in \Val$.
		Therefore, we can apply the \textsc{Retire-Store} and \textsc{Step} rules to $C_0$ and $C_0'$ as follows:
		\begin{align*}
			\tup{m_0,a_0,\tagged{\pstore{v}{n}}{\notags} \concat \buf_1,\CacheState_0,\BpState_0} &\muarchStep{\retire}{} \tup{m_0[n\mapsto v], a_0, \buf_1, \CacheUpdate(\CacheState_0,n),\BpState_0}\\
			\tup{m_0,a_0,\tagged{\pstore{v}{n}}{\notags} \concat \buf_1,\CacheState_0,\BpState_0, \SchedState_0} &\SeqProcMuarchStep{}{} \tup{m_0[n\mapsto v], a_0, \buf_1, \CacheUpdate(\CacheState_0,n),\BpState_0, \SchedUpdate(\SchedState_0, \BufProject{\buf_1})}\\
			\tup{m_0',a_0',\tagged{\pstore{v'}{n}}{\notags} \concat \buf_1',\CacheState_0',\BpState_0'} &\muarchStep{\retire}{} \tup{m_0'[n \mapsto v'], a_0', \buf_1', \CacheUpdate(\CacheState_0',n),\BpState_0'}\\
			\tup{m_0',a_0',\tagged{\pstore{v'}{n}}{\notags} \concat \buf_1',\CacheState_0',\BpState_0', \SchedState_0'} &\SeqProcMuarchStep{}{} \tup{m_0'[n\mapsto v'], a_0', \buf_1', \CacheUpdate(\CacheState_0',n),\BpState_0', \SchedUpdate(\SchedState_0', \BufProject{\buf_1'})}
		\end{align*}
		We now show that $C_1 = \tup{m_0[n \mapsto v], a_0, \buf_1, \CacheUpdate(\CacheState_0,n),\BpState_0, \SchedUpdate(\SchedState_0, \BufProject{\buf_1})}$ and $C_1' =  \tup{m_0'[n\mapsto v'], a_0', \buf_1', \CacheUpdate(\CacheState_0',n),\BpState_0', \SchedUpdate(\SchedState_0', \BufProject{\buf_1'})}$ are indistinguishable, i.e., $C_1 \sim C_1'$.
		For this, we need to show that:
		\begin{description}
			\item[$\apply{\buf_1}{a_0}(\pc) = \apply{\buf_1'}{a_0'}(\pc)$:]
			From (a), we have $\apply{\tagged{\pstore{n}{v}}{\notags} \concat \buf_1}{a_0}(\pc) = \apply{\tagged{\pstore{n}{v}}{\notags} \concat \buf_1'}{a_0'}(\pc)$.
			From this, we immediately get that $\apply{ \buf_1}{a_0}(\pc) = \apply{ \buf_1'}{a_0'}(\pc)$.

			\item[$\DeepProject{\buf_1} = \DeepProject{\buf_1'}$:] 
			This follows from $\buf_0 = \tagged{\pstore{v}{n}}{\notags} \concat \buf_1 $, $\buf_0' = \tagged{\pstore{v'}{n}}{\notags} \concat \buf_1' $, and (a).

			\item[$\CacheUpdate(\CacheState_0,n) = \CacheUpdate(\CacheState_0',n)$:]
			This follows from $\CacheState_0 = \CacheState_0'$, which, in turn, follows from (a).

			\item[$\BpState_0 = \BpState_0':$] 
			This follows from (a).

			\item[$\SchedUpdate(\SchedState_0, \BufProject{\buf_1}) = \SchedUpdate(\SchedState_0', \BufProject{\buf_1'})$:]
			From (a), we have   $\SchedState_0 = \SchedState_0'$.
			From $\DeepProject{\buf_1} = \DeepProject{\buf_1'}$ and Lemma~\ref{lemma:seq-processor:buffer-projections}, we have $\BufProject{\buf_1} = \BufProject{\buf_1'}$.
			Therefore, $\SchedUpdate(\SchedState_0, \BufProject{\buf_1}) = \SchedUpdate(\SchedState_0', \BufProject{\buf_1'})$.
		\end{description}
		Therefore, $C_1 \sim C_1'$ and (c) holds.
	\end{description}
	Therefore, (c) holds for all the cases.

\end{description}
Since (c) holds for all cases, this completes the proof of our lemma.
\end{proof}

\subsection{Main lemma}

\begin{definition}[Indistinguishability of hardware configurations]\label{def:seq-processor:indistinguishability}
	We say that two hardware configurations $\tup{\sigma,\mu} = \tup{m,a,\buf, \CacheState,\BpState, \SchedState}$ and $\tup{\sigma',\mu'} = \tup{m',a',\buf', \CacheState',\BpState', \SchedState'}$ are \emph{indistinguishable}, written $\tup{\sigma,\mu} \approx \tup{\sigma',\mu'}$, iff $\BufProject{\buf } = \BufProject{\buf'}$, $\CacheState = \CacheState'$, $\BpState = \BpState'$, and $\SchedState = \SchedState'$.
\end{definition}

\begin{lemma}[Deep-indistinguishability implies indistinguishability]\label{lemma:seq:deep-indistiguishability-implies-indistinguishability}
Let $C, C'$ be hardware configurations.
If $C \sim C'$, then $C \approx C'$.
\end{lemma}

\begin{proof}
It  follows from Definitions~\ref{def:seq-processor:deep-indistinguishability} and~\ref{def:seq-processor:indistinguishability}.
\end{proof}

\begin{lemma}\label{lemma:seq-processor:main-lemma}
Let $p$ be a well-formed program, $\CacheState_0$ be the initial cache state, $\BpState_0$ be the initial branch predictor state, and $\SchedState_0$ be the initial scheduler state, $\sigma_0 = \tup{m_0,a_0}, \sigma_0' = \tup{m_0',a_0'}$ be initial \archstate{}s, and $C_0 = \tup{m_0,a_0, \emptysequence, \CacheState_0, \BpState_0, \SchedState_0}$ and $C_0' = \tup{m_0',a_0', \emptysequence, \CacheState_0, \BpState_0, \SchedState_0}$ be hardware configurations.
	Furthermore, let $\crun := \sigma_0$ $\CtSeqInterfStep{o_1}{}$ $\sigma_1$ $\CtSeqInterfStep{o_2}{}$ $\ldots$  $\CtSeqInterfStep{o_{n-1}}{}$  $\sigma_n$ and $\crunp:=\sigma_0'$ $\CtSeqInterfStep{o_1'}{}$  $\sigma_1' $ $\CtSeqInterfStep{o_2'}{}$ $\ldots$  $\CtSeqInterfStep{o_{n-1}'}{} $ $\sigma_n$ be two runs for the $\CtSeqInterf{\cdot}$ contract where  $\sigma_n, \sigma_{n}$ are final \archstate{}s.
	If $o_i = o_i'$ for all $0 < i < n$, then there is a $k \in \Nat$ and $C_0, \ldots, C_k, C_0', \ldots, C_k'$ such that $C_0 \SeqProcMuarchStep{}{} C_1 \SeqProcMuarchStep{}{} \ldots \SeqProcMuarchStep{}{} C_k$, $C_0' \SeqProcMuarchStep{}{} C_1' \SeqProcMuarchStep{}{} \ldots \SeqProcMuarchStep{}{} C_k'$, and one of the following conditions hold:
	\begin{compactenum}
		\item $C_0, C_0'$ are initial states, $\forall 0 \leq i \leq k.\ C_{i} \approx C_{i}'$, and $C_, C_k'$ are final states, or
		\item $C_0, C_0'$ are initial states, $\forall 0 \leq i \leq k.\ C_{i} \approx C_{i}'$, and there are no $C_{k+1}$ such that $C_{k+1} \neq C_k \wedge C_k \SeqProcMuarchStep{}{} C_{k+1}$ and no $C_{k+1}'$ such that $C_{k+1}' \neq C_k' \wedge C_k' \SeqProcMuarchStep{}{} C_{k+1}'$.
	\end{compactenum}
\end{lemma}

\begin{proof}
Let $p$ be a well-formed program, $\CacheState_0$ be the initial cache state, $\BpState_0$ be the initial branch predictor state, and $\SchedState_0$ be the initial scheduler state, $\sigma_0 = \tup{m,a}, \sigma_0' = \tup{m',a'}$ be initial \archstate{}s, and $C_0 = \tup{m,a, \emptysequence, \CacheState_0, \BpState_0, \SchedState_0}$ and $C_0' = \tup{m',a', \emptysequence, \CacheState_0, \BpState_0, \SchedState_0}$ be hardware configurations.
	Furthermore, let $\crun := \sigma_0$ $\CtSeqInterfStep{o_1}{}$ $\sigma_1$ $\CtSeqInterfStep{o_2}{}$ $\ldots$  $\CtSeqInterfStep{o_{n-1}}{}$  $\sigma_n$ and $\crunp:=\sigma_0'$ $\CtSeqInterfStep{o_1'}{}$  $\sigma_1' $ $\CtSeqInterfStep{o_2'}{}$ $\ldots$  $\CtSeqInterfStep{o_{n-1}'}{} $ $\sigma_n$ be two runs for the $\CtSeqInterf{\cdot}$ contract where $\sigma_n, \sigma_{n}$ are final \archstate{}s.
	Finally, we assume that $o_i = o_i'$ for all $0 < i < n$.

	Consider the two hardware runs $\hrun, \hrunp$ obtained as follows: we start from $C_0= \tup{m_0,a_0, \emptysequence, \CacheState_0, \BpState_0, \SchedState_0}$ and $C_0'= \tup{m_0',a_0', \emptysequence, \CacheState_0, \BpState_0, \SchedState_0}$, and we apply one step of $\SeqProcMuarchStep{}{}$ until either $\hrun$ or $\hrunp$ reaches a final state or gets stuck.
	That is, for some $k \in \Nat$, we obtain the following runs:
	\begin{align*}
		C_0 &:=  \tup{m_0,a_0, \emptysequence, \CacheState_0, \BpState_0, \SchedState_0}\\
		C_0' &:= \tup{m_0',a_0', \emptysequence, \CacheState_0, \BpState_0, \SchedState_0}\\
		\hrun&:=C_0 \SeqProcMuarchStep{}{} C_1 \SeqProcMuarchStep{}{} \ldots  \SeqProcMuarchStep{}{} C_k \\
		\hrunp&:=C_0' \SeqProcMuarchStep{}{} C_1' \SeqProcMuarchStep{}{} \ldots  \SeqProcMuarchStep{}{} C_k' 
	\end{align*}
	Let $\map{\crun}{\hrun}{\cdot}$ and $\map{\crunp}{\hrunp}{\cdot}$ be the maps constructed according to Definition~\ref{def:seq:mapping}.
	We claim that that for all $0 \leq i \leq k$, $\map{\crun}{\hrun}{i} = \map{\crunp}{\hrunp}{i}$ and $ C_i \sim C_{i}'$ which implies $C_i \approx C_i'$ (see Lemma~\ref{lemma:seq:deep-indistiguishability-implies-indistinguishability}).
	Moreover, if $C_k$ is a stuck configuration (i.e., there is no $C'$ such that $C_k \SeqProcMuarchStep{}{} C'$) then so is $C_k'$ from $C_{k} \sim C_{k}'$.
	This concludes the proof of our lemma.
	
	We now prove our claim.
	That is, we show, by induction on $i$, that for all $0 \leq i \leq k$, $\map{\crun}{\hrun}{i} = \map{\crun'}{\hrun'}{i}$ and $ C_i \sim C_{i}'$.
	\begin{description}
		\item[Base case:]
		Then, $i = 0$.
		Therefore, $\map{\crun}{\hrun}{0} = \map{\crun'}{\hrun'}{0} = \{ 0 \mapsto 0\}$ by construction.
		Moreover, $C_{0} \sim C_{0}'$ immediately follows from $C_0 =  \tup{m_0,a_0, \emptysequence, \CacheState_0, \BpState_0, \SchedState_0}$ and $
		C_0'= \tup{m_0',a_0', \emptysequence, \CacheState_0, \BpState_0, \SchedState_0}$.
	
		\item[Induction step:]
		For the induction step, we assume that our claim holds for all $i' < i$, and we show that it holds for $i$ as well.
		From the induction hypothesis, we get that $\map{\crun}{\hrun}{i-1} = \map{\crunp}{\hrunp}{i-1}$ and $C_{i-1} \sim C_{i-1}'$.
		In the following, we denote $\map{\crun}{\hrun}{i-1} = \map{\crunp}{\hrunp}{i-1}$ as (IH.1) and $C_{i-1} \sim C_{i-1}'$ as (IH.2). 
		Moreover, from Lemma~\ref{lemma:seq-processor:mapping-is-correct} applied to $\map{\crun}{\hrun}{i-1} = \map{\crunp}{\hrunp}{i-1}$, we have that for all $\buf \in \prefixes{\buf_{i-1}}$ and all $\buf' \in \prefixes{\buf_{i-1}'}$, $C_{i-1} \bufEquiv{i} \crun( \map{\crun}{\hrun}{i-1}(|\buf|) ) $ and $C_{i-1}' \bufEquiv{i} \crunp( \map{\crunp}{\hrunp}{i-1}(|\buf'|) ) $.
		Let $\buf \in \prefixes{\buf_{i-1}}$ and all $\buf' \in \prefixes{\buf_{i-1}'}$ be arbitrary buffers such that $|\buf| = |\buf'|$.
		From (IH.1) and $|\buf| = |\buf'|$, we get that $\map{\crun}{\hrun}{i-1}(|\buf|)= \map{\crunp}{\hrunp}{i-1}(|\buf'|)$.
		Therefore,  $\crun( \map{\crun}{\hrun}{i-1}(|\buf|) )$ and $\crunp(\map{\crunp}{\hrunp}{i-1}(|\buf'|))$ are two configurations $\sigma_j, \sigma_j'$ for some $j = \map{\crun}{\hrun}{i-1}(|\buf|)$.
		From $\crun, \crunp$ having pairwise the same observations, we have $\sigma_{j} \CtSeqInterfStep{o_{j+1}}{} \sigma_{j+1}$, $\sigma_{j}' \CtSeqInterfStep{o_{j+1}'}{} \sigma_{j+1}'$, and $o_{j+1} = o_{j+1}'$ (note that we can always make a step in $\sigma_{j},\sigma_{j}'$ since if $j = k$ we can make silent steps thanks to the \textsc{Terminate} rule).
		Since $\buf,\buf'$ have been selected arbitrarily, we know that for all $\buf \in \prefixes{\buf_{i-1}}$ and all $\buf' \in \prefixes{\buf_{i-1}'}$ such that $|\buf| = |\buf'|$, we have $\sigma_{j} \CtSeqInterfStep{o_{j+1}}{} \sigma_{j+1}$, $\sigma_{j}' \CtSeqInterfStep{o_{j+1}'}{} \sigma_{j+1}'$, and $o_{j+1} = o_{j+1}'$ where $j = \map{\crun}{\hrun}{i-1}(|\buf|)$.
		From this and (IH.2), we can apply Lemma~\ref{lemma:seq-processor:trace-equiv-implies-stepwise-indistinguishability} to $C_{i-1}$ and $C_{i-1}'$.
		As a result, we obtain that either $C_{i} \sim C_{i}$ or that both computations are stuck.
		Moreover, $\map{\crun}{\hrun}{i} = \map{\crunp}{\hrunp}{i}$ follows from the fact that the maps at step $i$ are derived from the maps at step $i-1$, which are equivalent thanks to (IH.1), depending on the content of the projection of the buffer (which is the same thanks to $C_{i} \sim C_{i}$).
		This completes the proof of the induction step.
	\end{description}
\end{proof}

\newpage
\section{Eager load-delay processor $\muarchStyle{loadDelay}$ -- Proofs of Theorems~\ref{theorem:hni:load-delayone} and~\ref{theorem:hni:load-delaytwo}}\label{appendix:proofs:load-delay}

In this section, we prove the security guarantees of the eager load-delay processor $\muarchStyle{loadDelay}$ given in \S\ref{sec:countermeasures:load-delay}.

In the following, we assume given an arbitrary cache $C$, branch predictor $Bp$, and scheduler $S$, whose initial states are, respectively, $\CacheState_0$, $\BpState_0$, and $\SchedState_0$.
Hence, our proof holds for arbitrary caches $C$, branch predictors $Bp$, and scheduler $S$.

For simplicity, here we report the rules defining the $\muarchStyle{loadDelay}$ processor:
\begin{mathpar}
    \inferrule[Step-Others]
    {
        \tup{m,a,\buf, \CacheState,\BpState} \muarchStep{d}{} \tup{m',a',\buf', \CacheState',\BpState'}	\\
        d = \SchedNext(\SchedState)\\
        \SchedState' = \SchedUpdate(\SchedState,\BufProject{\buf'})\\\\
        { d \in \{\fetch{}, \retire\} \vee (d = \execute{i} \wedge \elt{\buf}{i} \neq \pload{x}{e}) }
    }
    {
        \tup{m,a,\buf, \CacheState,\BpState, \SchedState} \LoadDelayMuarchStep{}{} \tup{m',a',\buf', \CacheState',\BpState', \SchedState'}
    }
\end{mathpar}
\begin{mathpar}
    \inferrule[Step-Eager-Delay]
    {
        \tup{m,a,\buf, \CacheState,\BpState} \muarchStep{d}{} \tup{m',a',\buf', \CacheState',\BpState'}	\\
        d = \SchedNext(\SchedState)\\
        \SchedState' = \SchedUpdate(\SchedState,\BufProject{\buf'})\\
        {d = \execute{i}}\\
        {\elt{\buf}{i} = \pload{x}{e}}\\
        {\forall \tagged{\passign{\pc}{\lbl}}{\lbl'} \in \buf[0..i-1].\  \lbl' = \notags}
    }
    {
        \tup{m,a,\buf, \CacheState,\BpState, \SchedState } \LoadDelayMuarchStep{}{} \tup{m',a',\buf', \CacheState',\BpState', \SchedState'}
    }
\end{mathpar}

The semantics $\LoadDelayMuarchSem{p}$ for a program $p$ is defined as follows:
$\LoadDelayMuarchSem{p}(\tup{m,a})$ is $\tup{\emptysequence, \CacheState_0,\BpState_0, \SchedState_0}$ $ \cdot \tup{\BufProject{\buf_1}, \CacheState_1,\BpState_1, \SchedState_1}$ $  \cdot \tup{\BufProject{\buf_2}, \CacheState_2,\BpState_2, \SchedState_2}$ $  \cdot \ldots \cdot $ $ \tup{\BufProject{\buf_n}, \CacheState_n,\BpState_n, \SchedState_n}$ where $	\tup{m,a,\emptysequence, \CacheState_0,\BpState_0, \SchedState_0 }$ $ \LoadDelayMuarchStep{}{}$ $ \tup{m_1,a_1,\buf_1, \CacheState_1,\BpState_1, \SchedState_1} $ $ \LoadDelayMuarchStep{}{} $ $\tup{m_2,a_2,\buf_2, \CacheState_2,\BpState_2, \SchedState_2} $ $\LoadDelayMuarchStep{}{}$ $\ldots$ $\LoadDelayMuarchStep{}{} $ $ \tup{m_n,a_n,\buf_n, \CacheState_n,\BpState_n, \SchedState_n}$ is the complete hardware run obtained starting from $\tup{m,a}$ and terminating in $\tup{m_n,a_n,\buf_n, \CacheState_n,\BpState_n, \SchedState_n}$, which is a final hardware state. Otherwise, $\LoadDelayMuarchSem{p}(\tup{m,a})$ is undefined.

\subsection{Proof of Theorem~\ref{theorem:hni:load-delayone}}

\textbf{ASSUMPTION: } In the following, we assume that $\wInterf > \wMuarch + 1$.

\loadDelayOne*

\begin{proof}
    The proof of the theorem is similar to the one of Theorem~\ref{theorem:hni:all} (see Appendix~\ref{appendix:proofs:general}):
    \begin{compactitem}
        \item For the mapping lemma, we use the mapping function from Definition~\ref{def:vanilla:mapping} and the proof proceeds as in Lemma~\ref{lemma:vanilla:mapping-is-correct}.
        \item For the indistinguishability lemma, we use Lemma~\ref{lemma:load-delay:spec-ct:trace-equiv-implies-stepwise-indistinguishability}. 
        \item The main lemma proceeds in a similar way to Lemma~\ref{lemma:vanilla:main-lemma}.
    \end{compactitem}
\end{proof}

\subsubsection{Indistinguishability lemma}

\begin{definition}[Deep-indistinguishability of hardware configurations]\label{def:load-delay:spec-ct:deep-indistinguishability}
	We say that two hardware configurations $\tup{\sigma,\mu} = \tup{m,a,\buf, \CacheState,\BpState, \SchedState}$ and $\tup{\sigma',\mu'} = \tup{m',a',\buf', \CacheState',\BpState', \SchedState'}$ are \emph{deep-indistinguishable}, written $\tup{\sigma,\mu} \sim \tup{\sigma',\mu'}$, iff
	\begin{inparaenum}[(a)]
		\item $\apply{\buf}{a}(\pc) = \apply{\buf'}{a'}(\pc)$,
		\item $\DeepProject{\buf} = \DeepProject{\buf'}$,
		\item $\CacheState = \CacheState'$,
		\item $\BpState = \BpState'$, and
		\item $\SchedState = \SchedState'$,
		\end{inparaenum}
		where $\DeepProject{\buf}$ is inductively defined as follows:
		\begin{align*}
			\DeepProject{\tagged{\pskip{}}{T}}_k  &:=  \tagged{\pskip{}}{T}\\
			\DeepProject{\tagged{\pbarrier{}}{T}}_k  &:=  \tagged{\pbarrier{}}{T}\\
			\DeepProject{\tagged{\passign{x}{e}}{T}}_k  &:=
			{
				\begin{cases}
				\tagged{\passign{x}{\resolved}}{T} & \text{if}\ e \in \Val \wedge x \neq \pc\\
				\tagged{\passign{x}{e}}{T} & \text{if}\ e \not\in \Val \wedge x \neq \pc\\
				\tagged{\passign{x}{e}}{T} & \text{if}\ x = \pc
				\end{cases}
			}\\
			\DeepProject{\tagged{\pload{x}{e}}{T}}_k  &:=
			{
				\tagged{\pload{x}{e}}{T}
			}\\
			\DeepProject{\tagged{\pstore{x}{e}}{T}}_k  &:=
			{
				\begin{cases}
					\tagged{\pstore{x}{e}}{T} & \text{if}\ x \in \Var\\
					\tagged{\pstore{\resolved}{\resolved}}{T} & \text{if}\ x \in \Val \wedge k = 1 \\
					\tagged{\pstore{\resolved}{e}}{T} & \text{if}\ x \in \Val \wedge k = 0 
				\end{cases}
			}\\
			\DeepProject{\emptysequence} 	&:= \emptysequence\\
			\DeepProject{\tagged{i}{T} \concat \buf} &:= 
			{
				\begin{cases}
					\DeepProject{\tagged{i}{T}}_0 \concat \DeepProject{\buf}_1 & \text{if } i = \passign{\pc}{\ell'} \wedge T \neq \notags  \\
					\DeepProject{\tagged{i}{T}}_0 \concat \DeepProject{\buf} & \text{otherwise} 			
				\end{cases}
			}
		\end{align*}
\end{definition}

\begin{lemma}[Observation equivalence preserves deep-indistinguishability]\label{lemma:load-delay:spec-ct:trace-equiv-implies-stepwise-indistinguishability}
    Let $p$ be a well-formed program and $C_0 = \tup{m_0,a_0,\buf_0,\CacheState_0, \BpState_0, \SchedState_0}$, $C_0' = \tup{m_0',a_0',\buf_0',\CacheState_0', \BpState_0', \SchedState_0'}$ be reachable hardware configurations.
    If
    \begin{inparaenum}[(a)]
        \item $C_0 \sim C_0'$, and
        \item for all $\buf \in \prefixes{\buf_0}$, $\buf' \in \prefixes{\buf_0'}$ such that $|\buf| = |\buf'|$, 
        there are $s_0, s_0', s_1, s_1', \tau, \tau'$ such that $s_0 \CtPcSpecInterfStep{\tau}{} s_1$, $s_0' \CtPcSpecInterfStep{\tau'}{} s_1'$, $\headWindow{s_0} > 0$, $\headWindow{s_0'} >0$,  $\tau = \tau'$, $C_0 \bufEquiv{|\buf|} s_0$, and $C_0' \bufEquiv{|\buf'|} s_0'$,
    \end{inparaenum}
    then either there are $C_1, C_1'$ such that $C_0 \LoadDelayMuarchStep{}{} C_1$, $C_0' \LoadDelayMuarchStep{}{} C_1'$, and $C_1 \sim C_1'$ or there is no $C_1$ such that $C_0 \LoadDelayMuarchStep{}{} C_1$ and no $C_1'$ such that $C_0' \LoadDelayMuarchStep{}{} C_1'$.
\end{lemma}

\begin{proof}
    Let $p$ be a well-formed program, $C_0 = \tup{m_0,a_0,\buf_0,\CacheState_0, \BpState_0, \SchedState_0}$, and $C_0' = \tup{m_0',a_0',\buf_0',\CacheState_0', \BpState_0', \SchedState_0'}$.
    Moreover, we assume that conditions (a) and (b) holds. 
    In the following, we denote by (c) the post-condition ``either there are $C_1, C_1'$ such that $C_0 \LoadDelayMuarchStep{}{} C_1$, $C_0' \LoadDelayMuarchStep{}{} C_1'$, and $C_1 \sim C_1'$ or there is no $C_1$ such that $C_0 \LoadDelayMuarchStep{}{} C_1$ and no $C_1'$ such that $C_0' \LoadDelayMuarchStep{}{} C_1'$.''
    
    From (a), it follows that $\SchedState_0 = \SchedState_0$.
    Therefore, the directive obtained from the scheduler is the same in both cases, i.e., $\SchedNext(\SchedState_0) = \SchedNext(\SchedState_0')$.
    We proceed by case distinction on the directive $d = \SchedNext(\SchedState_0)$:
    \begin{description}
    \item[$d = \fetch{}$:]
    The proof of this case is similar to the corresponding case of Lemma~\ref{lemma:vanilla:trace-equiv-implies-stepwise-indistinguishability}.
            
    \item[$d = \execute{i}$:]
    Therefore, we can only apply one of the $\execute{}$ rules.
    There are two cases:
    \begin{description}
        \item[$i \leq |\buf_0| \wedge \pbarrier \not\in {\buf_0[0..i-1]}$:]
        There are several cases depending on the $i$-th command in the reorder buffer:
        \begin{description}
            \item[$\elt{\buf_0}{i} = \tagged{\pload{x}{e}}{T}$:]
            From (a), we also have that $\elt{\buf_0'}{i} =  \tagged{\pload{x}{e}}{T}$, $i \leq |\buf_0'|$, and $\pbarrier \not\in \buf_0'[0..i-1]$. 
            There are two cases:
            \begin{description}
                \item[$\pstore{x'}{e'} \not\in \buf_0{[0..i-1]} \wedge  \forall \tagged{\passign{\pc}{\lbl}}{\lbl'} \in \buf_0{[0..i-1]}.\  \lbl' = \notags$:]
                We now show that $\exprEval{e}{\apply{\buf_0[0..i-1]}{a_0}} = \exprEval{e}{\apply{\buf_0'[0..i-1]}{a_0'}}$.
                Since $C_0, C_0'$ are reachable configurations, the buffers $\buf_0, \buf_0'$ are well-formed, and therefore  $\buf_0[0..i-1] \in \prefixes{\buf_0}$ and  $\buf_0'[0..i-1] \in \prefixes{\buf_0'}$.
                From (b), therefore, there are configurations $s_0, s_0', s_1, s_1'$ such that $C_0 \bufEquiv{|\buf_0[0..i-1]|} s_0$, $C_0' \bufEquiv{|\buf_0[0..i-1]|} s_0'$, $s_0 \CtPcSpecInterfStep{\tau}{} s_1$,  $s_0' \CtPcSpecInterfStep{\tau'}{} s_1'$, $\headWindow{s_0} > 0$, $\headWindow{s_0'} >0$, and $\tau = \tau'$. 
                From (a), $C_0 \bufEquiv{|\buf_0[0..i-1]|} s_0$, $C_0' \bufEquiv{|\buf_0[0..i-1]|} s_0'$, and the well-formedness of the buffers, we know that $p(\headConf{s_0}(\pc)) = p(\headConf{s_0'}(\pc)) = \pload{x}{e}$.
                Moreover, from $\forall \tagged{\passign{\pc}{\lbl}}{\lbl'} \in \buf_0{[0..i-1]}.\  \lbl' = \notags$, we know that $\nrMispred{\tup{m_{0},a_{0}}{\buf_0{[0..i-1]}}} = 0$.
                From this and $C_0 \bufEquiv{|\buf_0[0..i-1]|} s_0$, we have that $\headWindow{s_0} = \infty$.
                Similarly, we obtain that $\headWindow{s_0'} = \infty$.
                From $\CtPcSpecInterf{\cdot}$, $\headWindow{s_0} = \infty$, $\headWindow{s_0'} =\infty$, and $p(\headConf{s_0}(\pc)) = p(\headConf{s_0'}(\pc)) = \pload{x}{e}$, we have that $\tau = \loadObs{ \exprEval{e}{\headConf{s_0} } }$  and $\tau' = \loadObs{ \exprEval{e}{ \headConf{s_0'} } }$ (because $s_0 \CtPcSpecInterfStep{\tau}{} s_1$ and $s_0' \CtPcSpecInterfStep{\tau'}{} s_1'$ have been obtained by applying the \textsc{Step} rule of $\CtPcSpecInterfStep{}{}$ and the \textsc{Load} rule of $\CtPcSpecInterfStep{}{}$).
                From $\tau=\tau'$, we get that $\exprEval{e}{\headConf{s_0}} = \exprEval{e}{\headConf{s_0'}}$.
                From this, $C_0 \bufEquiv{|\buf_0[0..i-1]|} s_0$, and $C_0' \bufEquiv{|\buf_0'[0..i-1]|} s_0'$, we finally get $\exprEval{e}{\apply{\buf_0{[0..i-1]}}{a_0}} = \exprEval{e}{\apply{\buf_0'{[0..i-1]}}{a_0'}(x)}$.
    
                Let $n = \exprEval{e}{\apply{\buf_0{[0..i-1]}}{a_0}} = \exprEval{e}{\apply{\buf_0'{[0..i-1]}}{a_0'}(x)}$.
                There are two cases:
                \begin{description}
                    \item[$\CacheAccess(\CacheState_0, \exprEval{e}{ \apply{\buf_0{[0..i-1]}}{a_0} }) = \CacheHit$:]
                    From (a), we have that $\CacheState_0 = \CacheState_0'$.
                    Moreover, we have already shown that $\exprEval{e}{ \apply{\buf_0{[0..i-1]}}{a_0} } = \exprEval{e}{ \apply{\buf_0'{[0..i-1]}}{a_0'} }$.
                    Therefore, we can apply the \textsc{Execute-Load-Hit} and \textsc{Step} rules to $C_0$ and $C_0'$ as follows:
                    \begin{align*}
                        \buf_0 &:= \buf_0[0..i-1] \concat \tagged{\pload{x}{e}}{T} \concat \buf_0[i+1 .. |\buf_0|]\\
                        \buf_1 &:= \buf_0[0..i-1] \concat \tagged{\passign{x}{m_0(n)}}{T} \concat \buf_0[i+1 .. |\buf_0|]\\
                        \tup{m_0,a_0,\buf_0, \CacheState_0, \BpState_0} &\muarchStep{\execute{i}}{} \tup{m_0,a_0,\buf_1, \CacheUpdate(\CacheState_0,n), \BpState_0}\\
                        \tup{m_0,a_0,\buf_0, \CacheState_0,  \BpState_0 , \SchedState_0} & \LoadDelayMuarchStep{}{} \tup{m_0,a_0,\buf_1, \CacheUpdate(\CacheState_0,n), \BpState_0,\SchedUpdate(\SchedState_0, \BufProject{\buf_0})}\\
                        \buf_0' &:= \buf_0'[0..i-1] \concat \tagged{\pload{x}{e}}{T} \concat \buf_0'[i+1 .. |\buf_0|]\\
                        \buf_1' &:= \buf_0'[0..i-1] \concat \tagged{\passign{x}{m_0'(n)}}{T} \concat \buf_0'[i+1 .. |\buf_0|]\\
                        \tup{m_0',a_0',\buf_0', \CacheState_0', \BpState_0'} &\muarchStep{\execute{i}}{} \tup{m_0',a_0',\buf_1', \CacheUpdate(\CacheState_0',n), \BpState_0'}\\
                        \tup{m_0',a_0',\buf_0', \CacheState_0', \BpState_0', \SchedState_0'} &\LoadDelayMuarchStep{}{} \tup{m_0',a_0',\buf_1', \CacheUpdate(\CacheState_0',n), \BpState_0',\SchedUpdate(\SchedState_0', \BufProject{\buf_0'})}
                    \end{align*}
                    We now show that $C_1 = \tup{m_0,a_0,\buf_1, \CacheUpdate(\CacheState_0,n), \BpState_0,\SchedUpdate(\SchedState_0, \BufProject{\buf_0})}$ and $C_1' = \tup{m_0',a_0',\buf_1', \CacheUpdate(\CacheState_0',n), \BpState_0',\SchedUpdate(\SchedState_0', \BufProject{\buf_0'})}$ are indistinguishable, i.e., i.e., $C_1 \sim C_1'$.
        
                    For this, we need to show that:
                        \begin{description}
                            \item[$\apply{\buf_1}{a_0}(\pc) = \apply{\buf_1'}{a_0'}(\pc)$:]
                            This immediately follows from (a) and the fact that \textbf{load}s do not modify $\pc$.
                
                            \item[$\DeepProject{\buf_1} = \DeepProject{\buf_1'}$:] 
                            This immediately follows from (a) and $x \neq \pc$ (from the well-formedness of the buffers).

                            \item[$\CacheUpdate(\CacheState_0,n) = \CacheUpdate(\CacheState_0',n)$:]
                            This follows from $\CacheState_0 = \CacheState_0'$, which follows from (a).
                
                            \item[$\BpState_0= \BpState_0':$] 
                            This follows from  (a).
                
                            \item[$\SchedUpdate(\SchedState_0, \BufProject{\buf_1}) = \SchedUpdate(\SchedState_0', \BufProject{\buf_1'})$:]
                            From (a), we have   $\SchedState_0 = \SchedState_0'$.
                            From $\DeepProject{\buf_1} = \DeepProject{\buf_1'}$ and Lemma~\ref{lemma:vanilla:buffer-projections}, we have $\BufProject{\buf_1} = \BufProject{\buf_1'}$.
                            Therefore, $\SchedUpdate(\SchedState_0, \BufProject{\buf_1}) = \SchedUpdate(\SchedState_0', \BufProject{\buf_1'})$.
                        \end{description}
                        Therefore, $C_1 \sim C_1'$ and (c) holds.
    
                    \item[$\CacheAccess(\CacheState_0, \exprEval{e}{\apply{\buf_0{[0..i-1]}}{a_0}}) = \CacheMiss$:]
                    The proof of this case is similar to the one for the $\CacheHit$ case (except that we apply the \textsc{Execute-Load-Miss} rule). 
                \end{description}
    
                \item[$\pstore{x'}{e'} \in \buf_0{[0..i-1]} \vee \forall \tagged{\passign{\pc}{\lbl}}{\lbl'} \in \buf_0{[0..i-1]}.\  \lbl' = \notags$:]
                From (a), we also have that $\pstore{x'}{e'} \in \buf_0'[0..i-1] \vee \forall \tagged{\passign{\pc}{\lbl}}{\lbl'} \in \buf_0'{[0..i-1]}.\  \lbl' = \notags$.
                Therefore, both computations are stuck and (c) holds.
    
            \end{description}
    
            \item[$\elt{\buf_0}{i} =  \tagged{\passign{\pc}{\lbl}}{\lbl_0} \wedge \ell_0 \neq \emptysequence$:]
            From (a), we also have that $\elt{\buf_0'}{i} =  \tagged{\passign{\pc}{\lbl}}{\lbl_0} \wedge \ell_0 \neq \emptysequence$, $i \leq |\buf_0'|$, and $\pbarrier \not\in \buf_0'[0..i-1]$. 
            Observe that $p(\lbl_0) = \pjz{x}{\lbl''}$.
    
            We now show that $\apply{\buf_0[0..i-1]}{a_0}(x) = \apply{\buf_0'[0..i-1]}{a_0'}(x)$.
            Since $C_0, C_0'$ are reachable configurations, the buffers $\buf_0, \buf_0'$ are well-formed (see Lemma~\ref{lemma:vanilla:buffers-well-formedness}), and therefore  $\buf_0[0..i-1] \in \prefixes{\buf_0}$ and  $\buf_0'[0..i-1] \in \prefixes{\buf_0'}$.
            From (b), therefore, there are configurations $s_0, s_0', s_1, s_1'$ such that $C_0 \bufEquiv{|\buf_0[0..i-1]|} s_0$, $C_0' \bufEquiv{|\buf_0[0..i-1]|} s_0'$, $s_0 \CtPcSpecInterfStep{\tau}{} s_1$,  $s_0' \CtPcSpecInterfStep{\tau'}{} s_1'$, $\headWindow{s_0} > 0$, $\headWindow{s_0'} >0$, and $\tau = \tau'$. 
            From (a), $C_0 \bufEquiv{|\buf_0[0..i-1]|} s_0$, $C_0' \bufEquiv{|\buf_0[0..i-1]|} s_0'$, and the well-formedness of the buffers, we know that $p(\headConf{s_0}(\pc)) = p(\headConf{s_0'}(\pc)) = \pjz{x}{\lbl''}$.
            From $\CtPcSpecInterf{\cdot}$, $\headWindow{s_0} > 0$, $\headWindow{s_0'} >0$, and  $p(\headConf{s_0}(\pc)) = p(\headConf{s_0'}(\pc)) = \pjz{x}{\lbl''}$, we have that $(\tau = \pcObs{ \lbl'' } \leftrightarrow \headConf{s_0}(x) = 0) \wedge (\tau = \pcObs{ \sigma_0(\pc)+1 } \leftrightarrow \headConf{s_0}(x) \neq 0)$  and $(\tau' = \pcObs{ \lbl'' } \leftrightarrow \headConf{s_0'}(x) = 0) \wedge (\tau' = \pcObs{ \sigma_0(\pc)+1 } \leftrightarrow \headConf{s_0'}(x) \neq 0)$ (because $s_0 \CtPcSpecInterfStep{\tau}{} s_1$ and  $s_0' \CtPcSpecInterfStep{\tau'}{} s_1'$ have been obtained by applying the \textsc{Branch} rule of $\CtPcSpecInterfStep{}{}$).
            From $\tau=\tau'$, we get that $\headConf{s_0}(x) = \headConf{s_0'}(x)$.
            From this, $C_0 \bufEquiv{|\buf_0[0..i-1]|} s_0$, and $C_0' \bufEquiv{|\buf_0'[0..i-1]|} s_0'$, we finally get $\apply{\buf_0{[0..i-1]}}{a_0}(x) = \apply{\buf_0'{[0..i-1]}}{a_0'}(x)$.
    
            Given that $\apply{\buf_0[0..i-1]}{a_0}(x) = \apply{\buf_0'[0..i-1]}{a_0'}(x)$, there are two cases:
            \begin{description}
                \item[$( \apply{\buf_0[0..i-1]}{a_0}(x) = 0 \wedge \lbl = \lbl'') \vee (\apply{\buf_0[0..i-1]}{a_0}(x) \in \Val \setminus \{0,\bot\} \wedge \lbl = \ell_0+1)$:]
                From $\apply{\buf_0{[0..i-1]}}{a_0}(x) = \apply{\buf_0'{[0..i-1]}}{a_0'}(x)$ and (a), we also get $( \apply{\buf_0'[0..i-1]}{a_0'}(x) = 0 \wedge \lbl = \lbl'') \vee (\apply{\buf_0'[0..i-1]}{a_0'}(x) \in \Val \setminus \{0,\bot\} \wedge \lbl = \ell_0+1)$.
                Therefore,  we can apply the \textsc{Execute-Branch-Commit} and \textsc{Step} rules to $C_0$ and $C_0'$ as follows:
                \begin{align*}
                    \buf_0 &:= \buf_0[0..i-1] \concat \tagged{\passign{\pc}{\ell}}{\ell_0} \concat \buf_0[i+1 .. |\buf_0|]\\
                    \buf_1 &:= \buf_0[0..i-1] \concat \tagged{\passign{\pc}{\ell}}{\notags} \concat \buf_0[i+1 .. |\buf_0|]\\
                    \tup{m_0,a_0,\buf_0, \CacheState_0, \BpState_0} &\muarchStep{\execute{i}}{} \tup{m_0,a_0,\buf_1, \CacheState_0, \BpState_0}\\
                    \tup{m_0,a_0,\buf_0, \CacheState_0,  \BpUpdate(\BpState_0, \ell_0, \ell) , \SchedState_0} &\LoadDelayMuarchStep{}{} \tup{m_0,a_0,\buf_1, \CacheState_0, \BpUpdate(\BpState_0, \ell_0, \ell),\SchedUpdate(\SchedState_0, \BufProject{\buf_0})}\\
                    \buf_0' &:= \buf_0'[0..i-1] \concat \tagged{\passign{\pc}{\ell}}{\ell_0} \concat \buf_0'[i+1 .. |\buf_0|]\\
                    \buf_1' &:= \buf_0'[0..i-1] \concat \tagged{\passign{\pc}{\ell}}{\notags} \concat \buf_0'[i+1 .. |\buf_0|]\\
                    \tup{m_0',a_0',\buf_0', \CacheState_0', \BpUpdate(\BpState_0', \ell_0, \ell)} &\muarchStep{\execute{i}}{} \tup{m_0',a_0',\buf_1', \CacheState_0', \BpState_0'}\\
                    \tup{m_0',a_0',\buf_0', \CacheState_0', \BpState_0', \SchedState_0'} &\LoadDelayMuarchStep{}{} \tup{m_0',a_0',\buf_1', \CacheState_0', \BpUpdate(\BpState_0', \ell_0, \ell),\SchedUpdate(\SchedState_0', \BufProject{\buf_0'})}
                \end{align*}
                We now show that $C_1 = \tup{m_0,a_0,\buf_1, \CacheState_0, \BpUpdate(\BpState_0, \ell_0, \ell),\SchedUpdate(\SchedState_0, \BufProject{\buf_0})}$ and $C_1' = \tup{m_0',a_0',\buf_1', \CacheState_0', \BpUpdate(\BpState_0', \ell_0, \ell),\SchedUpdate(\SchedState_0', \BufProject{\buf_0'})}$ are indistinguishable, i.e., i.e., $C_1 \sim C_1'$.
    
                For this, we need to show that:
                    \begin{description}
                        \item[$\apply{\buf_1}{a_0}(\pc) = \apply{\buf_1'}{a_0'}(\pc)$:]
                        This immediately follows from (a) and the fact that we set $\pc$ to $\ell$ in both computations.
            
                        \item[$\DeepProject{\buf_1} = \DeepProject{\buf_1'}$:] 
                        The only interesting case is if (1) there is $k > i$ such that $\elt{\buf_0}{k} = \pstore{v}{n}$ and (2) there is no other unresolved branch instruction in $\buf_0[0..k]$ except for $\elt{\buf_0}{i}$.
                        In all other cases, $\DeepProject{\buf_1} = \DeepProject{\buf_1'}$ directly follows from (a).
                        Assume that (1) and (2) hold.
                        From (a), we get also that (1') $\elt{\buf_0'}{k} = \pstore{v'}{n'}$ and (2') there is no other unresolved branch instruction in $\buf_0'[0..k]$ except for $\elt{\buf_0'}{i}$.
                        So, resolving the branch instruction (and, therefore, removing the tag), results in the $n,n'$ being visible in $\DeepProject{\buf_1} = \DeepProject{\buf_1'}$.
                        From (2),(2') and the fact that in both executions the branch instruction has been committed, it immediately follows that the \textbf{store} operations happen outside of a rolled-back speculative transaction.
                        As a result, we can use the observations produced by $\CtPcSpecInterf{\cdot}$ to prove that $n = n'$ and, therefore, that $\DeepProject{\buf_1} = \DeepProject{\buf_1'}$.
                        We refer to the case $\elt{\buf_0}{i} = \tagged{\pstore{x}{e}}{T}$ below for an in-depth proof of $n = n'$ using $\CtPcSpecInterf{\cdot}$'s observations.

                        \item[$\CacheState_0 = \CacheState_0'$:]
                        This follows from (a).
            
                        \item[$\BpUpdate(\BpState_0, \ell_0, \ell) = \BpUpdate(\BpState_0', \ell_0, \ell):$] 
                        This follows from $\BpState_0 = \BpState_0'$, which follows from (a).
            
                        \item[$\SchedUpdate(\SchedState_0, \BufProject{\buf_1}) = \SchedUpdate(\SchedState_0', \BufProject{\buf_1'})$:]
                        From (a), we have   $\SchedState_0 = \SchedState_0'$.
                        From $\DeepProject{\buf_1} = \DeepProject{\buf_1'}$ and Lemma~\ref{lemma:vanilla:buffer-projections}, we have $\BufProject{\buf_1} = \BufProject{\buf_1'}$.
                        Therefore, $\SchedUpdate(\SchedState_0, \BufProject{\buf_1}) = \SchedUpdate(\SchedState_0', \BufProject{\buf_1'})$.
                    \end{description}
                    Therefore, $C_1 \sim C_1'$ and (c) holds.
    
                \item[$( \apply{\buf_0[0..i-1]}{a_0}(x) = 0 \wedge \lbl \neq \lbl'') \vee (\apply{\buf_0[0..i-1]}{a_0}(x) \in \Val \setminus \{0,\bot\} \wedge \lbl \neq \ell_0+1)$:]
                The proof of this case is similar to the one of the $( \apply{\buf_0[0..i-1]}{a_0}(x) = 0 \wedge \lbl = \lbl'') \vee (\apply{\buf_0[0..i-1]}{a_0}(x) \in \Val \setminus \{0,\bot\} \wedge \lbl = \ell_0+1)$ (except that we apply the \textsc{Execute-Branch-Rollback} rule).
            \end{description}
    
            \item[$\elt{\buf_0}{i} = \tagged{\passign{x}{e}}{\notags}$:]
            From (a), we also have that $\elt{\buf_0'}{i} =  \tagged{\passign{x}{e'}}{\notags}$, $i \leq |\buf_0'|$, and $\pbarrier \not\in \buf_0'[0..i-1]$.
            There are two cases:
            \begin{description}
                \item[$\exprEval{e}{\apply{\buf_0[0..i-1]}{a_0}} \neq \bot$:]
                From (a), we also have that $\exprEval{e'}{\apply{\buf_0'[0..i-1]}{a_0'}} \neq \bot$ (indeed, if $e \in \Val$, then $e'$ has to be in $\Val$ as well, and if $e \not\in \Val$, then $e = e'$ and $e$'s dependencies must be resolved in both $\buf_0[0..i-1]$ and $\buf_0'[0..i-1]$ from (a)).
                Therefore,  we can apply the \textsc{Execute-Assignment} and \textsc{Step} rules to $C_0$ and $C_0'$ as follows:
                \begin{align*}
                    v &:= \exprEval{e}{\apply{\buf_0{[0..i-1]}}{a_0}}\\
                    \buf_0 &:= \buf_0[0..i-1] \concat \tagged{\passign{x}{e}}{T} \concat \buf_0[i+1 .. |\buf_0|]\\
                    \buf_1 &:= \buf_0[0..i-1] \concat \tagged{\passign{x}{v}}{T} \concat \buf_0[i+1 .. |\buf_0|]\\
                    \tup{m_0,a_0,\buf_0, \CacheState_0, \BpState_0} &\muarchStep{\execute{i}}{} \tup{m_0,a_0,\buf_1, \CacheState_0, \BpState_0}\\
                    \tup{m_0,a_0,\buf_0, \CacheState_0, \BpState_0, \SchedState_0} &\LoadDelayMuarchStep{}{} \tup{m_0,a_0,\buf_1, \CacheState_0, \BpState_0,\SchedUpdate(\SchedState_0, \BufProject{\buf_0})}\\
                    v' &:= \exprEval{e'}{\apply{\buf_0'{[0..i-1]}}{a_0'}}\\
                    \buf_0' &:= \buf_0'[0..i-1] \concat \tagged{\passign{x}{e'}}{T} \concat \buf_0'[i+1 .. |\buf_0|]\\
                    \buf_1' &:= \buf_0'[0..i-1] \concat \tagged{\passign{x}{v'}}{T} \concat \buf_0'[i+1 .. |\buf_0|]\\
                    \tup{m_0',a_0',\buf_0', \CacheState_0', \BpState_0'} &\muarchStep{\execute{i}}{} \tup{m_0',a_0',\buf_1', \CacheState_0', \BpState_0'}\\
                    \tup{m_0',a_0',\buf_0', \CacheState_0', \BpState_0', \SchedState_0'} &\LoadDelayMuarchStep{}{} \tup{m_0',a_0',\buf_1', \CacheState_0', \BpState_0',\SchedUpdate(\SchedState_0', \BufProject{\buf_0'})}
                \end{align*}
                We now show that $C_1 = \tup{m_0,a_0,\buf_1, \CacheState_0, \BpState_0,\SchedUpdate(\SchedState_0, \BufProject{\buf_0})}$ and $C_1' = \tup{m_0',a_0',\buf_1', \CacheState_0', \BpState_0',\SchedUpdate(\SchedState_0', \BufProject{\buf_0'})}$ are indistinguishable, i.e., i.e., $C_1 \sim C_1'$.
    
                There are two cases:
                \begin{description}
                    \item[$x \neq \pc$:]
                    For $C_1 \sim C_1'$, we need to show that:
                    \begin{description}
                        \item[$\apply{\buf_1}{a_0}(\pc) = \apply{\buf_1'}{a_0'}(\pc)$:]
                        This immediately follows from (a) and $x \neq \pc$.
            
                        \item[$\DeepProject{\buf_1} = \DeepProject{\buf_1'}$:] 
                        This immediately follows from (a) and $x \neq \pc$.

                        \item[$\CacheState_0 = \CacheState_0'$:]
                        This follows from (a).
            
                        \item[$\BpState_0 = \BpState_0':$] 
                        This follows from (a).
            
                        \item[$\SchedUpdate(\SchedState_0, \BufProject{\buf_1}) = \SchedUpdate(\SchedState_0', \BufProject{\buf_1'})$:]
                        From (a), we have   $\SchedState_0 = \SchedState_0'$.
                        From $\DeepProject{\buf_1} = \DeepProject{\buf_1'}$ and Lemma~\ref{lemma:vanilla:buffer-projections}, we have $\BufProject{\buf_1} = \BufProject{\buf_1'}$.
                        Therefore, $\SchedUpdate(\SchedState_0, \BufProject{\buf_1}) = \SchedUpdate(\SchedState_0', \BufProject{\buf_1'})$.
                    \end{description}
                    Therefore, $C_1 \sim C_1'$ and (c) holds.
    
                    \item[$x = \pc$:]
                    For $C_1 \sim C_1'$, we need to show that:
                    \begin{description}
                        \item[$\apply{\buf_1}{a_0}(\pc) = \apply{\buf_1'}{a_0'}(\pc)$:]
                        For this, we need to show that $v = v'$ (in case there are no later changes to the program counter).
                        Since $C_0, C_0'$ are reachable configurations, the buffers $\buf_0, \buf_0'$ are well-formed, and therefore  $\buf_0[0..i-1] \in \prefixes{\buf_0}$ and  $\buf_0'[0..i-1] \in \prefixes{\buf_0'}$.
                        From (b), therefore, there are configurations $s_0, s_0', s_1, s_1'$ such that $C_0 \bufEquiv{|\buf_0[0..i-1]|} s_0$, $C_0' \bufEquiv{|\buf_0[0..i-1]|} s_0'$, $s_0 \CtPcSpecInterfStep{\tau}{} s_1$,  $s_0' \CtPcSpecInterfStep{\tau'}{} s_1'$, $\headWindow{s_0} > 0$, $\headWindow{s_0'} >0$, and $\tau = \tau'$. 
                        There are two cases:
                        \begin{description}
                            \item[$e \in \Val$:] 
                            Then,  $e = e'$ follows from (a) and, therefore, we immediately have $v = v'$.
    
                            \item[$e \not\in \Val$:] 
                            Then, from (a) we have $e = e'$.
                            From (a), $C_0 \bufEquiv{|\buf_0[0..i-1]|} s_0$, $C_0' \bufEquiv{|\buf_0[0..i-1]|} s_0'$, and the well-formedness of the buffers, we know that $p(\headConf{s_0}(\pc)) = p(\headConf{s_0'}(\pc)) = \pjmp{e}$.
                            From $\CtPcSpecInterf{\cdot}$, $\headWindow{s_0} > 0$, and $\headWindow{s_0'} >0$, we have that $\tau = \pcObs{ \exprEval{e}{\headConf{s_0}} }$ and $\tau' = \pcObs{ \exprEval{e}{\headConf{s_0'}}}$ (because $s_0 \CtPcSpecInterfStep{\tau}{} s_1$ and $s_0' \CtPcSpecInterfStep{\tau'}{} s_1'$ have been obtained by applying the \textsc{Step} rule of $\CtPcSpecInterfStep{}{}$ and the \textsc{Jump} rule of $\CtSeqInterfStep{}{}$).
                            From $C_0 \bufEquiv{|\buf_0[0..i-1]|} s_0$ and $\tau = \pcObs{ \exprEval{e}{ \headConf{s_0} } }$, we have that $\tau = \pcObs{ \exprEval{e}{\apply{\buf_0{[0..i-1]}}{a_0}}}$.
                            Similarly, from $C_0' \bufEquiv{|\buf_0'[0..i-1]|} s_0'$ and $\tau' = \pcObs {\exprEval{e}{\headConf{s_0'}}}$, we have that $\tau' = \pcObs{\exprEval{e}{\apply{\buf_0'{[0..i-1]}}{a_0'}}}$.
                            Finally, from $\tau=\tau'$, we get $\exprEval{e}{\apply{\buf_0{[0..i-1]}}{a_0}} = \exprEval{e}{\apply{\buf_0'{[0..i-1]}}{a_0'}}$ and, therefore, $v = v'$.
                        \end{description}
                            
                        \item[$\DeepProject{\buf_1} = \DeepProject{\buf_1'}$:] 
                        This immediately follows from (a) and $v = v'$ (shown above).
            
                        \item[$\CacheState_0 = \CacheState_0'$:]
                        This follows from (a).
            
                        \item[$\BpState_0 = \BpState_0':$] 
                        This follows from (a).
            
                        \item[$\SchedUpdate(\SchedState_0, \BufProject{\buf_1}) = \SchedUpdate(\SchedState_0', \BufProject{\buf_1'})$:]
                        From (a), we have   $\SchedState_0 = \SchedState_0'$.
                        From $\DeepProject{\buf_1} = \DeepProject{\buf_1'}$ and Lemma~\ref{lemma:seq-processor:buffer-projections}, we have $\BufProject{\buf_1} = \BufProject{\buf_1'}$.
                        Therefore, $\SchedUpdate(\SchedState_0, \BufProject{\buf_1}) = \SchedUpdate(\SchedState_0', \BufProject{\buf_1'})$.
                    \end{description}
                    Therefore, $C_1 \sim C_1'$ and (c) holds.
                \end{description} 
    
                \item[$\exprEval{e}{\apply{\buf_0[0..i-1]}{a_0}} = \bot$:] 
                From this, it follows that $e \not\in \Val$.
                Therefore, from (a), we have that $e = e'$.
                
                Observe that $\exprEval{e}{\apply{\buf[0..i-1]}{a_0}} = \bot$ implies that one of the dependencies of $e$ is unresolved in $\buf_0[0..i-1]$.
                From this and (a), it follows that one of the dependencies of $e$ is unresolved in $\buf_0'[0..i-1]$.
                Therefore, $\exprEval{e}{\apply{\buf_0'{[0..i-1]}}{a_0'}} = \bot $ holds as well.
                Hence, both configurations are stuck and (c) holds.
    
            \end{description}

            \item[$\elt{\buf_0}{i} =   \tagged{\pmarkedassign{x}{e}}{\notags}$:]
            The proof of this case is similar to that of $\elt{\buf_0}{i} = \tagged{\passign{x}{e}}{\notags}$ (when $x = \pc$).
    
            \item[$\elt{\buf_0}{i} =  \tagged{\pstore{x}{e}}{T}$:]
            From (a), we also have that $\elt{\buf_0'}{i} =  \tagged{\pstore{x}{e}}{T}$, $i \leq |\buf_0'|$, and $\pbarrier \not\in \buf_0'[0..i-1]$.
            There are two cases:
            \begin{description}
                \item[$\exprEval{e}{\apply{\buf_0{[0..i-1]}}{a_0}} \neq \bot \wedge \apply{\buf_0{[0..i-1]}}{a_0}(x) \neq \bot$:] 
                From (a), we have that  $\exprEval{e}{\apply{\buf_0'{[0..i-1]}}{a_0'}} \neq \bot \wedge \apply{\buf_0'{[0..i-1]}}{a_0'}(x) \neq \bot$ holds as well.
                Therefore,  we can apply the \textsc{Execute-Store} and \textsc{Step} rules to $C_0$ and $C_0'$ as follows:
                \begin{align*}
                    v &:= \apply{\buf_0{[0..i-1]}}{a_0}(x) \\
                    n &:= \exprEval{e}{\apply{\buf_0{[0..i-1]}}{a_0}}\\
                    \buf_0 &:= \buf_0[0..i-1] \concat \tagged{\pstore{x}{e}}{T} \concat \buf_0[i+1 .. |\buf_0|]\\
                    \buf_1 &:= \buf_0[0..i-1] \concat \tagged{\pstore{v}{n}}{T} \concat \buf_0[i+1 .. |\buf_0|]\\
                    \tup{m_0,a_0,\buf_0, \CacheState_0, \BpState_0} &\muarchStep{\execute{i}}{} \tup{m_0,a_0,\buf_1, \CacheState_0, \BpState_0}\\
                    \tup{m_0,a_0,\buf_0, \CacheState_0, \BpState_0, \SchedState_0} &\muarchStep{}{} \tup{m_0,a_0,\buf_1, \CacheState_0, \BpState_0,\SchedUpdate(\SchedState_0, \BufProject{\buf_0})}\\
                    v' &:= \apply{\buf_0'{[0..i-1]}}{a_0'}(x) \\
                    n' &:= \exprEval{e}{\apply{\buf_0'{[0..i-1]}}{a_0'}}\\
                    \buf_0' &:= \buf_0'[0..i-1] \concat \tagged{\pstore{x}{e}}{T} \concat \buf_0'[i+1 .. |\buf_0|]\\
                    \buf_1' &:= \buf_0'[0..i-1] \concat \tagged{\pstore{v'}{n'}}{T} \concat \buf_0'[i+1 .. |\buf_0|]\\
                    \tup{m_0',a_0',\buf_0', \CacheState_0', \BpState_0'} &\muarchStep{\execute{i}}{} \tup{m_0',a_0',\buf_1', \CacheState_0', \BpState_0'}\\
                    \tup{m_0',a_0',\buf_0', \CacheState_0', \BpState_0', \SchedState_0'} &\muarchStep{}{} \tup{m_0',a_0',\buf_1', \CacheState_0', \BpState_0',\SchedUpdate(\SchedState_0', \BufProject{\buf_0'})}
                \end{align*}
                We now show that $C_1 = \tup{m_0,a_0,\buf_1, \CacheState_0, \BpState_0,\SchedUpdate(\SchedState_0, \BufProject{\buf_0})}$ and $C_1' = \tup{m_0',a_0',\buf_1', \CacheState_0', \BpState_0',\SchedUpdate(\SchedState_0', \BufProject{\buf_0'})}$ are indistinguishable, i.e., i.e., $C_1 \sim C_1'$.
                For this, we need to show that:
                \begin{description}
                    \item[$\apply{\buf_1}{a_0}(\pc) = \apply{\buf_1'}{a_0'}(\pc)$:]
                    This immediately follows from (a) and the fact that \textbf{store}s do not alter the value of $\pc$.
        
                    \item[$\DeepProject{\buf_1} = \DeepProject{\buf_1'}$:] 
                    There are two cases:
                    \begin{description}
                    	\item[$\forall \tagged{\passign{\pc}{\ell}}{T} \in \buf_0{[0..i-1]}.\ T = \notags$:]
                    	From (a), we also get $\forall \tagged{\passign{\pc}{\ell}}{T} \in \buf_0'{[0..i-1]}.\ T = \notags$.
                    	In this case, we need to show $n = n'$:
                    \begin{description}
                        \item[$n = n'$:]
                        Since $C_0, C_0'$ are reachable configurations, the buffers $\buf_0, \buf_0'$ are well-formed, and therefore  $\buf_0[0..i-1] \in \prefixes{\buf_0}$ and  $\buf_0'[0..i-1] \in \prefixes{\buf_0'}$.
                        From (b), therefore, there are configurations $s_0, s_0', s_1, s_1'$ such that $C_0 \bufEquiv{|\buf_0[0..i-1]|} s_0$, $C_0' \bufEquiv{|\buf_0[0..i-1]|} s_0'$, $s_0 \CtPcSpecInterfStep{\tau}{} s_1$,  $s_0' \CtPcSpecInterfStep{\tau'}{} s_1'$, $\headWindow{s_0} > 0$, $\headWindow{s_0'} >0$, and $\tau = \tau'$.  
                        From (a), $C_0 \bufEquiv{|\buf_0[0..i-1]|} s_0$, $C_0' \bufEquiv{|\buf_0[0..i-1]|} s_0'$, and the well-formedness of the buffers, we know that $p(\headConf{s_0}(\pc)) = p(\headConf{s_0'}(\pc)) = \pstore{x}{e}$.
                        Moreover, from $\forall \tagged{\passign{\pc}{\lbl}}{\lbl'} \in \buf{[0..i-1]}.\  \lbl' = \notags$, we know that $\nrMispred{\tup{m_{0},a_{0}}{\buf_0{[0..i-1]}}} = 0$.
		                From this and $C_0 \bufEquiv{|\buf_0[0..i-1]|} s_0$, we have that $\headWindow{s_0} = \infty$.
		                Similarly, we obtain that $\headWindow{s_0'} = \infty$.
                        From $\CtSpecInterf{\cdot}$, $\headWindow{s_0} =\infty 0$, $\headWindow{s_0'} =\infty$, we have that $\tau = \storeObs{ \exprEval{e}{\sigma_0} }$ and $\tau' = \storeObs{ \exprEval{e}{\sigma_0'}}$ (because $s_0 \CtPcSpecInterfStep{\tau}{} s_1$ and  $s_0' \CtPcSpecInterfStep{\tau'}{} s_1'$ have been obtained by applying the \textsc{Step} rule of $\CtPcSpecInterfStep{}{}$ and the \textsc{Store} rule of $\CtSeqInterfStep{}{}$).
                        From $C_0 \bufEquiv{|\buf_0[0..i-1]|} s_0$ and $\tau = \storeObs{ \exprEval{e}{\headConf{s_0}} }$, we have that $\tau = \storeObs{ \exprEval{e}{\apply{\buf_0{[0..i-1]}}{a_0}}}$.
                        Similarly, from $C_0' \bufEquiv{|\buf_0'[0..i-1]|} \sigma_0'$ and $\tau' = \storeObs {\exprEval{e}{\headConf{s_0'}}}$, we have that $\tau' = \storeObs{\exprEval{e}{\apply{\buf_0'{[0..i-1]}}{a_0'}}}$.
                        Finally, from $\tau=\tau'$, we get $\exprEval{e}{\apply{\buf_0{[0..i-1]}}{a_0}} = \exprEval{e}{\apply{\buf_0'{[0..i-1]}}{a_0'}}$ and, therefore, $n = n'$.
    
                    \end{description}
                    From (a), $n = n'$, $\buf_0 = \buf_0[0..i-1] \concat \tagged{\pstore{x}{e}}{T} \concat \buf_0[i+1 .. |\buf_0|]$, $				\buf_1 = \buf_0[0..i-1] \concat \tagged{\pstore{v}{n}}{T} \concat \buf_0[i+1 .. |\buf_0|]$, $\buf_0' = \buf_0'[0..i-1] \concat \tagged{\pstore{x}{e}}{T} \concat \buf_0'[i+1 .. |\buf_0|]$, and $				\buf_1' = \buf_0'[0..i-1] \concat \tagged{\pstore{v'}{n'}}{T} \concat \buf_0'[i+1 .. |\buf_0|]$, we get $\DeepProject{\buf_1} = \DeepProject{\buf_1'}$.

                    	\item[$\exists \tagged{\passign{\pc}{\ell}}{T} \in \buf_0{[0..i-1]}.\ T \neq \notags$:] 
                    	From (a), we also get $\exists \tagged{\passign{\pc}{\ell}}{T} \in \buf_0'{[0..i-1]}.\ T \neq \notags$.
                    	Therefore, from (a), $\buf_0 = \buf_0[0..i-1] \concat \tagged{\pstore{x}{e}}{T} \concat \buf_0[i+1 .. |\buf_0|]$, $				\buf_1 = \buf_0[0..i-1] \concat \tagged{\pstore{v}{n}}{T} \concat \buf_0[i+1 .. |\buf_0|]$, $\buf_0' = \buf_0'[0..i-1] \concat \tagged{\pstore{x}{e}}{T} \concat \buf_0'[i+1 .. |\buf_0|]$, and $				\buf_1' = \buf_0'[0..i-1] \concat \tagged{\pstore{v'}{n'}}{T} \concat \buf_0'[i+1 .. |\buf_0|]$, we get $\DeepProject{\buf_1} = \DeepProject{\buf_1'}$. 
                    \end{description}

                    \item[$\CacheState_0 = \CacheState_0'$:]
                    This follows from (a).
        
                    \item[$\BpState_0 = \BpState_0':$] 
                    This follows from (a).
        
                    \item[$\SchedUpdate(\SchedState_0, \BufProject{\buf_1}) = \SchedUpdate(\SchedState_0', \BufProject{\buf_1'})$:]
                    From (a), we have   $\SchedState_0 = \SchedState_0'$.
                    From $\DeepProject{\buf_1} = \DeepProject{\buf_1'}$ and Lemma~\ref{lemma:vanilla:buffer-projections}, we have $\BufProject{\buf_1} = \BufProject{\buf_1'}$.
                    Therefore, $\SchedUpdate(\SchedState_0, \BufProject{\buf_1}) = \SchedUpdate(\SchedState_0', \BufProject{\buf_1'})$.
                \end{description}
                Therefore, $C_1 \sim C_1'$ and (c) holds.
    
                \item[$\exprEval{e}{\apply{\buf_0{[0..i-1]}}{a_0}} = \bot \vee \apply{\buf_0{[0..i-1]}}{a_0}(x) = \bot$:] 
                Then, one of the dependencies of $e$ or $x$ are unresolved in $\buf_0[0..i-1]$.
                From this and (a), it follows that one of the dependencies of $e$ or $x$ are unresolved in $\buf_0'[0..i-1]$.
                Therefore, $\exprEval{e}{\apply{\buf_0'{[0..i-1]}}{a_0'}} = \bot \vee \apply{\buf_0'{[0..i-1]}}{a_0'}(x) = \bot$ holds as well.
                Hence, both configurations are stuck and (c) holds.
            \end{description}
    
            \item[$\elt{\buf_0}{i} =  \tagged{\pskip{}}{\notags}$:]
            The proof of this case is similar to that of $\elt{\buf_0}{i} =  \tagged{\passign{x}{e}}{\notags}$.
            \item[$\elt{\buf_0}{i} =  \tagged{\pbarrier}{T}$:]
            The proof of this case is similar to that of $\elt{\buf_0}{i} =  \tagged{\passign{x}{e}}{\notags}$.
        \end{description}
        Therefore, (c) holds in all cases.

        \item[$i > |\buf_0| \vee \pbarrier \in \buf_0{[0..i-1]}$:]
        From (a), it immediately follows that $i > |\buf_0'| \vee \pbarrier \in \buf_0'[0..i-1]$.
        Therefore, both configurations are stuck and (c) holds.
    \end{description}
    Therefore, (c) holds in all cases.

    \item[$d = \retire{}$:]
    Therefore, we can only apply one of the $\retire{}$ rules depending on the head of the reorder buffer in $\buf_0$.
        There are five cases:
        \begin{description}
            \item[$\buf_0 = \tagged{\pskip}{\notags} \concat \buf_1 $:] 
        	The proof of this case is similar to the corresponding case of Lemma~\ref{lemma:vanilla:trace-equiv-implies-stepwise-indistinguishability}.    
        
            \item[$\buf_0 = \tagged{\pbarrier}{\notags} \concat \buf_1 $:]
            The proof of this case is similar to that of the case $\buf_0 = \tagged{\pskip}{\notags} \concat \buf_1 $.
    
            \item[$\buf_0 = \tagged{\passign{x}{v}}{\notags} \concat \buf_1 $:] 
        	The proof of this case is similar to the corresponding case of Lemma~\ref{lemma:vanilla:trace-equiv-implies-stepwise-indistinguishability}.
    
            \item[$\buf_0 = \tagged{\pmarkedassign{x}{v}}{\notags} \concat \buf_1 $:]
            The proof fo this case is similar to that of the case $\buf_0 = \tagged{\passign{x}{v}}{\notags} \concat \buf_1 $.
    
            \item[$\buf_0 = \tagged{\pstore{v}{n}}{\notags} \concat \buf_1 $:] 
            From (a), we get that $\DeepProject{\buf_0} = \DeepProject{\buf_0'}$.
            Therefore, we have that $\buf_0' = \tagged{\pstore{v'}{n}}{\notags} \concat \buf_1' $ and $\DeepProject{\buf_1} = \DeepProject{\buf_1'}$.
            Observe that (1) $v \in \Val \leftrightarrow v' \in \Val$ from (a), (2) if $v \not\in \Val$ or $n \not\in \Val$, then both computations are stuck and (c) holds (since there is no $C_1$ such that $C_0 \muarchStep{}{} C_1$ and no $C_1'$ such that $C_0' \muarchStep{}{} C_1'$), and (3) if $v \in \Val$, then $v = v'$ (this follows from (a) and the fact that there are no tagged commands before the store under retirement).
            In the following, therefore, we assume that $v,v',n \in \Val$.
            Therefore, we can apply the \textsc{Retire-Store} and \textsc{Step} rules to $C_0$ and $C_0'$ as follows:
            \begin{align*}
                \tup{m_0,a_0,\tagged{\pstore{v}{n}}{\notags} \concat \buf_1,\CacheState_0,\BpState_0} &\muarchStep{\retire}{} \tup{m_0[n\mapsto v], a_0, \buf_1, \CacheUpdate(\CacheState_0,n),\BpState_0}\\
                \tup{m_0,a_0,\tagged{\pstore{v}{n}}{\notags} \concat \buf_1,\CacheState_0,\BpState_0, \SchedState_0} &\muarchStep{}{} \tup{m_0[n\mapsto v], a_0, \buf_1, \CacheUpdate(\CacheState_0,n),\BpState_0, \SchedUpdate(\SchedState_0, \BufProject{\buf_1})}\\
                \tup{m_0',a_0',\tagged{\pstore{v'}{n}}{\notags} \concat \buf_1',\CacheState_0',\BpState_0'} &\muarchStep{\retire}{} \tup{m_0'[n \mapsto v'], a_0', \buf_1', \CacheUpdate(\CacheState_0',n),\BpState_0'}\\
                \tup{m_0',a_0',\tagged{\pstore{v'}{n}}{\notags} \concat \buf_1',\CacheState_0',\BpState_0', \SchedState_0'} &\muarchStep{}{} \tup{m_0'[n\mapsto v'], a_0', \buf_1', \CacheUpdate(\CacheState_0',n),\BpState_0', \SchedUpdate(\SchedState_0', \BufProject{\buf_1'})}
            \end{align*}
            We now show that $C_1 = \tup{m_0[n \mapsto v], a_0, \buf_1, \CacheUpdate(\CacheState_0,n),\BpState_0, \SchedUpdate(\SchedState_0, \BufProject{\buf_1})}$ and $C_1' =  \tup{m_0'[n\mapsto v'], a_0', \buf_1', \CacheUpdate(\CacheState_0',n),\BpState_0', \SchedUpdate(\SchedState_0', \BufProject{\buf_1'})}$ are indistinguishable, i.e., $C_1 \sim C_1'$.
            For this, we need to show that:
            \begin{description}
                \item[$\apply{\buf_1}{a_0}(\pc) = \apply{\buf_1'}{a_0'}(\pc)$:]
                From (a), we have $\apply{\tagged{\pstore{n}{v}}{\notags} \concat \buf_1}{a_0}(\pc) = \apply{\tagged{\pstore{n}{v'}}{\notags} \concat \buf_1'}{a_0'}(\pc)$.
                From this, we immediately get that $\apply{ \buf_1}{a_0}(\pc) = \apply{ \buf_1'}{a_0'}(\pc)$.
    
                \item[$\DeepProject{\buf_1} = \DeepProject{\buf_1'}$:] 
                This follows from $\buf_0 = \tagged{\pstore{v}{n}}{\notags} \concat \buf_1 $, $\buf_0' = \tagged{\pstore{v'}{n}}{\notags} \concat \buf_1' $, and (a).
    
                \item[$\CacheUpdate(\CacheState_0,n) = \CacheUpdate(\CacheState_0',n)$:]
                This follows from $\CacheState_0 = \CacheState_0'$, which, in turn, follows from (a).
    
                \item[$\BpState_0 = \BpState_0':$] 
                This follows from (a).
    
                \item[$\SchedUpdate(\SchedState_0, \BufProject{\buf_1}) = \SchedUpdate(\SchedState_0', \BufProject{\buf_1'})$:]
                From (a), we have   $\SchedState_0 = \SchedState_0'$.
                From $\DeepProject{\buf_1} = \DeepProject{\buf_1'}$ and Lemma~\ref{lemma:seq-processor:buffer-projections}, we have $\BufProject{\buf_1} = \BufProject{\buf_1'}$.
                Therefore, $\SchedUpdate(\SchedState_0, \BufProject{\buf_1}) = \SchedUpdate(\SchedState_0', \BufProject{\buf_1'})$.
            \end{description}
            Therefore, $C_1 \sim C_1'$ and (c) holds.
        \end{description}
        Therefore, (c) holds for all the cases.

    \end{description}
    Since (c) holds for all cases, this completes the proof of our lemma.
    \end{proof}

\subsection{Proof of Theorem~\ref{theorem:hni:load-delaytwo}}

\loadDelayTwo*

\begin{proof}
    Let $p$ be an arbitrary well-formed program.
	Moreover, let $\sigma = \tup{m,a},\sigma' = \tup{m',a'}$ be two arbitrary initial configurations.
	There are two cases:
	\begin{compactitem}
	\item[$\ArchSeqInterf{\Prg}(\sigma) \neq \ArchSeqInterf{\Prg}(\sigma')$:] Then, 	$\ArchSeqInterf{\Prg}(\sigma) = \ArchSeqInterf{\Prg}(\sigma') \Rightarrow \LoadDelayMuarchSem{\Prg}(\sigma) = \LoadDelayMuarchSem{\Prg}(\sigma')$ trivially holds.
	\item[$\ArchSeqInterf{\Prg}(\sigma) = \ArchSeqInterf{\Prg}(\sigma')$:]
		By unrolling the notion of $\ArchSeqInterf{\Prg}(\sigma)$ (together with all changes to the program counter $\pc$ being visible on traces), we obtain that there are runs $\crun:= \sigma$ $\ArchSeqInterfStep{o_1}{}$ $\sigma_1$ $\ArchSeqInterfStep{o_2}{}$ $\ldots$  $\ArchSeqInterfStep{o_{n-1}}{}$  $\sigma_n$ and $\crunp:= \sigma'\ArchSeqInterfStep{o_1'}{} \sigma_1' \ArchSeqInterfStep{o_2'}{} \ldots  \ArchSeqInterfStep{o_{n-1}'}{}  \sigma_n'$ such that $o_i = o_i'$ for all $0 < i < n$.
		By applying Lemma~\ref{lemma:loadDelay:arch-seq:main-lemma}, we immediately get that $\LoadDelayMuarchSem{\Prg}(\sigma) = \LoadDelayMuarchSem{\Prg}(\sigma')$ (because either both run terminate producing indistinguishable sequences of processor configurations or they both get stuck).
		Therefore, $\ArchSeqInterf{\Prg}(\sigma) = \ArchSeqInterf{\Prg}(\sigma') \Rightarrow \LoadDelayMuarchSem{\Prg}(\sigma) = \LoadDelayMuarchSem{\Prg}(\sigma')$ holds.
	\end{compactitem}
	Hence, $\ArchSeqInterf{\Prg}(\sigma) = \ArchSeqInterf{\Prg}(\sigma') \Rightarrow \LoadDelayMuarchSem{\Prg}(\sigma) = \LoadDelayMuarchSem{\Prg}(\sigma')$ holds for all programs $p$ and initial configurations $\sigma,\sigma'$.
	Therefore, $\hsni{\ArchSeqInterf{\cdot}}{\LoadDelayMuarchSem{\cdot}}$ holds.
\end{proof}

\subsubsection{Preliminary definitions}

\begin{definition}[Deep-update for \muarchStyle{loadDelay} and $\ArchSeqInterf{\cdot}{}$]
    Let $p$ be a program,  $\tup{m,a}$ be an \archstate{}, and $\buf$ be a buffer.
    The \emph{deep-update of $\tup{m,a}$ given $\buf$} is defined as follows:
        \begin{align*}
        \update{\tup{m,a}}{\emptysequence} &:= \tup{m,a} \\
        \update{\tup{m,a}}{ \tagged{\passign{x}{e}}{T}}  &:= 
            \begin{cases}
                \tup{m, a[x \mapsto \exprEval{e}{a}]} & \text{if } x \neq \pc \\
                \tup{m, a[x \mapsto \exprEval{e}{a}]} & \text{if } x = \pc  \wedge T = \notags \\
                \tup{m, a[x \mapsto \ell']} & \text{if } x = \pc  \wedge T = \ell \wedge p(\ell) = \pjz{y}{\ell'} \wedge a(y) = 0 \\ 
                \tup{m, a[x \mapsto \ell+1]} & \text{if } x = \pc  \wedge T = \ell \wedge p(\ell) = \pjz{y}{\ell'} \wedge a(y) \neq 0 
            \end{cases}
        \\
        \update{\tup{m,a}}{ \tagged{\pmarkedassign{x}{e}}{T}}  &:= \tup{m, a[x \mapsto \exprEval{e}{a}]}\\
        \update{\tup{m,a}}{ \tagged{\pload{x}{e}}{T}} &:= \tup{m,a[x \mapsto m(\exprEval{e}{a})] } \\
        \update{\tup{m,a}}{ \tagged{\pstore{x}{e}}{T}} &:= \tup{m[\exprEval{e}{a} \mapsto a(x)],a}\\
        \update{\tup{m,a}}{\tagged{\pskip{}}{T}} &:= \tup{m,a}\\
        \update{\tup{m,a}}{\tagged{\pbarrier{}}{T}} &:= \tup{m,a}\\
        \update{\tup{m,a}}{(\tagged{i}{T} \concat \buf)} &:= 				\update{ (\update{\tup{m,a}}{\tagged{i}{T}}) }{ \buf }
        \end{align*}
\end{definition}

\begin{definition}[Well-formed buffers for $\muarchStyle{loadDelay}$ and $\ArchSeqInterf{\cdot}{}$]
    A reorder buffer $\buf$ is \emph{well-formed for $\muarchStyle{loadDelay}$, $\ArchSeqInterf{\cdot}{}$, a well-formed program $p$, and an \archstate{}  $\tup{m,a}$}, written $\wellformed{\buf,\tup{m,a}}$, if the following conditions hold:
    \begin{align*}
    	\wellformed{\emptysequence,\tup{m,a}} & \\
        \wellformed{ \tagged{\passign{\pc}{e}}{T} \concat \buf, \tup{m,a} } & \text{ if } \wellformed{\buf, \update{\tup{m,a}}{\tagged{\passign{\pc}{e}}{T}}} \wedge p(a(\pc)) \sim_{\tup{m,a}} \tagged{\passign{\pc}{e}}{\notags} \\
        \wellformed{\tagged{\passign{\pc}{\ell}}{\ell_0} \concat \buf, \tup{m,a}} & \text{ if } \wellformed{\buf, \update{\tup{m,a}}{\tagged{\passign{\pc}{\ell}}{\ell_0}}} \wedge  \ell_0 \in \Val \wedge p(\ell_0) = \pjz{x}{\ell'} \wedge \\& \quad \ell \in \{\ell', \ell_0+1\} \wedge p(a(\pc)) \sim_{\tup{m,a}} \tagged{\passign{\pc}{\ell}}{\ell_0} \wedge \ell_0 = a(\pc)\\
    	\wellformed{ \tagged{\pmarkedassign{\pc}{\ell}}{\notags} \concat \buf , \tup{m,a} } & \text{ if } \wellformed{\buf,\update{\tagged{\pmarkedassign{\pc}{\ell}}{\notags}}{\tup{m,a}}} \wedge \ell \in \Val \\
        \wellformed{ \tagged{i}{\notags} \concat \tagged{\pmarkedassign{\pc}{\ell}}{\notags}  \concat \buf , \tup{m,a} } & \text{ if }  
        \wellformed{\buf, \update{\tup{m,a}}{(\tagged{i}{\notags} \concat \tagged{\pmarkedassign{\pc}{\ell}}{\notags} )}} \wedge 
        \ell \in \Val \wedge (\forall e.\ i \neq \passign{\pc}{e}) \wedge  \\
        & \quad (\forall x,e.\ i \neq \pmarkedassign{x}{e}) \wedge (\forall x,e.\ i \neq \pload{\pc}{e}) \wedge p(a(\pc)) \sim_{\tup{m,a}} \tagged{i}{\notags} \wedge \\ & \quad \ell = a(\pc)+1
    \end{align*}
    where the instruction-compatibility relation $\sim_{\tup{m,a}}$ is defined as follows:
    \begin{align*}
        \pskip &\sim_{\tup{m,a}} \tagged{\pskip}{\notags} \\ \allowdisplaybreaks
        \pbarrier &\sim_{\tup{m,a}} \tagged{\pbarrier}{\notags} \\ \allowdisplaybreaks
        \passign{x}{e} &\sim_{\tup{m,a}} \tagged{\passign{x}{e}}{\notags} \\ \allowdisplaybreaks
        \passign{x}{e} &\sim_{\tup{m,a}} \tagged{\passign{x}{v}}{\notags} \text{ if } v = \exprEval{e}{a} \\ \allowdisplaybreaks
        \pload{x}{e} &\sim_{\tup{m,a}} \tagged{\pload{x}{e}}{\notags} \\ \allowdisplaybreaks
        \pload{x}{e} &\sim_{\tup{m,a}} \tagged{\passign{x}{v}}{\notags} \text{ if } v = m(\exprEval{e}{a})\\ \allowdisplaybreaks
        \pstore{x}{e} &\sim_{\tup{m,a}} \tagged{\pstore{x}{e}}{\notags} \\ \allowdisplaybreaks
        \pstore{x}{e} &\sim_{\tup{m,a}} \tagged{\pstore{v}{n}}{\notags} \text{ if } v = a(x) \wedge n = \exprEval{e}{a} \\ \allowdisplaybreaks
        \pjmp{e} &\sim_{\tup{m,a}} \tagged{\passign{\pc}{e}}{\notags} \\ \allowdisplaybreaks
        \pjmp{e} &\sim_{\tup{m,a}} \tagged{\passign{\pc}{v}}{\notags} \text{ if } v = \exprEval{e}{a} \\ \allowdisplaybreaks
        \pjz{x}{\ell} &\sim_{\tup{m,a}} \tagged{\passign{\pc}{\ell}}{a(\pc)} \\ \allowdisplaybreaks
        \pjz{x}{\ell} &\sim_{\tup{m,a}} \tagged{\passign{\pc}{a(\pc)+1}}{a(\pc)} \\ \allowdisplaybreaks
        \pjz{x}{\ell} &\sim_{\tup{m,a}} \tagged{\passign{\pc}{a(\pc)+1}}{\notags} \wedge a(x) \neq 0\\ \allowdisplaybreaks
        \pjz{x}{\ell} &\sim_{\tup{m,a}} \tagged{\passign{\pc}{\ell}}{\notags} \wedge a(x) = 0
    \end{align*}
    \end{definition}

    \begin{definition}[Prefixes of buffers]\label{def:loadDelay:arch-seq:prefixes}
        The prefixes of a well-formed buffer $\buf$, given an \archstate{} $\tup{m,a}$, are defined as follows:
        \begin{align*}
            \prefixes{\emptysequence, \tup{m,a}} &= \{ \emptysequence \}\\
            \prefixes{  \tagged{\passign{\pc}{e}}{\notags} \concat \buf , \tup{m,a}} &= 
            \{ \emptysequence \} \cup \{  \tagged{\passign{\pc}{e}}{T} \concat \buf' \mid \buf' \in \prefixes{\buf, \update{\tup{m,a}}{\tagged{\passign{\pc}{e}}{\notags}} } \} \\
            \prefixes{  \tagged{\passign{\pc}{\ell}}{\ell_0} \concat \buf , \tup{m,a}} &= 
            \{ \emptysequence, \tagged{\passign{\pc}{\ell}}{\ell_0} \} \cup \{  \tagged{\passign{\pc}{\ell}}{\ell_0} \concat \buf' \mid \buf' \in \prefixes{\buf, \update{\tup{m,a}}{\tagged{\passign{\pc}{\ell}}{\ell_0}} } \wedge \\
            & \qquad p(\ell_0) = \pjz{x}{\ell'} \wedge (a(x) = 0 \rightarrow \ell = \ell') \wedge (a(x) \neq 0 \rightarrow \ell = \ell_0+1 ) \} \\
            \prefixes{  \tagged{i}{\notags} \concat \tagged{\pmarkedassign{\pc}{\ell}}{\notags} \concat \buf } &= \{ \emptysequence \} \cup \{   \tagged{i}{\notags} \concat \tagged{\pmarkedassign{\pc}{\ell}}{\notags} \concat \buf' \mid \buf' \in \prefixes{\buf, \update{\tup{m,a}}{ (\tagged{i}{\notags} \concat \tagged{\pmarkedassign{\pc}{\ell}}{\notags}) } } \} \\
            \prefixes{   \tagged{\pmarkedassign{\pc}{\ell}}{\notags} \concat \buf } &= \{  \tagged{\pmarkedassign{\pc}{\ell}}{\notags} \} \cup \{ \tagged{\pmarkedassign{\pc}{\ell}}{\notags} \concat \buf' \mid \buf' \in \prefixes{\buf, \update{\tup{m,a}}{\tagged{\pmarkedassign{\pc}{\ell}}{\notags}}} \} 
        \end{align*}
    \end{definition}

\subsubsection{Mapping lemma}

Lemma~\ref{lemma:loadDelay:arch-seq:buffers-well-formedness} states that all reorder buffers occurring in hardware runs are well-formed.

\begin{lemma}[Reorder buffers are well-formed]\label{lemma:loadDelay:arch-seq:buffers-well-formedness}
Let $p$ be a well-formed program, $\sigma_0 = \tup{m,a}$ be an initial \archstate, $\CacheState_0$ be the initial cache state, $\BpState_0$ be the initial branch predictor state, and $\SchedState_0$ be the initial scheduler state.
For all hardware runs  $\hrun := C_0 \LoadDelayMuarchStep{}{} C_1 \LoadDelayMuarchStep{}{} C_2 \LoadDelayMuarchStep{}{} \ldots \LoadDelayMuarchStep{}{} C_k$ and all $0 \leq i \leq k$, then $\wellformed{\buf_i,\tup{m_i,a_i}}$, where $C_0 = \tup{m,a,\emptysequence, \CacheState_0, \BpState_0, \SchedState_0}$ and $C_i = \tup{m_i, a_i,\buf_i, \CacheState_i, \BpState_i, \SchedState_i}$.
\end{lemma}

\begin{proof}
The lemma follows by (1) induction on $i$, and (2) inspection of the rules defining $\LoadDelayMuarchStep{}{}$.
\end{proof}

\begin{definition}[$\crun-\hrun$ mapping for $\muarchStyle{loadDelay}$ and $\ArchSeqInterf{\cdot}$]\label{def:loadDelay:arch-seq:mapping}
	Let $p$ be a well-formed program, $\sigma_0 = \tup{m,a}$ be an initial \archstate, $\CacheState_0$ be the initial cache state, $\BpState_0$ be the initial branch predictor state, and $\SchedState_0$ be the initial scheduler state.
	Furthermore, let:
	\begin{compactitem}
		\item $\crun := \sigma_0 \ArchSeqInterfStep{o_1}{} \sigma_1 \ArchSeqInterfStep{o_2}{} \sigma_2 \ArchSeqInterfStep{o_3}{} \ldots \ArchSeqInterfStep{o_{n-1}}{} \sigma_n$ be the longest $\interfStyle{seq-arch}$ contract run obtained starting from $\sigma$.
		\item $\hrun := C_0 \LoadDelayMuarchStep{}{} C_1 \LoadDelayMuarchStep{}{} C_2 \LoadDelayMuarchStep{}{} \ldots \LoadDelayMuarchStep{}{} C_k$ be the longest $\muarchStyle{loadDelay}$ hardware run obtained starting from $C_0 = \tup{m,a,\emptysequence, \CacheState_0, \BpState_0, \SchedState_0}$.
		\item $\hrun(i)$ be the $i$-th hardware configuration in $\hrun$.
		\item $\crun(i)$ be the $i$-th contract configuration in $\crun$ (note that $\crun(i) = \sigma_n$ for all $i > n$).
	\end{compactitem}
	The \emph{$\crun-\hrun$ mapping}, which maps hardware configurations in $\hrun$ to contract configurations in $\crun$, is defined as follows:
	\begin{align*}
		\map{\crun}{\hrun}{0} &:= \{0 \mapsto 0\} \\ 
		\map{\crun}{\hrun}{i} &:= {
            \begin{cases}
            \map{\crun}{\hrun}{i-1} & \text{if } \SchedNext(\hrun(i-1)) = \fetch{} \wedge  \Mispred{\hrun(i-1)} \\
			\map{\crun}{\hrun}{i-1} & \text{if } \SchedNext(\hrun(i-1)) = \fetch{} \wedge ln(\hrun(i-1)) = ln(\hrun(i)) \wedge \neg \Mispred{\hrun(i-1)}\\
			fetch_{\crun,\hrun}(i) & \text{if } \SchedNext(\hrun(i-1)) = \fetch{} \wedge ln(\hrun(i-1)) < ln(\hrun(i)) \wedge \neg \Mispred{\hrun(i-1)}\\
			\map{\crun}{\hrun}{i-1} & \text{if } \SchedNext(\hrun(i-1)) = \execute{j} \wedge \neg \Mispred{\hrun(i-1)}\\ 
			shift(\map{\crun}{\hrun}{i -1}) & \text{if } \SchedNext(\hrun(i-1)) = \retire{} \wedge \neg \Mispred{\hrun(i-1)} 
			\end{cases}
		}\\ 
		fetch_{\crun,\hrun}(i) &=
				\map{\crun}{\hrun}{i-1}[ln(\hrun(i-1))+2 \mapsto \map{\crun}{\hrun}{i-1}(ln(\hrun(i-1)))+1]\\
				& \qquad \text{if }
							p(\mathit{lstPc}(\hrun(i -1))) \neq \pjz{x}{\lbl} \wedge 
							p(\mathit{lstPc}(\hrun(i -1))) \neq \pjmp{e} 
							\\
		fetch_{\crun,\hrun}(i) &=
				\map{\crun}{\hrun}{i-1}[ln(\hrun({i-1}))+1 \mapsto \map{\crun}{\hrun}{i-1}(ln(\hrun(i-1)))+1]\\
				& \qquad \text{if }
							p(\mathit{lstPc}(\hrun(i -1))) = \pjmp{e} \vee p(\mathit{lstPc}(\hrun(i -1))) = \pjz{x}{\lbl} 
							\\
		ln(\tup{m,a,\buf,\CacheState,\BpState,\SchedState}) &= |\buf|\\
		\SchedNext(\tup{m,a,\buf,\CacheState,\BpState,\SchedState}) &= \SchedNext(\SchedState)\\
		shift(map) &= \lambda i \in \Nat.\ map(i +1 )\\
		\mathit{lstPc}(\tup{m,a,\buf,\CacheState,\BpState,\SchedState}) &= (\update{\tup{m,a}}{\buf})(\pc)\\
		\Mispred{\tup{m,a,\buf,\CacheState,\BpState,\SchedState}} &= 
	{	\begin{cases}
			\top & \text{if } \forall 1 \leq i \leq |\buf|.\ (\elt{\buf}{i} =  \tagged{\passign{\pc}{\ell}}{\ell'}) \rightarrow (\ell = \mathit{correctPred}(\ell', \update{\tup{m,a}}{\buf[0..i-1]}) )\\
			\bot & \text{otherwise}
		\end{cases}}\\
		\mathit{correctPred}(\ell,a) &= 
	{	\begin{cases}
		\ell'	& \text{if } p(\ell) = \pjz{x}{\ell'} \wedge a(x) = 0\\
		\ell + 1 & \text{otherwise}
		\end{cases}}
	\end{align*}
    \end{definition}
    
    We are now ready to prove Lemma~\ref{lemma:loadDelay:arch-seq:mapping-is-correct}, the main lemma showing the correctness of the $\crun-\hrun$ mapping.

    \begin{lemma}[Correctness of $\crun-\hrun$ mapping]\label{lemma:loadDelay:arch-seq:mapping-is-correct}
        Let $p$ be a well-formed program, $\sigma_0 = \tup{m,a}$ be an initial \archstate, $\CacheState_0$ be the initial cache state, $\BpState_0$ be the initial branch predictor state, and $\SchedState_0$ be the initial scheduler state.
        Furthermore, let:
        \begin{compactitem}
            \item $\crun := \sigma_0 \ArchSeqInterfStep{o_1}{} \sigma_1 \ArchSeqInterfStep{o_2}{} \sigma_2 \ArchSeqInterfStep{o_3}{} \ldots \ArchSeqInterfStep{o_{n-1}}{} \sigma_n$ be the longest $\interfStyle{seq-arch}$ contract run obtained starting from $\sigma_0$.
            \item $\hrun := C_0 \LoadDelayMuarchStep{}{}  C_1 \LoadDelayMuarchStep{}{} C_2 \LoadDelayMuarchStep{}{} \ldots \LoadDelayMuarchStep{}{} C_k$ be the longest $\muarchStyle{loadDelay}$ hardware run obtained starting from $C_0 = \tup{m,a,\emptysequence, \CacheState_0, \BpState_0, \SchedState_0}$.
            \item $\hrun(i)$ be the $i$-th hardware configuration in $\hrun$.
            \item $\crun(i)$ be the $i$-th contract configuration in $\crun$ (note that $\crun(i) = \sigma_n$ for all $i > n$).
            \item $ \map{\crun}{\hrun}{\cdot}$ be the mapping from Definition~\ref{def:loadDelay:arch-seq:mapping}.
        \end{compactitem}
        The following conditions hold:
        \begin{compactenum}[(1)]
        \item $C_0$ is an initial hardware configuration.
        \item $C_k$ is a final hardware configuration or there is no $C_{k'}$ such that $C_k \LoadDelayMuarchStep{}{} C_{k'}$.
        \item for all $0 \leq i \leq k$, given $C_i = \tup{m_i,a_i, \buf_i, \CacheState_i, \BpState_i, \SchedState_i}$ the following conditions hold:
            \begin{compactenum}[(a)]
                \item for all $\buf \in \prefixes{\buf_i, \tup{m_i,a_i}}$,  $\update{\tup{m_i,a_i}}{\buf} = \crun( \map{\crun}{\hrun}{i}(|\buf|) )$.
            \end{compactenum}
        \end{compactenum}
    \end{lemma}

    \begin{proof}
    Let $p$ be a well-formed program, $\sigma_0 = \tup{m,a}$ be an initial \archstate, $\CacheState_0$ be the initial cache state, $\BpState_0$ be the initial branch predictor state, and $\SchedState_0$ be the initial scheduler state.
    Furthermore, let:
    \begin{compactitem}
        \item $\crun := \sigma_0 \ArchSeqInterfStep{o_1}{} \sigma_1 \ArchSeqInterfStep{o_2}{} \sigma_2 \ArchSeqInterfStep{o_3}{} \ldots \ArchSeqInterfStep{o_{n-1}}{} \sigma_n$ be the longest $\interfStyle{seq-arch}$ contract run obtained starting from $\sigma_0$.
        \item $\hrun := C_0 \LoadDelayMuarchStep{}{}  C_1 \LoadDelayMuarchStep{}{} C_2 \LoadDelayMuarchStep{}{} \ldots \LoadDelayMuarchStep{}{} C_k$ be the longest $\muarchStyle{loadDelay}$ hardware run obtained starting from $C_0 = \tup{m,a,\emptysequence, \CacheState_0, \BpState_0, \SchedState_0}$.
        \item $\hrun(i)$ be the $i$-th hardware configuration in $\hrun$.
        \item $\crun(i)$ be the $i$-th contract configuration in $\crun$  (note that $\crun(i) = \sigma_n$ for all $i > n$).
        \item $ \map{\crun}{\hrun}{\cdot}$ be the mapping from Definition~\ref{def:loadDelay:arch-seq:mapping}.
    \end{compactitem}
    Observe that $C_0$ is an initial hardware configuration by construction.
    Observe also that $C_k$ is either a final hardware configuration (for which the semantics cannot further proceed) or the computation is stuck (since $\hrun$ is the longest run by construction).
    Therefore, (1) and (2) hold.

    We now prove, by induction on $i$, that (3) holds, i.e., that for all $0 \leq i \leq k$ and for all $\buf \in \prefixes{\buf_i}$,  $\update{\tup{m_i,a_i}}{\buf} = \crun( \map{\crun}{\hrun}{i}(|\buf|) )$.
    \begin{description}
        \item[Base case:]
        For the base case, we consider $i = 0$.
        Observe that $C_0 = \tup{m,a,\emptysequence, \CacheState_0, \BpState_0, \SchedState_0}$.
        Therefore, $\buf_0 = \emptysequence$ and $\prefixes{\buf_0} = \{ \emptysequence \}$ (see Definition~\ref{def:seq-processor:prefixes}).
        Then, (3.a) immediately follows from $\update{m,a}{\emptysequence} = \tup{m,a} = \sigma_0 = \crun(0)$.

        \item[Induction step:] 
        For the induction step, we assume that the claim holds for all $i' < i$ and we show that it holds for $i$ as well.
        In the following, let $C_i = \tup{m_i,a_i, \buf_i, \CacheState_i, \BpState_i, \SchedState_i}$ and we refer to the induction hypothesis as H.
        Similarly, we write H.3.a to denote that fact 3.a holds for the induction hypothesis.
        
        We proceed by case distinction on the directive $\SchedNext(C_{i-1})$ used to derive $C_i$.
        There are three cases:
        \begin{description}
            \item[$\SchedNext(C_{i-1}) = \fetch{}$:]
            Then, $C_i$ has been derived using the \textsc{Step-Others} rule.
            We proceed by case distinction on the $\fetch{}$-rule:
            \begin{description}
                \item[Rule \textsc{Fetch-Branch-Hit}:]
			    Then, $\buf_{i} = \buf_{i-1} \concat \tagged{\passign{\pc}{\lbl'}}{a_{i-1}(\pc)}$, $m_{i} = m_{a-1}$, $a_{i} = a_{i-1}$, and $p(a_{i-1}(\pc)) = \pjz{x}{\lbl}$.
                There are two cases:
                \begin{description}
                    \item[$\Mispred{C_{i-1}}$:]
                    Then, $\map{\crun}{\hrun}{i} = \map{\crun}{\hrun}{i-1}$. 
                    From $\Mispred{C_{i-1}}$ and the well-formedness of buffers (see Lemma~\ref{lemma:loadDelay:arch-seq:buffers-well-formedness}), we have that $\prefixes{\buf_{i-1}, \tup{m_{i-1},a_{i-1}}}  = \prefixes{\buf_{i}, \tup{m_{i},a_{i}}}$ (because $m_{i} = m_{a-1}$, $a_{i} = a_{i-1}$, $\buf_{i} = \buf_{i-1} \concat \tagged{\passign{\pc}{\lbl'}}{a_{i-1}(\pc)}$, and one of the earlier branch instructions has been mispredicted).
                    Let $\buf$ be an arbitrary prefix in $\prefixes{\buf_{i},\tup{m_i,a_i}}$.
                    Then, (3.a) immediately follows from $\prefixes{\buf_{i-1}, \tup{m_{i-1},a_{i-1}}}  = \prefixes{\buf_{i}, \tup{m_{i},a_{i}}}$, $\buf \in \prefixes{\buf_{i}, \tup{m_i,a_i}}$, $\map{\crun}{\hrun}{i} = \map{\crun}{\hrun}{i-1}$, and (H.3.a).
                    Since $\buf$ is an arbitrary prefix in $\prefixes{\buf_{i},\tup{m_i,a_i}}$, (3.a) holds.

                    \item[$\neg \Mispred{C_{i-1}}$:]
                    Then,   $\map{\crun}{\hrun}{i} = \map{\crun}{\hrun}{i-1}[\mathit{ln}(|\buf_{i-1}|) +1 \mapsto \map{\crun}{\hrun}{i-1}(\mathit{ln}(|\buf_{i-1}|)) + 1]$.
                    Let $\buf$ be an arbitrary prefix in $\prefixes{\buf_{i},\tup{m_i,a_i}}$.
                    There are two cases:
                    \begin{description}
                        \item[$\buf \in  \prefixes{\buf_{i-1},\tup{m_{i-1},a_{i-1}}}$:]
                        From (H.3.a), we have that $\update{\tup{m_{i-1},a_{i-1}}}{\buf} = \map{\crun}{\hrun}{i-1}(|\buf|)$.
                        From $m_i = m_{i-1}$ and $a_i = a_{i-1}$, we have that $\update{\tup{m_{i},a_{i}}}{\buf} = \crun(\map{\crun}{\hrun}{i-1}(|\buf|))$.
                        Moreover, from $\buf \in  \prefixes{\buf_{i-1},\tup{m_{i-1},a_{i-1}}}$, we also have that $|\buf| \leq |\buf_{i-1}|$.
                        Therefore, $\map{\crun}{\hrun}{i-1}(|\buf|) = \map{\crun}{\hrun}{i}(|\buf|)$ since $\map{\crun}{\hrun}{i} = \map{\crun}{\hrun}{i-1}[\mathit{ln}(|\buf_{i-1}|) +1 \mapsto \map{\crun}{\hrun}{i-1}(\mathit{ln}(|\buf_{i-1}|)) + 1]$.
                        As a result, we have that $\update{\tup{m_{i},a_{i}}}{\buf} = \crun(\map{\crun}{\hrun}{i}(|\buf|))$.

                        \item[$\buf \not\in  \prefixes{\buf_{i-1},\tup{m_{i-1},a_{i-1}}}$:]
                        Then, $\buf = \buf_i$.
                        Since $\buf_{i-1} \in \prefixes{\buf_{i-1},\tup{m_{i-1},a_{i-1}}}$ (because $\neg \Mispred{C_{i-1}}$ holds), we can apply (H.3.a) to $\buf_{i-1}$ to get $\update{\tup{m_{i-1},a_{i-1}}}{\buf_{i-1}} = \crun(\map{\crun}{\hrun}{i-1}(|\buf_{i-1}|))$.
                        There are two cases:
                        \begin{description}
                            \item[$\update{\tup{m_{i-1},a_{i-1}}}{\buf_{i-1}}(x) = 0$:]
                            Then, $\update{\tup{m_{i},a_{i}}}{\buf} = \update{\tup{m_{i},a_{i}}}{\buf_{i-1}}[\pc \mapsto \lbl]$.
                            From $\update{\tup{m_{i-1},a_{i-1}}}{\buf_{i-1}} = \crun(\map{\crun}{\hrun}{i-1}(|\buf_{i-1}|))$, we therefore get $\crun(\map{\crun}{\hrun}{i-1}(|\buf_{i-1}|))(x) = 0$.
                            Hence, $\crun(\map{\crun}{\hrun}{i-1}(|\buf_{i-1}|)+ 1)$ is obtained from $\crun(\map{\crun}{\hrun}{i-1}(|\buf_{i-1}|))$ by applying the \textsc{Beqz} and \textsc{Beqz-Sat} rules.
                            Therefore, $\crun(\map{\crun}{\hrun}{i-1}(|\buf_{i-1}|)+ 1) = \crun(\map{\crun}{\hrun}{i-1}(|\buf_{i-1}|))[\pc \mapsto \ell]$.
                            Therefore, $\update{\tup{m_{i},a_{i}}}{\buf} = \crun(\map{\crun}{\hrun}{i-1}(|\buf_{i-1}|)+ 1)$.
                            From this, $|\buf| = |\buf_{i-1}|+1$, and $\map{\crun}{\hrun}{i} = \map{\crun}{\hrun}{i-1}[\mathit{ln}(|\buf_{i-1}|) +1 \mapsto \map{\crun}{\hrun}{i-1}(\mathit{ln}(|\buf_{i-1}|)) + 1]$, we have that $\update{\tup{m_{i},a_{i}}}{\buf} = \crun(\map{\crun}{\hrun}{i}(|\buf|))$.

                            \item[$\update{\tup{m_{i-1},a_{i-1}}}{\buf_{i-1}}(x) \neq 0$:]
                            The proof of this case is similar to that of the case $\update{\tup{m_{i-1},a_{i-1}}}{\buf_{i-1}}(x) = 0$.
                        \end{description}
                    \end{description}
                    Since $\buf$ has been selected arbitrarily, (3.a) holds.                    
                \end{description}

                \item[Rule \textsc{Fetch-Jump-Hit}:]
                The proof of this case is similar to the case \textsc{Fetch-Branch-Hit}.

                \item[Rule \textsc{Fetch-Others-Hit}:]
                The proof of this case is similar to the case \textsc{Fetch-Branch-Hit} (see also the corresponding case in Lemma~\ref{lemma:seq-processor:mapping-is-correct}).

                \item[Rule \textsc{Fetch-Miss}:]
			    Then, $\buf_{i} = \buf_{i-1}$, $m_{i} = m_{a-1}$, and $a_{i} = a_{i-1}$.
                Moreover, $\map{\crun}{\hrun}{i} = \map{\crun}{\hrun}{i-1}$.
                Let $\buf$ be an arbitrary prefix in $\prefixes{\buf_{i},\tup{m_i,a_i}}$.
                Then, (3.a) immediately follows from $\prefixes{\buf_i,\tup{m_i,a_i}} = \prefixes{\buf_{i-1}},\tup{m_i,a_i}$, $\buf \in \prefixes{\buf_{i-1}, \tup{m_i,a_i}}$, $m_{i-1} = m_i$, $a_{i-1} = a_i$, $\map{\crun}{\hrun}{i} = \map{\crun}{\hrun}{i-1}$, and (H.3.a).
                Since $\buf$ is an arbitrary prefix in $\prefixes{\buf_{i},\tup{m_i,a_i}}$, (3.a) holds.
                
            \end{description}

            \item[$\SchedNext(C_{i-1}) = \execute{j}$:]
            Then, $\map{\crun}{\hrun}{i} = \map{\crun}{\hrun}{i-1}$.
            We now proceed by case distinction on the applied $\execute{}$ rule:
            \begin{description}
                \item[Rule \textsc{Execute-Load-Hit}:]
                Therefore, we have that $\pbarrier \not\in \buf_{i-1}[0..j-1]$, $\pstore{x'}{e'} \not\in \buf_{i-1}[0..j-1]$, $m_i = m_{i-1}$, $a_i = a_{i-1}$, $\buf_{i} = \buf_{i-1}[0..j-1] \concat  \tagged{\passign{x}{m(\exprEval{e}{ \apply{\buf_{i-1}[0..j-1]}{a_{i-1} } }) }}{T} \concat \buf_{i-1}[j+1..|\buf_{i-1}|]$, and $\elt{\buf_{i-1}}{j} = \tagged{\pload{x}{e}}{T}$.
                Let $\buf$ be an arbitrary prefix in $\prefixes{\buf_{i}, \tup{m_i,a_i}}$.
                There are two cases:
                \begin{description}
                    \item[$|\buf| < j$:] 
                    Since $\buf \in  \prefixes{\buf_{i}, \tup{m_i,a_i}}$, there is no mispredicted branch instruction in $\buf$ (except potentially in the last position).
                    From this, $\buf_{i} = \buf_{i-1}[0..j-1] \concat  \tagged{\passign{x}{m(\exprEval{e}{ \apply{\buf_{i-1}[0..j-1]}{a_{i-1} } }) }}{T} \concat \buf_{i-1}[j+1..|\buf_{i-1}|]$, and $m_i = m_{i-1}$, $a_i = a_{i-1}$, we have that $\buf \in  \prefixes{\buf_{i-1}, \tup{m_{i-1},a_{i-1}}}$.
                    By applying (H.3.a) to $\buf$, we therefore obtain that $\update{\tup{m_{i-1},a_{i-1}}}{\buf} = \crun(\map{\crun}{\hrun}{i-1}(|\buf|))$.
                    From $m_i = m_{i-1}$, $a_i = a_{i-1}$, and $\map{\crun}{\hrun}{i} = \map{\crun}{\hrun}{i-1}$, we get that $\update{\tup{m_{i},a_{i}}}{\buf} = \crun(\map{\crun}{\hrun}{i}(|\buf|))$.

                    \item[$|\buf| \geq j$:]
                    Since $\buf \in  \prefixes{\buf_{i}, \tup{m_i,a_i}}$, there is no mispredicted branch instruction in $\buf$  (except potentially in the last position).
                    Therefore, the prefix $\buf' = \buf[0..j-1] \concat \tagged{\pload{x}{e}}{T} \concat \buf[j+1..|\buf|]$ belongs to $\prefixes{\buf_{i-1}, \tup{m_{i-1},a_{i-1}}}$.
                    By applying (H.3.a) to $\buf'$, we therefore obtain that $\update{\tup{m_{i-1},a_{i-1}}}{\buf'} = \crun(\map{\crun}{\hrun}{i-1}(|\buf'|))$.
                    From $m_i = m_{i-1}$, $a_i = a_{i-1}$, $\map{\crun}{\hrun}{i} = \map{\crun}{\hrun}{i-1}$, and $|\buf| = |\buf'|$, we get that $\update{\tup{m_{i},a_{i}}}{\buf'} = \crun(\map{\crun}{\hrun}{i}(|\buf|))$.
                    Finally, observe that $\update{\tup{m_{i},a_{i}}}{\buf'} = \update{\tup{m_{i},a_{i}}}{\buf}$ (because $\pstore{x'}{e'} \not\in \buf_{i-1}[0..j-1]$).
                    Hence, $\update{\tup{m_{i},a_{i}}}{\buf} = \crun(\map{\crun}{\hrun}{i}(|\buf|))$.

                \end{description}
                Since $\buf$ has been selected arbitrarily, (3.a) holds.
                
                \item[Rule \textsc{Execute-Load-Miss}:]
                Therefore, we have that $m_i = m_{i-1}$, $a_i = a_{i-1}$, $\buf_{i} = \buf_{i-1}$.
                Let $\buf$ be an arbitrary prefix in $\prefixes{\buf_{i}, \tup{m_i,a_i}}$.
                From $m_i = m_{i-1}$, $a_i = a_{i-1}$, $\buf_{i} = \buf_{i-1}$, we have that $\buf \in \prefixes{\buf_{i-1}, \tup{m_{i-1},a_{i-1}}}$.
                By applying (H.3.a) to $\buf$, we get $\update{\tup{m_{i-1},a_{i-1}}}{\buf} = \crun(\map{\crun}{\hrun}{i-1}(|\buf|))$.
                From $m_i = m_{i-1}$, $a_i = a_{i-1}$, $\map{\crun}{\hrun}{i} = \map{\crun}{\hrun}{i-1}$, we get $\update{\tup{m_{i},a_{i}}}{\buf} = \crun(\map{\crun}{\hrun}{i}(|\buf|))$.
                Since $\buf$ has been selected arbitrarily, (3.a) holds.

                \item[Rule \textsc{Execute-Branch-Commit}:]
                Therefore, we have that $\pbarrier \not\in \buf_{i-1}[0..j-1]$, $m_i = m_{i-1}$, $a_i = a_{i-1}$, $\buf_{i} = \buf_{i-1}[0..j-1] \concat  \tagged{\passign{\pc}{ \lbl  }}{\notags}   \concat \buf_{i-1}[j+1..|\buf_{i-1}|]$, $\elt{\buf_{i-1}}{j} = \tagged{\passign{\pc}{\lbl}}{\lbl_0} $, $p(\lbl_0) = \pjz{x}{\lbl''}$, and $( \apply{\buf_{i-1}[0..j-1]}{a_{i-1}}(x) = 0 \wedge \lbl = \lbl'') \vee (\apply{\buf_{i-1}[0..j-1]}{a_{i-1}}(x) \in \Val \setminus \{0,\bot\} \wedge \lbl = \ell_0+1)$.
                Let $\buf$ be an arbitrary prefix in $\prefixes{\buf_{i}, \tup{m_i,a_i}}$.
                There are two cases:
                \begin{description}
                    \item[$|\buf| < j$:] 
                    Since $\buf \in  \prefixes{\buf_{i}, \tup{m_i,a_i}}$, there is no mispredicted branch instruction in $\buf$ (except potentially in the last position).
                    From this, $\buf_{i} = \buf_{i-1}[0..j-1] \concat \tagged{\passign{\pc}{ \lbl  }}{\notags} \concat \buf_{i-1}[j+1..|\buf_{i-1}|]$, and $m_i = m_{i-1}$, $a_i = a_{i-1}$, we have that $\buf \in  \prefixes{\buf_{i-1}, \tup{m_{i-1},a_{i-1}}}$.
                    By applying (H.3.a) to $\buf$, we therefore obtain that $\update{\tup{m_{i-1},a_{i-1}}}{\buf} = \crun(\map{\crun}{\hrun}{i-1}(|\buf|))$.
                    From $m_i = m_{i-1}$, $a_i = a_{i-1}$, and $\map{\crun}{\hrun}{i} = \map{\crun}{\hrun}{i-1}$, we get that $\update{\tup{m_{i},a_{i}}}{\buf} = \crun(\map{\crun}{\hrun}{i}(|\buf|))$.

                    \item[$|\buf| \geq j$:]
                    Since $\buf \in  \prefixes{\buf_{i}, \tup{m_i,a_i}}$, there is no mispredicted branch instruction in $\buf$  (except potentially in the last position).
                    Therefore, the prefix $\buf' = \buf[0..j-1] \concat \tagged{\passign{\pc}{\lbl}}{\lbl_0} \concat \buf[j+1..|\buf|]$ belongs to $\prefixes{\buf_{i-1}, \tup{m_{i-1},a_{i-1}}}$.
                    By applying (H.3.a) to $\buf'$, we therefore obtain that $\update{\tup{m_{i-1},a_{i-1}}}{\buf'} = \crun(\map{\crun}{\hrun}{i-1}(|\buf'|))$.
                    From $m_i = m_{i-1}$, $a_i = a_{i-1}$, $\map{\crun}{\hrun}{i} = \map{\crun}{\hrun}{i-1}$, and $|\buf| = |\buf'|$, we get that $\update{\tup{m_{i},a_{i}}}{\buf'} = \crun(\map{\crun}{\hrun}{i}(|\buf|))$.
                    Finally, observe that $\update{\tup{m_{i},a_{i}}}{\buf'} = \update{\tup{m_{i},a_{i}}}{\buf}$ (because $( \apply{\buf_{i-1}[0..j-1]}{a_{i-1}}(x) = 0 \wedge \lbl = \lbl'') \vee (\apply{\buf_{i-1}[0..j-1]}{a_{i-1}}(x) \in \Val \setminus \{0,\bot\} \wedge \lbl = \ell_0+1)$).
                    Hence, $\update{\tup{m_{i},a_{i}}}{\buf} = \crun(\map{\crun}{\hrun}{i}(|\buf|))$.

                \end{description}
                Since $\buf$ has been selected arbitrarily, (3.a) holds.

                \item[Rule \textsc{Execute-Branch-Rollback}:]
                Therefore, we have that $\pbarrier \not\in \buf_{i-1}[0..j-1]$, $m_i = m_{i-1}$, $a_i = a_{i-1}$, $\buf_{i} = \buf_{i-1}[0..j-1] \concat  \tagged{\passign{\pc}{ \lbl  }}{\notags}$, $\elt{\buf_{i-1}}{j} = \tagged{\passign{\pc}{\lbl}}{\lbl_0} $, $p(\lbl_0) = \pjz{x}{\lbl''}$, and $(a'(x) = 0 \wedge \lbl \neq \lbl'') \vee (a'(x) \in \Val \setminus \{0,\bot\} \wedge \lbl \neq \ell_0+1)$.
                Let $\buf$ be an arbitrary prefix in $\prefixes{\buf_{i}, \tup{m_i,a_i}}$.
                There are two cases:
                \begin{description}
                    \item[$|\buf| < j$:] 
                    Since $\buf \in  \prefixes{\buf_{i}, \tup{m_i,a_i}}$, there is no mispredicted branch instruction in $\buf$  (except potentially in the last position).
                    From this, $\buf_{i} = \buf_{i-1}[0..j-1] \concat \tagged{\passign{\pc}{ \lbl  }}{\notags} \concat \buf_{i-1}[j+1..|\buf_{i-1}|]$, and $m_i = m_{i-1}$, $a_i = a_{i-1}$, we have that $\buf \in  \prefixes{\buf_{i-1}, \tup{m_{i-1},a_{i-1}}}$.
                    By applying (H.3.a) to $\buf$, we therefore obtain that $\update{\tup{m_{i-1},a_{i-1}}}{\buf} = \crun(\map{\crun}{\hrun}{i-1}(|\buf|))$.
                    From $m_i = m_{i-1}$, $a_i = a_{i-1}$, and $\map{\crun}{\hrun}{i} = \map{\crun}{\hrun}{i-1}$, we get that $\update{\tup{m_{i},a_{i}}}{\buf} = \crun(\map{\crun}{\hrun}{i}(|\buf|))$.

                    \item[$|\buf| = j$:]
                    From this, $\buf = \buf_{i}$.
                    Since $\buf \in  \prefixes{\buf_{i}, \tup{m_i,a_i}}$, there is no mispredicted branch instruction in $\buf$.
                    Therefore, the prefix $\buf' = \buf[0..j-1] \concat \tagged{\passign{\pc}{\lbl}}{\lbl_0}$ belongs to $\prefixes{\buf_{i-1}, \tup{m_{i-1},a_{i-1}}}$.
                    By applying (H.3.a) to $\buf'$, we therefore obtain that $\update{\tup{m_{i-1},a_{i-1}}}{\buf'} = \crun(\map{\crun}{\hrun}{i-1}(|\buf'|))$.
                    From $m_i = m_{i-1}$, $a_i = a_{i-1}$, $\map{\crun}{\hrun}{i} = \map{\crun}{\hrun}{i-1}$, and $|\buf| = |\buf'|$, we get that $\update{\tup{m_{i},a_{i}}}{\buf'} = \crun(\map{\crun}{\hrun}{i}(|\buf|))$.
                    Finally, observe that $\update{\tup{m_{i},a_{i}}}{\buf'} = \update{\tup{m_{i},a_{i}}}{\buf}$ (because $(a'(x) = 0 \wedge \lbl \neq \lbl'') \vee (a'(x) \in \Val \setminus \{0,\bot\} \wedge \lbl \neq \ell_0+1)$).
                    Hence, $\update{\tup{m_{i},a_{i}}}{\buf} = \crun(\map{\crun}{\hrun}{i}(|\buf|))$.

                \end{description}
                Since $\buf$ has been selected arbitrarily, (3.a) holds.

                \item[Rule \textsc{Execute-Assignment}:]
                The proof of this case is similar to that of the \textsc{Execute-Load-Hit} rule.

                \item[Rule \textsc{Execute-Marked-Assignment}:]
                The proof of this case is similar to that of the \textsc{Execute-Load-Hit} rule.

                \item[Rule \textsc{Execute-Store}:]
                The proof of this case is similar to that of the \textsc{Execute-Load-Hit} rule.

                \item[Rule \textsc{Execute-Skip}:]
                The proof of this case is similar to that of the \textsc{Execute-Load-Miss} rule.

                \item[Rule \textsc{Execute-Barrier}:]
                The proof of this case is similar to that of the \textsc{Execute-Load-Miss} rule.

            \end{description}
            This completes the proof for the $\execute{j}$ case.

            \item[$\SchedNext(C_{i-1}) = \retire{}$:]
            The proof of the $\retire{}$ case is similar to the proof of the corresponding case in Lemma~\ref{lemma:seq-processor:mapping-is-correct}.

        \end{description}
    \end{description}
    \end{proof}

\subsubsection{Indistinguishability lemma}

\begin{definition}[Deep-indistinguishability of hardware configurations]\label{def:loadDelay:arch-seq:deep-indistinguishability}
	We say that two hardware configurations $\tup{\sigma,\mu} = \tup{m,a,\buf, \CacheState,\BpState, \SchedState}$ and $\tup{\sigma',\mu'} = \tup{m',a',\buf', \CacheState',\BpState', \SchedState'}$ are \emph{deep-indistinguishable}, written $\tup{\sigma,\mu} \sim \tup{\sigma',\mu'}$, iff
    \begin{inparaenum}[(a)]
        \item $a = a'$,
		\item $\buf = \buf'$,
		\item $\CacheState = \CacheState'$,
		\item $\BpState = \BpState'$, and
		\item $\SchedState = \SchedState'$.
	\end{inparaenum}
\end{definition}

\begin{lemma}[Observation equivalence preserves deep-indistinguishability]\label{lemma:loadDelay:arch-seq:trace-equiv-implies-stepwise-indistinguishability}
Let $p$ be a well-formed program and $C_0 = \tup{m_0,a_0,\buf_0,\CacheState_0, \BpState_0, \SchedState_0}$, $C_0' = \tup{m_0',a_0',\buf_0',\CacheState_0', \BpState_0', \SchedState_0'}$ be reachable hardware configurations.
If
\begin{inparaenum}[(a)]
	\item $C_0 \sim C_0'$, and
	\item for all $\buf \in \prefixes{\buf_0,\tup{m_0,a_0}}$, $\buf' \in \prefixes{\buf_0',\tup{m_0',a_0'}}$ such that $|\buf| = |\buf'|$, 
	there are $\sigma_0, \sigma_0', \sigma_1, \sigma_1', \tau, \tau'$ such that $\sigma_0 \ArchSeqInterfStep{\tau}{} \sigma_1$, $\sigma_0' \ArchSeqInterfStep{\tau'}{} \sigma_1'$, $\tau = \tau'$, $C_0 \bufEquiv{|\buf|} \sigma_0$, and $C_0' \bufEquiv{|\buf'|} \sigma_0'$,
\end{inparaenum}
then either there are $C_1, C_1'$ such that $C_0 \LoadDelayMuarchStep{}{} C_1$, $C_0' \LoadDelayMuarchStep{}{} C_1'$, and $C_1 \sim C_1'$ or there is no $C_1$ such that $C_0 \LoadDelayMuarchStep{}{} C_1$ and no $C_1'$ such that $C_0' \LoadDelayMuarchStep{}{} C_1'$.
\end{lemma}

\begin{proof}
Let $p$ be a well-formed program and $C_0 = \tup{m_0,a_0,\buf_0,\CacheState_0, \BpState_0, \SchedState_0}$, $C_0' = \tup{m_0',a_0',\buf_0',\CacheState_0', \BpState_0', \SchedState_0'}$ be reachable hardware configurations.
Moreover, we assume that conditions (a) and (b) holds. 
In the following, we denote by (c) the post-condition ``either there are $C_1, C_1'$ such that $C_0 \LoadDelayMuarchStep{}{} C_1$, $C_0' \LoadDelayMuarchStep{}{} C_1'$, and $C_1 \sim C_1'$ or there is no $C_1$ such that $C_0 \LoadDelayMuarchStep{}{} C_1$ and no $C_1'$ such that $C_0' \LoadDelayMuarchStep{}{} C_1'$.''

From (a), it follows that $\SchedState_0 = \SchedState_0$.
Therefore, the directive obtained from the scheduler is the same in both cases, i.e., $\SchedNext(\SchedState_0) = \SchedNext(\SchedState_0')$.
We proceed by case distinction on the directive $d = \SchedNext(\SchedState_0)$:
\begin{description}
    \item[$d = \fetch{}$:]
	Therefore, we can only apply one of the $\fetch{}$ rules depending on the current program counter.
	There are two cases: 
	\begin{description}
		\item[$\apply{\buf_0}{a_0}(\pc) \neq \bot \wedge |\buf_0| < \wMuarch$:]
		There are several cases:
		\begin{description}
			\item[$\CacheAccess(\CacheState_0,  \apply{\buf_0}{a_0}(\pc)) = \CacheHit \wedge p(\apply{\buf_0}{a_0}(\pc)) = \pjz{x}{\lbl}$:] 
			From (a), we get that $\CacheState_0' = \CacheState_0$ and $\apply{\buf_0'}{a_0'}(\pc) = \apply{\buf_0}{a_0}(\pc)$ (because $\buf_0 = \buf_0'$ and $a_0 = a_0'$).
			Therefore, $\CacheAccess(\CacheState_0',  \apply{\buf_0'}{a_0'}(\pc)) = \CacheHit$.
			Moreover, from (a) we also get that $\apply{\buf_0}{a_0}(\pc)=\apply{\buf_0'}{a_0'}(\pc)$ and, therefore, $p(\apply{\buf_0'}{a_0'}(\pc)) = \pjz{x}{\lbl}$ as well.
			Therefore, we can apply the \textsc{Fetch-Branch-Hit} and \textsc{Step-Others} rules to $C_0$ and $C_0'$ as follows:
			\begin{align*}
				\lbl_0 &:= \BpPredict(\BpState_0, \apply{\buf_0}{a_0}(\pc))\\
				\lbl_0' &:= \BpPredict(\BpState_0', \apply{\buf_0'}{a_0'}(\pc))\\
				\buf_1 &:= \buf_0  \concat  \tagged{\passign{\pc}{\lbl_0}}{\apply{\buf_0}{a_0}(\pc)}\\
				\buf_1' &:= \buf_0'  \concat \tagged{\passign{\pc}{\lbl_0'}}{\apply{\buf_0'}{a_0'}(\pc)}\\
				\tup{m_0,a_0,\buf_0,\CacheState_0,\BpState_0} &\muarchStep{\fetch{}}{} \tup{m_0, a_0, \buf_1, \CacheUpdate(\CacheState_0, \apply{\buf_0}{a_0}(\pc)),\BpState_0}\\
				\tup{m_0,a_0,\buf_0,\CacheState_0,\BpState_0, \SchedState_0} &\LoadDelayMuarchStep{}{} \tup{m_0, a_0, \buf_1, \CacheUpdate(\CacheState_0, \apply{\buf_0}{a_0}(\pc)),\BpState_0, \SchedUpdate(\SchedState_0, \BufProject{\buf_1})}\\
				\tup{m_0',a_0',\buf_0',\CacheState_0',\BpState_0'} &\muarchStep{\fetch{}}{} \tup{m_0', a_0', \buf_1', \CacheUpdate(\CacheState_0',  \apply{\buf_0'}{a_0'}(\pc)),\BpState_0'}\\
				\tup{m_0',a_0',\buf_0',\CacheState_0',\BpState_0', \SchedState_0'} &\LoadDelayMuarchStep{}{} \tup{m_0', a_0', \buf_1', \CacheUpdate(\CacheState_0',  \apply{\buf_0'}{a_0'}(\pc)),\BpState_0', \SchedUpdate(\SchedState_0', \BufProject{\buf_1'})}
			\end{align*}
			We now show that $C_1 =  \tup{m_0, a_0, \buf_1, \CacheUpdate(\CacheState_0, \apply{\buf_0}{a_0}(\pc)),\BpState_0, \SchedUpdate(\SchedState_0, \BufProject{\buf_1})}$ and $C_1' = \tup{m_0', a_0', \buf_1', \CacheUpdate(\CacheState_0',  \apply{\buf_0'}{a_0'}(\pc)),\BpState_0', \SchedUpdate(\SchedState_0', \BufProject{\buf_1'})}$ are indistinguishable, i.e., $C_1 \sim C_1'$.
			For this, we need to show that:
            \begin{description}
                \item[$a_0 = a_0'$:]
                This follows from (a).

				\item[$\apply{\buf_1}{a_0}(\pc) = \apply{\buf_1'}{a_0'}(\pc)$:]
				We know that $\apply{\buf_1}{a_0}(\pc) = \lbl_0$ and $\apply{\buf_1'}{a_0'}(\pc) = \lbl_0'$.
				From $\lbl_0 = \BpPredict(\BpState_0, \apply{\buf_0}{a_0}(\pc))$, $\lbl_0' = \BpPredict(\BpState_0', \apply{\buf_0'}{a_0'}(\pc))$, and (a), we immediately get $\lbl_0 = \lbl_0'$.

				\item[${\buf_1} = {\buf_1'}$:] 
                Observe that (1) $\apply{\buf_0}{a_0}(\pc) = \apply{\buf_0}{a_0}(\pc)$ (which follows from $a_0 = a_0'$ and $\buf_0 = \buf_0'$ thanks to (a)), and (2) $\lbl_0 = \lbl_0'$ which follows $\lbl_0 = \BpPredict(\BpState_0, \apply{\buf_0}{a_0}(\pc))$, $\lbl_0' = \BpPredict(\BpState_0', \apply{\buf_0'}{a_0'}(\pc))$, and (a).
                From this, (a),  $\buf_1 = \buf_0  \concat  \tagged{\passign{\pc}{\lbl_0}}{\apply{\buf_0}{a_0}(\pc)}$, $\buf_1' = \buf_0'  \concat  \tagged{\passign{\pc}{\lbl_0'}}{\apply{\buf_0'}{a_0'}(\pc)}$, we get that ${\buf_1} = {\buf_1'}$.

				\item[$\CacheUpdate(\CacheState_0,  \apply{\buf_0}{a_0}(\pc)) = \CacheUpdate(\CacheState_0',  \apply{\buf_0'}{a_0'}(\pc))$:]
				This follows from $\CacheState_0 = \CacheState_0'$, which in turn follows from (a), and $\apply{\buf_0}{a_0}(\pc) = \apply{\buf_0'}{a_0'}(\pc)$.

				\item[$\BpState_0 = \BpState_0':$] 
				This follows from (a).

				\item[$\SchedUpdate(\SchedState_0, \BufProject{\buf_1}) = \SchedUpdate(\SchedState_0', \BufProject{\buf_1'})$:]
				From (a), we have   $\SchedState_0 = \SchedState_0'$.
				From ${\buf_1} = {\buf_1'}$, we have $\BufProject{\buf_1} = \BufProject{\buf_1'}$.
				Therefore, $\SchedUpdate(\SchedState_0, \BufProject{\buf_1}) = \SchedUpdate(\SchedState_0', \BufProject{\buf_1'})$.
			\end{description}
			Therefore, $C_1 \sim C_1'$ and (c) holds.

			\item[$\CacheAccess(\CacheState_0,  \apply{\buf_0}{a_0}(\pc)) = \CacheHit \wedge p(\apply{\buf_0}{a_0}(\pc)) = \pjmp{e}$:] 
			The proof of this case is similar to that of the case $\CacheAccess(\CacheState_0,  \apply{\buf_0}{a_0}(\pc)) = \CacheHit \wedge p(\apply{\buf_0}{a_0}(\pc)) = \pjz{x}{\lbl}$.
			
			\item[$\CacheAccess(\CacheState_0,  \apply{\buf_0}{a_0}(\pc)) = \CacheHit \wedge p(\apply{\buf_0}{a_0}(\pc)) \neq \pjz{x}{\lbl} \wedge p(\apply{\buf_0}{a_0}(\pc)) \neq \pjmp{e}$:]
			Observe that from (a) it follows that $|\buf_0| \geq \wMuarch-1$ iff $|\buf_0'| \geq \wMuarch-1$.
			Therefore, if $|\buf_0| \geq \wMuarch-1$, then (c) holds since both computations are stuck.
			In the following, we assume that $|\buf_0| < \wMuarch-1$ and $|\buf_0'| < \wMuarch-1$.
			
			From (a), we get that $\CacheState_0' = \CacheState_0$, $\buf_0 = \buf_0'$, and $a_0 = a_0'$.
			Therefore, $\CacheAccess(\CacheState_0',  \apply{\buf_0'}{a_0'}(\pc)) = \CacheHit$.
			Moreover, $p(\apply{\buf_0'}{a_0'}(\pc)) \neq \pjz{x}{\lbl} \wedge p(\apply{\buf_0'}{a_0'}(\pc)) \neq \pjmp{e}$ as well.
			Therefore, we can apply the \textsc{Fetch-Others-Hit} and \textsc{Step-Others} rules to $C_0$ and $C_0'$ as follows:
			\begin{align*}
				v &:= \apply{\buf_0}{a_0}(\pc) +1 \\
				v' &:= \apply{\buf_0'}{a_0'}(\pc) +1 \\
				\buf_1 &:= \buf_0  \concat \tagged{p(\apply{\buf_0}{a_0}(\pc))}{\notags} \concat \tagged{\pmarkedassign{\pc}{v}}{\notags}\\
				\buf_1' &:= \buf_0'  \concat \tagged{p(\apply{\buf_0'}{a_0'}(\pc))}{\notags} \concat \tagged{\pmarkedassign{\pc}{v'}}{\notags}\\
				\tup{m_0,a_0,\buf_0,\CacheState_0,\BpState_0} &\muarchStep{\fetch{}}{} \tup{m_0, a_0, \buf_1, \CacheUpdate(\CacheState_0, \apply{\buf_0}{a_0}(\pc)),\BpState_0}\\
				\tup{m_0,a_0,\buf_0,\CacheState_0,\BpState_0, \SchedState_0} &\LoadDelayMuarchStep{}{} \tup{m_0, a_0, \buf_1, \CacheUpdate(\CacheState_0, \apply{\buf_0}{a_0}(\pc)),\BpState_0, \SchedUpdate(\SchedState_0, \BufProject{\buf_1})}\\
				\tup{m_0',a_0',\buf_0',\CacheState_0',\BpState_0'} &\muarchStep{\fetch{}}{} \tup{m_0', a_0', \buf_1', \CacheUpdate(\CacheState_0',  \apply{\buf_0'}{a_0'}(\pc)),\BpState_0'}\\
				\tup{m_0',a_0',\buf_0',\CacheState_0',\BpState_0', \SchedState_0'} & \LoadDelayMuarchStep{}{} \tup{m_0', a_0', \buf_1', \CacheUpdate(\CacheState_0',  \apply{\buf_0'}{a_0'}(\pc)),\BpState_0', \SchedUpdate(\SchedState_0', \BufProject{\buf_1'})}
			\end{align*}
			We now show that $C_1 =  \tup{m_0, a_0, \buf_1, \CacheUpdate(\CacheState_0, \apply{\buf_0}{a_0}(\pc)),\BpState_0, \SchedUpdate(\SchedState_0, \BufProject{\buf_1})}$ and $C_1' = \tup{m_0', a_0', \buf_1', \CacheUpdate(\CacheState_0',  \apply{\buf_0'}{a_0'}(\pc)),\BpState_0', \SchedUpdate(\SchedState_0', \BufProject{\buf_1'})}$ are indistinguishable, i.e., $C_1 \sim C_1'$.
			For this, we need to show that:
            \begin{description}
                \item[$a_0 = a_0'$:]
                This follows from (a). 

				\item[${\buf_1} = {\buf_1'}$:] 
				This follows from ${\buf_0} = {\buf_0'}$, which in turn follows from (a), $\buf_1 = \buf_0  \concat \tagged{p(\apply{\buf_0}{a_0}(\pc))}{\notags} \concat \tagged{\pmarkedassign{\pc}{v}}{\notags}$, $\buf_1' = \buf_0'  \concat \tagged{p(\apply{\buf_0'}{a_0'}(\pc))}{\notags} \concat \tagged{\pmarkedassign{\pc}{v'}}{\notags}$, $\apply{\buf_0}{a_0}(\pc) = \apply{\buf_0'}{a_0'}(\pc)$, and $v = v'$.

				\item[$\CacheUpdate(\CacheState_0,  \apply{\buf_0}{a_0}(\pc)) = \CacheUpdate(\CacheState_0',  \apply{\buf_0'}{a_0'}(\pc))$:]
				This follows from $\CacheState_0 = \CacheState_0'$, which in turn follows from (a), and $\apply{\buf_0}{a_0}(\pc) = \apply{\buf_0'}{a_0'}(\pc)$.

				\item[$\BpState_0 = \BpState_0':$] 
				This follows from (a).

				\item[$\SchedUpdate(\SchedState_0, \BufProject{\buf_1}) = \SchedUpdate(\SchedState_0', \BufProject{\buf_1'})$:]
				From (a), we have   $\SchedState_0 = \SchedState_0'$.
				From ${\buf_1} = {\buf_1'}$, we have $\BufProject{\buf_1} = \BufProject{\buf_1'}$.
				Therefore, $\SchedUpdate(\SchedState_0, \BufProject{\buf_1}) = \SchedUpdate(\SchedState_0', \BufProject{\buf_1'})$.
			\end{description}
			Therefore, $C_1 \sim C_1'$ and (c) holds.

			\item[$\CacheAccess(\CacheState_0,  \apply{\buf_0}{a_0}(\pc)) = \CacheMiss$:]
			From (a), we get that $\CacheState_0' = \CacheState_0$ and $\apply{\buf_0'}{a_0'}(\pc) = \apply{\buf_0}{a_0}(\pc)$.
			Therefore, $\CacheAccess(\CacheState_0',  \apply{\buf_0'}{a_0'}(\pc)) = \CacheMiss$.
			Therefore, we can apply the \textsc{Fetch-Miss} and \textsc{Step-Others} rules to $C_0$ and $C_0'$ as follows:
			\begin{align*}
				\tup{m_0,a_0,\buf_0,\CacheState_0,\BpState_0} &\muarchStep{\fetch{}}{} \tup{m_0, a_0, \buf_0, \CacheUpdate(\CacheState_0, \apply{\buf_0}{a_0}(\pc)),\BpState_0}\\
				\tup{m_0,a_0,\buf_0,\CacheState_0,\BpState_0, \SchedState_0} &\LoadDelayMuarchStep{}{} \tup{m_0, a_0, \buf_0, \CacheUpdate(\CacheState_0,  \apply{\buf_0}{a_0}(\pc)),\BpState_0, \SchedUpdate(\SchedState_0, \BufProject{\buf_0})}\\
				\tup{m_0',a_0',\buf_0',\CacheState_0',\BpState_0'} &\muarchStep{\fetch{}}{} \tup{m_0', a_0', \buf_0', \CacheUpdate(\CacheState_0',  \apply{\buf_0'}{a_0'}(\pc)),\BpState_0'}\\
				\tup{m_0',a_0',\buf_0',\CacheState_0',\BpState_0', \SchedState_0'} &\LoadDelayMuarchStep{}{} \tup{m_0', a_0', \buf_0', \CacheUpdate(\CacheState_0',  \apply{\buf_0'}{a_0'}(\pc)),\BpState_0', \SchedUpdate(\SchedState_0', \BufProject{\buf_0'})}
			\end{align*}
			We now show that $C_1 = \tup{m_0, a_0, \buf_0, \CacheUpdate(\CacheState_0,  \apply{\buf_0}{a_0}(\pc)),\BpState_0, \SchedUpdate(\SchedState_0, \BufProject{\buf_0})}$ and $C_1' = \tup{m_0', a_0', \buf_0', \CacheUpdate(\CacheState_0',  \apply{\buf_0'}{a_0'}(\pc)),\BpState_0', \SchedUpdate(\SchedState_0', \BufProject{\buf_0'})}$ are indistinguishable, i.e., $C_1 \sim C_1'$.
			For this, we need to show that:
            \begin{description}
                \item[$a_0 = a_0'$:]
                This follows from (a). 

				\item[${\buf_0} = {\buf_0'}$:] 
				This follows from (a).

				\item[$\CacheUpdate(\CacheState_0,  \apply{\buf_0}{a_0}(\pc)) = \CacheUpdate(\CacheState_0',  \apply{\buf_0'}{a_0'}(\pc))$:]
				This follows from $\CacheState_0 = \CacheState_0'$, which in turn follows from (a), and $\apply{\buf_0}{a_0}(\pc) = \apply{\buf_0'}{a_0'}(\pc)$.

				\item[$\BpState_0 = \BpState_0':$] 
				This follows from (a).

				\item[$\SchedUpdate(\SchedState_0, \BufProject{\buf_0}) = \SchedUpdate(\SchedState_0', \BufProject{\buf_0'})$:]
				From (a), we have   $\SchedState_0 = \SchedState_0'$.
				From ${\buf_1} = {\buf_1'}$, we have $\BufProject{\buf_1} = \BufProject{\buf_1'}$.
				Therefore, $\SchedUpdate(\SchedState_0, \BufProject{\buf_1}) = \SchedUpdate(\SchedState_0', \BufProject{\buf_1'})$.
			\end{description}
			Therefore, $C_1 \sim C_1'$ and (c) holds.
		\end{description}
		
		\item[$\apply{\buf_0}{a_0}(\pc) = \bot \vee |\buf_0| \geq \wMuarch$:]
		Then, from (a), we immediately get that $\apply{\buf_0'}{a_0'}(\pc) = \bot \vee |\buf_0'| \geq \wMuarch$ holds as well.
		Therefore, both computations are stuck and (c) holds.
	\end{description}
    Therefore, (c) holds for all the cases.
    
    \item[$d = \execute{i}:$]
    Therefore, we can only apply one of the $\execute{}$ rules.
    There are two cases:
    \begin{description}
        \item[$i \leq |\buf_0| \wedge \pbarrier \not\in {\buf_0[0..i-1]}$:]
        There are several cases depending on the $i$-th command in the reorder buffer:
        \begin{description}
            \item[$\elt{\buf_0}{i} = \tagged{\pload{x}{e}}{T}$:]   
            From (a), we have $a_0 = a_0'$ and $\buf_0 = \buf_0'$.
            Therefore, we also have that $\elt{\buf_0'}{i} =  \tagged{\pload{x}{e}}{T}$, $i \leq |\buf_0'|$, and $\pbarrier \not\in \buf_0'[0..i-1]$. 
            There are two cases:
            \begin{description}
                \item[$\forall \tagged{\passign{\pc}{\ell}}{\ell'} \in \buf_0{[0..i-1]}.\ \ell' = \notags$: ] 
                Then, from $\buf_0 = \buf_0'$ we also have that $\forall \tagged{\passign{\pc}{\ell}}{\ell'} \in \buf_0'{[0..i-1]}.\ \ell' = \notags$.
                There are two cases:
                \begin{description}
                    \item[$\pstore{x'}{e'} \not\in \buf_0{[0..i-1]}$:]
                    Observe that $\exprEval{e}{\apply{\buf_0[0..i-1]}{a_0}} = \exprEval{e}{\apply{\buf_0'[0..i-1]}{a_0'}}$ follows from $a_0 = a_0'$ and $\buf_0 = \buf_0'$, which follow from (a).
                    In the following, we denote $\exprEval{e}{\apply{\buf_0[0..i-1]}{a_0}}$ as $n$.
                    We now show that $m_0(n) = m_0'(n')$.
                    Since $C_0, C_0'$ are reachable configurations, the buffers $\buf_0, \buf_0'$ are well-formed (see Lemma~\ref{lemma:loadDelay:arch-seq:buffers-well-formedness}).
                    Therefore,  $\buf_0[0..i-1] \in \prefixes{\buf_0, \tup{m_0,a_0}}$ and  $\buf_0'[0..i-1] \in \prefixes{\buf_0', \tup{m_0,a_0}}$ because there are no unresolved branch instructions in $\buf_0[0..i-1]$ and $\buf_0'[0..i-1]$ (because $\forall \tagged{\passign{\pc}{\ell}}{\ell'} \in \buf_0{[0..i-1]}.\ \ell' = \notags$ and $\forall \tagged{\passign{\pc}{\ell}}{\ell'} \in \buf_0'{[0..i-1]}.\ \ell' = \notags$).
                    From (b), therefore, there are configurations $\sigma_0, \sigma_0', \sigma_1, \sigma_1'$ such that $C_0 \bufEquiv{|\buf_0[0..i-1]|} \sigma_0$, $C_0' \bufEquiv{|\buf_0[0..i-1]|} \sigma_0'$, $\sigma_0 \ArchSeqInterfStep{\tau}{} \sigma_1$,  $\sigma_0' \ArchSeqInterfStep{\tau'}{} \sigma_1'$, and $\tau = \tau'$. 
                    From $C_0 \bufEquiv{|\buf_0[0..i-1]|} \sigma_0$, $C_0' \bufEquiv{|\buf_0[0..i-1]|} \sigma_0'$, and the well-formedness of the buffers, we know that $p(\sigma_0(\pc)) = p(\sigma_0'(\pc)) = \pload{x}{e}$.
                    From $\ArchSpecInterf{\cdot}$, we have that $\tau = \loadObs{ \exprEval{e}{\sigma_0} = \sigma_0(\exprEval{e}{\sigma_0}) }$  and $\tau' = \loadObs{ \exprEval{e}{\sigma_0'}  = \sigma_0'(\exprEval{e}{\sigma_0'}) }$.
                    From $\tau=\tau'$, we get that $\exprEval{e}{\sigma_0} = \exprEval{e}{\sigma_0'}$ and $\sigma_0'(\exprEval{e}{\sigma_0'}) = \sigma_0(\exprEval{e}{\sigma_0})$.
                    From this, $C_0 \bufEquiv{|\buf_0[0..i-1]|} \sigma_0$, and $C_0' \bufEquiv{|\buf_0'[0..i-1]|} \sigma_0'$, we finally get $\exprEval{e}{\apply{\buf_0{[0..i-1]}}{a_0}} = \exprEval{e}{\apply{\buf_0'{[0..i-1]}}{a_0'}(x)} = n$ and $m_0(n) = m_0'(n)$.

                    There are two cases:
                    \begin{description}
                        \item[$\CacheAccess(\CacheState_0, \exprEval{e}{ \apply{\buf_0{[0..i-1]}}{a_0} }) = \CacheHit$:]
                        From (a), we have that $\CacheState_0 = \CacheState_0'$.
                        Moreover, we have already shown that $\exprEval{e}{ \apply{\buf_0{[0..i-1]}}{a_0} } = \exprEval{e}{ \apply{\buf_0'{[0..i-1]}}{a_0'} }$.
                        Therefore, we can apply the \textsc{Execute-Load-Hit} and \textsc{Step-Eager-Delay} rules to $C_0$ and $C_0'$ as follows:
                        \begin{align*}
                            \buf_0 &:= \buf_0[0..i-1] \concat \tagged{\pload{x}{e}}{T} \concat \buf_0[i+1 .. |\buf_0|]\\
                            \buf_1 &:= \buf_0[0..i-1] \concat \tagged{\passign{x}{m_0(n)}}{T} \concat \buf_0[i+1 .. |\buf_0|]\\
                            \tup{m_0,a_0,\buf_0, \CacheState_0, \BpState_0} &\muarchStep{\execute{i}}{} \tup{m_0,a_0,\buf_1, \CacheUpdate(\CacheState_0,n), \BpState_0}\\
                            \tup{m_0,a_0,\buf_0, \CacheState_0,  \BpState_0 , \SchedState_0} &\LoadDelayMuarchStep{}{} \tup{m_0,a_0,\buf_1, \CacheUpdate(\CacheState_0,n), \BpState_0,\SchedUpdate(\SchedState_0, \BufProject{\buf_0})}\\
                            \buf_0' &:= \buf_0'[0..i-1] \concat \tagged{\pload{x}{e}}{T} \concat \buf_0'[i+1 .. |\buf_0|]\\
                            \buf_1' &:= \buf_0'[0..i-1] \concat \tagged{\passign{x}{m_0'(n)}}{T} \concat \buf_0'[i+1 .. |\buf_0|]\\
                            \tup{m_0',a_0',\buf_0', \CacheState_0', \BpState_0'} &\muarchStep{\execute{i}}{} \tup{m_0',a_0',\buf_1', \CacheUpdate(\CacheState_0',n), \BpState_0'}\\
                            \tup{m_0',a_0',\buf_0', \CacheState_0', \BpState_0', \SchedState_0'} & \LoadDelayMuarchStep{}{} \tup{m_0',a_0',\buf_1', \CacheUpdate(\CacheState_0',n), \BpState_0',\SchedUpdate(\SchedState_0', \BufProject{\buf_0'})}
                        \end{align*}
                        We now show that $C_1 = \tup{m_0,a_0,\buf_1, \CacheUpdate(\CacheState_0,n), \BpState_0,\SchedUpdate(\SchedState_0, \BufProject{\buf_0})}$ and $C_1' = \tup{m_0',a_0',\buf_1', \CacheUpdate(\CacheState_0',n), \BpState_0',\SchedUpdate(\SchedState_0', \BufProject{\buf_0'})}$ are indistinguishable, i.e., i.e., $C_1 \sim C_1'$.
            
                        For this, we need to show that:
                            \begin{description}
                                \item[$a_0 = a_0'$:]
                                This follows from (a). 
                    
                                \item[${\buf_1} = {\buf_1'}$:] 
                                This immediately follows from (a) and $m_0(n) = m_0'(n)$. 
                    
                                \item[$\CacheUpdate(\CacheState_0,n) = \CacheUpdate(\CacheState_0',n)$:]
                                This follows from $\CacheState_0 = \CacheState_0'$, which follows from (a).
                    
                                \item[$\BpState_0= \BpState_0':$] 
                                This follows from  (a).
                    
                                \item[$\SchedUpdate(\SchedState_0, \BufProject{\buf_1}) = \SchedUpdate(\SchedState_0', \BufProject{\buf_1'})$:]
                                From (a), we have   $\SchedState_0 = \SchedState_0'$.
                                From ${\buf_1} = {\buf_1'}$, we have $\BufProject{\buf_1} = \BufProject{\buf_1'}$.
                                Therefore, $\SchedUpdate(\SchedState_0, \BufProject{\buf_1}) = \SchedUpdate(\SchedState_0', \BufProject{\buf_1'})$.
                            \end{description}
                            Therefore, $C_1 \sim C_1'$ and (c) holds.

                        \item[$\CacheAccess(\CacheState_0, \exprEval{e}{\apply{\buf_0{[0..i-1]}}{a_0}}) = \CacheMiss$:]
                        The proof of this case is similar to the one for the $\CacheHit$ case (except that we apply the \textsc{Execute-Load-Miss} rule and we do not need to rely on observations produced by the $\ArchSeqInterf{\cdot}$ contract). 
                    \end{description}

                    \item[$\pstore{x'}{e'} \in \buf_0{[0..i-1]}$:]
                    From (a), we also have that $\pstore{x'}{e'} \in \buf_0'[0..i-1]$.
                    Therefore, both computations are stuck and (c) holds.

                \end{description}

                \item [$\neg \forall \tagged{\passign{\pc}{\ell}}{\ell'} \in \buf_0{[0..i-1]}.\ \ell' = \notags$: ]
                Then, $\neg \forall \tagged{\passign{\pc}{\ell}}{\ell'} \in \buf_0'{[0..i-1]}.\ \ell' = \notags$ holds as well (from $\buf_0 = \buf_0'$).
                Then, (c) holds since both configurations are stuck.
            \end{description}

            \item[$\elt{\buf_0}{i} =  \tagged{\passign{\pc}{\lbl}}{\lbl_0} \wedge \ell_0 \neq \emptysequence$:]
            From (a), we have that $a_0 = a_0'$ and $\buf_0 = \buf_0'$.
            Therefore, we have that $\elt{\buf_0'}{i} =  \tagged{\passign{\pc}{\lbl}}{\lbl_0} \wedge \ell_0 \neq \emptysequence$, $i \leq |\buf_0'|$, and $\pbarrier \not\in \buf_0'[0..i-1]$. 
            Observe that $p(\lbl_0) = \pjz{x}{\lbl''}$.
            Given that $\apply{\buf_0[0..i-1]}{a_0}(x) = \apply{\buf_0'[0..i-1]}{a_0'}(x)$ immediately follows from $a_0 = a_0'$ and $\buf_0 = \buf_0'$, there are two cases:
            \begin{description}
                \item[$( \apply{\buf_0[0..i-1]}{a_0}(x) = 0 \wedge \lbl = \lbl'') \vee (\apply{\buf_0[0..i-1]}{a_0}(x) \in \Val \setminus \{0,\bot\} \wedge \lbl = \ell_0+1)$:]
                From $\apply{\buf_0{[0..i-1]}}{a_0}(x) = \apply{\buf_0'{[0..i-1]}}{a_0'}(x)$ and (a), we also get $( \apply{\buf_0'[0..i-1]}{a_0'}(x) = 0 \wedge \lbl = \lbl'') \vee (\apply{\buf_0'[0..i-1]}{a_0'}(x) \in \Val \setminus \{0,\bot\} \wedge \lbl = \ell_0+1)$.
                Therefore,  we can apply the \textsc{Execute-Branch-Commit} and \textsc{Step-Others} rules to $C_0$ and $C_0'$ as follows:
                \begin{align*}
                    \buf_0 &:= \buf_0[0..i-1] \concat \tagged{\passign{\pc}{\ell}}{\ell_0} \concat \buf_0[i+1 .. |\buf_0|]\\
                    \buf_1 &:= \buf_0[0..i-1] \concat \tagged{\passign{\pc}{\ell}}{\notags} \concat \buf_0[i+1 .. |\buf_0|]\\
                    \tup{m_0,a_0,\buf_0, \CacheState_0, \BpState_0} &\muarchStep{\execute{i}}{} \tup{m_0,a_0,\buf_1, \CacheState_0, \BpState_0}\\
                    \tup{m_0,a_0,\buf_0, \CacheState_0,  \BpUpdate(\BpState_0, \ell_0, \ell) , \SchedState_0} &\LoadDelayMuarchStep{}{} \tup{m_0,a_0,\buf_1, \CacheState_0, \BpUpdate(\BpState_0, \ell_0, \ell),\SchedUpdate(\SchedState_0, \BufProject{\buf_0})}\\
                    \buf_0' &:= \buf_0'[0..i-1] \concat \tagged{\passign{\pc}{\ell}}{\ell_0} \concat \buf_0'[i+1 .. |\buf_0|]\\
                    \buf_1' &:= \buf_0'[0..i-1] \concat \tagged{\passign{\pc}{\ell}}{\notags} \concat \buf_0'[i+1 .. |\buf_0|]\\
                    \tup{m_0',a_0',\buf_0', \CacheState_0', \BpUpdate(\BpState_0', \ell_0, \ell)} &\muarchStep{\execute{i}}{} \tup{m_0',a_0',\buf_1', \CacheState_0', \BpState_0'}\\
                    \tup{m_0',a_0',\buf_0', \CacheState_0', \BpState_0', \SchedState_0'} & \LoadDelayMuarchStep{}{} \tup{m_0',a_0',\buf_1', \CacheState_0', \BpUpdate(\BpState_0', \ell_0, \ell),\SchedUpdate(\SchedState_0', \BufProject{\buf_0'})}
                \end{align*}
                We now show that $C_1 = \tup{m_0,a_0,\buf_1, \CacheState_0, \BpUpdate(\BpState_0, \ell_0, \ell),\SchedUpdate(\SchedState_0, \BufProject{\buf_0})}$ and $C_1' = \tup{m_0',a_0',\buf_1', \CacheState_0', \BpUpdate(\BpState_0', \ell_0, \ell),\SchedUpdate(\SchedState_0', \BufProject{\buf_0'})}$ are indistinguishable, i.e., i.e., $C_1 \sim C_1'$.

                For this, we need to show that:
                    \begin{description}
                        \item[$a_0 = a_0'$:]
                        This follows from (a). 
            
                        \item[${\buf_1} = {\buf_1'}$:] 
                        This follows from (a) and from that $\apply{\buf_0[0..i-1]}{a_0}(x) = \apply{\buf_0'[0..i-1]}{a_0'}(x)$.

                        \item[$\CacheState_0 = \CacheState_0'$:]
                        This follows from (a).
            
                        \item[$\BpUpdate(\BpState_0, \ell_0, \ell) = \BpUpdate(\BpState_0', \ell_0, \ell):$] 
                        This follows from $\BpState_0 = \BpState_0'$, which follows from (a).
            
                        \item[$\SchedUpdate(\SchedState_0, \BufProject{\buf_1}) = \SchedUpdate(\SchedState_0', \BufProject{\buf_1'})$:]
                        From (a), we have   $\SchedState_0 = \SchedState_0'$.
                        From ${\buf_1} = {\buf_1'}$, we have $\BufProject{\buf_1} = \BufProject{\buf_1'}$.
                        Therefore, $\SchedUpdate(\SchedState_0, \BufProject{\buf_1}) = \SchedUpdate(\SchedState_0', \BufProject{\buf_1'})$.
                    \end{description}
                    Therefore, $C_1 \sim C_1'$ and (c) holds.

                \item[$( \apply{\buf_0[0..i-1]}{a_0}(x) = 0 \wedge \lbl \neq \lbl'') \vee (\apply{\buf_0[0..i-1]}{a_0}(x) \in \Val \setminus \{0,\bot\} \wedge \lbl \neq \ell_0+1)$:]
                The proof of this case is similar to the one of the $( \apply{\buf_0[0..i-1]}{a_0}(x) = 0 \wedge \lbl = \lbl'') \vee (\apply{\buf_0[0..i-1]}{a_0}(x) \in \Val \setminus \{0,\bot\} \wedge \lbl = \ell_0+1)$ (except that we apply the \textsc{Execute-Branch-Rollback} rule).
            \end{description}

            \item[$\elt{\buf_0}{i} = \tagged{\passign{x}{e}}{\notags}$:]   
            From (a), we have that $\buf_0 = \buf_0'$ and $a_0 = a_0'$.
            Therefore, we also have that $\elt{\buf_0'}{i} =  \tagged{\passign{x}{e}}{\notags}$, $i \leq |\buf_0'|$, and $\pbarrier \not\in \buf_0'[0..i-1]$.
            There are two cases:
            \begin{description}
                \item[$\exprEval{e}{\apply{\buf_0[0..i-1]}{a_0}} \neq \bot$:]
                From (a), we also have that $\exprEval{e}{\apply{\buf_0'[0..i-1]}{a_0'}} \neq \bot$ (again from $\buf_0 = \buf_0'$ and $a_0 = a_0'$ from (a)).
                Therefore,  we can apply the \textsc{Execute-Assignment} and \textsc{Step-Others} rules to $C_0$ and $C_0'$ as follows:
                \begin{align*}
                    v &:= \exprEval{e}{\apply{\buf_0{[0..i-1]}}{a_0}}\\
                    \buf_0 &:= \buf_0[0..i-1] \concat \tagged{\passign{x}{e}}{T} \concat \buf_0[i+1 .. |\buf_0|]\\
                    \buf_1 &:= \buf_0[0..i-1] \concat \tagged{\passign{x}{v}}{T} \concat \buf_0[i+1 .. |\buf_0|]\\
                    \tup{m_0,a_0,\buf_0, \CacheState_0, \BpState_0} &\muarchStep{\execute{i}}{} \tup{m_0,a_0,\buf_1, \CacheState_0, \BpState_0}\\
                    \tup{m_0,a_0,\buf_0, \CacheState_0, \BpState_0, \SchedState_0} &\LoadDelayMuarchStep{}{} \tup{m_0,a_0,\buf_1, \CacheState_0, \BpState_0,\SchedUpdate(\SchedState_0, \BufProject{\buf_0})}\\
                    v' &:= \exprEval{e}{\apply{\buf_0'{[0..i-1]}}{a_0'}}\\
                    \buf_0' &:= \buf_0'[0..i-1] \concat \tagged{\passign{x}{e}}{T} \concat \buf_0'[i+1 .. |\buf_0|]\\
                    \buf_1' &:= \buf_0'[0..i-1] \concat \tagged{\passign{x}{v'}}{T} \concat \buf_0'[i+1 .. |\buf_0|]\\
                    \tup{m_0',a_0',\buf_0', \CacheState_0', \BpState_0'} &\muarchStep{\execute{i}}{} \tup{m_0',a_0',\buf_1', \CacheState_0', \BpState_0'}\\
                    \tup{m_0',a_0',\buf_0', \CacheState_0', \BpState_0', \SchedState_0'} & \LoadDelayMuarchStep{}{} \tup{m_0',a_0',\buf_1', \CacheState_0', \BpState_0',\SchedUpdate(\SchedState_0', \BufProject{\buf_0'})}
                \end{align*}
                We now show that $C_1 = \tup{m_0,a_0,\buf_1, \CacheState_0, \BpState_0,\SchedUpdate(\SchedState_0, \BufProject{\buf_0})}$ and $C_1' = \tup{m_0',a_0',\buf_1', \CacheState_0', \BpState_0',\SchedUpdate(\SchedState_0', \BufProject{\buf_0'})}$ are indistinguishable, i.e., i.e., $C_1 \sim C_1'$.
                
                For this, we need to show that:
                \begin{description}
                    \item[$a_0 = a_0'$:] 
                    This follows from (a).

                    \item[$\buf_1 = \buf_1'$:]
                    Observe that $v = v'$ directly follows from $v = \exprEval{e}{\apply{\buf_0{[0..i-1]}}{a_0}}$, $v' = \exprEval{e}{\apply{\buf_0'{[0..i-1]}}{a_0'}}$, $a_0 = a_0'$, and $\buf_0 = \buf_0'$ (the last two points follow from (a)).
                    Therefore, $\buf_1 = \buf_1'$ follows from $\buf_1 = \buf_0[0..i-1] \concat \tagged{\passign{x}{v}}{T} \concat \buf_0[i+1 .. |\buf_0|]$, $\buf_1' = \buf_0'[0..i-1] \concat \tagged{\passign{x}{v'}}{T} \concat \buf_0'[i+1 .. |\buf_0|]$, $v= v'$, and $\buf_0 = \buf_0'$.
                    
                    \item[$\BpState_0 = \BpState_0':$] 
                    This follows from (a).
        
                    \item[$\SchedUpdate(\SchedState_0, \BufProject{\buf_1}) = \SchedUpdate(\SchedState_0', \BufProject{\buf_1'})$:]
                    From (a), we have   $\SchedState_0 = \SchedState_0'$.
                    From ${\buf_1} = {\buf_1'}$, we have $\BufProject{\buf_1} = \BufProject{\buf_1'}$.
                    Therefore, $\SchedUpdate(\SchedState_0, \BufProject{\buf_1}) = \SchedUpdate(\SchedState_0', \BufProject{\buf_1'})$.
                \end{description}
                Therefore, $C_1 \sim C_1'$ and (c) holds.

                \item[$\exprEval{e}{\apply{\buf_0[0..i-1]}{a_0}} = \bot$:] 
                From this, it follows that $e \not\in \Val$.
                Therefore, from (a), we have that $e = e'$, $\buf_0 = \buf_0'$, and $a_0 = a_0'$.
                From $\exprEval{e}{\apply{\buf[0..i-1]}{a_0}} = \bot$, we therefore get that $\exprEval{e}{\apply{\buf_0'{[0..i-1]}}{a_0'}} = \bot $ holds as well.
                Hence, both configurations are stuck and (c) holds.

            \end{description}

            \item[$\elt{\buf_0}{i} =   \tagged{\pmarkedassign{x}{e}}{\notags}$:]
            The proof of this case is similar to that of $\elt{\buf_0}{i} = \tagged{\passign{x}{e}}{\notags}$.

            \item[$\elt{\buf_0}{i} =  \tagged{\pstore{x}{e}}{T}$:]   
            The proof of this case is similar to that of $\elt{\buf_0}{i} = \tagged{\passign{x}{e}}{\notags}$. 
        
            \item[$\elt{\buf_0}{i} =  \tagged{\pskip{}}{\notags}$:]
            The proof of this case is similar to that of $\elt{\buf_0}{i} =  \tagged{\passign{x}{e}}{\notags}$.
            
            \item[$\elt{\buf_0}{i} =  \tagged{\pbarrier}{T}$:]
            The proof of this case is similar to that of $\elt{\buf_0}{i} =  \tagged{\passign{x}{e}}{\notags}$.
        \end{description}
        Therefore, (c) holds in all cases.

        \item[$i > |\buf_0| \vee \pbarrier \in \buf_0{[0..i-1]}$:]
        From (a), it immediately follows that $i > |\buf_0'| \vee \pbarrier \in \buf_0'[0..i-1]$.
        Therefore, both configurations are stuck and (c) holds.
    \end{description}
    Therefore, (c) holds in all cases.

    \item[$d = \retire{}$:]
	Therefore, we can only apply one of the $\retire{}$ rules depending on the head of the reorder buffer in $\buf_0$.
	There are five cases:
	\begin{description}
		\item[$\buf_0 = \tagged{\pskip}{\notags} \concat \buf_1 $:] 
		From (a), we get that $\buf_0 = \buf_0'$.
		Therefore, we have that $\buf_0' = \tagged{\pskip}{\notags} \concat \buf_1' $ and $\buf_1 = \buf_1'$.
		Therefore, we can apply the \textsc{Retire-Skip} and \textsc{Step-Others} rules to $C_0$ and $C_0'$ as follows:
		\begin{align*}
			\tup{m_0,a_0,\tagged{\pskip}{\notags} \concat \buf_1,\CacheState_0,\BpState_0} &\muarchStep{\retire}{} \tup{m_0, a_0, \buf_1, \CacheState_0,\BpState_0}\\
			\tup{m_0,a_0,\tagged{\pskip}{\notags} \concat \buf_1,\CacheState_0,\BpState_0, \SchedState_0} &\LoadDelayMuarchStep{}{} \tup{m_0, a_0, \buf_1, \CacheState_0,\BpState_0, \SchedUpdate(\SchedState_0, \BufProject{\buf_1})}\\
			\tup{m_0',a_0',\tagged{\pskip}{\notags} \concat \buf_1',\CacheState_0',\BpState_0'} &\muarchStep{\retire}{} \tup{m_0', a_0', \buf_1', \CacheState_0',\BpState_0'}\\
			\tup{m_0',a_0',\tagged{\pskip}{\notags} \concat \buf_1',\CacheState_0',\BpState_0', \SchedState_0'} &\LoadDelayMuarchStep{}{} \tup{m_0', a_0', \buf_1', \CacheState_0',\BpState_0', \SchedUpdate(\SchedState_0', \BufProject{\buf_1'})}
		\end{align*}
		We now show that $C_1 = \tup{m_0, a_0, \buf_1, \CacheState_0,\BpState_0, \SchedUpdate(\SchedState_0, \BufProject{\buf_1})}$ and $C_1' = \tup{m_0', a_0', \buf_1', \CacheState_0',\BpState_0', \SchedUpdate(\SchedState_0', \BufProject{\buf_1'})}$ are indistinguishable, i.e., $C_1 \sim C_1'$.
		For this, we need to show that:
        \begin{description}
            \item[$a_0 = a_0'$:]
            This follows from (a).

			\item[$\buf_1 = \buf_1'$:] 
			This follows from $\buf_0 = \tagged{\pskip}{\notags} \concat \buf_1 $, $\buf_0' = \tagged{\pskip}{\notags} \concat \buf_1' $, and (a).

			\item[$\CacheState_0 = \CacheState_0'$:]
			This follows from (a).

			\item[$\BpState_0 = \BpState_0':$] 
			This follows from (a).

			\item[$\SchedUpdate(\SchedState_0, \BufProject{\buf_1}) = \SchedUpdate(\SchedState_0', \BufProject{\buf_1'})$:]
			From (a), we have   $\SchedState_0 = \SchedState_0'$.
			From ${\buf_1} = {\buf_1'}$, we have $\BufProject{\buf_1} = \BufProject{\buf_1'}$.
			Therefore, $\SchedUpdate(\SchedState_0, \BufProject{\buf_1}) = \SchedUpdate(\SchedState_0', \BufProject{\buf_1'})$.
		\end{description}
		Therefore, $C_1 \sim C_1'$ and (c) holds.

		\item[$\buf_0 = \tagged{\pbarrier}{\notags} \concat \buf_1 $:]
		The proof of this case is similar to that of the case $\buf_0 = \tagged{\pskip}{\notags} \concat \buf_1 $.

		\item[$\buf_0 = \tagged{\passign{x}{v}}{\notags} \concat \buf_1 $:] 
		From (a), we get that $\buf_0 = \buf_0'$.
		Therefore, we have that $\buf_0' = \tagged{\passign{x}{v'}}{\notags} \concat \buf_1' $ and ${\buf_1} = {\buf_1'}$.
		From ${\buf_0} = {\buf_0'}$, we also get that $v \in \Val$ iff $v' \in \Val$.
		Observe also that if $v \not\in \Val$ then both computations are stuck and (c) holds (since there is no $C_1$ such that $C_0 \LoadDelayMuarchStep{}{} C_1$ and no $C_1'$ such that $C_0' \LoadDelayMuarchStep{}{} C_1'$).
		In the following, therefore, we assume that $v,v' \in \Val$.
		Therefore, we can apply the \textsc{Retire-Assignment} and \textsc{Step-Others} rules to $C_0$ and $C_0'$ as follows:
		\begin{align*}
			\tup{m_0,a_0,\tagged{\passign{x}{v}}{\notags} \concat \buf_1,\CacheState_0,\BpState_0} &\muarchStep{\retire}{} \tup{m_0, a_0[x\mapsto v], \buf_1, \CacheState_0,\BpState_0}\\
			\tup{m_0,a_0,\tagged{\passign{x}{v}}{\notags} \concat \buf_1,\CacheState_0,\BpState_0, \SchedState_0} &\LoadDelayMuarchStep{}{} \tup{m_0, a_0[x \mapsto v], \buf_1, \CacheState_0,\BpState_0, \SchedUpdate(\SchedState_0, \BufProject{\buf_1})}\\
			\tup{m_0',a_0',\tagged{\passign{x}{v'}}{\notags} \concat \buf_1',\CacheState_0',\BpState_0'} &\muarchStep{\retire}{} \tup{m_0', a_0'[x \mapsto v'], \buf_1', \CacheState_0',\BpState_0'}\\
			\tup{m_0',a_0',\tagged{\passign{x}{v'}}{\notags} \concat \buf_1',\CacheState_0',\BpState_0', \SchedState_0'} &\LoadDelayMuarchStep{}{} \tup{m_0', a_0'[x \mapsto v'], \buf_1', \CacheState_0',\BpState_0', \SchedUpdate(\SchedState_0', \BufProject{\buf_1'})}
		\end{align*}
		We now show that $C_1 = \tup{m_0, a_0[x \mapsto v], \buf_1, \CacheState_0,\BpState_0, \SchedUpdate(\SchedState_0, \BufProject{\buf_1})}$ and $C_1' = \tup{m_0', a_0'[x \mapsto v'], \buf_1', \CacheState_0',\BpState_0', \SchedUpdate(\SchedState_0', \BufProject{\buf_1'})}$ are indistinguishable, i.e., $C_1 \sim C_1'$.
		For this, we need to show that:
        \begin{description}
            \item[$a_0{[x \mapsto v]} = a_0'{[x \mapsto v']}$:]
            From  ${\buf_0} = {\buf_0'}$, $\buf_0 = \tagged{\passign{x}{v}}{\notags} \concat \buf_1 $, and $\buf_0' = \tagged{\passign{x}{v'}}{\notags} \concat \buf_1' $, we get $v = v'$.
            From (a), we also have $a_0 = a_0'$.
            Therefore, $a_0[x \mapsto v] = a_0'[x \mapsto v']$.

			\item[${\buf_1} = {\buf_1'}$:] 
			This follows from $\buf_0 = \tagged{\passign{x}{v}}{\notags} \concat \buf_1 $, $\buf_0' = \tagged{\passign{x}{v'}}{\notags} \concat \buf_1' $, and (a).

			\item[$\CacheState_0 = \CacheState_0'$:]
			This follows from (a).

			\item[$\BpState_0 = \BpState_0':$] 
			This follows from (a).

			\item[$\SchedUpdate(\SchedState_0, \BufProject{\buf_1}) = \SchedUpdate(\SchedState_0', \BufProject{\buf_1'})$:]
			From (a), we have   $\SchedState_0 = \SchedState_0'$.
			From ${\buf_1} = {\buf_1'}$, we have $\BufProject{\buf_1} = \BufProject{\buf_1'}$.
			Therefore, $\SchedUpdate(\SchedState_0, \BufProject{\buf_1}) = \SchedUpdate(\SchedState_0', \BufProject{\buf_1'})$.
		\end{description}
		Therefore, $C_1 \sim C_1'$ and (c) holds.

		\item[$\buf_0 = \tagged{\pmarkedassign{x}{v}}{\notags} \concat \buf_1 $:]
		The proof fo this case is similar to that of the case $\buf_0 = \tagged{\passign{x}{v}}{\notags} \concat \buf_1 $.

		\item[$\buf_0 = \tagged{\pstore{v}{n}}{\notags} \concat \buf_1 $:] 
		From (a), we get that ${\buf_0} = {\buf_0'}$.
		Therefore, we have that $\buf_0' = \tagged{\pstore{v'}{n}}{\notags} \concat \buf_1' $ and ${\buf_1} = {\buf_1'}$.
		Observe that (1) $v \in \Val \leftrightarrow v' \in \Val$ from (a), and (2) if $v \not\in \Val$ or $n \not\in \Val$, then both computations are stuck and (c) holds (since there is no $C_1$ such that $C_0 \LoadDelayMuarchStep{}{} C_1$ and no $C_1'$ such that $C_0' \LoadDelayMuarchStep{}{} C_1'$).
		In the following, therefore, we assume that $v,v',n \in \Val$.
		Therefore, we can apply the \textsc{Retire-Store} and \textsc{Step-Others} rules to $C_0$ and $C_0'$ as follows:
		\begin{align*}
			\tup{m_0,a_0,\tagged{\pstore{v}{n}}{\notags} \concat \buf_1,\CacheState_0,\BpState_0} &\muarchStep{\retire}{} \tup{m_0[n\mapsto v], a_0, \buf_1, \CacheUpdate(\CacheState_0,n),\BpState_0}\\
			\tup{m_0,a_0,\tagged{\pstore{v}{n}}{\notags} \concat \buf_1,\CacheState_0,\BpState_0, \SchedState_0} &\LoadDelayMuarchStep{}{} \tup{m_0[n\mapsto v], a_0, \buf_1, \CacheUpdate(\CacheState_0,n),\BpState_0, \SchedUpdate(\SchedState_0, \BufProject{\buf_1})}\\
			\tup{m_0',a_0',\tagged{\pstore{v'}{n}}{\notags} \concat \buf_1',\CacheState_0',\BpState_0'} &\muarchStep{\retire}{} \tup{m_0'[n \mapsto v'], a_0', \buf_1', \CacheUpdate(\CacheState_0',n),\BpState_0'}\\
			\tup{m_0',a_0',\tagged{\pstore{v'}{n}}{\notags} \concat \buf_1',\CacheState_0',\BpState_0', \SchedState_0'} &\LoadDelayMuarchStep{}{} \tup{m_0'[n\mapsto v'], a_0', \buf_1', \CacheUpdate(\CacheState_0',n),\BpState_0', \SchedUpdate(\SchedState_0', \BufProject{\buf_1'})}
		\end{align*}
		We now show that $C_1 = \tup{m_0[n \mapsto v], a_0, \buf_1, \CacheUpdate(\CacheState_0,n),\BpState_0, \SchedUpdate(\SchedState_0, \BufProject{\buf_1})}$ and $C_1' =  \tup{m_0'[n\mapsto v'], a_0', \buf_1', \CacheUpdate(\CacheState_0',n),\BpState_0', \SchedUpdate(\SchedState_0', \BufProject{\buf_1'})}$ are indistinguishable, i.e., $C_1 \sim C_1'$.
		For this, we need to show that:
        \begin{description}
            \item[$a_0 = a_0'$:]
            This follows from (a).

			\item[${\buf_1} = {\buf_1'}$:] 
			This follows from $\buf_0 = \tagged{\pstore{v}{n}}{\notags} \concat \buf_1 $, $\buf_0' = \tagged{\pstore{v'}{n}}{\notags} \concat \buf_1' $, and (a).

			\item[$\CacheUpdate(\CacheState_0,n) = \CacheUpdate(\CacheState_0',n)$:]
			This follows from $\CacheState_0 = \CacheState_0'$, which, in turn, follows from (a).

			\item[$\BpState_0 = \BpState_0':$] 
			This follows from (a).

			\item[$\SchedUpdate(\SchedState_0, \BufProject{\buf_1}) = \SchedUpdate(\SchedState_0', \BufProject{\buf_1'})$:]
			From (a), we have   $\SchedState_0 = \SchedState_0'$.
			From ${\buf_1} = {\buf_1'}$, we have $\BufProject{\buf_1} = \BufProject{\buf_1'}$.
			Therefore, $\SchedUpdate(\SchedState_0, \BufProject{\buf_1}) = \SchedUpdate(\SchedState_0', \BufProject{\buf_1'})$.
		\end{description}
		Therefore, $C_1 \sim C_1'$ and (c) holds.
	\end{description}
	Therefore, (c) holds for all the cases.

\end{description}
Since (c) holds for all cases, this completes the proof of our lemma.
\end{proof}

\subsubsection{Main lemma}

\begin{definition}[Indistinguishability of hardware configurations]\label{def:loadDelay:arch-seq:indistinguishability}
	We say that two hardware configurations $\tup{\sigma,\mu} = \tup{m,a,\buf, \CacheState,\BpState, \SchedState}$ and $\tup{\sigma',\mu'} = \tup{m',a',\buf', \CacheState',\BpState', \SchedState'}$ are \emph{indistinguishable}, written $\tup{\sigma,\mu} \approx \tup{\sigma',\mu'}$, iff $\BufProject{\buf } = \BufProject{\buf'}$, $\CacheState = \CacheState'$, $\BpState = \BpState'$, and $\SchedState = \SchedState'$.
\end{definition}

\begin{lemma}[Deep-indistinguishability implies indistinguishability]\label{lemma:loadDelay:arch-seq:deep-indistinguishability-implies-indistinguishability}
Let $C, C'$ be hardware configurations.
If $C \sim C'$, then $C \approx C'$.
\end{lemma}

\begin{proof}
It  follows from Definitions~\ref{def:loadDelay:arch-seq:deep-indistinguishability} and~\ref{def:loadDelay:arch-seq:indistinguishability}.
\end{proof}

\begin{lemma}\label{lemma:loadDelay:arch-seq:main-lemma}
Let $p$ be a well-formed program, $\CacheState_0$ be the initial cache state, $\BpState_0$ be the initial branch predictor state, and $\SchedState_0$ be the initial scheduler state, $\sigma_0 = \tup{m_0,a_0}, \sigma_0' = \tup{m_0',a_0'}$ be initial \archstate{}s, and $C_0 = \tup{m_0,a_0, \emptysequence, \CacheState_0, \BpState_0, \SchedState_0}$ and $C_0' = \tup{m_0',a_0', \emptysequence, \CacheState_0, \BpState_0, \SchedState_0}$ be hardware configurations.
	Furthermore, let $\crun := \sigma_0$ $\ArchSeqInterfStep{o_1}{}$ $\sigma_1$ $\ArchSeqInterfStep{o_2}{}$ $\ldots$  $\ArchSeqInterfStep{o_{n-1}}{}$  $\sigma_n$ and $\crunp:=\sigma_0'$ $\ArchSeqInterfStep{o_1'}{}$  $\sigma_1' $ $\ArchSeqInterfStep{o_2'}{}$ $\ldots$  $\ArchSeqInterfStep{o_{n-1}'}{} $ $\sigma_n$ be two runs for the $\ArchSeqInterf{\cdot}$ contract where $\sigma_0, \sigma_0'$ are initial \archstate{}s and $\sigma_n, \sigma_{n}$ are final \archstate{}s.
	If $o_i = o_i'$ for all $0 < i < n$, then there is a $k \in \Nat$ and $C_1, \ldots, C_k, C_1', \ldots, C_k'$ such that $C_0 \LoadDelayMuarchStep{}{} C_1 \LoadDelayMuarchStep{}{} \ldots \LoadDelayMuarchStep{}{} C_k$, $C_0' \LoadDelayMuarchStep{}{} C_1' \LoadDelayMuarchStep{}{} \ldots \LoadDelayMuarchStep{}{} C_k'$, and one of the following conditions hold:
	\begin{compactenum}
		\item $C_0, C_0'$ are initial states, $\forall 0 \leq i \leq k.\ C_{i} \approx C_{i}'$, and $C_, C_k'$ are final states, or
		\item $C_0, C_0'$ are initial states, $\forall 0 \leq i \leq k.\ C_{i} \approx C_{i}'$, and there are no $C_{k+1}$ such that $C_{k+1} \neq C_k \wedge C_k \LoadDelayMuarchStep{}{} C_{k+1}$ and no $C_{k+1}'$ such that $C_{k+1}' \neq C_k' \wedge C_k' \LoadDelayMuarchStep{}{} C_{k+1}'$.
	\end{compactenum}
\end{lemma}

\begin{proof}
Let $p$ be a well-formed program, $\CacheState_0$ be the initial cache state, $\BpState_0$ be the initial branch predictor state, and $\SchedState_0$ be the initial scheduler state, $\sigma_0 = \tup{m_0,a_0}, \sigma_0' = \tup{m_0',a_0'}$ be initial \archstate{}s, and $C_0 = \tup{m_0,a_0, \emptysequence, \CacheState_0, \BpState_0, \SchedState_0}$ and $C_0' = \tup{m_0',a_0', \emptysequence, \CacheState_0, \BpState_0, \SchedState_0}$ be hardware configurations.
    Furthermore, let $\crun := \sigma_0$ $\ArchSeqInterfStep{o_1}{}$ $\sigma_1$ $\ArchSeqInterfStep{o_2}{}$ $\ldots$  $\ArchSeqInterfStep{o_{n-1}}{}$  $\sigma_n$ and $\crunp:=\sigma_0'$ $\ArchSeqInterfStep{o_1'}{}$  $\sigma_1' $ $\ArchSeqInterfStep{o_2'}{}$ $\ldots$  $\ArchSeqInterfStep{o_{n-1}'}{} $ $\sigma_n$ be two runs for the $\ArchSeqInterf{\cdot}$ contract where  $\sigma_n, \sigma_{n}$ are final \archstate{}s.
	Finally, we assume that $o_i = o_i'$ for all $0 < i < n$.

	Consider the two hardware runs $\hrun, \hrunp$ obtained as follows: we start from $C_0= \tup{m_0,a_0, \emptysequence, \CacheState_0, \BpState_0, \SchedState_0}$ and $C_0'= \tup{m_0',a_0', \emptysequence, \CacheState_0, \BpState_0, \SchedState_0}$, and we apply one step of $\LoadDelayMuarchStep{}{}$ until either $\hrun$ or $\hrunp$ reaches a final state or gets stuck.
	That is, for some $k \in \Nat$, we obtain the following runs:
	\begin{align*}
		C_0 &:=  \tup{m_0,a_0, \emptysequence, \CacheState_0, \BpState_0, \SchedState_0}\\
		C_0' &:= \tup{m_0',a_0', \emptysequence, \CacheState_0, \BpState_0, \SchedState_0}\\
		\hrun&:=C_0 \LoadDelayMuarchStep{}{} C_1 \LoadDelayMuarchStep{}{} \ldots  \LoadDelayMuarchStep{}{} C_k \\
		\hrunp&:=C_0' \LoadDelayMuarchStep{}{} C_1' \LoadDelayMuarchStep{}{} \ldots  \LoadDelayMuarchStep{}{} C_k' 
	\end{align*}
	Let $\map{\crun}{\hrun}{\cdot}$ and $\map{\crunp}{\hrunp}{\cdot}$ be the maps constructed according to Definition~\ref{def:loadDelay:arch-seq:mapping}.
	We claim that that for all $0 \leq i \leq k$, $\map{\crun}{\hrun}{i} = \map{\crunp}{\hrunp}{i}$ and $ C_i \sim C_{i}'$ which implies $C_i \approx C_i'$ (see Lemma~\ref{lemma:loadDelay:arch-seq:deep-indistinguishability-implies-indistinguishability}).
	Moreover, if $C_k$ is a stuck configuration (i.e., there is no $C'$ such that $C_k \SeqProcMuarchStep{}{} C'$) then so is $C_k'$ from $C_{k} \sim C_{k}'$.
	This concludes the proof of our lemma.
	
	We now prove our claim.
	That is, we show, by induction on $i$, that for all $0 \leq i \leq k$, $\map{\crun}{\hrun}{i} = \map{\crun'}{\hrun'}{i}$ and $ C_i \approx C_{i}'$.
	\begin{description}
		\item[Base case:]
		Then, $i = 0$.
		Therefore, $\map{\crun}{\hrun}{0} = \map{\crun'}{\hrun'}{0} = \{ 0 \mapsto 0\}$ by construction.
		Moreover, $C_{0} \sim C_{0}'$ immediately follows from $C_0 =  \tup{m_0,a_0, \emptysequence, \CacheState_0, \BpState_0, \SchedState_0}$, $
		C_0'= \tup{m_0',a_0', \emptysequence, \CacheState_0, \BpState_0, \SchedState_0}$, and $C_0,C_0'$ being initial \archstate{}s, i.e., $a_0 = a_0' = \lambda x \in \Var.\ 0$.
	
		\item[Induction step:]
		For the induction step, we assume that our claim holds for all $i' < i$, and we show that it holds for $i$ as well.
		From the induction hypothesis, we get that $\map{\crun}{\hrun}{i-1} = \map{\crunp}{\hrunp}{i-1}$ and $C_{i-1} \sim C_{i-1}'$.
		In the following, we denote $\map{\crun}{\hrun}{i-1} = \map{\crunp}{\hrunp}{i-1}$ as (IH.1) and $C_{i-1} \sim C_{i-1}'$ as (IH.2). 
		Moreover, from Lemma~\ref{lemma:loadDelay:arch-seq:mapping-is-correct} applied to $\map{\crun}{\hrun}{i-1} = \map{\crunp}{\hrunp}{i-1}$, we have that for all $\buf \in \prefixes{\buf_{i-1},\tup{m_{i-1},a_{i-1}}}$ and all $\buf' \in \prefixes{\buf_{i-1}',\tup{m_{i-1}',a_{i-1}'}}$, $C_{i-1} \bufEquiv{i} \crun( \map{\crun}{\hrun}{i-1}(|\buf|) ) $ and $C_{i-1}' \bufEquiv{i} \crunp( \map{\crunp}{\hrunp}{i-1}(|\buf'|) ) $.
		Let $\buf \in \prefixes{\buf_{i-1},\tup{m_{i-1},a_{i-1}}}$ and all $\buf' \in \prefixes{\buf_{i-1}',\tup{m_{i-1}',a_{i-1}'}}$ be arbitrary buffers such that $|\buf| = |\buf'|$.
		From (IH.1) and $|\buf| = |\buf'|$, we get that $\map{\crun}{\hrun}{i-1}(|\buf|)= \map{\crunp}{\hrunp}{i-1}(|\buf'|)$.
		Therefore,  $\crun( \map{\crun}{\hrun}{i-1}(|\buf|) )$ and $\crunp(\map{\crunp}{\hrunp}{i-1}(|\buf'|))$ are two configurations $\sigma_j, \sigma_j'$ for some $j = \map{\crun}{\hrun}{i-1}(|\buf|)$.
		From $\crun, \crunp$ having pairwise the same observations, we have $\sigma_{j} \ArchSeqInterfStep{o_{j+1}}{} \sigma_{j+1}$, $\sigma_{j}' \ArchSeqInterfStep{o_{j+1}'}{} \sigma_{j+1}'$, and $o_{j+1} = o_{j+1}'$ (note that we can always make a step in $\sigma_{j},\sigma_{j}'$ since if $j = k$ we can make silent steps thanks to the \textsc{Terminate} rule).
		Since $\buf,\buf'$ have been selected arbitrarily, we know that for all $\buf \in \prefixes{\buf_{i-1},\tup{m_{i-1},a_{i-1}}}$ and all $\buf' \in \prefixes{\buf_{i-1}',\tup{m_{i-1}',a_{i-1}'}}$ such that $|\buf| = |\buf'|$, we have $\sigma_{j} \ArchSeqInterfStep{o_{j+1}}{} \sigma_{j+1}$, $\sigma_{j}' \ArchSeqInterfStep{o_{j+1}'}{} \sigma_{j+1}'$, and $o_{j+1} = o_{j+1}'$ where $j = \map{\crun}{\hrun}{i-1}(|\buf|)$.
        From this and (IH.2), we can apply Lemma~\ref{lemma:loadDelay:arch-seq:trace-equiv-implies-stepwise-indistinguishability} to $C_{i-1}$ and $C_{i-1}'$.
		As a result, we obtain that either $C_{i} \sim C_{i}$ or that both configurations are stuck.
		Moreover, $\map{\crun}{\hrun}{i} = \map{\crunp}{\hrunp}{i}$ follows from the fact that the maps at step $i$ are derived from the maps at step $i-1$, which are equivalent thanks to (IH.1), depending on the content of the projection of the buffer (which is the same thanks to $C_{i} \sim C_{i}$).
		This completes the proof of the induction step.
	\end{description}
\end{proof}
\newpage
\section{Hardware taint tracking processor $\muarchStyle{tt}$  -- Proofs of Theorems~\ref{theorem:hni:stt:one} and~\ref{theorem:hni:stt:two}}\label{appendix:proofs:taint-tracking}

In this section, we prove the security guarantees of the hardware taint tracking processor $\muarchStyle{tt}$ given in \S\ref{sec:countermeasures:taint-tracking}.

In the following, we assume given an arbitrary cache $C$, branch predictor $Bp$, and scheduler $S$, whose initial states are, respectively, $\CacheState_0$, $\BpState_0$, and $\SchedState_0$.
Hence, our proof holds for arbitrary caches $C$, branch predictors $Bp$, and scheduler $S$.

For simplicity, here we report the rule defining the $\muarchStyle{tt}$ processor:
\begin{mathpar}
    \inferrule[Step]
    {
        d = \SchedNext(\SchedState)\\
        {\buf_{ul} = \unlabelNda{\buf}{d}}\\
        \tup{m,a, {\buf_{ul}}, \CacheState,\BpState} \muarchStep{d}{} \tup{m',a',{\buf_{ul}'}, \CacheState',\BpState'}	\\
        \SchedState' = \SchedUpdate(\SchedState,{\BufProject{{\buf'}}})\\
        {\buf' = \labelNda{\buf_{ul}'}{\buf}{d}}
    }
    {
        \tup{m,a,\buf, \CacheState,\BpState, \SchedState} \TtMuarchStep{}{} \tup{m',a'\buf', \CacheState',\BpState', \SchedState'}
    }
\end{mathpar}
We also report here the  unlabelling function $\unlabelNda{\cdot}{\cdot}$ (see Figure~\ref{figure:stt:unlabeling-apx}) and the labelling function $\labelNda{\cdot}{\cdot}{\cdot}$ (see Figure~\ref{figure:stt:labeling-apx}).

\begin{figure*}
    \begin{align*}
        \unlabelNda{\buf}{\fetch{}} &= \mask{\buf}\\
        \unlabelNda{\buf}{\retire{}} &= \drop{\buf}\\
        \unlabelNda{\buf}{\execute{i}} &=  
        {
            \begin{cases}
                \mask{\buf} & \text{if } \transmitGadget{\elt{\buf}{i}}\\
                \drop{\buf} & \text{otherwise}
            \end{cases}
        }
        \\
        \drop{\emptysequence} & := \emptysequence\\
        \drop{ \labelled{\tagged{i}{T}}{L} \concat \buf } & := \tagged{i}{T} \concat \drop{\buf}\\
        \mask{\emptysequence} &:= \emptysequence \\ 	
        \mask{ \labelled{\tagged{i}{T}}{L} \concat \buf } & := 
        {
            \begin{cases}
                \tagged{\passign{x}{\bot}}{T} \concat \mask{\buf} & \text{if } L = \emptyset  \wedge  i = \passign{x}{e}  \\ 
                \tagged{i}{T} \concat \mask{\buf} & \text{otherwise}
            \end{cases}
        }
    \end{align*}
    \caption{Unlabeling function $\unlabelNda{\buf}{d }$ for \muarchStyle{tt}}\label{figure:stt:unlabeling-apx}
    \end{figure*}

    \begin{figure*}
        \begin{align*}
            \labels{\buf} &= \deriveLabels{\buf}{\lambda x \in \Var.\ \emptyset}\\
            \labels{\buf}(x) &= \deriveLabels{\buf}{\lambda x \in \Var.\ \emptyset}(x)\\
            \labels{\buf}(e) &= \bigcup_{x \in \mathit{vars}(e)} \deriveLabels{\buf}{\lambda x \in \Var.\ \emptyset}(x)\\
            \\
            \deriveLabels{\emptysequence}{\Lambda} &= \Lambda \\ 
            \deriveLabels{ \labelled{\tagged{i}{T}}{L} \concat \buf }{\Lambda} &=  
            {
                \begin{cases}
                    \deriveLabels{  \buf }{\Lambda[x \mapsto L ]} & \text{if } i = \pload{x}{e} \vee i = \passign{x}{e}\\
                    \deriveLabels{  \buf }{\Lambda} & \text{otherwise}
                \end{cases}
            }\\
            \\
        \decrement{\emptysequence} &:= \emptysequence \\
        \decrement{\labelled{\tagged{i}{T}}{L} \concat \buf } &:=  \labelled{\tagged{i}{T}}{ \{ j - 1 \mid j \in L \wedge j-1 > 0\} } \concat \decrement{\buf} \\
        \\
        \mathit{strip}(\emptysequence, j) &:= \emptysequence\\
        \mathit{strip}(\labelled{\tagged{i}{T}}{L} \concat \buf , j) &:= \labelled{\tagged{i}{T}}{L \setminus \{ j \}} \concat  \mathit{strip}( \buf , j)\\
        \\
        \labelNda{\buf_{ul}}{\buf}{\fetch{}} &= \buf \\
                \text{where }& |\buf_{ul}| = |\buf| \\
        \labelNda{\buf_{ul} \concat \tagged{\passign{\pc}{e}}{\notags}}{\buf}{\fetch{}} &= \buf \concat \labelled{\tagged{\passign{\pc}{e}}{\notags}}{\emptyset}\\
        \text{where }& |\buf_{ul}| = |\buf| + 1 \\
        \labelNda{\buf_{ul} \concat \tagged{\passign{\pc}{\ell}}{\ell_0}}{\buf}{\fetch{}} &= \buf \concat \labelled{\tagged{\passign{\pc}{\ell}}{\ell_0}}{\emptyset}\\
        \text{where }& |\buf_{ul}| = |\buf| + 1 \\
        \labelNda{\buf_{ul} \concat \tagged{i}{T} \concat \tagged{\pmarkedassign{\pc}{\ell}}{\notags} }{\buf}{\fetch{}} &= \buf \concat \labelled{\tagged{i}{T}}{\emptyset} \concat \labelled{\tagged{\pmarkedassign{\pc}{\ell}}{\notags}}{\emptyset}\\
            \text{where }& 
            i \neq \pload{x}{e} \wedge i \neq \passign{x}{e} \wedge |\buf_{ul}| = |\buf|+2\\
        \labelNda{\buf_{ul} \concat \tagged{ \passign{x}{e} }{T} \concat \tagged{\pmarkedassign{\pc}{\ell}}{\notags} }{\buf}{\fetch{}} &= \buf \concat \labelled{\tagged{ \passign{x}{e} }{T}}{\labels{\buf}(e)} \concat \labelled{\tagged{\pmarkedassign{\pc}{\ell}}{\notags}}{\emptyset}\\
        \text{where }& |\buf_{ul}| = |\buf| + 2 \\
        \labelNda{\buf_{ul} \concat \tagged{\pload{x}{e}}{T} \concat \tagged{\passign{\pc}{\ell}}{\notags} }{\buf}{\fetch{}} &= \buf \concat \labelled{\tagged{\pload{x}{e}}{T}}{ \{ j \in \Nat \mid \exists y, \ell,\ell_0.\  (\elt{\buf}{j} = \tagged{\passign{y}{\ell}}{\ell_0} \wedge \ell_0 \neq \notags) \} } \concat \labelled{\tagged{\passign{\pc}{\ell}}{\notags}}{\emptyset}\\
        \text{where }& |\buf_{ul}| = |\buf| + 2 
            \\
            \\
        \labelNda{\buf_{ul}}{\buf}{\execute{i}} &= \buf[0..i-1] \concat \labelled{\elt{\buf_{ul}}{i}}{ L } \concat \buf[i+1 .. |\buf|] \\
        \text{where }& |\buf_{ul}| = |\buf| \wedge \elt{\buf}{i} = \labelled{c}{ L } \wedge 
        \forall \ell, \ell_0.\ 	(c = \tagged{\passign{\pc}{\ell}}{\ell_0} \rightarrow \ell_0 = \emptysequence)
        \\
        \labelNda{\buf_{ul}}{\buf}{\execute{i}} &= \buf[0..i-1] \concat \labelled{\elt{\buf_{ul}}{i}}{ \emptyset } \concat 
        \mathit{strip}(\buf[i+1 .. |\buf|], i) \\
        \text{where }& |\buf_{ul}| = |\buf| \wedge \elt{\buf}{i} = \labelled{\tagged{\passign{\pc}{\ell}}{\ell_0}}{ L } \wedge \ell_0 \neq  \notags \wedge \elt{\buf_{ul}}{i} = \tagged{\passign{\pc}{\ell'}}{\notags}
        \\
        \labelNda{\buf_{ul}}{\buf}{\execute{i}} &= \buf[0..i-1] \concat \labelled{\elt{\buf_{ul}}{i}}{ L } \concat \buf[i+1 .. |\buf|] \\
        \text{where }& |\buf_{ul}| = |\buf| \wedge \elt{\buf}{i} = \labelled{\tagged{\passign{\pc}{\ell}}{\ell_0}}{ L } \wedge \ell_0 \neq \notags  \wedge \elt{\buf_{ul}}{i} = \tagged{\passign{\pc}{\ell}}{\ell_0} 
        \\
        \labelNda{\buf_{ul}}{\buf}{\execute{i}} &= \buf[0..i-1] \concat \labelled{\elt{\buf_{ul}}{i}}{ \emptyset} \\
        \text{where }& |\buf_{ul}| = i \wedge \elt{\buf}{i} = \labelled{\tagged{\passign{\pc}{\ell}}{\ell_0}}{L} \wedge \ell_0 \neq \notags \wedge \elt{\buf_{ul}}{i} = \tagged{\passign{\pc}{\ell'}}{\notags}
        \\
        \\
        \labelNda{ \buf_{ul} }{\labelled{\tagged{i}{T}}{L} \concat \buf}{\retire{}} &= \decrement{\buf}
        \end{align*}
        \caption{Labeling function $\labelNda{\buf_{\mathit{ul}}}{\buf}{d}$ for \muarchStyle{tt}}\label{figure:stt:labeling-apx}
    \end{figure*}

The semantics $\TtMuarchSem{p}$ for a program $p$ is defined as follows:
$\TtMuarchSem{p}(\tup{m,a})$ is $\tup{\emptysequence, \CacheState_0,\BpState_0, \SchedState_0}$ $ \cdot \tup{\BufProject{\buf_1}, \CacheState_1,\BpState_1, \SchedState_1}$ $  \cdot \tup{\BufProject{\buf_2}, \CacheState_2,\BpState_2, \SchedState_2}$ $  \cdot \ldots \cdot $ $ \tup{\BufProject{\buf_n}, \CacheState_n,\BpState_n, \SchedState_n}$ where $	\tup{m,a,\emptysequence, \CacheState_0,\BpState_0, \SchedState_0 }$ $ \TtMuarchStep{}{}$ $ \tup{m_1,a_1,\buf_1, \CacheState_1,\BpState_1, \SchedState_1} $ $ \TtMuarchStep{}{} $ $\tup{m_2,a_2,\buf_2, \CacheState_2,\BpState_2, \SchedState_2} $ $\TtMuarchStep{}{}$ $\ldots$ $\TtMuarchStep{}{} $ $ \tup{m_n,a_n,\buf_n, \CacheState_n,\BpState_n, \SchedState_n}$ is the complete hardware run obtained starting from $\tup{m,a}$ and terminating in $\tup{m_n,a_n,\buf_n, \CacheState_n,\BpState_n, \SchedState_n}$, which is a final hardware state. Otherwise, $\TtMuarchSem{p}(\tup{m,a})$ is undefined.

\subsection{Proof of Theorem~\ref{theorem:hni:stt:one}}

\textbf{ASSUMPTION: } In the following, we assume that $\wInterf > \wMuarch + 1$.

\sttOne*

\begin{proof}
The proof of the theorem is identical to the one of Theorem~\ref{theorem:hni:all} (see Appendix~\ref{appendix:proofs:general}):
\begin{compactitem}
    \item For the mapping lemma, we use the mapping function from Definition~\ref{def:vanilla:mapping} and the proof proceeds as in Lemma~\ref{lemma:vanilla:mapping-is-correct}.
    \item For the indistinguishability lemma, we use the indistinguishability relation from Definition~\ref{def:vanilla:deep-indistinguishability} and the proof proceeds as in Lemma~\ref{lemma:vanilla:trace-equiv-implies-stepwise-indistinguishability}. 
    \item The main lemma proceeds as in Lemma~\ref{lemma:vanilla:main-lemma}.
\end{compactitem}
\end{proof}

\subsection{Proof of Theorem~\ref{theorem:hni:stt:two}}

\sttTwo*

\begin{proof}
    Let $p$ be an arbitrary well-formed program.
	Moreover, let $\sigma = \tup{m,a},\sigma' = \tup{m',a'}$ be two arbitrary initial configurations.
	There are two cases:
	\begin{compactitem}
	\item[$\ArchSeqInterf{\Prg}(\sigma) \neq \ArchSeqInterf{\Prg}(\sigma')$:] Then, 	$\ArchSeqInterf{\Prg}(\sigma) = \ArchSeqInterf{\Prg}(\sigma') \Rightarrow \TtMuarchSem{\Prg}(\sigma) = \TtMuarchSem{\Prg}(\sigma')$ trivially holds.
	\item[$\ArchSeqInterf{\Prg}(\sigma) = \ArchSeqInterf{\Prg}(\sigma')$:]
		By unrolling the notion of $\ArchSeqInterf{\Prg}(\sigma)$ (together with all changes to the program counter $\pc$ being visible on traces), we obtain that there are runs $\crun:= \sigma$ $\ArchSeqInterfStep{o_1}{}$ $\sigma_1$ $\ArchSeqInterfStep{o_2}{}$ $\ldots$  $\ArchSeqInterfStep{o_{n-1}}{}$  $\sigma_n$ and $\crunp:= \sigma'\ArchSeqInterfStep{o_1'}{} \sigma_1' \ArchSeqInterfStep{o_2'}{} \ldots  \ArchSeqInterfStep{o_{n-1}'}{}  \sigma_n'$ such that $o_i = o_i'$ for all $0 < i < n$.
		By applying Lemma~\ref{lemma:tt:arch-seq:main-lemma}, we immediately get that $\TtMuarchSem{\Prg}(\sigma) = \TtMuarchSem{\Prg}(\sigma')$ (because either both run terminate producing indistinguishable sequences of processor configurations or they both get stuck).
		Therefore, $\ArchSeqInterf{\Prg}(\sigma) = \ArchSeqInterf{\Prg}(\sigma') \Rightarrow \TtMuarchSem{\Prg}(\sigma) = \TtMuarchSem{\Prg}(\sigma')$ holds.
	\end{compactitem}
	Hence, $\ArchSeqInterf{\Prg}(\sigma) = \ArchSeqInterf{\Prg}(\sigma') \Rightarrow \TtMuarchSem{\Prg}(\sigma) = \TtMuarchSem{\Prg}(\sigma')$ holds for all programs $p$ and initial configurations $\sigma,\sigma'$.
	Therefore, $\hsni{\ArchSeqInterf{\cdot}}{\TtMuarchSem{\cdot}}$ holds.
\end{proof}

\subsubsection{Preliminary definitions}

\begin{definition}[Deep-update for \muarchStyle{tt} and $\ArchSeqInterf{\cdot}{}$]
    Let $p$ be a program,  $\tup{m,a}$ be an \archstate{}, and $\buf$ be a labelled buffer.
    The \emph{deep-update of $\tup{m,a}$ given $\buf$} is defined as follows:
        \begin{align*}
        \update{\tup{m,a}}{\emptysequence} &:= \tup{m,a} \\
        \update{\tup{m,a}}{ \labelled{\tagged{\passign{x}{e}}{T}}{L} }  &:= 
            \begin{cases}
                \tup{m, a[x \mapsto \exprEval{e}{a}]} & \text{if } x \neq \pc \\
                \tup{m, a[x \mapsto \exprEval{e}{a}]} & \text{if } x = \pc  \wedge T = \notags \\
                \tup{m, a[x \mapsto \ell']} & \text{if } x = \pc  \wedge T = \ell \wedge p(\ell) = \pjz{y}{\ell'} \wedge a(y) = 0 \\ 
                \tup{m, a[x \mapsto \ell+1]} & \text{if } x = \pc  \wedge T = \ell \wedge p(\ell) = \pjz{y}{\ell'} \wedge a(y) \neq 0 
            \end{cases}
        \\
        \update{\tup{m,a}}{ \labelled{\tagged{\pmarkedassign{x}{e}}{T}}{L} }  &:= \tup{m, a[x \mapsto \exprEval{e}{a}]}\\
        \update{\tup{m,a}}{ \labelled{\tagged{\pload{x}{e}}{T}}{L} } &:= \tup{m,a[x \mapsto m(\exprEval{e}{a})] } \\
        \update{\tup{m,a}}{ \labelled{\tagged{\pstore{x}{e}}{T}}{L} } &:= \tup{m[\exprEval{e}{a} \mapsto a(x)],a}\\
        \update{\tup{m,a}}{ \labelled{\tagged{\pskip{}}{T}}{L} } &:= \tup{m,a}\\
        \update{\tup{m,a}}{ \labelled{\tagged{\pbarrier{}}{T}}{L}} &:= \tup{m,a}\\
        \update{\tup{m,a}}{(\tagged{i}{T} \concat \buf)} &:= 				\update{ (\update{\tup{m,a}}{\tagged{i}{T}}) }{ \buf }
        \end{align*}
\end{definition}

\newcommand{\wellformedLabels}[1]{\mathit{wfL}(#1)}
\newcommand{\wellformedp}[1]{\mathit{wf'}(#1)}
\begin{definition}[Well-formed buffers for $\muarchStyle{tt}$ and $\ArchSeqInterf{\cdot}{}$]
    A labelled reorder buffer $\buf$ is \emph{well-formed for $\muarchStyle{tt}$, $\ArchSeqInterf{\cdot}{}$, a well-formed program $p$, and an \archstate{}  $\tup{m,a}$}, written $\wellformed{\buf,\tup{m,a}}$, if the following conditions hold:
    \begin{align*}
        \wellformed{\buf, \tup{m,a}} & \text{ if } \wellformedp{\buf, \tup{m,a}} \wedge \wellformedLabels{\buf} \\
    	\wellformedp{\emptysequence,\tup{m,a}} & \\
        \wellformedp{ \labelled{\tagged{\passign{\pc}{e}}{T}}{L} \concat \buf, \tup{m,a} } & \text{ if } \wellformedp{\buf, \update{\tup{m,a}}{ \labelled{\tagged{\passign{\pc}{e}}{T}}{L}} } \wedge p(a(\pc)) \sim_{\tup{m,a}} \tagged{\passign{\pc}{e}}{\notags} \\
        \wellformedp{ \labelled{\tagged{\passign{\pc}{\ell}}{\ell_0}}{L} \concat \buf, \tup{m,a}} & \text{ if } \wellformedp{\buf, \update{\tup{m,a}}{ \labelled{\tagged{\passign{\pc}{\ell}}{\ell_0}}{L} }} \wedge  \ell_0 \in \Val \wedge p(\ell_0) = \pjz{x}{\ell'} \wedge \\& \quad \ell \in \{\ell', \ell_0+1\} \wedge p(a(\pc)) \sim_{\tup{m,a}} \tagged{\passign{\pc}{\ell}}{\ell_0} \wedge \ell_0 = a(\pc)\\
    	\wellformedp{ \labelled{\tagged{\pmarkedassign{\pc}{\ell}}{\notags}}{L} \concat \buf , \tup{m,a} } & \text{ if } \wellformedp{\buf,\update{ \labelled{\tagged{\pmarkedassign{\pc}{\ell}}{\notags}}{L} }{\tup{m,a}}} \wedge \ell \in \Val \\
        \wellformedp{ \labelled{\tagged{i}{\notags}}{L} \concat \labelled{\tagged{\pmarkedassign{\pc}{\ell}}{\notags}}{L'}  \concat \buf , \tup{m,a} } & \text{ if }  
        \wellformedp{\buf, \update{\tup{m,a}}{(  \labelled{\tagged{i}{\notags}}{L} \concat \labelled{\tagged{\pmarkedassign{\pc}{\ell}}{\notags}}{L'} )} } \wedge 
        \ell \in \Val \wedge (\forall e.\ i \neq \passign{\pc}{e}) \wedge  \\
        & \quad (\forall x,e.\ i \neq \pmarkedassign{x}{e}) \wedge (\forall x,e.\ i \neq \pload{\pc}{e}) \wedge p(a(\pc)) \sim_{\tup{m,a}} \tagged{i}{\notags} \wedge \\ & \quad \ell = a(\pc)+1\\
        \wellformedLabels{\buf} & \text{ if } \forall 1 \leq i \leq |\buf|.\ \wellformedLabels{\elt{\buf}{i},\buf[0..i-1]} \\
        \wellformedLabels{ \labelled{\tagged{\pload{x}{e}}{T}}{L} ,\buf} &\text{ if } L = \{ j \mid \elt{\buf}{j} = \labelled{\tagged{\passign{\pc}{\ell}}{\ell_0} }{L} \wedge \ell_0 \neq \notags \}\\
        \wellformedLabels{ \labelled{\tagged{\passign{x}{e}}{T}}{L} ,\buf} &\text{ if } L = \labels{\buf}(e) \wedge x \neq \pc\\
        \wellformedLabels{ \labelled{\tagged{\pstore{x}{e}}{T}}{L} ,\buf} &\text{ if } L = \emptyset\\
        \wellformedLabels{ \labelled{\tagged{\passign{\pc}{e}}{T}}{L} ,\buf} &\text{ if } L = \emptyset\\
        \wellformedLabels{ \labelled{\tagged{\pmarkedassign{\pc}{e}}{T}}{L} ,\buf} &\text{ if } L = \emptyset\\
        \wellformedLabels{ \labelled{\tagged{\pskip}{T}}{L} ,\buf} &\text{ if } L = \emptyset\\
        \wellformedLabels{ \labelled{\tagged{\pbarrier}{T}}{L} ,\buf} &\text{ if } L = \emptyset
    \end{align*}
    where the instruction-compatibility relation $\sim_{\tup{m,a}}$ is defined as follows:
    \begin{align*}
        \pskip &\sim_{\tup{m,a}} \tagged{\pskip}{\notags} \\ \allowdisplaybreaks
        \pbarrier &\sim_{\tup{m,a}} \tagged{\pbarrier}{\notags} \\ \allowdisplaybreaks
        \passign{x}{e} &\sim_{\tup{m,a}} \tagged{\passign{x}{e}}{\notags} \\ \allowdisplaybreaks
        \passign{x}{e} &\sim_{\tup{m,a}} \tagged{\passign{x}{v}}{\notags} \text{ if } v = \exprEval{e}{a} \\ \allowdisplaybreaks
        \pload{x}{e} &\sim_{\tup{m,a}} \tagged{\pload{x}{e}}{\notags} \\ \allowdisplaybreaks
        \pload{x}{e} &\sim_{\tup{m,a}} \tagged{\passign{x}{v}}{\notags} \text{ if } v = m(\exprEval{e}{a})\\ \allowdisplaybreaks
        \pstore{x}{e} &\sim_{\tup{m,a}} \tagged{\pstore{x}{e}}{\notags} \\ \allowdisplaybreaks
        \pstore{x}{e} &\sim_{\tup{m,a}} \tagged{\pstore{v}{n}}{\notags} \text{ if } v = a(x) \wedge n = \exprEval{e}{a} \\ \allowdisplaybreaks
        \pjmp{e} &\sim_{\tup{m,a}} \tagged{\passign{\pc}{e}}{\notags} \\ \allowdisplaybreaks
        \pjmp{e} &\sim_{\tup{m,a}} \tagged{\passign{\pc}{v}}{\notags} \text{ if } v = \exprEval{e}{a} \\ \allowdisplaybreaks
        \pjz{x}{\ell} &\sim_{\tup{m,a}} \tagged{\passign{\pc}{\ell}}{a(\pc)} \\ \allowdisplaybreaks
        \pjz{x}{\ell} &\sim_{\tup{m,a}} \tagged{\passign{\pc}{a(\pc)+1}}{a(\pc)} \\ \allowdisplaybreaks
        \pjz{x}{\ell} &\sim_{\tup{m,a}} \tagged{\passign{\pc}{a(\pc)+1}}{\notags} \wedge a(x) \neq 0\\ \allowdisplaybreaks
        \pjz{x}{\ell} &\sim_{\tup{m,a}} \tagged{\passign{\pc}{\ell}}{\notags} \wedge a(x) = 0
    \end{align*}
    \end{definition}

    \begin{definition}[Prefixes of buffers]\label{def:tt:arch-seq:prefixes}
        The prefixes of a well-formed labelled buffer $\buf$, given an \archstate{} $\tup{m,a}$, are defined as follows:
        \begin{align*}
            \prefixes{\emptysequence, \tup{m,a}} &= \{ \emptysequence \}\\
            \prefixes{  \labelled{\tagged{\passign{\pc}{e}}{\notags}}{L} \concat \buf , \tup{m,a}} &= 
            \{ \emptysequence \} \cup \{  \labelled{\tagged{\passign{\pc}{e}}{T}}{L} \concat \buf' \mid \buf' \in \prefixes{\buf, \update{\tup{m,a}}{ \labelled{\tagged{\passign{\pc}{e}}{\notags}}{L} } } \} \\
            \prefixes{  \labelled{\tagged{\passign{\pc}{\ell}}{\ell_0}}{L} \concat \buf , \tup{m,a}} &= 
            \{ \emptysequence, \labelled{\tagged{\passign{\pc}{\ell}}{\ell_0}}{L} \} \cup \{  \labelled{\tagged{\passign{\pc}{\ell}}{\ell_0}}{L} \concat \buf' \mid \buf' \in \prefixes{\buf, \update{\tup{m,a}}{ \labelled{\tagged{\passign{\pc}{\ell}}{\ell_0}}{L} } } \wedge \\
            & \qquad p(\ell_0) = \pjz{x}{\ell'} \wedge (a(x) = 0 \rightarrow \ell = \ell') \wedge (a(x) \neq 0 \rightarrow \ell = \ell_0+1 ) \} \\
            \prefixes{  \labelled{\tagged{i}{\notags}}{L} \concat \labelled{\tagged{\pmarkedassign{\pc}{\ell}}{\notags}}{L'} \concat \buf } &= \{ \emptysequence \} \cup \{   \labelled{\tagged{i}{\notags}}{L} \concat \labelled{\tagged{\pmarkedassign{\pc}{\ell}}{\notags}}{L'} \concat \buf' \mid \buf' \in \prefixes{\buf, \update{\tup{m,a}}{ (\labelled{\tagged{i}{\notags}}{L} \concat \labelled{\tagged{\pmarkedassign{\pc}{\ell}}{\notags}}{L'} ) } } \} \\
            \prefixes{   \labelled{\tagged{\pmarkedassign{\pc}{\ell}}{\notags}}{L}  \concat \buf } &= \{  \labelled{\tagged{\pmarkedassign{\pc}{\ell}}{\notags}}{L} \} \cup \{  \labelled{\tagged{\pmarkedassign{\pc}{\ell}}{\notags}}{L} \concat \buf' \mid \buf' \in \prefixes{\buf, \update{\tup{m,a}}{ \labelled{\tagged{\pmarkedassign{\pc}{\ell}}{\notags}}{L} }} \} 
        \end{align*}
    \end{definition}

\subsubsection{Mapping lemma}

Lemma~\ref{lemma:loadDelay:arch-seq:buffers-well-formedness} states that all reorder buffers occurring in hardware runs are well-formed.

\begin{lemma}[Reorder buffers are well-formed]\label{lemma:tt:arch-seq:buffers-well-formedness}
Let $p$ be a well-formed program, $\sigma_0 = \tup{m,a}$ be an initial \archstate, $\CacheState_0$ be the initial cache state, $\BpState_0$ be the initial branch predictor state, and $\SchedState_0$ be the initial scheduler state.
For all hardware runs  $\hrun := C_0 \TtMuarchStep{}{} C_1 \TtMuarchStep{}{} C_2 \TtMuarchStep{}{} \ldots \TtMuarchStep{}{} C_k$ and all $0 \leq i \leq k$, then $\wellformed{\buf_i,\tup{m_i,a_i}}$, where $C_0 = \tup{m,a,\emptysequence, \CacheState_0, \BpState_0, \SchedState_0}$ and $C_i = \tup{m_i, a_i,\buf_i, \CacheState_i, \BpState_i, \SchedState_i}$.
\end{lemma}

\begin{proof}
The lemma follows by (1) induction on $i$, and (2) inspection of the rules defining $\TtMuarchStep{}{}$.
\end{proof}

\begin{definition}[$\crun-\hrun$ mapping for $\muarchStyle{tt}$ and $\ArchSeqInterf{\cdot}$]\label{def:tt:arch-seq:mapping}
	Let $p$ be a well-formed program, $\sigma_0 = \tup{m,a}$ be an initial \archstate, $\CacheState_0$ be the initial cache state, $\BpState_0$ be the initial branch predictor state, and $\SchedState_0$ be the initial scheduler state.
	Furthermore, let:
	\begin{compactitem}
		\item $\crun := \sigma_0 \ArchSeqInterfStep{o_1}{} \sigma_1 \ArchSeqInterfStep{o_2}{} \sigma_2 \ArchSeqInterfStep{o_3}{} \ldots \ArchSeqInterfStep{o_{n-1}}{} \sigma_n$ be the longest $\interfStyle{seq-arch}$ contract run obtained starting from $\sigma$.
		\item $\hrun := C_0 \TtMuarchStep{}{} C_1 \TtMuarchStep{}{} C_2 \TtMuarchStep{}{} \ldots \TtMuarchStep{}{} C_k$ be the longest $\muarchStyle{tt}$ hardware run obtained starting from $C_0 = \tup{m,a,\emptysequence, \CacheState_0, \BpState_0, \SchedState_0}$.
		\item $\hrun(i)$ be the $i$-th hardware configuration in $\hrun$.
		\item $\crun(i)$ be the $i$-th contract configuration in $\crun$  (note that $\crun(i) = \sigma_n$ for all $i > n$).
	\end{compactitem}
	The \emph{$\crun-\hrun$ mapping}, which maps hardware configurations in $\hrun$ to contract configurations in $\crun$, is defined as follows:
	\begin{align*}
		\map{\crun}{\hrun}{0} &:= \{0 \mapsto 0\} \\ 
		\map{\crun}{\hrun}{i} &:= {
            \begin{cases}
            \map{\crun}{\hrun}{i-1} & \text{if } \SchedNext(\hrun(i-1)) = \fetch{} \wedge  \Mispred{\hrun(i-1)} \\
			\map{\crun}{\hrun}{i-1} & \text{if } \SchedNext(\hrun(i-1)) = \fetch{} \wedge ln(\hrun(i-1)) = ln(\hrun(i)) \wedge \neg \Mispred{\hrun(i-1)}\\
			fetch_{\crun,\hrun}(i) & \text{if } \SchedNext(\hrun(i-1)) = \fetch{} \wedge ln(\hrun(i-1)) < ln(\hrun(i)) \wedge \neg \Mispred{\hrun(i-1)}\\
			\map{\crun}{\hrun}{i-1} & \text{if } \SchedNext(\hrun(i-1)) = \execute{j} \wedge \neg \Mispred{\hrun(i-1)}\\ 
			shift(\map{\crun}{\hrun}{i -1}) & \text{if } \SchedNext(\hrun(i-1)) = \retire{} \wedge \neg \Mispred{\hrun(i-1)} 
			\end{cases}
		}\\ 
		fetch_{\crun,\hrun}(i) &=
				\map{\crun}{\hrun}{i-1}[ln(\hrun(i-1))+2 \mapsto \map{\crun}{\hrun}{i-1}(ln(\hrun(i-1)))+1]\\
				& \qquad \text{if }
							p(\mathit{lstPc}(\hrun(i -1))) \neq \pjz{x}{\lbl} \wedge 
							p(\mathit{lstPc}(\hrun(i -1))) \neq \pjmp{e} 
							\\
		fetch_{\crun,\hrun}(i) &=
				\map{\crun}{\hrun}{i-1}[ln(\hrun({i-1}))+1 \mapsto \map{\crun}{\hrun}{i-1}(ln(\hrun(i-1)))+1]\\
				& \qquad \text{if }
							p(\mathit{lstPc}(\hrun(i -1))) = \pjmp{e} \vee p(\mathit{lstPc}(\hrun(i -1))) = \pjz{x}{\lbl} 
							\\
		ln(\tup{m,a,\buf,\CacheState,\BpState,\SchedState}) &= |\buf|\\
		\SchedNext(\tup{m,a,\buf,\CacheState,\BpState,\SchedState}) &= \SchedNext(\SchedState)\\
		shift(map) &= \lambda i \in \Nat.\ map(i +1 )\\
		\mathit{lstPc}(\tup{m,a,\buf,\CacheState,\BpState,\SchedState}) &= (\update{\tup{m,a}}{\buf})(\pc)\\
		\Mispred{\tup{m,a,\buf,\CacheState,\BpState,\SchedState}} &= 
	{	\begin{cases}
			\top & \text{if } \forall 1 \leq i \leq |\buf|.\ (\elt{\buf}{i} =  \tagged{\passign{\pc}{\ell}}{\ell'}) \rightarrow (\ell = \mathit{correctPred}(\ell', \update{\tup{m,a}}{\buf[0..i-1]}) )\\
			\bot & \text{otherwise}
		\end{cases}}\\
		\mathit{correctPred}(\ell,a) &= 
	{	\begin{cases}
		\ell'	& \text{if } p(\ell) = \pjz{x}{\ell'} \wedge a(x) = 0\\
		\ell + 1 & \text{otherwise}
		\end{cases}}
	\end{align*}
    \end{definition}
    
    We are now ready to prove Lemma~\ref{lemma:tt:arch-seq:mapping-is-correct}, the main lemma showing the correctness of the $\crun-\hrun$ mapping.

    \begin{lemma}[Correctness of $\crun-\hrun$ mapping]\label{lemma:tt:arch-seq:mapping-is-correct}
        Let $p$ be a well-formed program, $\sigma_0 = \tup{m,a}$ be an initial \archstate, $\CacheState_0$ be the initial cache state, $\BpState_0$ be the initial branch predictor state, and $\SchedState_0$ be the initial scheduler state.
        Furthermore, let:
        \begin{compactitem}
            \item $\crun := \sigma_0 \ArchSeqInterfStep{o_1}{} \sigma_1 \ArchSeqInterfStep{o_2}{} \sigma_2 \ArchSeqInterfStep{o_3}{} \ldots \ArchSeqInterfStep{o_{n-1}}{} \sigma_n$ be the longest $\interfStyle{seq-arch}$ contract run obtained starting from $\sigma_0$.
            \item $\hrun := C_0 \TtMuarchStep{}{}  C_1 \TtMuarchStep{}{} C_2 \TtMuarchStep{}{} \ldots \TtMuarchStep{}{} C_k$ be the longest $\muarchStyle{tt}$ hardware run obtained starting from $C_0 = \tup{m,a,\emptysequence, \CacheState_0, \BpState_0, \SchedState_0}$.
            \item $\hrun(i)$ be the $i$-th hardware configuration in $\hrun$.
            \item $\crun(i)$ be the $i$-th contract configuration in $\crun$  (note that $\crun(i) = \sigma_n$ for all $i > n$).
            \item $ \map{\crun}{\hrun}{\cdot}$ be the mapping from Definition~\ref{def:tt:arch-seq:mapping}.
        \end{compactitem}
        The following conditions hold:
        \begin{compactenum}[(1)]
        \item $C_0$ is an initial hardware configuration.
        \item $C_k$ is a final hardware configuration or there is no $C_{k'}$ such that $C_k \TtMuarchStep{}{} C_{k'}$.
        \item for all $0 \leq i \leq k$, given $C_i = \tup{m_i,a_i, \buf_i, \CacheState_i, \BpState_i, \SchedState_i}$ the following conditions hold:
            \begin{compactenum}[(a)]
                \item for all $\buf \in \prefixes{\buf_i, \tup{m_i,a_i}}$,  $\update{\tup{m_i,a_i}}{\buf} = \crun( \map{\crun}{\hrun}{i}(|\buf|) )$.
            \end{compactenum}
        \end{compactenum}
    \end{lemma}

\begin{proof}
The proof of this lemma is similar to the proof of Lemma~\ref{lemma:loadDelay:arch-seq:mapping-is-correct} (since the mapping in Definition~\ref{def:tt:arch-seq:mapping} is the same as the one in Definition~\ref{def:loadDelay:arch-seq:mapping}).
\end{proof}

\subsubsection{Indistinguishability lemma}

\begin{definition}[Deep-indistinguishability of hardware configurations]\label{def:tt:arch-seq:deep-indistinguishability}
	We say that two hardware configurations $\tup{\sigma,\mu} = \tup{m,a,\buf, \CacheState,\BpState, \SchedState}$ and $\tup{\sigma',\mu'} = \tup{m',a',\buf', \CacheState',\BpState', \SchedState'}$ are \emph{deep-indistinguishable}, written $\tup{\sigma,\mu} \sim \tup{\sigma',\mu'}$, iff
    \begin{inparaenum}[(a)]
        \item $a = a'$,
		\item $\DeepProject{\buf} = \DeepProject{ \buf'}$,
		\item $\CacheState = \CacheState'$,
		\item $\BpState = \BpState'$, and
		\item $\SchedState = \SchedState'$.
    \end{inparaenum}
    where $\DeepProject{\buf}$ is as follows:
    \begin{align*}
        \DeepProject{\emptysequence}	&:= \emptysequence\\
        \DeepProject{\labelled{\tagged{\pskip{}}{T}}{L}}  &:=  \labelled{\tagged{\pskip{}}{T}}{L}\\
        \DeepProject{ \labelled{\tagged{\pbarrier{}}{T}}{L}}  &:=  \labelled{\tagged{\pbarrier{}}{T}}{L}\\
        \DeepProject{ \labelled{\tagged{\passign{x}{e}}{T}}{L}}  &:=
        {
            \begin{cases}
            \labelled{\tagged{\passign{x}{\unresolved}}{T}}{L} & \text{if}\  L \neq \emptyset \wedge x \neq \pc\\
            \labelled{\tagged{\passign{x}{e}}{T}}{L} & \text{if}\  L = \emptyset\\
            \labelled{\tagged{\passign{x}{e}}{T}}{L} & \text{if}\ x = \pc
            \end{cases}
        }\\
        \DeepProject{\labelled{\tagged{\pload{x}{e}}{T}}{L}}  &:=
        {
            \labelled{\tagged{\pload{x}{e}}{T}}{L}
        }\\
        \DeepProject{\labelled{\tagged{\pstore{x}{e}}{T}}{L}}  &:=
        {
            \labelled{\tagged{\pstore{x}{e}}{T}}{L}
        }\\
        \DeepProject{ (\labelled{\tagged{i}{T}}{L} \concat \buf) } &:=\DeepProject{ (\labelled{\tagged{i}{T}}{L} } \concat \DeepProject{\buf}
    \end{align*}
\end{definition}

\begin{lemma}[Masking preserves deep-indistinguishablity]
\label{lemma:tt:arch-seq:masking-preserves-indistinguishability}
Let  $\tup{\sigma,\mu} = \tup{m,a,\buf, \CacheState,\BpState, \SchedState}$ and $\tup{\sigma',\mu'} = \tup{m',a',\buf', \CacheState',\BpState', \SchedState'}$ be two well-formed hardware configurations.
If $\tup{\sigma,\mu} \sim \tup{\sigma',\mu'}$, then for all $0 \leq i \leq |\buf|$ and
     all $x \in \Var$, $\apply{\tup{m,a}}{ \mask{\buf[0..i]} }(x) = \apply{\tup{m',a'}}{ \mask{\buf'[0..i]} }(x)$.
\end{lemma}

\begin{proof}
This follows by inspection of $\apply{\cdot}{ \cdot }$, $\mask{\cdot}$, and Definition~\ref{def:tt:arch-seq:deep-indistinguishability}.
Intuitively, $\sim$ ensures that (1) the initial assignments $a, a'$ are identical, and (2) all entries in $\buf, \buf'$ whose label is $\emptyset$ are the same and  $\mask{\cdot}$ replaces all values whose label is not $\emptyset$ with $\bot$.
Finally, $\apply{\cdot}{ \cdot }$ simply propagates changes through the buffer.
\end{proof}

\begin{lemma}[Mask and drop preserve values]\label{lemma:tt:arch-seq:mask-drop-preserve-values}
Let  $\tup{\sigma,\mu} = \tup{m,a,\buf, \CacheState,\BpState, \SchedState}$ be a well-formed hardware configurations.
For all $0 \leq i \leq |\buf|$ and all $x \in \Var$, $\apply{\tup{m,a}}{ \mask{\buf[0..i]} }(x) \neq \bot \rightarrow \apply{\tup{m,a}}{ \mask{\buf[0..i]} }(x) = \apply{\tup{m',a'}}{ {\buf'[0..i]} }(x)$ and $\apply{\tup{m,a}}{ \drop{\buf[0..i]} }(x)  = \apply{\tup{m',a'}}{ {\buf'[0..i]} }(x)$.
\end{lemma}

\begin{proof}
This follows by inspection of $\apply{\cdot}{ \cdot }$, $\mask{\cdot}$, and $\drop{\cdot}$.
\end{proof}

\begin{lemma}[Observation equivalence preserves deep-indistinguishability]\label{lemma:tt:arch-seq:trace-equiv-implies-stepwise-indistinguishability}
    Let $p$ be a well-formed program and $C_0 = \tup{m_0,a_0,\buf_0,\CacheState_0, \BpState_0, \SchedState_0}$, $C_0' = \tup{m_0',a_0',\buf_0',\CacheState_0', \BpState_0', \SchedState_0'}$ be reachable hardware configurations.
    If
    \begin{inparaenum}[(a)]
        \item $C_0 \sim C_0'$, and
        \item for all $\buf \in \prefixes{\buf_0,\tup{m_0,a_0}}$, $\buf' \in \prefixes{\buf_0',\tup{m_0',a_0'}}$ such that $|\buf| = |\buf'|$, 
        there are $\sigma_0, \sigma_0', \sigma_1, \sigma_1', \tau, \tau'$ such that $\sigma_0 \ArchSeqInterfStep{\tau}{} \sigma_1$, $\sigma_0' \ArchSeqInterfStep{\tau'}{} \sigma_1'$, $\tau = \tau'$, $C_0 \bufEquiv{|\buf|} \sigma_0$, and $Cs_0' \bufEquiv{|\buf'|} \sigma_0'$,
    \end{inparaenum}
    then either there are $C_1, C_1'$ such that $C_0 \TtMuarchStep{}{} C_1$, $C_0' \TtMuarchStep{}{} C_1'$, and $C_1 \sim C_1'$ or there is no $C_1$ such that $C_0 \TtMuarchStep{}{} C_1$ and no $C_1'$ such that $C_0' \TtMuarchStep{}{} C_1'$.
    \end{lemma}
    
    \begin{proof}
    Let $p$ be a well-formed program and $C_0 = \tup{m_0,a_0,\buf_0,\CacheState_0, \BpState_0, \SchedState_0}$, $C_0' = \tup{m_0',a_0',\buf_0',\CacheState_0', \BpState_0', \SchedState_0'}$ be reachable hardware configurations.
    Moreover, we assume that conditions (a) and (b) holds. 
    In the following, we denote by (c) the post-condition ``either there are $C_1, C_1'$ such that $C_0 \TtMuarchStep{}{} C_1$, $C_0' \TtMuarchStep{}{} C_1'$, and $C_1 \sim C_1'$ or there is no $C_1$ such that $C_0 \TtMuarchStep{}{} C_1$ and no $C_1'$ such that $C_0' \TtMuarchStep{}{} C_1'$.''
    
    From (a), it follows that $\SchedState_0 = \SchedState_0$.
    Therefore, the directive obtained from the scheduler is the same in both cases, i.e., $\SchedNext(\SchedState_0) = \SchedNext(\SchedState_0')$.
    We proceed by case distinction on the directive $d = \SchedNext(\SchedState_0)$:
    \begin{description}
        \item[$d = \fetch{}$:]
        Therefore, we know that $u\buf_0 = \unlabelNda{\buf_0}{d} = \mask{\buf_0}$ and we can only apply one of the $\fetch{}$ rules depending on the current program counter.
        There are two cases: 
        \begin{description}
            \item[$\apply{u\buf_0}{a_0}(\pc) \neq \bot \wedge |u\buf_0| < \wMuarch$:]
            From (a) and Lemma~\ref{lemma:tt:arch-seq:masking-preserves-indistinguishability}, we get that also $\apply{u\buf_0'}{a_0'}(\pc) \neq \bot \wedge |u\buf_0'| < \wMuarch$, where $u\buf_0' = \mask{\buf_0'}$.
            There are several cases:
            \begin{description}
                \item[$\CacheAccess(\CacheState_0,  \apply{u\buf_0}{a_0}(\pc)) = \CacheHit \wedge p(\apply{u\buf_0}{a_0}(\pc)) = \pjz{x}{\lbl}$:] 
                From (a) and Lemma~\ref{lemma:tt:arch-seq:masking-preserves-indistinguishability}, we get that $\CacheState_0' = \CacheState_0$ and $\apply{u\buf_0'}{a_0'}(\pc) = \apply{u\buf_0}{a_0}(\pc)$, where $u\buf_0' = \mask{\buf_0'}$.
                Therefore, $\CacheAccess(\CacheState_0',  \apply{u\buf_0'}{a_0'}(\pc)) = \CacheHit$ holds as well.
                Moreover, from (a)  and Lemma~\ref{lemma:tt:arch-seq:masking-preserves-indistinguishability}, we also get that $\apply{u\buf_0}{a_0}(\pc)=\apply{u\buf_0'}{a_0'}(\pc)$ and, therefore, $p(\apply{u\buf_0'}{a_0'}(\pc)) = \pjz{x}{\lbl}$ as well.
                Therefore, we can apply the \textsc{Fetch-Branch-Hit} and \textsc{Step} rules to $C_0$ and $C_0'$ as follows:
                \begin{align*}
                    u\buf_0 &:= \mask{\buf_0}\\
                    \lbl_0 &:= \BpPredict(\BpState_0, \apply{u\buf_0}{a_0}(\pc))\\
                    u\buf_1 &:= u\buf_0  \concat  \tagged{\passign{\pc}{\lbl_0}}{\apply{u\buf_0}{a_0}(\pc)}\\
                    \buf_1 &:= \buf_0 \concat \labelled{\tagged{\passign{\pc}{\lbl_0}}{\apply{u\buf_0}{a_0}(\pc)}}{\emptyset}\\
                    \tup{m_0,a_0,u\buf_0,\CacheState_0,\BpState_0} &\muarchStep{\fetch{}}{} \tup{m_0, a_0, u\buf_1, \CacheUpdate(\CacheState_0, \apply{u\buf_0}{a_0}(\pc)),\BpState_0}\\
                    \tup{m_0,a_0,\buf_0,\CacheState_0,\BpState_0, \SchedState_0} &\TtMuarchStep{}{} \tup{m_0, a_0, \buf_1, \CacheUpdate(\CacheState_0, \apply{u\buf_0}{a_0}(\pc)),\BpState_0, \SchedUpdate(\SchedState_0, \BufProject{\buf_1})}\\
                    u\buf_0' &:= \mask{\buf_0'}\\
                    \lbl_0' &:= \BpPredict(\BpState_0', \apply{u\buf_0'}{a_0'}(\pc))\\
                    u\buf_1' &:= u\buf_0'  \concat  \tagged{\passign{\pc}{\lbl_0'}}{\apply{u\buf_0'}{a_0'}(\pc)}\\
                    \buf_1' &:= \buf_0' \concat \labelled{\tagged{\passign{\pc}{\lbl_0'}}{\apply{u\buf_0'}{a_0'}(\pc)}}{\emptyset}\\
                    \tup{m_0',a_0',u\buf_0',\CacheState_0',\BpState_0'} &\muarchStep{\fetch{}}{} \tup{m_0', a_0', u\buf_1', \CacheUpdate(\CacheState_0',  \apply{u\buf_0'}{a_0'}(\pc)),\BpState_0'}\\
                    \tup{m_0',a_0',\buf_0',\CacheState_0',\BpState_0', \SchedState_0'} &\TtMuarchStep{}{} \tup{m_0', a_0', \buf_1', \CacheUpdate(\CacheState_0',  \apply{u\buf_0'}{a_0'}(\pc)),\BpState_0', \SchedUpdate(\SchedState_0', \BufProject{\buf_1'})}
                \end{align*}
                We now show that $C_1 =\tup{m_0, a_0, \buf_1, \CacheUpdate(\CacheState_0, \apply{u\buf_0}{a_0}(\pc)),\BpState_0, \SchedUpdate(\SchedState_0, \BufProject{\buf_1})}$ and $C_1' = \tup{m_0', a_0', \buf_1', \CacheUpdate(\CacheState_0',  \apply{u\buf_0'}{a_0'}(\pc)),\BpState_0', \SchedUpdate(\SchedState_0', \BufProject{\buf_1'})}$ are indistinguishable, i.e., $C_1 \sim C_1'$.

                For this, we need to show that:
                \begin{description}
                    \item[$a_0 = a_0'$:]
                    This follows from (a). 
    
                    \item[$\DeepProject{\buf_1} = \DeepProject{ \buf_1'}$:] 
                    This follows from (a), $\buf_1 = \buf_0 \concat \labelled{\tagged{\passign{\pc}{\lbl_0}}{\apply{u\buf_0}{a_0}(\pc)}}{\emptyset}$, $\buf_1' = \buf_0' \concat \labelled{\tagged{\passign{\pc}{\lbl_0'}}{\apply{u\buf_0'}{a_0'}(\pc)}}{\emptyset}$, and $\lbl_0 = \lbl_0'$, which in turn follows from (a) and Lemma~\ref{lemma:tt:arch-seq:masking-preserves-indistinguishability}.

                    \item[$\CacheUpdate(\CacheState_0,  \apply{u\buf_0}{a_0}(\pc)) = \CacheUpdate(\CacheState_0',  \apply{u\buf_0'}{a_0'}(\pc))$:]
                    This follows from $\CacheState_0 = \CacheState_0'$, which in turn follows from (a), and $\apply{u\buf_0}{a_0}(\pc) = \apply{u\buf_0'}{a_0'}(\pc)$.
    
                    \item[$\BpState_0 = \BpState_0':$] 
                    This follows from (a).
    
                    \item[$\SchedUpdate(\SchedState_0, \BufProject{\buf_1}) = \SchedUpdate(\SchedState_0', \BufProject{\buf_1'})$:]
                    From (a), we have $\SchedState_0 = \SchedState_0'$.
                    Moreover, we have already shown that $\DeepProject{\buf_1} = \DeepProject{ \buf_1'}$.
                    Therefore, $\SchedUpdate(\SchedState_0, \BufProject{\buf_1}) = \SchedUpdate(\SchedState_0', \BufProject{\buf_1'})$ follows from $\DeepProject{\buf_1} = \DeepProject{ \buf_1'} \rightarrow \BufProject{\buf_1} = \BufProject{\buf_1'}$.
                \end{description}
                Therefore, $C_1 \sim C_1'$ and (c) holds.

                \item[$\CacheAccess(\CacheState_0,  \apply{\buf_0}{a_0}(\pc)) = \CacheHit \wedge p(\apply{\buf_0}{a_0}(\pc)) = \pjmp{e}$:] 
                From (a) and Lemma~\ref{lemma:tt:arch-seq:masking-preserves-indistinguishability}, we get that $\CacheState_0' = \CacheState_0$ and $\apply{u\buf_0'}{a_0'}(\pc) = \apply{u\buf_0}{a_0}(\pc)$, where $u\buf_0' = \mask{\buf_0'}$.
                Therefore, $\CacheAccess(\CacheState_0',  \apply{u\buf_0'}{a_0'}(\pc)) = \CacheHit$ holds as well.
                Moreover, from (a)  and Lemma~\ref{lemma:tt:arch-seq:masking-preserves-indistinguishability}, we also get that $\apply{u\buf_0}{a_0}(\pc)=\apply{u\buf_0'}{a_0'}(\pc)$ and, therefore, $p(\apply{u\buf_0'}{a_0'}(\pc)) = \pjmp{e}$ as well.
                Therefore, we can apply the \textsc{Fetch-Jump-Hit} and \textsc{Step} rules to $C_0$ and $C_0'$ as follows:
                \begin{align*}
                    u\buf_0 &:= \mask{\buf_0}\\
                    u\buf_1 &:= u\buf_0  \concat  \tagged{\passign{\pc}{e}}{\notags}\\
                    \buf_1 &:= \buf_0 \concat \labelled{\tagged{\passign{\pc}{e}}{\notags } }{\emptyset}\\
                    \tup{m_0,a_0,u\buf_0,\CacheState_0,\BpState_0} &\muarchStep{\fetch{}}{} \tup{m_0, a_0, u\buf_1, \CacheUpdate(\CacheState_0, \apply{u\buf_0}{a_0}(\pc)),\BpState_0}\\
                    \tup{m_0,a_0,\buf_0,\CacheState_0,\BpState_0, \SchedState_0} &\TtMuarchStep{}{} \tup{m_0, a_0, \buf_1, \CacheUpdate(\CacheState_0, \apply{u\buf_0}{a_0}(\pc)),\BpState_0, \SchedUpdate(\SchedState_0, \BufProject{\buf_1})}\\
                    u\buf_0' &:= \mask{\buf_0'}\\
                    u\buf_1' &:= u\buf_0'  \concat  \tagged{\passign{\pc}{e}}{\notags}\\
                    \buf_1' &:= \buf_0' \concat \labelled{\tagged{\passign{\pc}{e}}{ \notags }}{\emptyset}\\
                    \tup{m_0',a_0',u\buf_0',\CacheState_0',\BpState_0'} &\muarchStep{\fetch{}}{} \tup{m_0', a_0', u\buf_1', \CacheUpdate(\CacheState_0',  \apply{u\buf_0'}{a_0'}(\pc)),\BpState_0'}\\
                    \tup{m_0',a_0',\buf_0',\CacheState_0',\BpState_0', \SchedState_0'} &\TtMuarchStep{}{} \tup{m_0', a_0', \buf_1', \CacheUpdate(\CacheState_0',  \apply{u\buf_0'}{a_0'}(\pc)),\BpState_0', \SchedUpdate(\SchedState_0', \BufProject{\buf_1'})}
                \end{align*}
                We now show that $C_1 =\tup{m_0, a_0, \buf_1, \CacheUpdate(\CacheState_0, \apply{u\buf_0}{a_0}(\pc)),\BpState_0, \SchedUpdate(\SchedState_0, \BufProject{\buf_1})}$ and $C_1' = \tup{m_0', a_0', \buf_1', \CacheUpdate(\CacheState_0',  \apply{u\buf_0'}{a_0'}(\pc)),\BpState_0', \SchedUpdate(\SchedState_0', \BufProject{\buf_1'})}$ are indistinguishable, i.e., $C_1 \sim C_1'$.

                For this, we need to show that:
                \begin{description}
                    \item[$a_0 = a_0'$:]
                    This follows from (a). 
    
                    \item[$\DeepProject{\buf_1} = \DeepProject{ \buf_1'}$:] 
                    This follows from (a), $\buf_1 = \buf_0 \concat \labelled{\tagged{\passign{\pc}{e}}{ \notags }}{\emptyset}$, and $\buf_1' = \buf_0' \concat \labelled{\tagged{\passign{\pc}{e}}{ \notags }}{\emptyset}$.

                    \item[$\CacheUpdate(\CacheState_0,  \apply{u\buf_0}{a_0}(\pc)) = \CacheUpdate(\CacheState_0',  \apply{u\buf_0'}{a_0'}(\pc))$:]
                    This follows from $\CacheState_0 = \CacheState_0'$, which in turn follows from (a), and $\apply{u\buf_0}{a_0}(\pc) = \apply{u\buf_0'}{a_0'}(\pc)$.
    
                    \item[$\BpState_0 = \BpState_0':$] 
                    This follows from (a).
    
                    \item[$\SchedUpdate(\SchedState_0, \BufProject{\buf_1}) = \SchedUpdate(\SchedState_0', \BufProject{\buf_1'})$:]
                    From (a), we have $\SchedState_0 = \SchedState_0'$.
                    Moreover, we have already shown that $\DeepProject{\buf_1} = \DeepProject{ \buf_1'}$.
                    Therefore, $\SchedUpdate(\SchedState_0, \BufProject{\buf_1}) = \SchedUpdate(\SchedState_0', \BufProject{\buf_1'})$ follows from $\DeepProject{\buf_1} = \DeepProject{ \buf_1'} \rightarrow \BufProject{\buf_1} = \BufProject{\buf_1'}$.
                \end{description}
                Therefore, $C_1 \sim C_1'$ and (c) holds.
                
                \item[$\CacheAccess(\CacheState_0,  \apply{\buf_0}{a_0}(\pc)) = \CacheHit \wedge p(\apply{\buf_0}{a_0}(\pc)) \neq \pjz{x}{\lbl} \wedge p(\apply{\buf_0}{a_0}(\pc)) \neq \pjmp{e}$:]
                Let $u\buf_0'$ be $\mask{\buf_0'}$.
                Observe that from (a) it follows that $|u\buf_0| \geq \wMuarch-1$ iff $|u\buf_0'| \geq \wMuarch-1$.
                Therefore, if $|u\buf_0| \geq \wMuarch-1$, then (c) holds since both computations are stuck.
                In the following, we assume that $|u\buf_0| < \wMuarch-1$ and $|u\buf_0'| < \wMuarch-1$.
                
                From (a) and Lemma~\ref{lemma:tt:arch-seq:masking-preserves-indistinguishability}, we get that $\apply{u\buf_0'}{a_0'}(\pc) = \apply{u\buf_0}{a_0}(\pc)$ and, therefore, $\CacheAccess(\CacheState_0',  \apply{u\buf_0'}{a_0'}(\pc)) = \CacheHit$.
                Moreover, $p(\apply{u\buf_0'}{a_0'}(\pc)) \neq \pjz{x}{\lbl} \wedge p(\apply{u\buf_0'}{a_0'}(\pc)) \neq \pjmp{e}$ as well.
                Therefore, we can apply the \textsc{Fetch-Others-Hit} and \textsc{Step} rules to $C_0$ and $C_0'$ as follows:
                \begin{align*}
                    u\buf_0 &:= \mask{\buf_0}\\
                    v &:= \apply{u\buf_0}{a_0}(\pc) +1 \\
                    u\buf_1 &:= u\buf_0 \concat \tagged{p(\apply{u\buf_0}{a_0}(\pc))}{\notags} \concat \tagged{\pmarkedassign{\pc}{v}}{\notags}\\
                    \buf_1 &:= \buf_0 \concat \tagged{p(\apply{u\buf_0}{a_0}(\pc))}{L} \concat \labelled{\tagged{\pmarkedassign{\pc}{v}}{\notags}}{\emptyset}\\
                    \tup{m_0,a_0,u\buf_0,\CacheState_0,\BpState_0} &\muarchStep{\fetch{}}{} \tup{m_0, a_0, u\buf_1, \CacheUpdate(\CacheState_0, \apply{u\buf_0}{a_0}(\pc)),\BpState_0}\\
                    \tup{m_0,a_0,\buf_0,\CacheState_0,\BpState_0, \SchedState_0} &\TtMuarchStep{}{} \tup{m_0, a_0, \buf_1, \CacheUpdate(\CacheState_0, \apply{u\buf_0}{a_0}(\pc)),\BpState_0, \SchedUpdate(\SchedState_0, \BufProject{\buf_1})}\\
                    u\buf_0' &:= \mask{\buf_0'}\\
                    u\buf_1' &:= u\buf_0' \concat \tagged{p(\apply{u\buf_0'}{a_0'}(\pc))}{\notags} \concat \tagged{\pmarkedassign{\pc}{v}}{\notags}\\
                    \buf_1' &:= \buf_0' \concat \labelled{\tagged{p(\apply{u\buf_0'}{a_0'}(\pc))}{\notags}}{L'} \concat \labelled{\tagged{\pmarkedassign{\pc}{v}}{\notags}}{\emptyset}\\
                    \tup{m_0',a_0',u\buf_0',\CacheState_0',\BpState_0'} &\muarchStep{\fetch{}}{} \tup{m_0', a_0', u\buf_1', \CacheUpdate(\CacheState_0',  \apply{u\buf_0'}{a_0'}(\pc)),\BpState_0'}\\
                    \tup{m_0',a_0',\buf_0',\CacheState_0',\BpState_0', \SchedState_0'} & \TtMuarchStep{}{} \tup{m_0', a_0', \buf_1', \CacheUpdate(\CacheState_0',  \apply{u\buf_0'}{a_0'}(\pc)),\BpState_0', \SchedUpdate(\SchedState_0', \BufProject{\buf_1'})}
                \end{align*}
                We now show that $C_1 =  \tup{m_0, a_0, \buf_1, \CacheUpdate(\CacheState_0, \apply{u\buf_0}{a_0}(\pc)),\BpState_0, \SchedUpdate(\SchedState_0, \BufProject{\buf_1})}$ and $C_1' = \tup{m_0', a_0', \buf_1', \CacheUpdate(\CacheState_0',  \apply{u\buf_0'}{a_0'}(\pc)),\BpState_0', \SchedUpdate(\SchedState_0', \BufProject{\buf_1'})}$ are indistinguishable, i.e., $C_1 \sim C_1'$.

                For this, we need to show that:
                \begin{description}
                    \item[$a_0 = a_0'$:]
                    This follows from (a). 
    
                    \item[$\DeepProject{\buf_1} = \DeepProject{ \buf_1'}$:] 
                    There are three cases:
                    \begin{description}
                        \item[$p(\apply{u\buf_0}{a_0}(\pc)) = \tagged{\passign{x}{e}}{\notags}$:]
                        Then, we have that $\buf_1 = \buf_0 \concat \labelled{\tagged{\passign{x}{e}}{\notags}}{ \labels{\buf_0}(e) } \concat \labelled{\tagged{\pmarkedassign{\pc}{v}}{\notags}}{\emptyset}$ and $\buf_1' = \buf_0' \concat \labelled{\tagged{ \passign{x}{e} }{\notags}}{ \labels{\buf_0'}(e) } \concat \labelled{\tagged{\pmarkedassign{\pc}{v}}{\notags}}{\emptyset}$.
                        Therefore, $\DeepProject{\buf_1} = \DeepProject{ \buf_1'}$ follows from $\DeepProject{\buf_0} = \DeepProject{ \buf_0'}$ and $\labels{\buf_0}(e) = \labels{\buf_0'}(e)$, which follows from (a).

                        \item[$p(\apply{u\buf_0}{a_0}(\pc)) = \tagged{\pload{x}{e}}{T}$:]
                        Then, we have that $\buf_1 = \buf_0 \concat \labelled{\tagged{\pload{x}{e}}{\notags}}{  \{ j \mid \elt{\buf_0'}{j} = \tagged{\passign{\pc}{\ell}}{\ell_0} \wedge \ell_0 \neq \notags \} } \concat \labelled{\tagged{\pmarkedassign{\pc}{v}}{\notags}}{\emptyset}$ and $\buf_1' = \buf_0' \concat \labelled{\tagged{ \pload{x}{e} }{\notags}}{  \{ j \mid \elt{\buf_0}{j} = \tagged{\passign{\pc}{\ell}}{\ell_0} \wedge \ell_0 \neq \notags \}  } \concat \labelled{\tagged{\pmarkedassign{\pc}{v}}{\notags}}{\emptyset}$.
                        Therefore, $\DeepProject{\buf_1} = \DeepProject{ \buf_1'}$ follows from $\DeepProject{\buf_0} = \DeepProject{ \buf_0'}$ and $\{ j \mid \elt{\buf_0'}{j} = \tagged{\passign{\pc}{\ell}}{\ell_0} \wedge \ell_0 \neq \notags \} = \{ j \mid \elt{\buf_0}{j} = \tagged{\passign{\pc}{\ell}}{\ell_0} \wedge \ell_0 \neq \notags \}$, which follows from (a).

                        \item[$p(\apply{u\buf_0}{a_0}(\pc)) \neq \tagged{\passign{x}{e}}{T} \wedge p(\apply{u\buf_0}{a_0}(\pc)) \neq \tagged{\pload{x}{e}}{T}$:] 
                        Then, $\buf_1 = \buf_0 \concat \tagged{p(\apply{u\buf_0}{a_0}(\pc))}{\emptyset} \concat \labelled{\tagged{\pmarkedassign{\pc}{v}}{\notags}}{\emptyset}$ and $\buf_1' = \buf_0' \concat \labelled{\tagged{p(\apply{u\buf_0'}{a_0'}(\pc))}{\notags}}{\emptyset} \concat \labelled{\tagged{\pmarkedassign{\pc}{v}}{\notags}}{\emptyset}$.
                        Therefore,  $\DeepProject{\buf_1} = \DeepProject{ \buf_1'}$ follows from $\DeepProject{\buf_0} = \DeepProject{ \buf_0'}$ and $\apply{u\buf_0}{a_0}(\pc)) = \apply{u\buf_0'}{a_0'}(\pc))$, which follows from (a) and Lemma~\ref{lemma:tt:arch-seq:masking-preserves-indistinguishability}.
                    \end{description}
    
                    \item[$\CacheUpdate(\CacheState_0,  \apply{u\buf_0}{a_0}(\pc)) = \CacheUpdate(\CacheState_0',  \apply{u\buf_0'}{a_0'}(\pc))$:]
                    This follows from $\CacheState_0 = \CacheState_0'$, which in turn follows from (a), and $\apply{u\buf_0}{a_0}(\pc) = \apply{u\buf_0'}{a_0'}(\pc)$.
    
                    \item[$\BpState_0 = \BpState_0':$] 
                    This follows from (a).
    
                    \item[$\SchedUpdate(\SchedState_0, \BufProject{\buf_1}) = \SchedUpdate(\SchedState_0', \BufProject{\buf_1'})$:]
                    From (a), we have $\SchedState_0 = \SchedState_0'$.
                    Moreover, we have already shown that$\DeepProject{\buf_1} = \DeepProject{ \buf_1'} $.
                    Therefore, $\SchedUpdate(\SchedState_0, \BufProject{\buf_1}) = \SchedUpdate(\SchedState_0', \BufProject{\buf_1'})$ follows from $\DeepProject{\buf_1} = \DeepProject{ \buf_1'} \rightarrow \BufProject{\buf_1} = \BufProject{\buf_1'}$.
                \end{description}
                Therefore, $C_1 \sim C_1'$ and (c) holds.
    
                \item[$\CacheAccess(\CacheState_0,  \apply{u\buf_0}{a_0}(\pc)) = \CacheMiss$:]
                From (a) and Lemma~\ref{lemma:tt:arch-seq:masking-preserves-indistinguishability}, we get that $\CacheState_0' = \CacheState_0$ and $\apply{u\buf_0'}{a_0'}(\pc) = \apply{u\buf_0}{a_0}(\pc)$, where $u\buf_0' = \mask{\buf_0'}$.
                Therefore, $\CacheAccess(\CacheState_0',  \apply{u\buf_0'}{a_0'}(\pc)) = \CacheMiss$.
                Therefore, we can apply the \textsc{Fetch-Miss} and \textsc{Step} rules to $C_0$ and $C_0'$ as follows:
                \begin{align*}
                    u\buf_0 &:= \mask{\buf_0}\\
                    \buf_1 &:= \labelNda{u\buf_0}{\buf_0}{\fetch{}}\\
                    \tup{m_0,a_0,u\buf_0,\CacheState_0,\BpState_0} &\muarchStep{\fetch{}}{} \tup{m_0, a_0, u\buf_0, \CacheUpdate(\CacheState_0, \apply{u\buf_0}{a_0}(\pc)),\BpState_0}\\
                    \tup{m_0,a_0,\buf_0,\CacheState_0,\BpState_0, \SchedState_0} &\TtMuarchStep{}{} \tup{m_0, a_0, \buf_1, \CacheUpdate(\CacheState_0,  \apply{u\buf_0}{a_0}(\pc)),\BpState_0, \SchedUpdate(\SchedState_0, \BufProject{\buf_1})}\\
                    u\buf_0' &:= \mask{\buf_0'}\\
                    \buf_1' &:= \labelNda{u\buf_0'}{\buf_0'}{\fetch{}}\\
                    \tup{m_0',a_0',u\buf_0',\CacheState_0',\BpState_0'} &\muarchStep{\fetch{}}{} \tup{m_0', a_0', u\buf_0', \CacheUpdate(\CacheState_0',  \apply{\buf_0'}{a_0'}(\pc)),\BpState_0'}\\
                    \tup{m_0',a_0',\buf_0',\CacheState_0',\BpState_0', \SchedState_0'} &\TtMuarchStep{}{} \tup{m_0', a_0', \buf_1', \CacheUpdate(\CacheState_0',  \apply{u\buf_0'}{a_0'}(\pc)),\BpState_0', \SchedUpdate(\SchedState_0', \BufProject{\buf_1'})}
                \end{align*}
                We now show that $C_1 = \tup{m_0, a_0, \buf_1, \CacheUpdate(\CacheState_0,  \apply{u\buf_0}{a_0}(\pc)),\BpState_0, \SchedUpdate(\SchedState_0, \BufProject{\buf_1})}$ and $C_1' = \tup{m_0', a_0', \buf_1', \CacheUpdate(\CacheState_0',  \apply{u\buf_0'}{a_0'}(\pc)),\BpState_0', \SchedUpdate(\SchedState_0', \BufProject{\buf_1'})}$ are indistinguishable, i.e., $C_1 \sim C_1'$.
                For this, we need to show that:
                \begin{description}
                    \item[$a_0 = a_0'$:]
                    This follows from (a). 
    
                    \item[$\DeepProject{\buf_1} = \DeepProject{ \buf_1'}$:] 
                    From $\buf_1 = \labelNda{u\buf_0}{\buf_0}{\fetch{}}$ and $|\buf_0| = |u\buf_0|$, we get that $\buf_1 = \buf_0$.
                    Similarly, $\buf_1' = \buf_0'$ follows from $\buf_1' = \labelNda{u\buf_0'}{\buf_0'}{\fetch{}}$ and $|\buf_0'| = |u\buf_0'|$.
                    Therefore, $\DeepProject{\buf_1} = \DeepProject{ \buf_1'}$ follows from (a).
    
                    \item[$\CacheUpdate(\CacheState_0,  \apply{u\buf_0}{a_0}(\pc)) = \CacheUpdate(\CacheState_0',  \apply{u\buf_0'}{a_0'}(\pc))$:]
                    This follows from $\CacheState_0 = \CacheState_0'$, which in turn follows from (a), and $\apply{u\buf_0}{a_0}(\pc) = \apply{u\buf_0'}{a_0'}(\pc)$.
    
                    \item[$\BpState_0 = \BpState_0':$] 
                    This follows from (a).
    
                    \item[$\SchedUpdate(\SchedState_0, \BufProject{\buf_1}) = \SchedUpdate(\SchedState_0', \BufProject{\buf_1'})$:]
                    From (a), we have $\SchedState_0 = \SchedState_0'$.
                    Moreover, we have already shown that $\buf_0 = \buf_1$ and $\buf_0' = \buf_1'$.
                    Therefore, $\SchedUpdate(\SchedState_0, \BufProject{\buf_1}) = \SchedUpdate(\SchedState_0', \BufProject{\buf_1'})$ follows from $\DeepProject{\buf_1} = \DeepProject{ \buf_1'} \rightarrow \BufProject{\buf_1} = \BufProject{\buf_1'}$.
                \end{description}
                Therefore, $C_1 \sim C_1'$ and (c) holds.

            \end{description}
            
            \item[$\apply{u\buf_0}{a_0}(\pc) = \bot$:]
            Then, from (a) and Lemma~\ref{lemma:tt:arch-seq:masking-preserves-indistinguishability}, we immediately get that $\apply{\mask{\buf_0'}}{a_0'}(\pc) = \bot$ holds as well.
            Therefore, both computations are stuck and (c) holds.

            \item[$|u\buf_0| \geq \wMuarch$:]
            Then, from (a), we immediately get that  $|\mask{\buf_0'}| \geq \wMuarch$.
            Therefore, both computations are stuck and (c) holds.

        \end{description}
        Therefore, (c) holds for all the cases.
        
        \item[$d = \execute{i}:$]
        Therefore, we can only apply one of the $\execute{}$ rules.
        There are two cases:
        \begin{description}
            \item[$\transmitGadget{\elt{\buf_0}{i}}$:]
            Then, from (a), we have that  $\transmitGadget{\elt{\buf_0'}{i}}$.
            Therefore, we have that $u\buf_0 = \unlabelNda{\buf_0}{\execute{i}} = \mask{\buf_0}$ and $u\buf_0' = \unlabelNda{\buf_0'}{\execute{i}} = \mask{\buf_0'}$.
            There are two cases:
            \begin{description}
                \item[$i \leq |u\buf_0| \wedge \tagged{\pbarrier}{T'} \not\in {u\buf_0[0..i-1]}$:]
                There are several cases depending on the $i$-th command in the reorder buffer: 
                \begin{description}
                    \item[$\elt{u\buf_0}{i} = \tagged{\pload{x}{e}}{T}$:]   
                    From (a), we have $\elt{u\buf_0'}{i} =  \tagged{\pload{x}{e}}{T}$ as well. 
                    There are two cases:
                    \begin{description}
                        \item[$\tagged{\pstore{x'}{e'}}{T''} \not\in u\buf_0{[0..i-1]}$:]
                        From (a), we also have that $\tagged{\pstore{x'}{e'}}{T''} \not\in u\buf_0'[0..i-1]$.
                        Observe that $\exprEval{e}{\apply{\buf_0[0..i-1]}{a_0}} =  \exprEval{e}{\apply{\buf_0'[0..i-1]}{a_0'}}$ follows from $a_0 = a_0'$, $\DeepProject{u\buf_0}_{\tup{m_0,a_0}} = \DeepProject{u\buf_0'}_{\tup{m_0',a_0'}}$, which follows from (a), and Lemma~\ref{lemma:tt:arch-seq:masking-preserves-indistinguishability}.
                        There are two cases:
                        \begin{description}
                            \item[$\exprEval{e}{\apply{u\buf_0{[0..i-1]}}{a_0}} = \bot$:]
                            Then, both configurations are stuck and (c) holds.

                            \item[$\exprEval{e}{\apply{u\buf_0{[0..i-1]}}{a_0}} \neq \bot$:] 
                            Then, let $n = \exprEval{e}{\apply{u\buf_0{[0..i-1]}}{a_0}} = \exprEval{e}{\apply{u\buf_0'{[0..i-1]}}{a_0'}} $.
                            There are two cases:
                            \begin{description}
                                \item[$\CacheAccess(\CacheState_0, n) = \CacheHit$:]
                                From (a), we have that $\CacheState_0 = \CacheState_0'$.
                                Therefore, $\CacheAccess(\CacheState_0', n) = \CacheHit$ holds as well, and we can apply the \textsc{Execute-Load-Hit} and \textsc{Step} rules to $C_0$ and $C_0'$ as follows:
                                \begin{align*}
                                    \buf_0 &:= \buf_0[0..i-1] \concat \labelled{\tagged{\pload{x}{e}}{T}}{L} \concat \buf_0[i+1 .. |\buf_0|]\\
                                    u\buf_0 &:= \mask{\buf_0[0..i-1]} \concat \tagged{\pload{x}{e}}{T} \concat \mask{\buf_0[i+1 .. |\buf_0|]}\\
                                    u\buf_1 &:= \mask{\buf_0[0..i-1]} \concat \tagged{\passign{x}{m_0(n)}}{T} \concat \mask{\buf_0[i+1 .. |\buf_0|]}\\
                                    \buf_1 &:= \buf_0[0..i-1] \concat \labelled{\tagged{\passign{x}{m_0(n)}}{T}}{L} \concat \buf_0[i+1 .. |\buf_0|]\\
                                    \tup{m_0,a_0,u\buf_0, \CacheState_0, \BpState_0} &\muarchStep{\execute{i}}{} \tup{m_0,a_0,u\buf_1, \CacheUpdate(\CacheState_0,n), \BpState_0}\\
                                    \tup{m_0,a_0,\buf_0, \CacheState_0,  \BpState_0 , \SchedState_0} &\TtMuarchStep{}{} \tup{m_0,a_0,\buf_1, \CacheUpdate(\CacheState_0,n), \BpState_0,\SchedUpdate(\SchedState_0, \BufProject{\buf_1})}\\
                                    \buf_0' &:= \buf_0'[0..i-1] \concat \labelled{\tagged{\pload{x}{e}}{T}}{L} \concat \buf_0'[i+1 .. |\buf_0'|]\\
                                    u\buf_0' &:= \mask{\buf_0'[0..i-1]} \concat \tagged{\pload{x}{e}}{T} \concat \mask{\buf_0'[i+1 .. |\buf_0'|]}\\
                                    u\buf_1' &:= \mask{\buf_0'[0..i-1]} \concat \tagged{\passign{x}{m_0'(n)}}{T} \concat \mask{\buf_0'[i+1 .. |\buf_0'|]}\\
                                    \buf_1' &:= \buf_0'[0..i-1] \concat \labelled{\tagged{\passign{x}{m_0'(n)}}{T}}{L} \concat \buf_0'[i+1 .. |\buf_0'|]\\
                                    \tup{m_0',a_0',\buf_0', \CacheState_0', \BpState_0'} &\muarchStep{\execute{i}}{} \tup{m_0',a_0',\buf_1', \CacheUpdate(\CacheState_0',n), \BpState_0'}\\
                                    \tup{m_0',a_0',\buf_0', \CacheState_0', \BpState_0', \SchedState_0'} & \TtMuarchStep{}{} \tup{m_0',a_0',\buf_1', \CacheUpdate(\CacheState_0',n), \BpState_0',\SchedUpdate(\SchedState_0', \BufProject{\buf_1'})}
                                \end{align*}
                                We now show that $C_1 = \tup{m_0,a_0,\buf_1, \CacheUpdate(\CacheState_0,n), \BpState_0,\SchedUpdate(\SchedState_0, \BufProject{\buf_1})}$ and $C_1' = \tup{m_0',a_0',\buf_1', \CacheUpdate(\CacheState_0',n), \BpState_0',\SchedUpdate(\SchedState_0', \BufProject{\buf_1'})}$ are indistinguishable, i.e., i.e., $C_1 \sim C_1'$.
                                For this, we need to show that:
                                    \begin{description}
                                        \item[$a_0 = a_0'$:]
                                        This follows from (a). 
                            
                                        \item[$\DeepProject{\buf_1} = \DeepProject{ \buf_1'}$:] 
                                        There are two cases:
                                        \begin{description}
                                            \item[$L \neq \emptyset$:]
                                            Then, the $i$-th entry in $\DeepProject{\buf_1}$ is $\labelled{\tagged{\passign{x}{\unresolved}}{T}}{L}$ and the $i$-th entry in $\DeepProject{\buf_1'}$ is $\labelled{\tagged{\passign{x}{\unresolved}}{T}}{L}$.
                                            From this and (a), we get that $\DeepProject{\buf_1} = \DeepProject{ \buf_1'}$. 
                        
                                            \item[$L = \emptyset$:]
                                            Then, the $i$-th entry in $\DeepProject{\buf_1}$ is $\labelled{\tagged{\passign{x}{m_0(n)}}{T}}{L}$ and the $i$-th entry in $\DeepProject{\buf_1'}$ is $\labelled{\tagged{\passign{x}{m_0'(n)}}{T}}{L}$.
                                            Therefore,  to show that $\DeepProject{\buf_1} = \DeepProject{ \buf_1'}_{\tup{m_0',a_0'}}$ holds, we need to show that $m_0(n) = m_0'(n)$. 
                                            
                                            We now show that $m_0(n) = m_0'(n')$.
                                            Observe, first, that from  the well-formedness of $\buf_0, \buf_0'$, we know that there are no unresolved branch instruction in $\buf_0{[0..i-1]}, \buf_0'{[0..i-1]}$.
                                            Since $C_0, C_0'$ are reachable configurations, the buffers $\buf_0, \buf_0'$ are well-formed (see Lemma~\ref{lemma:tt:arch-seq:buffers-well-formedness}).
                                            Therefore, $\buf_0[0..i-1] \in \prefixes{\buf_0, \tup{m_0,a_0}}$ and  $\buf_0'[0..i-1] \in \prefixes{\buf_0', \tup{m_0,a_0}}$ because there are no unresolved branch instructions in $\buf_0[0..i-1]$ and $\buf_0'[0..i-1]$.
                                            From (b), therefore, there are configurations $\sigma_0, \sigma_0', \sigma_1, \sigma_1'$ such that $C_0 \bufEquiv{|\buf_0[0..i-1]|} \sigma_0$, $C_0' \bufEquiv{|\buf_0[0..i-1]|} \sigma_0'$, $\sigma_0 \ArchSeqInterfStep{\tau}{} \sigma_1$,  $\sigma_0' \ArchSeqInterfStep{\tau'}{} \sigma_1'$, and $\tau = \tau'$. 
                                            From $C_0 \bufEquiv{|\buf_0[0..i-1]|} \sigma_0$, $C_0' \bufEquiv{|\buf_0[0..i-1]|} \sigma_0'$, and the well-formedness of the buffers, we know that $p(\sigma_0(\pc)) = p(\sigma_0'(\pc)) = \pload{x}{e}$.
                                            From $\ArchSpecInterf{\cdot}$, we have that $\tau = \loadObs{ \exprEval{e}{\sigma_0} = \sigma_0(\exprEval{e}{\sigma_0}) }$  and $\tau' = \loadObs{ \exprEval{e}{\sigma_0'}  = \sigma_0'(\exprEval{e}{\sigma_0'}) }$.
                                            From $\tau=\tau'$, we get that $\exprEval{e}{\sigma_0} = \exprEval{e}{\sigma_0'}$ and $\sigma_0'(\exprEval{e}{\sigma_0'}) = \sigma_0(\exprEval{e}{\sigma_0})$.
                                            From this, $C_0 \bufEquiv{|\buf_0[0..i-1]|} \sigma_0$,  $C_0' \bufEquiv{|\buf_0'[0..i-1]|} \sigma_0'$, and Lemma~\ref{lemma:tt:arch-seq:mask-drop-preserve-values}, we finally get $\exprEval{e}{\apply{\buf_0{[0..i-1]}}{a_0}} = \exprEval{e}{\apply{\buf_0'{[0..i-1]}}{a_0'}(x)} = n$ and $m_0(n) = m_0'(n)$.
                                        \end{description}
                            
                                        \item[$\CacheUpdate(\CacheState_0,n) = \CacheUpdate(\CacheState_0',n)$:]
                                        This follows from $\CacheState_0 = \CacheState_0'$, which follows from (a).
                            
                                        \item[$\BpState_0= \BpState_0':$] 
                                        This follows from  (a).
                            
                                        \item[$\SchedUpdate(\SchedState_0, \BufProject{\buf_1}) = \SchedUpdate(\SchedState_0', \BufProject{\buf_1'})$:]
                                        This follows from $\SchedState_0 = \SchedState_0'$, which follows from (a), $\DeepProject{\buf_1} = \DeepProject{ \buf_1'}$, and $\DeepProject{\buf_1} = \DeepProject{ \buf_1'} \rightarrow \BufProject{\buf_1} = \BufProject{\buf_1'}$.
                                    \end{description}
                                    Therefore, $C_1 \sim C_1'$ and (c) holds.
        
                                \item[$\CacheAccess(\CacheState_0, \exprEval{e}{\apply{\buf_0{[0..i-1]}}{a_0}}) = \CacheMiss$:]
                                The proof of this case is similar to the one for the $\CacheHit$ case (except that we apply the \textsc{Execute-Load-Miss} rule and we do not need to rely on observations produced by the $\ArchSeqInterf{\cdot}$ contract). 
                            \end{description}
                        \end{description}

                        \item[$\tagged{\pstore{x'}{e'}}{T''} \in u\buf_0{[0..i-1]}$:]
                        From (a), we also have that $\tagged{\pstore{x'}{e'}}{T''} \in u\buf_0'[0..i-1]$.
                        Therefore, both configurations are stuck and (c) holds. 
                    \end{description}

                    \item[$\elt{u\buf_0}{i} =  \tagged{\passign{\pc}{\lbl}}{\lbl_0} \wedge \ell_0 \neq \emptysequence$:]
                    From (a),  we have that $\elt{u\buf_0'}{i} =  \tagged{\passign{\pc}{\lbl}}{\lbl_0} \wedge \ell_0 \neq \emptysequence$ holds.
                    Observe that $p(\lbl_0) = \pjz{x}{\lbl''}$.
                    Moreover, from Lemma~\ref{lemma:tt:arch-seq:masking-preserves-indistinguishability}, we have that  $\apply{u\buf_0[0..i-1]}{a_0}(x) = \apply{u\buf_0'[0..i-1]}{a_0'}(x)$.
                    There are three cases:
                    \begin{description}
                        \item[$\apply{\buf_0[0..i-1]}{a_0}(x) = \bot$:]
                        Then, both configurations are stuck and (c) holds.

                        \item[$( \apply{u\buf_0[0..i-1]}{a_0}(x) = 0 \wedge \lbl = \lbl'') \vee (\apply{u\buf_0[0..i-1]}{a_0}(x) \in \Val \setminus \{0,\bot\} \wedge \lbl = \ell_0+1)$:]
                        From $\apply{u\buf_0{[0..i-1]}}{a_0}(x) = \apply{u\buf_0'{[0..i-1]}}{a_0'}(x)$ and (a), we also get $( \apply{u\buf_0'[0..i-1]}{a_0'}(x) = 0 \wedge \lbl = \lbl'') \vee (\apply{u\buf_0'[0..i-1]}{a_0'}(x) \in \Val \setminus \{0,\bot\} \wedge \lbl = \ell_0+1)$.
                        Therefore,  we can apply the \textsc{Execute-Branch-Commit} and \textsc{Step} rules to $C_0$ and $C_0'$ as follows:
                        \begin{align*}
                            \buf_0 &:= \buf_0[0..i-1] \concat \labelled{\tagged{\passign{\pc}{\ell}}{\ell_0}}{\emptyset} \concat \buf_0[i+1 .. |\buf_0|]\\
                            u\buf_0 &:= \mask{\buf_0[0..i-1]} \concat \tagged{\passign{\pc}{\ell}}{\ell_0} \concat \mask{\buf_0[i+1 .. |\buf_0|]}\\
                            u\buf_1 &:= \mask{\buf_0[0..i-1]} \concat \tagged{\passign{\pc}{\ell}}{\notags} \concat \mask{\buf_0[i+1 .. |\buf_0|]}\\
                            \buf_1 &:= \buf_0[0..i-1] \concat \tagged{\passign{\pc}{\ell}}{\notags} \concat \mathit{strips}(\buf_0[i+1 .. |\buf_0|],i)\\
                            \tup{m_0,a_0,u\buf_0, \CacheState_0, \BpState_0} &\muarchStep{\execute{i}}{} \tup{m_0,a_0,u\buf_1, \CacheState_0, \BpUpdate(\BpState_0, \ell_0, \ell)}\\
                            \tup{m_0,a_0,\buf_0, \CacheState_0, \BpState_0 , \SchedState_0} &\TtMuarchStep{}{} \tup{m_0,a_0,\buf_1, \CacheState_0, \BpUpdate(\BpState_0, \ell_0, \ell),\SchedUpdate(\SchedState_0, \BufProject{\buf_1})}\\
                            \buf_0' &:= \buf_0'[0..i-1] \concat \labelled{\tagged{\passign{\pc}{\ell}}{\ell_0}}{\emptyset} \concat \buf_0'[i+1 .. |\buf_0'|]\\
                            u\buf_0' &:= \mask{\buf_0'[0..i-1]} \concat \tagged{\passign{\pc}{\ell}}{\ell_0} \concat \mask{\buf_0'[i+1 .. |\buf_0|]}\\
                            u\buf_1' &:= \mask{\buf_0'[0..i-1]} \concat \tagged{\passign{\pc}{\ell}}{\notags} \concat \mask{\buf_0'[i+1 .. |\buf_0'|]}\\
                            \buf_1' &:= \buf_0'[0..i-1] \concat \tagged{\passign{\pc}{\ell}}{\notags} \concat \mathit{strips}(\buf_0'[i+1 .. |\buf_0'|],i)\\
                            \tup{m_0',a_0',u\buf_0', \CacheState_0', \BpState'} &\muarchStep{\execute{i}}{} \tup{m_0',a_0',u\buf_1', \CacheState_0', \BpUpdate(\BpState_0', \ell_0, \ell)}\\
                            \tup{m_0',a_0',\buf_0', \CacheState_0', \BpState_0', \SchedState_0'} & \TtMuarchStep{}{} \tup{m_0',a_0',\buf_1', \CacheState_0', \BpUpdate(\BpState_0', \ell_0, \ell),\SchedUpdate(\SchedState_0', \BufProject{\buf_1'})}
                        \end{align*}
                        We now show that $C_1 = \tup{m_0,a_0,\buf_1, \CacheState_0, \BpUpdate(\BpState_0, \ell_0, \ell),\SchedUpdate(\SchedState_0, \BufProject{\buf_1})}$ and $C_1' =  \tup{m_0',a_0',\buf_1', \CacheState_0', \BpUpdate(\BpState_0', \ell_0, \ell),\SchedUpdate(\SchedState_0', \BufProject{\buf_1'})}$ are indistinguishable, i.e., i.e., $C_1 \sim C_1'$.
        
                        For this, we need to show that:
                            \begin{description}
                                \item[$a_0 = a_0'$:]
                                This follows from (a). 
                    
                                \item[$\DeepProject{\buf_1} = \DeepProject{\buf_1'}$:] 
                                Observe that $\buf_1 = \buf_0[0..i-1] \concat \tagged{\passign{\pc}{\ell}}{\notags} \concat \mathit{strips}(\buf_0[i+1 .. |\buf_0|],i)$ and $\buf_1' = \buf_0'[0..i-1] \concat \tagged{\passign{\pc}{\ell}}{\notags} \concat \mathit{strips}(\buf_0'[i+1 .. |\buf_0'|],i)$.
                                For $\DeepProject{\buf_1} = \DeepProject{\buf_1'}$ to hold, we need to show that $\DeepProject{\mathit{strips}(\buf_0[i+1 .. |\buf_0|],i)} = \DeepProject{\mathit{strips}(\buf_0'[i+1 .. |\buf_0'|],i)}$.
                                Observe, first, that whenever $\elt{\DeepProject{\buf_0[i+1 .. |\buf_0|]}}{j} = \labelled{\tagged{i}{T}}{L'}$ and $L' \neq \{ i\}$, then  $\elt{\DeepProject{\buf_0[i+1 .. |\buf_0|]}}{j} = \elt{\DeepProject{\mathit{strips}(\buf_0[i+1 .. |\buf_0|],i)}}{j}$.
                                Therefore, we just need to show that $\elt{\DeepProject{\mathit{strips}(\buf_0[i+1 .. |\buf_0|],i)}}{j} = \elt{\DeepProject{\mathit{strips}(\buf_0'[i+1 .. |\buf_0'|],i)}}{j}$ for all $j$ where the label is $\{ i  \}$.
                                Let $j_1, \ldots, j_n$, where $j_1 < j_2 < \ldots < j_n$ be the positions of all entries where the label is $\{ i  \}$.
                                We now show, by induction on $1 \leq k \leq n$,  that $\elt{\DeepProject{\mathit{strips}(\buf_0[i+1 .. |\buf_0|],i)}}{j_k} = \elt{\DeepProject{\mathit{strips}(\buf_0'[i+1 .. |\buf_0'|],i)}}{j_k}$.
                                \begin{description}
                                    \item[Base case:]
                                    For the base case, let $k =1$.
                                    Then, $j_1$ is the smallest index such that the label is $\{ i\}$.
                                    From (a) and the well-formedness of buffers, therefore, we have that  $p(\apply{u\buf_0[0..j_1]}{a_0}(\pc)) = p(\apply{u\buf_0'[0..j_1]}{a_0'}(\pc)) = \pload{x}{e}$.
                                    If the $j_1$-th entry is still unresolved, then $\elt{\DeepProject{\mathit{strips}(\buf_0[i+1 .. |\buf_0|],i)}}{j_1} = \elt{\DeepProject{\mathit{strips}(\buf_0'[i+1 .. |\buf_0'|],i)}}{j_1}$ follows directly from $\elt{\DeepProject{\buf_0[i+1 .. |\buf_0|]}}{j_1} = \elt{\DeepProject{\buf_0'[i+1 .. |\buf_0'|]}}{j_1}$.
                                    If the $j_1$-th entry has been resolved (i.e., $\elt{\DeepProject{\mathit{strips}(\buf_0[i+1 .. |\buf_0|],i)}}{j_1} = \labelled{\tagged{\passign{x}{v}}{\notags}}{\{i\}}$ and $\elt{\DeepProject{\mathit{strips}(\buf_0'[i+1 .. |\buf_0'|],i)}}{j_1} = \labelled{\tagged{\passign{x}{v'}}{\notags}}{\{i\}}$), then $v = v'$ follows from the fact that the \textbf{load} operation is non-speculative and from the observations in the $\ArchSeqInterf{\cdot}$ contract (similarly to the proof of the case $\elt{u\buf_0}{i} = \tagged{\pload{x}{e}}{T}$ when $L = \emptyset$).

                                    \item[Induction step:]
                                    For the induction step, we assume that $\elt{\DeepProject{\mathit{strips}(\buf_0[i+1 .. |\buf_0|],i)}}{j_{m}} = \elt{\DeepProject{\mathit{strips}(\buf_0'[i+1 .. |\buf_0'|],i)}}{j_{m}}$ holds for all $m < k$ and we show that it holds for $k$ as well.
                                    From (a) and the well-formedness of buffers, therefore, we have that  $\apply{u\buf_0[0..j_1]}{a_0}(\pc) = \apply{u\buf_0'[0..j_1]}{a_0'}(\pc)$.
                                    There are two cases:
                                    \begin{description}
                                        \item[$p(\apply{u\buf_0[0..j_1]}{a_0}(\pc)) = \pload{x}{e}$:]
                                        Then, the proof of $\elt{\DeepProject{\mathit{strips}(\buf_0[i+1 .. |\buf_0|],i)}}{j_{k}} = \elt{\DeepProject{\mathit{strips}(\buf_0'[i+1 .. |\buf_0'|],i)}}{j_{k}}$ is similar to the base case.

                                        \item[$p(\apply{u\buf_0[0..j_1]}{a_0}(\pc)) = \passign{x}{e}$:]
                                        Again, we have two cases. 
                                        If the entry is still unresolved, then $\elt{\DeepProject{\mathit{strips}(\buf_0[i+1 .. |\buf_0|],i)}}{j_k} = \elt{\DeepProject{\mathit{strips}(\buf_0'[i+1 .. |\buf_0'|],i)}}{j_k}$ follows directly from $\elt{\DeepProject{\buf_0[i+1 .. |\buf_0|]}}{j_1} = \elt{\DeepProject{\buf_0'[i+1 .. |\buf_0'|]}}{j_1}$.
                                        If the entry has been resolved, then its value depends only on earlier resolved values labelled either $\emptyset$ (which are the same in both buffers thanks to $\DeepProject{\buf_0} = \DeepProject{\buf_0'}$) or $\{ i\}$ (which are the same from the induction hypothesis).
                                        So, $\elt{\DeepProject{\mathit{strips}(\buf_0[i+1 .. |\buf_0|],i)}}{j_k} = \elt{\DeepProject{\mathit{strips}(\buf_0'[i+1 .. |\buf_0'|],i)}}{j_k}$ holds.

                                    \end{description}
                                \end{description}
                                Therefore,  $\DeepProject{\mathit{strips}(\buf_0[i+1 .. |\buf_0|],i)} = \DeepProject{\mathit{strips}(\buf_0'[i+1 .. |\buf_0'|],i)}$.

                                \item[$\CacheState_0 = \CacheState_0'$:]
                                This follows from (a).
                    
                                \item[$\BpUpdate(\BpState_0, \ell_0, \ell) = \BpUpdate(\BpState_0', \ell_0, \ell):$] 
                                This follows from $\BpState_0 = \BpState_0'$, which follows from (a).
                    
                                \item[$\SchedUpdate(\SchedState_0, \BufProject{\buf_1}) = \SchedUpdate(\SchedState_0', \BufProject{\buf_1'})$:]
                                This follows from $\SchedState_0 = \SchedState_0'$, $\DeepProject{\buf_1}_{\tup{m_0,a_0}} = \DeepProject{\buf_1'}_{\tup{m_0',a_0'}}$, $\DeepProject{\buf_1}_{\tup{m_0,a_0}} = \DeepProject{\buf_1'}_{\tup{m_0',a_0'}} \rightarrow \BufProject{\buf_1} = \BufProject{\buf_1'}$.
                            \end{description}
                            Therefore, $C_1 \sim C_1'$ and (c) holds.
        
                        \item[$( \apply{u\buf_0[0..i-1]}{a_0}(x) = 0 \wedge \lbl \neq \lbl'') \vee (\apply{u\buf_0[0..i-1]}{a_0}(x) \in \Val \setminus \{0,\bot\} \wedge \lbl \neq \ell_0+1)$:]
                        The proof of this case is similar to the one of the $( \apply{u\buf_0[0..i-1]}{a_0}(x) = 0 \wedge \lbl = \lbl'') \vee (u\apply{\buf_0[0..i-1]}{a_0}(x) \in \Val \setminus \{0,\bot\} \wedge \lbl = \ell_0+1)$ (except that we apply the \textsc{Execute-Branch-Rollback} rule).
                        Observe that the equality of the buffers' deep projections follow from (1) both configurations result in a rollback, and (2) the remaining reorder buffers are discarded (this is essential to ensure equality of projections).

                    \end{description}

                    \item[$\elt{u\buf_0}{i} = \tagged{\passign{\pc}{e}}{\notags}$:]   
                    Let $u\buf_0'$ be $\mask{\buf_0'}$.
                    From (a), we also have that $\elt{u\buf_0'}{i} =  \tagged{\passign{x}{e}}{\notags}$.
                    There are two cases:
                    \begin{description}
                        \item[$\exprEval{e}{\apply{u\buf_0[0..i-1]}{a_0}} \neq \bot$:]
                        From (a) and Lemma~\ref{lemma:tt:arch-seq:masking-preserves-indistinguishability}, we have that $\exprEval{e}{\apply{u\buf_0[0..i-1]}{a_0}} = \exprEval{e}{\apply{\buf_0'[0..i-1]}{a_0'}}$.
                        Therefore,  we can apply the \textsc{Execute-Assignment} and \textsc{Step} rules to $C_0$ and $C_0'$ as follows:
                        \begin{align*}
                            \buf_0 &:= \buf_0[0..i-1] \concat \labelled{\tagged{\passign{x}{e}}{\notags}}{\emptyset} \concat \buf_0[i+1 .. |\buf_0|]\\
                            u\buf_0 &:= \mask{\buf_0[0..i-1]} \concat \labelled{\tagged{\passign{x}{e}}{\notags}}{\emptyset} \concat \mask{\buf_0[i+1..|\buf_0|]}\\
                            v &:= \exprEval{e}{\apply{u\buf_0{[0..i-1]}}{a_0}}\\
                            u\buf_1 &:= \mask{\buf_0[0..i-1]} \concat {\tagged{\passign{x}{v}}{\notags}} \concat \mask{\buf_0[i+1..|\buf_0|]}\\
                            \buf_1 &:= \buf_0[0..i-1] \concat \labelled{\tagged{\passign{x}{v}}{\notags}}{\emptyset} \concat \buf_0[i+1..|\buf_0|]\\
                            \tup{m_0,a_0,u\buf_0, \CacheState_0, \BpState_0} &\muarchStep{\execute{i}}{} \tup{m_0,a_0,u\buf_1, \CacheState_0, \BpState_0}\\
                            \tup{m_0,a_0,\buf_0, \CacheState_0, \BpState_0, \SchedState_0} &\TtMuarchStep{}{} \tup{m_0,a_0,\buf_1, \CacheState_0, \BpState_0,\SchedUpdate(\SchedState_0, \BufProject{\buf_1})}\\
                            \buf_0' &:= \buf_0'[0..i-1] \concat \labelled{\tagged{\passign{x}{e}}{\notags}}{\emptyset} \concat \buf_0'[i+1 .. |\buf_0'|]\\
                            u\buf_0' &:= \mask{\buf_0'[0..i-1]} \concat \labelled{\tagged{\passign{x}{e}}{\notags}}{\emptyset} \concat \mask{\buf_0'[i+1..|\buf_0'|]}\\
                            u\buf_1' &:= \mask{\buf_0'[0..i-1]} \concat {\tagged{\passign{x}{v}}{\notags}} \concat \mask{\buf_0'[i+1..|\buf_0'|]}\\
                            \buf_1' &:= \buf_0'[0..i-1] \concat \labelled{\tagged{\passign{x}{v}}{\notags}}{\emptyset} \concat \buf_0'[i+1..|\buf_0'|]\\
                            \tup{m_0',a_0',u\buf_0', \CacheState_0', \BpState_0'} &\muarchStep{\execute{i}}{} \tup{m_0',a_0',u\buf_1', \CacheState_0', \BpState_0'}\\
                            \tup{m_0',a_0',\buf_0', \CacheState_0', \BpState_0', \SchedState_0'} & \TtMuarchStep{}{} \tup{m_0',a_0',\buf_1', \CacheState_0', \BpState_0',\SchedUpdate(\SchedState_0', \BufProject{\buf_1'})}
                        \end{align*}
                        We now show that $C_1 = \tup{m_0,a_0,\buf_1, \CacheState_0, \BpState_0,\SchedUpdate(\SchedState_0, \BufProject{\buf_1})}$ and $C_1' = \tup{m_0',a_0',\buf_1', \CacheState_0', \BpState_0',\SchedUpdate(\SchedState_0', \BufProject{\buf_1'})}$ are indistinguishable, i.e., i.e., $C_1 \sim C_1'$.
                        
                        For this, we need to show that:
                        \begin{description}
                            \item[$a_0 = a_0'$:] 
                            This follows from (a).
        
                            \item[$\DeepProject{\buf_1} = \DeepProject{\buf_1'}$:]
                            This follows from (a), $\buf_1= \buf_0[0..i-1] \concat \labelled{\tagged{\passign{x}{v}}{\notags}}{\emptyset} \concat \buf_0[i+1..|\buf_0|]$, and $\buf_1' = \buf_0'[0..i-1] \concat \labelled{\tagged{\passign{x}{v}}{\notags}}{\emptyset} \concat \buf_0'[i+1..|\buf_0'|]$.

                            \item[$\CacheState_0 = \CacheState_0':$] 
                            This follows from (a).
                            
                            \item[$\BpState_0 = \BpState_0':$] 
                            This follows from (a).
                
                            \item[$\SchedUpdate(\SchedState_0, \BufProject{\buf_1}) = \SchedUpdate(\SchedState_0', \BufProject{\buf_1'})$:]
                            This follows from (a) and $\DeepProject{\buf_1} = \DeepProject{\buf_1'}$.
                        \end{description}
                        Therefore, $C_1 \sim C_1'$ and (c) holds.

                        \item[$\exprEval{e}{\apply{u\buf_0[0..i-1]}{a_0}} = \bot$:] 
                        From (a) and Lemma~\ref{lemma:tt:arch-seq:masking-preserves-indistinguishability}, it follows that $\exprEval{e}{\apply{u\buf_0'[0..i-1]}{a_0'}} = \bot$.
                        Therefore, both configurations are stuck and (c) holds.
        
                    \end{description}

                    \item[$\elt{u\buf_0}{i} =   \tagged{\pmarkedassign{\pc}{e}}{\notags}$:]
                    
                    The proof of this case is similar to that of $\elt{u\buf_0}{i} = \tagged{\passign{\pc}{e}}{\notags}$.

                    \item[$\elt{u\buf_0}{i} =  \tagged{\pstore{x}{e}}{T}$:]   
                    The proof of this case is similar to that of $\elt{u\buf_0}{i} = \tagged{\passign{\pc}{e}}{\notags}$. 

                \end{description}
                Therefore, (c) holds in all cases.

                \item[$i > |u\buf_0| \vee \tagged{\pbarrier}{T'} \in u\buf_0{[0..i-1]}$:]
                From (a), it immediately follows that $i > |u\buf_0'| \vee \tagged{\pbarrier}{T'} \in u\buf_0'[0..i-1]$.
                Therefore, both configurations are stuck and (c) holds.
            \end{description}

            \item[$\neg \transmitGadget{\elt{\buf_0}{i}}$:]
            
            Then, from (a), we have that  $\neg\transmitGadget{\elt{\buf_0'}{i}}$.
            Therefore, we have that $u\buf_0 = \unlabelNda{\buf_0}{\execute{i}} = \drop{\buf_0}$ and $u\buf_0' = \unlabelNda{\buf_0'}{\execute{i}} = \drop{\buf_0'}$.

            There are two cases:
            \begin{description}
                \item[$i \leq |u\buf_0|$:]
                There are several cases depending on the $i$-th command in the reorder buffer: 
                \begin{description}
                    \item[$\elt{u\buf_0}{i} = \tagged{\passign{x}{e}}{T}$:]   
                    From (a), we get that also $\elt{u\buf_0'}{i} = \tagged{\passign{x}{e}}{T}$.
                    Moreover, observe $x \neq \pc$ (since $\neg \transmitGadget{\elt{\buf_0}{i}}$).
                    From (a), we know that $\exprEval{e}{\apply{u\buf_0[0..i-1]}{a_0}} = \bot \leftrightarrow \exprEval{e}{\apply{u\buf_0'[0..i-1]}{a_0'}} \neq \bot$.
                    Therefore, there are two cases:
                    \begin{description}
                        \item[${\exprEval{e}{\apply{u\buf_0[0..i-1]}{a_0}} \neq \bot \wedge\tagged{\pbarrier}{T'} \not\in {u\buf_0[0..i-1]}}$:]
                        Then, from (a), we have that ${\exprEval{e}{\apply{u\buf_0'[0..i-1]}{a_0'}} \neq \bot \wedge\tagged{\pbarrier}{T'} \not\in {u\buf_0'[0..i-1]}}$.
                        Therefore,  we can apply the \textsc{Execute-Assignment} and \textsc{Step} rules to $C_0$ and $C_0'$ as follows:
                        \begin{align*} 
                            \buf_0 &:= \buf_0[0..i-1] \concat \labelled{\tagged{\passign{x}{e}}{\notags}}{L} \concat \buf_0[i+1 .. |\buf_0|]\\
                            u\buf_0 &:= \drop{\buf_0[0..i-1]} \concat \tagged{\passign{x}{e}}{\notags} \concat \drop{\buf_0[i+1 .. |\buf_0|]}\\
                            v &:= \exprEval{e}{\apply{\drop{\buf_0[0..i-1]}}{a_0}} \\
                            u\buf_1 &:= \drop{\buf_0[0..i-1]} \concat \tagged{\passign{x}{v}}{\notags} \concat \drop{\buf_0[i+1 .. |\buf_0|]}\\
                            \buf_1 &:= \buf_0[0..i-1] \concat \labelled{\tagged{\passign{x}{v}}{\notags}}{L} \concat \buf_0[i+1 .. |\buf_0|]\\
                            \tup{m_0,a_0,u\buf_0, \CacheState_0, \BpState_0} &\muarchStep{\execute{i}}{} \tup{m_0,a_0,u\buf_1, \CacheState_0, \BpState_0}\\
                            \tup{m_0,a_0,\buf_0, \CacheState_0, \BpState_0 , \SchedState_0} &\TtMuarchStep{}{} \tup{m_0,a_0,\buf_1, \CacheState_0, \BpState_0,\SchedUpdate(\SchedState_0, \BufProject{\buf_1})}\\
                            \buf_0' &:= \buf_0'[0..i-1] \concat \labelled{\tagged{\passign{x}{e}}{\notags}}{L} \concat \buf_0'[i+1 .. |\buf_0'|]\\
                            u\buf_0' &:= \drop{\buf_0'[0..i-1]} \concat \tagged{\passign{x}{e}}{\notags} \concat \drop{\buf_0'[i+1 .. |\buf_0'|]}\\
                            v' &:= \exprEval{e}{\apply{\drop{\buf_0'[0..i-1]}}{a_0'}} \\
                            u\buf_1' &:= \drop{\buf_0'[0..i-1]} \concat \tagged{\passign{x}{v'}}{\notags} \concat \drop{\buf_0'[i+1 .. |\buf_0'|]}\\
                            \buf_1' &:= \buf_0'[0..i-1] \concat \labelled{\tagged{\passign{x}{v'}}{\notags}}{L} \concat \buf_0'[i+1 .. |\buf_0'|]\\
                            \tup{m_0',a_0',u\buf_0', \CacheState_0', \BpState'} &\muarchStep{\execute{i}}{} \tup{m_0',a_0',u\buf_1', \CacheState_0', \BpState_0'}\\
                            \tup{m_0',a_0',\buf_0', \CacheState_0', \BpState_0', \SchedState_0'} & \TtMuarchStep{}{} \tup{m_0',a_0',\buf_1', \CacheState_0', \BpState_0',\SchedUpdate(\SchedState_0', \BufProject{\buf_1'})}
                        \end{align*}
                        We now show that $C_1 = \tup{m_0,a_0,\buf_1, \CacheState_0, \BpState_0,\SchedUpdate(\SchedState_0, \BufProject{\buf_1})}$ and $C_1' =  \tup{m_0',a_0',\buf_1', \CacheState_0', \BpState_0',\SchedUpdate(\SchedState_0', \BufProject{\buf_1'})}$ are indistinguishable, i.e., i.e., $C_1 \sim C_1'$.

                        For this, we need to show:
                        \begin{description}
                            \item[$a_0 = a_0'$:] 
                            This follows from (a).
                            
                            \item[$\DeepProject{\buf_1} = \DeepProject{\buf_1'}$:]
                            For this, we need to show that $\DeepProject{ \labelled{\tagged{\passign{x}{v}}{\notags}}{L} } = \DeepProject{ \labelled{\tagged{\passign{x}{v'}}{\notags}}{L} }$.
                            There are two cases:
                            \begin{description}
                                \item[$L = \emptyset$:]
                                Then, $\DeepProject{ \labelled{\tagged{\passign{x}{v}}{\notags}}{L} } = \labelled{\tagged{\passign{x}{v}}{\notags}}{L}$ and $\DeepProject{ \labelled{\tagged{\passign{x}{v'}}{\notags}}{L} } = \labelled{\tagged{\passign{x}{v'}}{\notags}}{L}$.
                                From the well-formedness of buffers, $e$ depends only on data labelled at most $L$, i.e., $\emptyset$.
                                From this and (a), therefore, we get that $v =v'$ (since values whose label is $\emptyset$ are visible in $\DeepProject{ \buf }$).
                                Hence, $\DeepProject{ \labelled{\tagged{\passign{x}{v}}{\notags}}{L} } = \DeepProject{ \labelled{\tagged{\passign{x}{v'}}{\notags}}{L} }$ holds.

                                \item[$L \neq \emptyset$:] 
                                Then, $\DeepProject{ \labelled{\tagged{\passign{x}{v}}{\notags}}{L} } = \labelled{\tagged{\passign{x}{\unresolved}}{\notags}}{L}$ and $\DeepProject{ \labelled{\tagged{\passign{x}{v'}}{\notags}}{L} } = \labelled{\tagged{\passign{x}{\unresolved}}{\notags}}{L}$.
                                Therefore, $\DeepProject{ \labelled{\tagged{\passign{x}{v}}{\notags}}{L} } = \DeepProject{ \labelled{\tagged{\passign{x}{v'}}{\notags}}{L} }$ holds.
                            \end{description}

                            \item[$\CacheState_0 = \CacheState_0':$] 
                            This follows from (a).
                            
                            \item[$\BpState_0 = \BpState_0':$] 
                            This follows from (a).
                
                            \item[$\SchedUpdate(\SchedState_0, \BufProject{\buf_1}) = \SchedUpdate(\SchedState_0', \BufProject{\buf_1'})$:]
                            This follows from (a) and $\DeepProject{\buf_1} = \DeepProject{\buf_1'}$.
                        \end{description}

                        \item[${\exprEval{e}{\apply{u\buf_0[0..i-1]}{a_0}} = \bot \vee \tagged{\pbarrier}{T'} \in {u\buf_0[0..i-1]}}$:] 
                        From (a), we have that $\exprEval{e}{\apply{u\buf_0'[0..i-1]}{a_0'}} = \bot \vee \tagged{\pbarrier}{T'} \in {u\buf_0'[0..i-1]}$ holds as well.
                        Therefore, both configurations are stuck and (c) holds.
                    \end{description}

                    \item[$\elt{u\buf_0}{i} = \tagged{\pskip{}}{T}$:]
                    The proof of this case is similar to that of the $\elt{u\buf_0}{i} = \tagged{\passign{x}{e}}{T}$ case.
                    \item[$\elt{u\buf_0}{i} = \tagged{\pbarrier{}}{T}$:]
                    The proof of this case is similar to that of the $\elt{u\buf_0}{i} = \tagged{\passign{x}{e}}{T}$ case.
                    
                \end{description}

                \item[$i > |u\buf_0|$:]
                From (a), we have that $i > |u\buf_0'|$ as well.
                Therefore, both configurations are stuck and (c) holds.
            \end{description}
        \end{description}
        Therefore, (c) holds in all cases.
     
        \item[$d = \retire{}$:]
        Therefore, we can only apply one of the $\retire{}$ rules depending on the head of the reorder buffer in $\buf_0$.
        There are five cases:
        \begin{description}
            \item[$\buf_0 = \labelled{\tagged{\pskip}{\notags}}{L} \concat \buf_1 $:] 
            From (a), we get that  $\buf_0' = \labelled{\tagged{\pskip}{\notags}}{L} \concat \buf_1' $ and $\DeepProject{\buf_1}_{\tup{m_0,a_0}} = \DeepProject{\buf_1'}_{\tup{m_0',a_0'}}$.
            Therefore, we can apply the \textsc{Retire-Skip} and \textsc{Step} rules to $C_0$ and $C_0'$ as follows:
            \begin{align*}
                \buf_0 &:= \labelled{\tagged{\pskip}{\notags}}{L} \concat \buf_1  \\
                u\buf_0 &:= \tagged{\pskip}{\notags} \concat \drop{\buf_1} \\
                u\buf_1 &:=  \drop{\buf_1} \\
                \tup{m_0,a_0, u\buf_0,\CacheState_0,\BpState_0} &\muarchStep{\retire}{} \tup{m_0, a_0, u\buf_1, \CacheState_0,\BpState_0}\\
                \tup{m_0,a_0, \labelled{\tagged{\pskip}{\notags}}{L} \concat \buf_1,\CacheState_0,\BpState_0, \SchedState_0} &\TtMuarchStep{}{} \tup{m_0, a_0, \buf_1, \CacheState_0,\BpState_0, \SchedUpdate(\SchedState_0, \BufProject{\buf_1})}\\
                \buf_0' &:= \labelled{\tagged{\pskip}{\notags}}{L} \concat \buf_1'  \\
                u\buf_0' &:= \tagged{\pskip}{\notags} \concat \drop{\buf_1'} \\
                u\buf_1' &:=  \drop{\buf_1'} \\
                \tup{m_0',a_0',u\buf_0',\CacheState_0',\BpState_0'} &\muarchStep{\retire}{} \tup{m_0', a_0', u\buf_1', \CacheState_0',\BpState_0'}\\
                \tup{m_0',a_0',\labelled{\tagged{\pskip}{\notags}}{L} \concat \buf_1',\CacheState_0',\BpState_0', \SchedState_0'} &\TtMuarchStep{}{} \tup{m_0', a_0', \buf_1', \CacheState_0',\BpState_0', \SchedUpdate(\SchedState_0', \BufProject{\buf_1'})}
            \end{align*}
            We now show that $C_1 = \tup{m_0, a_0, \buf_1, \CacheState_0,\BpState_0, \SchedUpdate(\SchedState_0, \BufProject{\buf_1})}$ and $C_1' = \tup{m_0', a_0', \buf_1', \CacheState_0',\BpState_0', \SchedUpdate(\SchedState_0', \BufProject{\buf_1'})}$ are indistinguishable, i.e., $C_1 \sim C_1'$.
            
            For this, we need to show that: 
            \begin{description}
              \item[$a_0 = a_0'$:]
              This follows from (a). 
              \item[$\DeepProject{\buf_1} = \DeepProject{ \buf_1'}$:]
              This follows from $\buf_0 = \labelled{\tagged{\pskip}{\notags}}{L} \concat \buf_1 $,  $\buf_0' = \labelled{\tagged{\pskip}{\notags}}{L} \concat \buf_1' $, and $\DeepProject{\buf_0} = \DeepProject{ \buf_0'}$.

              \item[$\CacheState_0 = \CacheState_0'$:]
              This follows from (a).
              \item[$\BpState_0 = \BpState_0'$:]
              This follows from (a).

              \item[$\SchedUpdate(\SchedState_0, \BufProject{\buf_1}) = \SchedUpdate(\SchedState_0', \BufProject{\buf_1'})$:]
              This follows from (a) and $\DeepProject{\buf_1} = \DeepProject{ \buf_1'}\rightarrow \BufProject{\buf_1} = \BufProject{\buf_1'}$.
            \end{description}
            Therefore, $C_1 \sim C_1'$ and (c) holds.

            \item[$\buf_0 = \labelled{\tagged{\pbarrier}{\notags}}{L} \concat \buf_1 $:]
            The proof of this case is similar to that of the case $\buf_0 = \labelled{\tagged{\pskip}{\notags}}{L} \concat \buf_1 $.
    
            \item[$\buf_0 = \labelled{\tagged{\passign{x}{v}}{\notags}}{L} \concat \buf_1 $:] 
            From (a), we get that  $\buf_0' =\labelled{\tagged{\passign{x}{v}}{\notags}}{L} \concat \buf_1' $ and $\DeepProject{\buf_1}_{\tup{m_1,a_1[x \mapsto v]}} = \DeepProject{\buf_1'}_{\tup{m_1',a_1'[x \mapsto v]}}$.
            Observe that if $v \not\in \Val$ then both computations are stuck and (c) holds (since there is no $C_1$ such that $C_0 \TtMuarchStep{}{} C_1$ and no $C_1'$ such that $C_0' \TtMuarchStep{}{} C_1'$).
            In the following, therefore, we assume that $v,v' \in \Val$.
            Therefore, we can apply the \textsc{Retire-Assignment} and \textsc{Step-Others} rules to $C_0$ and $C_0'$ as follows:
            \begin{align*}
                \buf_0 &:= \labelled{\tagged{\passign{x}{v}}{\notags}}{L} \concat \buf_1 \\
                u\buf_0 &:= \tagged{\passign{x}{v}}{\notags} \concat \drop{\buf_1}\\
                u\buf_1 &:= \drop{\buf_1}\\
                \tup{m_0,a_0,u\buf_0,\CacheState_0,\BpState_0} &\muarchStep{\retire}{} \tup{m_0, a_0[x\mapsto v], u\buf_1, \CacheState_0,\BpState_0}\\
                \tup{m_0,a_0, \labelled{\tagged{\passign{x}{v}}{\notags}}{L} \concat \buf_1,\CacheState_0,\BpState_0, \SchedState_0} &\TtMuarchStep{}{} \tup{m_0, a_0[x \mapsto v], \buf_1, \CacheState_0,\BpState_0, \SchedUpdate(\SchedState_0, \BufProject{\buf_1})}\\
                \buf_0' &:= \labelled{\tagged{\passign{x}{v}}{\notags}}{L} \concat \buf_1' \\
                u\buf_0' &:= \tagged{\passign{x}{v}}{\notags} \concat \drop{\buf_1'}\\
                u\buf_1' &:= \drop{\buf_1'}\\
                \tup{m_0',a_0',u\buf_0',\CacheState_0',\BpState_0'} &\muarchStep{\retire}{} \tup{m_0', a_0'[x \mapsto v], u\buf_1', \CacheState_0',\BpState_0'}\\
                \tup{m_0',a_0', \labelled{\tagged{\passign{x}{v}}{\notags}}{L} \concat \buf_1',\CacheState_0',\BpState_0', \SchedState_0'} &\TtMuarchStep{}{} \tup{m_0', a_0'[x \mapsto v], \buf_1', \CacheState_0',\BpState_0', \SchedUpdate(\SchedState_0', \BufProject{\buf_1'})}
            \end{align*}
            We now show that $C_1 = \tup{m_0, a_0[x \mapsto v], \buf_1, \CacheState_0,\BpState_0, \SchedUpdate(\SchedState_0, \BufProject{\buf_1})}$ and $C_1' = \tup{m_0', a_0'[x \mapsto v'], \buf_1', \CacheState_0',\BpState_0', \SchedUpdate(\SchedState_0', \BufProject{\buf_1'})}$ are indistinguishable, i.e., $C_1 \sim C_1'$.
            For this, we need to show that:
            \begin{description}
                \item[$a_0{[x \mapsto v]} = a_0'{[x \mapsto v]}$:]
                This follows from (a).
    
                \item[$\DeepProject{\buf_1} = \DeepProject{\buf_1'}$:] 
                This follows from $\buf_0 = \labelled{\tagged{\passign{x}{v}}{\notags}}{L} \concat \buf_1' $, $\buf_0' = \labelled{\tagged{\passign{x}{v}}{\notags}}{L} \concat \buf_1' $, and (a).
    
                \item[$\CacheState_0 = \CacheState_0'$:]
                This follows from (a).
    
                \item[$\BpState_0 = \BpState_0':$] 
                This follows from (a).
    
                \item[$\SchedUpdate(\SchedState_0, \BufProject{\buf_1}) = \SchedUpdate(\SchedState_0', \BufProject{\buf_1'})$:]
                This follows from (a) and $\DeepProject{\buf_1} = \DeepProject{ \buf_1'} \rightarrow \BufProject{\buf_1} = \BufProject{\buf_1'}$.
            \end{description}
            Therefore, $C_1 \sim C_1'$ and (c) holds.
    
            \item[$\buf_0 = \tagged{\pmarkedassign{x}{v}}{\notags} \concat \buf_1 $:]
            The proof of this case is similar to that of the case $\buf_0 = \tagged{\passign{x}{v}}{\notags} \concat \buf_1 $.
    
            \item[$\buf_0 = \tagged{\pstore{v}{n}}{\notags} \concat \buf_1 $:]
            The proof of this case is similar to that of the case $\buf_0 = \tagged{\passign{x}{v}}{\notags} \concat \buf_1 $ (the equality between the final cache states directly follow from $\buf_0'$ being of the form $\tagged{\pstore{v}{n}}{\notags} \concat \buf_1'$).
        
        \end{description}
        Therefore, (c) holds for all the cases.

    \end{description}
    Since (c) holds for all cases, this completes the proof of our lemma.
    \end{proof}

\subsubsection{Main lemma}

\begin{lemma}\label{lemma:tt:arch-seq:main-lemma}
	Let $p$ be a well-formed program, $\CacheState_0$ be the initial cache state, $\BpState_0$ be the initial branch predictor state, and $\SchedState_0$ be the initial scheduler state, $\sigma_0 = \tup{m,a}, \sigma_0' = \tup{m',a'}$ be initial \archstate{}s, and $C_0 = \tup{m,a, \emptysequence, \CacheState_0, \BpState_0, \SchedState_0}$ and $C_0' = \tup{m',a', \emptysequence, \CacheState_0, \BpState_0, \SchedState_0}$ be hardware configurations.
	Furthermore, let $\crun := \sigma_0$ $\ArchSeqInterfStep{o_1}{}$ $\sigma_1$ $\ArchSeqInterfStep{o_2}{}$ $\ldots$  $\ArchSeqInterfStep{o_{n-1}}{}$  $\sigma_n$ and $\crunp:=\sigma_0'$ $\ArchSeqInterfStep{o_1'}{}$  $\sigma_1' $ $\ArchSeqInterfStep{o_2'}{}$ $\ldots$  $\ArchSeqInterfStep{o_{n-1}'}{} $ $\sigma_n$ be two runs for the $\ArchSeqInterf{\cdot}$ contract where  $\sigma_n, \sigma_{n}$ are final \archstate{}s.
	If $o_i = o_i'$ for all $0 < i < n$, then there is a $k \in \Nat$ and $C_0, \ldots, C_k, C_0', \ldots, C_k'$ such that $C_0 \TtMuarchStep{}{} C_1 \TtMuarchStep{}{} \ldots \TtMuarchStep{}{} C_k$, $C_0' \TtMuarchStep{}{} C_1' \TtMuarchStep{}{} \ldots \TtMuarchStep{}{} C_k'$, and one of the following conditions hold:
	\begin{compactenum}
		\item $C_0, C_0'$ are initial states, $\forall 0 \leq i \leq k.\ C_{i} \approx C_{i}'$, and $C_, C_k'$ are final states, or
		\item $C_0, C_0'$ are initial states, $\forall 0 \leq i \leq k.\ C_{i} \approx C_{i}'$, and there are no $C_{k+1}$ such that $C_{k+1} \neq C_k \wedge C_k \TtMuarchStep{}{} C_{k+1}$ and no $C_{k+1}'$ such that $C_{k+1}' \neq C_k' \wedge C_k' \TtMuarchStep{}{} C_{k+1}'$.
	\end{compactenum}
\end{lemma}

\begin{proof}
The proof is similar to that of Lemma~\ref{lemma:loadDelay:arch-seq:main-lemma} (except that we use Lemmas~\ref{lemma:tt:arch-seq:mapping-is-correct} and~\ref{lemma:tt:arch-seq:trace-equiv-implies-stepwise-indistinguishability} instead of Lemmas Lemmas~\ref{lemma:loadDelay:arch-seq:mapping-is-correct} and~\ref{lemma:loadDelay:arch-seq:trace-equiv-implies-stepwise-indistinguishability}).
\end{proof}
}

\end{document}